\documentclass[english,openany,12pt]{article}

\usepackage[cp1251]{inputenc}
\usepackage{indentfirst}
\usepackage{url}

\usepackage{cmap}
\usepackage[T2A]{fontenc}

\usepackage[english]{babel}
\usepackage{amsfonts,amssymb,amsmath,amscd,amsthm}

\usepackage{cases}

\usepackage{pgf,tikz,circuitikz}\usetikzlibrary{arrows}
\usepackage{wrapfig}
\usepackage{geometry}
\usepackage{titling}

\usepackage{color}
\usepackage{xcolor}

\newcommand{\mscomm}[1]{\begingroup\color{green}\endgroup}
\newcommand{\new}[1]{\begingroup#1\endgroup}

\newcommand{\blue}[1]{#1}
\newcommand{\red}[1]{#1}

\usepackage{graphicx}

\usepackage[matrix,arrow,curve]{xy}
\usepackage{pdfpages}

\usepackage{diagbox}
\usepackage{makecell}
\usepackage{faktor}
\usepackage[breaklinks,pdfstartview=FitH,colorlinks,linktocpage,bookmarks=false]{hyperref}

\makeatletter
\newcommand{\l@myshrink}[2]{\vspace{-0.4cm}}
\newcommand{\l@myshrinkalt}[2]{\vspace{-0.05cm}}
\makeatother

\long\def\comment#1\endcomment{}
\long\def\solutions#1\endsolutions{#1}
\long\def\lktgonly#1\endlktgonly{}
\long\def\iumonly#1\endiumonly{}
\long\def\tmpcomment#1\endtmpcomment{}
\long\def\skipprint#1\endskipprint{#1}

\theoremstyle{theorem}
\newtheorem{theorem}{Theorem}
\newtheorem{lemma}{Lemma}
\newtheorem{corollary}{Corollary}
\newtheorem{proposition}{Proposition}
\newtheorem{example}{Example}
\newtheorem*{theorem2prime*}{Theorem~5$'$}

\theoremstyle{remark}
\newtheorem{remark}{Remark}

\theoremstyle{definition}
\newtheorem{definition}{Definition}
\newtheorem*{definition*}{Definition}
\newtheorem{algorithm}{Algorithm}
\newtheorem{pr}{Problem}

\newenvironment {th*}[1]
    {\gdef\thname{#1} \begin{thn}}%
    {\end{thn}}
\newtheorem*{thn}{\thname}

\def\No{\textnumero}

\begin{document}

\title{\ \\[-2.5cm]Feynman checkers:\\[-0.2cm]
{\Large towards algorithmic quantum theory}\vspace{-0.4cm}}

\author{M. Skopenkov and A. Ustinov}

\date{}

\maketitle

\vspace{-1.9cm}

\begin{abstract}
We survey and develop the most elementary model of electron motion introduced by R.Feynman. In this game, a checker moves on a checkerboard by simple rules, and we count the turns. Feynman checkers are also known as a one-dimensional quantum walk or an Ising model at imaginary temperature. We solve mathematically a problem by R.Feynman from 1965, which was to prove that the discrete model (for large time, small average velocity, and small lattice step) is consistent with the continuum one. We study asymptotic properties of the model (for small lattice step and large time) improving the results
by J.Narlikar from 1972 and by T.Sunada--T.Tate from 2012.
For the first time we observe and prove concentration of measure in the small-lattice-step limit.  We perform the second quantization of the model.


\end{abstract}



\vspace{-0.8cm}

\tableofcontents

\newpage

\addcontentsline{toc}{myshrink}{}
\addcontentsline{toc}{myshrink}{}

\section{Introduction}



We survey and develop the most elementary model of electron motion introduced by R.~Feynman (see Figure~\ref{P-contour}).
In this game, a checker moves on a checkerboard by simple rules, and we count the turns (see Definition~\ref{def-mass}). Feynman checkers can be viewed as a one-dimensional quantum walk, or an Ising model, or count of Young diagrams of certain type.

\begin{figure}[htbp]
  \begin{tabular}{ccc}
  \hspace{-0.2cm}
  \includegraphics[width=0.31\textwidth]
  {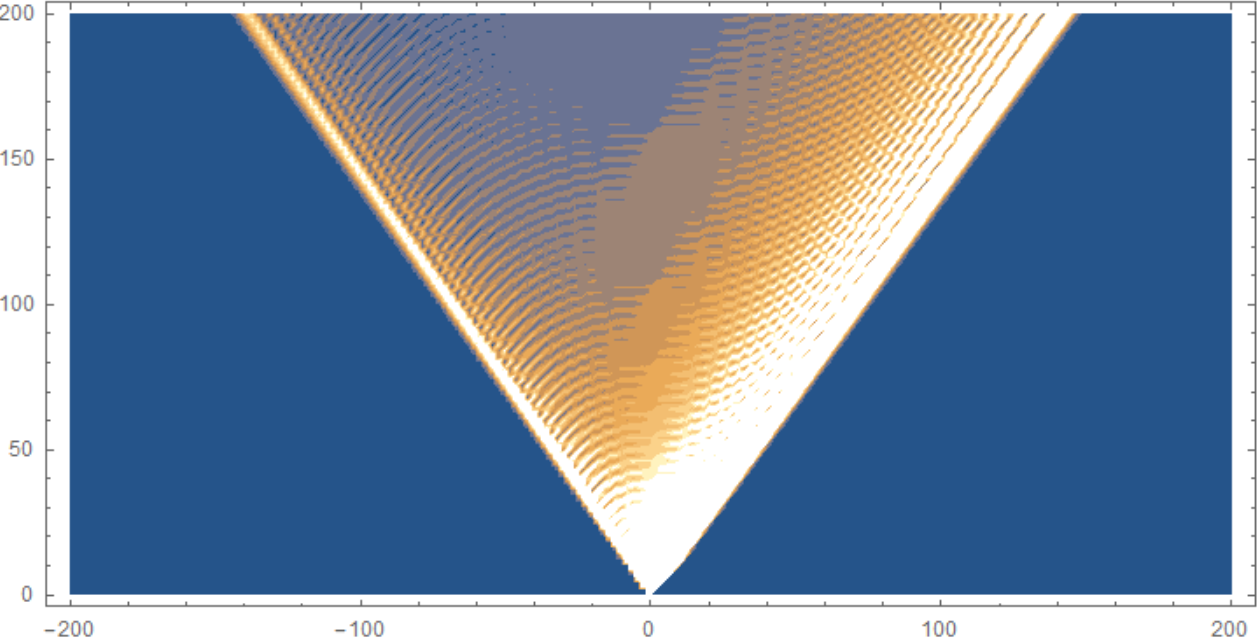}
  \hspace{-0.2cm} & \hspace{-0.2cm}
  \includegraphics[width=0.31\textwidth]
  {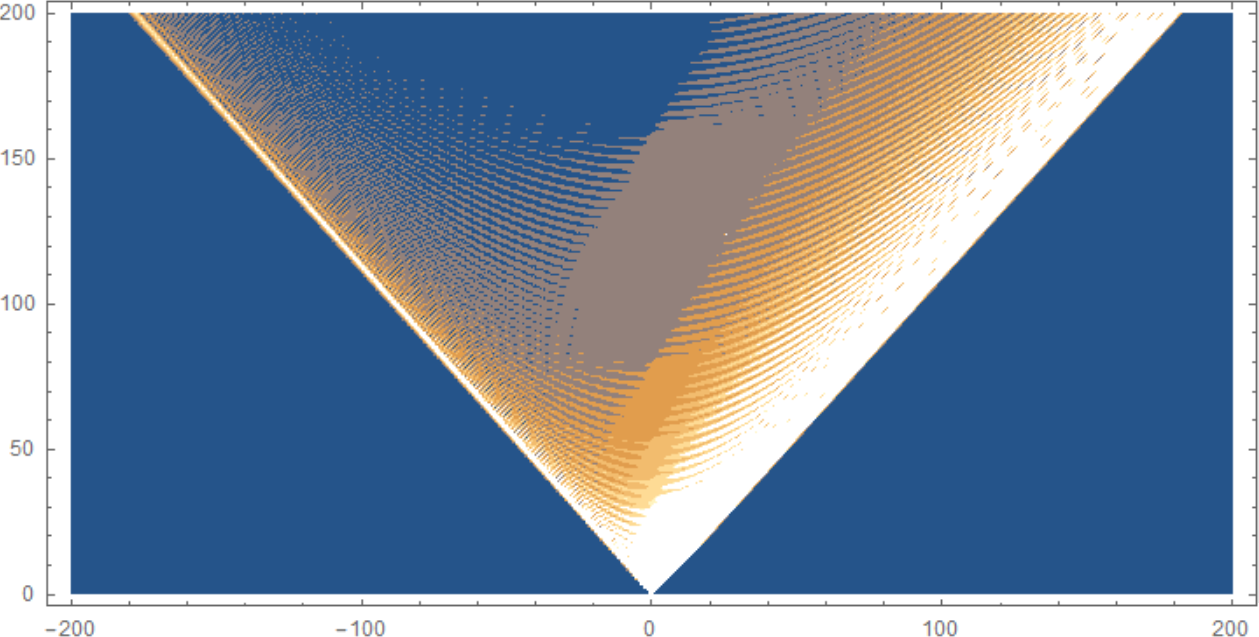}
  \hspace{-0.2cm} & \hspace{-0.2cm}
  \includegraphics[width=0.31\textwidth]
  {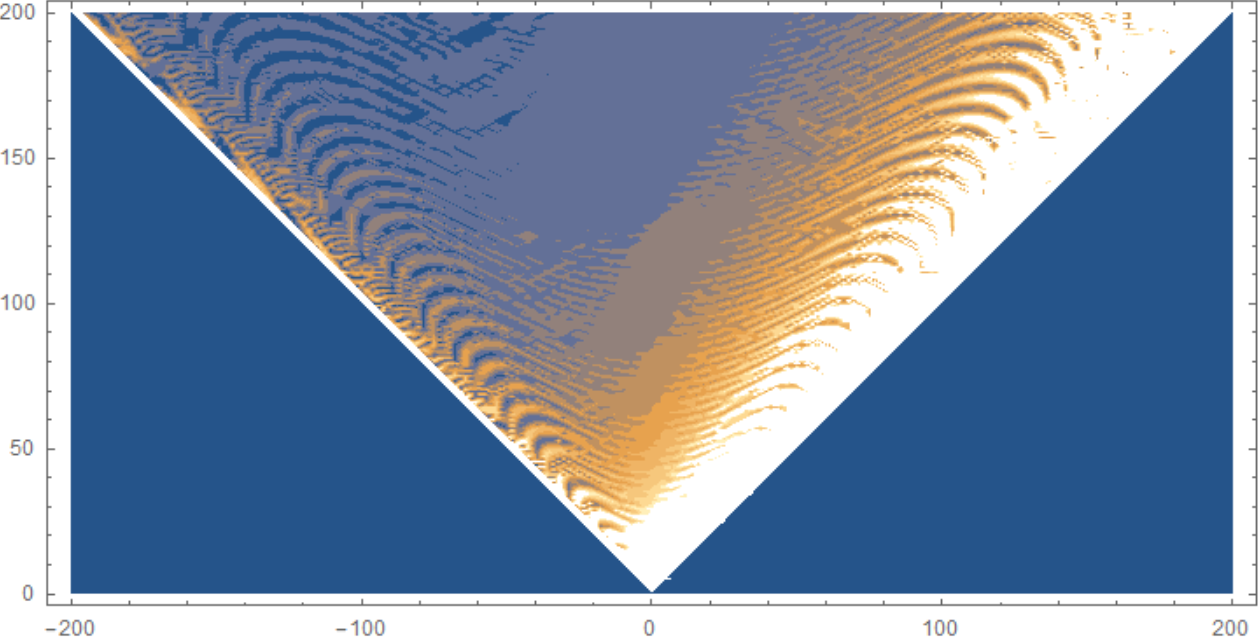}
  \hspace{-0.2cm} \\
  \includegraphics[width=0.20\textwidth]
  {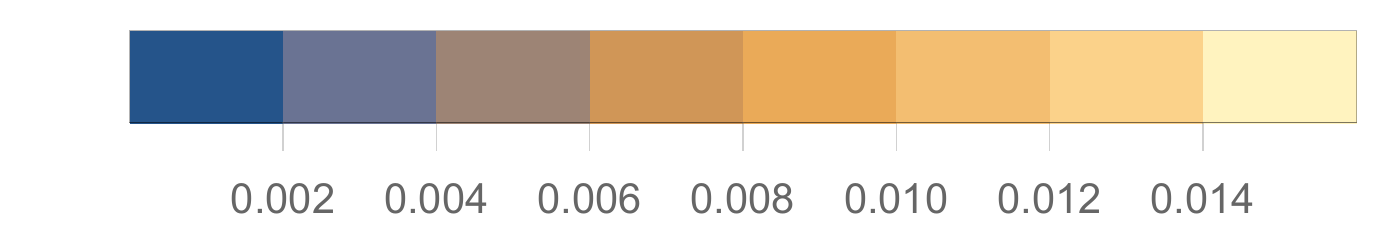}  &
  \includegraphics[width=0.15\textwidth]
  {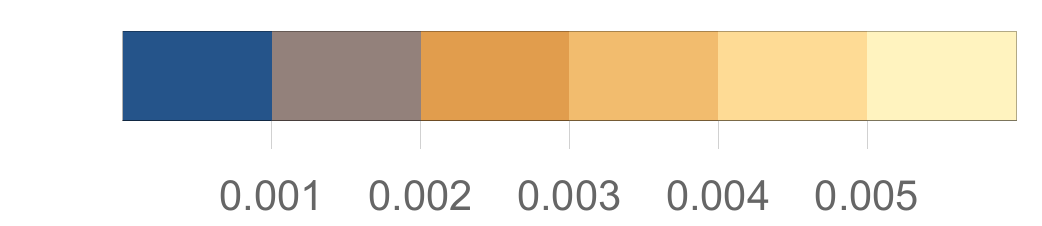} &
  \includegraphics[width=0.22\textwidth]
  {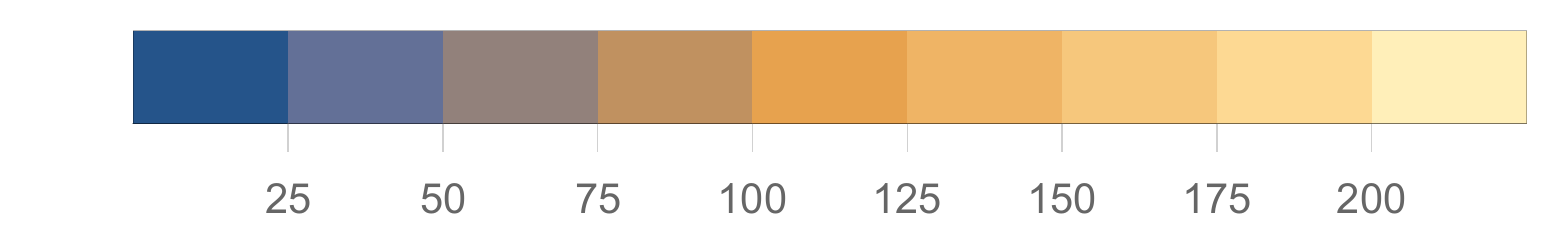}  \\
  \end{tabular}
  \caption{The probability to find an electron in a small square around a given point 
  (white depicts strong oscillations of the probability). Left:
  in the basic model from \S\ref{sec-basic}  (cf.~\cite[Figure~6]{Yepez-05}).
  Middle: 
  in the upgrade from~\S\ref{sec-mass} for smaller square side.
  Right: in continuum theory. For the latter, the relative probability density is depicted.
  } \label{P-contour}
  \vspace{-0.2cm}
\end{figure}

\addcontentsline{toc}{myshrinkalt}{}

\subsection{Motivation}

The simplest way to understand what is the model about is the classical \emph{double-slit experiment} (see Figure~\ref{Double-slit}). In this experiment, a (\emph{coherent}) beam of electrons is directed towards a plate pierced by two parallel slits, and the part of the beam passing through the slits is observed on a screen behind the plate. If one of the slits is closed, then the beam illuminates a spot on the screen. If both slits are open, one would expect a larger spot, but in fact one observes a sequence of bright and dark bands (\emph{interferogram}).

This shows 
that electrons behave like a wave: the waves travel through both slits, and the contributions of the two paths either amplify or cancel each other
depending on the final phases. 

Further, if the electrons are sent through the slits one at time, then single dots appear on the screen, as expected. Remarkably, however, the same interferogram with bright and dark bands emerges when the electrons are allowed to build up one by one. One cannot predict where a particular electron hits the screen; all we can do is to compute the probability to find the electron at a given place.

The \emph{Feynman sum-over-paths} (or \emph{path integral}) method of computing such probabilities is to assign 
phases to all possible paths and to sum up the resulting waves (see \cite{Feynman, Feynman-Gibbs}). \emph{Feynman checkers} (or \emph{Feynman checkerboard}) is a particularly simple combinatorial rule for those phases in the case of an electron
freely moving (or better jumping) along a line.

\begin{figure}[htbp]
  \centering
  \includegraphics[width=0.24\textwidth]{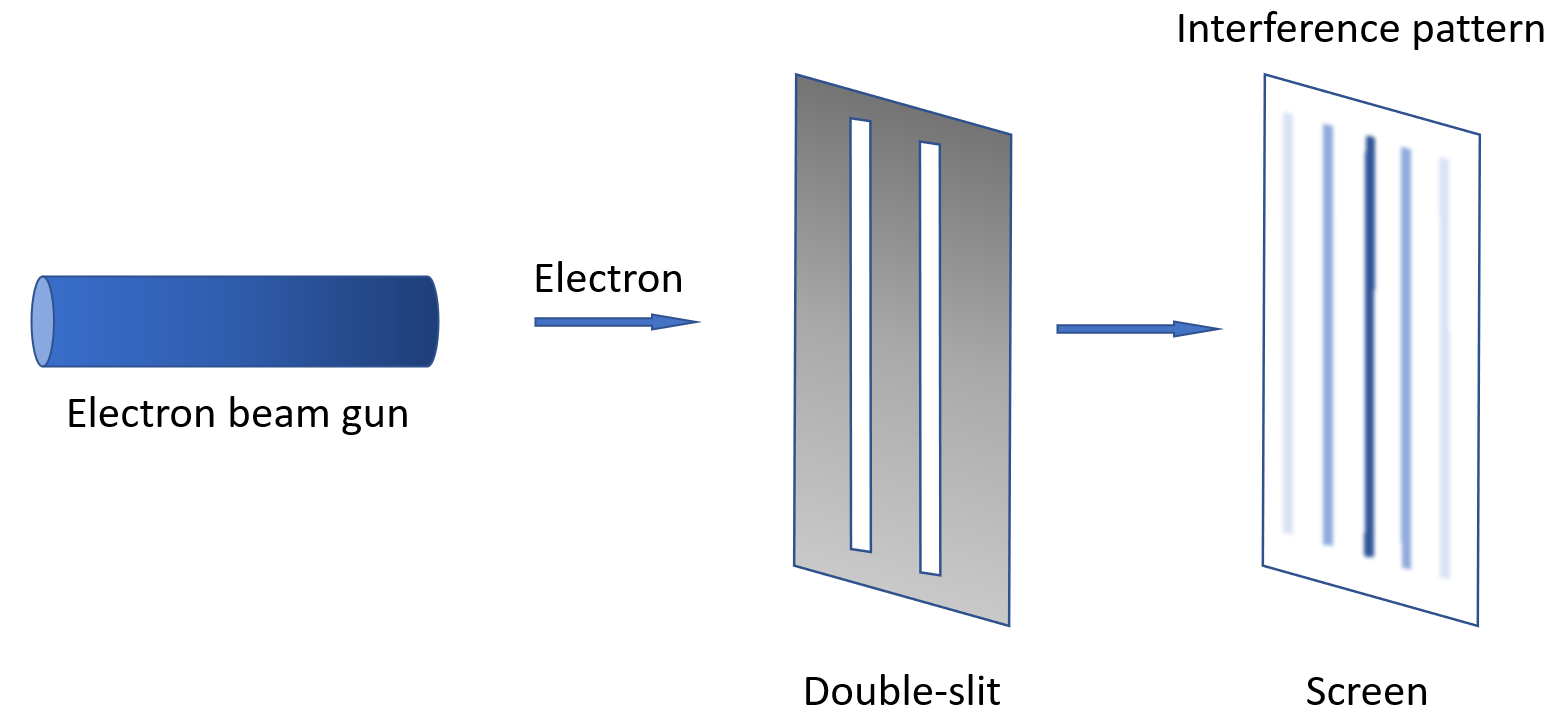}
  \caption{
  Double-slit experiment}\label{Double-slit}
  \vspace{-0.2cm}
\end{figure}

\addcontentsline{toc}{myshrinkalt}{}

\subsection{Background} \label{ssec-background}

\textbf{The beginning.}
The checkers model was invented by R.~Feynman in 1940s \cite{Schweber-86} and first published in 1965 \cite{Feynman-Gibbs}. 
In Problem~2.6 there, a function on a lattice of small step $\varepsilon$ was constructed (called \emph{kernel}; see~\eqref{eq-def-mass}) and the following task was posed:
\begin{quotation}
  If the time interval is very long ($t_b-t_a\gg \hbar/mc^2$) and the average velocity is small [$x_b-x_a\ll c(t_b-t_a)$], show that the resulting kernel is approximately the same as that for a free particle [given in Eq.~(3-3)], except for a factor $\exp[(imc^2/\hbar)(t_b-t_a)]$.
\end{quotation}
Mathematically, this means that the kernel (divided by $2i\varepsilon\exp[(-imc^2/\hbar)(t_b-t_a)]$) asymptotically equals \emph{free-particle kernel}~\eqref{eq-free-particle-kernel} (this is Eq.~(3-3) from \cite{Feynman-Gibbs}) in the triple limit when time tends to infinity, whereas the average velocity and the lattice step tend to zero (see Table~\ref{table-propagators} and Figure~\ref{fig-triple-limit}). Both scaling by the lattice step and tending it to zero were understood, otherwise the mentioned ``exceptional'' factor would be different (see Example~\ref{p-Feynman-couterexample}).
We show that the assertion, although incorrect literally, holds under mild assumptions (see Corollary~\ref{cor-feynman-problem}).

Although the Feynman problem might seem self-evident for theoretical physicists, even the first step of a mathematical solution (disproving the assertion as stated) is not found in literature. As usual, the difficulty is to prove the convergence rather than to guess the limit.


\begin{table}[hb]
  \centering
  \begin{tabular}{|l|c|c|c|l|}
    \hline
    propagator & \hspace{-0.2cm}continuum\hspace{-0.2cm} & \hspace{-0.1cm}lattice\hspace{-0.1cm} & context & references \\
    \hline
    free-particle kernel & \eqref{eq-free-particle-kernel} &
    - & 
    quantum mechanics &  \cite[(3-3)]{Feynman-Gibbs} \\
    \hline
    spin-$1/2$ 
    retarded propagator & \eqref{eq-relativistic-propagator},\eqref{eq-double-fourier-retarded} & \eqref{eq-def-mass} &
    relativistic 
    & cf.~\cite[(13)]{Jacobson-Schulman-84}
    \\
    & & & quantum mechanics & and~\cite[(2-27)]{Feynman-Gibbs}
    \\
    \hline
    spin-$1/2$ 
    Feynman propagator & \eqref{eq-feynman-propagator},\eqref{eq-double-fourier-feynman} & \eqref{eq-def-anti1} &
    quantum field theory & cf.~\cite[\S9F]{Bender-etal-94} \\
    \hline
  \end{tabular}
  \caption{Expressions for the \emph{propagators} of a particle freely moving in $1$ space and $1$ time dimension. The meaning of the norm square of a propagator is the relative probability density to find the particle at a particular point, or alternatively, the charge density at the point.} \label{table-propagators}
\end{table}

\begin{figure}[htbp]
  \vspace{-0.5cm}
  \centering
  \includegraphics[width=0.98\textwidth]{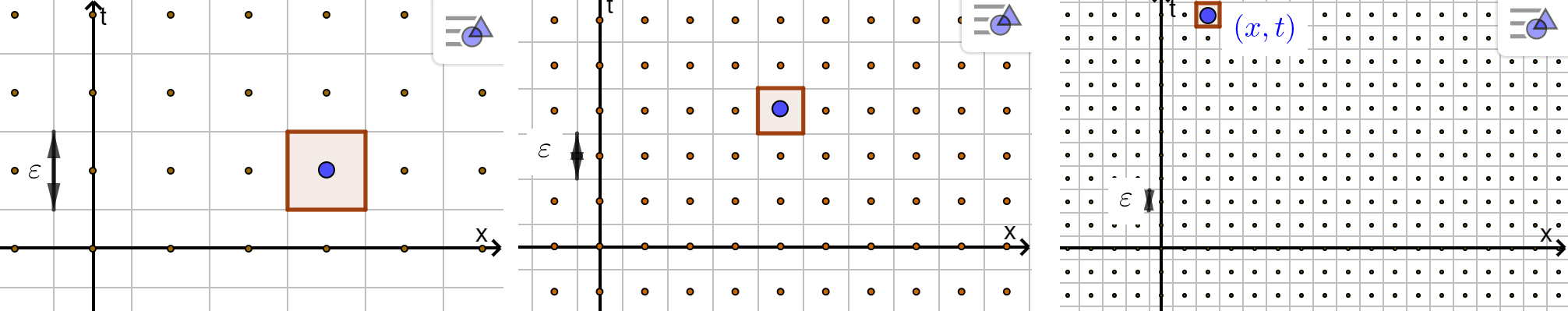}
  \caption{
  The Feynman triple limit: $t\to+\infty$, $x/t\to 0$, $\varepsilon\to 0$
  }\label{fig-triple-limit}
\end{figure}

In 1972 J.~Narlikar discovered that the above kernel reproduces the \emph{spin-$1/2$ retarded propagator} in the different limit when the lattice step tends to zero but time stays fixed \cite{Narlikar-72} (see Table~\ref{table-propagators}, Figures~\ref{fig-limit} and~\ref{P-contour}, Corollary~\ref{cor-uniform}).
In 1984 T.Jacobson--L.Schulman repeated this derivation, applied \emph{stationary phase method} among other bright ideas, and found the probability of changing the movement direction \cite{Jacobson-Schulman-84} (cf.~Theorem~\ref{p-right-prob}). The remarkable works of that period contain no mathematical proofs, but approximate computations without estimating the error.


\begin{figure}[htbp]
  \centering
  \includegraphics[width=0.285\textwidth]{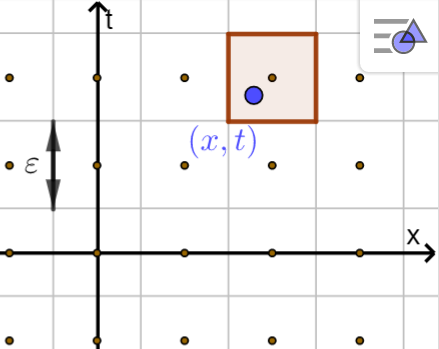}
  \includegraphics[width=0.293\textwidth]{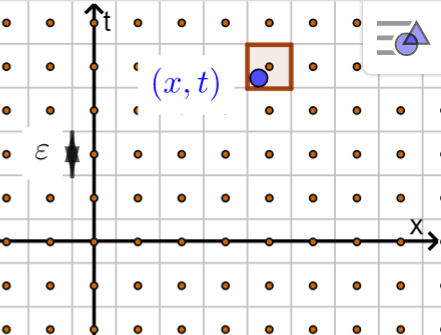}
  \includegraphics[width=0.30\textwidth]{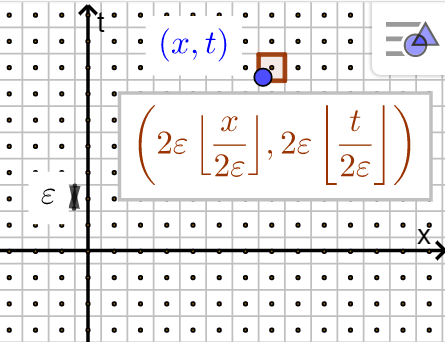}
  \caption{Continuum limit: the point $(x,t)$ stays fixed while the lattice step $\varepsilon$ tends to zero
  }\label{fig-limit}
\end{figure}

\textbf{Ising model.} In 1981 H.~Gersch noticed that Feynman checkers can be viewed as a $1$-dimensional Ising model with \emph{imaginary} temperature or edge weights (see~\S\ref{ssec-interpretation} and \cite{Gersch-81}, \cite[\S3]{Jacobson-Schulman-84}). Imaginary values of these quantities are usual in physics (e.g., in quantum field theory or in alternating current networks). Due to the imaginarity, contributions of most configurations cancel each other,
which makes the model highly nontrivial in spite of being $1$-dimensional. 
In particular, the model exhibits a phase transition (see Figures~\ref{P-contour} and \ref{fig-distribution}).
Surprisingly, the latter seems to have never been reported before. Phase transitions were studied only in more complicated $1$-dimensional Ising models~\cite[\S III]{Jones-66}, \cite{Matveev-Shrock-97},
in spite of a known equivalent result, which we are going to discuss now (see Theorem~\ref{th-limiting-distribution}(B) and~Corollary~\ref{cor-free}).




\begin{figure}[ht]
  \centering
  \includegraphics[width=0.3\textwidth]
  {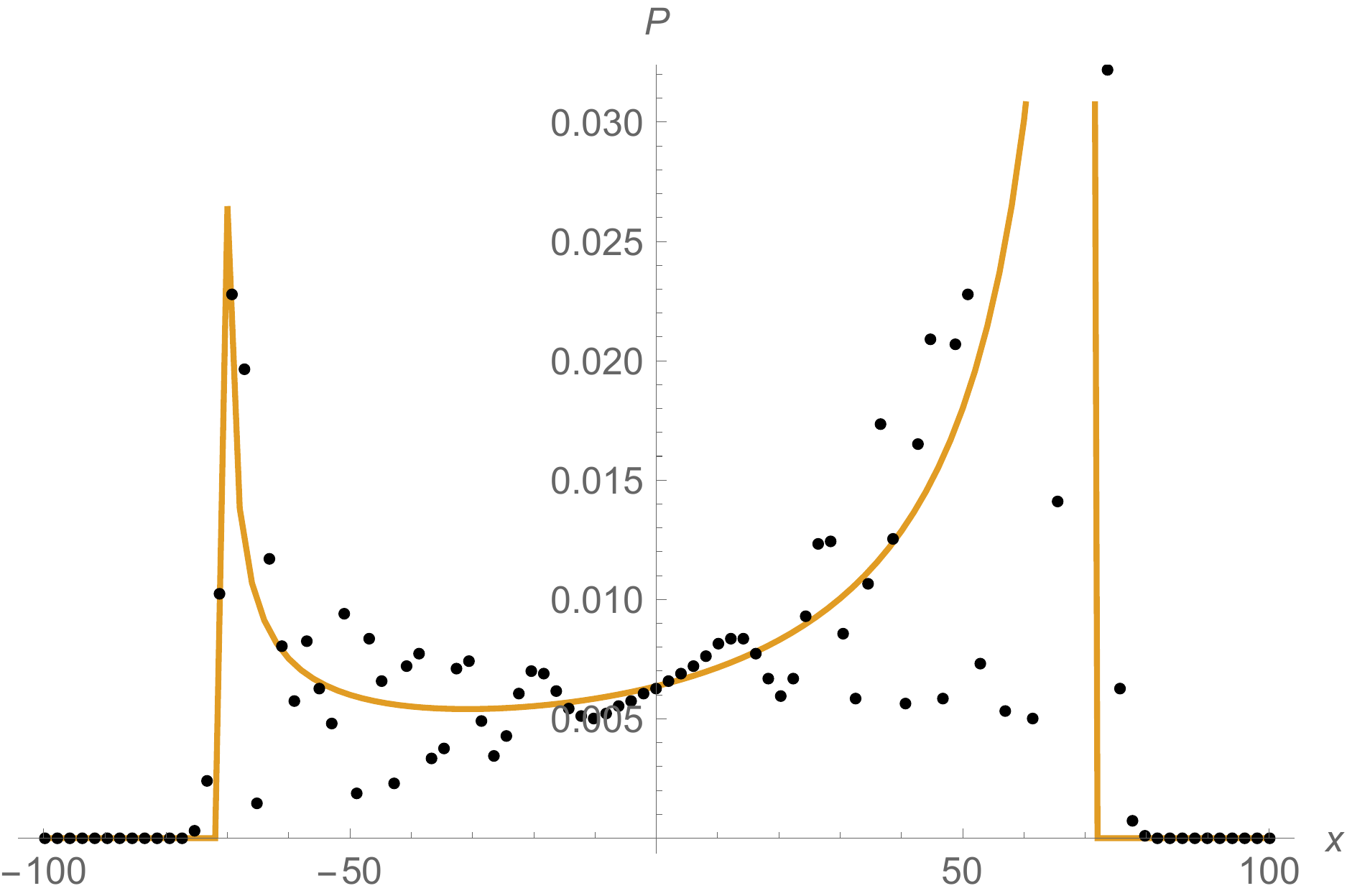}
  \includegraphics[width=0.3\textwidth]
  {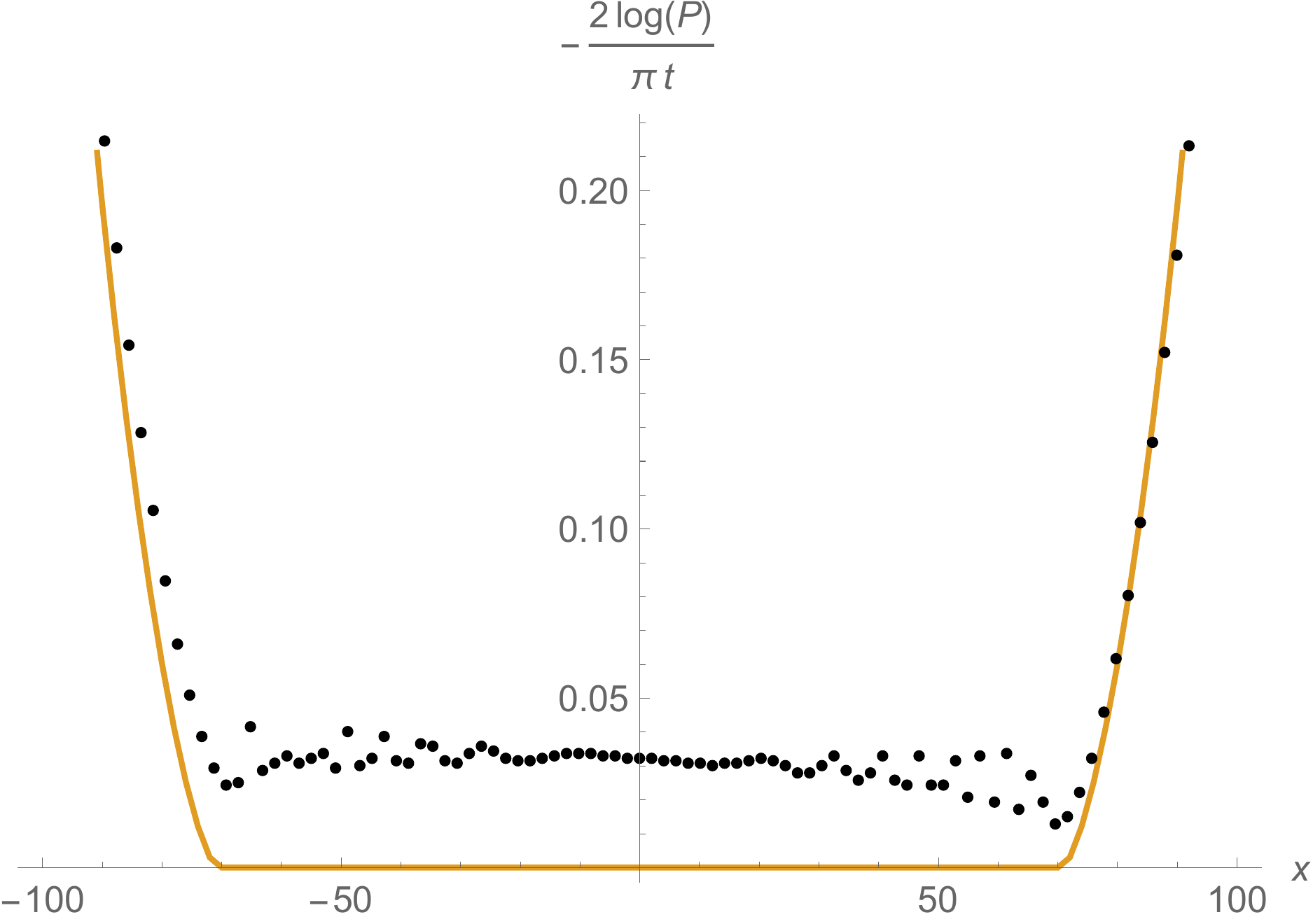}
  \includegraphics[width=0.3\textwidth]
  {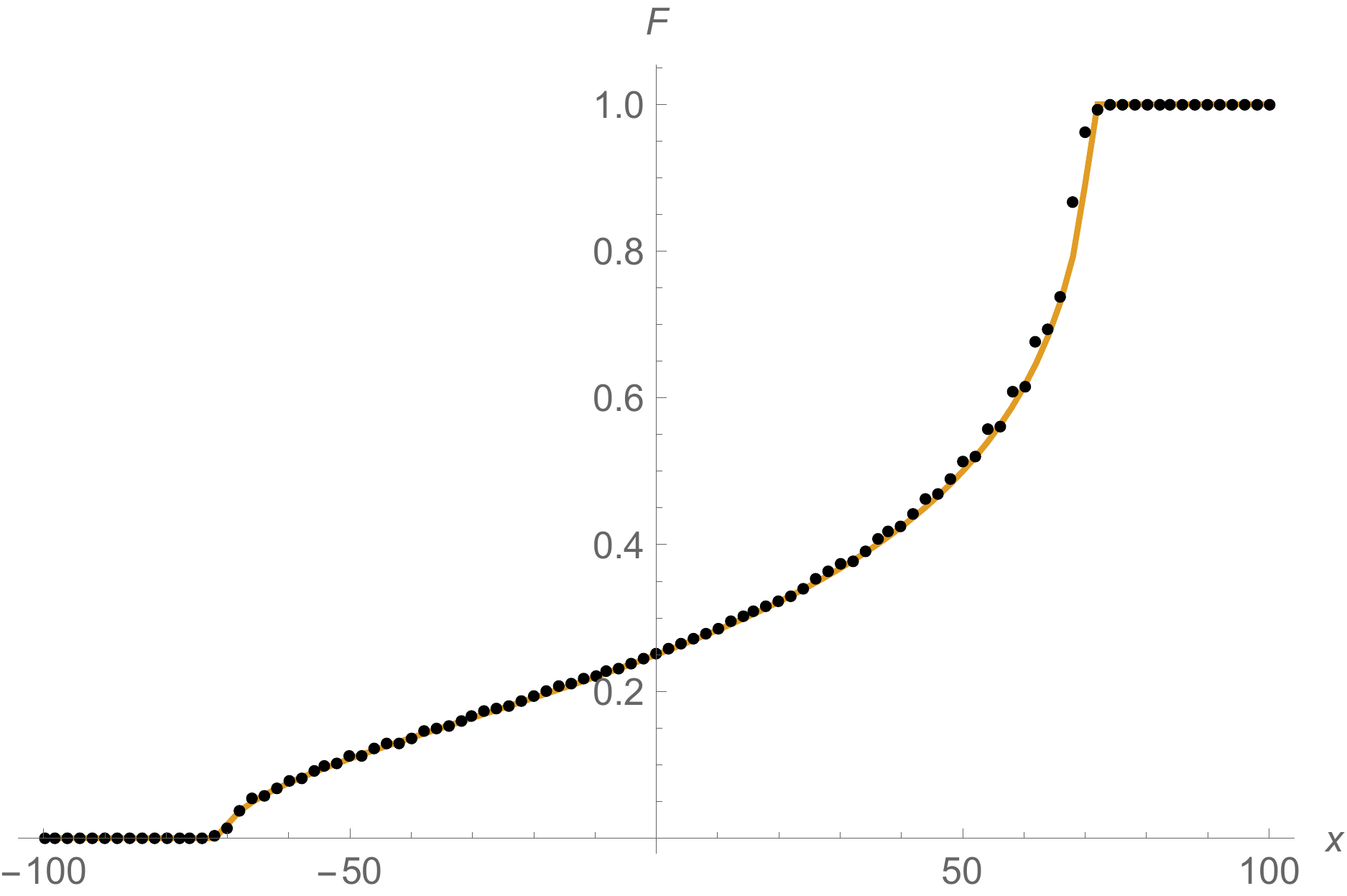}
  \caption{The distribution of the electron position~$x$ at time $t=100$ in natural units for the basic model from~\S\ref{sec-basic} (left, dots). Its normalized logarithm (middle, dots) and cumulative distribution function (right, dots). Their (weak) scaling limits as $t\to\infty$ (curves). 
  The middle plot is also (minus the imaginary part of) the limiting free energy density in the Ising model.
  The non-analyticity of the curves reflects a phase transition.}
  \label{fig-distribution}
\end{figure}

\textbf{Quantum walks.} In 2001 A.~Ambainis, E.~Bach, A.~Nayak, A.~Vishwanath, and J.~Watrous performed a breakthrough \cite{Ambainis-etal-01}. They studied Feynman checkers under names \emph{one-dimensional quantum walk} and \emph{Hadamard walk}; although those cleverly defined models were completely equivalent to Feynman's simple one.
They computed the large-time limit of the model (see Theorem~\ref{th-ergenium}). They discovered several striking properties having sharp contrast with both continuum theory and the classical random walks. First, the most probable average electron velocity in the model equals $1/\sqrt{2}$ of the speed of light, and the probability of exceeding this value is very small (see Figures~\ref{fig-distribution} and~\ref{P-contour} to the left and Theorem~\ref{th-limiting-distribution}(B)).
Second, if an absorbing boundary is put immediately to the left of the starting position, then the probability that the electron is absorbed is $2/\pi$. Third, if an additional absorbing boundary is put at location $x>0$, the probability that the electron is absorbed to the left actually increases, approaching $1/\sqrt{2}$ in the limit $x\to+\infty$. Recall that in the classical case both absorption probabilities are $1$. In addition, they found many combinatorial identities and expressed the above kernel through the values of Jacobi polynomials at a particular point (see Remark~\ref{rem-hypergeo}; cf.~\cite[\S2]{Stanley-04}).

N.~Konno studied a \emph{biased quantum walk} \cite{Konno-05, Konno-08}, which is still essentially equivalent to Feynman checkers (see Remark~\ref{rem-qw}). He found the distribution of the electron position in the (weak) large-time limit (see Figure~\ref{fig-distribution} and Theorem~\ref{th-limiting-distribution}(B)). This result was proved mathematically by G.R.~Grimmett--S.~Janson--P.F.~Scudo \cite{Grimmett-Janson-Scudo-04}. In their seminal paper \cite{Sunada-Tate-12} from 2012, T.~Sunada--T.~Tate found and proved a large-time asymptotic formula for the distribution (see Theorems~\ref{th-ergenium}--\ref{th-outside}). This was a powerful result but it still could not solve the Feynman problem because the error estimate was not uniform in the lattice step. \new{In 2018 M.~Maeda et al.~proved a bound for the maximum of the distribution for large time \cite[Theorem~2.1]{Maeda-etal-18}.}

\new{} Quantum walks were generalized to arbitrary graphs and applied to quantum algorithms (see Figure~\ref{tab-implementation} and~\cite{Georgopoulus-etal-21} for an implementation).
We refer to the surveys by
M.J.~Cantero--F.A.~Gr\"unbaum--L.~Moral--L.~Vel\'azquez,
N.~Konno, J.~Kempe, and S.E.~Venegas-Andraca
\cite{Cantero-etal-12, Konno-08, Kempe-09, Venegas-Andraca-12, Konno-20} for further details in this direction.

\textbf{Lattice quantum field theories.} In a more general context, this is a direction towards creation of \emph{Minkowskian} lattice quantum field theory, with both space and time being discrete~\cite{Bender-etal-94}. In 1970s F.~Wegner and K.~Wilson introduced \emph{lattice gauge theory} as a computational tool for gauge theory describing all known interactions (except gravity); see \cite{Maldacena-16} for a popular-science introduction.
This culminated in determining the proton mass theoretically with error less than $2\%$ in a sense.
This theory is \emph{Euclidean} in the sense that it involves \emph{imaginary} time.
Likewise, an asymptotic formula for the propagator for the (massless) \emph{Euclidean} lattice Dirac equation \cite[Theorem~4.3]{Kenyon-02}
played a crucial role in the continuum limit of the Ising model performed by D.~Chelkak--S.~Smirnov 
\cite{Chelkak-Smirnov-11}. 
Similarly, asymptotic formulae for the \emph{Minkowskian} one (Theorems~\ref{th-ergenium}--\ref{th-main}) can be useful for missing Minkowskian lattice quantum field theory. Several authors argue that Feynman checkers have the advantage of \emph{no fermion doubling} inherent in Euclidean lattice theories and avoid the Nielsen--Ninomiya no-go theorem \cite{Bialynicki-Birula-94, Foster-Jacobson-17}.

Several upgrades of Feynman checkers have been discussed. For instance, around 1990s B.~Gaveau--L.~Schulman and G.~Ord added electromagnetic field to the model  \cite{Schulman-Gaveau-89, Ord}. That time they achieved neither exact charge conservation nor generalization to non-Abelian gauge fields; this is fixed in Definition~\ref{def-external}.
Another example is adding mass matrix by P.~Jizba \cite{Jizba-15}.


\begin{figure}[htb]
\vspace{-0.3cm}
\begin{center}
\includegraphics[width=0.45\textwidth]{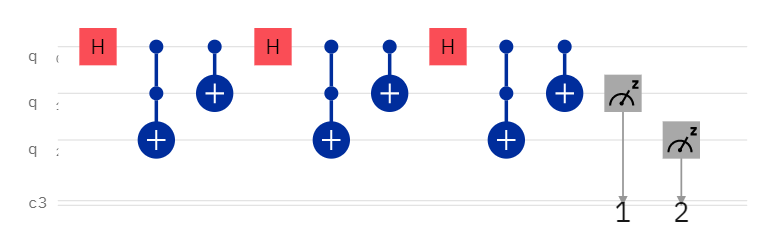}
\includegraphics[width=0.2\textwidth]{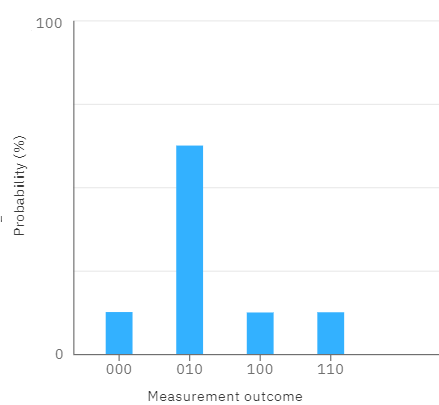}
\includegraphics[width=0.2\textwidth]{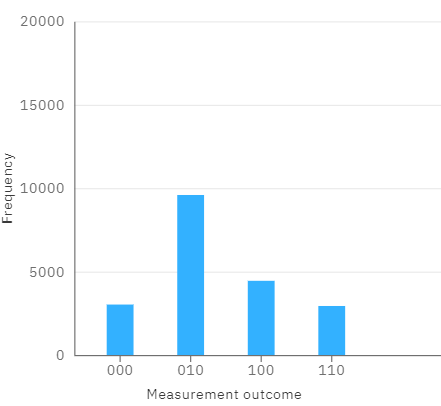}
\end{center}
\vspace{-0.7cm}
\caption{Implementation of the basic model from \S\ref{sec-basic} on a quantum computer using quantum-circuit language (left). The output is a random bit-string coding  electron position $x$ at time $t=4$. The strings 000, 010, 100, 110 code $x=4,2,0,-2$ respectively.
Distribution of $x$ (middle) and a histogram for 
quantum computer IBM-Lima (right) \cite[\S19]{SU-2}. See \cite{Georgopoulus-etal-21} for details.}
\label{tab-implementation}
\end{figure}

It was an old dream to incorporate also checker paths turning backwards in time or forming cycles \cite[p.~481--483]{Schweber-86}, \cite{Jacobson-85}; this would mean creation 
of electron-positron~pairs, celebrating a passage from quantum mechanics to quantum field theory. One looked for a combinatorial model reproducing the \emph{Feynman propagator} rather than the \emph{retarded} one in the continuum limit (see~Table~\ref{table-propagators}). 
Known constructions (such as \emph{hopping expansion}) did not lead to the Feynman propagator because of certain lattice artifacts
(e.g., the title of \cite{Ord-Gualtieri} is misleading: the Feynman propagator is not discussed there). In the \emph{massless} case, a \emph{noncombinatorial} construction of Feynman propagator on the lattice was provided by C.~Bender--L.~Mead--K.~Milton--D.~Sharp in \cite[\S9F]{Bender-etal-94} and~\cite[\S IV]{Bender-etal-85}. 
In Definition~\ref{def-anti-combi} the desired combinatorial construction is finally given.

Another long-standing open problem is to generalize the model to the \emph{$4$-dimensional} real world. In his Nobel prize lecture, R.~Feynman mentioned his own unsuccessful attempts.
There are several recent approaches, e.g., by B.~Foster--T.~Jacobson from 2017 \cite{Foster-Jacobson-17}. Those are not yet as simple and beautiful as the original $2$-dimensional model, as it is written in \cite[\S7.1]{Foster-Jacobson-17} itself.

\textbf{On physical and mathematical works.}
The physical literature on the subject is quite extensive \cite{Venegas-Andraca-12}, and we cannot mention all remarkable works in this brief overview. Surprisingly, in the whole literature we have not found the shouting property of \emph{concentration of measure} for lattice step tending to zero (see Corollary~\ref{cor-concentration}). Many papers are well-written, insomuch that the \emph{physical} theorems and proofs there could be carelessly taken for \emph{mathematical} ones (see the end of \S\ref{sec-proofs}).
We are aware of just a few mathematical works on the subject, such as \cite{Grimmett-Janson-Scudo-04,Sunada-Tate-12,Maeda-etal-18}.


\addcontentsline{toc}{myshrinkalt}{}

\subsection{Contributions} \label{ssec-contributions}

We solve mathematically a problem by R.~Feynman from~1965, which was to prove that his model is consistent with the continuum one, namely, reproduces the usual quantum-mechanical free-particle kernel for large time, small average velocity, and small lattice step (see Corollary~\ref{cor-feynman-problem}). We compute large-time and small-lattice-step
asymptotic formulae for the lattice propagator, uniform in the model parameters (see Theorems~\ref{th-ergenium} and~\ref{th-main}).
For the first time we observe and prove concentration of measure in the continuum limit: the average velocity of an electron emitted by a point source is close to the speed of light with high probability (see Corollary~\ref{cor-concentration}).
The results can be interpreted as asymptotic properties of Young diagrams (see Corollary~\ref{cor-young}) and Jacobi polynomials (see Remark~\ref{rem-hypergeo}).

All these results 
are proved mathematically for the first time. For their statements, just Definition~\ref{def-mass} is sufficient.
In Definitions~\ref{def-external} and~\ref{def-anti-combi} we perform a coupling to lattice gauge theory and the second quantization of the model, 
promoting Feynman checkers to a full-fledged lattice quantum field theory.

\addcontentsline{toc}{myshrinkalt}{}

\subsection{Organization of the paper and further directions}

First we give the definitions and precise statements of the results, and in the process provide a zero-knowledge examples for basic concepts of quantum theory. These are precisely those examples that R.~Feynman presents first in his own books: Feynman checkers (see~\S\ref{sec-mass}) is the first specific example in the whole book~\cite{Feynman-Gibbs}.
The thin-film reflection (see \S\ref{sec-medium}) is the first example in~\cite{Feynman}; see Figures 10--11 there. Thus we hope that these examples could be enlightening to readers unfamiliar with quantum theory.

We start with the simplest (and rough) particular case of the model and upgrade it step by step in each subsequent section. Before each upgrade, we summarize which physical \emph{question} does it address, which simplifying \emph{assumptions} does it resolve or impose additionally, and which experimental or theoretical \emph{results} does it reproduce. Some upgrades (\S\S\ref{sec-medium}--\ref{sec-creation}) are just announced to be discussed in a subsequent publication.
 Our aim is 
\emph{(1+1)-dimensional lattice quantum electrodynamics} (``QED'') but the last step on this way (mentioned in \S\ref{sec-QED}) is still undone.
Open problems are collected in~\S\ref{sec-open}.
For easier navigation, we present the upgrades-dependence chart: 
$$
\xymatrix{
\boxed{\text{\ref{sec-basic}. Basic model}} \ar[r]\ar[d]\ar[rd]
&
\boxed{\text{\ref{sec-mass}. Mass}} \ar[r]\ar[rd]
&
\boxed{\text{\ref{sec-source}. Source}} \ar[r]
&
\boxed{\text{\ref{sec-medium}. Medium}}
\\
\boxed{\text{\ref{sec-spin}. Spin}} \ar[rd]
&
\boxed{\text{\ref{sec-external field}. External field }} \ar[rd] 
&
\boxed{\text{\red{\ref{sec-creation}}. Antiparticles}} \ar[d]
&
\\
&
\boxed{\text{\red{\ref{sec-IdParAnt}.} Identical particles}} \ar[r]
&
\boxed{\text{\red{\ref{sec-QED}}. QED}}
&
}
$$
\comment
$$
\xymatrix{
\boxed{\text{\ref{sec-basic}. Basic model}} \ar[dr]
&
\boxed{\text{\ref{sec-external field}. External field }} \ar[r]
&
\boxed{\text{\red{\ref{sec-interaction}. Interaction 
}}} \ar[dr] 
&
\\
\boxed{\text{\ref{sec-medium}. Medium}}
&
\boxed{\text{\ref{sec-spin}. Spin}} \ar[r]\ar[d]\ar[u]
&
\boxed{\text{\red{\ref{sec-IdParAnt}.} Identical particles}} \ar[d]
\red{\ar[u]}
&
\boxed{\text{\red{\ref{sec-QED}}. QED}}
\\
\boxed{\text{\ref{sec-source}. Source}} \ar[u]
&
\boxed{\text{\ref{sec-mass}. Mass}} \ar[r]\ar[l]
&
\boxed{\text{\red{\ref{sec-creation}}. Antiparticles}} \ar[ur]
}
$$
\endcomment

Hopefully this is a promising path to making quantum field theory 
rigorous and algorithmic. An \emph{algorithmic} quantum field theory would be a one which, given an experimentally observable quantity and a number $\delta>0$, would provide a \emph{precise statement} of an algorithm predicting a value for the quantity within accuracy~$\delta$. (Surely, the predicted value needs not to agree with the experiment for $\delta$ less than accuracy of theory itself.) See Algorithm~\ref{alg-main} for a toy example. This would be an extension of \emph{constructive} quantum field theory (currently far from being algorithmic). Application of quantum theory to computer science is in mainstream now, but the opposite direction could provide benefit as well. (Re)thinking algorithmically is a way to make a subject available to nonspecialists, as it is happening with, say, algebraic topology. \mscomm{Add ref!}

The paper is written in a mathematical level of rigor, in the sense that all the definitions, conventions, and theorems (including corollaries, propositions, lemmas) should be understood literally. Theorems remain true, even if cut out from the text.
The proofs of theorems use the statements but not the proofs of the other ones. Most statements are much less technical than the proofs; hence the proofs are kept in a separate section (\S\ref{sec-proofs}) and long computations are kept in~\cite{SU-2}. In the process of the proofs, we give a zero-knowledge introduction to the main tools to study the model: combinatorial identities, the Fourier transform, the method of moments, the stationary phase method.
Remarks are informal and usually not used elsewhere (hence skippable). 
Text outside definitions, theorems, proofs is neither used formally.


\addcontentsline{toc}{myshrink}{}

\section{Basic model (Hadamard walk)} \label{sec-basic}


{
\hrule
\footnotesize
\noindent\textbf{Question:} what is the probability to find an electron in the square $(x,t)$, if it was emitted from 
$(0,0)$?

\noindent\textbf{Assumptions:} no self-interaction, no creation of electron-positron pairs, fixed mass and lattice step, point source;
the electron moves in a plane ``drifting uniformly along the $t$-axis'' or along a line (and then $t$ is time).

\noindent\textbf{Results:} double-slit experiment (qualitative explanation), charge conservation, large-time limiting distribution.
\hrule
}

\medskip

\addcontentsline{toc}{myshrinkalt}{}

\subsection{Definition and examples}

We first give an informal definition of the model in the spirit of~\cite{Feynman} and then a precise one. 

On an infinite checkerboard, a checker moves to the diagonal-neighboring squares, either upwards-right or upwards-left (see Figure~\ref{Checker-paths} to the left). To each path $s$ of the checker, assign a vector ${a}(s)$ as follows (see Figure~\ref{Checker-paths} to the middle). Take a stopwatch that can time the checker as it moves. Initially the stopwatch hand points upwards. While the checker moves straightly, the hand does not rotate, but each time when the checker changes the direction, the hand rotates through $90^\circ$ clockwise (independently of the direction the checker turns). The final direction of the hand is the direction of the required vector ${a}(s)$. The length of the vector is set to be $1/2^{(t-1)/2}$, where $t$ is the total number of moves (this is just a normalization).  For instance, for the path $s$ in Figure~\ref{Checker-paths} to the top-middle, the vector ${a}(s)=(0,-1/2)$.

\begin{figure}[htbp]
\vspace{-0.3cm}
\begin{center}
\begin{tabular}{l}
\includegraphics[width=3.5cm]{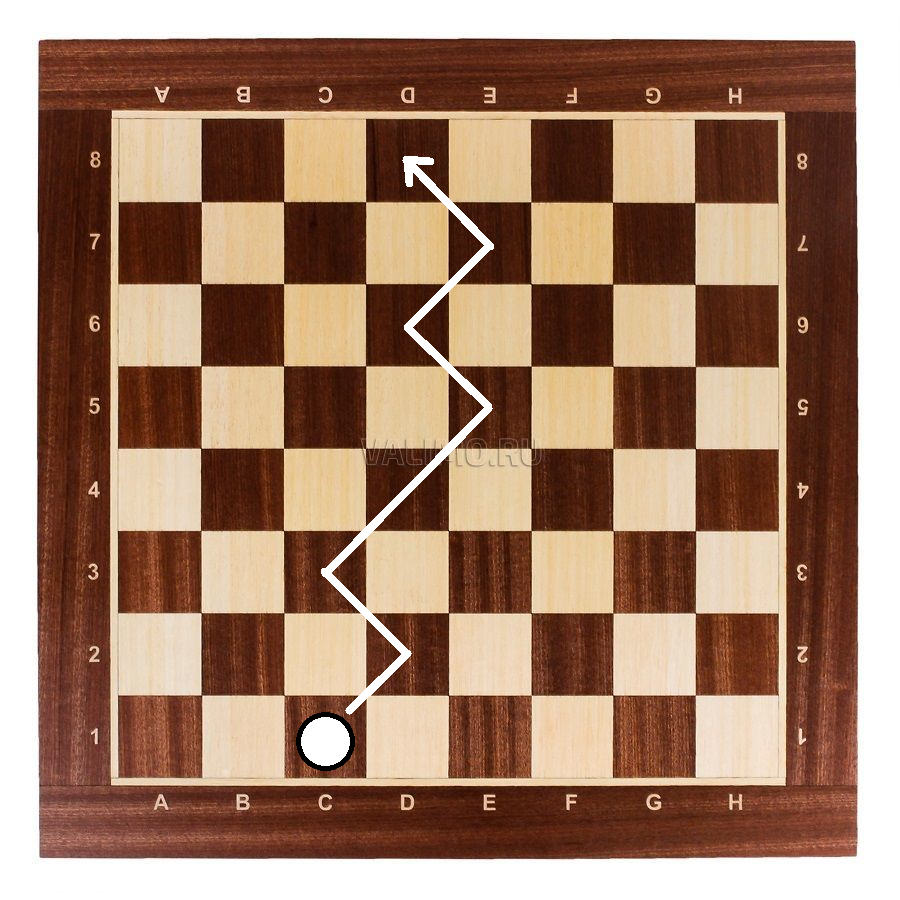}
\end{tabular}
\begin{tabular}{lccc}
\includegraphics[width=1.4cm]{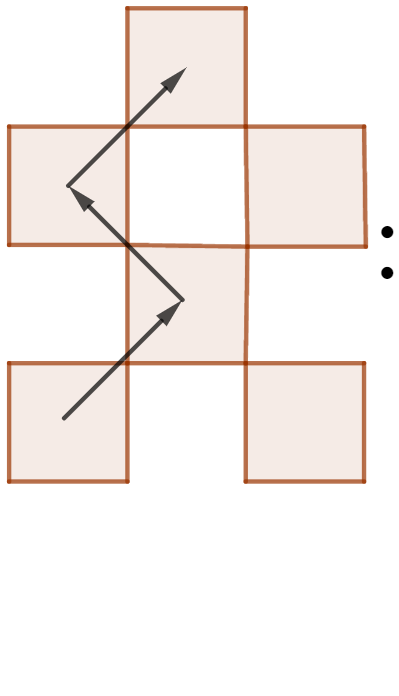} &
\includegraphics[width=1.5cm]{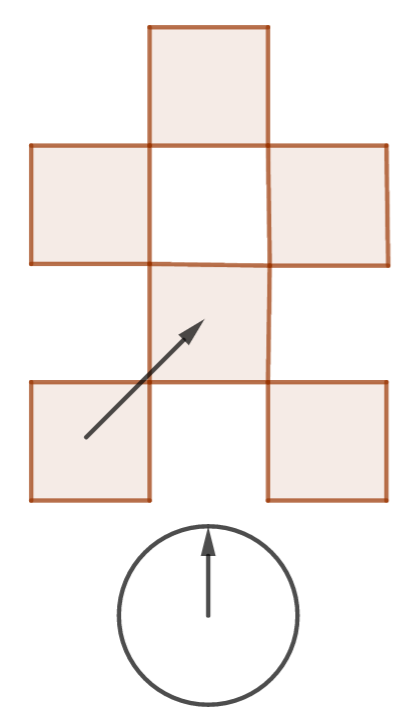} &
\includegraphics[width=1.5cm]{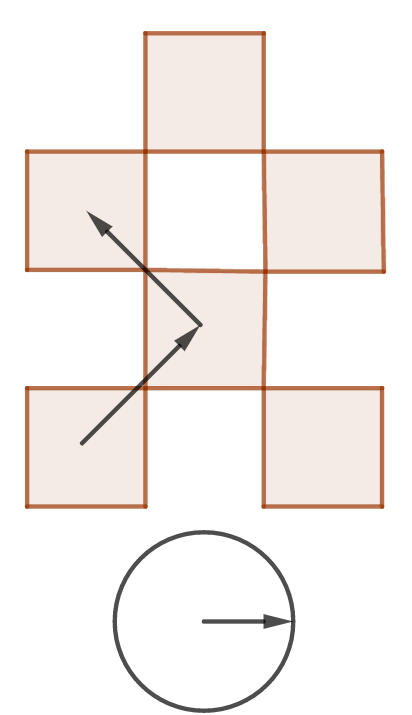} &
\includegraphics[width=1.5cm]{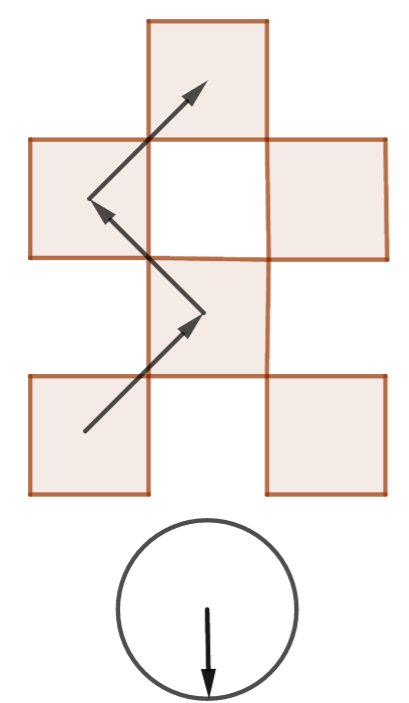} \\
\includegraphics[width=1.5cm]{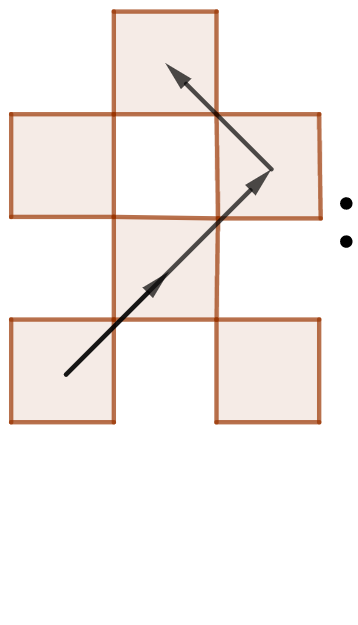} &
\includegraphics[width=1.5cm]{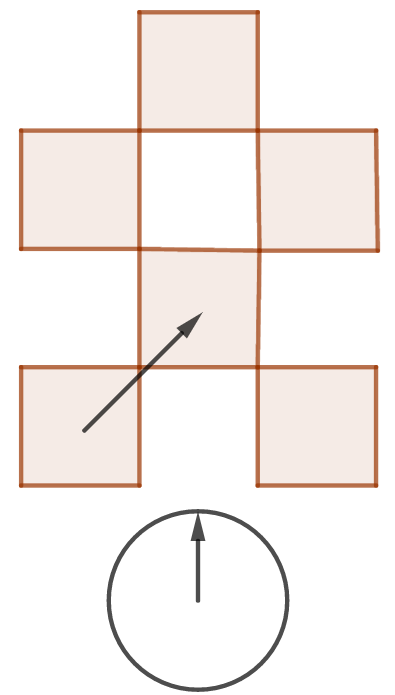} &
\includegraphics[width=1.5cm]{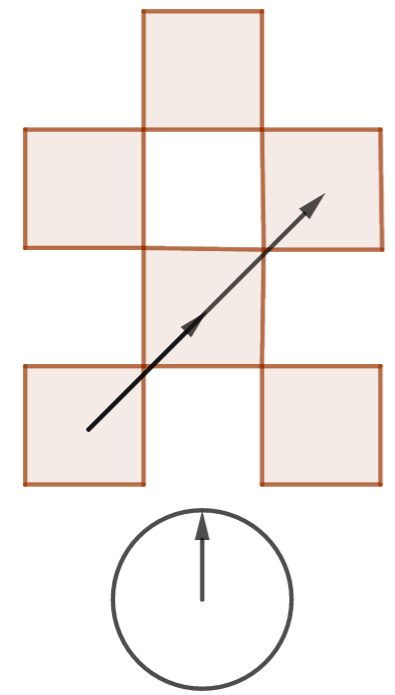} &
\includegraphics[width=1.5cm]{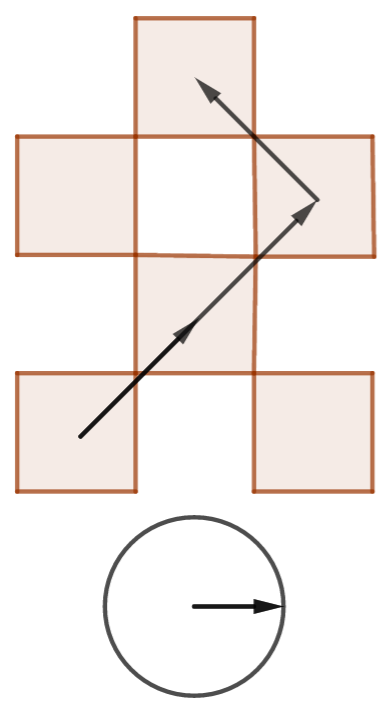}
\end{tabular}
\begin{tabular}{r}
\new{}
\includegraphics[width=3.8cm]{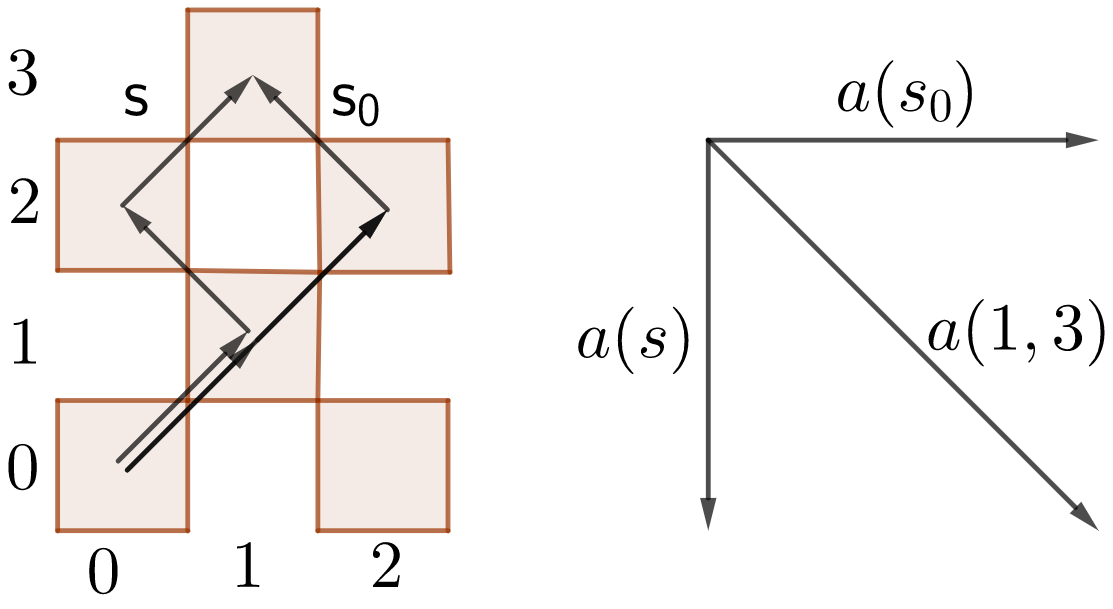}
\end{tabular}
\end{center}
\vspace{-0.8cm}
\caption{ {\new{A} checker path (left). \new{The} vectors assigned to paths (middle) and \new{a} square (right).}}
\label{Checker-paths}
\vspace{-0.3cm}
\end{figure}

Denote by ${a}(x,t):=\sum_s {a}(s)$ the sum over all the checker paths from the square $(0,0)$ to the square $(x,t)$, \emph{starting with the upwards-right move}. 
For instance, ${a}(1,3)=(0,-1/2)+(1/2,0)=(1/2,-1/2)$;
see Figure~\ref{Checker-paths} to the right.
The length square of the vector ${a}(x,t)$ is called the \emph{probability to find an electron in the square $(x,t)$, if it was emitted from $(0,0)$}  {(see \S\ref{ssec-interpretation} for a discussion of the terminology). The vector ${a}(x,t)$ itself is called the \emph{arrow} \cite[Figure~6]{Feynman}.}

Let us summarize this construction rigorously.

\begin{definition} \label{def-basic} A \emph{checker path} is a finite sequence of integer points in the plane such that the vector from each point (except the last one) to the next one equals either $(1,1)$ or $(-1,1)$. A \emph{turn} is a point of the path (not the first and not the last one) such that the vectors from the point to the next and to the previous ones are orthogonal.
 {The \emph{arrow} is the complex number}
$$
{a}(x,t):=2^{(1-t)/2}\,i\,\sum_s (-i)^{\mathrm{turns}(s)},
$$
where the sum over all checker paths $s$ from $(0,0)$ to $(x,t)$ with the first step to $(1,1)$, and $\mathrm{turns}(s)$ is the number of turns in $s$.
Hereafter an empty sum is $0$ by definition.
Denote $$
P(x,t):=|{a}(x,t)|^2,
\qquad  {a_1(x,t):=\mathrm{Re}\,{a}(x,t)},
\qquad  {a_2(x,t):=\mathrm{Im}\,{a}(x,t)}.
$$
Points (or squares) $(x,t)$ with even and odd $x+t$ are called \emph{black} and \emph{white} respectively.
\end{definition}

\begin{figure}[htbp]
\vspace{-0.4cm}
\begin{center}
\includegraphics[width=0.40\textwidth]{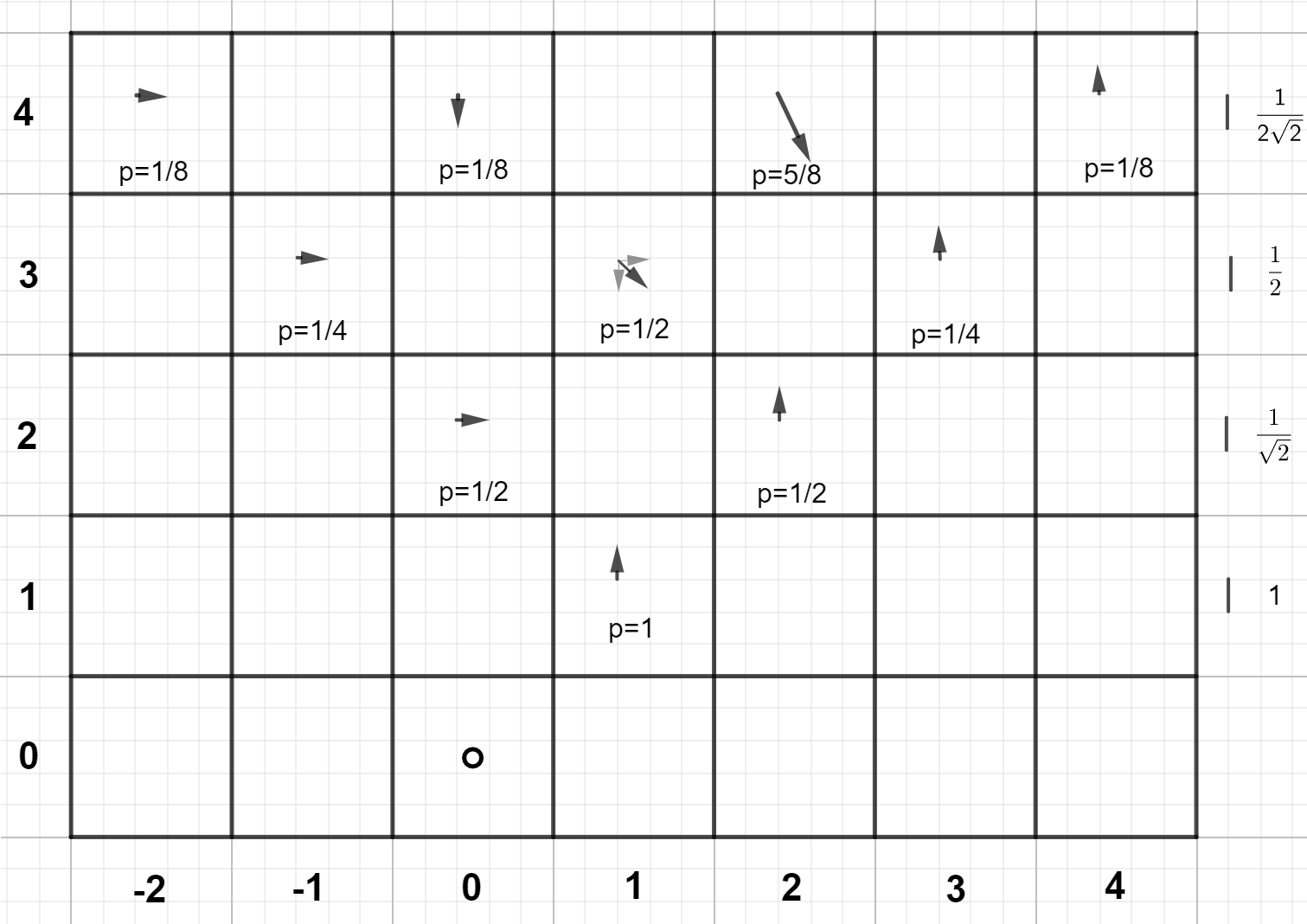}
\includegraphics[width=0.50\textwidth]{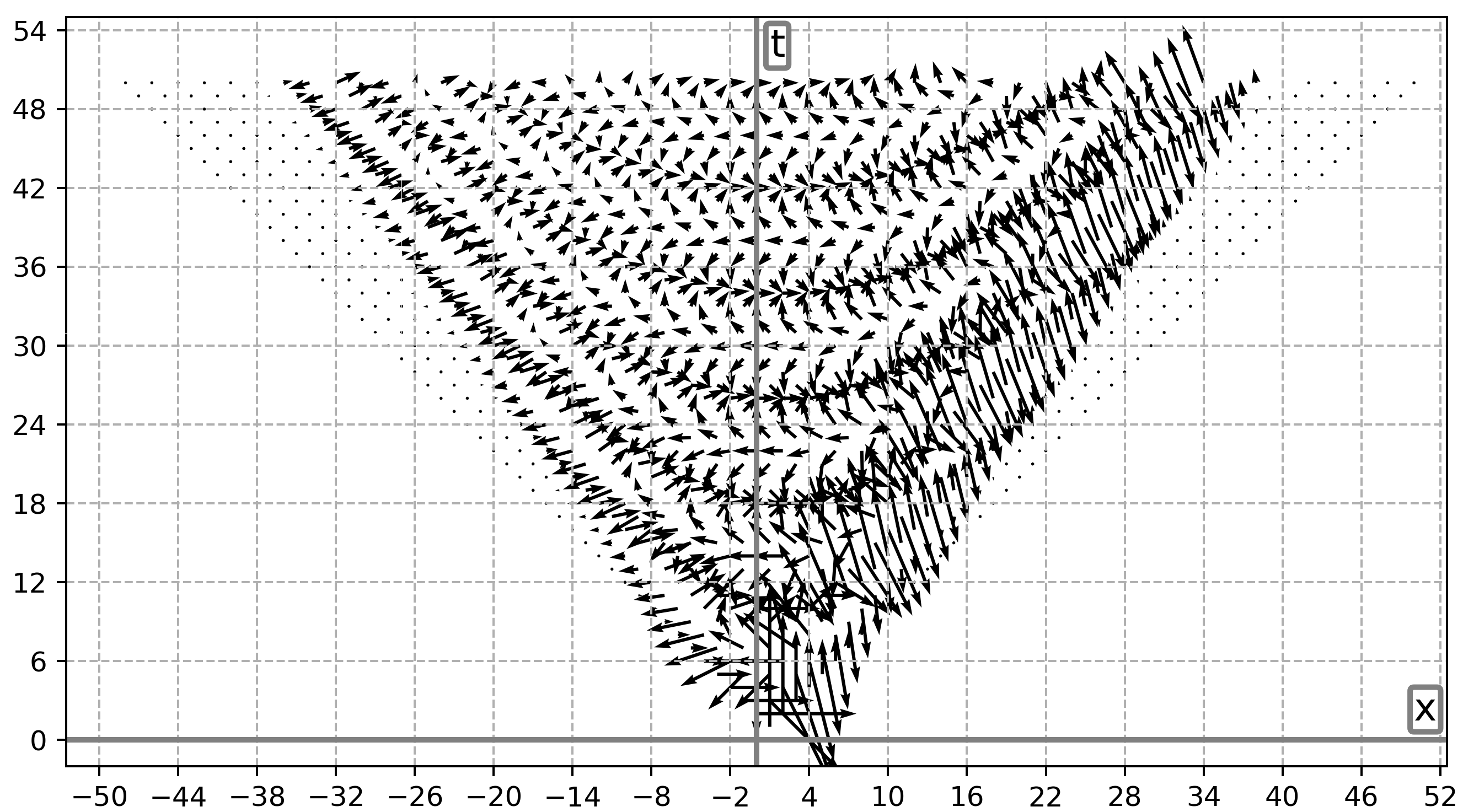}
\end{center}
\vspace{-0.4cm}
\caption{ {
The arrows $a(x,t)$ and the probabilities $P(x,t)$ for small $x,t$ (left)\new{;} the scale depends on the row. The arrows $10\cdot a(x,t)$ for $t\le 50$ (right).}}
\label{fig-arrows}
\vspace{-0.4cm}
\end{figure}


\begin{table}[htbp]
\begin{tabular}{|c*{8}{|>{\centering\arraybackslash}p{55pt}}|}
\hline
	$4$&$\frac{1}{2\sqrt{2}}$ &&$-\frac{i}{2\sqrt{2}}$ &&$\frac{1-2i}{2\sqrt{2}}$ &&$\frac{i}{2\sqrt{2}}$ \\
	\hline
	$3$&&$\frac{1}{2}$&&$\frac{1-i}{2}$&&$\frac{i}{2}$&\\
	\hline
	$2$&&& $\frac{1}{\sqrt{2}}$&&$\frac{i}{\sqrt{2}}$&&\\
	\hline
	$1$&&&&$i$&&&\\
	\hline
	\diagbox[dir=SW,height=18pt]{t}{x}&$-2$&$-1$&$0$&$1$&$2$&$3$&$4$\\
	\hline
\end{tabular}
\caption{ {The arrows $a(x,t)$ for small $x,t$.}}
\label{table-a}
\end{table}

 {Figure~\ref{fig-arrows} and Table~\ref{table-a} depict the arrows $a(x,t)$ and the probabilities $P(x,t)$ for small $x,t$.
Figure}~\ref{a-1000} depicts the graphs of $P(x,1000)$, {$a_1(x,1000)$, and $a_2(x,1000)$} as functions in an even number~$x$.  {We see that} under variation of the final position $x$ at a fixed large time $t$, right after the peak the probability falls to very small although still nonzero {values}. {What is particularly interesting is the unexpected position of the peak, far from $x=1000$.} In Figure~\ref{P-contour} to the left, the color of a point $(x,t)$ with even $x+t$ depicts the value $P(x,t)$. Notice that the sides of the apparent angle are \emph{not} the lines $t=\pm x$ but the lines
$t=\pm \sqrt{2}x$ (see Theorem~\ref{th-limiting-distribution}(A)).

\begin{figure}[htbp]
\begin{tabular}{ccc}
\includegraphics[width=0.31\textwidth]{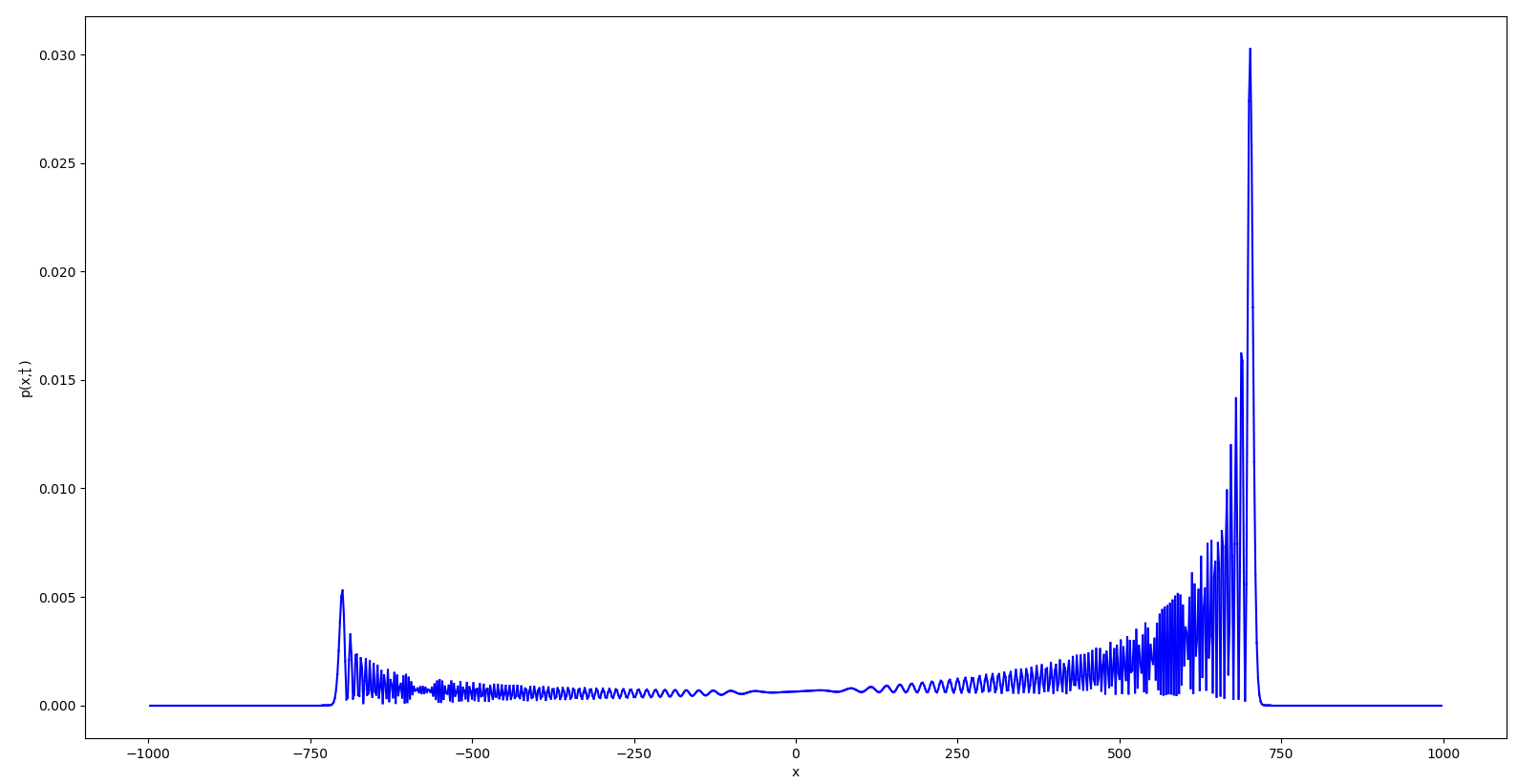}&
\includegraphics[width=0.31\textwidth]{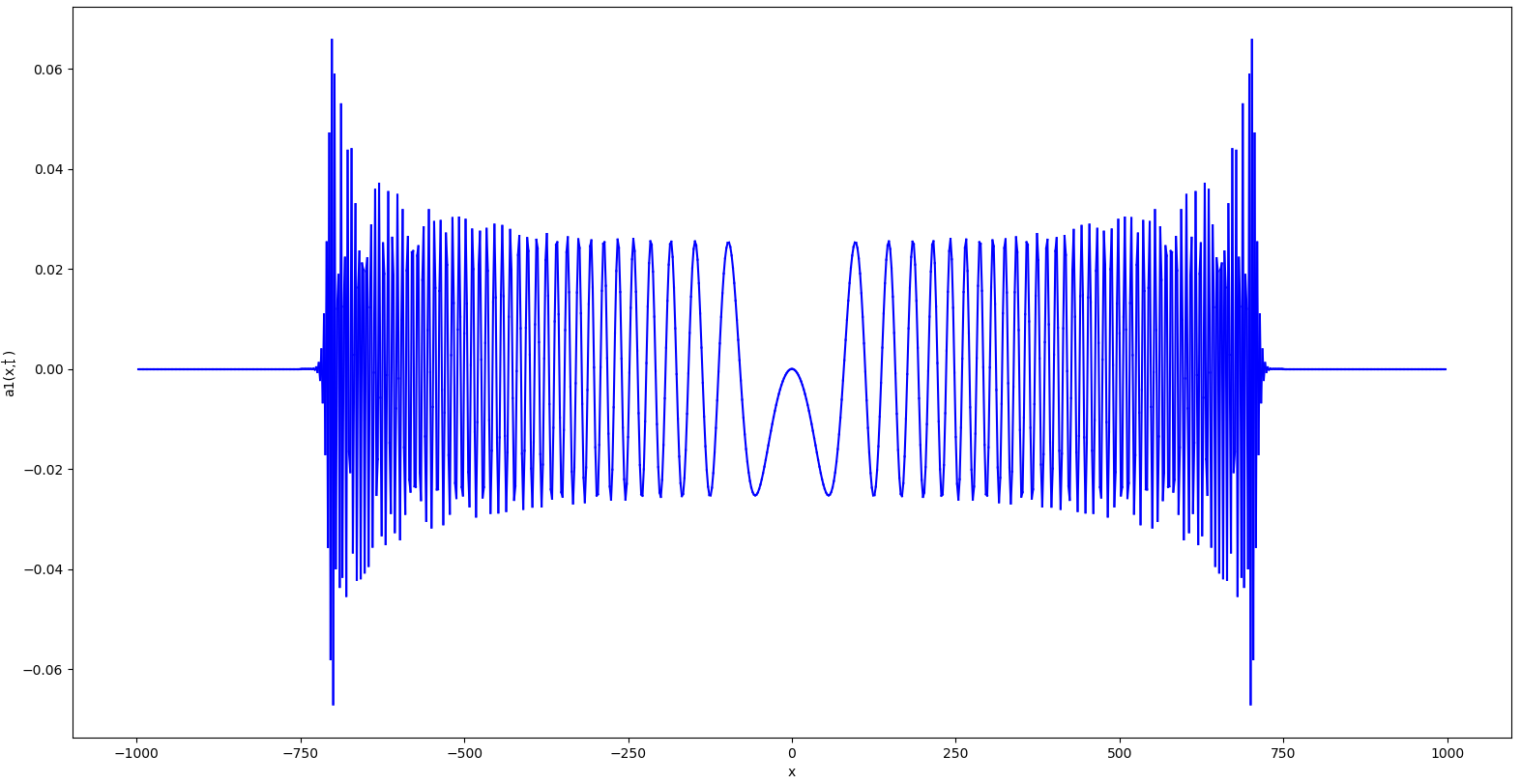} &
\includegraphics[width=0.31\textwidth]{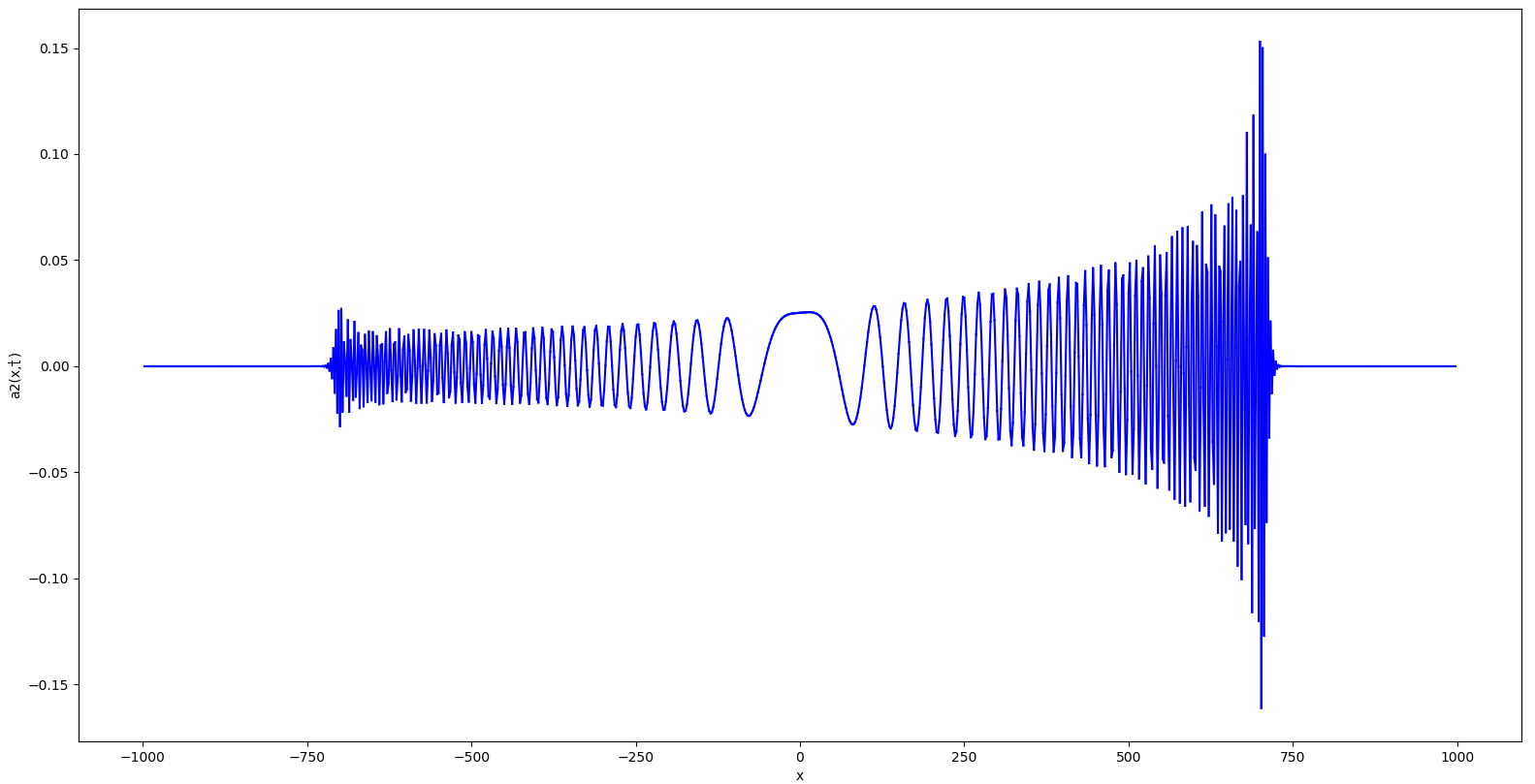}
\end{tabular}
\vspace{-0.2cm}
\caption{
The plots of $P(x,1000),a_1(x,1000),a_2(x,1000)$}
\label{a-1000}
\end{figure}

\comment
\mscomm{Which of the numerous tables and figures should remain, and which should be removed?}

In the following table, the number in the cell $(x,t)$ is $P (x,t)$, and an empty cell $(x,t)$ means that $P(x,t)=0$.
\smallskip

\begin{tabular}{|c*{8}{|>{\centering\arraybackslash}p{55pt}}|}
	\hline
	$4$&$1/8$ &&$1/8$ &&$5/8$ &&$1/8$ \\
	\hline
	$3$&&$1/4$&&$ 1/2$&&$1/4$&\\
	\hline
	$2$&&& $1/2$&&$1/2$&&\\
	\hline
	$1$&&&&$1$&&&\\
	\hline
	\diagbox[dir=SW,height=18pt]{t}{x}&$-2$&$-1$&$0$&$1$&$2$&$3$&$4$\\
	\hline
\end{tabular}
\vskip8pt

 {Rotating the checkerboard (see Figure~\ref{Checker-paths} to the right) through $45^\circ$ clockwise, writing $a(x,t)$ in each black square, and dropping the normalization factor, we get a compact way to show the vectors (Table~\ref{table}).}


{\small
$$
\begin{array}{cccccccc}
 (0,1) & (0,1) & (0,1) & (0,1) & (0,1) & (0,1) & (0,1) & (0,1) \\
 (1,0) & (1,-1) & (1,-2) & (1,-3) & (1,-4) & (1,-5) & (1,-6) & (1,-7) \\
 (1,0) & (0,-1) & (-1,-1) & (-2,0) & (-3,2) & (-4,5) & (-5,9) & (-6,14) \\
 (1,0) & (-1,-1) & (-2,0) & (-2,2) & (-1,4) & (1,5) & (4,4) & (8,0) \\
 (1,0) & (-2,-1) & (-2,1) & (0,3) & (3,3) & (6,0) & (8,-6) & (8,-14) \\
 (1,0) & (-3,-1) & (-1,2) & (3,3) & (6,0) & (6,-6) & (2,-12) & (-6,-14) \\
 (1,0) & (-4,-1) & (1,3) & (6,2) & (6,-4) & (0,-10) & (-10,-10) & (-20,0) \\
 (1,0) & (-5,-1) & (4,4) & (8,0) & (2,-8) & (-10,-10) & (-20,0) & (-20,20) \\
\end{array}
$$
}

\begin{table}[htbp]
{\small
\begin{tabular}{|c|c*{8}{>{\centering\arraybackslash}p{45pt}}c|}
    \hline
 $ 8$ & $ 1 $ & $ -6-i $ & $ 8+5 i $ & $ 8-3 i $ & $ -6-11 i $ & $ -20-5 i $ & $ -20+15 i $ & $ 35 i $ & $ 35+35 i $ \\
 $ 7$ & $ 1 $ & $ -5-i $ & $ 4+4 i $ & $ 8 $ & $ 2-8 i $ & $ -10-10 i $ & $ -20 $ & $ -20+20 i $ & $ -5+40 i $ \\
 $ 6$ & $ 1 $ & $ -4-i $ & $ 1+3 i $ & $ 6+2 i $ & $ 6-4 i $ & $ -10 i $ & $ -10-10 i $ & $ -20 $ & $ -25+20 i $ \\
 $ 5$ & $ 1 $ & $ -3-i $ & $ -1+2 i $ & $ 3+3 i $ & $ 6 $ & $ 6-6 i $ & $ 2-12 i $ & $ -6-14 i $ & $ -17-8 i $ \\
 $ 4$ & $ 1 $ & $ -2-i $ & $ -2+i $ & $ 3 i $ & $ 3+3 i $ & $ 6 $ & $ 8-6 i $ & $ 8-14 i $ & $ 5-22 i $ \\
 $ 3$ & $ 1 $ & $ -1-i $ & $ -2 $ & $ -2+2 i $ & $ -1+4 i $ & $ 1+5 i $ & $ 4+4 i $ & $ 8 $ & $ 13-8 i $ \\
 $ 2$ & $ 1 $ & $ -i $ & $ -1-i $ & $ -2 $ & $ -3+2 i $ & $ -4+5 i $ & $ -5+9 i $ & $ -6+14 i $ & $ -7+20 i $ \\
 $ 1$ & $ 1 $ & $ 1-i $ & $ 1-2 i $ & $ 1-3 i $ & $ 1-4 i $ & $ 1-5 i $ & $ 1-6 i $ & $ 1-7 i $ & $ 1-8 i $ \\
 $ 0$ & $ i $ & $ i $ & $ i $ & $ i $ & $ i $ & $ i $ & $ i $ & $ i $ & $ i $ \\
 	\hline	\diagbox[dir=SW,height=18pt]{$n$}{$k$}&
 $1$&$2$&$3$&$4$&$5$&$6$&$7$&$8$&$9$\\
	\hline
\end{tabular}
}
\caption{The values $2^{(n+k-1)/2}{a}(k-n,n+k)$ for $0\le n\le 8$, $1\le k\le 9$.}
\label{table}
\end{table}

\endcomment

\addcontentsline{toc}{myshrinkalt}{}

\subsection{Physical interpretation}\label{ssec-interpretation}

Let us comment on the physical interpretation of the model and ensure that it captures unbelievable behavior of electrons. {There are two different interpretations; see Table~\ref{table-interpretations}.}

\begin{table}[htbp]
 {
\begin{center}
\begin{tabular}{|l|p{0.37\textwidth}|p{0.48\textwidth}|}
  \hline
  object & standard interpretation & spin-chain interpretation\\
  \hline
  $s$ & path & configuration of ``$+$'' and ``$-$'' in a row \\
  $\mathrm{turns}(s)$ & number of turns & half of the configuration energy \\
  $t$ & time & volume \\
  $x$ & position & difference between the number of ``$+$'' and ``$-$'' \\
  $x/t$ & average velocity &  magnetization \\
  $a(x,t)$ & probability amplitude & partition function up to constant \\
  $P(x,t)$ & probability & partition function norm squared\\
  $\dfrac{4i}{\pi t}\log a(x,t)$ &
  normalized logarithm of amplitude &
  free energy density\\
  $\dfrac{i\,a_2(x,t)}{a(x,t)}$ &
  \vspace{-0.6cm}
  conditional probability amplitude
  of the last move upwards-right &
  \vspace{-0.6cm}
  ``probability'' of equal signs
  at the ends of the spin chain
  \\  \hline
\end{tabular}
\end{center}
}
  \caption{ {Physical interpretations of Feynman checkers}} \label{table-interpretations}
\end{table}

 {\bf Standard interpretation.}
Here the $x$- and $t$-coordinates are interpreted as   {the electron} position and time respectively.   {Sometimes} (e.g., in Example~\ref{p-double-slit}) we a bit informally interpret both as position,  {and assume that the electron performs a ``uniform classical motion'' along the $t$-axis. We work in the natural system of units, where the speed of light, the Plank {and the Boltzmann} constants equal $1$. Thus the lines $x=\pm t$ represent motion with the speed of light. Any checker path lies above both lines,  {i.e. in the light cone}, which means agreement with relativity: the speed of electron cannot exceed the speed of light.}

To think of $P(x,t)$ as a probability, consider  {the $t$-coordinate} as fixed, and the squares $(-t,t), (-t+2,t), \dots, (t,t)$ as all the possible outcomes of an experiment. For instance, the $t$-th horizontal might be a  {screen} detecting the electron.  {We shall see that all the numbers $P(x,t)$ on one horizontal sum up to $1$ (Proposition~\ref{p-probability-conservation}), thus indeed can be considered as probabilities. Notice that the probability to find the electron in a set $X\subset\mathbb{Z}$ is
$P(X,t):=\sum_{x\in X}P(x,t)=\sum_{x\in X}|a(x,t)|^2$ rather than $\left|\sum_{x\in X}a(x,t)\right|^2$ (cf.~\cite[Figure~50]{Feynman}).
}

 {In reality, one cannot measure the electron position exactly. A fundamental limitation
is the electron \emph{reduced Compton wavelength} $\lambda=1/m\approx 4\cdot 10^{-13}$ meters, where $m$ is the electron mass. Physically, the basic model approximates the 
continuum by a lattice of step exactly $\lambda$. But that is still a rough approximation:}
one needs even smaller step to prevent accumulation of approximation error at larger distances and times.
For instance, Figures~\ref{P-contour} and~\ref{fig-correlation} to the left  {show a finite-lattice-step effect} (\emph{renormalization of speed of light}): the average velocity $x/t$ cannot exceed $1/\sqrt{2}$ of the speed of light with high probability. (An explanation in physical terms: lattice regularization cuts off distances smaller than the lattice step,  {hence small wavelengths}, hence large momenta, {and} hence large velocities.) 
A more  {precise} 
model is given in~\S\ref{sec-mass}: compare the plots in Figure~\ref{P-contour}.


\begin{figure}[htbp]
  \centering
  \includegraphics[width=0.24\textwidth]
  {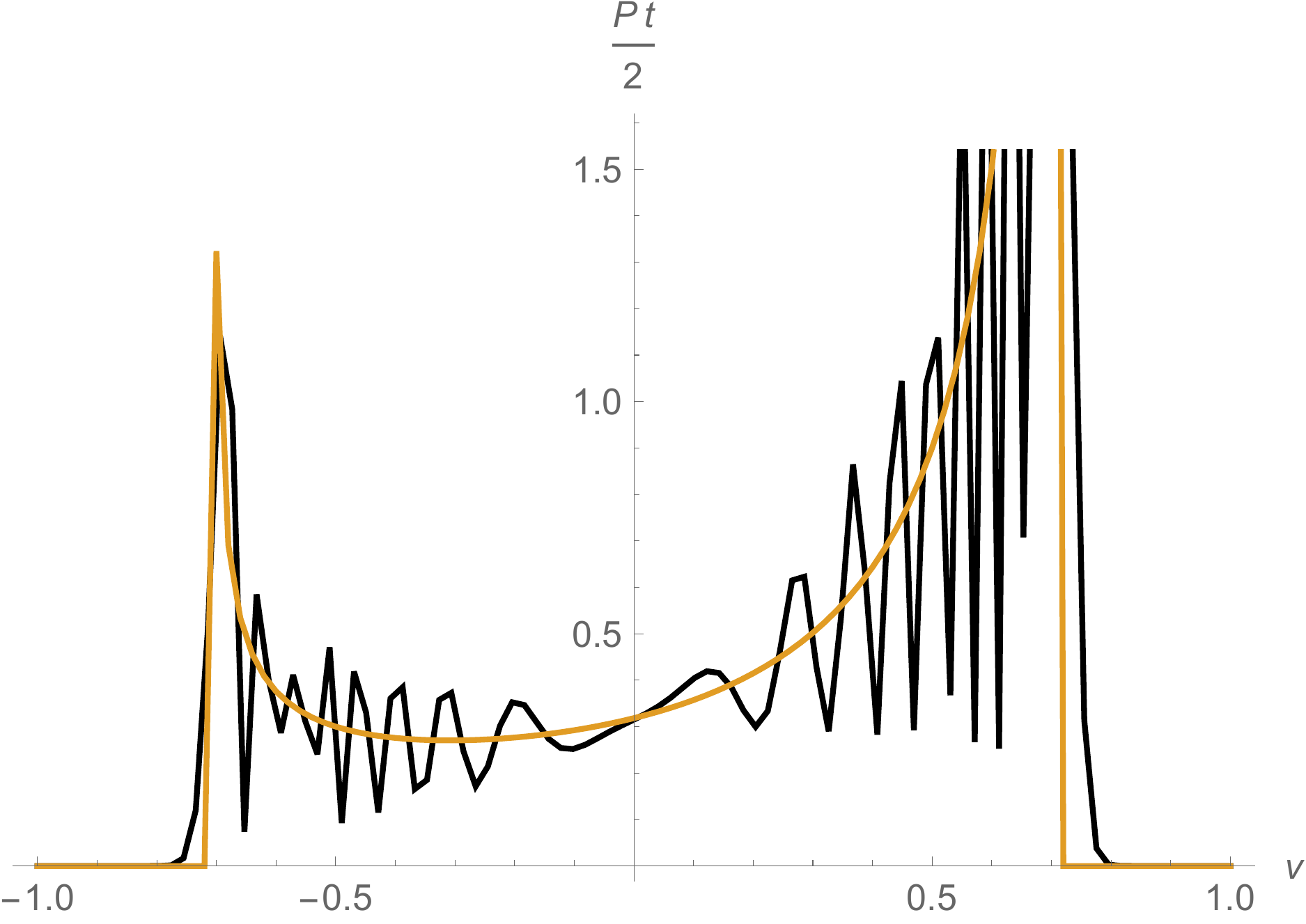}
  \includegraphics[width=0.24\textwidth]
  {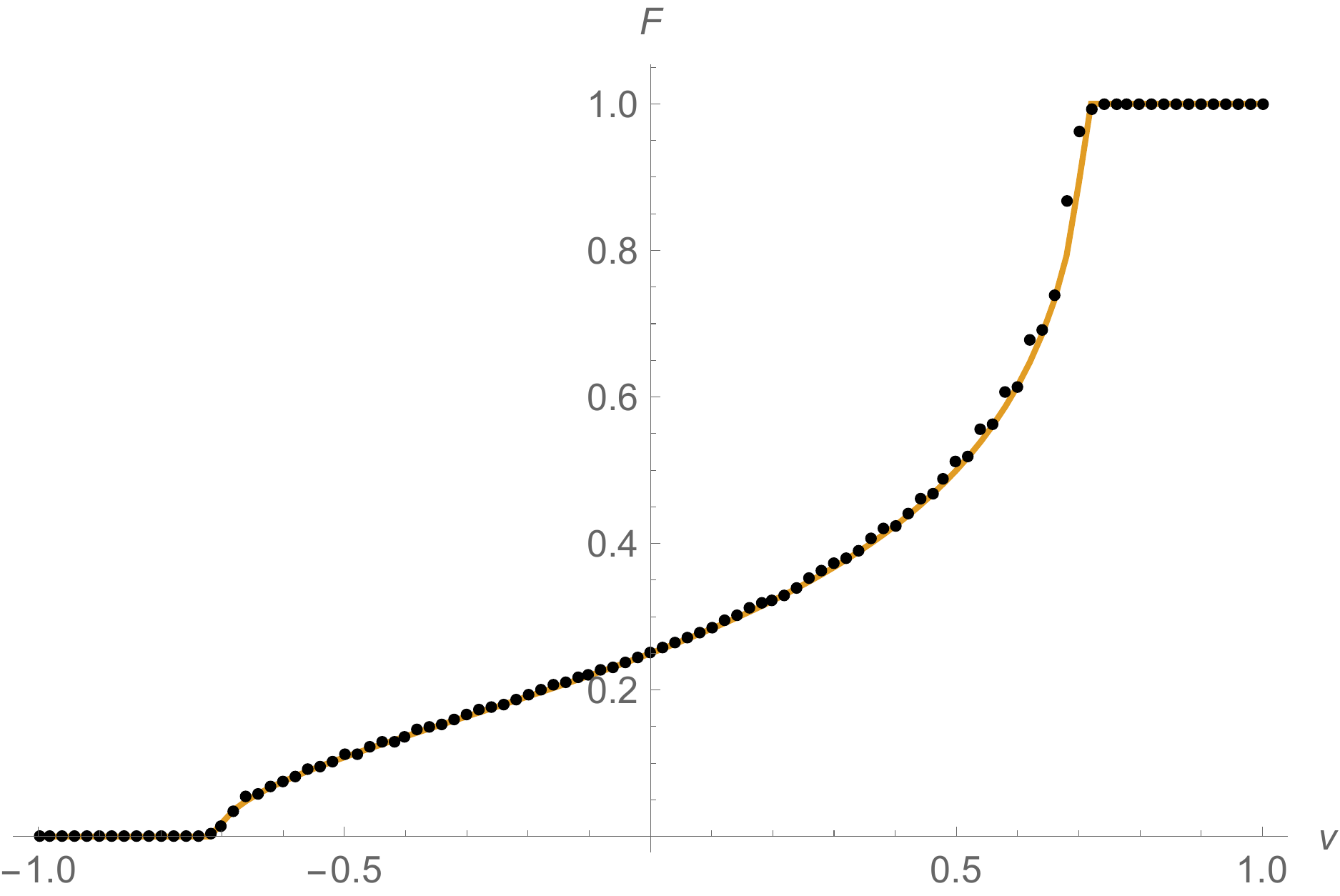}
  \includegraphics[width=0.24\textwidth]
  {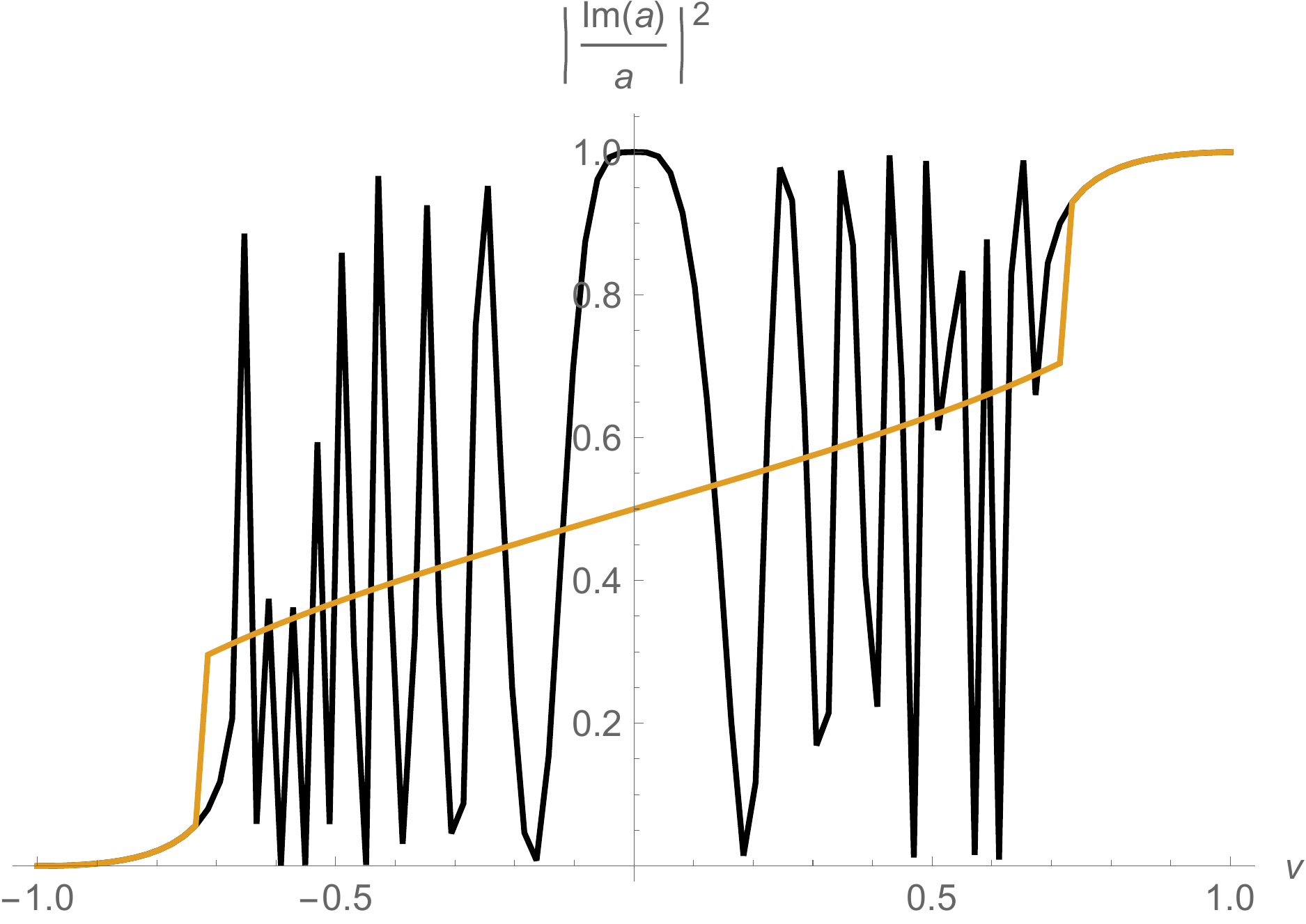}
  \includegraphics[width=0.24\textwidth]
  {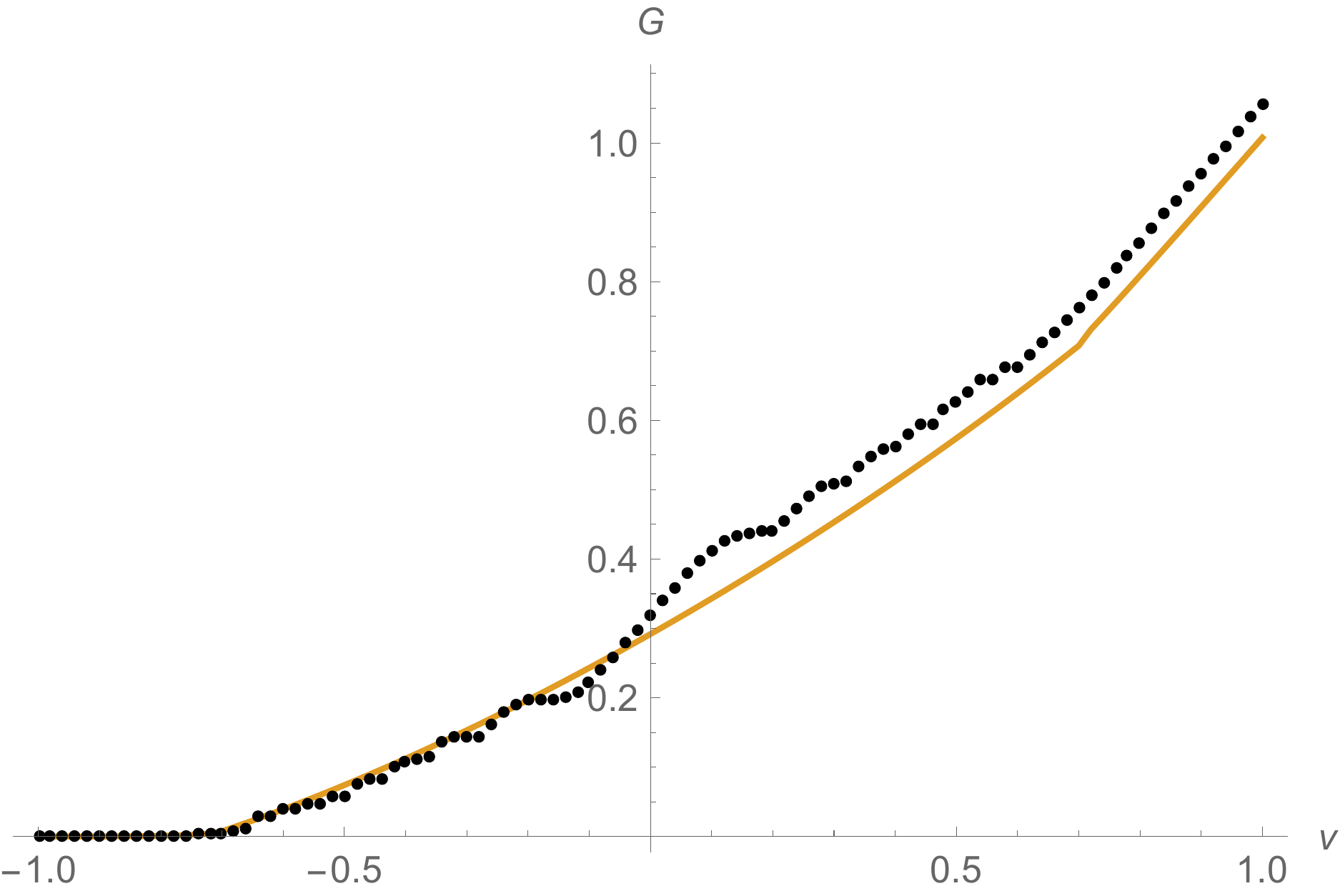}\\
  \includegraphics[width=0.24\textwidth]
  {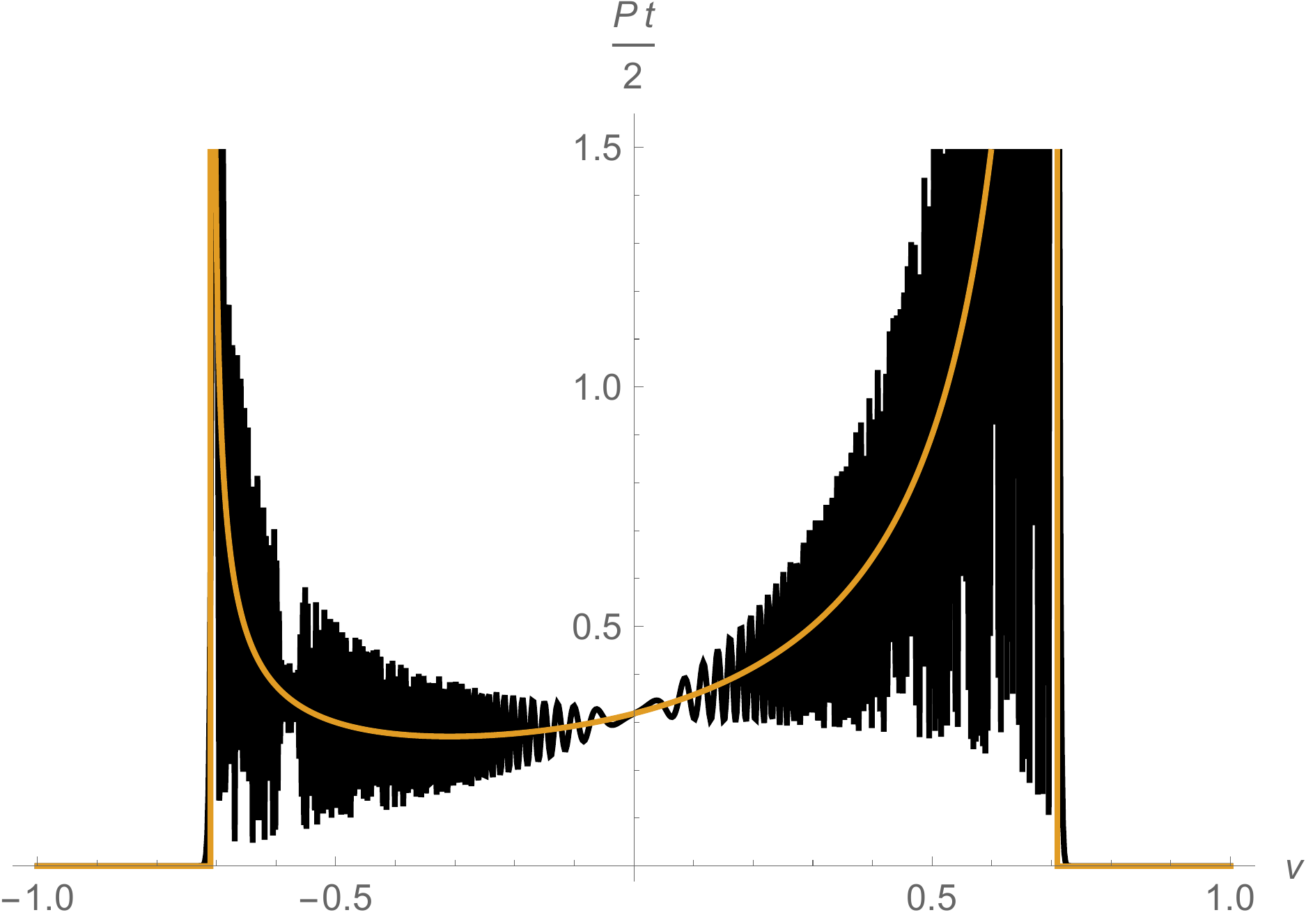}
  \includegraphics[width=0.24\textwidth]
  {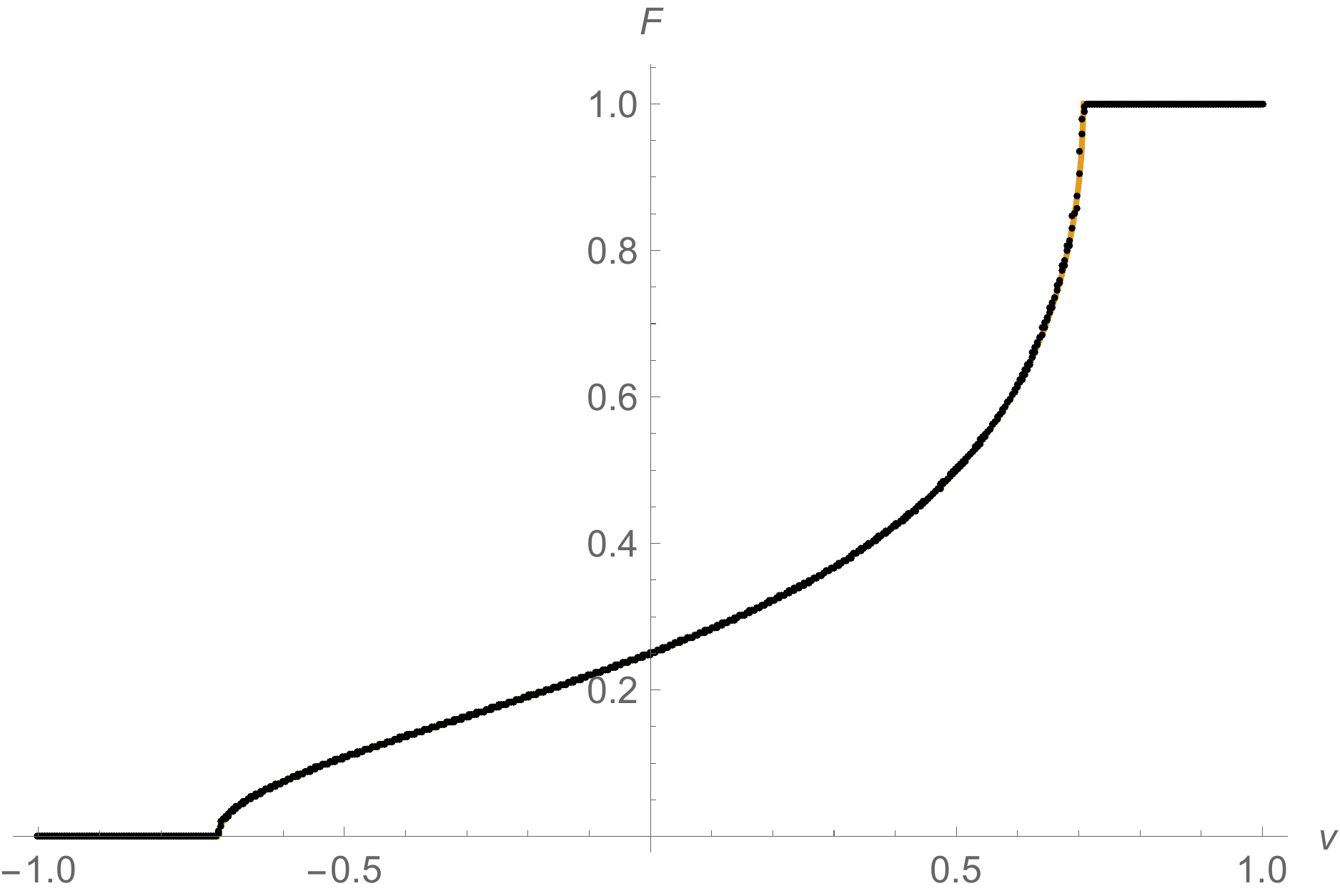}
  \includegraphics[width=0.24\textwidth]
  {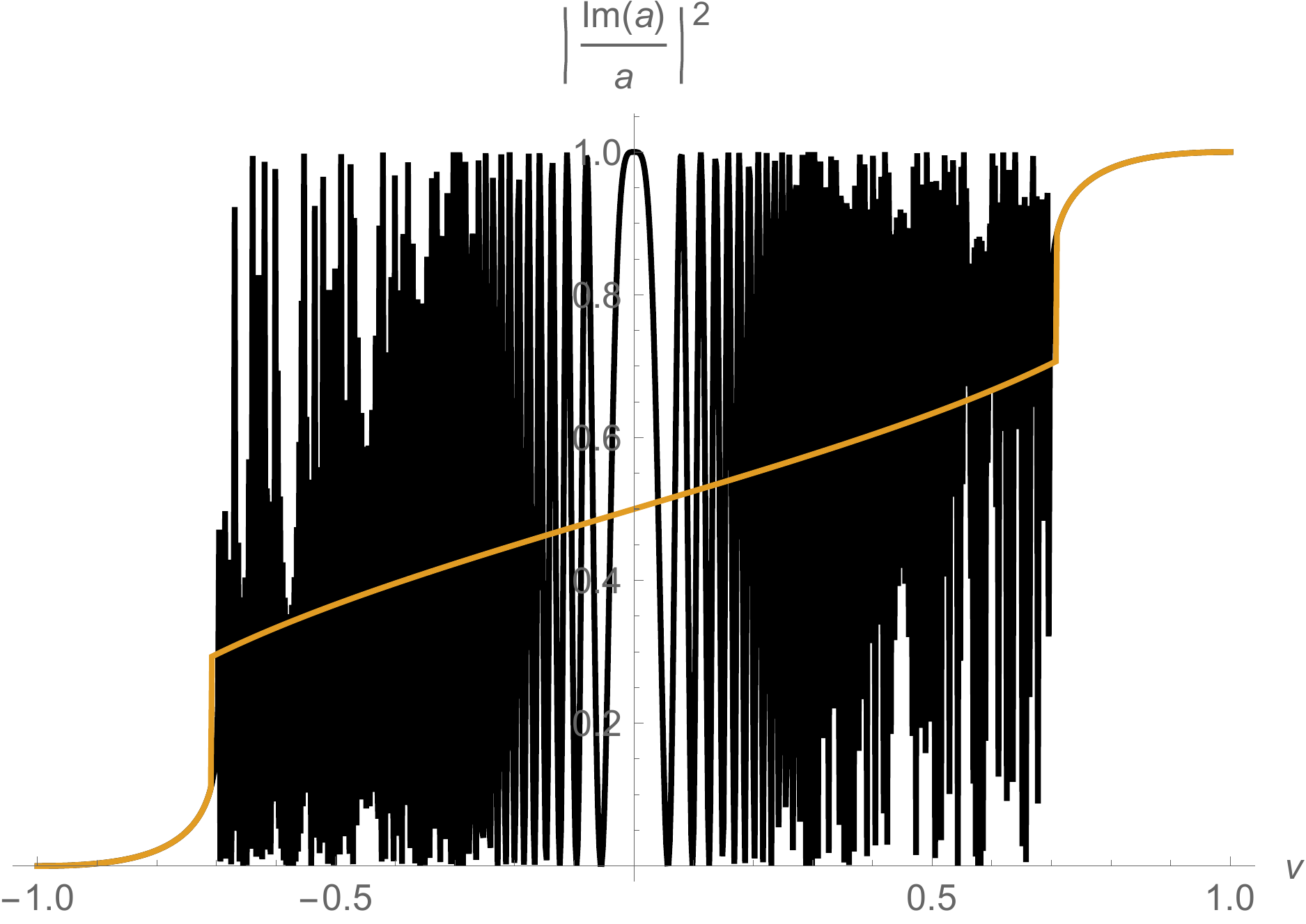}
  \includegraphics[width=0.24\textwidth]
  {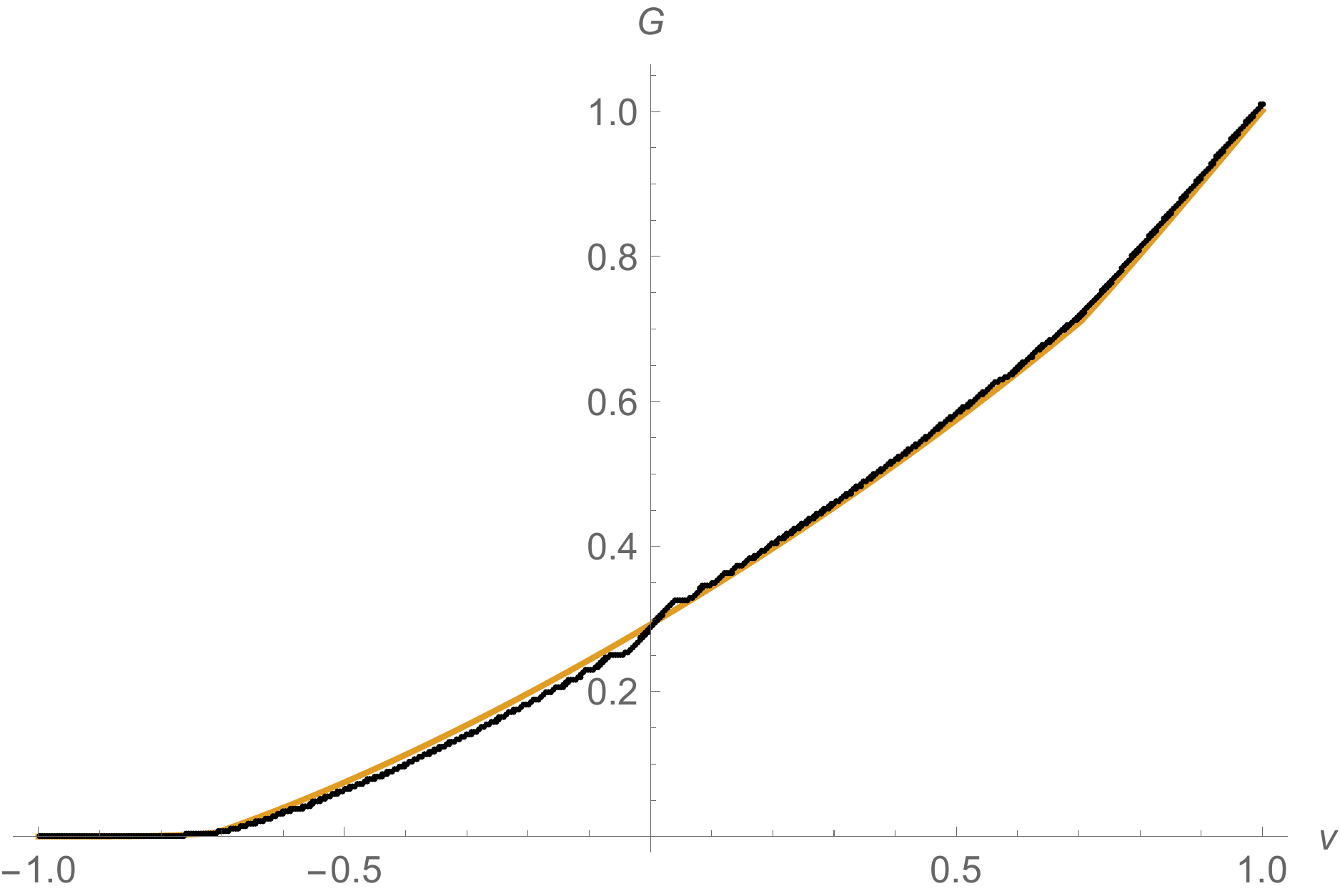}
  \caption{ {The plots of $\frac{t}{2}P(2\lceil \frac{vt}{2}\rceil,t)$, $F_t(v):=\sum\limits_{x\le vt}P(x,t)$, $\left|\frac{a_2(2\lceil {vt}/{2}\rceil,t)}{a(2\lceil {vt}/{2}\rceil,t)}\right|^2$,
  $G_t(v):=\sum\limits_{x\le vt} \frac{2}{t}\left|\frac{a_2(x,t)}{a(x,t)}\right|^2$ (dark) 
  for $t=100$ (top), $t=1000$ (bottom), and their distributional limits as $t\to\infty$ (light). 
  \mscomm{!!! Colors !!!}}}
  \label{fig-correlation}
\end{figure}

As we shall see now, the model qualitatively captures unbelievable behavior of electrons.  {(For correct quantitative results like exact shape of an interferogram, an upgrade involving a \emph{coherent} source is required; see \S\ref{sec-source}.)}


The \emph{probability to find an electron in the square
$(x,t)$ subject to absorption in a subset $B\subset\mathbb{Z}^2$}
is defined analogously to $P(x,t)$, only the summation is over {checker} paths $s$ that have no common points with $B$ possibly except $(0,0)$ and $(x,t)$.
The probability is denoted by $P(x,t \text{ bypass } B)$. Informally, this means an additional outcome of the experiment: the electron has been absorbed and has not reached the~ {screen}.  {In the following example, we view the two black squares $(\pm1,1)$ on the horizontal $t=1$ as two slits in a horizontal plate (cf.~Figure~\ref{Double-slit}).}

\begin{example}[Double-slit experiment] \label{p-double-slit}  {Distinct paths 
cannot be viewed as ``mutually exclusive'':}
$$P(0,4) \ne P(0,4 \text{ bypass } \{(2,2)\})
+P(0,4 \text{ bypass } \{(0,2)\}).$$
Absorption can increase probabilities at some places:
$P(0,4) = 1/8 < 1/4 = P(0,4 \text{ bypass } \{(2,2)\})$.
\end{example}

The standard interpretation of Feynman checkers is also known as the \emph{Hadamard walk}, or more generally, the \emph{1-dimensional quantum walk} or \emph{quantum lattice gas}. Those are all equivalent but lead to generalizations 
in distinct directions \cite{Venegas-Andraca-12, Konno-20, Yepez-05}. For instance, a unification of the upgrades from \S\ref{sec-mass}--\ref{sec-medium} gives a general inhomogeneous 1-dimensional quantum walk.

The striking properties of quantum walks discussed in~\S\ref{ssec-background} are stated precisely as follows:
$$
\sum_{t=1}^{\infty}P(0,t \text{ bypass } \{x=0\})
= \frac{2}{\pi}
\le \sum_{t=1}^{\infty}P(0,t \text{ bypass } \{x=0\}\cup \{x=n\})
\to\frac{1}{\sqrt{2}}\text{ as }n\to+\infty.
$$
Recently M.~Dmitriev \cite{Dmitriev-22} has found $\sum_{t=1}^{\infty}P(n,t \text{ bypass } \{x=n\})$ for a few values $n\ne 0$ (see~Problem~\ref{p-pi}). Similar numbers appear in the simple random walk on $\mathbb{Z}^2$ \cite[Table~2]{Pakharev-Skopenkov-Ustinov}.

\smallskip
 {\textbf{Spin-chain interpretation.} There is a \emph{very different} physical interpretation of the same model: a $1$-dimensional Ising model with imaginary temperature and fixed magnetization.}

 {Recall that a \emph{configuration} in the Ising model is a sequence $\sigma=(\sigma_1,\dots,\sigma_t)$ of $\pm 1$ of fixed length. The \emph{magnetization} and the \emph{energy} of the configuration are $\sum_{k=1}^{t}\sigma_k/t$ and $H(\sigma)=\sum_{k=1}^{t-1}(1-\sigma_k\sigma_{k+1})$ respectively.
The \emph{probability} of the configuration is $e^{-\beta H(\sigma)}/Z(\beta)$,
where the \emph{inverse temperature} $\beta=1/T>0$ is a parameter and the \emph{partition function} $Z(\beta):=\sum_\sigma e^{-\beta H(\sigma)}$ is a normalization factor. Additional restrictions on configurations $\sigma$ are usually imposed.}

 {
Now, moving the checker along a path $s$, write ``$+$'' for each upwards-right move, and ``$-$'' for each upwards-left one; see Figure~\ref{fig-young} to the left. The resulting sequence of signs is a configuration in the Ising model, the number of turns in $s$ is one half of the configuration energy, and the ratio of the final $x$- and $t$-coordinates is the magnetization. 
Thus $a(x,t)=\sum_s a(s)$ coincides (up to a factor not depending on $x$) with the partition function for the Ising model at the \emph{imaginary} inverse temperature $\beta=i\pi/4$ under the fixed magnetization $x/t$:}
\begin{equation*}
a(x,t)=2^{(1-t)/2}\,i\,
\sum_{\substack{(\sigma_1,\dots,\sigma_t)\in\{+1,-1\}^t:\\
\sigma_1=+1,\quad\sum_{k=1}^{t-1}\sigma_k=x}}
\exp\left(\frac{i\pi}{4}\sum_{k=1}^{t-1}(\sigma_k\sigma_{k+1}-1)\right).
\end{equation*}

\begin{figure}[htbp]
  \centering
  \includegraphics[width=0.2\textwidth]{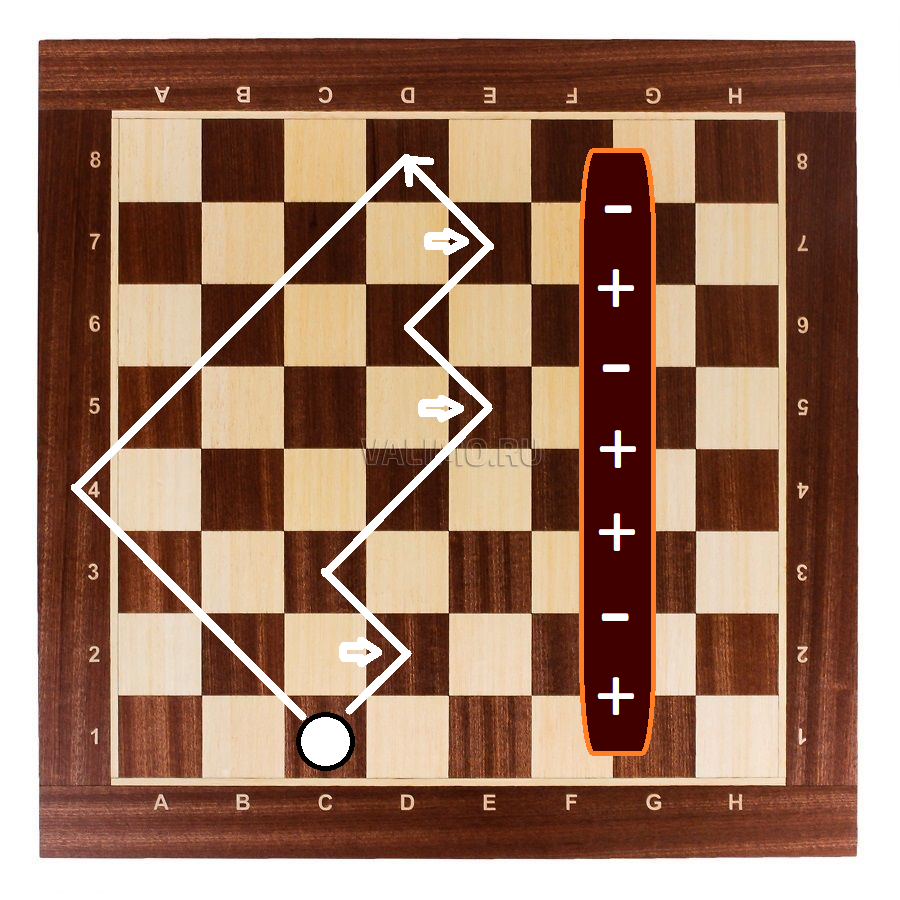}
  \includegraphics[width=0.2\textwidth]{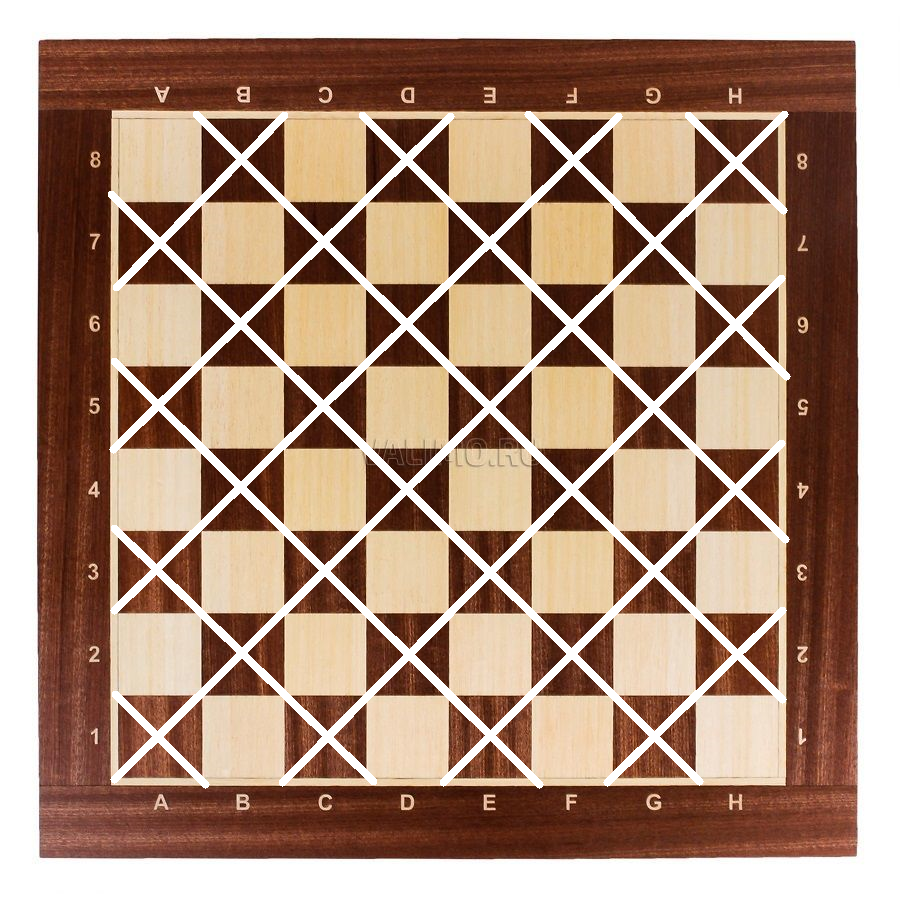}
  \caption{ {\new{A} Young diagram\new{} (the arrows point to steps) and the Ising model \new{(left)}. The auxiliary grid \new{(right)}.}}
  \label{fig-young}
\end{figure}

 {Notice a crucial difference of the resulting spin-chain interpretation from both the usual Ising model and the above standard interpretation. In the latter two models, the magnetization~$x/t$ and the average velocity~$x/t$ were \emph{random variables}; now the magnetization~$x/t$ (not to be confused with an external magnetic field) is an \emph{external condition}. The configuration space in the spin-chain interpretation consists of sequences of ``$+$'' and ``$-$'' with \emph{fixed} numbers of ``$+$'' and~``$-$''. Summation over configurations with different~$x$ or~$t$ would make no sense: e.g., recall that $P(X,t)=\sum_{x\in X}|a(x,t)|^2$ rather than $\left|\sum_{x\in X}a(x,t)\right|^2$.}

 {Varying the magnetization $x/t$, viewed as an external condition, we observe a \emph{phase transition}:
 (the imaginary part of) the limiting \emph{free energy density $-\log a(x,t)/\beta t$}
 is non-analytic when $x/t$ passes through $\pm 1/\sqrt{2}$ (see Figure~\ref{fig-distribution} to the middle and Corollary~\ref{cor-free}). The phase transition emerges as $t\to\infty$.
 \label{equal-signs}
  It is interesting to study other order parameters, for instance, the \emph{``probability'' $i\,a_2(x,t)/a(x,t)$ of equal signs at the endpoints of the spin chain}
  (see Figure~\ref{fig-correlation} and Problems~\ref{p-correlation}--\ref{p-correlation2}). These quantities are complex just because the temperature is imaginary.
}

 {(A comment for specialists: the phase transition is not related to accumulation of zeroes of the partition function in the plane of complex parameter $\beta$ as in \cite[\S III]{Jones-66} and~\cite{Matveev-Shrock-97}. In our situation, $\beta=i\pi/4$ is fixed, the \emph{real} parameter $x/t$ varies, and the partition function $a(x,t)$ has no 
zeroes at all~\cite[Theorem~1]{Novikov-20}.)
}



\smallskip
 {\textbf{Young-diagram interpretation.} Our results have also a combinatorial interpretation.

 {
The \emph{number of steps} (or \emph{inner corners})
in a Young diagram with $w$ columns of heights $x_1,\dots,x_w$ is the number of elements in the set $\{x_1,\dots,x_w\}$; see Figure~\ref{fig-young} \new{to the left}. \label{page-young} Then the value $2^{(h+w-1)/2}\,a_1(h-w,h+w)$ is the difference between the number of Young diagrams with an odd and an even number of steps, having exactly $w$ columns and $h$ rows.}
}

 {
Interesting behaviour starts already for $h=w$ (see Proposition~\ref{cor-coefficients}). For $h=w$ even, the difference vanishes. For $h=w=2n+1$ odd, it is  $(-1)^n\binom{2n}{n}$. Such $4$-periodicity roughly remains for
$h$ close to $w$ \cite[Theorem~2]{SU-20}.
For fixed half-perimeter $h+w$, the difference slowly oscillates as 
$h/w$ increases, attains a peak at $h/w\approx 3+2\sqrt{2}$, and then harshly falls to very small values (see Corollary~\ref{cor-young} and~Theorems~\ref{th-ergenium}--\ref{th-outside}).
}

 {
Similarly, $2^{(h+w-1)/2}a_2(h-w,h+w)$ is the difference between the number of Young diagrams with an even and an odd number of steps, having exactly $w$ columns and \emph{less} than $h$ rows. The behaviour is similar. 
The upgrade in~\S\ref{sec-mass} is related to \emph{Stanley character polynomials}~\cite[\S2]{Stanley-04}.
}

\smallskip
 {\textbf{Discussion of the definition.} Now compare Definition~\ref{def-basic} with the ones in the literature.}

The notation ``$a$'' comes from ``\emph{a}rrow'' and ``probability \emph{a}mplitude'';  {other names are ``wavefunction'',  ``kernel'', ``the Green function'', ``propagator''}. More traditional  {notations are ``$\psi$'',  ``$K$'', ``$G$'', ``$\Delta$'', ``$S$'' depending on the context. We prefer a neutral one suitable for all contexts.}

The factor of $i$ and the minus sign in  {the definition} are irrelevant (and absent in the original definition \cite[Problem~2.6]{Feynman-Gibbs}). They come from the ordinary starting direction and rotation direction of the stopwatch hand, and reduce the number of minus signs in what follows.  {}

 {The normalization factor $2^{(1-t)/2}$ can be explained by analogy to the classical random walk.} If the checker were performing just a random walk, choosing one of the two possible directions at each step (after the obligatory first upwards-right move), then $|a(s)|^2$ {$=2^{1-t}$} would be the probability of a path~$s$.  {This analogy should be taken with a grain of salt: in quantum theory, the ``probability of a path'' has absolutely \emph{no} sense (recall Example~\ref{p-double-slit}). The reason is that the path is not something one can measure: a measurement of the electron position at one moment strongly affects the motion for all later moments.} 

Conceptually, one should also fix the direction of the \emph{last} move of the path $s$  {(see \cite[bottom of p.35]{Feynman-Gibbs})}. Luckily, this is not required in the present paper (and thus is not done), but becomes crucial in further upgrades (see \S\ref{sec-spin}  {for an explanation}).

One could ask where does the definition come from. Following Feynman, we do not try to explain or ``derive'' it physically. This quantum model cannot be obtained from a classical one by the standard Feynman sum-over-paths approach: there is simply \emph{no} clear classical analogue of a spin $1/2$ particle  (cf.~\S\ref{sec-spin}) and \emph{no} true understanding of spin.
``Derivations'' in \cite{Ambainis-etal-01, Bialynicki-Birula-94, Narlikar-72} appeal to much more complicated notions than the model itself.
The  {true} motivation for the model is its simplicity, consistency with basic principles (like probability conservation), and agreement with experiment (which  here means the correct continuum limit; see Corollary~\ref{cor-uniform}).



\addcontentsline{toc}{myshrinkalt}{}

\subsection{Identities and asymptotic formulae}

\mscomm{!!! Update translation: moved several propositions to \S\ref{sec-mass} !!!}

Let us state several well-known basic properties of the model. The proofs are given in~\S\ref{ssec-proofs-basic}.

\mscomm{!!! Removed sentence !!!}
First, the arrow coordinates  {$a_1(x,t)$ and $a_2(x,t)$} satisfy the following recurrence relation.

\begin{proposition}[Dirac equation] \label{p-Dirac}  For each integer $x$ and each positive integer $t$ we have
\begin{align*}
a_1(x,t+1)&=\frac{1}{\sqrt{2}}a_2(x+1,t)+\frac{1}{\sqrt{2}}a_1(x+1,t);\\
a_2(x,t+1)&=\frac{1}{\sqrt{2}}a_2(x-1,t)-\frac{1}{\sqrt{2}}a_1(x-1,t).
\end{align*}
\end{proposition}

This mimics the $(1+1)$-dimensional Dirac equation in the Weyl basis  {\cite[(19.4) and~(3.31)]{Peskin-Schroeder}}
 {
\begin{equation}\label{eq-continuum-Dirac}
\begin{pmatrix}
m  & \partial/\partial x-\partial/\partial t \\
\partial/\partial x+\partial/\partial t & m
\end{pmatrix}
\begin{pmatrix}
a_2(x,t) \\ a_1(x,t)
\end{pmatrix}=0,
\end{equation}
only} the derivatives are replaced by finite differences, $m$ is set to $1$, and the normalization factor $1/\sqrt{2}$ is added.
For the upgrade  {in}~\S\ref{sec-mass}, this analogy becomes 
transparent (see Remark~\ref{rem-equation-limit}).
 {
The Weyl basis is not unique, thus there are several forms of equation~\eqref{eq-continuum-Dirac}; cf.~\cite[(1)]{Jacobson-Schulman-84}.}

 {The Dirac equation implies the conservation of probability.}

\begin{proposition}[Probability/charge conservation] \label{p-probability-conservation} For each integer $t\ge 1$ we get $\sum\limits_{x\in\mathbb{Z}}\!P(x,t)=1$.
\end{proposition}

\comment

\begin{proposition}[Klein--Gordon equation]\label{p-Klein-Gordon} For each integer $x$ and each integer $t\ge 2$ we have
\begin{align*}
\sqrt{2}\, a(x,t+1)+\sqrt{2}\, a(x,t-1)-a(x-1,t)-a(x+1,t)=0.
\end{align*}
\end{proposition}

The analogy to the usual Klein--Gordon equation is going to become transparent in \S\ref{sec-mass}; see Proposition~\ref{p-Klein-Gordon-mass}.

\begin{proposition}[Symmetry]\label{p-symmetry}
For each integer $x$ and each positive integer $t$ we have $$a_1(x,t)=a_1(-x,t)\qquad\text{and}\qquad
a_2(x,t)+a_1(x,t)=a_2(2-x,t)+a_1(2-x,t).
$$
\end{proposition}

\begin{proposition}[Huygens' principle]\label{p-Huygens} For each integers $x$ and $t>t'>0$ we have
\begin{align*}
a_1(x,t)  &=\sum \limits_{x'\in\mathbb{Z}} \left[ a_2(x',t')a_1(x-x'+1,t-t'+1) + a_1(x',t')a_2(x'-x+1,t-t'+1) \right],\\
a_2(x,t)  &= \sum \limits_{x'\in\mathbb{Z}} \left[ a_2(x',t')a_2(x-x'+1,t-t'+1) - a_1(x',t')a_1(x'-x+1,t-t'+1) \right].
\end{align*}
\end{proposition}

Informally, Huygens' principle means that each black square $(x',t')$ on the $(t'/\varepsilon)$-th horizontal acts like an independent point source, with the amplitude and phase determined by $a(x',t')$.



 {To write a formula for $a(x,t)$, set $\binom{n}{k}:=0$ for $k<0<n$ or $k>n>0$. Denote
$\theta(x):=\begin{cases}
                            1, & \mbox{if } x\ge0, \\
                            0, & \mbox{if } x<0.
                          \end{cases}$
}

\endcomment

 {
For $a_1(x,t)$ and $a_2(x,t)$, there is an ``explicit'' formula 
(more ones are given in Appendix~\ref{app-formula}).
}

\begin{proposition}[``Explicit'' formula] \label{Feynman-binom}
For each integers $|x|<t$ such that $x+t$ is even we have
\begin{align*}
a_1(x,t)&=2^{(1-t)/2}\sum_{r=0}^{ {(t-|x|)/2}}(-1)^r \binom{(x+t-2)/2}{r}\binom{(t-x-2)/2}{r},\\
a_2(x,t)&=2^{(1-t)/2}\sum_{r=1}^{ {(t-|x|)/2}}(-1)^r \binom{(x+t-2)/2}{r}\binom{(t-x-2)/2}{r-1}.
\end{align*}
\end{proposition}





The following proposition is a straightforward corollary of the \new{``explicit'' formula}.

\begin{proposition}[Particular values]
\label{cor-coefficients}
For each \new{integer} $1\le k\le t-1$ the numbers $a_1(-t+2k,t)$ and $a_2(-t+2k,t)$ are the coefficients \new{at} $z^{t-k-1}$ and $z^{t-k}$ in the expansion of the polynomial $2^{(1-t)/2}(1+z)^{t-k-1}(1-z)^{k-1}$. In particular,
\begin{align*}
a_1(0,4n+2)&=\frac{(-1)^n}{2^{(4n+1)/2}}\binom{2n}{n},
&
a_1(0,4n)&=0,\\
a_2(0,4n+2)&=0,
&
a_2(0,4n)&=\frac{(-1)^n}{2^{(4n-1)/2}}\binom{2n-1}{n}.
\end{align*}
\end{proposition}

In \S\ref{ssec-mass-properties} we give more identities.
The sequences $2^{(t-1)/2}a_1(x,t)$ and $2^{(t-1)/2}a_2(x,t)$
are present in the
on-line encyclopedia of integer sequences~\cite[A098593 and A104967]{oeis}.

\comment

\begin{corollary}
\label{cor-alt-symmetry} 
For each $x,t\in\mathbb{Z}$, $t>0$, we have   $(t-x)\,a_2(x,t)=(t+x-2)\,a_2(2-x,t)$.
\end{corollary}

\begin{remark}\label{cor-hypergeometric}
By definition of the Gauss hypergeometric function,
for each $|x|<t$ such that $x+t$ is even,
\begin{align*}
a_1(x,t)&=2^{(1-t)/2}
\,{}_2F_1\left(1 - \frac{x+t}{2}, 1 + \frac{x-t}{2}, 1; -1\right),\\
a_2(x,t)&=2^{(1-t)/2 }\left(1 - \frac{x+t}{2}\right)
\,{}_2F_1\left(2 - \frac{x+t}{2}, 1 + \frac{x-t}{2}, 2; -1\right).
\end{align*}
\end{remark}


Propositions~\ref{p-Dirac} and~\ref{Feynman-binom} and their generalizations in \S\ref{ssec-mass-properties} are versions of well-known results. Surprisingly, Proposition~\ref{p-probability-conservation} seems to be new; cf.~the massless case in~\cite{Bialynicki-Birula-94}.
For the other formulae from~\S\ref{ssec-mass-properties} and Appendix~\ref{app-formula}, we have found no analogues in the literature.
Numerous other identities have been found by I.~Novikov in numeric experiments \cite[\S4]{Novikov-20}.

\endcomment





The following remarkable result was observed in \cite[\S4]{Ambainis-etal-01} (see Figures~\ref{fig-distribution} and~\ref{fig-correlation}), stated precisely and derived heuristically in~\cite[Theorem~1]{Konno-05}, and proved mathematically in~\cite[Theorem~1]{Grimmett-Janson-Scudo-04}. 
See a short exposition of the latter proof in \S\ref{ssec-proofs-moments}, and generalizations 
in~\S\ref{ssec-mass-asymptotic}.



\begin{theorem}[Large-time limiting distribution; see Figure~\ref{fig-correlation}]
\label{th-limiting-distribution}
$\mathbf{(A)}$ For each $v\in\mathbb{R}$ we have
\begin{equation*}
\lim_{t\to\infty}\sum_{x\le vt}P(x,t)
= F(v):=
\begin{cases}
  0, & \mbox{if } v\le -1/\sqrt{2}; \\
  \frac{1}{\pi}\arccos\frac{1-2v}{\sqrt{2}(1-v)},
  & \mbox{if } |v|<1/\sqrt{2}; \\ 
  1, & \mbox{if }v\ge 1/\sqrt{2}.
\end{cases}
\end{equation*}
$\mathbf{(B)}$ We have the following convergence in distribution as $t\to\infty$:
$$
tP(\lceil vt\rceil,t)\overset{d}\to F'(v)=
\begin{cases}
  \frac{1}{\pi (1-v)\sqrt{1-2v^2}}, & \mbox{if } |v|< 1/\sqrt{2}; \\
  0, & \mbox{if } |v|\ge 1/\sqrt{2}.
\end{cases}
$$
$\mathbf{(C)}$ For each integer $r\ge 0$ we have
$
\lim_{t\to\infty}\sum_{x\in\mathbb{Z}} \left(\frac{x}{t}\right)^r P(x,t)=
\int_{-1}^{1}
v^r F'(v)\,dv. 
$
\end{theorem}

Theorem~\ref{th-limiting-distribution}(B) demonstrates a phase transition in Feynman checkers, if interpreted as an Ising model at imaginary temperature and fixed magnetization. Recall that then the magnetization $v$ is an external condition (rather than a random variable) and $P(\lceil vt\rceil,t)$ is the norm square of the partition function (rather than a probability). The distributional limit of $tP(\lceil vt\rceil,t)$ is discontinuous at $v=\pm 1/\sqrt{2}$, reflecting a phase transition (cf.~Corollary~\ref{cor-free}).

\comment

Our first new result is an analytic approximation of ${a}(x,t)$ accurate for small $|x|/t$ (see Figure~\ref{fig-approximation}).
This solves an analogue of the Feynman problem 
for the basic model (cf.~Corollary~\ref{cor-feynman-problem}).

\begin{theorem}[Large-time limit near the origin] \label{Feynman-convergence}
For each integers $x,t$ such that $|x|<t^{3/4}$ and $x+t$ is even we have 
\begin{align}\label{eq-Feynman-convergence}
{a}(x,t)&=i\sqrt{\frac{2}{\pi t}}
\exp \left(-\frac{i\pi t}{4}+\frac{i x^2}{2t}\right)
+i\frac{2x}{\sqrt{\pi t^3}}\,
\cos\left(-\frac{\pi (t+1)}{4}+\frac{x^2}{2t}\right)
+{O}\left(\frac{\log^{2}t}{t^{3/2}}+\frac{x^{4}}{t^{7/2}}\right),
\\ \label{eq-Feynman-convergence2}
P(x,t)&=\frac{2}{\pi t}\left[1+\frac{x}{t}\left(1+\sqrt{2}
\cos\left(\frac{\pi (2t+1)}{4}-\frac{x^2}{t}\right)\right)
+{O}\left(\frac{\log^{2}t}{t}+\frac{x^{4}}{t^3}\right)\right].
\end{align}
\end{theorem}

Recall that 
$f(x,t)={O}\left(g(x,t)\right)$ means that there is a constant $C$ (not depending on $x,t$) such that for each $x,t$ satisfying the assumptions of the theorem we have $|f(x,t)|\le C\,g(x,t)$.

\begin{figure}[htbp]
\begin{center}
\includegraphics[width=0.32\textwidth]{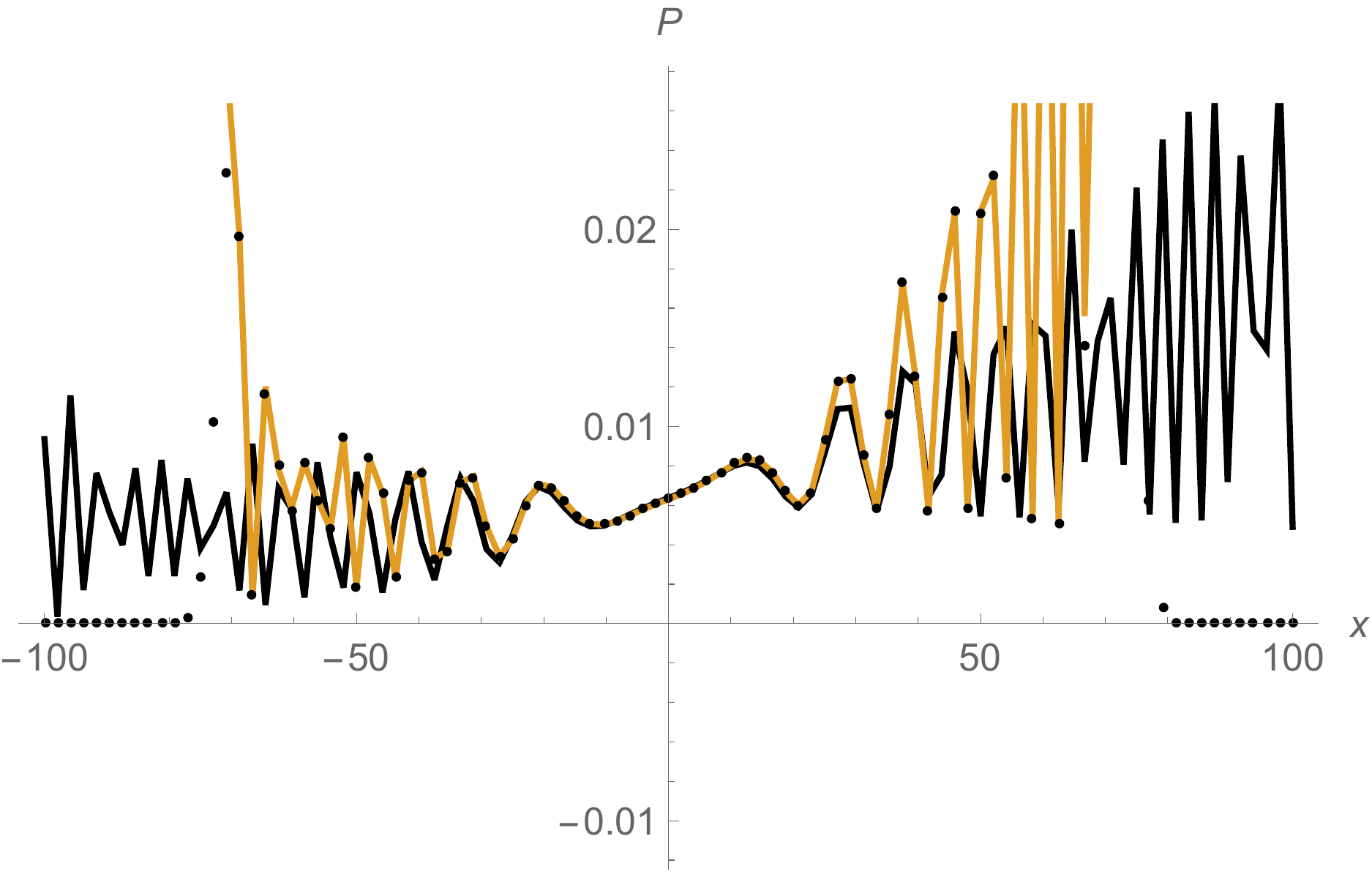}
\includegraphics[width=0.32\textwidth]{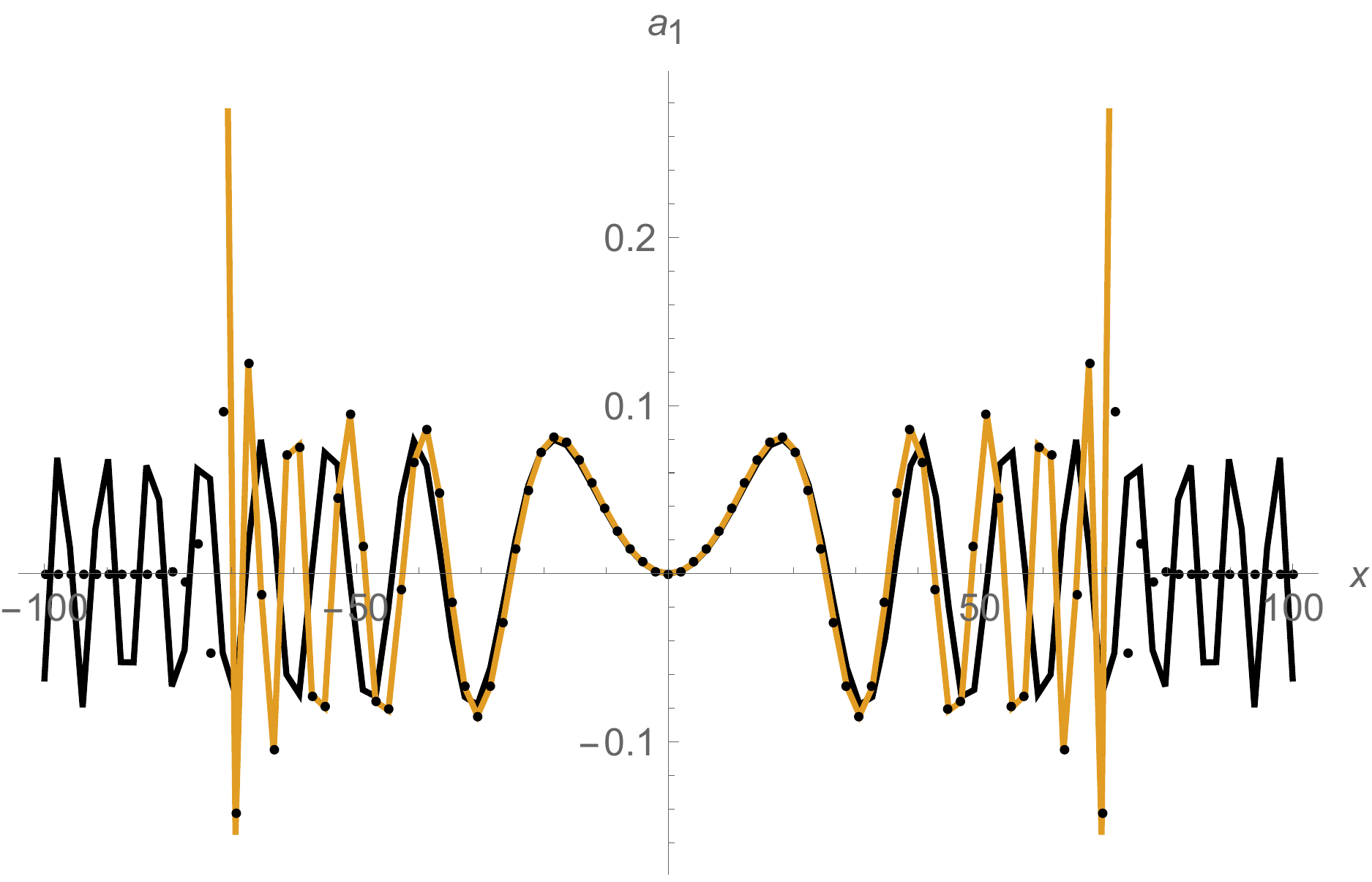}
\includegraphics[width=0.32\textwidth]{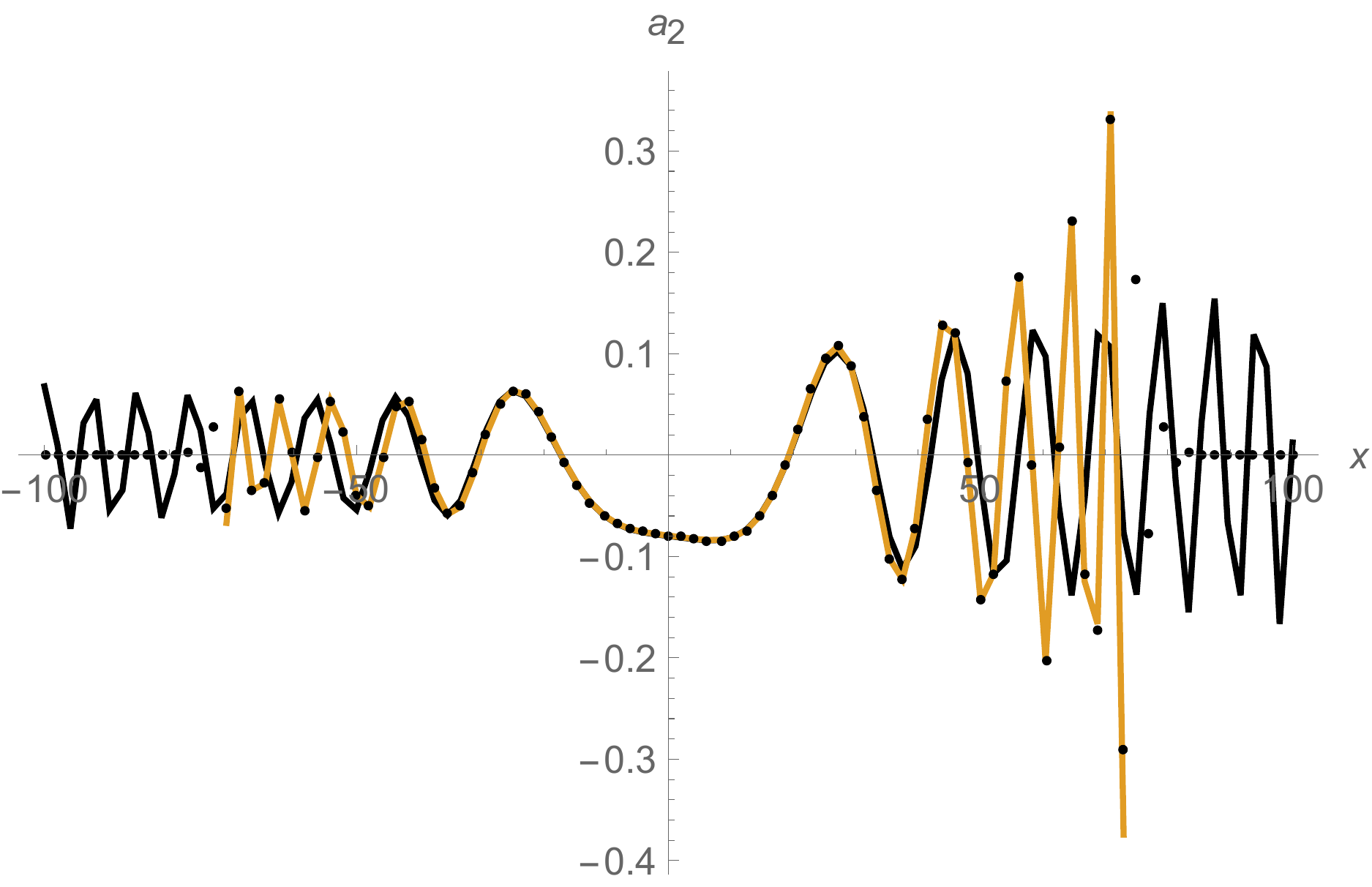}
\end{center}
\vspace{-0.6cm}
\caption{The graphs of $P(x,100)$, $a_1(x,100)$, $a_2(x,100)$ (dots) and their analytic approximations from Theorem~\ref{Feynman-convergence} (dark) and Theorem~\ref{th-ergenium} (light).}
\label{fig-approximation}
\end{figure}

Thus for 
$x/t=o(t^{-1/4})$ the model asymptotically reproduces the \emph{free-particle kernel}
\begin{equation}\label{eq-free-particle-kernel}
K(x,t)=\sqrt{\frac{m}{2\pi t}}
\exp \left(\frac{i m x^2}{2t}-\frac{i\pi}{4}\right).
\end{equation}
The first term in~\eqref{eq-Feynman-convergence} is kernel~\eqref{eq-free-particle-kernel} for $m=1$ times $2i^{3/2}e^{-i\pi t/4}$. The latter ``exceptional'' factor 
differs from the one in the Feynman problem: $-it$ is multiplied by $\pi/4$ in the exponential. This additional $\pi/4$ is due to the finite lattice step $1/m$. The second term in~\eqref{eq-Feynman-convergence} is a kind of relativistic correction responsible for the dependence of the probability $P(x,t)$ on $x$. 
The second term coincides with the one in asymptotic expansion of the Bessel functions in spin-$1/2$ retarded propagator~\eqref{eq-relativistic-propagator} for $m=1$, up to a factor of $\pi/4$ in the argument of cosine. We shall see that the assumption $|x|<t^{3/4}$ is essential in the Feynman problem (see Example~\ref{p-Feynman-couterexample}). 
Although Theorem~\ref{Feynman-convergence} is an easy corollary of Theorem~\ref{th-ergenium}, we give a direct alternative proof in~\S\ref{ssec-proofs-alternative}.

\endcomment

\addcontentsline{toc}{myshrink}{}

\section{Mass (biased quantum walk)} \label{sec-mass}


{
\hrule
\footnotesize
\noindent\textbf{Question:} what is the probability to find an electron of mass $m$ in the square $(x,t)$, if it was emitted from~$(0,0)$?

\noindent\textbf{Assumptions:} the mass and the lattice step are now arbitrary.

\noindent\textbf{Results:} analytic expressions for the probability for large time or small lattice step, concentration of measure.
\hrule
}

\bigskip

\addcontentsline{toc}{myshrinkalt}{}

\subsection{Identities}\label{ssec-mass-properties}


\begin{definition} \label{def-mass} 
Fix $\varepsilon>0$ and $m\ge 0$ called \emph{lattice step} and \emph{particle mass} respectively. Consider the lattice $\varepsilon\mathbb{Z}^2
=\{\,(x,t):x/\varepsilon,t/\varepsilon\in\mathbb{Z}\,\}$ (see Figure~\ref{fig-triple-limit}). A \emph{checker path} $s$ is a finite sequence of lattice points such that the vector from each point (except the last one) to the next one equals either $(\varepsilon,\varepsilon)$ or $(-\varepsilon,\varepsilon)$. Denote by $\mathrm{turns}(s)$ the number of points in $s$ (not the first and not the last one) such that the vectors from the point to the next and to the previous ones are orthogonal.
For each $(x,t)\in\varepsilon\mathbb{Z}^2$, where $t>0$, denote by
\begin{equation}\label{eq-def-mass}
{a}(x,t,m,\varepsilon)
:=(1+m^2\varepsilon^2)^{(1-t/\varepsilon)/2}\,i\,\sum_s
(-im\varepsilon)^{\mathrm{turns}(s)}
\end{equation}
the sum over all checker paths $s$ on $\varepsilon\mathbb{Z}^2$ from $(0,0)$ to $(x,t)$ with the first step to $(\varepsilon,\varepsilon)$. Denote
$$P(x,t,m,\varepsilon):=|{a}(x,t,m,\varepsilon)|^2,
\quad
a_1(x,t,m,\varepsilon):=\mathrm{Re}\,a(x,t,m,\varepsilon),
\quad
a_2(x,t,m,\varepsilon):=\mathrm{Im}\,a(x,t,m,\varepsilon).
$$
\end{definition}

\begin{remark}\label{rem-particular}
In particular, $P(x,t)=P(x,t,1,1)$ and 
$a(x,t)=a(x,t,1,1)
=a(x\varepsilon,t\varepsilon,1/\varepsilon,\varepsilon)$.
\end{remark}

One interprets $P(x,t,m,\varepsilon)$ as the probability to find an electron of mass $m$ in the square $\varepsilon\times\varepsilon$ with the center $(x,t)$, if the electron was emitted from the origin. Notice that the value $m\varepsilon$, hence $P(x,t,m,\varepsilon)$, is dimensionless in the natural units, where $\hbar=c=1$. In Figure~\ref{P-contour} to the middle, the color of a point $(x,t)$ depicts the value $P(x,t,1,0.5)$ (if $x+t$ is an integer). Table~\ref{table-am} shows $a(x,t,m,\varepsilon)$ for small $x$ and $t$.
Recently I.~Novikov elegantly proved that the probability vanishes nowhere inside the light cone: $P(x,t,m,\varepsilon)\ne 0$ for $m>0$, $|x|<t$ and even $(x+t)/\varepsilon$ \cite[Theorem~1]{Novikov-20}.

\begin{table}[htbp]
\begin{tabular}{|c|c|c|c|c|c|c|c|}
\hline
$4\varepsilon$&$\frac{m \varepsilon}{(1+m^2\varepsilon^2)^{3/2}}$&&$\frac{(m \varepsilon- m^3 \varepsilon^3) - m^2 \varepsilon^2 i}{(1+m^2\varepsilon^2)^{3/2}}$&&$\frac{m \varepsilon - 2 m^2 \varepsilon^2 i}{(1+m^2\varepsilon^2)^{3/2}}$&&$\frac{\new{i}}{(1+m^2\varepsilon^2)^{3/2}}$\\
\hline
$3\varepsilon$&&$\frac{m \varepsilon}{1+m^2\varepsilon^2}$&&$\frac{m \varepsilon - m^2 \varepsilon^2 i}{1+m^2\varepsilon^2}$&&$\frac{\new{i}}{(1+m^2\varepsilon^2)}$&\\
\hline
$2\varepsilon$&&&$\frac{m\varepsilon}{\sqrt{1+m^2\varepsilon^2}}$&&$\frac{\new{i}}{\sqrt{1+m^2\varepsilon^2}}$&&\\
\hline
$\varepsilon$&&&&$i$&&&\\
\hline
\diagbox[dir=SW,height=21pt]{$t$}{$x$}&$-2\varepsilon$&$-\varepsilon$&$0$&$\varepsilon$&$2\varepsilon$&$3\varepsilon$&$4\varepsilon$ \\
\hline
\end{tabular}
\caption{(by I.~Novikov \cite{Novikov-20}) The values $a(x,t,m,\varepsilon)$ for small $x$ and $t$.}
\label{table-am}
\end{table}

\begin{example}[Boundary values] \label{ex-boundary-values}
We have
$a(t,t,m,\varepsilon)=i(1+m^2\varepsilon^2)^{(1-t/\varepsilon)/2}$ and
$a(2\varepsilon-t,t,m,\varepsilon)
=m\varepsilon(1+m^2\varepsilon^2)^{(1-t/\varepsilon)/2}$
for each $t\in \varepsilon\mathbb{Z}$, $t>0$, and
$a(x,t,m,\varepsilon)=0$ for each $x>t$ or $x\le -t$.
\end{example}

\begin{example}[Massless and heavy particles]\label{p-massless}  For each $(x,t)\in \varepsilon\mathbb{Z}^2$, where $t>0$, we have
$$P(x,t,0,\varepsilon)=\begin{cases}
1, &\text{\!for }x=t;\\
0,   &\text{\!for }x\ne t.
\end{cases}
\quad\text{and}\quad
\lim_{m \to \infty}P(x,t,m,\varepsilon)=
\begin{cases}
1, &\text{\!for } x=0 \text{ or } \varepsilon, \text{ and }\frac{x+t}{\varepsilon}\text{ even};\\
0,   &\text{\! otherwise}.
\end{cases}$$
\end{example}


Let us list known combinatorial properties of the model \cite{Venegas-Andraca-12, Konno-08}; see \S\ref{ssec-proofs-basic} for simple proofs.

\begin{proposition}[Dirac equation]\label{p-mass}
For each $(x,t)\in \varepsilon\mathbb{Z}^2$, where $t>0$, we have
\begin{align}\label{eq-Dirac-mass1}
a_1(x,t+\varepsilon,m, \varepsilon) &= \frac{1}{\sqrt{1+m^2\varepsilon^2}}
(a_1(x+\varepsilon,t,m, \varepsilon)
+ m \varepsilon\, a_2(x+\varepsilon,t,m, \varepsilon)),\\
\label{eq-Dirac-mass2}
a_2(x,t+\varepsilon,m, \varepsilon) &= \frac{1}{\sqrt{1+m^2\varepsilon^2}}
(a_2(x-\varepsilon,t,m, \varepsilon)
- m \varepsilon\, a_1(x-\varepsilon,t,m, \varepsilon)).
\end{align}
\end{proposition}

\begin{remark}\label{rem-equation-limit} This equation reproduces Dirac equation~\eqref{eq-continuum-Dirac} in the continuum limit $\varepsilon\to 0$: for $C^2$ functions $a_1,a_2\colon\mathbb{R}\times (0,+\infty)\to \mathbb{R}$ satisfying \eqref{eq-Dirac-mass1}--\eqref{eq-Dirac-mass2} on $\varepsilon\mathbb{Z}^2$, the left-hand side of~\eqref{eq-continuum-Dirac} is
${O}_m\left(\varepsilon\cdot(\|a_1\|_{C^2}+\|a_2\|_{C^2})\right)$.
\end{remark}

\begin{proposition}[Probability conservation] \label{p-mass2}
For each $t\in\varepsilon\mathbb{Z}$, $t>0$, we get $\sum\limits
_{x\in\varepsilon\mathbb{Z}}\!P(x,t,m, \varepsilon)=1$.
\end{proposition}

\begin{proposition}[Klein--Gordon equation] \label{p-Klein-Gordon-mass} For each $(x,t)\in \varepsilon\mathbb{Z}^2$, where $t>\varepsilon$, we have
\begin{align*}
\sqrt{1+m^2\varepsilon^2}\,a(x,t+\varepsilon,m, \varepsilon)
+\sqrt{1+m^2\varepsilon^2}\,a(x,t-\varepsilon,m, \varepsilon)
-a(x+\varepsilon,t,m, \varepsilon)-a(x-\varepsilon,t,m, \varepsilon)=0.
\end{align*}
\end{proposition}

This equation reproduces the \emph{Klein--Gordon equation}
$\tfrac{\partial^2 a}{\partial t^2}-\tfrac{\partial^2 a}{\partial x^2}+m^2a=0$
in the limit $\varepsilon\to 0$.

\begin{proposition}[Symmetry]\label{p-symmetry-mass} 
For each $(x,t)\in \varepsilon\mathbb{Z}^2$, where $t>0$, we have
\begin{gather*}
  a_1(x,t,m, \varepsilon)=a_1(-x,t,m, \varepsilon),
  \qquad\qquad
  (t-x)\,a_2(x,t,m, \varepsilon)
  =(t+x-2\varepsilon)\,a_2(2\varepsilon-x,t,m, \varepsilon),
  \\
  a_1(x,t,m, \varepsilon)+m\varepsilon\, a_2(x,t,m,\varepsilon)
  =a_1(2\varepsilon-x,t,m, \varepsilon)+m\varepsilon\, a_2(2\varepsilon-x,t,m, \varepsilon).
\end{gather*}
\end{proposition}

\begin{proposition}[Huygens' principle]\label{p-Huygens} For each $x,t,t'\in\varepsilon\mathbb{Z}$, where $t>t'>0$, we have
\begin{align*}
\hspace{-1cm}a_1(x,t,m, \varepsilon)  &=\sum \limits_{x'\in\varepsilon\mathbb{Z}}
\left[ a_2(x',t',m, \varepsilon) a_1(x-x'+\varepsilon,t-t'+\varepsilon,m, \varepsilon) + a_1(x',t',m, \varepsilon)
a_2(x'-x+\varepsilon,t-t'+\varepsilon,m, \varepsilon) \right],\\
\hspace{-1cm}a_2(x,t,m, \varepsilon)  &= \sum \limits_{x'\in\varepsilon\mathbb{Z}}
\left[ a_2(x',t',m, \varepsilon)
a_2(x-x'+\varepsilon,t-t'+\varepsilon,m, \varepsilon) - a_1(x',t',m, \varepsilon)
a_1(x'-x+\varepsilon,t-t'+\varepsilon,m, \varepsilon) \right].
\end{align*}
\end{proposition}

Informally, Huygens' principle means that each black square $(x',t')$ on one 
horizontal acts like an independent point source, with the amplitude and phase determined by $a(x',t',m, \varepsilon)$.

\begin{proposition}[Equal-time recurrence relation]
\label{p-equal-time} For each $(x,t)\in\varepsilon\mathbb{Z}^2$, where $t>0$, we have
\begin{multline}\label{eq-p-equal-time}
(x+\varepsilon)((x-\varepsilon)^2-(t-\varepsilon)^2)
a_1(x-2\varepsilon,t,m,\varepsilon) +(x-\varepsilon)((x+\varepsilon)^2-(t-\varepsilon)^2)
a_1(x+2\varepsilon,t,m,\varepsilon)=
\\
=2x \left((1+2m^2\varepsilon^2)(x^2-\varepsilon^2)-(t-\varepsilon)^2\right)
a_1(x,t,m,\varepsilon),
\end{multline}
\vspace{-1.0cm}
\begin{multline*}
x((x-2\varepsilon)^2-t^2)a_2(x-2\varepsilon,t,m,\varepsilon) +(x-2\varepsilon)(x^2-(t-2\varepsilon)^2)
a_2(x+2\varepsilon,t,m,\varepsilon)=
\\
=2(x-\varepsilon)
\left((1+2m^2\varepsilon^2)(x^2-2\varepsilon x)-t^2+2\varepsilon t \right)
a_2(x,t,m,\varepsilon).
\end{multline*}
\end{proposition}

This allows to compute $a_{1}(x,t)$ and $a_2(x,t)$ quickly
on far horizontals, starting from $x=2\varepsilon-t$ and $x=t$ respectively (see Example~\ref{ex-boundary-values}).



\begin{proposition}[``Explicit'' formula] \label{p-mass3}
For each $(x,t)\in\varepsilon\mathbb{Z}^2$ with $|x|<t$  and $(x+t)/\varepsilon$ even,
\begin{align}
a_1(x,t,m,\varepsilon) &=
(1+m^2\varepsilon^2)^{(1-t/\varepsilon)/2}
\sum_{r=0}^{\frac{t-|x|}{2\varepsilon}}(-1)^r \binom{(x+t)/2\varepsilon-1}{r}\binom{(t-x)/2\varepsilon-1}{r}
(m\varepsilon)^{2r+1},
\label{eq1-p-mass}\\
a_2(x,t,m,\varepsilon)&=
(1+m^2\varepsilon^2)^{(1-t/\varepsilon)/2}
\sum_{r=1}^{\frac{t-|x|}{2\varepsilon}}(-1)^r \binom{(x+t)/2\varepsilon-1}{r}\binom{(t-x)/2\varepsilon-1}{r-1}
(m\varepsilon)^{2r}\new{.}
\label{eq2-p-mass}
\end{align}
\end{proposition}

\begin{remark}\label{rem-hypergeo}
  For each $|x|\ge t$ we have $a(x,t,m,\varepsilon)=0$ unless \new{}$0<x=t\in\varepsilon\mathbb{Z}$, which gives $a(t,t,m,\varepsilon)=(1+m^2\varepsilon^2)^{(1-t/\varepsilon)/2}i$. Beware that the proposition 
  is \emph{not} applicable for $|x|\ge t$.
%

By the definition of Gauss hypergeometric function, we can rewrite the formula as follows: 
\begin{align*}
a_1(x,t,m,\varepsilon)
&=m\varepsilon\left({1+m^2\varepsilon^2}\right)^{(1-t/\varepsilon)/2}
\cdot{}_2F_1\left(1 - \frac{x+t}{2\varepsilon}, 1 + \frac{x-t}{2\varepsilon}; 1; -m^2\varepsilon^2\right),\\
a_2(x,t,m,\varepsilon)
&=m^2\varepsilon^2\left({1+m^2\varepsilon^2}\right)^{(1-t/\varepsilon)/2}
\cdot{}_2F_1\left(2 - \frac{x+t}{2\varepsilon}, 1 + \frac{x-t}{2\varepsilon}; 2; -m^2\varepsilon^2\right)\cdot\left(1 - \frac{x+t}{2\varepsilon}\right).
\end{align*}
This gives a lot of identities. E.g., Gauss contiguous relations connect the values $a(x,t,m,\varepsilon)$ at any $3$ neighboring lattice points; cf.~Propositions~\ref{p-mass} and~\ref{p-equal-time}.
In terms of the Jacobi polynomials,
\begin{align*}
a_1(x,t,m,\varepsilon)
&=
m\varepsilon(1+m^2\varepsilon^2)^{(x/\varepsilon-1)/2}
P_{(x+t)/2\varepsilon-1}^{(0, -x/\varepsilon)} \left(\frac{1-m^2\varepsilon^2}{1+m^2\varepsilon^2}\right),\\
a_2(x,t,m,\varepsilon)
&=
-m^2\varepsilon^2(1+m^2\varepsilon^2)^{(x/\varepsilon - 3)/2}
P_{(x+t)/2\varepsilon - 2}^{(1, 1 - x/\varepsilon)} \left(\frac{1-m^2\varepsilon^2}{1+m^2\varepsilon^2}\right).
\end{align*}
There is a similar expression through Kravchuk polynomials (cf.~Proposition~\ref{cor-coefficients}).
In terms of Stanley character polynomials (defined in \cite[\S2]{Stanley-04}),
$$
a_2(0,t,m,\varepsilon)
=(-1)^{t/2\varepsilon-1}(1+m^2\varepsilon^2)^{(1-t/\varepsilon)/2}
\left(\tfrac{t}{2\varepsilon}-1\right)G_{t/2\varepsilon-1}(1;m^2\varepsilon^2).
$$
\end{remark}

\begin{proposition}[Fourier integral] \label{cor-fourier-integral}
Set $\omega_p:=\frac{1}{\varepsilon}\arccos(\frac{\cos p\varepsilon}{\sqrt{1+m^2\varepsilon^2}})$. Then for each $m>0$ and $(x,t)\in \varepsilon\mathbb{Z}^2$ such that $t>0$ and $(x+t)/\varepsilon$ is even we have
\begin{align*}
  a_1(x, t,m,\varepsilon)&=
  \frac{im\varepsilon^2}{2\pi}
  \int_{-\pi/\varepsilon}^{\pi/\varepsilon}
  \frac{e^{i p x-i\omega_p(t-\varepsilon)}\,dp}
  {\sqrt{m^2\varepsilon^2+\sin^2(p\varepsilon)}},\\
  a_2(x, t,m,\varepsilon)&=
  \frac{\varepsilon}{2\pi}\int_{-\pi/\varepsilon}^{\pi/\varepsilon}
  \left(1+
  \frac{\sin (p\varepsilon)} {\sqrt{m^2\varepsilon^2+\sin^2(p\varepsilon)}}\right) e^{ip(x-\varepsilon)-i\omega_p(t-\varepsilon)}\,dp.
\end{align*}
\end{proposition}


Fourier integral represents a wave emitted by a point source as a superposition of waves of wavelength $2\pi/p$ and frequency~$\omega_p$.

\begin{proposition}[Full space-time Fourier transform] \label{cor-double-fourier}
Denote $\delta_{x\varepsilon}:=1$, if $x=\varepsilon$, and $\delta_{x\varepsilon}:=0$, if $x\ne\varepsilon$.
For each $m>0$ and $(x,t)\in \varepsilon\mathbb{Z}^2$ such that $t>0$ and $(x+t)/\varepsilon$ is even we get
\begin{align*}
  a_1(x,t,m,\varepsilon)&=
  \lim_{\delta\to+0}
  \frac{m\varepsilon^3}{4\pi^2}
  \int_{-\pi/\varepsilon}^{\pi/\varepsilon}
  \int_{-\pi/\varepsilon}^{\pi/\varepsilon}
  \frac{e^{i p x-i\omega(t-\varepsilon)}\,d\omega dp} {\sqrt{1+m^2\varepsilon^2}\cos(\omega\varepsilon)
  -\cos(p\varepsilon)-i\delta},\\
  a_2(x,t,m,\varepsilon)&=
  \lim_{\delta\to+0}
  \frac{-i\varepsilon^2}{4\pi^2}
  \int_{-\pi/\varepsilon}^{\pi/\varepsilon}
  \int_{-\pi/\varepsilon}^{\pi/\varepsilon}
  \frac{\sqrt{1+m^2\varepsilon^2}\sin(\omega\varepsilon)+\sin(p\varepsilon)} {\sqrt{1+m^2\varepsilon^2}\cos(\omega\varepsilon)
  -\cos(p\varepsilon)-i\delta}
  e^{i p (x-\varepsilon)-i\omega(t-\varepsilon)}
  \,d\omega dp+\delta_{x\varepsilon}\delta_{t\varepsilon}.
\end{align*}
\end{proposition}


\addcontentsline{toc}{myshrinkalt}{}

\subsection{Asymptotic formulae}\label{ssec-mass-asymptotic}

\textbf{Large-time limit.} (See Figure~\ref{fig-approx-error}) For fixed $m,\varepsilon$, and large time $t$\new{,} the function $a(x,t,m,\varepsilon)$
\begin{description}
  \item[(A)] oscillates in the interval between the (approximate) peaks $x=\pm \frac{t}{\sqrt{1+m^2\varepsilon^2}}$ (see Theorem~\ref{th-ergenium});
  \item[(B)] is approximated by the Airy function around the peaks (see Theorem~\ref{th-Airy});
  \item[(C)] exponentially decays outside the interval bounded by the peaks (see Theorem~\ref{th-outside}).
\end{description}

We start with discussing regime~\textbf{(A)}. Let us state our main theorem: an analytic approximation of $a(x,t,m,\varepsilon)$, \new{which is} very accurate between the peaks not too close to them
(see Figure~\ref{fig-approx-error} to the left).

\begin{figure}[htbp]
  \centering
 \begin{tabular}{ccc}
\includegraphics[width=0.3\textwidth]{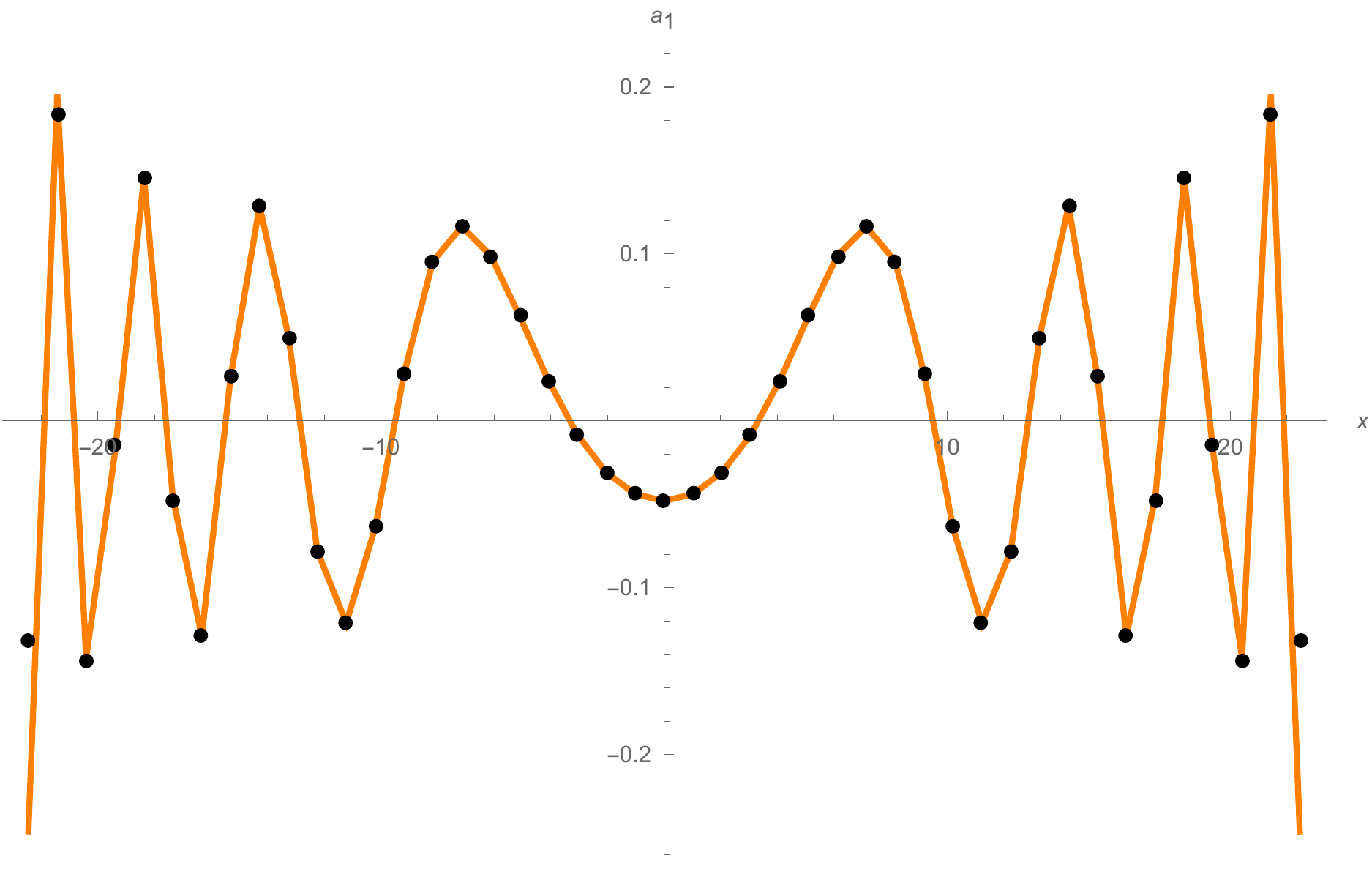}
&\includegraphics[width=0.3\textwidth]{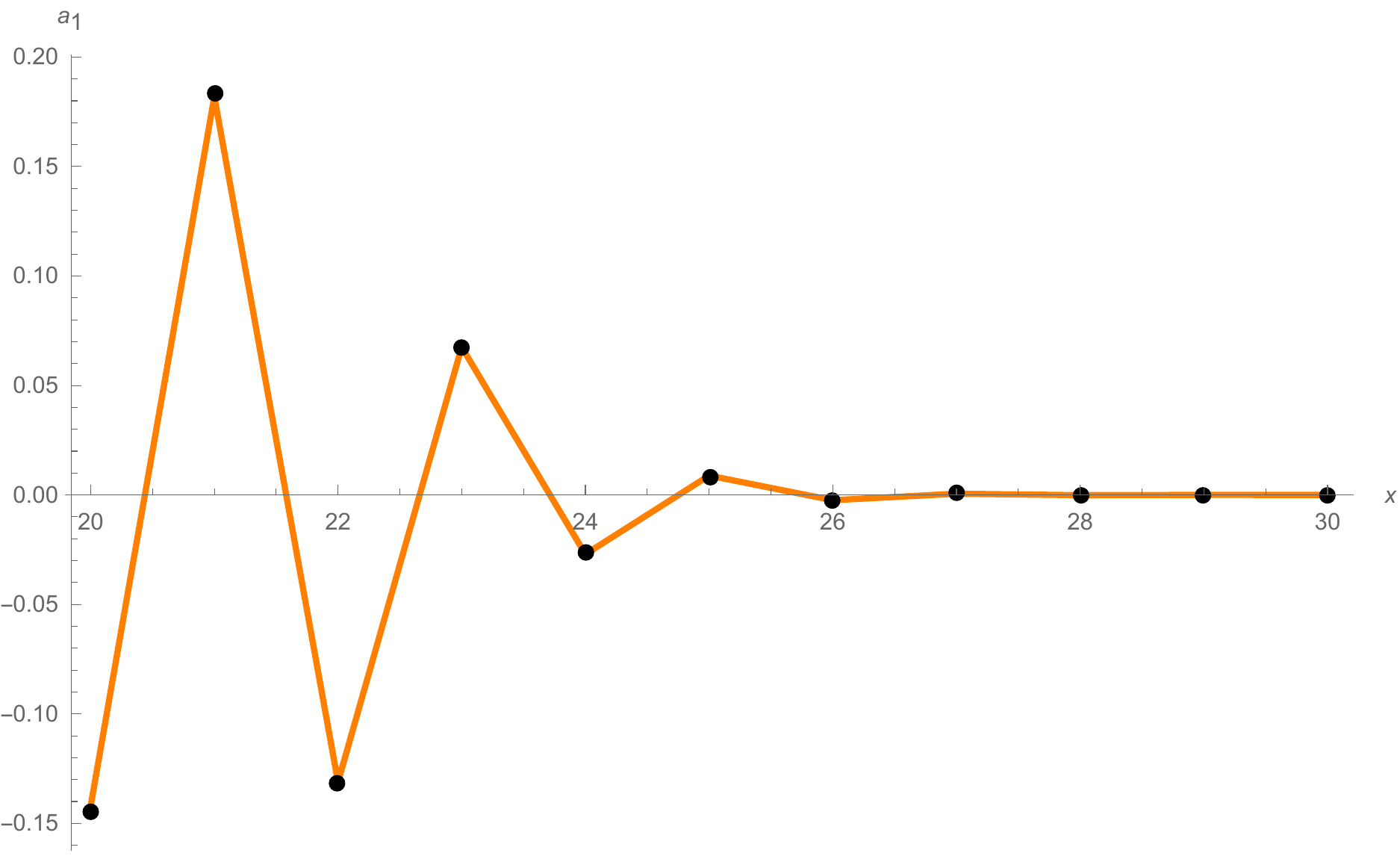} &\includegraphics[width=0.3\textwidth]{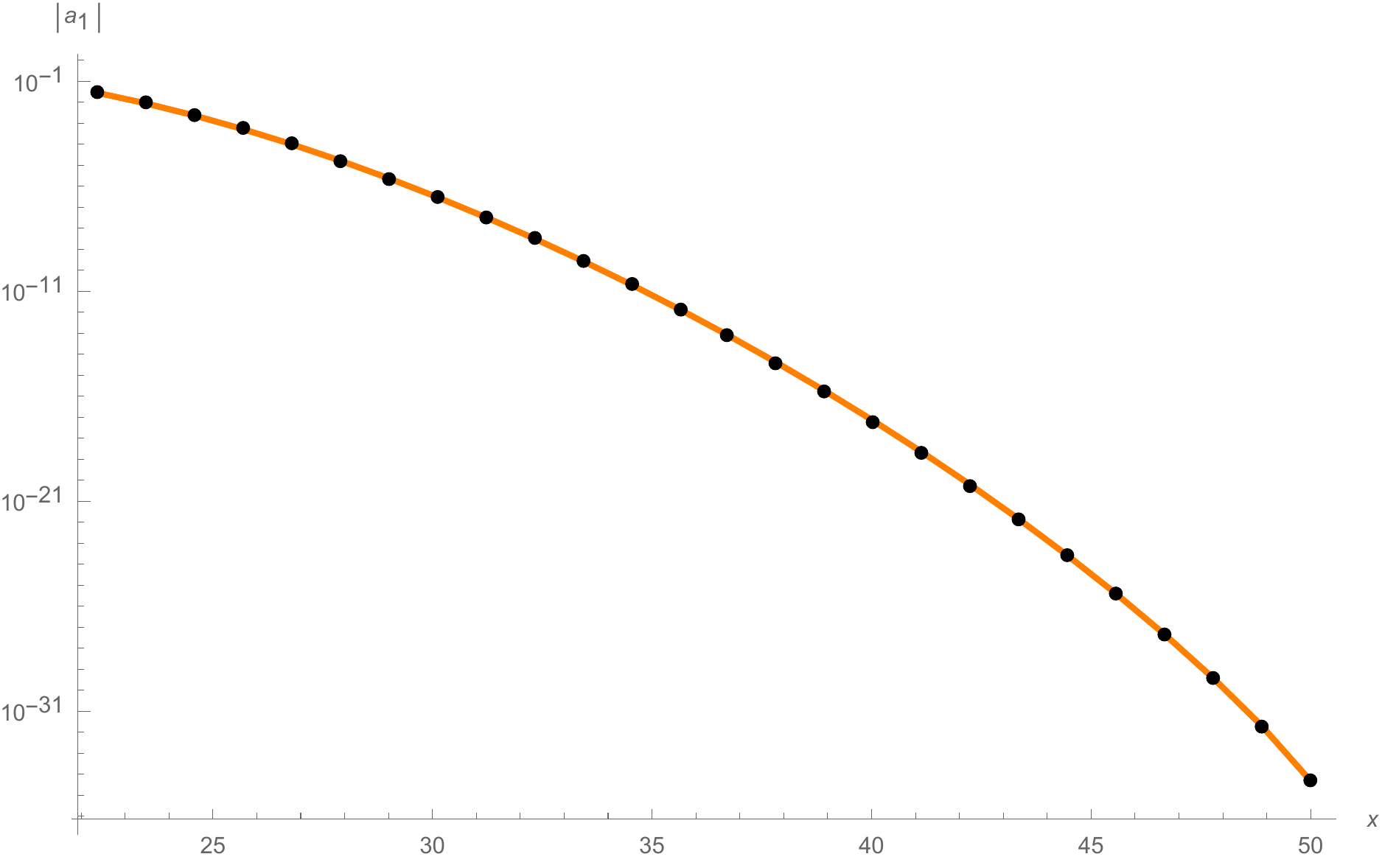}\\
\includegraphics[width=0.3\textwidth]{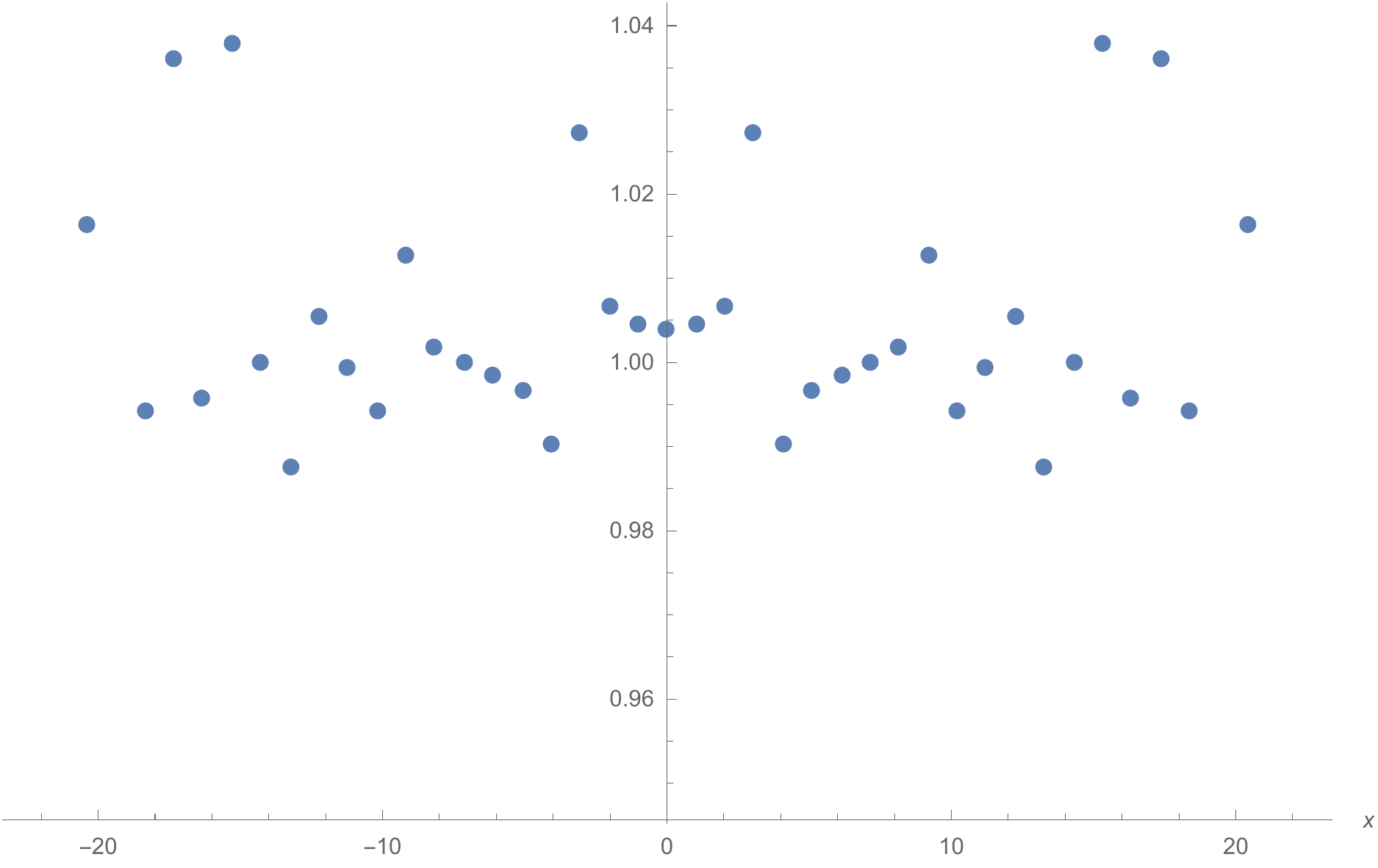}
&
\includegraphics[width=0.3\textwidth]{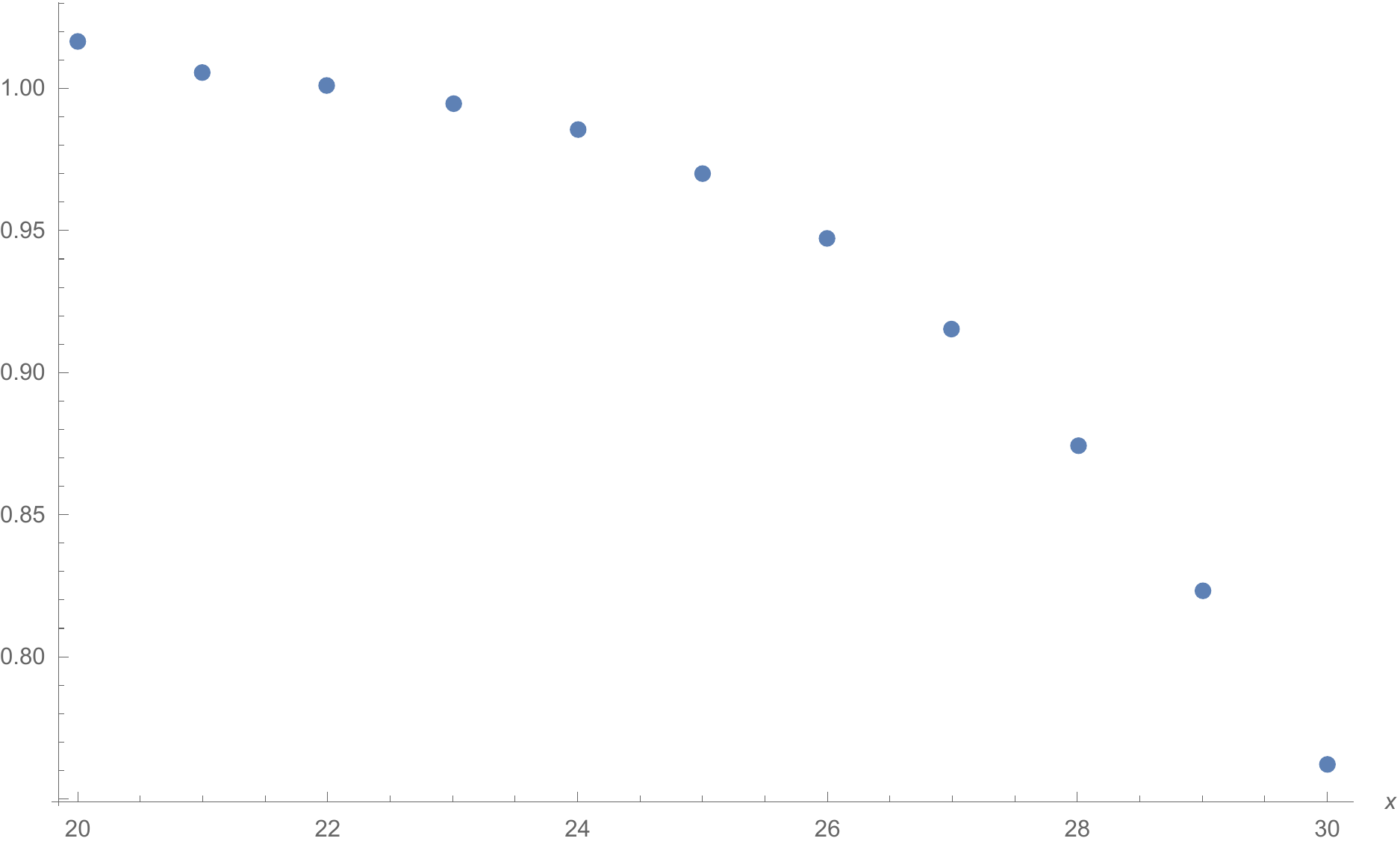}
&
\includegraphics[width=0.3\textwidth]{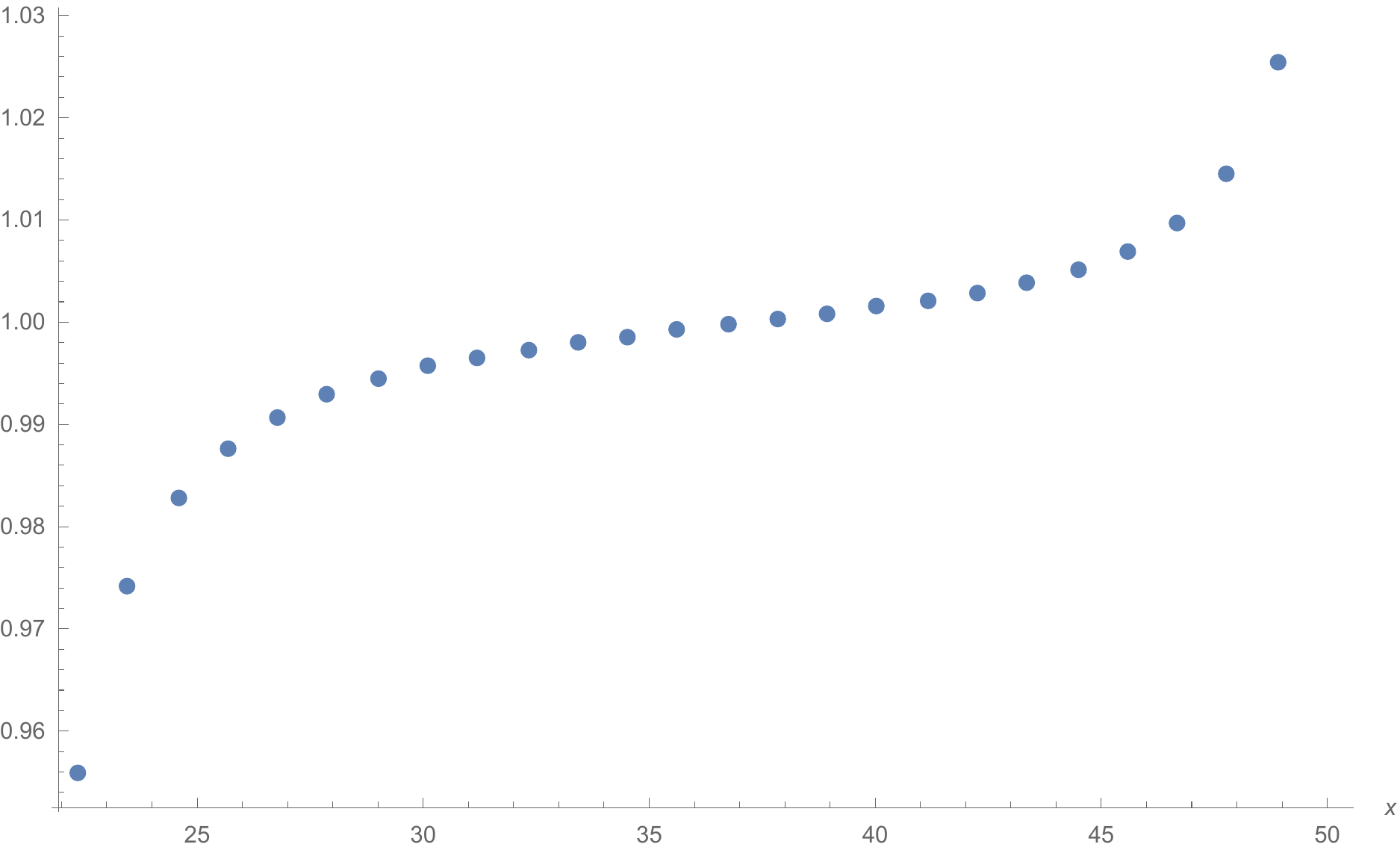}
\end{tabular}
  \caption{Graphs of $a_1(x,50,4,0.5)$ (top, dots), its analytic approximation (top, curves) given by Theorems~\ref{th-ergenium}
  (left), \ref{th-Airy} (middle), and~\ref{th-outside} (right), their ratio (bottom). The top-right graph depicts the absolute value of the functions in the logarithmic scale}
  \label{fig-approx-error}
\end{figure}

\begin{theorem}[Large-time asymptotic formula between the peaks; see Figure~\ref{fig-approx-error} to the left] \label{th-ergenium}
For each $\delta>0$ there is $C_\delta>0$ such that for each $m,\varepsilon>0$ and each $(x,t)\in\varepsilon\mathbb{Z}^2$ satisfying
\begin{equation}\label{eq-case-A}
 |x|/t<1/\sqrt{1+m^2\varepsilon^2}-\delta, \qquad
 \varepsilon\le 1/m, \qquad t>C_\delta/m,
\end{equation}
we have
\begin{align}\label{eq-ergenium-re}
{a}_1\left(x,t+\varepsilon,m,\varepsilon\right)
&={\varepsilon}\sqrt{\frac{2m}{\pi}}
\left(t^2-(1+m^2\varepsilon^2)x^2\right)^{-1/4}
\sin \theta(x,t,m,\varepsilon)
+O_\delta\left(\frac{\varepsilon}{m^{1/2}t^{3/2}}\right),\\
\label{eq-ergenium-im}
{a}_2\left(x+\varepsilon,t+\varepsilon,m,\varepsilon\right)
&={\varepsilon}\sqrt{\frac{2m}{\pi}}
\left(t^2-(1+m^2\varepsilon^2)x^2\right)^{-1/4}\sqrt{\frac{t+x}{t-x}}
\cos \theta(x,t,m,\varepsilon)
+O_\delta\left(\frac{\varepsilon}{m^{1/2}t^{3/2}}\right),
\intertext{for $(x+t)/\varepsilon$ odd and even respectively, where}
\label{eq-theta}
\theta(x,t,m,\varepsilon)&:=
\frac{t}{\varepsilon}\arcsin
\frac{m\varepsilon t} {\sqrt{\left(1+m^2\varepsilon^2\right)\left(t^2-x^2\right)}}
-\frac{x}{\varepsilon}\arcsin
\frac{m\varepsilon x}{\sqrt{t^2-x^2}}
+\frac{\pi }{4}.
\end{align}
\end{theorem}


Hereafter the notation $f(x,t,m,\varepsilon)
=g(x,t,m,\varepsilon)+{O}_\delta\left(h(x,t,m,\varepsilon)\right)$ means that there is a constant $C(\delta)$ (depending on $\delta$ but \emph{not} on $x,t,m,\varepsilon$) such that for each $x,t,m,\varepsilon,\delta$ satisfying the assumptions of the theorem we have $|f(x,t,m,\varepsilon)-g(x,t,m,\varepsilon)|\le C(\delta)\,h(x,t,m,\varepsilon)$. 

The main terms in Theorem~\ref{th-ergenium} 
were \new{first} computed in~\cite[Theorem~2]{Ambainis-etal-01} in the particular case $m=\varepsilon=1$.
The error terms were estimated in~\cite[Proposition~2.2]{Sunada-Tate-12}; that estimate had the same order in $t$ but was \emph{not} uniform in $m$ and $\varepsilon$ (and \emph{cannot} be uniform as it stands, otherwise one would get a contradiction with Corollary~\ref{cor-uniform} as $\varepsilon\to 0$). Getting uniform estimate~\eqref{eq-ergenium-re}--\eqref{eq-ergenium-im} has required a refined version of \new{the} stationary phase method and additional Steps~3--4 of the proof (see \S\ref{ssec-proofs-main}).  Different approaches (Darboux asymptotic formula for Legendre polynomials and Hardy--Littlewood circle method) were used in \cite[Theorem~3]{Bogdanov-20} and \cite[Theorem~2]{SU-20} to get~\eqref{eq-ergenium-re} in the particular case $x=0$ and a rougher approximation for $x$ close to $0$ \new{respectively}.


Theorem~\ref{th-ergenium} has several interesting corollaries. First, it allows to pass to the large-time distributional limit (see Figure~\ref{fig-correlation}). Compared to Theorem~\ref{th-limiting-distribution}, it provides convergence in a stronger sense, not immediately accessible by the method of moments and by~\cite[Theorem~1.1]{Sunada-Tate-12}.

\begin{corollary}[Large-time limiting 
distribution] 
\label{th-limiting-distribution-mass}
For each $m>0$ and $\varepsilon\le 1/m$ we have
\begin{equation*}
\sum_{\substack{x\le vt\\ x\in\varepsilon\mathbb{Z}}}P(x,t,m,\varepsilon)
\rightrightarrows F(v,m,\varepsilon):=
  \begin{cases}
  0, & \mbox{if } v\le -1/\sqrt{1+m^2\varepsilon^2}; \\
  \frac{1}{\pi}\arccos\frac{1-(1+m^2\varepsilon^2)v}{\sqrt{1+m^2\varepsilon^2}(1-v)},
     & \mbox{if } |v|<1/\sqrt{1+m^2\varepsilon^2}; \\
  1, & \mbox{if }v\ge 1/\sqrt{1+m^2\varepsilon^2}
\end{cases}
\end{equation*}
as $t\to\infty$, $t\in\varepsilon\mathbb{Z}$, uniformly in $v$.
\end{corollary}

Recall that all our results can be stated in terms of Young diagrams or Jacobi polynomials.

\begin{corollary}[Steps of Young diagrams; see~Figure~\ref{fig-young}] \label{cor-young} Denote by $n_+(h\times w)$ and $n_-(h\times w)$ the number of Young diagrams with exactly $h$ rows and $w$ columns, having an even and an odd number of steps (defined in page~\pageref{page-young}) respectively.
Then for almost every $r>1$ we have
\begin{equation*}
\limsup_{\substack{w\to\infty}}
\frac{\sqrt{w}}{2^{(r+1)w/2}}
\left|n_+(\lceil rw\rceil\times w)-n_-(\lceil rw\rceil\times w)\right|=
\begin{cases}
  \frac{1}{\sqrt{\pi}}(6r-r^2-1)^{-1/4},
  & \mbox{if } r< 3+2\sqrt{2}; \\
  0, & \mbox{if }r> 3+2\sqrt{2}.
\end{cases}
\end{equation*}
\end{corollary}

For the regimes~\textbf{(B)} and~\textbf{(C)} around and outside the peaks, we state two theorems by Sunada--Tate without discussing the proofs in detail. (To be precise,
the main terms in literature~\cite[Theorem~2]{Ambainis-etal-01}, ~\cite[Propositions~2.2, 3.1, 4.1]{Sunada-Tate-12} are slightly different from~\eqref{eq-ergenium-re}, \eqref{eq-Airy-re}, \eqref{eq-outside-re}; but that difference is within the error term. Practically, the latter three approximations are better by several orders of magnitude.)

Regime~\textbf{(B)} around the peaks \new{is described by} the \emph{Airy function}
$$
\mathrm{Ai}(\lambda) := \frac{1}{\pi} \int_{0}^{+\infty} \cos\left(\lambda p + \frac{p^{3}}{3}\right) \,dp.
$$

\begin{theorem}[Large-time asymptotic formula around the peaks; see Figure~\ref{fig-approx-error} to the middle]  \label{th-Airy} \textup{\cite[Proposition~3.1]{Sunada-Tate-12}}
\new{}
For each $m,\varepsilon,\Delta>0$ and each sequence $(x_n,t_n)\in\varepsilon\mathbb{Z}^2$ satisfying
\begin{equation}\label{eq-Airy}
|x_n/t_n-1/\sqrt{1+m^2\varepsilon^2}|<\Delta\, t_n^{-2/3}, \qquad t_n=n\varepsilon,
\end{equation}
we have
\begin{align}\label{eq-Airy-re}
{a}_1\left(\pm x_n,t_n+\varepsilon,m,\varepsilon\right)
&=
\left(-1\right)^{\frac{t_n - x_n - \varepsilon}{2\varepsilon}}  m\varepsilon\left( \frac{2}{m^2\varepsilon t_n} \right)^{1/3}
\mathrm{Ai} \left( \Delta(x_n,t_n,m,\varepsilon) \right)+ O_{m,\varepsilon,\{x_n\}}\left(n^{-2/3}\right),\\
\label{eq-Airy-im}
\hspace{-0.6cm}
a_2\left(\pm x_n+\varepsilon,t_n+\varepsilon,m,\varepsilon \right)
&= \left(-1\right)^{\frac{t_n - x_n}{2\varepsilon}}\left( \sqrt{1+m^2\varepsilon^2} \pm 1 \right)\left( \frac{2}{m^2\varepsilon t_n} \right)^{1/3} \mathrm{Ai} \left(\Delta(x_n,t_n,m,\varepsilon) \right)+ O_{m,\varepsilon,\{x_n\}}\left(n^{-2/3}\right)
\intertext{for $(x_n+t_n)/\varepsilon$ odd and even respectively, where the signs $\pm$ agree and}
\label{eq-Delta}
\Delta(x_n,t_n,m,\varepsilon)&:=\left( \frac{2}{m^2\varepsilon t_n} \right)^{1/3} \frac{\sqrt{1+m^2\varepsilon^2}\,x_n-t_n}{\varepsilon}.
\end{align}
\end{theorem}

We write $f(x_n,t_n,m,\varepsilon)
=O_{m,\varepsilon,\{x_n\}}\left(g(n)\right)$, if there is a constant $C(m,\varepsilon,\{x_n\})$ (depending on $m,\varepsilon$, and the whole sequence $\{x_n\}_{n=1}^\infty$ but \emph{not} on $n$) such that for each $n$ we have $|f(x_n,t_n,m,\varepsilon)|\le C(m,\varepsilon,\{x_n\})g(n)$.

Recently P.~Zakorko has extended \eqref{eq-Airy-re}--\eqref{eq-Airy-im} to a larger neighborhood of the peak \new{}\cite{Zakorko-21}. 

Exponential decay outside the peaks was stated without proof in~\cite[Theorem~1]{Ambainis-etal-01}. A proof appeared only a decade later, when the following asymptotic formula was established.

\begin{theorem}[Large-time asymptotic formula outside the peaks; see Figure~\ref{fig-approx-error} to the right] \label{th-outside} \textup{\cite[Proposition~4.1]{Sunada-Tate-12}}\new{}
For each $m,\varepsilon,\Delta>0$, $v\in (-1;1)$ and each sequence $(x_n,t_n)\in\varepsilon\mathbb{Z}^2$ satisfying
\begin{equation}\label{eq-outside}
 1/\sqrt{1+m^2\varepsilon^2}<|v|<1, \qquad
 |x_n-vt_n|<\Delta, \qquad t_n=n\varepsilon,
\end{equation}
we have\new{}\new{}
\begin{align}\notag
{a}_1\left(x_n,t_n+\varepsilon,m,\varepsilon\right)
&=
{\varepsilon}\sqrt{\frac{m}{2\pi t_n}}
\frac{(-1)^{(t_n-|x_n|-\varepsilon)/2\varepsilon}}{\left((1+m^2\varepsilon^2)v^2-1\right)^{1/4}}\cdot
\\
\label{eq-outside-re}
&\cdot \exp\left(-\frac{t_n}{2\varepsilon}
H\left(\frac{x_n}{t_n},m,\varepsilon\right)\right)
\left(1+O_{m,\varepsilon,\{x_n\}}\left(\frac{1}{n}\right)\right),
\\
\notag
{a}_2\left(x_n+\varepsilon,t_n+\varepsilon,m,\varepsilon\right)
&={\varepsilon}\sqrt{\frac{m}{2\pi t_n}}
\frac{(-1)^{(t_n-|x_n|)/2\varepsilon}}
{\left((1+m^2\varepsilon^2)v^2-1\right)^{1/4}}
\sqrt{\frac{1+v}{1-v}}\cdot
\\
\label{eq-outside-im}
&\cdot \exp\left(-\frac{t_n}{2\varepsilon}
H\left(\frac{x_n}{t_n},m,\varepsilon\right)\right)
\left(1+O_{m,\varepsilon,\{x_n\}}\left(\frac{1}{n}\right)\right)
\intertext{for $(x_n+t_n)/\varepsilon$ odd and even respectively, where}
\label{eq-hermite}
H(v,m,\varepsilon)&:=
-2\,\mathrm{arcosh}\,
\frac{m\varepsilon} {\sqrt{\left(1+m^2\varepsilon^2\right)\left(1-v^2\right)}}
+2|v|\,\mathrm{arcosh}\,
\frac{m\varepsilon |v|}{\sqrt{1-v^2}}.
\end{align}
\end{theorem}

Function~\eqref{eq-hermite} is positive and convex in $v$ in the region $\frac{1}{\sqrt{1+m^2\varepsilon^2}}<|v|<1$ \cite[Theorem~1.4]{Sunada-Tate-12}.

\begin{corollary}[Limiting free energy density; see Figure~\ref{fig-distribution} to the middle] \label{cor-free}
For each $m,\varepsilon>0$, $v\in(-1;1)$, and $H(v,m,\varepsilon)$ given by~\eqref{eq-hermite} we have
\begin{equation*}
\lim_{\substack{t\to\infty\\ t\in 2\varepsilon\mathbb{Z}}}
\frac{1}{t}\log P\left(2\varepsilon\left\lceil \frac{vt}{2\varepsilon}\right\rceil,t,m,\varepsilon\right)
=
\begin{cases}
-H(v,m,\varepsilon)/\varepsilon,
  & \mbox{if } |v|>1/\sqrt{1+m^2\varepsilon^2};\\
0,& \mbox{if } |v|<1/\sqrt{1+m^2\varepsilon^2}.
\end{cases}
\end{equation*}
\end{corollary}

This means phase transition in the Ising model in a strong sense: the limiting free energy density (with the imaginary part proportional to the left side) is nonanalytic in $v$.

\textbf{Feynman triple limit.} Theorem~\ref{th-ergenium} allows to pass to the limit $(1/t,x/t,\varepsilon)\to 0$ as follows.

\begin{corollary}[Simpler and rougher asymptotic formula between the peaks] \label{cor-intermediate-asymptotic-form}
Under the assumptions of Theorem~\ref{th-ergenium} 
 we have
\begin{equation}\label{eq-intermediate-asymptotic-form}
{a}\left(x,t,{m},{\varepsilon}\right)
=\varepsilon\sqrt{\frac{2m}{\pi t}}
\exp\left(-im\sqrt{t^2-x^2}+\frac{i\pi}{4}\right) \left(1+{O}_\delta
\left(\frac{1}{mt}+\frac{|x|}{t}+m^3\varepsilon^2t\right)\right).
\end{equation}
\end{corollary}

\begin{corollary}[Feynman triple limit; see Figure~\ref{fig-triple-limit}] \label{cor-feynman-problem}  
For each $m\ge 0$ and each sequence $(x_n,t_n,\varepsilon_n)$ such that $(x_n,t_n)\in \varepsilon_n\mathbb{Z}^2$, $(x_n+t_n)/\varepsilon_n$ is even, and
\begin{equation}\label{eq-feynman-problem-assumptions}
1/t_n,\quad x_n/t_n^{3/4},\quad \varepsilon_nt_n^{1/2}\quad\to\quad 0
 \qquad \text{as }\quad n\to\infty,
\end{equation}
we have the equivalence
\begin{equation}\label{eq-feynman-problem}
\frac{1}{2i\varepsilon_n}{a}\left(x_n,t_n,{m},{\varepsilon_n}\right)
\sim\sqrt{\frac{m}{2\pi t_n}}
\exp\left(-imt_n-\frac{i\pi}{4}+\frac{imx_n^2}{2t_n}\right)
\qquad\text{ as }\quad n\to\infty.
\end{equation}
\end{corollary}

For equivalence~\eqref{eq-feynman-problem}, assumptions~\eqref{eq-feynman-problem-assumptions} are essential and sharp, as the next example shows.

\begin{example}\label{p-Feynman-couterexample}
Equivalence~\eqref{eq-feynman-problem} holds for
neither $(x_n,t_n,\varepsilon_n)=(n^3,n^4,1/n^4)$ nor $(0,4n^2,1/{2n})$ nor $(0,2n\varepsilon,\varepsilon)$, where $\varepsilon=\mathrm{const}<1/m$: the limit of the ratio of the left and the right side equals $e^{im/8}$, $e^{im^3/3}$, and does not exist respectively, rather than equals~$1$. (The absence of the latter limit confirms that tending $\varepsilon$ to $0$ was understood in the Feynman problem.)
\end{example}

Corollary~\ref{cor-feynman-problem} solves the Feynman problem (and moreover corrects the statement, by revealing the required sharp assumptions). The main difficulty here is that it concerns \emph{triple} rather than \emph{iterated} limit. We are not aware of any approach which could solve the problem without proving the whole Theorem~\ref{th-ergenium}.  E.g., the Darboux asymptotic formula for the Jacobi polynomials (see Remark~\ref{rem-hypergeo}) is suitable for the iterated limit when first $t\to+\infty$, then $\varepsilon\to 0$, giving a (weaker) result already independent on $x$.
Neither the Darboux nor Mehler--Heine nor more recent asymptotic formulae~\cite{Kuijlaars-Martinez-Finkelshtein-04} are applicable when $1/m\varepsilon$ or $x/\varepsilon$ is unbounded. The same concerns \cite[Proposition~2.2]{Sunada-Tate-12} because the error estimate there is not uniform in $m$ and $\varepsilon$.
Conversely, the next theorem is suitable for the iterated limit when first $\varepsilon\to 0$, then $x/t\to 0$, then $t\to+\infty$, but not for the triple limit because the error term blows up as $t\to+\infty$.

%



\textbf{Continuum limit.}
The limit $\varepsilon\to 0$ 
\new{is described by} the \emph{Bessel functions of the first kind}: 
\begin{equation*}
\vspace{-1.0cm}
\end{equation*}
\begin{align*}
J_0(z)&:=\sum_{k=0}^\infty (-1)^k\frac{(z/2)^{2k}}{(k!)^2},
&
J_1(z)&:=\sum_{k=0}^\infty (-1)^k\frac{(z/2)^{2k+1}}{k!(k+1)!}.
\end{align*}

\begin{theorem}[Asymptotic formula in the continuum limit] \label{th-main}
For each $m,\varepsilon,\delta>0$ and $(x,t)\in\varepsilon\mathbb{Z}^2$ such that $(x+t)/\varepsilon$ even, $t-|x|\ge\delta$, and $\varepsilon<
\delta \,e^{-3ms}/16$, where $s:=\sqrt{t^2-x^2}$, we have
$$
{a}\left(x,t,{m},{\varepsilon}\right)
={m}{\varepsilon}\left(J_0(ms)
-i\,\frac{t+x}{s}\,J_1(ms)
+O\left(\frac{\varepsilon}{\delta}\log^2\frac{\delta}{\varepsilon}
\cdot e^{m^2t^2}\right)
\right).
$$
\end{theorem}

Recall that 
$f(x,t,{m},{\varepsilon})={O}\left(g(x,t,{m},{\varepsilon})\right)$ means that there is a constant $C$ (not depending on $x,t,{m},{\varepsilon}$) such that for each $x,t,{m},{\varepsilon}$ satisfying the assumptions of the theorem we have $|f(x,t,{m},{\varepsilon})|\le C\,g(x,t,{m},{\varepsilon})$.


The main term in Theorem~\ref{th-main} was computed in \cite[\S1]{Narlikar-72}. Numerical experiment shows that the error term decreases faster than asserted 
(see Table~\ref{table-error} computed in \cite[\S14]{SU-2}).

In the next corollary, we approximate a fixed point $(x,t)$ in the plane by the 
lattice point $\left(2\varepsilon\!\left\lceil \frac{x}{2\varepsilon}\right\rceil,2\varepsilon\!\left\lceil \frac{t}{2\varepsilon}\right\rceil\right)$ 
(see Figure~\ref{fig-limit}). The factors of $2$ make
the latter accessible for the checker.

\begin{corollary}[Uniform continuum limit; see Figure~\ref{fig-limit}] \label{cor-uniform}  For each fixed $m\ge 0$ we have
\begin{equation}\label{eq-uniform-limit}
\frac{1}{2\varepsilon}\,{a}\left(2\varepsilon\!\left\lceil \frac{x}{2\varepsilon}\right\rceil,2\varepsilon\!\left\lceil \frac{t}{2\varepsilon}\right\rceil,{m},{\varepsilon}\right)
\rightrightarrows \frac{m}{2}\,J_0(m\sqrt{t^2-x^2})
-i\,\frac{m}{2}\sqrt{\frac{t+x}{t-x}}\,J_1(m\sqrt{t^2-x^2})
\end{equation}
as $\varepsilon\to 0$ uniformly on compact subsets of the angle $|x|<t$.
\end{corollary}

The proof of \emph{pointwise} convergence is simpler and is presented in Appendix~\ref{app-pointwise}.

\begin{corollary}[Concentration of measure]\label{cor-concentration} 
For each $t,m,\delta>0$ 
we have
\begin{equation*}
\vspace{-0.8cm}
\end{equation*}
\begin{equation*}
  \sum_{x\in\varepsilon\mathbb{Z}\,:\,
  0\le t-|x|\le \delta} P(x,t,m,\varepsilon)\to 1\qquad\text{as }\quad
  \varepsilon\to 0
  \quad\text{ so that }\quad\frac{t}{2\varepsilon}\in \mathbb{Z}.
\end{equation*}
\end{corollary}

This result, although expected, is not found in literature. An elementary proof is given in~\S\ref{ssec-proofs-continuum}. We remark a sharp contrast between the continuum and the large-time limit here: by Corollary~\ref{th-limiting-distribution-mass}, there is \emph{no} concentration of measure as $t\to\infty$ for fixed $\varepsilon$.



\begin{table}[htb]
  \centering
  \begin{tabular}{ccc}
    \hline
    $\varepsilon$ & $5\varepsilon\log_{10}^2(5\varepsilon)$ &
    $\max\limits_{x\in (-0.8,0.8)\cap2\varepsilon\mathbb{Z}}
    \left|\frac{1}{2\varepsilon}a(x,1,10,\varepsilon)
    -G^R_{11}(x,1)-iG^R_{12}(x,1)\right|$
    \\
    \hline
    0.02 & 0.1 & 1.1 \\
    0.002 & 0.04 & 0.06 \\
    0.0002 & 0.009 &  0.006 \\
    \hline
  \end{tabular}
  \caption{
  Approximation of spin-$1/2$ retarded propagator~\eqref{eq-relativistic-propagator} by Feynman checkers ($m=10$, $\delta=0.2$, $t=1$)}\label{table-error}
\end{table}

\addcontentsline{toc}{myshrinkalt}{}

\subsection{Physical interpretation} \label{ssec-mass-interpretation}

Let us discuss the meaning of the Feynman triple limit and the continuum limit. In this subsection we omit some technical definitions not used in the sequel.

The right side of~\eqref{eq-feynman-problem} is up to the factor $e^{-imt_n}$ the \emph{free-particle kernel}
\begin{equation}\label{eq-free-particle-kernel}
K(x,t)=\sqrt{\frac{m}{2\pi t}}
\exp \left(\frac{i m x^2}{2t}-\frac{i\pi}{4}\right)\new{,}
\end{equation}
describing \new{the} motion of \new{a} non-relativistic particle emitted from the origin along a line.

Limit~\eqref{eq-uniform-limit} reproduces the spin-$1/2$ retarded propagator describing motion of an electron along a line.
More precisely, the \emph{spin-$1/2$ retarded propagator}, or the \emph{retarded Green function for Dirac equation}~\eqref{eq-continuum-Dirac} is a matrix-valued tempered distribution $G^R(x,t)=(G^R_{kl}(x,t))$ on $\mathbb{R}^2$ vanishing for $t<|x|$ and satisfying
\begin{equation}\label{eq-Green}
\begin{pmatrix}
m  & \partial/\partial x-\partial/\partial t \\
\partial/\partial x+\partial/\partial t & m
\end{pmatrix}
\begin{pmatrix}
G^R_{11}(x,t) & G^R_{12}(x,t) \\ G^R_{21}(x,t) & G^R_{22}(x,t)
\end{pmatrix}=\delta(x)\delta(t)\begin{pmatrix}
1 & 0 \\ 0 & 1
\end{pmatrix},
\end{equation}
where $\delta(x)$ is the Dirac delta function. 
The propagator is given by  (cf.~\cite[(13)]{Jacobson-Schulman-84}, \cite[(3.117)]{Peskin-Schroeder})
\begin{equation}\label{eq-relativistic-propagator}
G^R(x,t)=\frac{m}{2}\,
\begin{pmatrix}
J_0(m\sqrt{t^2-x^2}) &
-\sqrt{\frac{t+x}{t-x}}\,J_1(m\sqrt{t^2-x^2}) \\ \sqrt{\frac{t-x}{t+x}}\,J_1(m\sqrt{t^2-x^2}) &
J_0(m\sqrt{t^2-x^2})
\end{pmatrix}\qquad\text{for }|x|<t.
\end{equation}
In addition, $G^R(x,t)$ involves a generalized function supported on the lines $t=\pm x$, not observed in the limit~\eqref{eq-uniform-limit} and not specified here.
A more common expression is 
(cf.~Proposition~\ref{cor-double-fourier})
\begin{equation}\label{eq-double-fourier-retarded}
G^R(x,t)=
  \frac{1}{4\pi^2}
  \int_{-\infty}^{+\infty}
  \int_{-\infty}^{+\infty}
  \lim_{\delta\to+0}
  \begin{pmatrix}
  m & -ip-i\omega \\
  -ip+i\omega & m
  \end{pmatrix}
  \frac{ e^{i p x-i\omega t}\,dpd\omega}
  {m^2+p^2-(\omega+i\delta)^2},
\end{equation}
where the limit is taken in the weak topology of matrix-valued tempered distributions and the integral is understood as the Fourier transform of tempered distributions (cf.~\cite[(6.47)]{Folland}).

The propagator ``square'' $G^R_{11}(x,t)^2+G^R_{12}(x,t)^2$ is ill-defined (because of the square of the Dirac delta-function supported on the lines $t=\pm x$ involved). Thus the propagator lacks probabilistic interpretation, and global charge conservation (Proposition~\ref{p-mass2})
has no continuum analogue. For instance,
$\int_{(-t,t)} (G^R_{11}(x,t)^2+G^R_{12}(x,t)^2)\,dx\new{\approx m^2}t/2\ne\mathrm{const}$ paradoxically. A physical explanation: the line $t=x$ carries infinite charge flowing inside the angle $|x|<t$. One can interpret the propagator ``square'' for $|x|\ne t$ as a relative probability density or charge density (see Figure~\ref{P-contour} \new{to the right}). In the spin-chain interpretation, the propagator is the limit of the partition function for one-dimensional Ising model at the inverse temperature $\beta=i\pi/4-\log(m\varepsilon)/2$. Those are essentially the values of $\beta$ for which phase transition is possible~\cite{Matveev-Shrock-97}.

The normalization factor $1/2\varepsilon$ before ``${a}$'' in \eqref{eq-uniform-limit}
can be explained as division by the length associated to one black lattice point on a horizontal line.
On a deeper level,
it comes from the normalization of~$G^R(x,t)$ arising from~\eqref{eq-Green}.

Theorem~\ref{th-main} is a toy result in \emph{algorithmic quantum field theory}: it determines the lattice step to compute the propagator with given accuracy. So far this is not a big deal, because the propagator has a known analytic expression and is not really experimentally-measurable; neither the efficiency of the algorithm is taken into account. But that is a first step.

\begin{algorithm}[Approximation algorithm for  
spin-$1/2$ retarded propagator~\eqref{eq-relativistic-propagator}] \label{alg-main}

Input: mass $m>0$, coordinates $|x|<t$, accuracy level $\Delta$.

Output: an approximate value $G_{kl}$ of $G_{kl}^R(x,t)$ within distance $\Delta$ from the true value~\eqref{eq-relativistic-propagator}.

Algorithm: compute $G_{kl}=\tfrac{(-1)^{(k-1)l}}{2\varepsilon}
\,{a}_{(k+l)\!\!\mod2+1}\left(2\varepsilon\!\left\lceil \tfrac{(-1)^{(k-1)l}x}{2\varepsilon}\right\rceil,2\varepsilon\!\left\lceil \tfrac{t}{2\varepsilon}\right\rceil,{m},{\varepsilon}\right)$
by~\eqref{eq-def-mass} for
$$
\varepsilon=(t-|x|)\min\left\{
\frac{1}{16\, e^{3mt}},\left(\frac{\Delta}{9C\, m\,e^{m^2t^2}}\right)^3
\right\},\qquad\text{where }\quad C=100.
$$
\end{algorithm}

Here we used an explicit estimate for the constant $C$ understood in the big-O notation in Theorem~\ref{th-main}; it is easily extracted from the proof. The theorem and the estimate remain true, if ${a}\left(x,t,{m},{\varepsilon}\right)$ with $(x,t)\in\varepsilon\mathbb{Z}^2$ is replaced by
${a}\left(2\varepsilon\!\left\lceil \tfrac{x}{2\varepsilon}\right\rceil,2\varepsilon\!\left\lceil \tfrac{t}{2\varepsilon}\right\rceil,{m},{\varepsilon}\right)$ with arbitrary $(x,t)\in\mathbb{R}^2$.




\begin{remark}\label{rem-qw} A general (homogeneous) \emph{one-dimensional quantum walk} essentially reduces to the upgrade studied in this section. Namely, consider the equation
\begin{align*}
\psi_1(x,t+1)&=\hphantom{-}a\,\psi_1(x+1,t)+b \,\psi_2(x+1,t), \\
\psi_2(x,t+1)&=-\bar b\,\psi_1(x-1,t)+\bar a \,\psi_2(x-1,t);
\end{align*}
where $a,b\in\mathbb{C}-\{0\}$ and $|a|^2+|b|^2=1$, with the initial condition $\psi_1(x,0)=\psi_2(x,0)=0$ for each $x\ne 0$ \cite[\S1]{Sunada-Tate-12}. By direct checking using \eqref{eq-Dirac-mass1}--\eqref{eq-Dirac-mass2} one finds the solution:
\begin{align*}
  \psi_1(x,t)&=
  \left(\frac{a}{|a|}\right)^{-x}
  \left(
  \frac{|a|b}{a|b|}\,a_1\left(x+1,t+1,\left|\frac{b}{a}\right|,1\right)\psi_2(0,0)
  +a_2\left(1-x,t+1,\left|\frac{b}{a}\right|,1\right)\psi_1(0,0)
  \right),\\
  \psi_2(x,t)&=
  \left(\frac{a}{|a|}\right)^{-x}
  \left(
  a_2\left(x+1,t+1,\left|\frac{b}{a}\right|,1\right)\psi_2(0,0)
  -\frac{a|b|}{|a|b}\,a_1\left(1-x,t+1,\left|\frac{b}{a}\right|,1\right)\psi_1(0,0)
  \right).
\end{align*}
\end{remark}

\addcontentsline{toc}{myshrink}{}

\section{Spin} \red{\label{sec-spin}}

\addcontentsline{toc}{myshrink}{}

{
\hrule
\footnotesize
\noindent\textbf{Question:} what is the probability to find a right electron at $(x,t)$, if a right electron was emitted from $(0,0)$?

\noindent\textbf{Assumptions:} electron chirality is now taken into account.

\noindent\textbf{Results:} the probability of chirality flip.
\hrule
}
\bigskip


A feature of the model is that the electron \emph{spin} emerges naturally rather than is added artificially.

It goes almost without saying to view the electron as being in one of the two states depending on the last-move direction: \emph{right-moving} or \emph{left-moving} (or just `\emph{right}' or `\emph{left}' for brevity).

The \emph{probability to find a right 
electron in the square $(x,t)$, if a right electron was emitted from the square $(0,0)$}, is the length square of the vector $\sum_s {a}(s)$, where the sum is over only those paths from $(0,0)$ to $(x,t)$, which both start and finish with an upwards-right move. 
The \emph{probability to find a left electron} is defined analogously, only the sum is taken over paths which start with an upwards-right move but finish with an upwards-left move. Clearly, these probabilities equal $a_2(x,t)^2$ and $a_1(x,t)^2$ respectively, because the last move is directed upwards-right if and only if the number of turns is even.

These right and left electrons are exactly the $(1+1)$-dimensional analogue of \emph{chirality} states for a spin $1/2$ particle \cite[\S19.1]{Peskin-Schroeder}. Indeed, it is known that the components $a_2(x,t)$ and $a_1(x,t)$ in Dirac equation in the Weyl basis~\eqref{eq-continuum-Dirac} are interpreted as wave functions of right- and left-handed particles respectively. The relation to the movement direction
becomes transparent for $m=0$: \new{then} a general solution of~\eqref{eq-continuum-Dirac} is $(a_2(x,t),a_1(x,t))=(a_2(x-t,0),a_1(x+t,0))$; thus the maxima of $a_2(x,t)$ and $a_1(x,t)$ (if any) move to the right and to the left respectively as $t$ increases. 
Beware that in $3$ or more dimensions, spin is \emph{not} the movement direction
and cannot be explained in nonquantum terms.

This gives a more conceptual interpretation of the model:
an experiment outcome is a pair (final $x$-coordinate, last-move direction), whereas the final $t$-coordinate is fixed.
The probabilities to find a right/left electron 
are the fundamental ones. In further upgrades, $a_1(x,t)$ and $a_2(x,t)$ become complex numbers and $P(x,t)$ should be \emph{defined} as 
$|{a}_1(x,t)|^2+|{a}_2(x,t)|^2$
rather than by the above formula $P(x,t)=|{a}(x,t)|^2=|{a}_1(x,t)+i{a}_2(x,t)|^2$, being a coincidence.

\begin{theorem}[Probability of chirality flip] \label{p-right-prob}
For integer $t>0$ we get
$\sum\limits_{x\in\mathbb{Z}}a_{1}(x,t)^2
=\frac{1}{2\sqrt{2}}+{O}\left(\frac{1}{\sqrt{t}}\right).$
\end{theorem}

See Figure~\ref{spin-reversal} for an illustration and comparison with the upgrade from~\S\ref{sec-external field}. The physical interpretation of the theorem is limited: in continuum theory, the probability of chirality flip (for an electron emitted by a point source) is ill-defined similarly to the propagator ``square'' (see \S\ref{ssec-mass-interpretation}). A related more reasonable quantity is studied in~\cite[p.~381]{Jacobson-Schulman-84} (cf.~Problem~\ref{p-correlation2}).
Recently I.~Bogdanov has generalized the theorem to an arbitrary mass and lattice step (see Definition~\ref{def-mass}): if $0\le m \varepsilon\le 1$ then
$\lim_{t\to+\infty,t\in\varepsilon\mathbb{Z}} \sum_{x\in\varepsilon\mathbb{Z}} a_1(x, t, m,\varepsilon)^2 = \frac{m\varepsilon}{2\sqrt{1+m^2\varepsilon^2}}$ \cite[Theorem~2]{Bogdanov-20}.
This has confirmed a conjecture by I.~Gaidai-Turlov--T.~Kovalev--A.~Lvov.



\begin{figure}[htbp]
\begin{center}
\includegraphics[width=0.3\textwidth]{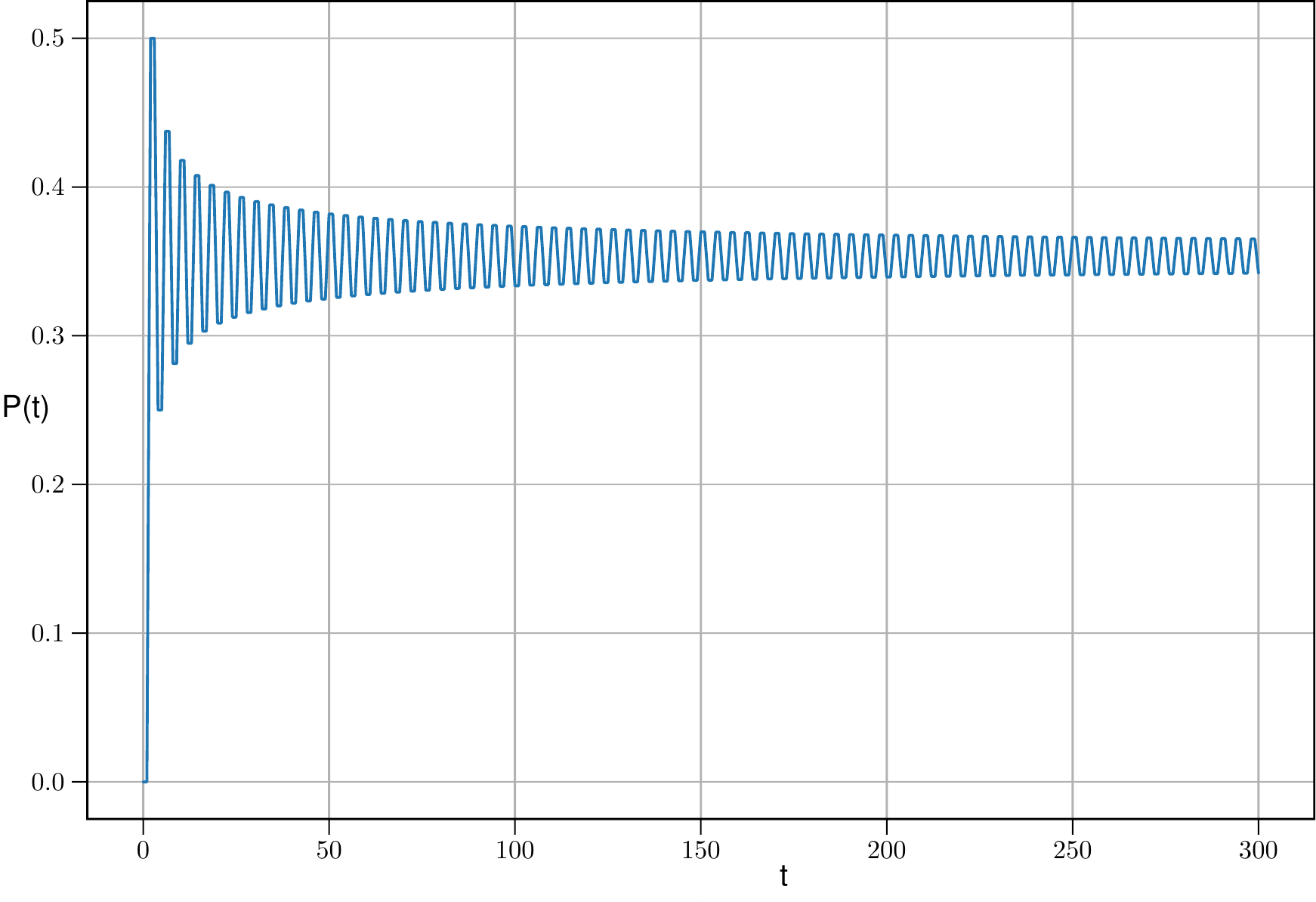}
\includegraphics[width=0.3\textwidth]{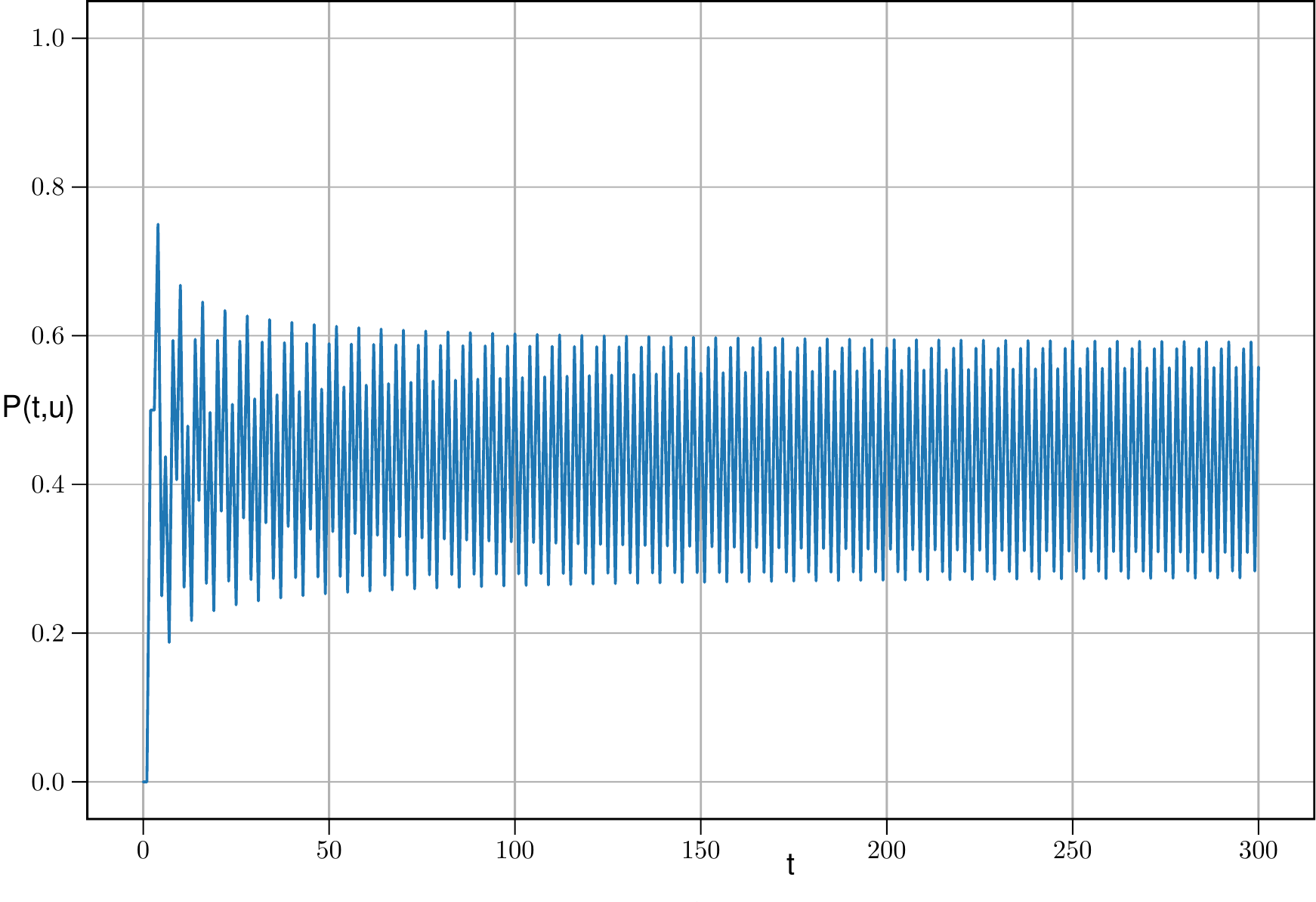}
\end{center}
\vspace{-0.8cm}
\caption{
The graphs of the probabilities $P(t)=\sum_{x\in\mathbb{Z}}a_{1}(x,t)^2$ and $P(t,u)=\sum_{x\in\mathbb{Z}}a_{1}(x,t,u)^2$
of chirality flip
with electromagnetic field off and on respectively}
\label{spin-reversal}
\end{figure}



\section{External field (inhomogeneous quantum walk)}\label{sec-external field}

\addcontentsline{toc}{myshrink}{}

{
\hrule
\footnotesize
\noindent\textbf{Question:} what is the probability to find {an electron} at $(x,t)$, if it moves in a given electromagnetic field $u$?

\noindent\textbf{Assumptions:} the electromagnetic field vanishes outside the $xt$-plane; it is not affected by the electron.

\noindent\textbf{Results:} ``spin precession'' in a magnetic field (qualitative explanation), charge conservation.
\hrule
}

\bigskip

{
\begin{wrapfigure}{r}{6.2cm}
\vspace{-0.4cm}
\includegraphics[width=6cm]{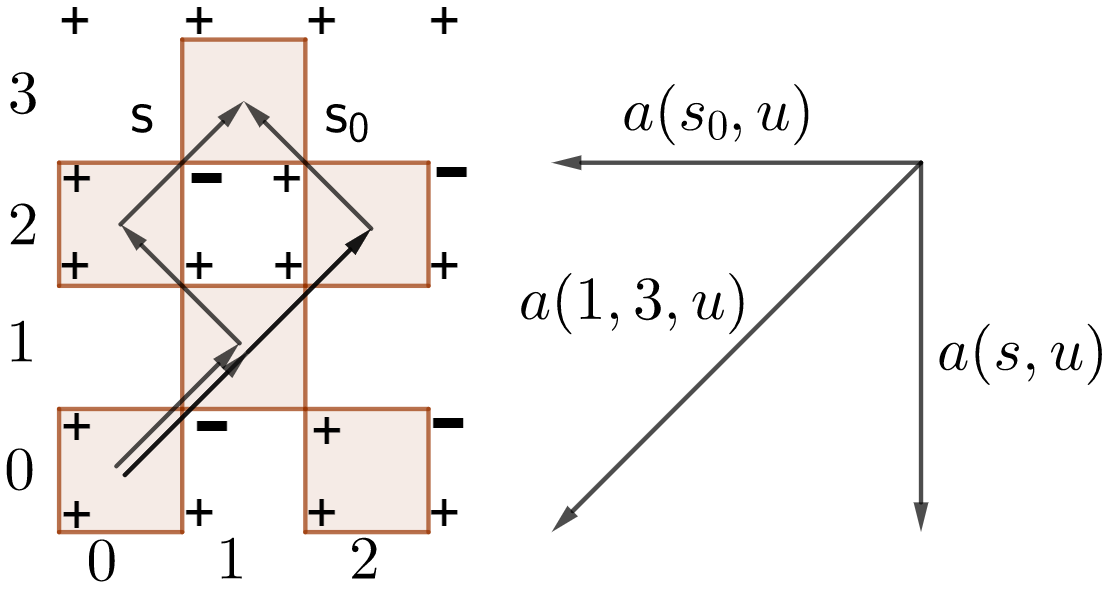}
\vspace{-0.2cm}
\caption{Paths in a field\new{}}
\label{homogeneous-field}
\vspace{-0.2cm}
\end{wrapfigure}
Another feature of the model is that external electromagnetic field emerges naturally rather than is added artificially. We start with an informal  definition, then give a precise one, and finally show exact charge conservation.

In the basic model, the stopwatch hand did not rotate while the checker moved straightly. It goes without saying to modify the model, rotating the hand uniformly during the motion. This does not change the model essentially: since all the paths from the initial to the final position have the same length, their vectors are rotated through the same angle, not affecting probabilities. 
A more interesting modification is when the current rotation angle depends on the checker position. This is exactly what electromagnetic field does. In what follows, the rotation angle assumes only the two values $0^\circ$ and $180^\circ$ for simplicity, meaning just multiplication by $\pm 1$.

}

Thus an electromagnetic field is viewed as a fixed assignment $u$ of numbers $+1$ and $-1$ to all the vertices of the squares.
For instance, in Figure~\ref{homogeneous-field}, the field equals $-1$ at the top-right vertex of each square $(x,t)$ with both $x$ and $t$ even.
Modify the definition of the vector ${a}(s)$ by reversing the direction each time when the checker passes through a vertex with the field $-1$. Denote by ${a}(s,u)$ the resulting vector. Define ${a}(x,t,u)$ and $P(x,t,u)$ analogously to ${a}(x,t)$ and $P(x,t)$ replacing ${a}(s)$ by ${a}(s,u)$ in the definition. For instance, if $u=+1$ identically, then $P(x,t,u)=P(x,t)$.

Let us slightly rephrase this construction, making the relation to lattice gauge theory more transparent. We introduce an auxiliary grid with the vertices at the centers of black squares (see Figure~\ref{fig-young} to the right). It is the graph where the checker actually moves. 

\begin{definition} \label{def-external}
An \emph{auxiliary edge} is a segment joining nearest-neighbor integer points with even sum of coordinates. 
Let $u$ be a map from the set of all auxiliary edges to $\{+1,-1\}$. Denote by
$$
{a}(x,t,u):=2^{(1-t)/2}\,i\,\sum_s (-i)^{\mathrm{turns}(s)}u(s_0s_1)u(s_1s_2)\dots u(s_{t-1}s_t)
$$
the sum over all checker paths $s=(s_0,s_1,\dots,s_t)$ with $s_0=(0,0)$, $s_1=(1,1)$, and $s_t=(x,t)$.
Set $P(x,t,u):=|{a}(x,t,u)|^2$.
 {Define $a_1(x,t,u)$ and $a_2(x,t,u)$ analogously to $a(x,t,u)$, only add the condition $s_{t-1}=(x+1,t-1)$ and $s_{t-1}=(x-1,t-1)$ respectively.}
For half-integers $x,t$ denote by $u(x,t)$ the value of $u$ on the auxiliary edge with the midpoint $(x,t)$.
\end{definition}


\begin{remark}\label{rem-external}
 Here  {the field $u$} is a fixed external classical field not affected by the electron.

 This definition is analogous to one of the first constructions of gauge theory by Weyl--Fock--London, and gives a coupling of Feynman checkers to the Wegner--Wilson {$\mathbb{Z}/2\mathbb{Z}$} lattice gauge theory. In particular, it reproduces the correct spin $1$ for the electromagnetic field: a function defined on the set of edges is a discrete analogue of a \emph{vector} field, i.e., a \emph{spin} $1$ field. Although this way of coupling is classical, it has never been  {explicitly} applied to Feynman checkers {(cf.~\cite[p.~36]{Schulman-Gaveau-89})}, and is very different from both the approach of \cite{Ord} and Feynman-diagram intuition~\cite{Feynman}. 

 For instance, the field $u$ in Figure~\ref{homogeneous-field} has form $u(s_1s_2)=\exp\left(-i\int_{s_1}^{s_2}(A_0\,dt+A_1\,dx)\right)$ for each auxiliary edge $s_1s_2$, where $(A_0,A_1):=(\pi/2)(x+1/2,x+1/2)$ is the vector-potential of a constant homogeneous electromagnetic field.

  {
 For an arbitrary \emph{gauge group},
 $a_1(x,t,u)$ and $a_2(x,t,u)$ are defined analogously, only $u$ becomes a map from the set of auxiliary edges to a matrix group, e.g., $U(1)$ or $SU(n)$. Then  $P(x,t,u):=\sum_{k}\left(|(a_1(x,t,u))_{k1}|^2
 +|(a_2(x,t,u))_{k1}|^2\right)$, where $(a_j)_{kl}$ are the entries of a matrix~$a_j$.}
\end{remark}

\begin{example}[``Spin precession'' in a magnetic field] \label{p-precession}
  Let $u(x+1/2,t+1/2)=-1$, if $x$ and $t$ \new{are} even, and $u(x+1/2,t+1/2)=+1$ otherwise
  (``\emph{constant homogeneous field}'';  {see Figure~\ref{homogeneous-field}}). 
  Then the probability {$P(t,u):=\sum_{x\in\mathbb{Z}}a_1(x,t,u)^2$} of detecting a {left} electron (see~\S\ref{sec-spin}) is plotted in Figure~\ref{spin-reversal} to the right. It apparently
  tends to a ``periodic regime''  {as} $t\to\infty$ (see Problem~\ref{p-precession-prove}).
\end{example}

The following propositions are proved analogously to Propositions~\ref{p-mass}--\ref{p-mass2}, only a factor of $u(x\pm\tfrac{1}{2},t+\tfrac{1}{2})$ is added due to the last step of the path passing through  {the vertex} $(x\pm\tfrac{1}{2},t+\tfrac{1}{2})$.

\begin{proposition}[Dirac equation in electromagnetic field] \label{p-Dirac-external}  For each integers $x$ and $t\ge 1$,  
\begin{align*}
a_1(x,t+1,u) &= \frac{1}{\sqrt{2}} u\left(x+\frac{1}{2},t+\frac{1}{2}\right)(a_1(x+1,t,u) + a_2(x+1,t,u)),\\
a_2(x,t+1,u) &= \frac{1}{\sqrt{2}} u\left(x-\frac{1}{2},t+\frac{1}{2}\right)( {a_2(x-1,t,u)}-a_1(x-1,t,u)).
\end{align*}
\end{proposition}

\begin{proposition}[Probability/charge conservation] \label{p-probability-conservation-external} For each integer $t\ge 1$, 
$\sum\limits_{x\in\mathbb{Z}}P(x,t,u)=1$.
\end{proposition}

Recently F.~Ozhegov has found analogues of the ``explicit'' formula (Proposition~\ref{Feynman-binom-alt}) and the continuum limit (Corollary~\ref{cor-uniform}) for the ``{homogeneous field}'' in Example~\ref{p-precession} \cite[Theorems~3--4]{Ozhegov-21}.


%
%
%
%

\section{Source}\label{sec-source}

\addcontentsline{toc}{myshrink}{}

{
\hrule
\footnotesize
\noindent\textbf{Question:} what is the probability to find an electron at $(x,t)$, if it was emitted by a source of wavelength~$\lambda$?

\noindent\textbf{Assumptions:} the source is now realistic.

\noindent\textbf{Results:}
wave propagation, dispersion relation.
\hrule
}

\bigskip

A realistic source produces a wave rather than electrons localized at $x=0$ (as in the basic model). This means solving Dirac equation~\eqref{eq-Dirac-mass1}--\eqref{eq-Dirac-mass2} with \emph{(quasi-)periodic initial conditions}. 

To state the result, it is convenient to rewrite Dirac equation~\eqref{eq-Dirac-mass1}--\eqref{eq-Dirac-mass2} using the notation
$$
\tilde a_1(x,t)=a_1(x,t+\varepsilon,m, \varepsilon), \qquad
\tilde a_2(x,t)=a_2(x+\varepsilon,t+\varepsilon,m, \varepsilon),
$$
so that it gets 
form 
\begin{align}\label{eq-Dirac-source1}
\tilde a_1(x,t) &= \frac{1}{\sqrt{1+m^2\varepsilon^2}}
(\tilde a_1(x+\varepsilon,t-\varepsilon)
+ m \varepsilon\, \tilde a_2(x,t-\varepsilon)),\\
\label{eq-Dirac-source2}
\tilde a_2(x,t) &= \frac{1}{\sqrt{1+m^2\varepsilon^2}}
(\tilde a_2(x-\varepsilon,t-\varepsilon)
- m \varepsilon\, \tilde a_1(x,t-\varepsilon)).
\end{align}
The following proposition is proved by direct checking (available in \cite[\S12]{SU-2}).

\begin{proposition}[Wave propagation, dispersion relation]
Equations~\eqref{eq-Dirac-source1}--\eqref{eq-Dirac-source2} with the initial condition
\begin{align*}
&\\[-1.0cm]
 \tilde a_1(x,0)&= \tilde a_1(0,0)e^{2\pi ix/\lambda},\\
 \tilde a_2(x,0)&= \tilde a_2(0,0)e^{2\pi ix/\lambda};\\[-1.0cm]
\end{align*}
have the unique solution
\begin{align}\label{eq-solution1}
\tilde a_1(x,t)&= \hphantom{i}a\cos\tfrac{\alpha}{2}\, e^{2\pi i(x/\lambda+t/T)} +
\hphantom{i}b\sin\tfrac{\alpha}{2}\,e^{2\pi i(x/\lambda-t/T)},\\
         \label{eq-solution2}
\tilde a_2(x,t)&= i a\sin\tfrac{\alpha}{2}\, e^{2\pi i(x/\lambda+t/T)}
-ib\cos\tfrac{\alpha}{2}\, e^{2\pi i(x/\lambda-t/T)},
\end{align}
where the numbers $T\ge 2$, $\alpha\in [0,\pi]$, and $a,b\in\mathbb{C}$ are given by
\begin{align*}
\cos (2\pi\varepsilon/T)
&=\frac{\cos (2\pi\varepsilon/\lambda)}{\sqrt{1+m^2\varepsilon^2}},&
\cot \alpha
&=\frac{\sin (2\pi\varepsilon/\lambda)}{m\varepsilon}, &
\begin{aligned}
a&= \tilde a_1(0,0)\cos\tfrac{\alpha}{2}-i\tilde a_2(0,0)\sin\tfrac{\alpha}{2},\\
b&= \tilde a_1(0,0)\sin\tfrac{\alpha}{2}+i\tilde a_2(0,0)\cos\tfrac{\alpha}{2}.
\end{aligned}
\end{align*}
\end{proposition}

\begin{remark}
The solution of continuum Dirac equation~\eqref{eq-continuum-Dirac} is given by the same expression~\eqref{eq-solution1}--\eqref{eq-solution2},
only $2\pi/T$ and $\alpha$ are redefined by $\tfrac{4\pi^2}{T^2}=\tfrac{4\pi^2}{\lambda^2}+m^2$ and $\cot\alpha=2\pi/m\lambda$ instead. In both continuum and discrete setup, these are the hypotenuse and the angle in a right triangle with one leg $2\pi/\lambda$ and another leg either $m$ or $(\arctan m\varepsilon)/\varepsilon$ respectively, lying in the plane or a sphere of radius $1/\varepsilon$ respectively. This spherical-geometry interpretation is new and totally unexpected.

\new{Physically, \eqref{eq-solution1}--\eqref{eq-solution2} describe a wave with the period $T$ and the wavelength $\lambda$. The relation between $T$ and $\lambda$ is called the \emph{dispersion relation}. \emph{Plank} and \emph{de Broglie relations} assert that $E:=2\pi \hbar/T$ and $p:=2\pi \hbar/\lambda$ are the energy and the momentum of the wave (recall that $\hbar=c=1$ in our units). The above dispersion relation tends to \emph{Einstein formula} $E=mc^2$ as $\varepsilon\to 0$ and $\lambda\to\infty$.
}

A comment for specialists: replacing $a$ and $b$ by creation and annihilation operators, i.e., the second quantization of the lattice Dirac equation, leads to the model from~\S\ref{sec-creation}.
\end{remark}

For the next upgrades, we just announce results to be discussed in subsequent publications.

\section{Medium}\label{sec-medium}

\addcontentsline{toc}{myshrink}{}

{
\hrule

\footnotesize
\noindent\textbf{Question:} which part of light of given color is reflected from a glass plate of given width?

\noindent\textbf{Assumptions:} right angle of incidence, no polarization of light; mass now depends on $x$ but not on the color.

\noindent\textbf{Results:} thin-film reflection (quantitative explanation). 

\hrule
}

\bigskip

Feynman checkers can be applied to describe propagation of light in transparent media such as glass. Light propagates as if it had acquired some nonzero mass plus potential energy (depending on the refractive index) inside the media, while both remain zero outside. 
In general the model is inappropriate to describe light; 
partial reflection is a remarkable exception.
Notice that similar classical phenomena are described by quantum models~\cite[\S2.7]{Venegas-Andraca-12}.

In Feynman checkers, we announce a rigorous derivation 
of the following well-known formula for the fraction $P$ of light of wavelength $\lambda$ reflected from a transparent plate of width $L$ and refractive index~$n$:
$$
P=\frac{(n^2-1)^2}{(n^2+1)^2+4n^2\cot^2(2\pi Ln/\lambda)}.
$$
This makes Feynman's popular-science discussion of partial reflection \cite{Feynman} completely rigorous and shows that his model has experimentally-confirmed predictions in the real world, not just a $2$-dimensional one.

\section{Identical particles} \label{sec-IdParAnt}

\addcontentsline{toc}{myshrink}{}

{
\hrule
\footnotesize
\noindent\textbf{Question:} what is the probability to find
electrons 
at $F$ and $F'$,
if they were emitted from $A$ and~$A'$?

\noindent\textbf{Assumptions:} there are several moving electrons.

\noindent\textbf{Results:} exclusion principle, locality, charge conservation.
\hrule
}

\bigskip

We announce a simple-to-define upgrade describing the motion of several electrons, respecting exclusion principle, locality, and probability conservation (cf.~\cite[\S4.2]{Yepez-05}).

\begin{definition}\label{def-identical}
Fix integer points $A=(0,0)$, $A'=(x_0,0)$, \new{}$F=(x,t)$, $F'=(x',t)$\new{} and their diagonal neighbors $B=(1,1)$, $B'=(x_0+1,1)$,\new{}
$E=(x-1,t-1)$, \new{}$E'=(x'-1,t-1)$, where $x_0\ne 0$ and $x'\ge x$.
Denote
$$
{a}(AB,A'B'\to EF,E'F'):=
\sum_{\substack{s:AB\to EF\\s':A'B'\to E'F'}} {a}(s){a}(s')-
\sum_{\substack{s:AB\to E'F'\\s':A'B'\to EF}} {a}(s){a}(s'),
$$
where the first sum is over all pairs consisting of a checker path $s$ starting with the move $AB$ and ending with the move $EF$, and a path $s'$ starting with the move $A'B'$ and ending with the move $E'F'$, whereas in the second sum the final moves are interchanged.

The length square $P(AB,A'B'\to EF,E'F'):=\left|{a}(AB,A'B'\to EF,E'F')\right|^2$ is called the \emph{probability 
to find right electrons at $F$ and $F'$, if they are emitted from $A$ and~$A'$}.
Define $P(AB,A'B'\to EF,E'F')$ analogously also for \new{} $E=(x\pm 1,t-1)$, \new{}$E'=(x'\pm 1,t-1)$. Here we require $x'\ge x$, if both signs are the same, and allow arbitrary $x'$ and $x$, otherwise.
\end{definition}




\section{Antiparticles} \label{sec-creation}


{
\hrule
\footnotesize
\noindent\textbf{Question:} what is the expected charge in the square $(x,t)$, if an electron was emitted from the square $(0,0)$?

\noindent\textbf{Assumptions:} electron-positron pairs \new{are} now created and annihilated, the $t$-axis is time.

\noindent\textbf{Results:} 
spin-$1/2$ Feynman propagator in the continuum limit, an analytic expression for the large-time limit.
\hrule
}

\bigskip





\addcontentsline{toc}{myshrinkalt}{}

\subsection{Identities and asymptotic formulae}

Finally, we introduce a completely new upgrade (\emph{Feynman anti-checkers}), allowing creation and annihilation of electron-positron pairs during the motion. The upgrade is defined just by allowing odd $(x+t)/\varepsilon$ in the Fourier integral (Proposition~\ref{cor-fourier-integral}), that is, computing the same integral in white checkerboard squares in addition to black ones. This is equivalent to the second quantization of lattice Dirac equation~\eqref{eq-Dirac-source1}--\eqref{eq-Dirac-source2}, but we do not need to work out this procedure (cf.~\cite[\S9F]{Bender-etal-94} and \cite[\S IV]{Bender-etal-85} for the massless case). Anyway, the true motivation of the upgrade is consistency with the initial model and appearance of spin-$1/2$ Feynman propagator~\eqref{eq-feynman-propagator} in the continuum limit (see Figure~\ref{fig-approx-b}). We also give a combinatorial definition (\new{see} Definition~\ref{def-anti-combi}).


\begin{figure}[htbp]
  \centering
\includegraphics[width=0.24\textwidth]{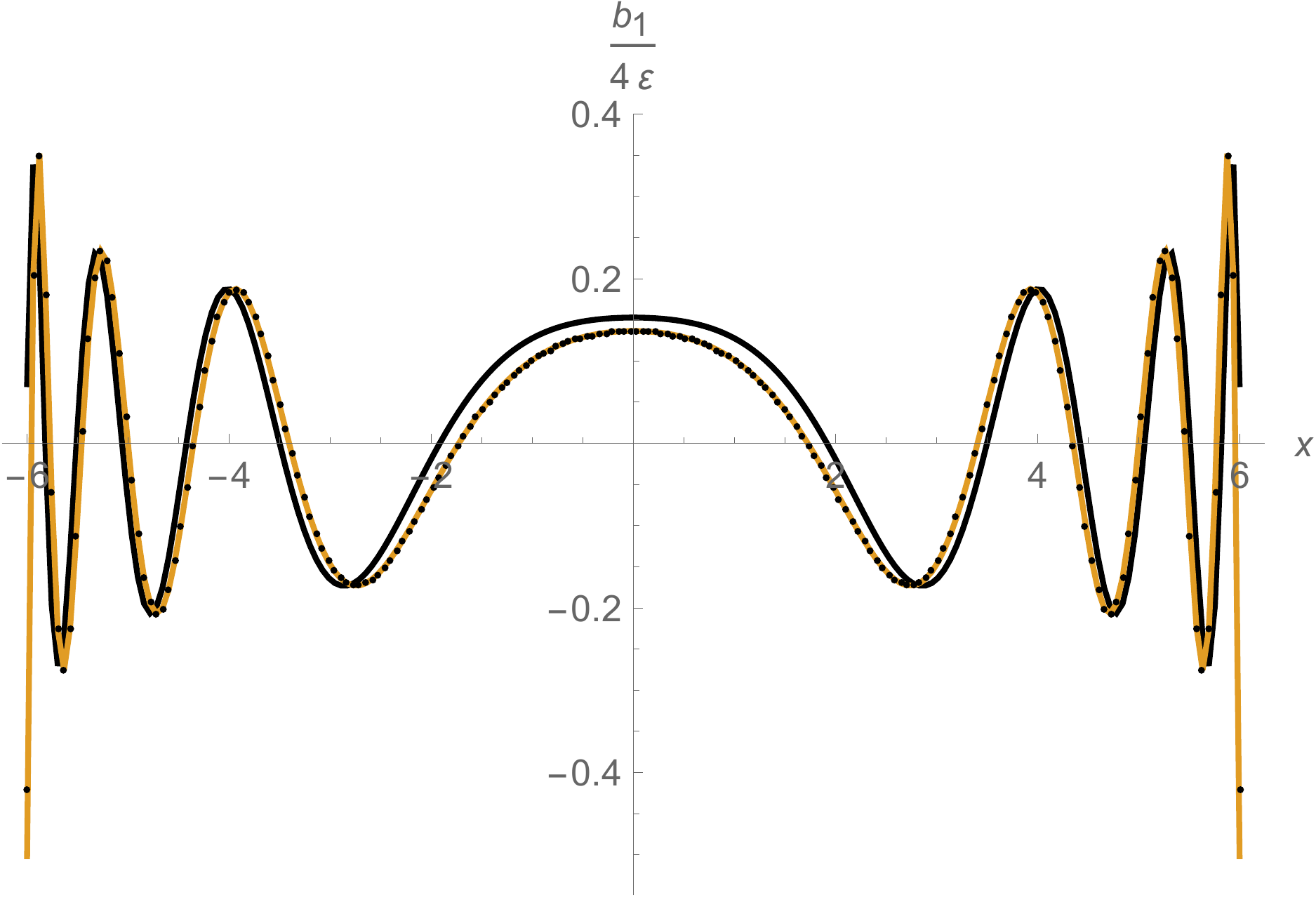}
\includegraphics[width=0.24\textwidth]{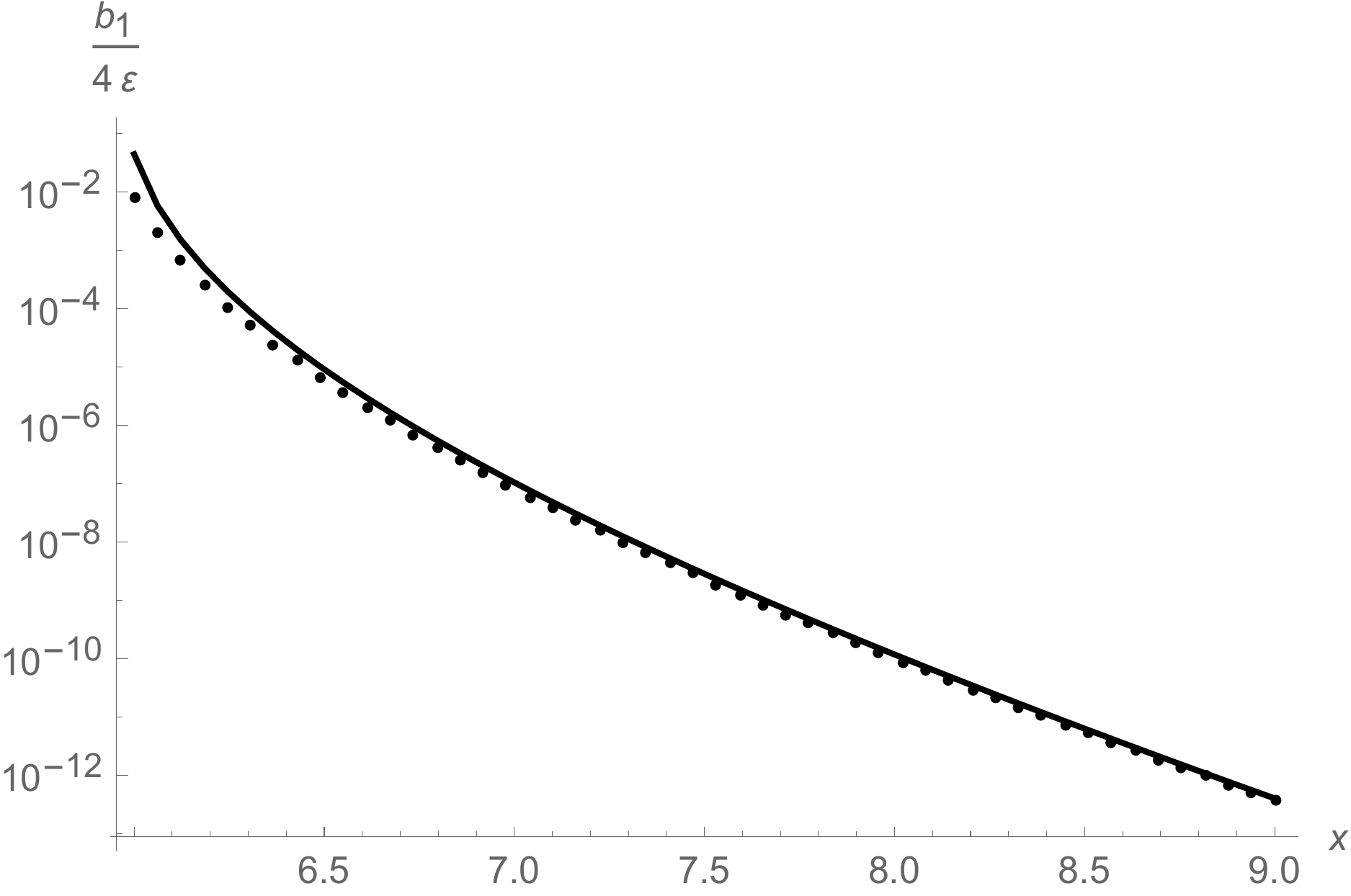}
\includegraphics[width=0.24\textwidth]{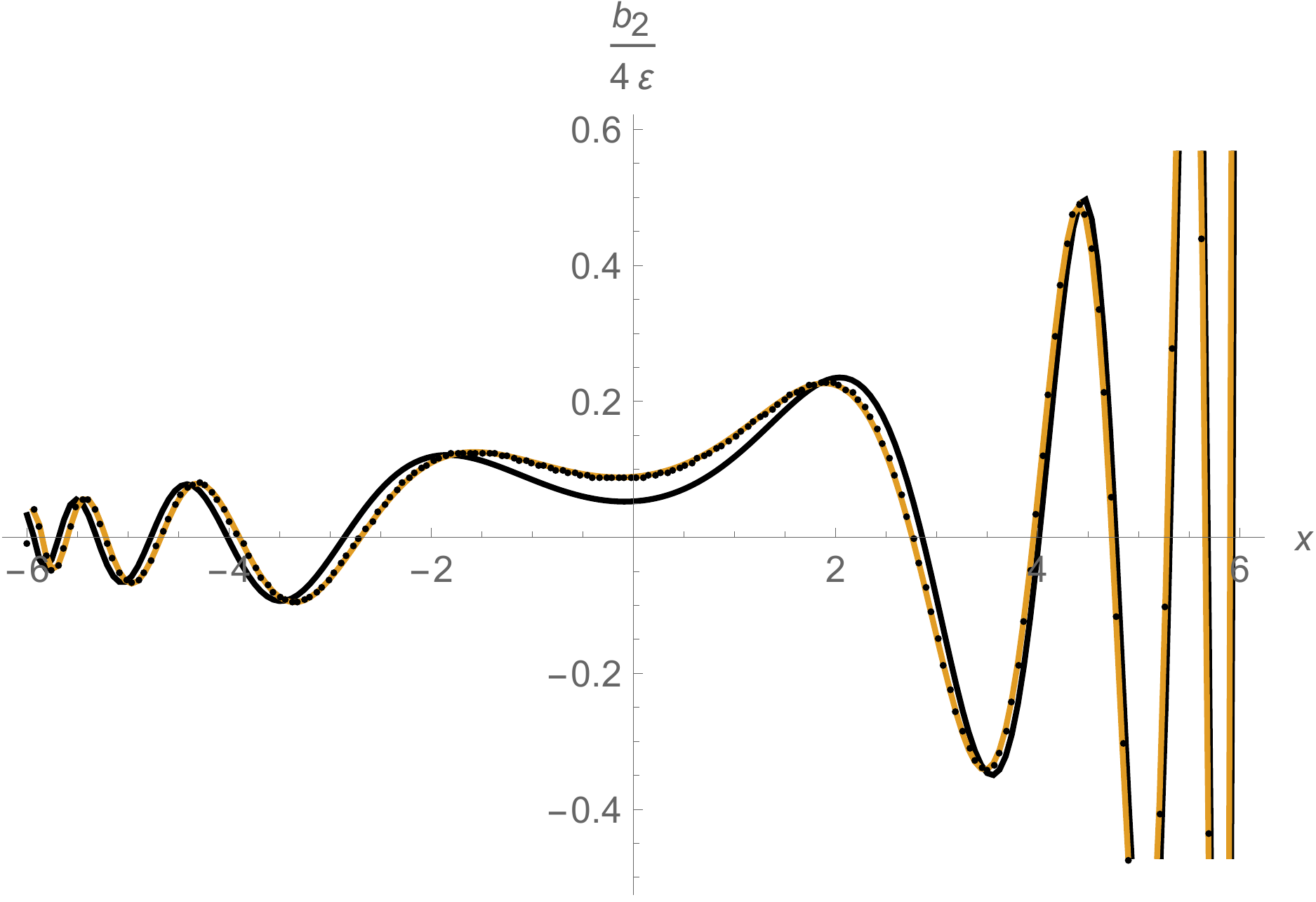}
\includegraphics[width=0.24\textwidth]{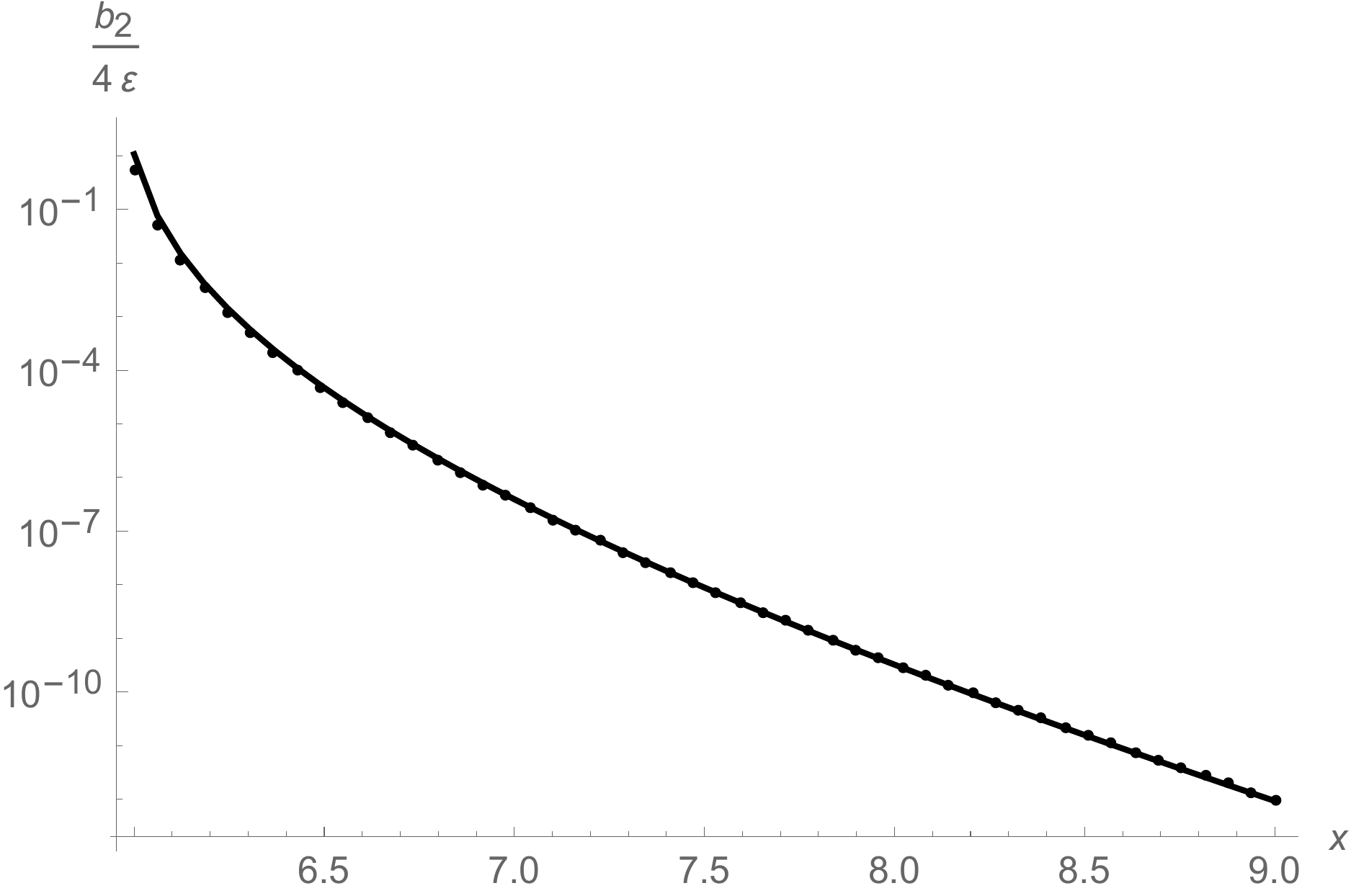}
  \caption{Plots of $b_1(x,6,4,0.03)/0.12$ (left, dots),
  $b_2(x,6,4,0.03)/0.12$ (right, dots), their
  analytic approximation from Theorem~\ref{th-anti-ergenium} (light), the imaginary part of the Feynman propagator
  $\mathrm{Im}\,G^F_{11}(x,6)$ (left, dark) and $\mathrm{Im}\,G^F_{12}(x,6)$ (right, dark) given by~\eqref{eq-feynman-propagator} for $m=4$ and $t=6$.} 
  \label{fig-approx-b}
\end{figure}


\begin{definition} \label{def-anti} (Cf.~Proposition~\ref{cor-fourier-integral}, see Figure~\ref{fig-approx-b}.)
Fix $m,\varepsilon>0$. For each $(x,t)\in \varepsilon\mathbb{Z}^2$ denote $\omega_p:=\frac{1}{\varepsilon}\arccos(\frac{\cos p\varepsilon}{\sqrt{1+m^2\varepsilon^2}})$ and
\begin{equation}
  \label{eq-def-anti1}
  \begin{aligned}
  A_1(x,t,m,\varepsilon)&:=
  \pm\frac{im\varepsilon^2}{2\pi}
  \int_{-\pi/\varepsilon}^{\pi/\varepsilon}
  \frac{e^{i p x-i\omega_p(t-\varepsilon)}\,dp}
  {\sqrt{m^2\varepsilon^2+\sin^2(p\varepsilon)}};\\
  A_2(x,t,m,\varepsilon)&:=
  \pm\frac{\varepsilon}{2\pi}\int_{-\pi/\varepsilon}^{\pi/\varepsilon}
  \left(1+
  \frac{\sin (p\varepsilon)} {\sqrt{m^2\varepsilon^2+\sin^2(p\varepsilon)}}\right) e^{ip(x-\varepsilon)-i\omega_p(t-\varepsilon)}\,dp,
  \end{aligned}
\end{equation}
where in both cases the overall minus sign is taken when $t\le 0$ and $(x+t)/\varepsilon$ is even. For $m=0$ define $A_2(x,t,m,\varepsilon)$ by the same formula and set $A_1(x,t,0,\varepsilon):=0$. In particular, $A_k(x,t,m,\varepsilon)=a_k(x,t,m,\varepsilon)$ for $(x+t)/\varepsilon$ even, $t>0$, and $k=1,2$.  Denote $A_k(x,t,m,\varepsilon)=:ib_k(x,t,m,\varepsilon)$ for $(x+t)/\varepsilon$ odd. Set $b_k(x,t,m,\varepsilon):=0$ for $(x+t)/\varepsilon$ even.
\end{definition}

One can show that $A_k(x,t,m,\varepsilon)$ \new{is} purely imaginary for $(x+t)/\varepsilon$ odd. Thus the real and the imaginary parts ``live'' on the black and white squares respectively, analogously to how discrete analytic functions are defined (see~\cite{Chelkak-Smirnov-11}). The sign convention for $t\le 0$ is dictated by the analogy to continuum theory (see~\eqref{eq-def-anti2} and~\eqref{eq-double-fourier-feynman}).

\begin{example} \label{ex-simplest-values} The value $b_1(0,1,1,1)=\Gamma(\tfrac{1}{4})^2/(2\pi)^{3/2}=
\tfrac{2}{\pi}K(i)=:G\approx 0.83463$ is the Gauss constant
and $-b_2(0,1,1,1)=2\sqrt{2\pi}/\Gamma(\tfrac{1}{4})^2
=\tfrac{2}{\pi}(E(i)-K(i))=1/\pi G=:L'\approx 0.38138$ is the inverse lemniscate constant, where $K(z)$ and $E(z)$ are the complete elliptic integrals of the 1st and 2nd kind respectively (cf.~\textup{\cite[\S6.1]{Finch-03}}).
\end{example}

The other values are even more complicated irrationalities (see Table~\ref{table-a4}).

\begin{table}[ht]
  \centering
  $b_1(x,t,1,1)$\\[0.1cm]
\begin{tabular}{|c*{6}{|>{\centering\arraybackslash}p{55pt}}|}
\hline
$2$& $\frac{G-L'}{\sqrt{2}}$ & & $\frac{G-L'}{\sqrt{2}}$ & & $\frac{7G-15L'}{3\sqrt{2}}$ \\
\hline
$1$& & $G$ & & $G-2L'$ & \\
\hline
$0$& $\frac{G-L'}{\sqrt{2}}$ & & $\frac{G-L'}{\sqrt{2}}$ & & $\frac{7G-15L'}{3\sqrt{2}}$ \\
\hline
$-1$& & $-L'$ & & $\frac{2G-3L'}{3}$ & \\
\hline
\diagbox[dir=SW,height=18pt]{$t$}{$x$}&$-1$&$0$&$1$&$2$&$3$\\
\hline
\end{tabular}
%
\\[0.2cm]
   $b_2(x,t,1,1)$\\[0.1cm]
\begin{tabular}{|c*{6}{|>{\centering\arraybackslash}p{55pt}}|}
\hline
$2$& $\frac{G-3L'}{3\sqrt{2}}$ & & $\frac{-G-L'}{\sqrt{2}}$ & & $\frac{-G+3L'}{\sqrt{2}}$ \\
\hline
$1$& & $-L'$ & & $L'$ & \\
\hline
$0$& $\frac{G-3L'}{\sqrt{2}}$ & & $\frac{G+L'}{\sqrt{2}}$ & & $\frac{-G+3L'}{3\sqrt{2}}$ \\
\hline
$-1$& & $G$ & & $\frac{G}{3}$ & \\
\hline
\diagbox[dir=SW,height=18pt]{$t$}{$x$}&$-1$&$0$&$1$&$2$&$3$\\
\hline
\end{tabular}
\qquad
  \caption{The values $b_1(x,t,1,1)$ and $b_2(x,t,1,1)$ for small $x,t$ (see Definition~\ref{def-anti} and Example~\ref{ex-simplest-values})}
  \label{table-a4}
\end{table}

We announce that the analogues of Propositions~\ref{p-mass}--\ref{p-symmetry-mass} and~\ref{p-equal-time} remain true literally, if $a_1$ and $a_2$ are replaced by $b_1$ and $b_2$ respectively (the assumption $t>0$ can then be dropped). As a consequence, $2^{(t-1)/2}b_1(x,t,1,1)$ and $2^{(t-1)/2}b_2(x,t,1,1)$ for all $(x,t)\in \mathbb{Z}^2$ are rational linear combinations of 
the Gauss constant~$G$ and the inverse lemniscate constant~$L'$. In general, we announce that $b_1$ and $b_2$ can be ``explicitly'' expressed through Gauss hypergeometric function: for $m,\varepsilon>0$, $(x,t)\in\varepsilon\mathbb{Z}^2$ with $(x+t)/\varepsilon$ odd we get
\begin{align*}
b_1(x,t,m,\varepsilon)
&=\left({1+m^2\varepsilon^2}\right)^{\frac{1}{2}-\frac{t}{2\varepsilon}}
(-m^2\varepsilon^2)^{\frac{t-|x|}{2\varepsilon}-\frac{1}{2}}
\binom{\frac{t+|x|}{2\varepsilon}-1}{|x|/\varepsilon}
\\&
\cdot{}_2F_1\left(1 + \frac{|x|-t}{2\varepsilon}, 1 + \frac{|x|-t}{2\varepsilon}; 1+\frac{|x|}{\varepsilon}; -\frac{1}{m^2\varepsilon^2}\right),\\
b_2(x+\varepsilon,t+\varepsilon,m,\varepsilon)
&=\left({1+m^2\varepsilon^2}\right)^{-\frac{t}{2\varepsilon}}
(m\varepsilon)^{\frac{t-|x|}{\varepsilon}}
(-1)^{\frac{t-|x|}{2\varepsilon}+\frac{1}{2}}
\binom{\frac{t+|x|}{2\varepsilon}-1+\theta(x)}{|x|/\varepsilon}
\,{}
\\&\cdot{}_2F_1\left(\frac{|x|-t}{2\varepsilon}, 1+\frac{|x|-t}{2\varepsilon}; 1+\frac{|x|}{\varepsilon}; -\frac{1}{m^2\varepsilon^2}\right),\qquad
\text{where }
\theta(x):=\begin{cases}
                            1, & \mbox{if } x\ge0, \\
                            0, & \mbox{if } x<0;
                          \end{cases}
\end{align*}
and ${}_2F_{1}(p,q;r;z)$ is the principal branch of the hypergeometric function. The idea of the proof is induction on $t/\varepsilon\ge 1$: the base is given by \new{}\cite[9.112, 
9.131.1, 9.134.3, \red{9.137.15}]{Gradstein-Ryzhik-63} and the step is given by the analogue of~\eqref{eq-Dirac-mass1}--\eqref{eq-Dirac-mass2} for $b_1$ and $b_2$ plus \new{}\cite[9.137.\red{11,12,18}]{Gradstein-Ryzhik-63}.

\begin{remark} \label{rem-jacobi-functions} (Cf.~Remark~\ref{rem-hypergeo}) These expressions can be rewritten as the \emph{Jacobi functions of the second kind} of half-integer order (see the definition in \cite[(4.61.1)]{Szego-39}). For instance, for each $(x,t)\in\varepsilon\mathbb{Z}^2$ such that $|x|>t$ and $(x+t)/\varepsilon$ is odd we have
$$
b_1(x,t,m,\varepsilon)
=\frac{2m\varepsilon}{\pi}
\left({1+m^2\varepsilon^2}\right)^{(t/\varepsilon-1)/2}
Q_{(|x|-t)/2\varepsilon}^{(0,t/\varepsilon-1)}(1+2m^2\varepsilon^2).
$$
\end{remark}

%
%
%

\begin{remark} \label{rem-probability-to-die} The number $b_1(x,\varepsilon,m,\varepsilon)$ equals $(1+\sqrt{1+m^2\varepsilon^2})/m\varepsilon$ times the probability that a planar simple random walk over white squares dies at
$(x,\varepsilon)$, if it starts at $(0,\varepsilon)$ and dies with the probability $1-1/\sqrt{1+m^2\varepsilon^2}$ before each step. Nothing like that is known for $b_{1}(x,t,m,\varepsilon)$ and $b_{2}(x,t,m,\varepsilon)$ with $t\ne \varepsilon$.
\end{remark}


The following two results are proved almost literally as Proposition~\ref{cor-double-fourier} and Theorem~\ref{th-ergenium}. 
(The only difference is change of the sign of the summands involving $f_-(p)$ in~\eqref{eq-oscillatory-integral}, \eqref{eq-main-term} \eqref{eq-boundary-term}, \eqref{eq-oscillatory-integral-a2}; the analogues of Lemmas~\ref{l-main-term} and~\ref{l-boundary-term} are then obtained by direct checking.)



\begin{proposition}[Full space-time Fourier transform] \label{cor-double-fourier-anti}
Denote $\delta_{x\varepsilon}:=1$, if $x=\varepsilon$, and $\delta_{x\varepsilon}:=0$, if $x\ne\varepsilon$.
For each $m>0$ and $(x,t)\in \varepsilon\mathbb{Z}^2$ we get
\begin{equation}
  \label{eq-def-anti2}
  \begin{aligned}
  A_1(x,t,m,\varepsilon)&=
  \lim_{\delta\to+0}
  \frac{m\varepsilon^3}{4\pi^2}
  \int_{-\pi/\varepsilon}^{\pi/\varepsilon}
  \int_{-\pi/\varepsilon}^{\pi/\varepsilon}
  \frac{ e^{i p x-i\omega(t-\varepsilon)}\,d\omega dp} {\sqrt{1+m^2\varepsilon^2}\cos(\omega\varepsilon)
  -\cos(p\varepsilon)-i\delta},\\
  \hspace{-0.5cm}A_2(x,t,m,\varepsilon)&=
  \lim_{\delta\to+0}
  \frac{-i\varepsilon^2}{4\pi^2}
  \int_{-\pi/\varepsilon}^{\pi/\varepsilon}
  \int_{-\pi/\varepsilon}^{\pi/\varepsilon}
  \frac{\sqrt{1+m^2\varepsilon^2}\sin(\omega\varepsilon)+\sin(p\varepsilon)} {\sqrt{1+m^2\varepsilon^2}\cos(\omega\varepsilon)
  -\cos(p\varepsilon)-i\delta}
  e^{i p (x-\varepsilon)-i\omega(t-\varepsilon)}
  \,d\omega dp+\delta_{x\varepsilon}\delta_{t\varepsilon}.
  \end{aligned}
  \end{equation}
\end{proposition}

\begin{theorem}
[Large-time asymptotic formula \new{between the peaks}; see Figure~\ref{fig-approx-b}]
\label{th-anti-ergenium}
For each $\delta>0$ there is $C_\delta>0$ such that for each $m,\varepsilon>0$ and each $(x,t)\in\varepsilon\mathbb{Z}^2$ satisfying~\eqref{eq-case-A} we have
\begin{align*}
{b}_1\left(x,t+\varepsilon,m,\varepsilon\right)
&={\varepsilon}\sqrt{\frac{2m}{\pi}}
\left(t^2-(1+m^2\varepsilon^2)x^2\right)^{-1/4}
\cos \theta(x,t,m,\varepsilon)
+O_\delta\left(\frac{\varepsilon}{m^{1/2}t^{3/2}}\right),\\
{b}_2\left(x+\varepsilon,t+\varepsilon,m,\varepsilon\right)
&=-{\varepsilon}\sqrt{\frac{2m}{\pi}}
\left(t^2-(1+m^2\varepsilon^2)x^2\right)^{-1/4}\sqrt{\frac{t+x}{t-x}}
\sin \theta(x,t,m,\varepsilon)
+O_\delta\left(\frac{\varepsilon}{m^{1/2}t^{3/2}}\right),
\end{align*}
for $(x+t)/\varepsilon$ even and odd respectively, where $\theta(x,t,m,\varepsilon)$ is given by~\eqref{eq-theta}.
\end{theorem}



\addcontentsline{toc}{myshrinkalt}{}

\subsection{Physical interpretation}

One interprets $\frac{1}{2}|{A}_1\left(x,t,m,\varepsilon\right)|^2+
\frac{1}{2}|{A}_2\left(x,t,m,\varepsilon\right)|^2$
as the \emph{expected charge} in a square $(x,t)$ with $t>0$, in the units of electron charge.
The numbers cannot be anymore interpreted as probabilities to find the electron in the square. The reason is that now the outcomes of the experiment are not mutually exclusive: one can detect an electron in two distinct squares simultaneously. 
There is nothing mysterious about that: Any measurement necessarily influences the electron. This influence might be enough to create an electron-positron pair from the vacuum. Thus one can detect a newborn electron in addition to the initial one; and there is no way to distinguish one from another. (A more formal explanation for specialists: two states in the Fock space representing the electron localized at distant regions are not mutually orthogonal; their scalar product 
is essentially provided by the Feynman propagator.) 

We announce that the model reproduces the Feynman propagator rather than the retarded one in the continuum limit (see Figure~\ref{fig-approx-b}). 
The \emph{spin-$1/2$ Feynman propagator} equals
\begin{equation}\label{eq-feynman-propagator}
G^F(x,t)=
\begin{cases}
\dfrac{m}{4}\,
\begin{pmatrix}
J_0(ms)-iY_0(ms) &
-{\frac{t+x}{s}}\left(J_1(ms)-iY_1(ms)\right) \\ {\frac{t-x}{s}}\left(J_1(ms)-iY_1(ms)\right) &
J_0(ms)-iY_0(ms)
\end{pmatrix}, &\text{if }|x|<|t|;\\
\dfrac{im}{2\pi}\,
\begin{pmatrix}
K_0(ms) & {\frac{t+x}{s}}\,K_1(ms)  \\
{\frac{x-t}{s}}\,K_1(ms) & K_0(ms)
\end{pmatrix}, &\text{if }|x|>|t|;
\end{cases}
\end{equation}
where $Y_n(z)$ and $K_n(z)$ are Bessel functions of the 2nd kind and modified Bessel functions of the 2nd kind, and $s:=\sqrt{\left|t^2-x^2\right|}$.
In addition, there is a generalized function supported on the lines $t=\pm x$\new{,} which we do not specify. The Feynman propagator satisfies~\eqref{eq-Green}. We see that it has additional imaginary part (and an overall factor of $1/2$) compared to retarded one~\eqref{eq-relativistic-propagator}. In particular, it does not vanish for $|x|>|t|$: annihilation of electron at one point and creation at another one may result in apparent motion faster than light.

A more common expression is the Fourier transform 
(cf.~\eqref{eq-double-fourier-retarded} and~\cite[(6.51)]{Folland}) 
\begin{equation}\label{eq-double-fourier-feynman}
G^F(x,t)=
  \frac{1}{4\pi^2}
  \int_{-\infty}^{+\infty}
  \int_{-\infty}^{+\infty}
  \lim_{\delta\to+0}
  \begin{pmatrix}
  m & -ip-i\omega \\
  -ip+i\omega & m
  \end{pmatrix}
  \frac{ e^{i p x-i\omega t}\,dpd\omega}
  {m^2+p^2-\omega^2-i\delta}.
\end{equation}

Overall, a small correction introduced by the upgrade reflects some fundamental limitations on measurement rather than adds something meaningful to description of the motion. The upgrade should only be viewed as an ingredient for more realistic models with interaction.

\addcontentsline{toc}{myshrinkalt}{}

\subsection{Combinatorial definition}

Informally, the combinatorial definition of Feynman anti-checkers (Definition~\ref{def-anti-combi}) is obtained from the definition of Feynman checkers (Definition~\ref{def-mass}) by the following four-step modification:
\begin{description}
\item[Step 1:] \new{the} particle mass acquires small imaginary part which we eventually tend to zero;
\item[Step 2:] just like the real mass is ``responsible'' for turns in the centers of the squares, the imaginary one allows turns at their \emph{vertices} (see Figure~\ref{fig-1x1} to the left);
\item[Step 3:] the infinite lattice is replaced by a torus with the size eventually tending to infinity;
\item[Step 4:] the sum over lattice paths is replaced by a ratio of sums over loop configurations.
\end{description}
Here Step~2 is completely new whereas the other ones are standard. It reflects a general principle that the real and the imaginary part of a quantity should be always put on dual lattices.

\begin{figure}[htbp]
  \centering
\includegraphics[height=2.5cm]{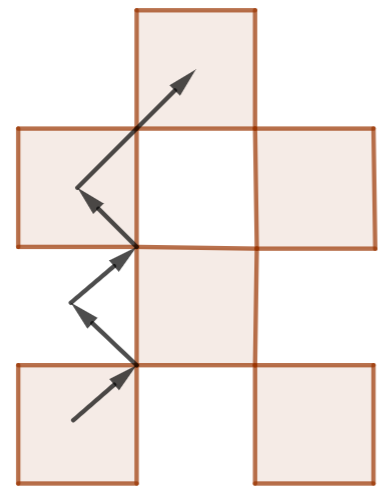}\qquad
\includegraphics[height=2.5cm]{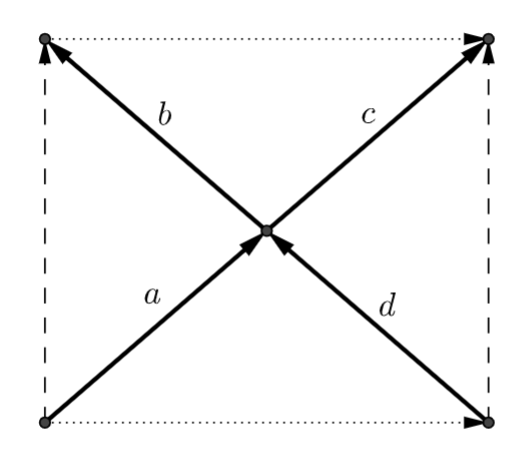}
\includegraphics[height=2.5cm]{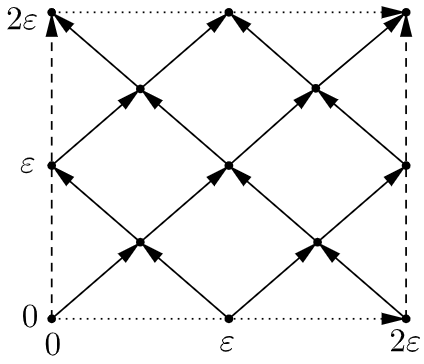}
\caption{A generalized checker path (left) and lattices of sizes $1$ and $2$ (right); see Example~\ref{ex-1x1}.
}
  \label{fig-1x1}
\end{figure}

\begin{definition} \label{def-anti-combi} (see Figure~\ref{fig-1x1})
Fix $T\in\mathbb{Z}$ and $\varepsilon,m,\delta>0$ called \emph{lattice size, lattice step, particle mass}, and \emph{small imaginary mass} respectively. Assume $T>0$ and $\delta<1/2$. The \emph{lattice} is the quotient set
$$
\faktor{
\{\,(x,t)\in[0,T\varepsilon]^2:
2x/\varepsilon,2t/\varepsilon,(x+t)/\varepsilon\in\mathbb{Z}\,\}
}
{
\forall x,t:(x,0)\sim (x,T\varepsilon)\,\&\,(0,t)\sim (T\varepsilon,t).
}
$$
(This is a \new{finite subset of the} torus obtained from the square $[0,T\varepsilon]^2$ by an identification of the opposite \new{sides}.)
A lattice point $(x,t)$ is \emph{even} (respectively, \emph{odd}), if $2x/\varepsilon$ is even (respectively, odd).
An \emph{edge} is a vector starting from a lattice point $(x,t)$ and ending at the lattice point $(x+\varepsilon/2,t+\varepsilon/2)$ or $(x-\varepsilon/2,t+\varepsilon/2)$.

A \emph{generalized checker path} is a finite sequence of distinct edges such that the endpoint of each edge is the starting point of the next one. A \emph{cycle} is defined analogously, only the sequence has the unique repetition: the first and the last edges coincide, and there is at least one edge in between. (In particular, a generalized checker path such that the endpoint of the last edge is the starting point of the first one is \emph{not} yet a cycle; coincidence of the first and the last \emph{edges} is required. The first and the last edges of a generalized checker \emph{path} coincide only if the path has a single edge. \new{Thus} in our setup, a path is \emph{never} a cycle.) \emph{Changing the starting edge} of a cycle means removal of the first edge from the sequence, then a cyclic permutation, and then adding the last edge of the resulting sequence at the beginning. A \emph{loop} is a cycle up to changing of the starting edge.

A \emph{node} of a path or loop $s$ is \new{an ordered} pair of consecutive edges in $s$ \new{(the order of the edges in the pair is the same as in~$s$)}. 
A \emph{turn} is a node such that the two edges are orthogonal. A node or turn is \emph{even} (respectively, \emph{odd}), if the \new{endpoint of the first edge in the pair}
is even (respectively, odd). Denote by $\mathrm{eventurns}(s)$, $\mathrm{oddturns}(s)$, $\mathrm{evennodes}(s)$, $\mathrm{oddnodes}(s)$ the number of even and odd turns and nodes in $s$. The \emph{arrow} (or \emph{weight)} of $s$ is
\begin{equation*}
{A}(s,m,\varepsilon,\delta)
:=\pm\frac{(-im\varepsilon)^{\mathrm{oddturns}(s)}
(-\delta)^{\mathrm{eventurns}(s)}}
{(1+m^2\varepsilon^2)^{\mathrm{oddnodes}(s)/2}
(1-\delta^2)^{\mathrm{evennodes}(s)/2}},
\end{equation*}
where the overall minus sign is taken when $s$ is a loop.

A set of checker paths or loops is \emph{edge-disjoint}, if
no two of them have a common edge. An edge-disjoint set of loops is a \emph{loop configuration}. A \emph{loop configuration with a source \new{$a$} and a sink $f$} is an edge-disjoint set of any number of loops and exactly one generalized checker path starting with the edge \new{$a$} and ending with the edge $f$. 
The \emph{arrow} $A(S,m,\varepsilon,\delta)$ of a loop configuration $S$ (possibly with a source and a sink) is the product of arrows of all loops and paths in the configuration. An empty product is set to be $1$.

The \emph{arrow from an edge \new{$a$} to an edge $f$} (or \emph{finite-lattice propagator}) is \new{}\new{}
\begin{equation*}
{A}(a\to f,m,\varepsilon,\delta,T)
:=\frac{\sum\limits_{\substack{\text{loop configurations $S$}\\
\text{with the source $a$ and the sink $f$}}}A(S,m,\varepsilon,\delta)}
{\sum\limits_{\substack{\text{loop configurations $S$}}}A(S,m,\varepsilon,\delta)}.
\end{equation*}

Now take a point $(x,t)\in (\varepsilon\mathbb{Z})^2$ and set $ x':=x\mod T\varepsilon$, $t':=t\mod T\varepsilon$. Let \new{} $a_0,f_1,f_2$ be the edges starting at $(0,0)$, $(x',t')$, $(x',t')$ and ending at $(\varepsilon/2,\varepsilon/2)$, $(x'-\varepsilon/2,t'+\varepsilon/2)$, $(x'+\varepsilon/2,t'+\varepsilon/2)$ respectively.
The \emph{arrow of the point $(x,t)$} (or \emph{infinite-lattice propagator}) is the pair of complex numbers \new{}
\begin{equation*}
\widetilde{A}_k(x,t,m,\varepsilon):=-2(-i)^k\,
\lim_{\delta\searrow 0}\lim_{\substack{T\to\infty}}
{A}(a_0\to f_k,m,\varepsilon,\delta,T)
\qquad\text{for $k=1,2$.}
\end{equation*}
\end{definition}

\begin{example} \label{ex-1x1} (\new{See} Figure~\ref{fig-1x1} to the middle) \new{The} lattice of size $1$ \new{lies on} the square $[0,\varepsilon]^2$ with the opposite sides identified.
\new{The lattice has $2$ points: the midpoint and the identified vertices of the square.}
It \new{has} $4$ edges $a,b,c,d$. The generalized checker paths $abdc$, $acdb$, $bacd$ are distinct although they contain the same edges. Their arrows are $\tfrac{-m^2\varepsilon^2}{\sqrt{1-\delta^2}(1+m^2\varepsilon^2)}$,
$\tfrac{-\delta}{\sqrt{1-\delta^2}(1+m^2\varepsilon^2)}$,
$\tfrac{\delta^2}{(1-\delta^2)\sqrt{1+m^2\varepsilon^2}}$
respectively. Those paths are distinct from the cycles $acdba$, $bacdb$. The two cycles determine the same loop with the arrow $\tfrac{-\delta^2}{(1-\delta^2)(1+m^2\varepsilon^2)}$. There are totally $9$ loop configurations: $\varnothing$, $\{aba\}$, $\{cdc\}$, $\{aca\}$, $\{bdb\}$, $\{abdca\}$, $\{acdba\}$, $\{aba,cdc\}$, $\{aca,bdb\}$. Their arrows are $1$, ${-im\varepsilon\delta}/n$, ${-im\varepsilon\delta}/n$,
${-1}/n$, ${-1}/n$, ${m^2\varepsilon^2}/n^2$, ${-\delta^2}/n^2$,
${-m^2\varepsilon^2\delta^2}/n^2$, ${1}/n^2$ respectively, where $n:={\sqrt{1-\delta^2}\sqrt{1+m^2\varepsilon^2}}$.
\end{example}


Informally, the loops form the \emph{Dirac sea} of electrons filling the whole space, and the edges not in the loops form paths of holes in the sea, that is, antiparticles.

We announce that Definitions~\ref{def-anti} and~\ref{def-anti-combi} are equivalent in the sense that $\widetilde{A}_1(x,t,m,\varepsilon)={A}_1(x,t+\varepsilon,m,\varepsilon)$ and $\widetilde{A}_2(x,t,m,\varepsilon)={A}_2(x+\varepsilon,t+\varepsilon,m,\varepsilon)$;
cf.~\S\ref{sec-source}. The idea is that both definitions actually construct matrix elements of the ``inverse'' of the same lattice Dirac operator: the former via the Fourier transform, and the latter via the ratio of determinants.

\comment


{
\hrule
\small
\noindent\textbf{Question:} the same as in \S\ref{sec-IdParAnt} ``Identical particles''.

\noindent\textbf{Assumptions:} the electrons now generate an electromagnetic field affecting the motion.


\noindent\textbf{Results:} experimental: repulsion of like charges and attraction of opposite charges is expected.
\hrule
}

\bigskip

\endcomment


\addcontentsline{toc}{myshrink}{}

\section{Towards $(1+1)$-dimensional quantum electrodynamics}
\label{sec-QED}

\addcontentsline{toc}{myshrink}{}

{
\hrule
\footnotesize
\noindent\textbf{Question:} what is the probability to find
electrons (or an electron and a positron) with momenta $q$ and $q'$ in the far future, if they were emitted with momenta $p$ and $p'$ in the far past?

\noindent\textbf{Assumptions:} interaction now switched on; all simplifying assumptions removed except the default ones:\\
no nuclear forces, no gravitation, electron moves along the $x$-axis only, and the $t$-axis is time.

\noindent\textbf{Results:} repulsion of like charges and attraction of opposite charges 
(qualitative explanation expected).
\hrule
}

\bigskip

Construction of the required model is a widely open problem  because in particular it requires the missing mathematically rigorous construction of the \emph{Minkowskian} lattice gauge theory.


\section{Open problems}\label{sec-open}

\addcontentsline{toc}{myshrink}{}






\comment

\begin{pr}\label{p-ugol-eng} (A.Daniyarkhodzhaev--F.Kuyanov)
Do the sides of the ``apparent angle'' in Figure~\ref{P-contour} to the left lie on the lines $t=\pm\sqrt{2}\,x$? (not~$t=\pm x$ as one could expect!)
\end{pr}

For fixed $t$, the random variable $x/t$ is called \emph{average electron velocity}. Informally, the latter problem means that with very high probability the average velocity \emph{in the basic model} does not exceed $1/\sqrt{2}$ of the speed of light.
A physical explanation is that lattice regularization cuts off distances smaller than the lattice step, hence large momenta, hence large velocities. In particular, finite lattice step affects the behavior not only at small distances but at very large distances as well.

\mscomm{??? Remove the next problem and the paragraphs around ???}

Likewise, the random variable equal to $\pm 1$ depending on the direction of the last checker move is called
the \emph{instantaneous velocity} (so that the probability of the value $-1$ equals $\sum_{x\in\mathbb{Z}}a_1(x,t)^2$).

\begin{pr} \label{p-average} (D.~Treschev)
Does the expectation of the average electron velocity equal the time-average of the expectation of the instantaneous velocity?
\end{pr}

This is not at all automatic because the model does \emph{not} give any probability distribution on the set of checker paths.

\endcomment


We start with problems relying on Definition~\ref{def-basic}.
The plots suggest that for fixed $t$, the most probable position of the electron is near to $x=t/\sqrt{2}$ (see Figure~\ref{a-1000} to the left, Theorems~\ref{th-limiting-distribution}(B) and~\ref{th-ergenium}). Although this was noticed 20 years ago, the following question is still open.


\begin{pr} \label{P-ugol-eng} (A.Daniyarkhodzhaev--F.Kuyanov; see Figure~\ref{a-1000} to the left) Denote by $x_{\max}(t)$ a point where $P(x,t)$ has a maximum for fixed $t$. Is $x_{\max}(t)-t/\sqrt{2}$ bounded as $t\to\infty$?
\end{pr}

What makes the problem hard is that the behavior of $P(x,t)$ is only known near to $x=t/\sqrt{2}$ (Theorem~\ref{th-Airy}) and far from $x=t/\sqrt{2}$ (Theorems~\ref{th-ergenium} and~\ref{th-outside}) but not at intermediate distances.

\begin{pr} (Cf.~\cite{Sunada-Tate-12}) Find a asymptotic formula for $P(x,t)$ as $t\to\infty$ uniform in $x\in [-t;t]$.
\end{pr}

\begin{pr} (S.~Nechaev; see Figure~\ref{a-1000} to the left) Find the positions of ``wide gaps'' (intervals, where oscillations are smaller) in the plot of $P(x,t)$ for fixed large~$t$. 
(Cf.~
formulae~\eqref{eq-ergenium-re}--\eqref{eq-ergenium-im}.)
\end{pr}

%

The aim of the next two problems is to study the phase transition by means of various order parameters (see~page~\pageref{equal-signs}). Specifically, we conjecture that the limiting ``probability'' of equal signs at the endpoints of the spin-chain, as well as the limiting ``probability'' of equal signs at the endpoints and the midpoint, are nonanalytic at $v=\pm 1/\sqrt{2}$.


\begin{pr} \label{p-correlation} (See Figure~\ref{fig-correlation}) Prove that for each $0<v<1/\sqrt{2}$ we have
$$
\lim_{t\to\infty}\sum\limits_{0\le x\le vt} \frac{2}{t}\left|\frac{a_2(x,t)}{a(x,t)}\right|^2=
\frac{1}{2} \left(1+v-\sqrt{1-v^2}+\log \frac{1+\sqrt{1-v^2}}{2}\right).
$$
Compute the same limit for $1/\sqrt{2}<v<1$.
(Cf.~the proof of Corollary~\ref{th-limiting-distribution-mass} in~\S\ref{ssec-proofs-phase}.) \mscomm{??? Problem on Young diagram ???}
\end{pr}

\begin{pr} \label{p-correlation2} (Cf.~\cite[p.~381]{Jacobson-Schulman-84}.)
Find the weak limit 
$\lim_{t\to\infty}\left|\sum_{x\in\mathbb{Z}}
\dfrac{a_2(x,t)^2}{a_2(2\lceil vt\rceil-1,2t-1)}\right|^2$.
\end{pr}




The next problem is on absorption probabilities; it relies on the definition before~Example~\ref{p-double-slit}.

\begin{pr} \label{p-pi} (G.\,Minaev--I.\,Russkikh; cf.~\cite[\S5]{Ambainis-etal-01},\,\cite{Dmitriev-22},\,\cite[\S4]{Novikov-20})
Find $\sum_{t=1}^{\infty}P(n,t \,\text{bypass}\, \{x\!\!=\!\!n\})$ for all $n\in\mathbb{Z}$.
Find the weak limit and an asymptotic formula for
$P(x,t \text{ bypass } \{x=0\})$ as $t\to\infty$. (Cf.~Theorems~\ref{th-limiting-distribution}--\ref{th-ergenium}.)
\end{pr}


The following problem generalizes and specifies Problem~\ref{P-ugol-eng} above; it relies on Definition~\ref{def-mass}.

\begin{pr} \label{P-ugol-mass} (A.Daniyarkhodzhaev--F.Kuyanov, cf.~\cite[\S4]{Ambainis-etal-01}, \cite{Sunada-Tate-12}, \cite{Zakorko-21}) Denote by $x_{\max}=x_{\max}(t,m,\varepsilon)$ the point where $P(x):=P(x,t,m,\varepsilon)$ has a maximum. Is $x_{\max}/\varepsilon-t/\varepsilon\sqrt{1+m^2\varepsilon^2}$  uniformly bounded? Does $P(x)$ decrease for $x>x_{\max}$? Find an asymptotic formula for $a(x,t,m,\varepsilon)$ as $t\to \infty$ valid for all $x\in [-t,t]$ and  uniform in $x,m,\varepsilon$.
\end{pr}

\begin{pr} (M.~Blank--S.~Shlosman) Is the number of times the function $a_1(x):=a_1(x,t,m,\varepsilon)$ changes the sign on $[-t,t]$ bounded as $\varepsilon\to 0$ for fixed $t,m$?
\end{pr}

Corollary~\ref{cor-uniform} gives uniform limit on compact subsets of the angle $|x|<t$, hence misses the main contribution to the probability. Now we ask for the weak limit detecting the peak.

\begin{pr} \label{p-weak} Find the weak limits $\lim\limits_{\varepsilon\to 0}  \frac{1}{2\varepsilon}\, {a}\left(2\varepsilon\!\left\lceil \frac{x}{2\varepsilon}\right\rceil,2\varepsilon\!\left\lceil \frac{t}{2\varepsilon}\right\rceil,{m},{\varepsilon}\right)$ and $\lim\limits_{\varepsilon\to 0}  \frac{1}{4\varepsilon^2}\, {P}\left(2\varepsilon\!\left\lceil \frac{x}{2\varepsilon}\right\rceil,2\varepsilon\!\left\lceil \frac{t}{2\varepsilon}\right\rceil,{m},{\varepsilon}\right)$
on the whole $\mathbb{R}^2$. Is the former limit equal to propagator~\eqref{eq-double-fourier-retarded} including the generalized function supported on the lines $t=\pm x$?
What is the physical interpretation of the latter limit (providing a value to the ill-defined square of the propagator)?
\end{pr}


%


The following problem is to construct a continuum analogue of Feynman checkers.

\begin{pr} (M.~Lifshits)
Consider $(-im\varepsilon)^{\mathrm{turns}(s)}$ as a charge on the set of all checker paths $s$ from $(0,0)$ to $(x,t)$ starting and ending with an upwards-right move. Does the charge converge (weakly or in another sense) to a charge on the space of all continuous functions $[0,t]\to\mathbb{R}$ with boundary values $0$ and $x$ respectively as $\varepsilon\to 0$?
\end{pr}

\comment

The following problem generalizes Theorem~\ref{p-right-prob} to the upgrade with external field; see \S\ref{sec-external field}. 

\begin{pr} \label{p-spin-reversal} (I.Gaidai-Turlov,T.Kovalev,A.Lvov) Is it true that $\lim_{t\to+\infty} \sum_{x\in\mathbb{Z}} a_1(x, t, m,\varepsilon)^2 = \frac{m\varepsilon}{2\sqrt{1+m^2\varepsilon^2}}$?
\end{pr}

\endcomment

The following problem relying on Definition~\ref{def-external} would demonstrate ``spin precession''.

\begin{pr} \label{p-precession-prove} (See Figure~\ref{spin-reversal} to the right; cf.~\cite{Ozhegov-21})
Is $P(x)=\sum_{x\in\mathbb{Z}} a_1(x, t, u)^2$ a periodic function asymptotically as $t\to\infty$ for $u(x+\tfrac{1}{2},t+\tfrac{1}{2})=(-1)^{(x-1)(t-1)}$?
Find the weak limit of $P(x,t,u)$ and asymptotic formulae for
$a_k(x,t,u)$ as $t\to\infty$. (Cf.~Theorems~\ref{th-limiting-distribution}--\ref{th-ergenium}.)
\end{pr}




Define ${a}(x,t,m,\varepsilon,u)$ analogously to ${a}(x,t,m,\varepsilon)$ and ${a}(x,t,u)$, unifying Definitions~\ref{def-mass}--\ref{def-external} and Remark~\ref{rem-external}. The next problem asks if this reproduces Dirac equation in electromagnetic field.


\begin{pr} (Cf.~\cite{Schulman-Gaveau-89}) Fix $A_0(x,t),A_1(x,t)\in C^2(\mathbb{R}^2)$. For each auxiliary edge $s_1s_2$ set
$$u(s_1s_2)
:=\exp\left(-i\int_{s_1}^{s_2}\left(A_0(x,t)\,dt+A_1(x,t)\,dx\right)\right).
$$
Denote $\psi_k(x,t):=\lim_{\varepsilon\to 0}  \frac{1}{2\varepsilon}\, {a}_k\left(2\varepsilon\!\left\lceil \frac{x}{2\varepsilon}\right\rceil,2\varepsilon\!\left\lceil \frac{t}{2\varepsilon}\right\rceil,{m},{\varepsilon},u\right)$ for $k=1,2$.
Does the limit satisfy 
$$
\begin{pmatrix}
m  & \partial/\partial x-\partial/\partial t +iA_0(x,t)-iA_1(x,t)\\
\partial/\partial x+\partial/\partial t - iA_0(x,t)-iA_1(x,t) & m
\end{pmatrix}
\begin{pmatrix}
\psi_2(x,t) \\ \psi_1(x,t)
\end{pmatrix}=0\quad\text{for }t>0?
$$
\end{pr}

The next two problems rely on Definition~\ref{def-anti}.

\begin{pr} (Cf.~Corollary~\ref{th-limiting-distribution-mass})
Prove that
$
\lim\limits_{\substack{t\to\infty\\t\in\varepsilon\mathbb{Z}}}
\sum\limits_{\substack{x\le vt\\x\in\varepsilon\mathbb{Z}}}
\dfrac{\left|A_1(x,t,{m},{\varepsilon})\right|^2
+\left|A_2(x,t,{m},{\varepsilon})\right|^2}{2}=F(v,m,\varepsilon)$.
\end{pr}

\begin{pr} (Cf.~Theorems~\ref{th-outside} and~\ref{th-anti-ergenium}) Find an asymptotic formula for $A_k(x,t,{m},{\varepsilon})$  for $|x|>|t|$ as $t\to \infty$.
\end{pr}


The last problem is informal; it stands for half a century. 

\begin{pr} (R.~Feynman; cf.~\cite{Foster-Jacobson-17}) Generalize the model to $4$ dimensions so that 
$$\lim_{\varepsilon\to 0 }  \frac{1}{2\varepsilon}\, {a}\left(2\varepsilon\!\left\lceil \frac{x}{2\varepsilon}\right\rceil,2\varepsilon\!\left\lceil \frac{y}{2\varepsilon}\right\rceil,2\varepsilon\!\left\lceil \frac{z}{2\varepsilon}\right\rceil,2\varepsilon\!\left\lceil \frac{t}{2\varepsilon}\right\rceil,{m},{\varepsilon}\right)$$ coincides with the spin-$1/2$ retarded propagator, now in
$3$ space- and $1$ time-dimension.
\end{pr}

\section{Proofs}\label{sec-proofs}



Let us present a chart showing
the dependence of the above results and further subsections:
{\footnotesize
$$
\xymatrix{
\boxed{\text{\ref{ssec-proofs-moments}. (Theorem~\ref{th-limiting-distribution})}}
&
\boxed{\text{\ref{ssec-proofs-basic}. (Propositions~\ref{p-Dirac}--\ref{cor-double-fourier})}}
\ar[l]_{\ref{p-mass},\ref{cor-fourier-integral}}
\ar[ld]_{\ref{p-Dirac}-\ref{cor-coefficients},\ref{p-symmetry-mass},\ref{p-Huygens}}
\ar[d]^{\ref{p-mass2},\ref{p-mass3}}
\ar[rrd]^{\ref{p-mass2}}
\ar[r]^{\ref{p-mass},\ref{cor-fourier-integral}}
&
\boxed{\text{\ref{ssec-proofs-main}. (Theorem~\ref{th-ergenium})}} \ar[r]\ar[rd]
&
\boxed{\text{\ref{ssec-proofs-feynman}. (Corollaries~\ref{cor-intermediate-asymptotic-form}--\ref{cor-feynman-problem})}}
\\
\boxed{\text{\ref{ssec-proofs-spin}. (Theorem~\ref{p-right-prob})}}
&
\boxed{\text{\ref{ssec-proofs-continuum}. (Theorem~\ref{th-main}, Corollaries~\ref{cor-uniform}--\ref{cor-concentration})}}
&
&
\boxed{\text{\ref{ssec-proofs-phase}. (Corollaries~\ref{th-limiting-distribution-mass}--\ref{cor-free})}}
}
$$
}
Particular proposition numbers are shown above the arrows. Propositions~\ref{p-Klein-Gordon-mass}, \ref{p-equal-time}, and~\ref{cor-double-fourier} are not used in the main results. 

In the process of the proofs, we give a zero-knowledge introduction to the used methods. Some proofs are simpler than the original ones.

\addcontentsline{toc}{myshrinkalt}{}

\subsection{Identities: elementary combinatorics (Propositions~\ref{p-Dirac}--\ref{cor-double-fourier}) }
\label{ssec-proofs-basic}

Let us prove the identities from~\S\ref{sec-mass}; the ones from \S\ref{sec-basic} are the particular case $m=\varepsilon=1$. 

\begin{proof}[Proof of Propositions~\ref{p-Dirac} and~\ref{p-mass}]
Let us derive a recurrence for $a_2(x,t,m,\varepsilon)$. Take a path $s$ on $\varepsilon\mathbb{Z}^2$ from $(0,0)$ to $(x,t)$ with the first step to $(\varepsilon,\varepsilon)$. Set $a(s,m \varepsilon):=i(-im \varepsilon)^{\mathrm{turns}(s)}(1+m^2\varepsilon^2)^{(1-t/\varepsilon)/2}$. 

The last move in the path $s$ is made either from $(x-\varepsilon,t-\varepsilon)$ or from $(x+\varepsilon,t-\varepsilon)$. If it is from $(x+\varepsilon,t-\varepsilon)$, then $\mathrm{turns}(s)$ must be odd, hence $s$ does not contribute to $a_2(x,t,m,\varepsilon)$.
Assume further that the last move in $s$ is made from $(x-\varepsilon,t-\varepsilon)$.
Denote by $s'$ the path $s$ without the last move. If the directions of the last moves in $s$ and $s'$ coincide, then $a(s,m \varepsilon) = \frac{1}{\sqrt{1+m^2\varepsilon^2}} a(s',m \varepsilon)$, otherwise $a(s,m \varepsilon) = \frac{-im \varepsilon}{\sqrt{1+m^2\varepsilon^2}}a(s',m \varepsilon)
= \frac{m \varepsilon}{\sqrt{1+m^2\varepsilon^2}}(\mathrm{Im}\,a(s',m \varepsilon)-i\mathrm{Re}\,a(s',m \varepsilon))$.

Summation over all paths $s'$ gives the required equation
\begin{multline*}
a_2(x,t,m, \varepsilon) =
\mathrm{Im}\sum_{s\ni(x-\varepsilon,t-\varepsilon)} a(s,m\varepsilon)
=\sum_{s'\ni(x-2\varepsilon,t-2\varepsilon)} \frac{\mathrm{Im}\,a(s',m\varepsilon)}{\sqrt{1+m^2\varepsilon^2}}-
\sum_{s'\ni(x,t-2\varepsilon)} \frac{m\varepsilon\,\mathrm{Re}\,a(s',m\varepsilon)}
{\sqrt{1+m^2\varepsilon^2}} \\
=\frac{1}{\sqrt{1+m^2\varepsilon^2}}\left( a_2(x-\varepsilon,t-\varepsilon,m, \varepsilon)
-  m \varepsilon\, a_1(x-\varepsilon,t-\varepsilon,m, \varepsilon)\right).
\end{multline*}
The recurrence for $a_1(x,t+\varepsilon,m,\varepsilon)$ is proved analogously.
\end{proof}

\begin{proof}[Proof of Propositions~\ref{p-probability-conservation} and~\ref{p-mass2}]
The proof is by induction \new{on} $t/\varepsilon$. The base $t/\varepsilon=1$ is
obvious. 
The step of induction follows immediately from the following computation using Proposition~\ref{p-mass}:
\begin{multline*}
    \sum \limits_{x \varepsilon\in \mathbb{Z}} P(x,t+\varepsilon,m, \varepsilon)
    = \sum \limits_{x \in \varepsilon\mathbb{Z}}
    \left[ a_1(x, t+\varepsilon,m, \varepsilon)^2
    + a_2(x, t+\varepsilon,m, \varepsilon)^2 \right]
    \\ \hspace{-0.8cm} = \tfrac{1}{1+m^2\varepsilon^2}
    \left(\sum \limits_{x \in \varepsilon\mathbb{Z}}
    [a_1(x+\varepsilon,t,m, \varepsilon)
    + m \varepsilon\, a_2(x+\varepsilon, t,m, \varepsilon)]^2  + \sum \limits_{x \in \varepsilon\mathbb{Z}}
    [a_2(x-\varepsilon,t,m, \varepsilon)
    - m \varepsilon\, a_1(x-\varepsilon, t,m, \varepsilon)]^2\right)
    \\= \tfrac{1}{1+m^2\varepsilon^2}\left(
    \sum \limits_{x \in \varepsilon\mathbb{Z}}
    [a_1(x,t,m, \varepsilon)
    + m \varepsilon\, a_2(x, t,m, \varepsilon)]^2  + \sum \limits_{x \in \varepsilon\mathbb{Z}}
    [a_2(x,t,m, \varepsilon)
    - m \varepsilon\, a_1(x, t,m, \varepsilon)]^2\right)
    \\= \sum \limits_{x \in \varepsilon\mathbb{Z}}
    \left[a_1(x,t,m, \varepsilon)^2
    + a_2(x, t,m, \varepsilon)^2\right]
    = \sum \limits_{x \in \varepsilon\mathbb{Z}}P(x,t,m, \varepsilon).\\[-1.5cm]
\end{multline*}
\end{proof}

\begin{lemma}[Adjoint Dirac equation] \label{p-Dirac-conjugate} For each $(x,t)\in \varepsilon\mathbb{Z}^2$, where $t>\varepsilon$, we have
\begin{align*}
a_1(x,t-\varepsilon,m,\varepsilon)
&=\frac{1}{\sqrt{1+m^2\varepsilon^2}}
(a_1(x-\varepsilon,t,m,\varepsilon)
-m\varepsilon\, a_2(x+\varepsilon,t,m,\varepsilon));\\
a_2(x,t-\varepsilon,m,\varepsilon)
&=\frac{1}{\sqrt{1+m^2\varepsilon^2}}
(m\varepsilon\, a_1(x-\varepsilon,t,m,\varepsilon)
+a_2(x+\varepsilon,t,m,\varepsilon)).
\end{align*}
\end{lemma}

\begin{proof}[Proof of Lemma~\ref{p-Dirac-conjugate}]
The second equation is obtained from  Proposition~\ref{p-mass} by substituting $(x,t)$ by $(x-\varepsilon,t-\varepsilon)$ and $(x+\varepsilon,t-\varepsilon)$ in~\eqref{eq-Dirac-mass1} and~\eqref{eq-Dirac-mass2} respectively
and adding them with the coefficients $m\varepsilon/\sqrt{1+m^2\varepsilon^2}$ and $1/\sqrt{1+m^2\varepsilon^2}$. The first equation is obtained analogously.
\end{proof}

\begin{proof}[Proof of Proposition
~\ref{p-Klein-Gordon-mass}]
The real part of the desired equation is the sum of the first equations of Lemma~\ref{p-Dirac-conjugate} and Proposition~\ref{p-mass}. The imaginary part is the sum of the  second ones.
\end{proof}

\begin{proof}[Proof of Proposition
~\ref{p-symmetry-mass}]
Let us prove the first identity.
For a path $s$ denote by $s'$ the reflection of $s$ with respect to the $t$ axis, and by $s''$ the path consisting of the same moves as $s'$, but in the opposite order. 

Take a path $s$ from $(0,0)$ to $(x,t)$ with the first move upwards-right such that $\mathrm{turns}(s)$ is odd (the ones with $\mathrm{turns}(s)$ even do not contribute to $a_1(x,t,m,\varepsilon)$). Then the last move in $s$ is upwards-left. Therefore, the last move in $s'$ is upwards-right, hence the first move in $s''$ is upwards-right. The endpoint of both $s'$ and $s''$ is $(-x,t)$, because reordering of moves does not affect the endpoint. Thus $s\mapsto s''$ is a bijection between the paths to $(x,t)$ and to $(-x,t)$ with  $\mathrm{turns}(s)$ odd. Thus $a_1(x,t,m,\varepsilon)=a_1(-x,t,m,\varepsilon)$.

We prove the second identity by induction on $t/\varepsilon$ (this proof was found and written by E.~Kolpakov). The base of induction ($t/\varepsilon=1$ and $t/\varepsilon=2$) is obvious.

Step of induction:  take $t\ge 3\varepsilon$. Applying the inductive hypothesis for the three points $(x-\varepsilon,t-\varepsilon), (x+\varepsilon,t-\varepsilon), (x,t-2\varepsilon)$ and the identity just proved, we get
\begin{align*}
  (t-x)a_2(x-\varepsilon,t-\varepsilon, m, \varepsilon)&=(x+t-4\varepsilon)a_2(3\varepsilon-x,t-\varepsilon, m, \varepsilon),\\
  (t-x-2\varepsilon)a_2(x+\varepsilon,t-\varepsilon, m, \varepsilon)&=(x+t-2\varepsilon)a_2(\varepsilon-x,t-\varepsilon, m, \varepsilon),\\
  (t-x-2\varepsilon)a_2(x,t-2\varepsilon, m, \varepsilon)&=(x+t-4\varepsilon)a_2(2\varepsilon-x,t-2\varepsilon, m, \varepsilon),\\
   a_1(x-\varepsilon,t-\varepsilon, m, \varepsilon)&= a_1(\varepsilon-x,t-\varepsilon, m, \varepsilon).
\end{align*}
Summing up the $4$ equations with the coefficients $1,1,-\sqrt{1+m^2\varepsilon^2},-2m\varepsilon^2$ respectively, we get
\begin{multline*}
(t-x)\left(a_2(x-\varepsilon,t-\varepsilon, m, \varepsilon)+a_2(x+\varepsilon,t-\varepsilon, m, \varepsilon)-\sqrt{1+m^2\varepsilon^2}\,a_2(x,t-2\varepsilon, m, \varepsilon)\right)\\
 -2m\varepsilon^2\, a_1(x-\varepsilon,t-\varepsilon, m, \varepsilon) -2\varepsilon\, a_2(x+\varepsilon,t-\varepsilon, m, \varepsilon)+2\varepsilon\sqrt{1+m^2\varepsilon^2}\,a_2(x,t-2\varepsilon, m, \varepsilon)= \\
=-2m\varepsilon^2\, a_1(\varepsilon-x,t-\varepsilon, m, \varepsilon) -2\varepsilon\, a_2(3\varepsilon-x,t-\varepsilon, m, \varepsilon)+2\varepsilon\sqrt{1+m^2\varepsilon^2}\,a_2(2\varepsilon-x,t-2\varepsilon, m, \varepsilon)\\
 +(t+x-2\varepsilon)\left(a_2(3\varepsilon-x,t-\varepsilon, m, \varepsilon)+a_2(\varepsilon-x,t-\varepsilon, m, \varepsilon)-\sqrt{1+m^2\varepsilon^2}\,a_2(2\varepsilon-x,t-2\varepsilon, m, \varepsilon)\right).
\end{multline*}
Here the 3 terms in the 2nd line, as well as the 3 terms in the 3rd line, cancel each other by Lemma~\ref{p-Dirac-conjugate}.
Applying the Klein--Gordon equation (Proposition~\ref{p-Klein-Gordon-mass}) to the expressions in the 1st and 4th line and cancelling the common factor $\sqrt{1+m^2\varepsilon^2}$, we get the desired identity
$$(t-x)a_2(x,t, m, \varepsilon) =(t+x-2\varepsilon)a_2(2\varepsilon-x,t,m, \varepsilon).$$

The third identity follows from the first one and Proposition~\ref{p-mass}:
\begin{multline*}
a_1(x,t,m,\varepsilon)+m\varepsilon\, a_2(x,t,m,\varepsilon)
=\sqrt{1+m^2\varepsilon^2}\, a_1(x-\varepsilon,t+\varepsilon,m,\varepsilon)
=\\=\sqrt{1+m^2\varepsilon^2}\, a_1(\varepsilon-x,t+\varepsilon,m,\varepsilon)
=a_1(2\varepsilon-x,t,m,\varepsilon)
+m\varepsilon\, a_2(2\varepsilon-x,t,m,\varepsilon).\\[-1.5cm]
\end{multline*}
\end{proof}

The 1st and the 3rd identities can also be proved simultaneously by induction on $t/\varepsilon$ using Proposition~\ref{p-mass}.

\begin{figure}[htb]
  \centering
  \includegraphics[width=0.15\textwidth]{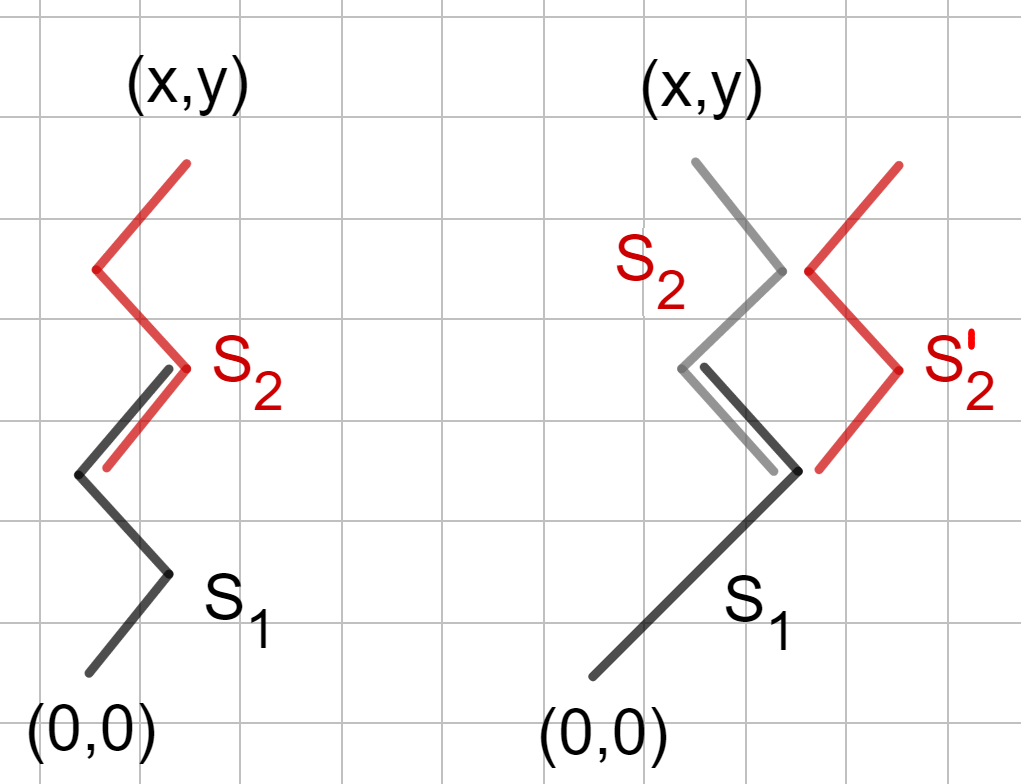}
  \vspace{-0.3cm}
  \caption{
  The path cut into two parts (see the proof of Proposition~\ref{p-Huygens}). 
  }\label{fig-huygens}
\end{figure}

\begin{proof}[Proof of Proposition~\ref{p-Huygens}]
Take a checker path $s$ from $(0,0)$ to $(x,t)$.
Denote by $(x',t')$ the point where $s$ intersects the line $t = t'$. Denote by $s_1$ the part of $s$ that joins $(0,0)$ with $(x',t')$. Denote by $s_2$ the part starting at the intersection point of $s$ with the line $t = t'-\varepsilon$ and ending at $(x,t)$ (see Figure~\ref{fig-huygens}). Translate the path $s_2$ so that it starts at $(0,0)$.
Set
$a(s,m \varepsilon)
:=i(-im \varepsilon)^{\mathrm{turns}(s)}
(1+m^2\varepsilon^2)^{(1-t/\varepsilon)/2}$.
Since $\mathrm{turns}(s)=\mathrm{turns}(s_1)+\mathrm{turns}(s_2)$, it follows that
\begin{equation*}
\mathrm{Re}\,a(s,m\varepsilon) =
 \begin{cases}
   \mathrm{Re}\,a(s_1,m\varepsilon)
   \mathrm{Im}\,a(s_2,m\varepsilon), & \text{if the move to $(x',t')$ is upwards-left},\\
   \mathrm{Im}\,a(s_1,m\varepsilon)
   \mathrm{Re}\,a(s_2,m\varepsilon), & \text{if the move to $(x',t')$ is upwards-right}.
 \end{cases}
\end{equation*}
In the former case replace the path $s_2$ by the path $s_2'$ obtained by the reflection with respect to the line $x=0$ (and starting at the origin). We have  $\mathrm{Im}\,a(s_2',m\varepsilon)= \mathrm{Im}\,a(s_2,m\varepsilon)$. Therefore,
\blue{
\begin{multline*}
a_1(x,t,m,\varepsilon) = \sum\limits_{s}\mathrm{Re}\,a(s,m\varepsilon)
= \sum \limits_{x'} \sum \limits_{s \ni (x',t')} \mathrm{Re}\,a(s,m\varepsilon) \\
= \sum \limits_{x'} \left(\sum \limits_{s \ni (x',t'),(x'-\varepsilon,t'-\varepsilon)}
\mathrm{Im}\,a(s_1,m\varepsilon)
\mathrm{Re}\,a(s_2,m\varepsilon)
+ \sum \limits_{s \ni (x',t'),(x'+\varepsilon,t'-\varepsilon)}
\mathrm{Re}\,a(s_1,m\varepsilon)
\mathrm{Im}\,a(s_2',m\varepsilon)\right) \\
= \sum \limits_{x'} \left[ a_2(x',t',m,\varepsilon)a_1(x-x'+\varepsilon,t-t'+\varepsilon,m,\varepsilon) +a_1(x',t',m,\varepsilon)a_2(x'-x+\varepsilon,t-t'+\varepsilon,m,\varepsilon) \right].
\end{multline*}
}
The formula for $a_2(x,t,m,\varepsilon)$ is proven analogously.
\end{proof}


\comment

\begin{remark} \mscomm{!!! Notational conflict -- letter $t$ !!!}
Formulae illustrating the Huygens principle have the form of a convolution.
Define two series of Laurent polynomials: 
$$P_n=P_n(t) = \sum \limits_{m = -n+2}^{n}  a_1(m,n)t^m,\qquad Q_n=Q_n(t) = \sum \limits_{m = -n+2}^{n}  a_2(m,n)t^m.$$
These polynomials can be considered as generating functions of the sequences $a_{1,2}(m,n)$ with fixed $n$.
By Propositions~\ref{p-symmetry-mass} and~\ref{p-Huygens} we have
\begin{multline*}
a_2(x,t)  = \sum \limits_{x'} \left[ a_2(x',t')a_2(x-x'+1,t-t'+1) - a_1(x',t')a_1(x'-x+1,t-t'+1) \right] =  \\ = \sum \limits_{x'} \left[ a_2(x',t')a_2(x-x'+1,t-t'+1) - a_1(x',t')a_1(x-x'-1,t-t'+1) \right].
\end{multline*}
This formula is equivalent to the following \textit{addition theorem} for $Q$-polynomials:
\blue{$$Q_{n+m} = \frac{1}{t}Q_{n}Q_{m+1} - tP_{n}P_{m+1}.$$}
Also from the Propositions~\ref{p-symmetry-mass} and~\ref{p-Huygens} we have
\begin{multline*}
a_1(x,t)  =\sum \limits_{x'} \left[ a_2(x',t')a_1(x-x'+1,t-t'+1) + a_1(x',t')a_2(x'-x+1,t-t'+1) \right] =\\= \sum \limits_{x'} [ a_2(x',t')a_1(x-x'+1,t-t'+1) +\\+ a_1(x',t')(a_2(x-x'+1,t-t'+1)+a_1(x-x'+1,t-t'+1) - a_1(x-x'-1,t-t'+1) )].
\end{multline*}
This formula respectively is equivalent to the following addition theorem for $P$-polynomials:
\blue{$$P_{n+m} = \frac{1}{t}\left(Q_nP_{m+1} +P_nQ_{m+1}\right)+ \left( \frac{1}{t} - t\right)P_nP_{m+1}.$$}
In particular, substitution of $m=1$ gives the Dirac equation written in terms of
$P$- and $Q$-polynomials:
\begin{equation*}
\begin{cases}
P_{n+1} &=  \frac{1}{t}Q_nP_{2} + P_n \cdot (\frac{1}{t}Q_{2} + \frac{1}{t}P_{2} - tP_{2}) = \frac{1}{t\sqrt{2}}(Q_n + P_n),\\
Q_{n+1}&=  \frac{1}{t}Q_{n}Q_{2} - tP_{n}P_{2} = \frac{t}{\sqrt{2}}(Q_n - P_n).
\end{cases}
\end{equation*}
It follows immediately that
\begin{multline*}
\begin{cases}
t\sqrt{2}P_{n+1} - Q_n &= P_n = Q_n - \frac{\sqrt{2}}{t}Q_{n+1}  \\
t\sqrt{2}P_{n+1} - P_n &= Q_n = P_n + \frac{\sqrt{2}}{t}Q_{n+1}
\end{cases}
\Rightarrow
\begin{cases}
P_{n+1} &= \frac{\sqrt{2}}{t}Q_n -  \frac{Q_{n+1}}{t^2}\\
Q_{n+1}&= t^2P_{n+1} - t\sqrt{2}P_n
\end{cases}
\Rightarrow
\blue{
\begin{cases}
P_{n+1} &=  \frac{1}{\sqrt{2}}(t+\frac{1}{t})P_n  - P_{n-1} \\
Q_{n+1} &=  \frac{1}{\sqrt{2}}(t + \frac{1}{t})Q_n - Q_{n-1}.
\end{cases}}
\end{multline*}
The latter formulae give us yet another recursive definitions for $P_n$ and $Q_n$.
\end{remark}

\endcomment

\begin{proof}[Proof of Proposition~\ref{p-equal-time}]
Denote by $f(x,t)$ the difference between the left- and the right-hand side of~\eqref{eq-p-equal-time}. Introduce the operator
$$[\square_m f](x,t)
:=\sqrt{1+m^2\varepsilon^2}\,f(x,t+\varepsilon)
+\sqrt{1+m^2\varepsilon^2}\,f(x,t-\varepsilon)
-f(x+\varepsilon,t)-f(x-\varepsilon,t).$$
It suffices to prove that
\begin{equation}\label{eq-sq5}
[\square_m^4 f](x,t)=0 \qquad \text{for }t\ge 5\varepsilon.
\end{equation}
Then~\eqref{eq-p-equal-time} will follow by induction on~$t/\varepsilon$: \eqref{eq-sq5} expresses  $f(x,t+4\varepsilon)$ as a linear combination of $f(x',t')$ with smaller $t'$; it remains to check $f(x,t)=0$ for $t\le 8\varepsilon$, which is done in~\cite[\S11]{SU-2}.

To prove~\eqref{eq-sq5}, write
$$
f(x,t)=:p_1(x,t)a(x-2\varepsilon,t,m,\varepsilon)
+p_2(x,t)a(x+2\varepsilon,t,m,\varepsilon)
+p_3(x,t)a(x,t,m,\varepsilon)
$$
for certain cubic polynomials $p_k(x,t)$ (see~\eqref{eq-p-equal-time}), and observe the \emph{Leibnitz rule}
$$
\square_m(fg)=f\cdot\square_m g
+\sqrt{1+m^2\varepsilon^2}
(\nabla_{t+}f\cdot T_{t+}g-\nabla_{t-}f\cdot T_{t-}g)
-\nabla_{x+}f\cdot T_{x+}g+\nabla_{x-}f\cdot T_{x-}g,
$$
where $[\nabla_{t\pm}f](x,t):=\pm(f(x,t\pm\varepsilon)-f(x,t))$
and $[\nabla_{x\pm}f](x,t):=\pm(f(x\pm\varepsilon,t)-f(x,t))$ are the finite difference operators, $[T_{t\pm}g](x,t):=g(x,t\pm\varepsilon)$
and $[T_{x\pm}g](x,t):=g(x\pm\varepsilon,t)$ are the \new{shift} operators. Since $\square_m\, a(x,t,m,\varepsilon)=0$ by Proposition~\ref{p-Klein-Gordon-mass}, each operator $\nabla_{t\pm},\nabla_{x\pm}$ decreases $\deg p_k(x,t)$, and all the above operators commute, by the Leibnitz rule we get~\eqref{eq-sq5}. This proves the first identity in the proposition; the second one is proved analogously (the induction base is checked in~\cite[\S11]{SU-2}).
\end{proof}

Alternatively, Proposition~\ref{p-equal-time} can be derived by applying the Gauss contiguous relations to the hypergeometric expression from Remark~\ref{rem-hypergeo} seven times.

\begin{proof}[Proof of Propositions~\ref{Feynman-binom} and~\ref{p-mass3}]
Let us find $a_1(x,t,m,\varepsilon)$. 
Consider a path with an odd number of turns; the other ones do not contribute to $a_1(x,t,m,\varepsilon)$. Denote by $2r+1$ the number of turns in the path. Denote by $R$ and $L$ the number of upwards-right and upwards-left moves respectively. Let $x_1, x_2, \dots, x_{r+1}$ be the number of upwards-right moves \emph{before} the first, the third, \dots, the last turn respectively. Let $y_1, y_2, \dots, y_{r+1}$ be the number of upwards-left moves \emph{after} the first, the third, \dots, the last turn respectively. Then $x_k,y_k\ge 1$ for $1\le k\le r+1$ and
\begin{align*}
R&=x_1+\dots+x_{r+1};\\
L&=y_1+\dots+y_{r+1}.
\end{align*}
The problem now reduces to a combinatorial one: the number of paths with $2r+1$ turns equals the number of positive integer solutions of the resulting equations. For the first equation, this number is the number of ways to put $r$ sticks between $R$ coins in a row, that is, $\binom{R-1}{r}$. Thus
$$
a_1(x,t,m,\varepsilon) = (1+m^2\varepsilon^2)^{(1-t/\varepsilon)/2}
\sum_{r=0}^{\min\{R,L\}}(-1)^r \binom{R-1}{r}\binom{L-1}{r}(m\varepsilon)^{2r+1}.
$$
Thus~\eqref{eq1-p-mass} follows from
$L+R=t/\varepsilon$ and $R-L=x/\varepsilon$. Formula~\eqref{eq2-p-mass} is derived analogously.
\end{proof}


\begin{proof}[Proof of Proposition~\ref{cor-fourier-integral}]
The proof is by induction on $t/\varepsilon$.

The base $t/\varepsilon=1$ is obtained by the change of variable $p\mapsto p+\pi/\varepsilon$ so that the integrals over $[0;\pi/\varepsilon]$ and $[-\pi/\varepsilon;0]$ cancel for $x/\varepsilon$ odd and there remains $(\varepsilon/2\pi)\int_{-\pi/\varepsilon}^{\pi/\varepsilon}e^{ip(x-\varepsilon)}=\delta_{x\varepsilon}$.

The inductive step is the following computation and an analogous \new{one} for $a_2(x, t+\varepsilon,m,\varepsilon)$:
\begin{multline*}
  a_1(x, t+\varepsilon,m,\varepsilon)
  = \frac{1}{\sqrt{1+m^2\varepsilon^2}} \left(
  a_1(x+\varepsilon,t,m,\varepsilon)
  + m\varepsilon \, a_2(x+\varepsilon,t,m,\varepsilon)\right)
  \\
  = \frac{m\varepsilon^2}{2\pi\sqrt{1+m^2\varepsilon^2}}
  \int_{-\pi/\varepsilon}^{\pi/\varepsilon}
  \left(
  \frac{ie^{ip\varepsilon}}
  {\sqrt{m^2\varepsilon^2+\sin^2(p\varepsilon)}}
  +1+\frac{\sin (p\varepsilon)}
  {\sqrt{m^2\varepsilon^2+\sin^2(p\varepsilon)}}\right)
  e^{i p x-i\omega_p(t-\varepsilon)}\,dp\\
  = \frac{m\varepsilon^2}{2\pi}
  \int_{-\pi/\varepsilon}^{\pi/\varepsilon}
  \frac{\left(i\cos(\omega_p\varepsilon)
  +\sin(\omega_p\varepsilon)\right)
  e^{i p x-i\omega_p(t-\varepsilon)}\,dp}
  {\sqrt{m^2\varepsilon^2+\sin^2(p\varepsilon)}}
  = \frac{im\varepsilon^2}{2\pi}
  \int_{-\pi/\varepsilon}^{\pi/\varepsilon}
  \frac{e^{i p x-i\omega_pt}\,dp}  {\sqrt{m^2\varepsilon^2+\sin^2(p\varepsilon)}}.
\end{multline*}
Here the 1st equality is Proposition~\ref{p-mass}. The 2nd one is the inductive hypothesis. The 3rd one follows from
$\cos\omega_p\varepsilon=\frac{\cos p\varepsilon}{\sqrt{1+m^2\varepsilon^2}}$
and $\sin\omega_p\varepsilon
=\sqrt{1-\frac{\cos^2 p\varepsilon}{1+m^2\varepsilon^2}}=
\sqrt{\frac{m^2\varepsilon^2+\sin^2 p\varepsilon}{1+m^2\varepsilon^2}}$. 
\end{proof}

Alternatively, Proposition~\ref{cor-fourier-integral} can be derived by integration of~\eqref{eq-solution1}--\eqref{eq-solution2} over $p=2\pi/\lambda$ for $\tilde a_1(0,0)=0$, $\tilde a_2(0,0)=1$.

\begin{proof}[Proof of Proposition~\ref{cor-double-fourier}]
To prove the formula for $a_1(x,t,m,\varepsilon)$, we do the $\omega$-integral:
\begin{multline*}
\frac{\varepsilon}{2\pi}\int\limits_{-\pi/\varepsilon}^{\pi/\varepsilon}
  \frac{e^{-i\omega(t-\varepsilon)}\,d\omega} {\sqrt{1+m^2\varepsilon^2}\cos(\omega\varepsilon)
  -\cos(p\varepsilon)-i\delta}
\overset{(*)}{=}\frac{1}{2\pi i}\oint\limits_{|z|=1}
  \frac{2\,z^{t/\varepsilon-1}\,dz} {\sqrt{1+m^2\varepsilon^2}z^2
  -2(\cos(p\varepsilon)+i\delta)z+\sqrt{1+m^2\varepsilon^2}}\\
\overset{(**)}{=}\frac{\left(\left(\cos p\varepsilon+i\delta
- i\sqrt{m^2\varepsilon^2+\sin^2p\varepsilon+\delta^2
-2i\delta\cos p\varepsilon}\right)/ \sqrt{1+m^2\varepsilon^2}\right)^{t/\varepsilon-1}}
{-i\sqrt{m^2\varepsilon^2+\sin^2p\varepsilon+\delta^2
-2i\delta\cos p\varepsilon}}
\overset{(***)}{\rightrightarrows}
\frac{i\,e^{-i\omega_p(t-\varepsilon)}}
{\sqrt{m^2\varepsilon^2+\sin^2p\varepsilon}}
\end{multline*}
as $\delta\to 0$ uniformly in $p$. Here we assume that $m,t,\delta>0$ and $\delta$ is sufficiently small. Equality~$(*)$ is obtained by the change of variables $z=e^{-i\omega\varepsilon}$ and then the change of the contour direction to the counterclockwise one. To prove~$(**)$, we find the roots of the denominator
$$
z_\pm=\frac{\cos p\varepsilon+i\delta
\pm i\sqrt{m^2\varepsilon^2+\sin^2p\varepsilon+\delta^2
-2i\delta\cos p\varepsilon}}{\sqrt{1+m^2\varepsilon^2}},
$$
where $\sqrt{z}$ denotes the value of the square root with positive real part. Then $(**)$ follows from the residue formula: the expansion
\new{}
$$
z_\pm=\frac{\cos p\varepsilon\pm i\sqrt{m^2\varepsilon^2+\sin^2p\varepsilon}}{\sqrt{1+m^2\varepsilon^2}}
\left(1\pm \frac{\delta}{\sqrt{m^2\varepsilon^2+\sin^2p\varepsilon}}
+{O}_{m,\varepsilon}\left(\delta^2\right)
\right)
$$
shows that $z_-$ is inside the unit circle, whereas $z_+$ is outside, for \new{$\delta>0$ sufficiently small in terms of $m$ and $\varepsilon$}. In~$(***)$ we denote
$\omega_p:=\frac{1}{\varepsilon}\arccos(\frac{\cos p\varepsilon}{\sqrt{1+m^2\varepsilon^2}})$
so that $\sin\omega_p\varepsilon
=\sqrt{\frac{m^2\varepsilon^2+\sin^2 p\varepsilon}{1+m^2\varepsilon^2}}$
and pass to the limit $\delta\to 0$ which is uniform in $p$ by the assumption~$m>0$.

The resulting uniform convergence allows to interchange the limit and the $p$-integral, and we arrive at Fourier integral for $a_1(x,t,m,\varepsilon)$ in Proposition~\ref{cor-fourier-integral}.
The formula for $a_2(x,t,m,\varepsilon)$ is proved analogously, with the case $t=\varepsilon$ considered separately.
\end{proof}


\comment

\begin{proof}[Proof of Proposition~\ref{p-probability-conservation-external}]
Use induction over $t$.
The step of induction follows immediately from the following computation:
\begin{multline*}
\sum \limits_{x \in \mathbb{Z}} P(x,t+1,u) = \sum \limits_{x \in \mathbb{Z}} \left[ a_1(x, t+1,u)^2 + a_2(x, t+1,u)^2 \right] = \sum \limits_{x \in \mathbb{Z}} a_1(x, t+1,u)^2 + \sum \limits_{x \in \mathbb{Z}} a_2(x, t+1,u)^2 = \\ = \frac{1}{2}\sum \limits_{x \in \mathbb{Z}}u(x+\frac{1}{2},t+\frac{1}{2})^2( a_1(x+1,t,u) + a_2(x+1, t,u))^2  + \frac{1}{2}\sum \limits_{x \in \mathbb{Z}}u(x-\frac{1}{2},t+\frac{1}{2})^2( a_2(x-1,t,u) - a_1(x-1, t,u))^2 = \\ =  \sum \limits_{x \in \mathbb{Z}}\frac{( a_1(x,t,u) + a_2(x, t,u))^2}{2}  + \sum \limits_{x \in \mathbb{Z}}\frac{( a_2(x,t,u) - a_1(x, t,u))^2}{2} = \sum \limits_{x \in \mathbb{Z}} \left[a_1(x,t,u)^2 + a_2(x, t,u)^2\right] = \sum \limits_{x \in \mathbb{Z}} P(x,t,u).
\end{multline*}
\end{proof}

\endcomment

\addcontentsline{toc}{myshrinkalt}{}

\subsection{Phase transition: the method of moments (Theorem~\ref{th-limiting-distribution})}
\label{ssec-proofs-moments}

In this subsection we give a simple exposition of the proof of Theorem~\ref{th-limiting-distribution} from \cite{Grimmett-Janson-Scudo-04} using the \emph{method of moments}. The theorem also follows from Corollary~\ref{th-limiting-distribution-mass} obtained by another method in~\S\ref{ssec-proofs-phase}. We rely on the following well-known result.

\begin{lemma}\label{l-method-moments}
\textup{(See~\cite[Theorems~30.1--30.2]{Billingsley-95})}
Let $f_t\colon\mathbb{R}\to [0,+\infty)$, where $t=0,1,2,\dots$, be piecewise continuous functions such that $\alpha_{r,t}:=\int_{-\infty}^{+\infty}v^rf_t(v)\,dv$ is finite and $\alpha_{0,t}=1$ for each $r,t=0,1,2,\dots$. If the series $\sum_{r=0}^{\infty}\alpha_{r,0}z^r/r!$ has positive radius of convergence and $\lim_{t\to\infty}\alpha_{r,t}=\alpha_{r,0}$ for each $r=0,1,2,\dots$, then $f_t$ converges to $f_0$ in distribution.
\end{lemma}

\begin{proof}[Proof of Theorem~\ref{th-limiting-distribution}]
Let us prove 
(C); then  (A) and (B) will follow from Lemma~\ref{l-method-moments} for $f_0(v):=F'(v)$ and \new{}$f_t(v):=tP(\lceil vt \rceil,t)$ because \new{$F'(v)=0$} for $|v|>1$, hence \new{$\alpha_{r,0}\le \int_{-1}^{+1}|F'(v)|\,dv=1$}.

Rewrite Proposition~\ref{cor-fourier-integral} in a form, valid for \emph{each} $x,t\in\mathbb{Z}$ independently on the parity: 
\new{}\new{}
\begin{equation}\label{eq-rewritten-fourier}
\begin{pmatrix}
  a_1(x,t) \\
  a_2(x,t)
\end{pmatrix}
=\int_{-\pi}^{\pi}
\begin{pmatrix}
  \hat a_1(p,t) \\
  \hat a_2(p,t)
\end{pmatrix}
e^{ip(x-1)}\,\frac{dp}{2\pi}
:=\int_{-\pi}^{\pi}
\begin{pmatrix}
  \hat a_{1+}(p,t)+\hat a_{1-}(p,t) \\
  \hat a_{2+}(p,t)+\hat a_{2-}(p,t)
\end{pmatrix}
e^{ip(x-1)}\,\frac{dp}{2\pi},
\end{equation}
where
\begin{equation}\label{eq-fourier-transform}
\begin{aligned}
\hat a_{1\pm}(p,t)
&= \mp\frac{ie^{ip}}{2\sqrt{1+\sin^2 p}}
e^{\pm i\omega_p(t-1)};\\
\hat a_{2\pm}(p,t)
&= \frac{1}{2}\left(1\mp\frac{\sin p}{\sqrt{1+\sin^2 p}}\right)
e^{\pm i\omega_p(t-1)};
\end{aligned}
\end{equation}
and $\omega_p:=\arccos\frac{\cos p}{\sqrt{2}}$. Now~\eqref{eq-rewritten-fourier} holds for each $x,t\in\mathbb{Z}$: Indeed, the identity
\new{}\new{}\new{}
\begin{multline*}
\exp\left({-i\omega_{p+\pi}(t-1)+i(p+\pi)(x-1)}\right)
=\exp\left({-i(\pi-\omega_{p})(t-1)+ip(x-1)+i\pi(x-1)}\right)
=\\=(-1)^{(x+t)}\exp\left({i\omega_p(t-1)+ip(x-1)}\right)
\end{multline*}
shows that the contributions of the two summands $\hat a_{k\pm}(p,t)$ to integral~\eqref{eq-rewritten-fourier} are equal for $t+x$ even and cancel for $t+x$ odd. The summand
$\hat a_{k-}(p,t)$ contributes $a_k(x,t)/2$ by Proposition~\ref{cor-fourier-integral}.

By the derivative property of Fourier series and the Parseval theorem, we have
\begin{multline}\label{eq-moment}
\sum_{x\in\mathbb{Z}} \frac{x^r}{t^r} P(x,t)
=
\sum_{x\in\mathbb{Z}}
\begin{pmatrix}
  a_1(x,t) \\
  a_2(x,t)
\end{pmatrix}^*
\frac{x^r}{t^r}
\begin{pmatrix}
  a_1(x,t) \\
  a_2(x,t)
\end{pmatrix}
=
\int\limits_{-\pi}^{\pi}
\begin{pmatrix}
  \hat a_1(p,t) \\
  \hat a_2(p,t)
\end{pmatrix}^*
\frac{i^r}{t^r}
\frac{\partial^r}{\partial p^r}
\begin{pmatrix}
  \hat a_1(p,t) \\
  \hat a_2(p,t)
\end{pmatrix}
\,\frac{dp}{2\pi}.
\end{multline}

The derivative is estimated as follows:
\begin{equation}\label{eq-derivatives}
\frac{\partial^r}{\partial p^r}\hat a_{k\pm}(p,t)=
\left(\pm i(t-1)\frac{\partial \omega_p}{\partial p}
\right)^r\hat a_{k\pm}(p,t)+O_r(t^{r-1})=
\left(\pm \frac{i(t-1) \sin p}{\sqrt{1+\sin^2 p}}\right)^r\hat a_{k\pm}(p,t)+O_r(t^{r-1}).
\end{equation}
Indeed, differentiate~\eqref{eq-fourier-transform} $r$  times using the Leibnitz rule. If we differentiate the exponential factor $e^{\pm i\omega_p(t-1)}$ each time, then we get the main term. If we differentiate a factor rather than the exponential $e^{\pm i\omega_p(t-1)}$ at least once, then we get less than $r$ factors of $(t-1)$, hence the resulting term is $O_r(t^{r-1})$ by compactness because it is continuous and $2\pi$-periodic in $p$.

Substituting~\eqref{eq-derivatives} into~\eqref{eq-moment}
we arrive at
\begin{multline*}
\sum_{x\in\mathbb{Z}} \frac{x^r}{t^r} P(x,t)=
\int_{-\pi}^{\pi}
\begin{pmatrix}
  \hat a_1(p,t) \\
  \hat a_2(p,t)
\end{pmatrix}^*
\left(\frac{\sin p}{\sqrt{1+\sin^2 p}}\right)^r
\begin{pmatrix}
  (-1)^r \hat a_{1+}(p,t)+\hat a_{1-}(p,t) \\
  (-1)^r \hat a_{2+}(p,t)+\hat a_{2-}(p,t)
\end{pmatrix}
\,\frac{dp}{2\pi}
+O_r\left(\frac{1}{t}\right)\\
=\int_{-\pi}^{\pi}
\left(\frac{\sin p}{\sqrt{1+\sin^2 p}}\right)^r
\frac{1}{2}\left((-1)^r
\left(1-\frac{\sin p}{\sqrt{1+\sin^2 p}}\right)
+1+\frac{\sin p}{\sqrt{1+\sin^2 p}}
\right)
\,\frac{dp}{2\pi}
+O_r\left(\frac{1}{t}\right)\\
=\int_{-\pi/2}^{\pi/2}
\left(\frac{\sin p}{\sqrt{1+\sin^2 p}}\right)^r
\left(1+\frac{\sin p}{\sqrt{1+\sin^2 p}}\right)
\,\frac{dp}{\pi}
+O_r\left(\frac{1}{t}\right)
=
\int_{-1/\sqrt{2}}^{1/\sqrt{2}}
\frac{v^r\,dv}{\pi(1-v)\sqrt{1-2v^2}}
+O_r\left(\frac{1}{t}\right).
\end{multline*}
Here the 2nd equality follows from
$\hat a_{1\pm }(p,t)^*\hat a_{1\pm}(p,t)+\hat a_{2\pm }(p,t)^*\hat a_{2\pm}(p,t)
=\frac{1}{2}\left(1\mp\frac{\sin p}{\sqrt{1+\sin^2 p}}\right)$
and
$\hat a_{1\pm }(p,t)^*\hat a_{1\mp}(p,t)+\hat a_{2\pm }(p,t)^*\hat a_{2\mp}(p,t)
=0$.
The 3rd one is obtained by the changes of variables $p\mapsto -p$ and $p\mapsto \pi-p$ applied to the integral over $[-\pi/2,\pi/2]$. The 4th one is obtained by the change of variables $v=\sin p/\sqrt{1+\sin^2 p}$ so that $dp=d\arcsin\frac{v}{\sqrt{1-v^2}}=dv/(1-v^2)\sqrt{1-2v^2}$.
\end{proof}

\comment

\addcontentsline{toc}{myshrinkalt}{}

\subsection{Large-time limit near the origin: the circle method (Theorem~\ref{Feynman-convergence})}
\label{ssec-proofs-alternative}

In this subsection we prove Theorem~\ref{Feynman-convergence}. Although the theorem follows easily from Theorem~\ref{th-ergenium} by the Taylor expansion of~\eqref{eq-theta} in $x$ up to order $4$, we present a direct proof relying on Proposition~\ref{cor-coefficients} only.
It uses the \emph{Hardy--Littlewood circle method} in the simplest form, when the main contribution to an integral comes from just $4$ stationary points. 

Let us give the plan of the proof: introduce some notation and state some lemmas.


\begin{lemma}\label{Le_111}
Take integers
\begin{equation}\label{Def-xy}
n\ge k\ge 0, \qquad x:=2k-n, \qquad\text{and}\qquad t:=n+2.
\end{equation}
Then
\begin{align}
\label{a_1}
a_1(x,t)=&2^{{(n-1)}/{ 2}}i^{-k}\hat f(-x),\\
\label{a_2}
a_2(x,t)=&2^{{(n-1)}/{ 2}}i^{-k}\hat f(2-x),
\end{align}
where
\begin{equation}\label{Def-f}
\hat f(q):=
\frac{1}{2\pi}\int_{-\pi}^{\pi} f(p)e^{-ipq}\,dp
\qquad\text{and}\qquad
f(p):=\cos^{n-k}p\,\sin^k p.
\end{equation}
\end{lemma}

Notice that $f(p)$ is \emph{not} proportional to the Fourier series $\sum_{x\in\mathbb{Z}}a_{1}(x,t)e^{-ipx}$ because both $x$ and $f(p)$ depend on $k$, and thus Fourier inversion formula is not applicable.

Throughout this subsection 
assume that
\begin{equation}\label{eq-ineq}
\frac{|x|}{t}<\frac{1}{20}.
\end{equation}
This inequality follows from the assumption $|x|<t^{3/4}$ in Theorem~\ref{Feynman-convergence} for sufficiently large $t$.

The Fourier integral $\hat f(q)$ is estimated in several steps. The main contribution comes from
the sharp extrema of the function $f(p)$. The derivative is
$$f'(p)=f(p)(k\cot p-(n-k)\tan p).$$
Here $k\ne 0,n$ by \eqref{eq-ineq}; we ignore the points where $f(p)=0$ because they cannot be global extrema.
Thus the global extrema on $[-\pi,\pi]$ are the $4$ points
$\pm c$ and $\pm c\mp\pi$, where
\begin{equation}\label{eq-def-c}
c:=\arctan\sqrt{\frac{k}{n-k}}.
\end{equation}

\begin{lemma}\label{l-c}
  Assume \eqref{Def-xy}--\eqref{eq-def-c}. Then $|c-\pi/4|<1/9$ and $9/10<\tan c<10/9$.
\end{lemma}

Decompose $[-\pi,\pi]=E_0\sqcup E_1\sqcup \dots \sqcup E_4$, where
$E_1,\dots,E_4$ are $\delta$-neighborhoods of these extrema and $E_0=[-\pi,\pi]\setminus (E_1\cup \dots\cup E_4)$ is the remaining set. The parameter $\delta$ is going to be determined later.
Write 
\begin{equation}
\label{I(l)}
\hat f(q)=\hat f_0(q)+\dots+\hat f_4(q),
\end{equation}
where
\begin{equation}\label{eq-def-fj}
\hat f_j(q)=\frac{1}{2\pi}\int_{E_j}f(p)e^{-ipq}\,dp
\qquad \text{for each }j=0,1,2,3,4.
\end{equation}


We start with a rough estimate of the function $f(p)$ through the distance to the extremum~$c$.

\begin{lemma}\label{Le:I_3}
Assume \eqref{Def-xy}--\eqref{eq-def-fj} and $0<c+p<\pi/2$. Then
$$
f(c+p)\le f(c)e^{-np^2/4}.
$$
\end{lemma}

This immediately gives an estimate for the integral on the set $E_0$:
\begin{equation}
\label{I_3}
\hat f_0(q)= f(c)O(e^{-n\delta^2/4}).
\end{equation}

Next we expand the function $f(p)$ in a neighborhood of $c$. We consider complex $p$ as well.

\begin{lemma}\label{Le:asymp} Assume \eqref{Def-xy}--\eqref{eq-def-fj}. Then for each complex number $p$ with $|p|<1/2$ 
we have
\begin{equation}
\label{asymp}
f(c+p)=f(c)\exp\left({-np^2}+O(|x|\,|p|^3+t|p|^4)\right).
\end{equation}
\end{lemma}

Recall that notation $f(p,x,t)=\exp (O(g(p,x,t)))$ means that there is a number $C>0$ (not depending on $p,x,t$) and a complex-valued function $h(p,x,t)$ such that $f(p,x,t)=e^{h(p,x,t)}$ and $|h(p,x,t)|\le C g(p,x,t)$ for each $p,x,t$ satisfying the assumptions of the lemma.

Next we substitute expansion~\eqref{asymp} into~\eqref{eq-def-fj} and apply the \emph{method of steepest descent}.

\begin{lemma}\label{Le:I_j}
Assume \eqref{Def-xy}--\eqref{eq-def-fj} and $|q|/t<\delta/2<1/8$. 
Then
\begin{equation}
\label{I_12}
\hat f_1(q)=\frac{f(c)}{2\sqrt{\pi n}}\,
\exp\left(
-\frac{q^2}{ 4n}-icq+O(|x|\delta^3+t\delta^4)\right)
\left(1+O\left(e^{-t\delta^2/4}\right)\right)
\end{equation}
\end{lemma}

The integrals $\hat f_2(q),\hat f_3(q),\hat f_4(q)$ are estimated analogously. Finally, applying Lemma~\ref{Le:asymp} for $p=\pi/4-c$, we arrive at the following lemma.

\begin{lemma}\label{Sl2}
Assume \eqref{Def-xy}--\eqref{eq-def-fj}.
Then
\begin{equation}
\label{f(c)}
f(c)=2^{-n/2}
\exp\left(\frac{x^2}{4n}+O\left(\frac{x^4}{ t^3}\right)\right).
\end{equation}
\end{lemma}

Substituting the resulting asymptotic formulae~\eqref{I_12}, \eqref{f(c)}, and estimate~\eqref{I_3} into~\eqref{I(l)}, we obtain an asymptotic formula for $\hat f(q)$ and thus for $a(x,t)$.

\smallskip
Now realize this plan.

\begin{proof}[Proof of Lemma~\ref{Le_111}]
Applying the Cauchy formula to Proposition~\ref{cor-coefficients}, we get
$$a_1(-n+2k,n+2)=\frac{2^{-(n+1)/ 2}}{ 2\pi ir}\int_{\gamma}\frac{(1+z)^{n-k}(1-z)^k}{ z^{n-k+1}}dz,$$
where $\gamma$ is a contour performing $r$ counterclockwise
turns around the origin. Taking the contour $z=e^{2 i p}$, where $p\in[-\pi,\pi]$, performing $r=2$ turns, we get
\begin{gather*}
(1+z)^{n-k}(1-z)^k=2^ni^{-k}e^{ipn}f(p)\qquad\text{and}\qquad z^{-n+k-1}\,dz=2i\, e^{-2i (n-k)p}\,dp.
\end{gather*}
This gives~\eqref{a_1}; analogously one gets~\eqref{a_2}.
\end{proof}

\begin{remark} A contour $\gamma$ performing just one turn gives the formula $a_1(x,t)=\frac{1}{\pi}2^{{(n-1)}/{ 2}}i^{-k}\int_{0}^{\pi} f(p)e^{ipx}\,dp$. It can be alternatively used in the proof of Theorem~\ref{Feynman-convergence}; nothing changes essentially.
\end{remark}

\begin{proof}[Proof of Lemma~\ref{l-c}]
By~\eqref{eq-ineq} we get $(n-2k)/n\le 2|x|/t<1/10$, thus
$9/11<k/(n-k)<11/9$, hence $9/10<\tan c=\sqrt{\frac{ k}{n-k}}<10/9$. By the Lagrange theorem, $\tan c-\tan(\pi/4)=(c-\pi/4)/\cos^2\zeta$ for some $\zeta\in (c,\pi/4)$. Hence
$|c-\pi/4|\le|\tan c-\tan(\pi/4)|< |10/9-1|=1/9$.
\end{proof}

\begin{proof}[Proof of Lemma~\ref{Le:I_3}]
Since the sine and the cosine are concave on the interval $[0,\pi/2]$, they can be estimated from above by a linear function representing the tangent line at the point $c$:
\begin{gather*}
\sin\pi(c+p)\le \sin c(1+p\cot c), \qquad
\cos\pi(c+p)\le \cos c(1-p\tan c).
\end{gather*}
Thus
\begin{gather*}
|f(c+p)|\le f(c)(1+p\cot c)^k(1-p\tan c)^{n-k}=f(c)\exp\left(k\log (1+p\cot c)+(n-k)\log(1-p\tan c)\right).
\end{gather*}
To estimate the logarithms, we apply the inequality
\begin{equation}\label{eq-log}
\log(1+z)\le z-\frac{z^2}{4}\qquad\text{for }z\in (-1;1).
\end{equation}
The inequality follows from
$$
e^{z}e^{-z^2/4}\ge
\left(1+z+\frac{z^2}{2}+\frac{z^3}{6}\right)\left(1-\frac{z^2}{4}\right)
=1+z+\frac{z^2}{4}\left(1-\frac{z}{3}-\frac{z^2}{2}-\frac{z^3}{6}\right)
\ge 1+z.
$$

We are going to apply inequality~\eqref{eq-log} at the points
$z_1=p\cot c$ 
and $z_2=-p\tan c$. 
Let us check that indeed $z_1,z_2\in (-1;1)$.
Since $0<c+p<\pi/2$, by Lemma~\ref{l-c} it follows that
$$
|z_2|=|p|\tan c<\left(\frac{\pi}{4}+\left(\frac{\pi}{4}-c\right)\right)\tan c
<\left(\frac{\pi}{4}+\frac{1}{9}\right)\cdot\frac{10}{9}<1.
$$
Analogously $|z_1|<1$, as required.

By inequality~\eqref{eq-log} we get
\begin{gather*}
f(c+p)\le f(c)
\exp\left(k\left(z_1-\frac{z_1^2}{4}\right)
+(n-k)\left(z_2-\frac{z_2^2}{4}\right)\right)=\\=
f(c)\exp\left(p\sqrt{k(n-k)}-p\sqrt{k(n-k)}
-\frac{p^2}{4}(n-k)-\frac{p^2}{4}k\right)=
f(c)e^{-np^2/4}.\\[-1.8cm]
\end{gather*}
\end{proof}

\begin{proof}[Proof of Lemma~\ref{Le:asymp}]
Take the Taylor expansions with remainders in the Lagrange form: \begin{align}
\label{sin1}
\sin(c+p)&=\sin c+p \cos c-\frac{p^2}{2}\sin c
-\frac{p^3}{6}\cos c+\frac{p^4}{ 24} \sin(c+\zeta p),\\
\label{cos1}
\cos(c+p)&=\cos c-p \sin c-\frac{p^2}{2}\cos c
+\frac{p^3}{6}\sin c+\frac{p^4}{ 24} \cos(c+\eta p),
\end{align}
where $0\le \zeta,\eta\le 1$.
Since $|p|<1/2$ it follows that $|\sin(c+\zeta p)|\le e^{\zeta|p|}<2$ and $|\cos(c+\eta p)|\le e^{\eta|p|}<2$.
By definition of $c$, we have
$$
\sin c=\sqrt{\frac{k}{ n}},\quad
\cos c=\sqrt{\frac{n-k}{ n}},\quad
\tan c=\sqrt{\frac{k}{ n-k}},\quad
\cot c=\sqrt{\frac{n-k}{ k}}.
$$
By~\eqref{eq-ineq} we get $|n-2k|\le n/2$, thus
$n/4\le k\le 3n/4$, hence
$\sin c
\ge 1/2$ and
$\cos c
\ge 1/2$.
Thus
\begin{align}
\label{sin2}
\sin(c+p)&=\sin c
\left(1+p \cot c-\frac{p^2}{2}
-\frac{p^3}{6}\cot c+\frac{|p|^4}{6}\zeta'\right)=:\sin c
\left(1+z_1\right),\\
\label{cos2}
\cos(c+p)&=\cos c
\left(1-p \tan c-\frac{p^2}{2}
+\frac{p^3}{6}\tan c+\frac{|p|^4}{6}\eta'\right)=:\cos c\left(1+z_2\right),
\end{align}
for some complex $\zeta',\eta'$ of absolute value $<1$.
Since $|p|<1/2$ and $9/10<\tan c<10/9$ by Lemma~\ref{l-c}, we get
$$
|z_{1,2}|\le |p| \max\{\cot c,\tan c\}+\frac{|p|^2}{2}
+\frac{|p|^3}{6}\max\{\cot c,\tan c\}+\frac{|p|^4}{6}<\frac{3}{4}.
$$
Substituting expansions~\eqref{sin2} and \eqref{cos2} into the
Taylor expansion
$$
\log(1+z)=z-\frac{z^2}{2}+\frac{z^3}{3}+O(z^4)\qquad\text{for }|z|<3/4,
$$
and then into~\eqref{Def-f}, by a direct computation (available in \cite[\S10]{SU-2}) we get
\begin{gather*}
\frac{f(c+p)}{f(c) }=\exp\left(-np^2+\frac{(n-2k)np^3}{ 3\sqrt{k(n-k)}}+O(n|p|^4)\right)=
\exp\left(-np^2+O(|x||p|^3)+O(t|p|^4)\right)
.
\end{gather*}
Here we used that $n/3<\sqrt{k(n-k)}<n$ by assumption~\eqref{eq-ineq}.
\end{proof}

\begin{proof}[Proof of Lemma~\ref{Le:I_j}]
Denote $g(p):=f(c+p)e^{np^2}/f(c)$, so that
\begin{align*}
\hat f_1(q)
=\int\limits_{-\delta}^{\delta}f(c+p)e^{-i(c+p)q}\,\frac{dp}{2\pi}
=\frac{f(c)}{2\pi}e^{-icq}\int\limits_{-\delta}^{\delta}e^{-np^2-ipq}g(p)
\,dp
=\frac{f(c)}{2\pi}e^{-q^2/4n-icq}
\int\limits_{-\delta}^{\delta}e^{-\left(\sqrt{n}p+ iq/2\sqrt{n}\right)^2}g(p)\,dp.
\end{align*}
Switch to complex variable $z=\sqrt{n}p+iq/2\sqrt{n}$. Then the limits of integration become $B_1:=-\delta \sqrt{n}+iq/2\sqrt{n}$ and $B_2:=\delta \sqrt{n}+iq/2\sqrt{n}$:
\begin{equation}
\label{I_j}
\hat f_1(q)=
\frac{f(c)}{2\pi\sqrt{n}}e^{-q^2/4n-icq}
\int_{B_1}^{B_2}e^{-z^2}g\left(\frac{z}{\sqrt{n}}-\frac{iq}{2n}\right)\,dz.
\end{equation}
The function under the integral is analytic in the whole complex plane, thus the integral over the interval $B_1B_2$ can be replaced by the integral over the broken line $B_1A_1A_2B_2$ with the corners at $A_1:=-\delta\sqrt{n}$ and  $A_2:=\delta\sqrt{n}$.

Estimate the contribution from the integrals over the intervals
$A_1B_1$ and $A_2B_2$. Since $|q|/t<\delta/2<1/8$, it follows that for each $z\in B_1A_1A_2B_2$ we have
\begin{equation}
\label{bound}
\left|\frac{z}{\sqrt{n}}-\frac{iq}{2n}\right|< 2\delta<\frac{1}{2}.
\end{equation}
Then by Lemma~\ref{Le:asymp} for each $z\in B_1A_1A_2B_2$ we get
\begin{equation}\label{eq-g-bound}
g\left(\frac{z}{\sqrt{n}}-\frac{iq}{2n}\right)=
\exp\left(O(|x|\delta^3+t\delta^4)\right).
\end{equation}
Then
\begin{gather*}
\int_{A_\nu}^{B_\nu}
e^{-z^2}g\left(\frac{z}{\sqrt{n}}-\frac{iq}{2n}\right)\,dz
\le
\frac{|q|}{2\sqrt{n}}\max_{z\in A_\nu B_\nu}
\left|e^{-z^2}\right|
\max_{z\in A_\nu B_\nu}\left|g\left(\frac{z}{\sqrt{n}}-\frac{iq}{n}\right)\right|
=\\
=\frac{|q|}{2\sqrt{n}}
\exp\left(-n\delta^2+\frac{|q|^2}{4n} +O(|x|\delta^3+t\delta^4)\right)
=O\left(\exp\left(-\frac{n\delta^2}{4}\right)\right)
\exp\left( O(|x|\delta^3+t\delta^4)\right).
\end{gather*}
In the latter estimate we used the inequality $|q|/t<\delta/2$, which implies $|q|<n\delta$, so that
$$
\frac{|q|}{2\sqrt{n}}
\exp\left(-n\delta^2+\frac{|q|^2}{4n}\right)\le
\frac{\sqrt{n}\delta}{2} \exp\left(-\frac{n\delta^2}{2}\right)\le
\exp\left(-\frac{n\delta^2}{4}\right).
$$

Now compute the contribution from the integral over the interval $A_1A_2$. Using asymptotic formula~\eqref{eq-g-bound}, let us show that
$$
\int_{A_1}^{A_2}e^{-z^2}
g\left(\frac{z}{\sqrt{n}}-\frac{iq}{2n}\right) \,dz
=
\exp\left( O(|x|\delta^3+t\delta^4)\right)
\int_{-\delta\sqrt{n}}^{\delta\sqrt{n}}e^{-z^2}\,dz.
$$
Indeed, if $|x|\delta^3+t\delta^4\le 2\pi$, then $\exp\left( O(|x|\delta^3+t\delta^4)\right)$ is the same as $1+O(|x|\delta^3+t\delta^4)$, and the formula follows. If $|x|\delta^3+t\delta^4> 2\pi$, then $\exp\left( O(|x|\delta^3+t\delta^4)\right)$ is the same as $O\left( \exp(|x|\delta^3+t\delta^4)\right)$, and the formula follows  again.

It remains to compute
\begin{align*}
\int_{-\delta\sqrt{n}}^{\delta\sqrt{n}}e^{-z^2}\,dz
=
\int_{-\infty}^{+\infty}e^{-z^2}\,dz
-2\int_{\delta\sqrt{n}}^{+\infty}e^{-z^2}\,dz
=
\sqrt{\pi}+O\left(e^{-n\delta^2}\right),
\end{align*}
where we applied the estimate for the \emph{complimentary error function} \cite[Eq.~7.1.13]{Abramowitz-Stegun-92} 
$$
\int_{N}^{\infty}e^{-z^2}\,dz\le\frac{e^{-N^2}}{N+1}
=O\left(e^{-N^2}\right) \qquad \text{for }N>0.
$$
Combining the estimates for the integrals over $A_1B_1$, $A_1A_2$, $A_2B_2$, we get the desired result.
\end{proof}

\begin{proof}[Proof of Lemma~\ref{Sl2}]
Take the Taylor expansion 
$$\arcsin \sqrt{\frac{1}{ 2}+z}=\frac{\pi}{ 4}+z+O(|z|^3)\qquad \text{for }z\in \left(-\frac{1}{4},\frac{1}{4}\right).$$
Substituting $z=x/2n$ so that $|z|=|x|/2n\le |x|/t<1/4$ by~\eqref{eq-ineq}, we get
\begin{equation}
\label{as:t_1}
c=\arcsin \sqrt{\frac{k}{ n}}=\frac{\pi}{ 4}+\frac{x}{2n}+O\left(\frac{|x|^3}{n^3}\right).
\end{equation}
Apply Lemma~\ref{Le:asymp} for $p=\pi/4-c$. Then by~\eqref{asymp} we get
\begin{gather*}
f\left(\frac{\pi}{ 4}\right)=f(c)\exp\left(-n\left(\frac{x}{2 n}+O\left(\frac{|x|^3}{ n^3}\right)\right)^2+O\left(\frac{x^4}{n^3}\right)\right)
=f(c)\exp\left(-\frac{x^2}{ 4n}+O\left(\frac{x^4}{ t^3}\right)\right).
\end{gather*}
It remains to notice that
\begin{gather*}
f\left(\frac{\pi}{ 4}\right)=\left(\sin\frac{\pi}{ 4}\right)^k\left(\cos\frac{\pi}{ 4}\right)^{n-k}=2^{-{n}/{ 2}}.\\[-1.5cm]
\end{gather*}
\end{proof}

\begin{proof}[Proof of Theorem~\ref{Feynman-convergence}]
We may assume $t>10^6$ because
only finitely many pairs $(x,t)$ with $t\le 10^6$ satisfy the assumptions of the theorem. For $t>10^6$ the assumption $|x|<t^{3/4}$ implies inequality~\eqref{eq-ineq}.

We need to find asymptotic formulae for the integral $\hat f(q)$ for $q=-x$ and $q=2-x$ (see Lemma~\ref{Le_111}). Thus assume further that $|q|\le |x|+2$ and $q+t$ is even.
To guarantee the inequality $|q|/t<\delta/2$ required for Lemma~\ref{Le:I_j}, we need to take $\delta>2(|x|+2)/t$. To guarantee that the right-hand side of~\eqref{I_3} is small enough, we need to take  $\delta>4\sqrt{\tfrac{\log n}{n}}$. The remainder under the first exponential in~\eqref{I_12} increases with $\delta$, thus we assign the least possible value 
$$\delta:=\max\left\{4\sqrt{\frac{\log n}{ n}},2\frac{|x|+2}{ t}\right\}.$$
Then clearly $\delta<1/4$ by the inequalities $t>10^6$ and~\eqref{eq-ineq}. The assumption $|x|<t^{3/4}$ implies that
$|x|\delta^3+t\delta^4=O(1)$, hence
$$
\exp\left(O(|x|\delta^3+t\delta^4)\right)=1+O(R(t)),
\qquad\text{where }\quad
R(t)=t^{-1}\log^{2}t+x^4t^{-3}.
$$
For such choice of $\delta$, formula~\eqref{I_12} becomes
$$
\hat f_1(q)=\frac{f(c)}{2\sqrt{\pi n}}\,
\exp\left(
-\frac{q^2}{ 4n}-icq\right)
\left(1+O(R(t))\right),
$$
Analogous formulae hold for $\hat f_2(q), \hat f_3(q), \hat f_4(q)$ with $c$ replaced by $-c$ and $\pm \pi\mp c$.
Estimate~\eqref{I_3} becomes
\begin{gather*}
\hat f_0(q)= f(c)O(e^{-n\delta^2/4})
=\frac{f(c)}{2\sqrt{\pi n}}\,
e^{-q^2/4n}O\left(\sqrt{n}\exp\left({\frac{q^2}{4n}-\frac{n\delta^2}{4}}\right)\right)
=\frac{f(c)}{2\sqrt{\pi n}}\,e^{-q^2/4n}O(R(t)),
\end{gather*}
because
$
\sqrt{n}\exp\left({\frac{q^2}{4n}-\frac{n\delta^2}{4}}\right)
\le \sqrt{n}\exp\left(-\frac{n\delta^2}{8}\right)
\le n^{-3/2}
$
by the inequalities $|q|\le |x|+2<\delta t/2< \delta \sqrt{2}n/2$ for $t>10^6$ and $\delta>4\sqrt{\tfrac{\log n}{n}}$.

Substituting the resulting formulae for $\hat f_0(q),\dots,\hat f_4(q)$ into~\eqref{I(l)}, then applying Lemma~\ref{Sl2} and substituting~\eqref{as:t_1}, and using that $q+n$ even and $|q|\le |x|+2$, we find
\begin{multline*}
\hat f(q)=\frac{f(c)}{ \sqrt{\pi n}}e^{-\frac{q^2}{ 4n}}\left(e^{-icq}+(-1)^k e^{icq}+O(R(t))\right)
=2i^k\frac{f(c)}{ \sqrt{\pi n}}e^{-\frac{q^2}{ 4n}} \left(\cos\left(cq+\frac{\pi k}{ 2}\right)+O(R(t))\right)\\
=i^k\frac{2^{\frac{2-n}{ 2}}}{ \sqrt{\pi n}}e^{\frac{x^2-q^2}{4 n}} \left(\cos\left(cq+\frac{\pi k}{ 2}\right)+O(R(t))\right)
=i^k\frac{2^{\frac{2-n}{ 2}}}{ \sqrt{\pi n}}e^{\frac{x^2-q^2}{4 n}}\left(\cos\left(\frac{\pi (2k+q)}{ 4}+\frac{xq}{2 n}\right)+O\left(R(t)\right)\right).
\end{multline*}
%
Substituting the resulting expression into equations~\eqref{a_1}--\eqref{a_2}, we get
\begin{align*}
  a_1(x,t) &= \sqrt{\frac{2}{\pi n}}
  \left(\sin \left(\frac{\pi t}{ 4}-\frac{x^2}{2 n}\right)+O\left(R(t)\right)\right); \\
  a_2(x,t) &= \sqrt{\frac{2}{\pi n}}e^{(x-1)/n}
  \left(\cos \left(\frac{\pi t}{ 4}-\frac{x^2}{2 n}
  +\frac{x}{ n}\right)+O\left(R(t)\right)\right).
\end{align*}

Let us simplify these expressions further. Denote
$\theta:=\pi t/ 4-x^2/2 n$. By the Cauchy inequality, $x^2/t^2\le 1/t+{x^4}/{t^3}\le R(t)$ for all $x$. Thus $e^{(x-1)/n}=1+{x}/{n}+O(R(t))$ and
\begin{gather*}
\cos\left(\theta+\frac{x}{ n}\right)
=\cos\theta\cos\frac{x}{n} -\sin\theta\sin\frac{x}{n}
=\cos\theta-\frac{x}{ n}\sin\theta+O\left(R(t)\right).
\end{gather*}
This allows to rewrite
\begin{align*}
  a_2(x,t) = \sqrt{\frac{2}{\pi n}}
  \left(1+\frac{x}{n}\right)\left(\cos\theta-\frac{x}{n} \sin\theta+O\left(R(t)\right)\right)
  =
  \sqrt{\frac{2}{\pi n}}
  \left(\cos\theta+\frac{\sqrt{2}x}{n} \cos\left(\theta+\frac{\pi}{4}\right)+O\left(R(t)\right)\right).
\end{align*}
Here we can replace $n$ by $t=n+2$ (inside $\theta$ as well) because
\begin{gather*}
  \frac{1}{\sqrt n} = \frac{1}{\sqrt t}
  \sqrt{\frac{n+2}{n}}=\frac{1}{\sqrt t}
  \left(1+O\left(\frac{1}{n}\right)\right)
  =\frac{1}{\sqrt t} \left(1+O\left(R(t)\right)\right), \\
  \frac{x}{ n} = \frac{x}{t}\cdot\frac{n+2}{n}
  =\frac{x}{t}\left(1+O\left(R(t)\right)\right),
  \qquad\text{and}\qquad
  \frac{x^2}{ n} = \frac{x^2}{t}+\frac{2x^2}{nt}=\frac{x^2}{t}+O\left(R(t)\right).
\end{gather*}
We arrive at the required result.
\end{proof}


\endcomment

\addcontentsline{toc}{myshrinkalt}{}

\subsection{The main result: the stationary phase method (Theorem~\ref{th-ergenium})}
\label{ssec-proofs-main}

In this subsection we prove Theorem~\ref{th-ergenium}. 
First we
outline the plan of the argument, then prove the theorem modulo some technical lemmas, and finally the lemmas themselves.

The plan is to apply the Fourier transform and the \emph{stationary phase method} to the resulting oscillatory integral. The proof has $4$ steps, with the first two known before (\cite[\S4]{Ambainis-etal-01}, \cite[\S2]{Sunada-Tate-12}):
\begin{description}
  \item[Step 1:] computing the main term in the asymptotic formula;
  \item[Step 2:] estimating approximation error arising from neighborhoods of stationary points;
  \item[Step 3:] estimating approximation error arising from a neighborhood of the origin; 
  \item[Step 4:] estimating error arising from the complements to those neighborhoods.
\end{description}
\smallskip

\begin{proof}[Proof of Theorem~\ref{th-ergenium} modulo some lemmas]
Derive the asymptotic formula for
$a_1\left(x,t+\varepsilon,{m},{\varepsilon}\right)$; the derivation for $a_2\left(x+\varepsilon,t+\varepsilon,{m},{\varepsilon}\right)$ is analogous and is discussed at the end of the proof.
By Proposition~\ref{cor-fourier-integral}
and the identity $e^{i\omega_{p+\pi/\varepsilon}t-i(p+\pi/\varepsilon)x}
=e^{i(\pi/\varepsilon-\omega_{p})t-ipx-i\pi x/\varepsilon}
=-e^{-i\omega_pt-ipx}$
for $(t+x)/\varepsilon$ odd, we get
\begin{equation}
a_1(x,t+\varepsilon,m,\varepsilon)=
  \frac{m\varepsilon^2}{2\pi i}
  \int_{-\pi/\varepsilon}^{\pi/\varepsilon}
  \frac{e^{i\omega_pt-i p x}} {\sqrt{m^2\varepsilon^2+\sin^2(p\varepsilon)}}\,dp
  =
  \int_{-\pi/2\varepsilon}^{\pi/2\varepsilon}g(p)(e(f_+(p))-e(f_-(p)))\,dp,
  \label{eq-oscillatory-integral}
\end{equation}
where $e(z):=e^{2\pi i z}$ and
\begin{align}\label{eq-parameter-values0}
f_{\pm}(p)&=
\frac{1}{2\pi}
\left(-px\pm\frac{t}{\varepsilon}
\arccos\frac{\cos(p\varepsilon)}{\sqrt{1+m^2\varepsilon^2}}\right),
\\ \label{eq-parameter-values1/2}
g(p)&=
\frac{m\varepsilon^2}{2\pi i \sqrt{m^2\varepsilon^2+\sin^2(p\varepsilon)}}.
\end{align}


\smallskip
\textbf{Step 1.} We estimate oscillatory integral~\eqref{eq-oscillatory-integral} using the following known technical result.

\begin{lemma}[Weighted stationary phase integral] \label{l-weighted-stationary-phase}
\textup{\cite[Lemma~5.5.6]{Huxley-96}}
Let $f(p)$ be a real function,
four times continuously differentiable for $\alpha\le p\le \beta$, and let $g(p)$ be a real
function, three times continuously differentiable for $\alpha\le p\le \beta$. Suppose that there
are positive parameters $M, N, T, U$, with
$$
M\ge \beta-\alpha, \qquad N\ge M/\sqrt{T},
$$
and positive constants $C_r$ such that, for $\alpha\le p\le \beta$,
$$
\left|f^{(r)}(p)\right|\le C_rT/M^r, \qquad \left|g^{(s)}(p)\right|\le C_sU/N^s,
$$
for $r = 2, 3, 4$ and $s = 0, 1, 2, 3$, and
$$
f'' (p) \ge T/C_2M^2.
$$
Suppose also that $f'(p)$ changes sign from negative to positive at a point $p = \gamma$
with $\alpha< \gamma< \beta$. If $T$ is sufficiently large in terms of the constants $C_r$, then we
have
\begin{multline}\label{eq-l-weighted-stationary-phase}
  \int_\alpha^\beta g(p)e(f(p))\,dp  =\frac{g(\gamma)e(f(\gamma)+1/8)}{\sqrt{f''(\gamma)}}
  +\frac{g(\beta)e(f(\beta))}{2\pi i f'(\beta)}
  -\frac{g(\alpha)e(f(\alpha))}{2\pi i f'(\alpha)}\\
  +{O}_{C_0,\dots,C_4}
  \left(
  \frac{M^4U}{T^2}\left(1+\frac{M}{N}\right)^2
  \left(\frac{1}{(\gamma-\alpha)^3}+\frac{1}{(\beta-\gamma)^3}
  +\frac{\sqrt{T}}{M^3}\right)
  \right).
\end{multline}
\end{lemma}

Here the first term involving the values at the stationary point $\gamma$ is the main term, and the boundary terms involving the values at the endpoints $\alpha$ and $\beta$ are going to cancel out in Step~3.


\begin{lemma}\label{l-stationary-point} \textup{(Cf.~\cite[(25)]{Jacobson-Schulman-84}, \cite[\S4]{Ambainis-etal-01})} Assume~\eqref{eq-case-A};
then on 
$[-\frac{\pi}{2\varepsilon},\frac{\pi}{2\varepsilon}]$, the function $f_{\pm}(p)$ given by~\eqref{eq-parameter-values0} has a unique critical point
\begin{equation}\label{eq-parameter-values3}
\gamma_\pm=\pm\frac{1}{\varepsilon}\arcsin\frac{m\varepsilon x}{\sqrt{t^2-x^2}}.
\end{equation}
\end{lemma}

To estimate integral~\eqref{eq-oscillatory-integral}, we are going to apply Lemma~\ref{l-weighted-stationary-phase} twice, for the functions $f(p)=\pm f_\pm(p)$ in appropriate neighborhoods of their critical points $\gamma_\pm$. In the case of $f(p)=-f_-(p)$, we perform complex conjugation of both sides of~\eqref{eq-l-weighted-stationary-phase}. Then the total contribution of the two resulting main terms is
\begin{equation}\label{eq-main-term}
\mathrm{MainTerm}:=\frac{g(\gamma_+)e(f_{+}(\gamma_+)+1/8)}{\sqrt{f_{+}''(\gamma_+)}}
-\frac{g(\gamma_-)e(f_{-}(\gamma_-)-1/8)}{\sqrt{-f_{-}''(\gamma_-)}}.
\end{equation}
A direct but long computation (see \cite[\S2]{SU-2}) then gives the desired main term in the theorem:

\begin{lemma}\label{l-main-term} \textup{(See~\cite[\S2]{SU-2})} Assume~\eqref{eq-case-A}, \eqref{eq-theta}, \eqref{eq-parameter-values0}--\eqref{eq-parameter-values1/2}, \eqref{eq-parameter-values3}; then expression~\eqref{eq-main-term} equals
\begin{equation*}
\mathrm{MainTerm}=
\varepsilon\sqrt{\frac{2 m}{\pi}}
\left(t^2-(1+m^2\varepsilon^2)x^2\right)^{-1/4}
\sin\theta(x,t,m,\varepsilon).
\end{equation*}
\end{lemma}

\smallskip
\textbf{Step 2}. To estimate the approximation error, we need to specify the particular values of parameters which
Lemma~\ref{l-weighted-stationary-phase} is applied for:
\begin{align}\label{eq-parameter-values1}
M&=N=m, & T&=mt, & U&=\varepsilon.
\end{align}

\begin{lemma}\label{l-checking-the-assumptions}
If $\varepsilon\le 1/m$ then functions~\eqref{eq-parameter-values0}--\eqref{eq-parameter-values1/2}
and  parameters~\eqref{eq-parameter-values1}
satisfy the inequalities
$$
\left|f_\pm^{(r)}(p)\right|\le 3T/M^r, \qquad \left|g^{(s)}(p)\right|\le 3U/N^s \qquad\text{for $p\in\mathbb{R}$, \ $r = 2, 3, 4$, \ $s = 0, 1, 2, 3$.}
$$
\end{lemma}

We also need to specify the interval
\begin{equation}\label{eq-parameter-values4}
[\alpha_\pm,\beta_\pm]:=[\gamma_\pm-m\delta/2,\gamma_\pm+m\delta/2].
\end{equation}
To estimate the derivative $|f_{\pm}''(p)|$ from below, we make sure that we are apart its roots $\pm\pi/2\varepsilon$.

\begin{lemma}\label{l-interval}
Assume~\eqref{eq-case-A}, \eqref{eq-parameter-values3}; then
interval~\eqref{eq-parameter-values4} is contained in $[-\pi/2\varepsilon+m\delta/2,\pi/2\varepsilon-m\delta/2]$.
\end{lemma}

The wise choice of the interval provides the following more technical estimate.

\begin{lemma}\label{l-second-derivative}
Assume~\eqref{eq-case-A}, \eqref{eq-parameter-values0}, \eqref{eq-parameter-values3}, and \eqref{eq-parameter-values4}. Then for each $p\in[\alpha_\pm,\beta_\pm]$  we have
$$
|f_\pm''(p)|\ge \frac{t\delta^{3/2}}{24\pi m}.
$$
\end{lemma}

This gives $|f_\pm''(p)|\ge T/C_2M^2$ for $C_2:=24\pi\delta^{-3/2}$ under  notation~\eqref{eq-parameter-values1}.
Now all the assumptions of Lemma~\ref{l-weighted-stationary-phase} have been verified
($M\ge \beta_\pm-\alpha_\pm$ and $N\ge M/\sqrt{T}$ are automatic because $\delta\le 1$ and $t>C_\delta/m$ by~\eqref{eq-case-A}).
Apply the lemma to $g(p)$ and $\pm f_\pm(p)$ on $[\alpha_\pm,\beta_\pm]$ (the minus sign before $f_-(p)$ guarantees the inequality $f''(p)>0$ and the factor of $i$ inside $g(p)$ is irrelevant for application of the lemma). We arrive at the following estimate for the approximation error on those intervals.

\begin{lemma}\label{l-error-stationary} \textup{(See~\cite[\S4]{SU-2})}
Parameters~\eqref{eq-parameter-values3} and \eqref{eq-parameter-values1}--\eqref{eq-parameter-values4} satisfy
$$
\frac{M^4U}{T^2}\left(1+\frac{M}{N}\right)^2
  \left(\frac{1}{(\gamma_\pm-\alpha_\pm)^3}
  +\frac{1}{(\beta_\pm-\gamma_\pm)^3}
  +\frac{\sqrt{T}}{M^3}\right)
={O}_\delta\left(\frac{\varepsilon}{m^{1/2}t^{3/2}}\right).
$$
\end{lemma}

Although that is only a part of the error, it is already of the same order as in the theorem. 

\smallskip
\textbf{Step 3.}
To estimate the approximation error outside 
$[\alpha_\pm,\beta_\pm]$, we use another known technical result.

\begin{lemma}[Weighted first derivative test] \label{l-weighted-first-derivative-test}
\textup{\cite[Lemma~5.5.5]{Huxley-96}}
Let $f(p)$ be a real function,
three times continuously differentiable for $\alpha\le p\le \beta$, and let $g(p)$ be a real
function, twice continuously differentiable for $\alpha\le p\le \beta$. Suppose that there
are positive parameters $M, N, T, U$, with
$
M\ge \beta-\alpha,
$
and positive constants $C_r$ such that, for $\alpha\le p\le \beta$,
$$
\left|f^{(r)}(p)\right|\le C_rT/M^r, \qquad \left|g^{(s)}(p)\right|\le C_sU/N^s,
$$
for $r = 2, 3$ and $s = 0, 1, 2$.
If $f'(p)$ and $f''(p)$ do not change sign on the interval $[\alpha,\beta]$, then
\begin{multline*}
  \int_\alpha^\beta g(p)e(f(p))\,dp  =\frac{g(\beta)e(f(\beta))}{2\pi i f'(\beta)}
  -\frac{g(\alpha)e(f(\alpha))}{2\pi i f'(\alpha)}\\
  +{O}_{C_0,\dots,C_3}
  \left(
  \frac{TU}{M^2}\left(1+\frac{M}{N}+\frac{M^3\min|f'(p)|}{N^2T}\right)
  \frac{1}{\min|f'(p)|^3}
  \right).
\end{multline*}
\end{lemma}


This lemma in particular requires the interval to be sufficiently small. By this reason we decompose the initial interval $[-\pi/2\varepsilon,\pi/2\varepsilon]$ into a large number of intervals by the points
$$\hspace{-1cm}-\frac{\pi}{2\varepsilon}=
\alpha_{-K}<\beta_{-K}=\alpha_{-K+1}<\beta_{-K+1}
=\alpha_{-K+2}<\dots =\alpha_{i}<\hat\beta_{i}=\alpha_\pm<\beta_\pm=\hat\alpha_j<\beta_j=\dots<
\beta_{K-1}=\frac{\pi}{2\varepsilon}.
$$
Here $\alpha_\pm$ and $\beta_\pm$ are given by~\eqref{eq-parameter-values4} above. The other points are given by
\begin{equation}\label{eq-other-points}
\alpha_k=\frac{k\pi}{2\varepsilon K}, \quad \beta_k=\frac{(k+1)\pi}{2\varepsilon K}, \quad
\text{where}\quad K=2\left\lceil\frac{\pi}{m\varepsilon}\right\rceil
\quad\text{and}\quad k=-K,\dots,i,j+1,\dots,K-1.
\end{equation}
The indices $i$ and $j$ are the minimal ones such that $\frac{(i+1)\pi}{2\varepsilon K}>\alpha_\pm$ and $\frac{(j+1)\pi}{2\varepsilon K}>\beta_\pm$.
Thus all the resulting intervals except $[\alpha_\pm,\beta_\pm]$ and its neighbors have the same length $\frac{\pi}{2\varepsilon K}$. (Although it is more conceptual to decompose using a geometric sequence rather than arithmetic one, this does not affect the final estimate here.)

We have already applied Lemma~\ref{l-weighted-stationary-phase} to $[\alpha_\pm,\beta_\pm]$ and we are going to apply Lemma~\ref{l-weighted-first-derivative-test} to each of the remaining intervals in the decomposition for $f(p)= f_\pm(p)$ (this time it is not necessary to change the sign of $f_-(p)$). After summation of the resulting estimates, all the terms involving the values $f_\pm(\alpha_k)$ and $f_\pm(\beta_k)$ at the endpoints, except the leftmost and the rightmost ones, are going to cancel out.
The remaining boundary terms give
\begin{equation}\label{eq-boundary-term}
\mathrm{BoundaryTerm}:=
\frac{g(\frac{\pi}{2\varepsilon})e(f_+(\frac{\pi}{2\varepsilon}))}
{2\pi i f'_+(\frac{\pi}{2\varepsilon})}
-\frac{g(-\frac{\pi}{2\varepsilon})e(f_+(-\frac{\pi}{2\varepsilon}))}
{2\pi i f'_+(-\frac{\pi}{2\varepsilon})}
-\frac{g(\frac{\pi}{2\varepsilon})e(f_-(\frac{\pi}{2\varepsilon}))}
{2\pi i f'_-(\frac{\pi}{2\varepsilon})}
+\frac{g(-\frac{\pi}{2\varepsilon})e(f_-(-\frac{\pi}{2\varepsilon}))}
{2\pi i f'_-(-\frac{\pi}{2\varepsilon})}.
\end{equation}

\begin{lemma}\label{l-boundary-term} \textup{(See~\cite[\S5]{SU-2})} For $(x,t)\in\varepsilon\mathbb{Z}^2$ such that $(x+t)/\varepsilon$ is odd, expression~\eqref{eq-boundary-term} vanishes.
\end{lemma}

It remains to estimate the error terms.
We start estimates with the central intervals $[\alpha_0,\beta_0]$ and $[\alpha_{-1},\beta_{-1}]$ 
(possibly without parts cut out by $[\alpha_\pm,\beta_\pm]$); they require a special treatment. Apply Lemma~\ref{l-weighted-first-derivative-test} to the intervals for the same functions~\eqref{eq-parameter-values0}--\eqref{eq-parameter-values1/2} and the same values~\eqref{eq-parameter-values1} of the parameters $M,N,T,U$ as in Step 2. All the assumptions of the lemma have been already verified in Lemma~\ref{l-checking-the-assumptions}; we have
$\beta_0-\alpha_0\le {\pi}/{2\varepsilon K}={\pi}/4\varepsilon \lceil\frac{\pi}{m\varepsilon}\rceil<m=M$ and $f''_\pm(p)\ne 0$ as well.
We are thus left to estimate $|f_\pm'(p)|$ from below.

\begin{lemma}\label{l-first-derivative}
Assume~\eqref{eq-case-A}, \eqref{eq-parameter-values0}, \eqref{eq-parameter-values3}, \eqref{eq-parameter-values4}; then for each
$p\in [-\pi /2\varepsilon,\pi/2\varepsilon] \setminus[\alpha_\pm,\beta_\pm]$
we get $$|f_{\pm}'(p)|\ge t\delta^{5/2}/48\pi.$$
\end{lemma}

Then the approximation error on the central intervals is estimated as follows.

\begin{lemma}\label{l-error-central} \textup{(See~\cite[\S6]{SU-2})}
Assume~\eqref{eq-case-A}, \eqref{eq-parameter-values3}, and \eqref{eq-parameter-values4}. Then parameters~\eqref{eq-parameter-values1} and functions~\eqref{eq-parameter-values0} satisfy
$$
\frac{TU}{M^2}\left(1+\frac{M}{N}
+\frac{M^3\min_{p\in[\alpha_\pm,\beta_\pm]}|f_\pm'(p)|}{N^2T}\right)
  \frac{1}{\min_{p\in[\alpha_\pm,\beta_\pm]}|f_\pm'(p)|^3}
={O}\left(\frac{\varepsilon}{mt^{2}\delta^{15/2}}\right).
$$
\end{lemma}

This value is ${O}_\delta\left({\varepsilon}/{m^{1/2}t^{3/2}}\right)$ because $t>C_\delta/m$ by~\eqref{eq-case-A}. Hence the approximation error on the central intervals is within the remainder of the theorem.

\bigskip\textbf{Step 4.} To estimate the approximation error on the other intervals $[\alpha_k,\beta_k]$, where we assume that $k>0$ without loss of generality, we apply Lemma~\ref{l-weighted-first-derivative-test} with slightly different parameters:
\begin{align}\label{eq-parameter-values-alt}
T&=mt/k, & M&=mk, & U&=\varepsilon/k, & N&=mk. 
\end{align}

\begin{lemma}\label{l-checking-the-assumptions-alt}
For $0<k<K$ and $\varepsilon\le 1/m$, parameters~\eqref{eq-parameter-values-alt} and \eqref{eq-other-points},  functions~\eqref{eq-parameter-values0}--\eqref{eq-parameter-values1/2} on $[\alpha_k,\beta_k]$ satisfy all the assumptions of Lemma~\ref{l-weighted-first-derivative-test} possibly except the
one on the sign of $f'(p)$.
\end{lemma}

Since the neighborhood $[\alpha_\pm,\beta_\pm]$ of the root of $f'(p)$ is cut out, it follows that $f'(p)$ has constant sign on the remaining intervals, and by Lemma~\ref{l-weighted-first-derivative-test} their contribution to the error is estimated as follows.

\begin{lemma}\label{l-error-alt} \textup{(See~\cite[\S7]{SU-2})}
Assume~\eqref{eq-case-A}, \eqref{eq-parameter-values3}, \eqref{eq-parameter-values4}, $0<k<K$. Then
functions~\eqref{eq-parameter-values0} and parameters~\eqref{eq-parameter-values-alt} satisfy
$$
\frac{TU}{M^2}\left(1+\frac{M}{N}
+\frac{M^3\min_{p\not\in[\alpha_\pm,\beta_\pm]}|f_\pm'(p)|}{N^2T}\right)
  \frac{1}{\min_{p\not\in[\alpha_\pm,\beta_\pm]}|f_\pm'(p)|^3}
={O}\left(\frac{\varepsilon}{k^2mt^{2}\delta^{15/2}}\right).
$$
\end{lemma}

Summation over all $k$ gives the approximation error
$$
\sum_{k=1}^{K}
{O}\left(\frac{\varepsilon}{k^2mt^2\delta^{15/2}}\right)=
{O}\left(\frac{\varepsilon}{mt^2\delta^{15/2}}
\sum_{k=1}^{\infty}\frac{1}{k^2}\right)
={O}_\delta\left(\frac{\varepsilon}{m^{1/2}t^{3/2}}\right).
$$
because the series inside big-O converges and $t>C_\delta/m$. Thus the total approximation error on all the intervals is within the \new{error term in} the theorem, which completes the proof \new{}
of~\eqref{eq-ergenium-re}.\new{}

The derivation of the asymptotic formula for $a_2\left(x+\varepsilon,t+\varepsilon,{m},{\varepsilon}\right)$ is analogous. By Proposition~\ref{cor-fourier-integral} for $(x+t)/\varepsilon$ even we get 
\begin{multline}\label{eq-oscillatory-integral-a2}
  a_2\left(x+\varepsilon,t+\varepsilon,{m},{\varepsilon}\right)=
  \frac{\varepsilon}{2\pi}\int_{-\pi/\varepsilon}^{\pi/\varepsilon}
  \left(1+\frac{\sin (p\varepsilon)} {\sqrt{m^2\varepsilon^2+\sin^2(p\varepsilon)}}\right)
  e^{i\omega_pt-i p x}\,dp\\
  =
  \int_{-\pi/2\varepsilon}^{\pi/2\varepsilon}
  \left[g_+(p)e(f_+(p))+g_-(p)e(f_-(p))\right]dp,
\end{multline}
where $f_\pm(p)$ are the same as above (see~\eqref{eq-parameter-values0}) and
\begin{align}\label{eq-parameter-values7}
g_\pm(p)&=\frac{\varepsilon}{2\pi}\left(1\pm
\frac{\sin (p\varepsilon)} {\sqrt{m^2\varepsilon^2+\sin^2(p\varepsilon)}}\right).
\end{align}
One repeats the argument of Steps~1--4 with $g(p)$ replaced by $g_\pm(p)$. The particular form of $g(p)$ was only used in Lemmas~\ref{l-main-term}, \ref{l-checking-the-assumptions}, \ref{l-boundary-term}, \ref{l-checking-the-assumptions-alt}. The analogues of Lemmas~\ref{l-main-term} and~\ref{l-boundary-term} for $g_\pm(p)$ are obtained by direct checking \cite[\S13]{SU-2}. Lemma~\ref{l-checking-the-assumptions} holds for $g_\pm(p)$: one needs not to repeat the proof because $g_\pm(p)=(\varepsilon/t)(f_\pm'(p)+(x+t)/2\pi)$ \cite[\S1]{SU-2}.
But parameters~\eqref{eq-parameter-values-alt} and Lemma~\ref{l-checking-the-assumptions-alt} should be replaced by the following ones (then the analogue of Lemma~\ref{l-error-alt} 
holds):
\begin{align}\label{eq-parameter-values-alt-a2}
T&=mt/k, & M&=mk, & U&=\varepsilon, & N&=mk^{3/2}. 
\end{align}

\begin{lemma}\label{l-checking-the-assumptions-alt-a2}
For $0<k<K$ and $\varepsilon\le 1/m$, parameters~\eqref{eq-parameter-values-alt-a2} and \eqref{eq-other-points},  functions~\eqref{eq-parameter-values0} and~\eqref{eq-parameter-values7} on $[\alpha_k,\beta_k]$ satisfy all the assumptions of Lemma~\ref{l-weighted-first-derivative-test} possibly except the
one on the sign of~$f'(p)$.
\end{lemma}

\new{Again, one needs not to repeat the proof:
Lemma~\ref{l-checking-the-assumptions-alt-a2} follows from Lemma~\ref{l-checking-the-assumptions-alt} and the expression of $g_\pm$ through $f_\pm'$.
}

This concludes the proof of Theorem~\ref{th-ergenium} modulo the lemmas.
\end{proof}

\bigskip
Now we prove the lemmas. Lemmas~\ref{l-main-term}, \ref{l-error-stationary}, \ref{l-boundary-term}, \ref{l-error-central}, 
\ref{l-error-alt} are proved by direct checking
~\cite{SU-2}. The following expressions \cite[\S1,3]{SU-2} are used frequently in the proofs of the other lemmas:
\begin{align}\label{eq-f'}
f_\pm'(p)&=\frac{1}{2\pi}\left(-x\pm\frac{t \sin  p\varepsilon} {\sqrt{ m^2\varepsilon^2+\sin^2 p\varepsilon}}\right);\\
\label{eq-f''}
f_\pm''(p)&=
\pm\frac{ m^2 \varepsilon^3 t \cos (p \varepsilon)}
{2 \pi \left(m^2\varepsilon^2+\sin^2(p\varepsilon)\right)^{3/2}}
\end{align}

\begin{proof}[Proof of Lemma~\ref{l-stationary-point}]
Using~\eqref{eq-f'} and solving the quadratic equation $f_\pm'(p)=0$ in $\sin p\varepsilon$, we get~\eqref{eq-parameter-values3}. The assumption~$|x|/t<1/\sqrt{1+m^2\varepsilon^2}$ from~\eqref{eq-case-A} guarantees that the arcsine exists.
\end{proof}


\begin{proof}[Proof of Lemma~\ref{l-checking-the-assumptions}]
By the computation of the derivatives in~\cite[\S3]{SU-2}
and the assumption $m\varepsilon\le 1$
we get
\begin{align*}
|g(p)|&=\frac{m\varepsilon^2}
{2\pi\sqrt{m^2\varepsilon^2+\sin^2(p\varepsilon)}}
\le\frac{m\varepsilon^2}
{2\pi\sqrt{m^2\varepsilon^2+0}}
\le \varepsilon=U, \\
|g^{(1)}(p)|&=\frac{m\varepsilon^3
|\sin(p\varepsilon)\cos(p\varepsilon)|}
{2\pi\left(m^2\varepsilon^2+\sin^2(p\varepsilon)\right)^{3/2}}
\le\frac{m\varepsilon^3|\sin(p\varepsilon)\cos(p\varepsilon)|}
{2\pi\left(m^2\varepsilon^2+0\right)
\left(0+\sin^2(p\varepsilon)\right)^{1/2}}
\le \frac{\varepsilon}{m}=\frac{U}{N}, \\
|g^{(2)}(p)|&=\frac{m\varepsilon^4
\left|m^2\varepsilon^2+\sin^4(p \varepsilon) -2(1+m^2\varepsilon^2)\sin^2(p \varepsilon)\right|}
{2\pi\left(m^2\varepsilon^2+\sin^2(p\varepsilon)\right)^{5/2}}
\le \frac{m\varepsilon^4
\left(3m^2\varepsilon^2+3\sin^2(p \varepsilon) 
\right)}
{2\pi\left(m^2\varepsilon^2+\sin^2(p\varepsilon)\right)^{5/2}}
\le \frac{\varepsilon}{m^2}=\frac{U}{N^2}, \\
|g^{(3)}(p)|&=\frac{m\varepsilon^5
|\sin(p\varepsilon)\cos(p\varepsilon)|\cdot
\left|4m^4\varepsilon^4+9m^2\varepsilon^2+\sin^4(p \varepsilon)
-(6+10m^2\varepsilon^2)\sin^2(p \varepsilon)\right|}
{2\pi\left(m^2\varepsilon^2+\sin^2(p\varepsilon)\right)^{7/2}}
\le\frac{3\varepsilon}{m^3}=\frac{3U}{N^3},
\\
|f^{(2)}_\pm(p)|&=
\frac{ m^2 \varepsilon^3 t \cos (p \varepsilon)}
{2 \pi \left(m^2\varepsilon^2+\sin^2(p\varepsilon)\right)^{3/2}}
\le \frac{t}{m}=\frac{T}{M^2}, \\
|f^{(3)}_\pm(p)|&=\frac{m^2 \varepsilon^4 t |\sin (p \varepsilon)| \left(m^2\varepsilon^2+\cos (2 p \varepsilon)+2\right)}
{2 \pi\left(m^2\varepsilon^2+\sin^2(p\varepsilon)\right)^{5/2}}
\le \frac{4m^2 \varepsilon^4 t |\sin (p \varepsilon)| }
{2 \pi\left(m^2\varepsilon^2+\sin^2(p\varepsilon)\right)^{5/2}}
\le \frac{t}{m^2}=\frac{T}{M^3},\\
|f^{(4)}_\pm(p)|&=\frac{m^2\varepsilon^5t\cos(p\varepsilon) \left|m^4\varepsilon^4+3m^2\varepsilon^2
+4\sin^4(p \varepsilon)-2(6+5m^2\varepsilon^2)\sin^2(p \varepsilon)\right|}
{2 \pi\left(m^2\varepsilon^2+\sin^2(p\varepsilon)\right)^{7/2}}
\le\frac{3t}{m^3}=\frac{3T}{M^4}.\\[-1.4cm]
\end{align*}
\end{proof}

\begin{proof}[Proof of Lemma~\ref{l-interval}]
The lemma follows from the sequence of estimates
\begin{multline*}
\frac{\pi}{2\varepsilon}-|\gamma_\pm|
= \frac{\sin(\pi/2)-\sin|\gamma_\pm\varepsilon|}
{\varepsilon\cos(\theta\varepsilon)}
\ge \frac{\sin(\pi/2)-\sin|\gamma_\pm\varepsilon|}
{\varepsilon\cos(\gamma_\pm\varepsilon)}
= \frac{1-m\varepsilon |x|/\sqrt{t^2-x^2}} {\varepsilon\sqrt{1-m^2\varepsilon^2 x^2/(t^2-x^2)}}\\
= \frac{\sqrt{1-x^2/t^2}-m\varepsilon |x|/t} {\varepsilon\sqrt{1-(1+m^2\varepsilon^2) x^2/t^2}}
\ge\frac{1}{\varepsilon}\left(\sqrt{1-\frac{x^2}{t^2}}- \frac{m\varepsilon |x|}{t}\right)\\
\ge
\frac{1}{\varepsilon}\left(\sqrt{1-\frac{1}{1+m^2\varepsilon^2}}
-\frac{m\varepsilon |x|}{t}\right)
= m \left(\frac{1}{\sqrt{1+m^2\varepsilon^2}}-\frac{|x|}{t}\right)
\ge {m\delta}.
\end{multline*}
Here the first equality holds for some $\theta\in[|\gamma_\pm|,\pi/2\varepsilon]$ by the Lagrange theorem. The next inequality holds because the cosine is decreasing on the interval. The next one is obtained by substituting~\eqref{eq-parameter-values3}. The rest is straightforward because $|x|/t<1/\sqrt{1+m^2\varepsilon^2}-\delta$ by~\eqref{eq-case-A}.
\end{proof}

\begin{proof}[Proof of Lemma~\ref{l-second-derivative}]
Let us prove the lemma for $f_+(p)$ and $\gamma_+\ge 0$; the other signs are considered analogously. Omit the index $+$ in the notation of $f_+,\alpha_+,\beta_+,\gamma_+$. The lemma follows from 
\begin{align*}
|f^{(2)}(p)|&\overset{(*)}{\ge} |f^{(2)}(\beta)|=
\frac{ m^2 \varepsilon^3 t \cos (\beta \varepsilon)}
{2 \pi \left(m^2\varepsilon^2+\sin^2(\beta\varepsilon)\right)^{3/2}}
\overset{(**)}{\ge}
\frac{ m^2 \varepsilon^3 t \cos (\gamma \varepsilon)}
{4 \pi \left(m^2\varepsilon^2+\sin^2(\gamma\varepsilon)
+ 2m^2\varepsilon^2t^2/(t^2-x^2)\right)^{3/2}}\\
&\overset{(***)}{=}
\frac{ m^2 \varepsilon^3 t \sqrt{t^2-(1+m^2\varepsilon^2)x^2}(t^2-x^2)}
{4 \pi \left(3m^2\varepsilon^2t^2\right)^{3/2}}
\ge
\frac{t\sqrt{1-(1+m^2\varepsilon^2)x^2/t^2}(1-x^2/t^2)}
{24 \pi m }
\ge \frac{t\delta^{3/2}}{24 \pi m }.
\end{align*}

Here inequality $(*)$ is proved as follows. By~\eqref{eq-f''},  $f^{(2)}(p)$ is increasing on $[-\pi/2\varepsilon,0]$ and decreasing on $[0,\pi/2\varepsilon]$, because it is even, the numerator is decreasing on $[0,\pi/2\varepsilon]$ and the denominator is increasing on $[0,\pi/2\varepsilon]$.
Thus $|f^{(2)}(p)|\ge \min\{|f^{(2)}(\beta)|,|f^{(2)}(\alpha)|\}$ for $p\in[\alpha,\beta]$ by Lemma~\ref{l-interval}. But since $f^{(2)}(p)$ is even and $\gamma\ge 0$, by~\eqref{eq-parameter-values4} we get
$$
|f^{(2)}(\alpha)|=|f^{(2)}(\gamma-m\delta/2)|=
|f^{(2)}(|\gamma-m\delta/2|)|\ge |f^{(2)}(\gamma+m\delta/2)|
=|f^{(2)}(\beta)|.
$$

Inequality $(**)$ follows from the following two estimates. First, by Lemma~\ref{l-interval} and the convexity of the cosine on the interval $[\gamma\varepsilon,\pi/2]$ we obtain
$$
\cos(\beta\varepsilon)
\ge \cos\left(\frac{\gamma\varepsilon}{2}+\frac{\pi}{4}\right)
\ge \frac{1}{2}
\left(\cos (\gamma\varepsilon)+\cos\frac{\pi}{2}\right)
=\frac{\cos (\gamma\varepsilon)}{2}.
$$
Second, using the inequality $\sin z-\sin w\le z-w$ for $0\le w\le z\le \pi/2$, then $\delta\le 1$ and~\eqref{eq-parameter-values3}--\eqref{eq-parameter-values4}, we get
\begin{multline*}
  \sin^2(\beta\varepsilon)-\sin^2(\gamma\varepsilon)
  \le
  \varepsilon(\beta-\gamma)
  \left(\sin(\beta\varepsilon)+\sin(\gamma\varepsilon)\right)
  \le \varepsilon(\beta-\gamma)
  \left(\varepsilon(\beta-\gamma)+2\sin(\gamma\varepsilon)\right)\\
  = \frac{m\varepsilon\delta}{2}
  \left(\frac{m\varepsilon\delta}{2}
  +\frac{2m\varepsilon x}{\sqrt{t^2-x^2}}\right)
  \le \frac{m\varepsilon t}{2\sqrt{t^2-x^2}}
  \left(\frac{2m\varepsilon t}{\sqrt{t^2-x^2}}
  +\frac{2m\varepsilon t}{\sqrt{t^2-x^2}}\right)
  =\frac{2m^2\varepsilon^2t^2}{t^2-x^2}.
\end{multline*}

Equality $(***)$ is obtained from~\eqref{eq-parameter-values3}. The remaining estimates are straightforward.
\end{proof}

\begin{proof}[Proof of Lemma~\ref{l-first-derivative}]
By Lemmas~\ref{l-stationary-point} and~\ref{l-second-derivative}, for $p\in [\beta_+,\pi/2\varepsilon]$ we have
$$
f_+'(p)=f_+'(\gamma_+)+\int_{\gamma_+}^{p}f_+''(p)\,dp\ge 0+(p-\gamma_+)\frac{t\delta^{3/2}}{24 \pi m }
\ge (\beta_+-\gamma_+)\frac{t\delta^{3/2}}{24 \pi m }
=\frac{t\delta^{5/2}}{48 \pi}
$$
because
$
f_+''(p)\ge 0
$
by~\eqref{eq-f''}.
For $p\in [-\pi/2\varepsilon,\alpha_+]$ and for $f'_-(p)$ the proof is analogous.
\end{proof}

\begin{proof}[Proof of Lemma~\ref{l-checking-the-assumptions-alt}] Take $p\in [\alpha_k,\beta_k]$.
By~\eqref{eq-other-points}, the inequalities $\sin z\ge z/2$ for $z\in [0,\pi/2]$, and $m\varepsilon\le 1$ we get
\begin{align*}
|g(p)|&=\frac{m\varepsilon^2}
{2\pi\sqrt{m^2\varepsilon^2+\sin^2(p\varepsilon)}}
\le\frac{m\varepsilon^2}{2\pi\sin(p\varepsilon)}
\le\frac{m\varepsilon^2}{\pi p\varepsilon}
\le\frac{2m\varepsilon^2 K}{\pi^2 k}
=\frac{4m\varepsilon^2}{\pi^2 k} \left\lceil\frac{\pi}{m\varepsilon}\right\rceil
={O}\left(\frac{\varepsilon}{k}\right)
={O}\left({U}\right), \\
|g^{(1)}(p)|&=\frac{m\varepsilon^3
|\sin(p\varepsilon)\cos(p\varepsilon)|}
{2\pi\left(m^2\varepsilon^2+\sin^2(p\varepsilon)\right)^{3/2}}
\le\frac{m\varepsilon^3}{2\pi\sin^2(p\varepsilon)}
={O}\left(\frac{\varepsilon}{mk^2}\right)
={O}\left(\frac{U}{N}\right), \\
|g^{(2)}(p)|&=\frac{m\varepsilon^4
\left|m^2\varepsilon^2+\sin^4(p \varepsilon) -2(1+m^2\varepsilon^2)\sin^2(p \varepsilon)\right|}
{2\pi\left(m^2\varepsilon^2+\sin^2(p\varepsilon)\right)^{5/2}}
\le \frac{m\varepsilon^4
\left(3m^2\varepsilon^2+3\sin^2(p \varepsilon) 
\right)}
{2\pi\left(m^2\varepsilon^2+\sin^2(p\varepsilon)\right)^{5/2}}
={O}\left(\frac{U}{N^2}\right), \\
|g^{(3)}(p)|&=\frac{m\varepsilon^5
|\sin(p\varepsilon)\cos(p\varepsilon)|\cdot
\left|4m^4\varepsilon^4+9m^2\varepsilon^2+\sin^4(p \varepsilon)
-(6+10m^2\varepsilon^2)\sin^2(p \varepsilon)\right|}
{2\pi\left(m^2\varepsilon^2+\sin^2(p\varepsilon)\right)^{7/2}}
={O}\left(\frac{U}{N^3}\right),
\\
|f^{(2)}_\pm(p)|&=
\frac{ m^2 \varepsilon^3 t \cos (p \varepsilon)}
{2 \pi \left(m^2\varepsilon^2+\sin^2(p\varepsilon)\right)^{3/2}}
={O}\left(\frac{t}{mk^3}\right)
={O}\left(\frac{T}{M^2}\right), \\
|f^{(3)}_\pm(p)|
&=\frac{m^2 \varepsilon^4 t |\sin (p \varepsilon)| \left(m^2\varepsilon^2+\cos (2 p \varepsilon)+2\right)}
{2 \pi\left(m^2\varepsilon^2+\sin^2(p\varepsilon)\right)^{5/2}}
={O}\left(\frac{t}{m^2k^4}\right)
={O}\left(\frac{T}{M^3}\right),\\
|f^{(4)}_\pm(p)|&=\frac{m^2\varepsilon^5t\cos(p\varepsilon) \left|m^4\varepsilon^4+3m^2\varepsilon^2
+4\sin^4(p \varepsilon)-2(6+5m^2\varepsilon^2)\sin^2(p \varepsilon)\right|}
{2 \pi\left(m^2\varepsilon^2+\sin^2(p\varepsilon)\right)^{7/2}}
={O}\left(\frac{t}{m^3k^5}\right)
={O}\left(\frac{T}{M^4}\right).
\end{align*}
Further, $f_\pm''(p)$ does not change sign on the interval $[\alpha_k,\beta_k]$ because it vanishes only at $\pm\pi/2\varepsilon$. We also have
$\beta_k-\alpha_k\le {\pi}/{2\varepsilon K}={\pi}/4\varepsilon \lceil\frac{\pi}{m\varepsilon}\rceil<m\le mk=M$.
\end{proof}



\begin{remark} \label{rem-rough}
Analogously to Steps 3--4 above (with a lot of simplifications because there are no stationary points), one can prove that \emph{for each $m,\varepsilon,\delta>0$ and each   $(x,t)\in\varepsilon\mathbb{Z}^2$ satisfying
$|x|/t>1/\sqrt{1+m^2\varepsilon^2}+\delta$ and $\varepsilon\le 1/m$, we~have
$
{a}\left(x,t,m,\varepsilon\right)
=O\left(\frac{\varepsilon}{mt^{2}\delta^3}\right)$}
\cite[Theorem~3B]{SU-20}.
\end{remark}

\addcontentsline{toc}{myshrinkalt}{}

\subsection{Large-time limit: the stationary phase method again (Corollaries~\ref{th-limiting-distribution-mass}--\ref{cor-free})}
\label{ssec-proofs-phase}


In this section we prove Corollaries~\ref{th-limiting-distribution-mass}--\ref{cor-free}. First we outline the plan of the argument, then prove Corollary~\ref{th-limiting-distribution-mass} modulo a technical lemma, then the lemma itself, and finally Corollaries~\ref{cor-young}--\ref{cor-free}.


The plan of the proof of Corollary~\ref{th-limiting-distribution-mass}
(and results such as Problems~\ref{p-correlation}--\ref{p-correlation2}) consists of $4$ steps
\new{(cf.~arxiv preprint of \cite{Sunada-Tate-12} for a different realization of Steps~1--3)}:
\begin{description}
  \item[Step 1:] computing the main contribution to the sum, using asymptotic formulae~\eqref{eq-ergenium-re}--\eqref{eq-ergenium-im};
  \item[Step 2:] estimating the contribution coming from a trigonometric sum;
  \item[Step 3:] estimating the error coming from replacing  sum by an integral;
  \item[Step 4:] estimating the contribution coming from outside of the interval where \eqref{eq-ergenium-re}--\eqref{eq-ergenium-im} hold.
\end{description}

\begin{proof}[Proof of Corollary~\ref{th-limiting-distribution-mass} modulo some lemmas]
\textbf{Step 1.}
Fix $m,\varepsilon,\delta>0$, denote
$n:=1+m^2\varepsilon^2$, \new{$F(v):=F(v,m,\varepsilon)$}, $V:=1/\sqrt{n}-\delta$, and fix $v$ such that $-V\le v\le V$. Let us prove that if  $t$ is sufficiently large in terms of $\delta,m,\varepsilon$, then
\begin{equation}\label{eq-distribution-between-peaks}
\sum_{\substack{-Vt < x\le vt\\x\in\varepsilon\mathbb{Z}}} P(x,t,m,\varepsilon)
= F(v)-F(-V)+{O}_{\delta,m,\varepsilon}\left(\sqrt{\frac{\varepsilon}{t}}\right).
\end{equation}
This follows from the sequence of asymptotic formulae: 
\begin{multline}\label{eq-a1-distribution-between-peaks}
\hspace{-0.5cm}
\sum_{\substack{-V(t+\varepsilon) < x\le v(t+\varepsilon)\\ x\in\varepsilon\mathbb{Z}}}
a_1^2(x,t+\varepsilon,m,\varepsilon)
\overset{(*)}=
\sum_{\substack{-V(t+\varepsilon) < x\le v(t+\varepsilon)\\(x+t)/\varepsilon\text{ odd}}}
\left(\frac{2m\varepsilon^2}{\pi t}
\left(1-\frac{nx^2}{t^2}\right)^{-1/2}\sin^2\theta(x,t,m,\varepsilon)
+{O}_{\delta}\left(\frac{\varepsilon^2}{t^{2}}\right)
\right)
\\
\overset{(**)}=
\sum_{\substack{-Vt< x< vt\\(x+t)/\varepsilon\text{ odd}}}
\frac{m\varepsilon^2}{\pi t}
\left(1-\frac{nx^2}{t^2}\right)^{-1/2}-
\sum_{\substack{-Vt< x< vt\\(x+t)/\varepsilon\text{ odd}}}
\frac{m\varepsilon^2}{\pi t}
\left(1-\frac{nx^2}{t^2}\right)^{-1/2}
\cos 2\theta(x,t,m,\varepsilon)
+{O}_{\delta}\left(\frac{\varepsilon}{t}\right)
\\
\overset{(***)}=
\sum_{\substack{-Vt< x< vt\\(x+t)/\varepsilon\text{ odd}}}
\frac{m\varepsilon\cdot 2\varepsilon/t}{2\pi\sqrt{1-nx^2/t^2}}
+{O}_{\delta,m,\varepsilon}\left(\sqrt{\frac{\varepsilon}{t}}\right)
\overset{(****)}=
\int\limits_{-V}^{v}\frac{m\varepsilon\,dv}{2\pi\sqrt{1-nv^2}}
+{O}_{\delta,m,\varepsilon}\left(\sqrt{\frac{\varepsilon}{t}}\right)
\\
=m\varepsilon\frac{\arcsin(\sqrt{n}v)-\arcsin(-\sqrt{n}V)} {2\pi\sqrt{n}}
+{O}_{\delta,m,\varepsilon}\left(\sqrt{\frac{\varepsilon}{t}}\right)
\end{multline}
and an analogous asymptotic formula
\begin{multline*}
\sum_{\substack{-V(t+\varepsilon) < x\le v(t+\varepsilon)\\ x\in\varepsilon\mathbb{Z}}}a_2^2(x,t+\varepsilon,m,\varepsilon)
=\int_{-V}^{v}\frac{m\varepsilon(1+v)\, dv} {2\pi(1-v)\sqrt{1-nv^2}}
+{O}_{\delta,m,\varepsilon}\left(\sqrt{\frac{\varepsilon}{t}}\right)
\\
=F(v)-F(-V)
-m\varepsilon\frac{\arcsin(\sqrt{n}v)-\arcsin(-\sqrt{n}V)} {2\pi\sqrt{n}}
+{O}_{\delta,m,\varepsilon}\left(\sqrt{\frac{\varepsilon}{t}}\right).
\end{multline*}
Here $(*)$ follows from Theorem~\ref{th-ergenium} because $|x|/t\le |v|(t+\varepsilon)/t<V+\delta/2=1/\sqrt{n}-\delta/2$ for large enough $t$; the product of the main term and the error term in~\eqref{eq-ergenium-re} is estimated by $\varepsilon^2/t^2$. Asymptotic formula~$(**)$ holds because the number of summands is less than $t/\varepsilon$ and the (possibly) dropped first and last summands are less than $m\varepsilon^2/t\sqrt{\delta}$. 

\smallskip\textbf{Step 2.} Let us prove formula~$(***)$. We use the following simplified version of the stationary phase method.

\begin{lemma}\label{l-trigonometric-sum} \textup{\cite[Corollary from Theorem~4 in p.~17]{Karatsuba-93}} Under the assumptions of Lemma~\ref{l-weighted-stationary-phase} (except the ones on $f'(p)$, $g^{(3)}(p)$, and the inequality $N\ge M/\sqrt{T}$), if $M=N$ and $M/C\le T\le CM^2$ for some $C>0$, then
$$
\sum_{\alpha<p<\beta}g(p)e(f(p))
={O}_{C,C_0,\dots,C_4}
  \left(
  \frac{(\beta-\alpha)U\sqrt{T}}{M}+\frac{UM}{\sqrt{T}}
  \right).
$$
\end{lemma}

For notational convenience, assume that $t/\varepsilon$ is odd; otherwise the proof is analogous. Then the summation index $x=2p\varepsilon$ for some integer $p$. We apply Lemma~\ref{l-trigonometric-sum} for the functions
\begin{align}\label{eq-functions-trig}
  f_{\pm}(p)=\pm\frac{1}{\pi}\theta(2p\varepsilon,t,m,\varepsilon) \qquad\text{and}\qquad
  g(p)=\frac{m\varepsilon^2}{\pi t}
\left(1-\frac{4n\varepsilon^2p^2}{t^2}\right)^{-1/2}
\end{align}
and the parameter values
\begin{align}\label{eq-parameter-values-trig}
M&=N=T=t/\varepsilon,  & U&=\varepsilon/t, & \alpha&=-Vt/2\varepsilon, & \beta&=vt/2\varepsilon.
\end{align}

\begin{lemma}\label{l-checking-the-assumptions-trig}
For $\varepsilon\le 1/m$ there exist $C,C_0,\dots,C_4$ depending on $\delta,m,\varepsilon$ but not $v,p$ such that parameters~\eqref{eq-parameter-values-trig} and functions~\eqref{eq-functions-trig}
satisfy all the assumptions of Lemma~\ref{l-trigonometric-sum}. \end{lemma}

Since parameters~\eqref{eq-parameter-values-trig} satisfy
$\frac{(\beta-\alpha)U\sqrt{T}}{M}+\frac{UM}{\sqrt{T}}
={O}\left(\sqrt{\frac{\varepsilon}{t}}\right)$, formula~$(***)$ follows.

\smallskip\textbf{Step 3.} Let us prove formula~(****). We use yet another known result.

\begin{lemma}[Euler summation formula] \label{l-Euler-summation-formula} \textup{\cite[Remark to Theorem~1 in p.~3]{Karatsuba-93}}
If $g(p)$ is continuously differentiable on $[\alpha,\beta]$ and $\rho(p):=1/2-\{p\}$, then
$$
\sum_{\alpha<p<\beta}g(p)
=\int_{\alpha}^{\beta}g(p)\,dp+\rho(\beta)g(\beta)-\rho(\alpha)g(\alpha)
+\int_{\alpha}^{\beta}\rho(p)g'(p)\,dp.
$$
\end{lemma}

Again assume without loss of generality that $t/\varepsilon$ is odd. Apply Lemma~\ref{l-Euler-summation-formula} to the same $\alpha,\beta,g(p)$ (given by~\eqref{eq-functions-trig}--\eqref{eq-parameter-values-trig}) as in Step~2.  By Lemma~\ref{l-checking-the-assumptions-trig} we have $g(p)=O_{\delta,m,\varepsilon}(\varepsilon/t)$ and $g'(p)=O_{\delta,m,\varepsilon}(\varepsilon^2/t^2)$. Hence by Lemma~\ref{l-Euler-summation-formula} the difference between the sum and the integral in~(****) is $O_{\delta,m,\varepsilon}(\varepsilon/t)$, 
and~\eqref{eq-distribution-between-peaks} follows.

\smallskip\textbf{Step 4.} Let us prove the corollary for arbitrary $v\in (-1/\sqrt{n};1/\sqrt{n})$. By~\eqref{eq-distribution-between-peaks}, for each
$\delta,m,\varepsilon$ there are $C_1(\delta,m,\varepsilon)$ and $C_2(\delta,m,\varepsilon)$ such that for each $v\in [-1/\sqrt{n}+\delta,1/\sqrt{n}-\delta]$ and each $t\ge C_1(\delta,m,\varepsilon)$
 we have
\begin{equation*}
\left|\sum_{(-1/\sqrt{n}+\delta)t< x\le vt}P(x,t,m,\varepsilon)-F(v)\right|\le F\left(-\frac{1}{\sqrt{n}}+\delta\right)
+C_2(\delta,m,\varepsilon)\sqrt{\frac{\varepsilon}{t}}.
\end{equation*}

Clearly, we may assume that $C_1(\delta,m,\varepsilon)$ and $C_2(\delta,m,\varepsilon)$ are decreasing functions in $\delta$: the larger is the interval $[-\frac{1}{\sqrt{n}}+\delta,\frac{1}{\sqrt{n}}-\delta]$, the weaker is our error estimate in (*)--(****). Take $\delta(t)$ tending to $0$ slowly enough so that
$C_1(\delta(t),m,\varepsilon)\le t$ for $t$ sufficiently large in terms of $m,\varepsilon$ and
$C_2(\delta(t),m,\varepsilon)\sqrt{\frac{\varepsilon}{t}}\to 0$ as $t\to\infty$. Denote $V(t):=\frac{1}{\sqrt{n}}-\delta(t)$.
Then since $F\left(-\frac{1}{\sqrt{n}}+\delta\right)\to F\left(-\frac{1}{\sqrt{n}}\right)=0$ as $\delta \to 0$ by the definition of $F(v)$, it follows that
\begin{equation}\label{eq-auxiliary-uniform-convergence}
\sum_{-V(t)t< x\le vt}P(x,t,m,\varepsilon)\rightrightarrows F(v)\qquad\text{as }\quad t\to\infty
\end{equation}
uniformly in $v\in(-\frac{1}{\sqrt{n}};\frac{1}{\sqrt{n}})$.
Similarly, since $F\left(\frac{1}{\sqrt{n}}-\delta\right)\to F\left(\frac{1}{\sqrt{n}}\right)=1$ as $\delta \to 0$, we get
\begin{equation*}
\sum_{-V(t)t< x\le V(t)t}P(x,t,m,\varepsilon)\to 1\qquad\text{as }\quad t\to\infty.
\end{equation*}
Then by Proposition~\ref{p-mass2} we get
\begin{align*}
\sum_{x\le -V(t)t}P(x,t,m,\varepsilon)
&=1-\sum_{x>-V(t)t}P(x,t,m,\varepsilon)
\le 1-\sum_{-V(t)t< x\le V(t)t}P(x,t,m,\varepsilon)
\to 0.
\end{align*}
With~\eqref{eq-auxiliary-uniform-convergence}, this implies the corollary for  $v\in(-1/\sqrt{n};1/\sqrt{n})$.
For $v\le -1/\sqrt{n}$ and similarly for $v\ge 1/\sqrt{n}$,
the corollary follows from
$\sum_{x\le vt}P(x,t,m,\varepsilon)\le
\sum_{x\le -V(t)t}P(x,t,m,\varepsilon)\to 0.$
\end{proof}

Now we prove the lemma and the remaining corollaries.

\begin{proof}[Proof of Lemma~\ref{l-checking-the-assumptions-trig}] The inequalities $M/C\le T\le CM^2$ and $M\ge \beta-\alpha$ 
are automatic for $C=1$ because $t/\varepsilon$ is a positive integer and $|V|,|v|\le 1$.
We estimate the derivatives (computed in \cite[\S9]{SU-2}) as follows, using the assumption $\varepsilon\le 1/m$, $\alpha\le p\le \beta$, and setting $C_2:=\max\{1/m\varepsilon,2/\delta^{3/2}\}$:
\begin{align*}
|g(p)|&=\frac{m\varepsilon^2}{\pi t}
\left(1-\frac{4n\varepsilon^2p^2}{t^2}\right)^{-1/2}
\le \frac{m\varepsilon^2}{\pi t\sqrt{1-nV^2}}
\le \frac{m\varepsilon^2}{t\sqrt{\delta}}
\le \frac{\varepsilon}{t\sqrt{\delta}}
={O}_{\delta}\left({U}\right), \\
|g^{(1)}(p)|&=\frac{4m\varepsilon^4 n|p|}{\pi t^3}
\left(1-\frac{4n\varepsilon^2p^2}{t^2}\right)^{-3/2}
\le \frac{2m\varepsilon^3 nVt}
{\pi t^3(1-nV^2)^{3/2}}
\le \frac{m\varepsilon^3 n}{t^2\delta^{3/2}}
\le \frac{2\varepsilon^2}{t^2\delta^{3/2}}
={O}_{\delta}\left(\frac{U}{N}\right), \\
|g^{(2)}(p)|
&=\frac{4m\varepsilon^4 n(8\varepsilon^2 np^2+t^2)}{\pi t^5}
\left(1-\frac{4n\varepsilon^2p^2}{t^2}\right)^{-5/2}
\le \frac{4m\varepsilon^4 n(2nV^2+1)t^2}{\pi t^5(1-nV^2)^{5/2}}
={O}_{\delta}\left(\frac{\varepsilon^3}{t^3}\right)
={O}_{\delta}\left(\frac{U}{N^2}\right), \\
|f^{(2)}(p)|&=\frac{4m\varepsilon^2}{\pi t} \left(1-\frac{4\varepsilon^2p^2}{t^2}\right)^{-1}
\left(1-\frac{4n\varepsilon^2p^2}{t^2}\right)^{-1/2}
\ge\frac{m\varepsilon^2}{t}\ge \frac{T}{C_2M^2},\\
|f^{(2)}(p)|&=\frac{4m\varepsilon^2}{\pi t} \left(1-\frac{4\varepsilon^2p^2}{t^2}\right)^{-1}
\left(1-\frac{4n\varepsilon^2p^2}{t^2}\right)^{-1/2}
\le
\frac{4m\varepsilon^2}{\pi t \left(1-V^2\right)\sqrt{1-nV^2}}
\le
\frac{2m\varepsilon^2}{t\delta^{3/2}}
\le \frac{C_2T}{M^2}, \\
|f^{(3)}(p)|
&=\frac{16m\varepsilon^4
\left|(n+2)pt^2-12n\varepsilon^2 p^3\right|}{\pi t^5}
\left(1-\frac{4\varepsilon^2p^2}{t^2}\right)^{-2}
\left(1-\frac{4n\varepsilon^2p^2}{t^2}\right)^{-3/2}
={O}_{\delta}\left(\frac{T}{M^3}\right),\\
|f^{(4)}(p)|&=
\frac{16m\varepsilon^4\left|768n^2\varepsilon^6p^6 - 48n(2n+5) \varepsilon^4p^4 t^2 + 8(n^2-n+3)\varepsilon^2 p^2 t^4 +(n+2) t^6\right|} {\pi t^9
\left(1-{4\varepsilon^2p^2}/{t^2}\right)^{3}
\left(1-{4n\varepsilon^2p^2}/{t^2}\right)^{5/2}}
={O}_{\delta}\left(\frac{T}{M^4}\right).\\[-1.0cm]
\end{align*}
\end{proof}


\begin{proof}[Proof of Corollary~\ref{cor-young}]
We have $n_+(h\times w)-n_-(h\times w)=-2^{(w+h-1)/2}a_1(w-h,w+h)$ by the obvious bijection between the Young diagrams with exactly $h$ rows and $w$ columns, and checker paths from $(0,0)$ to $(w-h,w+h)$ passing through $(1,1)$ and $(w-h+1,w+h-1)$. Set $h:=\lceil rw\rceil$. Apply
Theorem~\ref{th-ergenium} and Remark~\ref{rem-rough} (or Theorem~\ref{th-outside}) for
$$
\delta=\tfrac{1}{2}\left|\tfrac{1}{\sqrt{2}}-\tfrac{r-1}{r+1}\right|,
\quad m=\varepsilon=1, \quad x=w-h, \quad t=w+h-1.
$$
This completes the proof in the case when $r>3+2\sqrt{2}$.
It remains to show that for $r<3+2\sqrt{2}$ the value~\eqref{eq-theta} is not bounded from $\frac{\pi}{2}+\pi\mathbb{Z}$ as $w\to\infty$.

Denote $v:=\frac{h-w}{w+h-1}$ and $v_0:=\frac{r-1}{r+1}$. Write
$$
\theta(vt,t,1,1)=
t\left(\arcsin\frac{1}{\sqrt{2-2v^2}}
-v\arcsin\frac{v}{\sqrt{1-v^2}}\right)+\frac{\pi}{4}
=:t\theta(v)+\frac{\pi}{4}.
$$
Since $\theta(v)\in C^2[0;1/\sqrt{2}-\delta]$, by the Taylor expansion it follows that
$$
\theta(vt,t,1,1)
=\frac{\pi}{4}+t\theta(v_0)+t(v-v_0)\theta'(v_0)
+O_\delta\left(t(v-v_0)^2\right).
$$
Substituting
$$
v-v_0=\frac{h-w}{w+h-1}-\frac{r-1}{r+1}
=\frac{2h-2rw+r-1}{(r+1)(w+h-1)}=\frac{2\{-rw\}+r-1}{(r+1)t},
$$
where $h=\lceil rw\rceil=rw+\{-rw\}$, we get
\begin{multline*}
\theta(vt,t,1,1)
=\frac{\pi}{4}+(w+rw+\{-rw\}-1)\theta(v_0)
+\frac{2\{-rw\}+r-1}{(r+1)}\theta'(v_0)
+O_\delta\left(\frac{1}{w}\right)\\
=\frac{\pi}{4}-\theta(v_0)+v_0\theta'(v_0)+w(r+1)\theta(v_0)
+\{-rw\}\left(\theta(v_0)+\tfrac{2}{(r+1)}\theta'(v_0)\right)
+O_\delta\left(\tfrac{1}{w}\right)\\=:\pi(\alpha(r) w+\beta(r)\{-rw\}+\gamma(r))+O_\delta\left(\tfrac{1}{w}\right).
\end{multline*}
For almost every $r$, the numbers
$1,r,\alpha(r)$ are linearly independent over the rational numbers because the graph of the function $\alpha(r)=(r+1)\theta\left(\frac{r-1}{r+1}\right)$ has \new{just} a countable number of intersection points with lines given by equations with rational coefficients. Hence by the Kronecker theorem for each $\Delta>0$ there are infinitely many $w$ such that
$$
\{-rw\}<\Delta\qquad\text{and}
\qquad \left|\{\alpha(r)w\}+\gamma(r)-\tfrac{1}{2}\right|<\Delta.
$$
By~\eqref{eq-ergenium-re}, the corollary follows because
those $w$ satisfy
\begin{multline*}
\left|\sin\theta(vt,t,1,1)\right|=
1+O\left((1+\beta(r))\Delta\right)+O_\delta\left(\tfrac{1}{w}\right)
\qquad\text{and}\qquad
2^{(r+1)w/2}\le 2^{(w+h)/2}\le 2^{(r+1)w/2+\Delta}.\\[-1.0cm]
\end{multline*}
\end{proof}

%
%
%
%

We are going to deduce Corollary~\ref{cor-free} from the results of~\cite{Sunada-Tate-12}.

\begin{proof}[Proof of Corollary~\ref{cor-free}]
We apply \cite[Corollary~1.5 and Theorem~1.1]{Sunada-Tate-12} for
$$
n=t/\varepsilon-1, \quad
y_n=2\varepsilon\left\lceil \frac{vt}{2\varepsilon}\right\rceil-1,  \quad
\xi=v, \quad
\phi=(0,1)^T, \quad
a=1/\sqrt{1+m^2\varepsilon^2}, \quad b=m\varepsilon/\sqrt{1+m^2\varepsilon^2}.
$$
Case $|v|>1/\sqrt{1+m^2\varepsilon^2}$ follows from \cite[Corollary~1.5]{Sunada-Tate-12}.
Case $|v|<1/\sqrt{1+m^2\varepsilon^2}$ follows from\new{}\new{}
\begin{equation}\label{eq-probability-bound}
\frac{2m\varepsilon^2}{\pi t(1-v)}
\sqrt{\frac{a-|v|}{a+|v|}}
+{O}_{m,\varepsilon,v}\left(\frac{1}{t^2}\right)
\le P\left(2\varepsilon\left\lceil\frac{vt}{2\varepsilon}\right\rceil,t,m,\varepsilon\right)
\le\frac{2m\varepsilon^2}{\pi t(1-v)}
\sqrt{\frac{a+|v|}{a-|v|}}
+{O}_{m,\varepsilon,v}\left(\frac{1}{t^2}\right),
\end{equation}
where $t\in 2\varepsilon\mathbb{Z}$ and $a:=1/\sqrt{1+m^2\varepsilon^2}$. Estimate~\eqref{eq-probability-bound} follows from \cite[Theorem~1.1]{Sunada-Tate-12} because \cite[Eq.~(1.11)]{Sunada-Tate-12} satisfies
$
|\mathrm{OSC}_n(\xi)|\le \sqrt{A(\xi)^2+B(\xi)^2}=|\xi|(1+\xi)/|a|
$
(checked in \cite[\S18]{SU-2}).
\end{proof}

Alternatively, \eqref{eq-probability-bound} can be deduced from Theorem~\ref{th-ergenium} using the method of \S\ref{ssec-proofs-feynman}.

\addcontentsline{toc}{myshrinkalt}{}

\subsection{Solution of the Feynman problem: Taylor expansions (Corollaries~\ref{cor-intermediate-asymptotic-form}--\ref{cor-feynman-problem})}
\label{ssec-proofs-feynman}

Here we deduce the solution of the Feynman problem from Theorem~\ref{th-ergenium}. For that purpose we approximate the functions in Theorem~\ref{th-ergenium} by a few terms of their Taylor expansions.

\begin{proof}[Proof of Corollary~\ref{cor-intermediate-asymptotic-form}]
First derive an asymptotic formula for the function $\theta(x,t,m,\varepsilon)$ given by~\eqref{eq-theta}.
Denote $n:=1+m^2\varepsilon^2$. Since $1/\sqrt{1+z^2}=1+O(z^2)$,
$\arcsin z=z+{O}\left(z^3\right)$ for $z\in[-1;1]$, and $t/\sqrt{t^2-x^2}<1/\sqrt{1-\sqrt{n}x/t}<1/\sqrt{\delta}$, we get
\begin{align*}
\arcsin
\frac{m\varepsilon t} {\sqrt{n\left(t^2-x^2\right)}}
&=\frac{m\varepsilon t} {\sqrt{1+m^2\varepsilon^2}\sqrt{t^2-x^2}}+
{O}\left(
\frac{m^3\varepsilon^3}{n^{3/2}} \left(\frac{t}{\sqrt{t^2-x^2}}\right)^3\right)
=
\frac{m\varepsilon t} {\sqrt{t^2-x^2}}+
{O}_\delta\left(m^3\varepsilon^3\right).
\end{align*}
Combining with a similar asymptotic formula for $\arcsin
\frac{m\varepsilon x}{\sqrt{t^2-x^2}}$, we get
$$
\theta(x,t,m,\varepsilon)
=\frac{m t^2} {\sqrt{t^2-x^2}}
-\frac{m x^2} {\sqrt{t^2-x^2}}+\frac{\pi}{4}+
\left(\frac{t+|x|}{\varepsilon}\right)
{O}_\delta\left(m^3\varepsilon^3\right)=
m\sqrt{t^2-x^2}+\frac{\pi}{4}+
{O}_\delta\left(m^3\varepsilon^2t\right).
$$
Since
\begin{equation*}
 \left|\frac{\partial \sqrt{t^2-x^2}}{\partial t}\right|
 =\frac{t}{\sqrt{t^2-x^2}}<\frac{1}{\sqrt{\delta}}
 \qquad\text{and}\qquad
 \left|\frac{\partial \sqrt{t^2-x^2}}{\partial x}\right|
 =\frac{|x|}{\sqrt{t^2-x^2}}<\frac{1}{\sqrt{\delta}},
\end{equation*}
by the Lagrange theorem it follows that
\begin{align*}
\theta(x,t-\varepsilon,m,\varepsilon)
&=m\sqrt{t^2-x^2}+\frac{\pi}{4}
+{O}_\delta\left(m\varepsilon+m^3\varepsilon^2t\right),\\
\theta(x-\varepsilon,t-\varepsilon,m,\varepsilon)
&=m\sqrt{t^2-x^2}+\frac{\pi}{4}
+{O}_\delta\left(m\varepsilon+m^3\varepsilon^2t\right).
\end{align*}
Consider the remaining factors in \eqref{eq-ergenium-re}--\eqref{eq-ergenium-im}.
By the Lagrange theorem, for some $\eta\in[0,nx^2/t^2]$ we get
\begin{multline*}
\left(1-\frac{nx^2}{t^2}\right)^{-1/4}-1
=\frac{nx^2}{t^2}\frac{(1-\eta)^{-5/4}}{4}
\le \frac{nx^2}{t^2}
\left(1-\frac{nx^2}{t^2}\right)^{-5/4}\\
\le  \frac{x^2}{t^2} \left(\frac{1}{\sqrt{n}}-\frac{x}{t}\right)^{-5/4}
\left(\frac{1}{\sqrt{n}}+\frac{x}{t}\right)^{-5/4}
\le \frac{x^2}{t^2}\delta^{-5/2}
={O}_\delta\left(\frac{|x|}{t}\right).
\end{multline*}
Hence for $t\ge 2\varepsilon$ we get
\begin{equation}\label{eq-auxest1}
\left(1-\frac{nx^2}{(t-\varepsilon)^2}\right)^{-1/4}
=1+{O}_\delta\left(\frac{|x|}{t}\right)
\quad\text{and}\quad
\left(1-\frac{n(x-\varepsilon)^2}{(t-\varepsilon)^2}\right)^{-1/4}
=1+{O}_\delta\left(\frac{|x|+\varepsilon}{t}\right).
\end{equation}
We also have
\begin{equation}\label{eq-auxest2}
\sqrt{\frac{t-\varepsilon+x-\varepsilon}{t-x}}
=\sqrt{1+2\frac{x-\varepsilon}{t-x}}
=1+{O}\left(\frac{x-\varepsilon}{t-x}\right)
=1+{O}_\delta\left(\frac{|x|+\varepsilon}{t}\right).
\end{equation}
Substituting all the resulting asymptotic formulae into~\eqref{eq-ergenium-re}--\eqref{eq-ergenium-im}, we get
\begin{align*}
\mathrm{Re}\,{a}\left(x,t,m,\varepsilon\right)
&={\varepsilon}\sqrt{\frac{2m}{\pi t}}
\left(\sin \left(m\sqrt{t^2-x^2}+\frac{\pi}{4}\right)
+O_\delta\left(
\frac{1}{mt}+\frac{|x|+\varepsilon}{t}+m\varepsilon+m^3\varepsilon^2t
\right)\right),\\
\mathrm{Im}\,{a}\left(x,t,m,\varepsilon\right)
&={\varepsilon}\sqrt{\frac{2m}{\pi t}}
\left(\cos \left(m\sqrt{t^2-x^2}+\frac{\pi}{4}\right)
+O_\delta\left(
\frac{1}{mt}+\frac{|x|+\varepsilon}{t}+m\varepsilon+m^3\varepsilon^2t
\right)
\right).
\end{align*}
Since $m\varepsilon\le \frac{1}{mt}+m^3\varepsilon^2t$ and $\frac{\varepsilon}{t}\le \frac{1}{mt}$ by the assumption $\varepsilon\le 1/m$, it follows that the error terms can be rewritten in the required form.
\end{proof}

\begin{proof}[Proof of Corollary~\ref{cor-feynman-problem}]
This follows directly from Corollary~\ref{cor-intermediate-asymptotic-form} by plugging in the Taylor expansion
\begin{multline*}
\sqrt{t^2-x^2}=t\left(1-\frac{x^2}{2t^2}
+{O}_\delta\left(\frac{x^4}{t^4}\right)\right)
\qquad\text{for }\quad \frac{|x|}{t}<1-\delta.
\\[-1.5cm]
\end{multline*}
\end{proof}

\begin{proof}[Proof of Example~\ref{p-Feynman-couterexample}]
The case $(x_n,t_n,\varepsilon_n)=(n^3,n^4,1/n^4)$ follows from Corollary~\ref{cor-intermediate-asymptotic-form} by plugging in the Taylor expansion
$$
\sqrt{t^2-x^2}=t\left(1-\frac{x^2}{2t^2}
-\frac{x^4}{8t^4}+{O}\left(\frac{x^6}{t^6}\right)\right)
\qquad\text{for }\quad \frac{|x|}{t}\le
\frac{1}{2}.
$$

In the remaining cases, we need to estimate $\theta(\varepsilon,t,m,\varepsilon)$ given by~\eqref{eq-theta} for $t\ge 2\varepsilon$ and $\varepsilon\le 1/m$. Since $m\varepsilon t/\sqrt{(1+m^2\varepsilon^2)(t^2-\varepsilon^2)}\le \sqrt{2/3}$ for such $t,m,\varepsilon$, and
$\arcsin z-\arcsin w={O}\left(z-w\right)$ for $0<w<z<\sqrt{2/3}$, and
$1/\sqrt{1-z^2}-1=O(z^2)$ for $|z|\le 1/2$, we get
\begin{multline*}
\theta(\varepsilon,t,m,\varepsilon)-\theta(0,t,m,\varepsilon)
=\\=\frac{t}{\varepsilon}\arcsin
\frac{m\varepsilon t} {\sqrt{\left(1+m^2\varepsilon^2\right)\left(t^2-\varepsilon^2\right)}}
-\frac{t}{\varepsilon}\arcsin\frac{m\varepsilon} {\sqrt{1+m^2\varepsilon^2}}
-\arcsin
\frac{m\varepsilon^2}{\sqrt{t^2-\varepsilon^2}}
= {O}\left(\frac{m\varepsilon^2}{t}\right)
={O}\left(\frac{1}{mt}\right).
\end{multline*}
Then by Theorem~\ref{th-ergenium} and~\eqref{eq-auxest1}--\eqref{eq-auxest2} for $x=0$ and $\delta=1/\sqrt{2}$, we get
\begin{equation}\label{eq-intermediate-asymptotic-form2}
{a}\left(0,t,{m},{\varepsilon}\right)
=\varepsilon\sqrt{\frac{2m}{\pi t}}
\exp\left(-i\left(\frac{t}{\varepsilon}-1\right)
\arctan(m\varepsilon)+\frac{i\pi}{4}\right)
\left(1+{O}\left(\frac{1}{mt}\right)\right).
\end{equation}

In the case when $\varepsilon=\varepsilon_n=1/2n$ and $t=t_n=(2n)^2$ the right side of~\eqref{eq-intermediate-asymptotic-form2} is equivalent to the right side of~\eqref{cor-feynman-problem} times $e^{im^3/3}$ because $\arctan(m\varepsilon)=m\varepsilon-m^3\varepsilon^3/3
+O(m^5\varepsilon^5)$.

In the case when $\varepsilon=\varepsilon_n=\mathrm{const}$ and $t=t_n=2n\varepsilon$ the ratio of the right sides of~\eqref{eq-intermediate-asymptotic-form2} and~\eqref{cor-feynman-problem} has no limit because $\arctan(m\varepsilon)-m\varepsilon$ is not an integer multiple of $\pi$ for $0<\varepsilon<1/m$.
\end{proof}

\addcontentsline{toc}{myshrinkalt}{}

\subsection{Continuum limit: the tail-exchange method (Theorem~\ref{th-main} and Corollaries~\ref{cor-uniform}--\ref{cor-concentration})}
\label{ssec-proofs-continuum}


\begin{proof}[Proof of Theorem~\ref{th-main}] The proof is based on the \emph{tail-exchange method} and consists of 
$5$ steps:
\begin{description}
\item[Step 1:] dropping the normalization factor in~\eqref{eq1-p-mass}--\eqref{eq2-p-mass}, which is of order $1$.

\item[Step 2:] dropping the summands in~\eqref{eq1-p-mass}--\eqref{eq2-p-mass} starting from a number $T$ (
    we take $T=\lceil\log \frac\delta\varepsilon\rceil$).

\item[Step 3:] replacing the binomial coefficients by powers in each of the remaining summands.

\item[Step 4:] replacing the resulting sum by infinite power series.

\item[Step 5:] combining the error bounds 
in the previous steps to get the total approximation error.
\end{description}
\smallskip


Let us derive the asymptotic formula for
$a_1\left(x,t,{m},{\varepsilon}\right)$; the argument for $a_2\left(x,t,{m},{\varepsilon}\right)$ is analogous.

\smallskip
\textbf{Step 1.} Consider the 1st factor in~\eqref{eq1-p-mass}. We have $0\ge (1-t/\varepsilon)/2\ge -t/\varepsilon$ because $t\ge \delta\ge\varepsilon$. Exponentiating, we get
$$1\ge \left(1+{m^2}{\varepsilon^2}\right)^{(1-t/\varepsilon)/2}
\ge \left(1+{m^2}{\varepsilon^2}\right)^{-t/\varepsilon}
\ge e^{-{m^2}{\varepsilon^2}t/\varepsilon}\ge 1-m^2 t\varepsilon,$$
where in the latter two inequalities we used that $e^a\ge 1+a$ for each $a\in\mathbb{R}$. Thus
$$\left(1+{m^2}{\varepsilon^2}\right)^{(1-t/\varepsilon)/2}=
1+O(m^2t\varepsilon).$$

\textbf{Step 2.} Consider the $T$-th partial sum in~\eqref{eq1-p-mass} with $T=\lceil\log \frac\delta\varepsilon\rceil$ summands.
The total number of summands is indeed at least $T$
because $(t-|x|)/2\varepsilon\ge \delta/2\varepsilon\ge \log (\delta/\varepsilon)$ by
the inequalities $t-|x|\ge \delta>\varepsilon$ and $e^a\ge 1+a+a^2/2\ge 2a$ for each $a\ge 0$.

For $r\ge T$ the ratio of consecutive summands in~\eqref{eq1-p-mass} equals
\begin{align*}
\left({m}{\varepsilon}\right)^2
\frac{((t+x)/2\varepsilon-1-r)
((t-x)/2\varepsilon-1-r)}{(r+1)^2}
&<\left({m}{\varepsilon}\right)^2
\cdot\frac{(t+x)}{2\varepsilon T}\cdot \frac{(t-x)}{2\varepsilon T}=\frac{m^2s^2}{4T^2}<\frac{1}{2},
\end{align*}
where the latter inequality follows from
$T=\lceil\log \frac\delta\varepsilon\rceil>\lceil\log e^{3ms}\rceil\ge 3ms$.
Therefore, 
the error term (i.e., the sum over $r\ge T$) is less then the sum of geometric series with ratio $\frac{1}{2}.$ Thus by Proposition~\ref{p-mass3} we get
\begin{align*}
a_1\left(x,t,{m},{\varepsilon}\right)
&={m}{\varepsilon}
\left(1+O\left({m^2 t\varepsilon}\right)\right)\cdot
\left[\sum_{r=0}^{T-1}(-1)^r
\binom{(t+x)/2\varepsilon-1}{r}
\binom{(t-x)/2\varepsilon-1}{r}
\left({m}{\varepsilon}\right)^{2r}
+
\right.
\\&+\left.O\left(
\binom{(t+x)/2\varepsilon-1}{T}
\binom{(t-x)/2\varepsilon-1}{T}
\left({m}{\varepsilon}\right)^{2T}
\right)\right].
\end{align*}

\textbf{Step 3.} To approximate the sum, take integers $L:=(t\pm x)/2\varepsilon$, $r<T$, and transform binomial coefficients as follows:
$$
\binom{L-1}{r}
=\frac{(L-1)
\cdots(L-r)}{r!}=\frac{L^r}{r!}
\left(1-\frac{1}{L}\right)
\cdots\left(1-\frac{r}{L}\right).
$$
Here
$$\frac{r}{L}= \frac{2r\varepsilon}{t\pm x}<\frac{2T\varepsilon}{\delta}=
\frac{2\varepsilon}{\delta}
\left\lceil\log \frac{\delta}{\varepsilon}\right\rceil
\le
\frac{2\varepsilon}{\delta}\left(\log \frac{\delta}{\varepsilon}+1\right)
< \frac{1}{2},
$$
because $\delta/\varepsilon\ge 16$, and $2(\log a+1)/a$ decreases for $a\ge 16$ and is less than $1/2$ for $a=16$.
Applying the inequality $1-a\ge e^{-2a}$ for $0\le a\le 1/2$,  then the inequalities $1-a\le e^{-a}$ and $L\ge\delta/2\varepsilon$, we get
$$
\left(1-\frac{1}{L}\right)\cdots\left(1-\frac{r}{L}\right)
\ge
e^{{-2/L}}e^{{-4/L}}\cdots e^{{-2r/L}}=
e^{{-r(r+1)/L}}\ge
e^{-T^2/L}\ge 1-\frac{T^2}{L}
\ge 1-\frac{2T^2\varepsilon}{\delta}.
$$
Therefore,
\begin{align*}
  \frac{(t\pm x)^r}{r!(2\varepsilon)^r} \ge\binom{(t\pm x)/2\varepsilon-1}{r}
  \ge \frac{(t\pm x)^r}{r!(2\varepsilon)^r}
  \left(1-\frac{2T^2 \varepsilon}{\delta}\right).
\end{align*}
Inserting the result into the expression for $a_1(x,t,m,\varepsilon)$ from Step~2, we get
\begin{multline*}
a_1\left(x,t,{m},{\varepsilon}\right)
={m}{\varepsilon}\left(1+O\left({m^2t}{\varepsilon}\right)\right)\cdot\\
\cdot
\left[\sum_{r=0}^{T-1}(-1)^r
\left(
\frac{m}{2}\right)^{2r}\frac{(t^2-x^2)^r}{(r!)^2}
\left(1+O\left(\frac{T^2 \varepsilon}{\delta}\right)
\right)
+O\left(
\left(\frac{m}{2}\right)^{2T}\frac{(t^2-x^2)^T}{(T!)^2}
\left(1+\frac{T^2 \varepsilon}{\delta}
\right)
\right)
\right].
\end{multline*}

The latter error term in the formula is estimated as follows.
Since $T!\ge (T/3)^T$
and
$$
T\ge \log \frac\delta\varepsilon\ge 3m\sqrt{t^2-x^2}\ge
\frac{3m}{2}\sqrt{t^2-x^2}\sqrt{e},
$$
it follows that
$$
\frac{(t^2-x^2)^T}{(T!)^2}\cdot\left(\frac{m}{2}\right)^{2T}\le\frac{(t^2-x^2)^T}{(T)^{2T}}\cdot\left(\frac{3m}{2}\right)^{2T}\le e^{-T}\le \frac\varepsilon\delta.
$$
We have $(\varepsilon/\delta)\left(1+T^2 \varepsilon/\delta\right)=O\left(T^2 \varepsilon/\delta\right)$
because $T\ge 1$ and $\varepsilon<\delta$.
Thus the error term in question
can be absorbed into the $0$-th summand $O\left(T^2 \varepsilon/\delta\right)$. We get
$$
a_1\left(x,t,{m},{\varepsilon}\right)
={m}{\varepsilon}\left(1+O\left({m^2t}{\varepsilon}\right)\right)\cdot
\sum_{r=0}^{T-1}(-1)^r
\left(
\frac{m}{2}\right)^{2r}\frac{(t^2-x^2)^r}{(r!)^2}
\left(1+O\left(\frac{T^2 \varepsilon}{\delta}\right)
\right).
$$
Notice that by our notational convention the constant understood in $O\left(T^2 \varepsilon/\delta\right)$  does not depend on $r$.

\textbf{Step 4.} Now we can replace the sum with $T$ summands by an infinite sum because the ``tail'' of alternating series with decreasing absolute value of the summands can be estimated by the first summand (which has just been estimated):
$$\left|\sum_{r=T}^{\infty}(-1)^r
\left(
\frac{m}{2}\right)^{2r}\frac{(t^2-x^2)^r}{(r!)^2}
\right|
\le\frac{(t^2-x^2)^T}{(T!)^2}\cdot\left(\frac{m}{2}\right)^{2T}\le \frac\varepsilon\delta=O\left(\frac{T^2 \varepsilon}{\delta}\right).
$$
Since the constant understood in each summand $O\left(\frac{T^2 \varepsilon}{\delta}\right)$ is the same (see Step~3), we get
\begin{align*}
a_1\left(x,t,{m},{\varepsilon}\right)
&={m}{\varepsilon}\left(1+O\left({m^2t\varepsilon}\right)\right)\cdot
\sum_{r=0}^{\infty}(-1)^r
\left(
\frac{m}{2}\right)^{2r}\frac{(t^2-x^2)^r}{(r!)^2}
\left[1+O\left(\frac{T^2 \varepsilon}{\delta}\right)
\right]\\
&=
{m}{\varepsilon}\left(1+O\left({m^2t\varepsilon}\right)\right)\cdot
\left(J_0(ms)+O\left(\frac{T^2 \varepsilon}{\delta}I_0(ms)\right)\right),
\end{align*}
where we use the
\emph{modified Bessel functions of the first kind}:
\begin{align*}
I_0(z)&:=\sum_{k=0}^\infty \frac{(z/2)^{2k}}{(k!)^2},
&
I_1(z)&:=\sum_{k=0}^\infty \frac{(z/2)^{2k+1}}{k!(k+1)!}.
\end{align*}

\textbf{Step 5.} We have $m^2t\delta\le
m^2(t+|x|)(t-|x|)=m^2s^2\le 9m^2s^2\le
 T^2$. Thus ${m^2t\varepsilon}\, J_0(ms)\le T^2\varepsilon \,I_0(ms)/\delta$ and
$
{m^2t\varepsilon}\le
{T^2\varepsilon}/{\delta}<(\log(\delta/\varepsilon)+1)^2
{\varepsilon}/{\delta}< 2
$
because $(a+1)^2/2<e^a$ for $a\ge 0$.
We arrive at the formula
\begin{align*}
a_1\left(x,t,{m},{\varepsilon}\right)
&={m}{\varepsilon}\left(J_0(ms)+O\left(\frac{\varepsilon}{\delta}
\log^2\frac{\delta}{\varepsilon}\, I_0(ms)\right)\right).
\\
\intertext{Analogously, }
a_2\left(x,t,{m},{\varepsilon}\right)
&={m}{\varepsilon}\left(1+O\left({m^2t\varepsilon}\right)\right)
\cdot\frac{t+x}{\sqrt{t^2-x^2}}\cdot
\sum_{r=1}^{T-1}(-1)^r
\left(
\frac{m}{2}\right)^{2r-1}\frac{(t^2-x^2)^{\frac{2r-1}{2}}}
{(r-1)!r!}
\left[1+O\left(\frac{T^2 \varepsilon}{\delta}\right)\right]\\
&=-{m}{\varepsilon}\cdot\frac{t+x}{s}
\left(J_1(ms)+O\left(\frac{\varepsilon}{\delta}
\log^2\frac{\delta}{\varepsilon}\,I_1(ms)\right)\right).
\end{align*}
This gives the required asymptotic formula for
$a\left(x,t,{m},{\varepsilon}\right)$
because
\begin{align*}
I_0(ms)&\le \sum_{k=0}^\infty \frac{(ms/2)^{2k}}{k!}=e^{m^2s^2/4}\le e^{m^2t^2},\\
\frac{t+x}{s}I_1(ms)&\le \frac{t+x}{s}\cdot\frac{ms}{2}\sum_{k=0}^\infty \frac{(ms/2)^{2k}}{k!}= m\frac{t+x}{2}\,e^{m^2s^2/4}\le
mt\,e^{m^2t^2/4}\le e^{m^2t^2/2}e^{m^2t^2/4}\le e^{m^2t^2}.\\[-1.7cm]
\end{align*}
\end{proof}

\begin{proof}[Proof of Corollary~\ref{cor-uniform}]
This follows from Theorem~\ref{th-main} because the right-hand side of~\eqref{eq-uniform-limit} is uniformly continuous on each compact subset of the angle $|x|<t$.
\end{proof}

\begin{proof}[Proof of Corollary~\ref{cor-concentration}]
Since the right-hand side of~\eqref{eq-uniform-limit} is continuous on
$[-t+\delta; t-\delta]$, it is bounded there. Since a sequence
uniformly converging to a bounded function is uniformly bounded,
by Corollary~\ref{cor-uniform} the absolute value of the left-hand side
of~\eqref{eq-uniform-limit} is less than some constant $C_{t,m,\delta}$
depending on $t,m,\delta$ but not on $x, \varepsilon$.
Then by Proposition~\ref{p-mass2} for ${t}/{2\varepsilon}\in \mathbb{Z}$
we get
\begin{multline*}
  1-\sum_{\substack{x\in\varepsilon\mathbb{Z}:
  |x|\ge t-\delta}} P(x,t,m,\varepsilon)
  =\sum_{\substack{x\in\varepsilon\mathbb{Z}:
  |x|<t-\delta}} P(x,t,m,\varepsilon)
  =\sum_{\substack{x\in\varepsilon\mathbb{Z}:
  |x|<t-\delta}} 4\varepsilon^2
  \left|\frac{1}{2\varepsilon}a(x,t,m,\varepsilon)\right|^2
  \\
  < 4\varepsilon^2C_{t,m,\delta}^2\frac{t-\delta}{\varepsilon}\to 0
  \quad\text{as }\varepsilon\to 0.
  \\[-1.8cm]
\end{multline*}
\end{proof}

\addcontentsline{toc}{myshrinkalt}{}

\subsection{Probability of chirality flip: combinatorial identities (Theorem~\ref{p-right-prob})}
\label{ssec-proofs-spin}

Although Theorem~\ref{p-right-prob} can be deduced from~\eqref{eq-a1-distribution-between-peaks}, we give a direct proof relying on~\S\ref{ssec-proofs-basic} only.


\begin{proof}[Proof of Theorem~\ref{p-right-prob}]
Denote
$
S_1(t)=\sum_xa_1^2(x,t);\;\;S_2(t)=\sum_xa_2^2(x,t);\;\;S_{12}(t)=\sum_xa_1(x,t)a_2(x,t).
$

By Propositions \ref{p-Dirac}, \ref{p-symmetry-mass}, and \ref {p-Huygens} we have
$$
a_1(0,2t)=\frac{1}{\sqrt{2}}\sum_xa_1(x,t)(a_2(x,t)-a_1(x,t))+a_2(x,t)(a_2(x,t)+a_1(x,t))=\frac{1}{\sqrt{2}}(S_2(t)+2S_{12}(t)-S_1(t)).
$$
By definition and Proposition \ref{p-Dirac} we have
$$
S_1(t+1)-S_2(t+1)=2S_{12}(t).
$$
Hence,
$$
S_1(t+1)-S_2(t+1)=S_1(t)-S_2(t)+a_1(0,2t)\sqrt{2}.
$$
Since $S_1(t)+S_2(t)=1$ by Proposition~\ref{p-probability-conservation}, we have the recurrence relation
$S_1(t+1)=S_1(t)+\frac{1}{\sqrt{2}}a_1(0,2t)$; cf.~\cite[(33)]{Jacobson-Schulman-84}. Then
Proposition~\ref{cor-coefficients} implies by induction that
$$
S_1(t)
=\frac{1}{2}\sum_{k=0}^{\lfloor t/2\rfloor-1}
\frac{1}{(-4)^k}\binom{2k}{k}.
$$

By the Newton binomial theorem we get $\sum_{k=0}^{\infty}\binom{2k}{k}x^k=\frac{1}{\sqrt{1-4x}}$ for each $x\in \left[-\frac{1}{4},\frac{1}{4}\right)$. Setting $x=-\frac{1}{4}$ we obtain
$\lim_{t\to \infty}\frac{1}{2}\sum_{k=0}^{\lfloor t/2\rfloor-1}\binom{2k}{k}\left(-\frac{1}{4}\right)^k=\frac{1}{2\sqrt{2}}.$
Using the Stirling formula we estimate the convergence rate:
\begin{gather*}
\left|\sum_{x\in\mathbb{Z}}a_{1}(x,t)^2-\frac{1}{2\sqrt{2}}\right|< \frac{1}{2\cdot 4^{\lfloor t/2 \rfloor}}\binom{2\lfloor t/2 \rfloor}{\lfloor t/2 \rfloor} <\frac{e}{2\pi\sqrt{2\lfloor t/2 \rfloor}}<\frac{1}{2\sqrt{t}}.\\[-1.6cm]
\end{gather*}
\end{proof}

\subsection*{Underwater rocks}


Finally, let us warn a mathematically-oriented reader. The outstanding papers \cite{Ambainis-etal-01, Konno-05, Konno-08} are well-written, insomuch that the physical theorems and proofs there could be carelessly taken for mathematical ones, although some of them are wrong as written. The main source of those issues is actually application of a wrong theorem from a mathematical paper~\cite[Theorem~3.3]{Chen-Ismail-91}.

A simple counterexample to \cite[Theorem~3.3]{Chen-Ismail-91} is $a=b=\alpha=\beta=x=0$ and $n$ odd. Those values automatically satisfy the assumptions of the theorem, that is, condition~(ii) of \cite[Lemma~3.1]{Chen-Ismail-91}. Then by Remark~\ref{rem-hypergeo} and Proposition~\ref{cor-coefficients}, the left-hand side of \cite[(2.16)]{Chen-Ismail-91} vanishes. Thus it cannot be equivalent to the nonvanishing sequence in the right-hand side. Here we interpret the ``$\approx$'' sign in \cite[(2.16)]{Chen-Ismail-91} as the equivalence of sequences, following \cite{Konno-05}. An attempt to interpret the sign so that the difference between the left- and the right-hand sides of \cite[(2.16)]{Chen-Ismail-91} tends to zero would void the result because each of the sides clearly tends to zero separately. \new{}

Although \cite{Ambainis-etal-01, Konno-05, Konno-08} report minor errors in \cite{Chen-Ismail-91}, the issue is more serious. The known asymptotic formulae for Jacobi polynomials are never stated as an equivalence but rather contain an additive error term. Estimating the error term is hard even in particular cases studied in~\cite{Kuijlaars-Martinez-Finkelshtein-04}, and the case from Remark~\ref{rem-hypergeo} is reported as more difficult \cite[bottom p.~198]{Kuijlaars-Martinez-Finkelshtein-04}. Thus
\cite[Theorem~3.3]{Chen-Ismail-91} should be viewed as an interesting physical but not mathematical result.

\subsection*{Acknowledgements}

The work is supported by Ministry of Science and Higher Education of the Russian Federation, agreement \No075-15-2019-1619.
This work was presented as courses at Summer Conference of Tournament of Towns in Arandelovac, Summer School in Contemporary Mathematics in Dubna, Faculty of Mathematics in Higher School of Economics in Moscow, and Independent Moscow University in 2019. The authors are grateful to participants of those events for their contribution, especially to  M.~Dmitriev, I.~Novikov, F.~Ozhegov, A.~Voropaev for numerous remarks and the Russian translation, to E.~Akhmedova, R.~Valieva for typesetting some parts of the text, to I.~Bogdanov, A.~Daniyarkhodzhaev, M.~Fedorov, I.~Gaidai-Turlov, T.~Kovalev, F.~Kuyanov, G.~Minaev, I.~Russkikh, V.~Skopenkova for figures, to A.~Kudryavtsev, A.~Lvov for writing appendices (those two authors were less than 16 years old that time). The authors are grateful to V.~Akulin, T.~Batenev, A.~Belavin, M.~Bershtein, M.~Blank, A.~Borodin, V.~Buchstaber, G.~Chelnokov, V.~Chernyshev, I.~Dynnikov, I.~Ibragimov, I.~Ivanov, T.~Jacobson, D.U.~Kim, M.~Khristoforov, E.~Kolpakov, A.B.J.~Kuijlaars, S.~Lando, M.~Lifshits, \new{M.~Maeda,} V.~Nazaikinskii, S.~Nechaev, S.~Novikov, G.~Olshanski, Yu.~Petrova, I.~Polekhin, P.~Pylyavskyy, A.~Rybko, I.~Sabitov, A.~Semenov, L.~Schulman, S.~Shlosman, \new{T.~Tate,} S.~Tikhomirov, D.~Treschev, \new{L.~Vel\'azquez,} A.~Vershik, P.~Zakorko for useful discussions.

\appendix

\comment
\section{G.~Minaev and I.~Russkikh. More general conservation law \mscomm{! to be deleted !}} \label{app-conservation}

\addcontentsline{toc}{myshrink}{}

\mscomm{This makes sense, only if rewritten using our notation and with a more intuitive proof via the Kirchhoff current law}

\emph{A generalization of conservation law by Gleb Minaev and Ivan Russkikh, participants of Summer conference of Tournament of towns.} Perform the change coordinates $(x,t)\mapsto(t,u)=(\frac{x+t}{2}-1, \frac{t-x}{2})$, i.e., rotate the coordinate system through $45$ degrees clockwise and shift it by the vector $(-1,0)$.

Denote $\vec{b}(t,u):=\vec{a}(t-u+1,t+u+1)$, $Q(t,u):=P(t-u+1,t+u+1)$, $\vec{B}(t,u):=(1+m^2\varepsilon^2)^{(t+u)/2}\:\vec{b}(t,u)$. The coordinates of the vectors $\vec{b}(t,u)$ and $\vec{B}(t,u)$ denote by $b_1(t,u)$, $b_2(t,u)$, $B_1(t,u)$, $B_2(t,u)$.

\textit{Remark.} In the new coordinate system, a checker moves between neighbouring points of the integer lattice rather than
along the diagonals between black squares. Also we suppose that we start at $(0,0)$ and move to any point of $(\mathbb{N}\cup\{0\})^2$, and the additional ``pre-move'' from $(-1,0)$ to $(0,0)$ is taken into account only to compute the number of turns.

For a subset $M\subset(\mathbb{N}\cup\{0\})^2$ denote $\vec{b}(t,u\text{ bypass }M):=\sum_{s}\vec{b}(s)$, where we sum is over all paths $s$ from $(0,0)$ to $(t,u)$, which bypass the points of the set $M$. Analogously, define $Q(t,u\text{ bypass }M)$, $\vec{b}(t,u\text{ bypass }M,\pm)$, $Q(t,u\text{ bypass }M,\pm)$.
Also for a set $M$ denote
$$
Q(M) := \sum_{p\in M}Q(p\text{ bypass }M\setminus\{p\}).
$$

\textit{Remark.}
Note that for an infinite set $M$ the sum becomes infinite as well. The order of a summation is irrelevant because all the summands are positive.

\begin{theorem}
For each finite set $M\subset(\mathbb{N}\cup\{0\})^2$ such that there are no infinite paths from $(0;0)$ bypassing the points of $M$ we have $Q(M)=1$.
\end{theorem}

\begin{proof}
Prove the theorem by induction over $m:=\max_{(t,u)\in M}(t+u)$, i.e., the maximal number of a downwards-right diagonal containing at least one point of the set $M$. The diagonal is called \emph{maximal}.

\textit{Base}: $m=0.$ In this case $M={(0,0)}$ and $Q(M)=Q(0,0)=1$.

\textit{Step.} Let us prove a lemma.

\begin{lemma}\label{up_right_recur_lemm}
If a set $A\subset(\mathbb{N}\cap\{0\})^2$, which does not contain the points $(t,u), (t,u+1),(t+1,u)\in(\mathbb{N}\cap\{0\})^2$ , then
$$
Q(t,u\text{ bypass }A)=Q(t,u+1\text{ bypass }A, -)+Q(t+1,u\text{ bypass }A, +).
$$
\end{lemma}

\begin{proof}
This is straightforward:
\begin{multline*}
Q(t,u+1\text{ bypass }A,-)+Q(t+1,u\text{ bypass }A,+)
=b_1(t,u+1\text{ bypass }A)^2+b_2(t+1,u\text{ bypass }A)^2=\\
=\frac{(b_1(t,u\text{ bypass }A)+ b_2(t,u\text{ bypass }A))^2+(b_2(t,u\text{ bypass }A)- b_1(t,u\text{ bypass }A))^2}{2}=\\
=b_1(t,u\text{ bypass }A)^2+b_2(t,u\text{ bypass }A)^2
=Q(t,u\text{ bypass }A).\\[-1.5cm]
\end{multline*}
\end{proof}

Suppose that the point $(t,u) \in M$ is such that $t+u$ is maximal. Suppose that there is a checker path starting at the point $(0,0)$ and ending at $(t,u)$ with the last move, say, in the upwards direction bypassing all the points of the set $M\setminus\{(t,u)\}$. Then there exists a path to the point $(t,u-1)$, bypassing all the points of the set $M\setminus\{(t,u-1)\}$. Notice that in this case $(t,u-1)\notin M$, because the move to $(t;u)$ is in the upwards direction. Hence there exists a path to the point $(t+1,u-1)$, bypassing all the points of the set $M\setminus\{(t+1,u-1)\}$. If $(t+1,u-1)\notin M$ then there exists a path going to infinity passing through $(t+1,u-1)$ and bypassing all the points of the set $M$. For example, the path turning right at the point $(t,u)$ and only going right from there passes through $(t,u)$ and bypasses all the points in $M$ because the diagonal is maximal. Therefore $(t+1,u-1)\in M$. Notice that $Q(a,b\text{ bypass }A)=Q(a,b\text{ bypass }A,+)+Q(a,b\text{ bypass }A,-)$ for any set $A\subset(\mathbb{N}\cap\{0\})^2$. Denote $K:=M\setminus\{(t,u);(t+1,u-1)\}$. Then we have the following chain of equalities (where we use \ref{up_right_recur_lemm} for the set $A=K$):
\begin{multline*}
Q(M)=Q(K)+Q(t,u\text{ bypass }M\setminus\{(t,u)\})+Q(t+1,u-1\text{ bypass }M\setminus\{(t+1,u-1)\})=\\
=Q(K)+Q(t,u\text{ bypass }M\setminus\{(t,u)\},+)+Q(t,u\text{ bypass }M\setminus\{(t,u)\},-)+\\
+Q(t+1,u-1\text{ bypass }M\setminus\{(t+1,u-1)\},+)+Q(t+1,u-1\text{ bypass }M\setminus\{(t+1,u-1)\},-)=\\
=Q(K)+Q(t,u\text{ bypass }M\setminus\{(t,u)\},+)+\\
+Q(t,u-1\text{ bypass }M\setminus\{(t,u-1)\})+Q(t+1,u-1\text{ bypass }M\setminus\{(t+1,u-1)\},-)
=Q(M\cup \{(t,u-1)\}).
\end{multline*}
Thus if the point $(t,u-1)$ is added to the set $M$ then the probability $Q(M)$ is not changed. 

This way we put new points onto the diagonal $t+u=m-1,$ therefore, the maximal diagonal is not changed. Thus if we add a few points to the set $M$, then the paths bypassing other points of the set $M$ bypass the points of $M$ in the maximal diagonal as well. Therefore, if we remove all the points from the maximal diagonal, then $Q(M)$ is not changed and no infinite path bypassing the points of the set $M$ appears. This way we change $M$ keeping $Q(M)$ fixed but decreasing the number of a maximal diagonal. By the inductive hypothesis the new $Q(M)$ equals $1$, hence the old one also equals $1$.
\end{proof}

\endcomment

\addcontentsline{toc}{myshrink}{}

\section{A.~Kudryavtsev. Alternative ``explicit'' formulae} \label{app-formula}

\addcontentsline{toc}{myshrink}{}


Set $\binom{n}{k}:=0$ for integers $k<0<n$ or $k>n>0$. Denote
$\theta(x):=\begin{cases}
                            1, & \mbox{if } x\ge0, \\
                            0, & \mbox{if } x<0.
                         \end{cases}$

\begin{proposition}[``Explicit'' formula] \label{Feynman-binom-alt}
 {For each integers $|x|<t$ such that $x+t$ is even we have:
\begin{align*}
&\mathbf{(A)} &
a_1(x,t)&= {2^{(1-t)/2}\sum\limits_{r=0}^{(t-|x|)/2}
(-2)^{r}\binom{(x+t-2)/2}{r}\binom{t-r-2}{(x+t-2)/2},}\\
& &
a_2(x,t)&= {2^{(1-t)/2}\sum\limits_{r=0}^{(t-|x|)/2}
(-2)^{r}\binom{(x+t-2)/2}{r}\binom{t-r-2}{(x+t-4)/2}};\\
\end{align*}\new{}
\begin{align*}
&\mathbf{(B)} &
a_1(x,t)&=2^{(1-t)/2}
\sum_{r=0}^{(t-|x|)/2} (-1)^r\binom{(t-|x|-2)/2}{r}
\binom{|x|}{(t+|x|-4r-2)/2},\\
& &
a_2(x,t)&=2^{(1-t)/2}
\sum_{r=0}^{(t-|x|)/2} (-1)^r\binom{(t-|x|-2)/2}{r-\theta(x)}
\binom{|x|}{(t+|x|-4r)/2}.
\end{align*}
}
\end{proposition}

\begin{proof}[Proof of Proposition~\ref{Feynman-binom-alt}(A)]
Introduce the generating functions
$$\hat a_1(p,q):=2^{n/2}\,\sum\limits_{n>k\ge 0}a_1(2k-n+1,n+1)p^kq^n
\quad\text{and}\quad
\hat a_2(p,q):=2^{n/2}\,\sum\limits_{n>k\ge 0}a_2(2k-n+1,n+1)p^kq^n.$$
By Proposition~\ref{p-Dirac} we get
\begin{equation*}
 \begin{cases}
   \hat a_1(p,q)-\hat a_1(p,0)=q\cdot(\hat a_2(p,q)+\hat a_1(p,q));\\
   \hat a_2(p,q)-\hat a_2(p,0)=pq\cdot(\hat a_2(p,q)-\hat a_1(p,q)).
 \end{cases}
\end{equation*}
Since $\hat a_1(p,0)=0$ and $\hat a_2(p,0)=1$, 
the solution of this system is 
$$\hat a_2(p,q)=\dfrac {1-q}{1-q-pq+2pq^2}, 
\qquad 
\hat a_1(p,q)=\dfrac q{1-q-pq+2pq^2}=q+q^2(1+p-2pq)+q^3(1+p-2pq)^2+\dots
$$
The coefficient at $p^kq^n$ in $\hat a_1(p,q)$ equals $$\sum\limits_{j=\max(k,n-k)}^n (-2)^{n-j-1}\cdot\binom{j}{n-j-1\quad k-n+j+1 \quad j-k},$$
because we must take exactly one combination of factors from every summand of the form $q^{j+1}(1+p-2pq)^j$:
\begin{itemize}
    \item for the power of $q$ to be equal to $n$, the number of factors $-2pq$ must be $n-j-1$;
    \item for the power of $p$ to be equal to $k$, the number of factors $p$ must be $k-(n-j-1)$;
    \item 
    the number of remaining factors $1$ must be $j-(k-(n-j-1))-(n-j-1)=j-k$.
\end{itemize}
Changing the summation variable to $r=n-j-1$, we arrive at the required formula for $a_1(x,t)$.

The formula for $a_2(x,t)$ follows from the one for $a_1(x,t)$, Proposition~\ref{p-Dirac}, and the Pascal rule:
\begin{align*}
a_2(x,t)&=\sqrt{2}\,a_1(x-1,t+1)-a_1(x,t)\\
&=2^{(1-t)/2}\sum\limits_{r=0}^{(t-|x|)/2}
(-2)^{r}\binom{(x+t-2)/2}{r}
\left(\binom{t-r-1}{(x+t-2)/2}-\binom{t-r-2}{(x+t-2)/2}\right)\\
&=2^{(1-t)/2}\sum\limits_{r=0}^{(t-|x|)/2}
(-2)^{r}\binom{(x+t-2)/2}{r}\binom{t-r-2}{(x+t-4)/2}.\\[-1.4cm]
\end{align*}
\end{proof}

\begin{proof}[Proof of Proposition~\ref{Feynman-binom-alt}(B)]
 {(by A.~Voropaev) By Proposition~\ref{cor-coefficients}, for each $|x|<t$ the numbers $a_1(x,t)$ and $a_2(x,t)$ are the coefficients \new{at} $z^{(t-x-2)/2}$ and $z^{(t-x)/2}$ respectively in the expansion of the polynomial 
\begin{equation*}
	2^{(1-t)/2}(1+z)^{(t-x-2)/2}(1-z)^{(t+x-2)/2}=
	\begin{cases}
		2^{(1-t)/2}(1-z^2)^{\frac{t-x-2}2}(1-z)^x,
&\text{for } x\ge 0;\\
		2^{(1-t)/2}(1-z^2)^{\frac{t+x-2}2}(1+z)^{-x},
&\text{for } x< 0.
	\end{cases}
\end{equation*}
For $x<0$, this implies the required proposition immediately. For $x\ge 0$, we first change the summation variable to $r'=(t-x-2)/2-r$ or $r'=(t-x)/2-r$ for $a_1(x,t)$ and $a_2(x,t)$ respectively.}
\end{proof}

\section{A.~Lvov. Pointwise continuum limit} \label{app-pointwise}

\addcontentsline{toc}{myshrink}{}



\begin{theorem}[Pointwise continuum limit]  \label{p-convergence}
For  {each real $m\ge 0$ and} $|x|<t$ we have
\begin{align*}
\lim_{n\to \infty} n\, a_1\left(\frac{2}{n}\left\lfloor \frac{nx}{2}\right\rfloor,\frac{2}{n}\left\lfloor \frac{nt}{2}\right\rfloor,m,\frac{1}{n}\right)
&= m\,J_0(m\sqrt{t^2-x^2});\\
\lim_{n\to \infty} n\, a_2\left(\frac{2}{n}\left\lfloor \frac{nx}{2}\right\rfloor,\frac{2}{n}\left\lfloor \frac{nt}{2}\right\rfloor,m,\frac{1}{n}\right)
&= -m\frac{x+t}{\sqrt{t^2-x^2}}J_1(m\sqrt{t^2-x^2}).
\end{align*}
\end{theorem}


\comment

\begin{remark} \label{p-convergence}
Let us present direct restatement of the theorem: For $|p| < q$ we have
\begin{align*}
\lim_{n\to \infty}  (n^2 + m^2)^{1/2- \left\lfloor nq\right\rfloor}\sum_{r=0}^{\left\lfloor nq\right\rfloor}(-1)^r \binom{\left\lfloor nq\right\rfloor + \left\lfloor np\right\rfloor - 1}{r}\binom{\left\lfloor nq \right\rfloor - \left\lfloor np\right\rfloor - 1}{r}\frac{m^{2r}}{n^{2r+1 - 2\left\lfloor nq\right\rfloor}} = \\
= \sum_{k=0}^\infty (-1)^k\frac{m^{2k}(q^2-p^2)^{k}}{(k!)^2};\\
\lim_{n\to \infty}  (n^2 + m^2)^{1/2- \left\lfloor nq\right\rfloor}\sum_{r=0}^{\left\lfloor nq\right\rfloor}(-1)^r \binom{\left\lfloor nq\right\rfloor + \left\lfloor np\right\rfloor - 1}{r}\binom{\left\lfloor nq\right\rfloor - \left\lfloor np\right\rfloor - 1}{r - 1}\frac{m^{2r - 1}}{n^{2r - 2\left\lfloor nq\right\rfloor}} = \\
= -(p+q)\sum_{k=0}^\infty (-1)^k\frac{m^{2k+1}(q^2 - p^2)^k}{k!(k+1)!}.
\end{align*}
\end{remark}

\endcomment

\begin{proof}[Proof of Theorem~\ref{p-convergence}]
Denote $A := \lfloor \frac{nx}{2} \rfloor+\lfloor \frac{nt}{2} \rfloor$ and $B := \lfloor \frac{nt}{2} \rfloor-\lfloor \frac{nx}{2} \rfloor$.
The first limit is computed as follows:
\begin{multline*} 
n\, a_1\left(\frac{2}{n}\left\lfloor \frac{nx}{2}\right\rfloor,\frac{2}{n}\left\lfloor \frac{nt}{2}\right\rfloor,m,\frac{1}{n}\right)
=
n\left(1 + \frac{m^2}{n^2}\right)^{\lfloor \frac{nt}{2} \rfloor - \frac{1}{2}} \cdot\sum \limits _{r = 0}^{\infty} (-1)^{r}\binom{A-1}{r}\binom{B-1}{r}\left(\frac{m}{n}\right)^{2r+1}
\\
\thicksim
\sum \limits _{r = 0}^{\infty} (-1)^{r}\binom{A-1}{r}\binom{B-1}{r}\frac{m^{2r+1}}{n^{2r}} =
\sum\limits_{r = 0; 2 | r}^{\infty} \binom{A-1}{r}\binom{B-1}{r}\frac{m^{2r+1}}{n^{2r}} -
\sum\limits_{r = 0; 2 \nmid r}^{\infty} \binom{A-1}{r}\binom{B-1}{r}\frac{m^{2r+1}}{n^{2r}}
\\ \to
\sum\limits_{r = 0; 2 | r}^{\infty}
\frac{(x + t)^r(t - x)^r m^{2r + 1}}{2^{2r}(r!)^2} -
\sum\limits_{r = 0; 2 \nmid r}^{\infty}
\frac{(x + t)^r(t - x)^r m^{2r + 1}}{2^{2r}(r!)^2} =
m\, J_0(m\sqrt{t^2 - x^ 2}) \qquad\text{as } n \rightarrow \infty.
\end{multline*}
Here the equality in the 1st line is Proposition~\ref{p-mass3}.
The equivalence in the 2nd line follows from
$$
1 \le \left(1 + \frac{m^2}{n^2}\right)^{\lfloor \frac{nt}{2} \rfloor - \frac{1}{2}}
\le \left(1 + \frac{m^2}{n^2}\right)^{nt}
= \sqrt[n]{\left(1 + \frac{m^2}{n^2}\right)^{n^2t}} \thicksim \sqrt[n]{e^{m^2t}} \rightarrow 1 \qquad\text{as } n \rightarrow \infty,
$$
by the squeeze theorem. The equality in the 2nd line holds because
$\binom{A-1}{r}\binom{B-1}{r}\frac{m^{2r+1}}{n^{2r}} = 0$ for $r > \max\{A, B\}$, hence all the three sums involved are finite. The convergence in the 3rd line  is established in Lemmas~\ref{r2}--\ref{l} below.
The second limit in the theorem is computed analogously.
\end{proof}



\begin{lemma}
\label{r2}
For each positive integer $r$ we have
$\lim\limits_{n \rightarrow \infty} \binom{A-1}{r}\binom{B-1}{r}\frac{m^{2r+1}}{n^{2r}} =
 \frac{(x + t)^r(t - x)^r m^{2r + 1}}{2^{2r}(r!)^2}$.
\end{lemma}

\begin{proof}
We have $$ 
\binom{A-1}{r}\binom{B-1}{r}\frac{m^{2r+1}}{n^{2r}} =
\frac{(A - 1)\dots(A - r)\cdot(B - 1)\dots(B - r)}{(r!)^2} \cdot \frac{m^{2r + 1}}{n^{2r}} \to
\left(\frac{x + t}{2}\right)^r\left(\frac{t - x}{2}\right)^r \frac{m^{2r + 1}}{(r!)^2}$$
as $n \rightarrow \infty$ because for each $1\le i\le r$
$$\lim\limits_{n \rightarrow \infty} \frac{A - i}{n} = \lim\limits_{n \rightarrow \infty} \frac{A}{n} =
\lim\limits_{n \rightarrow \infty} \frac{\lfloor \frac{nx}{2} \rfloor + \lfloor \frac{nt}{2} \rfloor}{n} = \lim\limits_{n \rightarrow \infty} \frac{\frac{nx}{2} + \frac{nt}{2} + o(n)}{n} = \frac{x + t}{2}$$
and analogously, $\lim\limits_{n \rightarrow \infty} \frac{B - i}{n} = \frac{t - x}{2}$.
\end{proof}

\begin{lemma}
\label{r3} For each positive integer $r$ we have
$\binom{A-1}{r}\binom{B-1}{r}\frac{m^{2r+1}}{n^{2r}} \leq \frac{(x + t)^r(t - x)^r m^{2r + 1}}{2^{2r}(r!)^2}$.
\end{lemma}

\begin{proof} This follows analogously because for each $1\le i\le r<\min\{A,B\}$ we have
\begin{align*}
(A-i)(B-i) \leq 
\left(\left\lfloor \frac{nx}{2} \right\rfloor + \left\lfloor \frac{nt}{2} \right\rfloor - 1\right)\left(\left\lfloor \frac{nt}{2} \right\rfloor - \left\lfloor \frac{nx}{2} \right\rfloor - 1\right) \leq
\left(\frac{nx}{2} + \frac{nt}{2}\right)\left(\frac{nt}{2} - \frac{nx}{2}\right).
\\[-1.2cm]
\end{align*}
\end{proof}










\begin{lemma}
\label{l}
Suppose 
$\{a_k(n)\}_{k=0}^\infty$ is a sequence of nonnegative sequences such that $\lim\limits_{n \rightarrow \infty} a_k(n) = b_k$ for each $k$; $a_k(n) \leq b_k$ for each $k, n$; and $\sum\limits_{k = 0}^{\infty}b_k$ is finite. Then $\lim\limits_{n \rightarrow \infty} \sum\limits_{k = 0}^{\infty}a_k(n) = \sum\limits_{k = 0}^{\infty}b_k$.
\end{lemma}

\begin{proof}
Denote $b := \sum\limits_{k = 0}^{\infty}b_k$.
Then for each $n$ we have $\sum\limits_{k = 0}^{\infty}a_k(n) \leq b$.
Take any $\varepsilon > 0$. Take such $N$ that $\sum\limits_{k = 0}^{N}b_k > b - \varepsilon$. For each $k \leq N$ take $M_k$ such that for each $n \geq M_k$ we have $ a_k(n) > b_k - \frac{\varepsilon}{2^{k + 1}}$.  Then for each
$n > \max\limits_{0 \leq k \leq N}M_k$
we have $\sum\limits_{k = 0}^{\infty}a_k(n) > b - 2\varepsilon$.
So, $\lim\limits_{n \rightarrow \infty} \sum\limits_{k = 0}^{\infty}a_k(n) = b.$
\end{proof}


{
\footnotesize

\noindent
\textsc{Mikhail Skopenkov\\
HSE University (Faculty of Mathematics) and\\
Institute for Information Transmission Problems, Russian Academy of Sciences} 
\\
\texttt{mikhail.skopenkov\,@\,gmail$\cdot $com} \quad \url{https://users.mccme.ru/mskopenkov/}

\vspace{0.3cm}
\noindent
\textsc{Alexey Ustinov\\
HSE University (Faculty of Computer Science) and\\
Khabarovsk Division of the Institute for Applied Mathematics,\\
Far-Eastern Branch,
Russian Academy of Sciences, Russia} 
\\
\texttt{Ustinov.Alexey\,@\,gmail$\cdot $com} \quad
\url{http://iam.khv.ru/staff/Ustinov.php}

}

\end{document}



\title{\ \\[-2.5cm]Шашки Фейнмана:\\[-0.2cm]
{\Large к алгоритмической квантовой теории}\vspace{-0.4cm}}

\author{М. Скопенков и А. Устинов}

\date{}

\maketitle

\vspace{-1.9cm}

\begin{abstract}
Статья посвящена шашкам Фейнмана~---  элементарной модели движения электрона, предложенной Р. Фейнманом. В этой игре шашка движется  по простым правилам по клетчатой доске, а мы следим за ее поворотами. Шашки Фейнмана также известны как одномерное квантовое блуждание или модель Изинга при мнимой температуре. Мы приводим математическое доказательство гипотезы Фейнмана 1965 года о том, что эта дискретная модель (при больших временах, малой средней скорости и малом размере
клетки) согласована с непрерывной.  Мы исследуем асимптотические свойства модели (при малом размере клетки и больших временах), усиливая результаты Дж. Нарликара 1972 года и Т.~Сунады --- Т.~Татэ 2012 года. Мы впервые замечаем и доказываем концентрацию меры в пределе при уменьшении размера клетки. Производится вторичное квантование модели.

Приводится обзор известных результатов о шашках Фейнмана.

\textbf{Ключевые слова.} Шахматная доска Фейнмана, квантовое блуждание, модель Изинга, диаграмма Юнга, уравнение Дирака, метод стационарной фазы

\textbf{MSC2010:} 82B20, 11L03, 68Q12, 81P68, 81T25, 81T40, 05A17, 11P82, 33C45.
\end{abstract}



\vspace{-0.8cm}

\tableofcontents


\addcontentsline{toc}{myshrink}{}
\addcontentsline{toc}{myshrink}{}

\section{Введение}

Статья посвящена шашкам Фейнмана~---  элементарной модели движения электрона, предложенной Р. Фейнманом (см.~рисунок~\ref{P-contour}). В этой игре шашка движется  по простым правилам по клетчатой доске, а мы следим за ее поворотами (см.~определение~\ref{def-mass}). Шашки Фейнмана также можно рассматривать как одномерное квантовое блуждание, или модель Изинга, или подсчет числа диаграмм Юнга определенного вида.

\begin{figure}[htbp]
  \begin{tabular}{ccc}
  \hspace{-0.2cm}
  \includegraphics[width=0.31\textwidth]
  {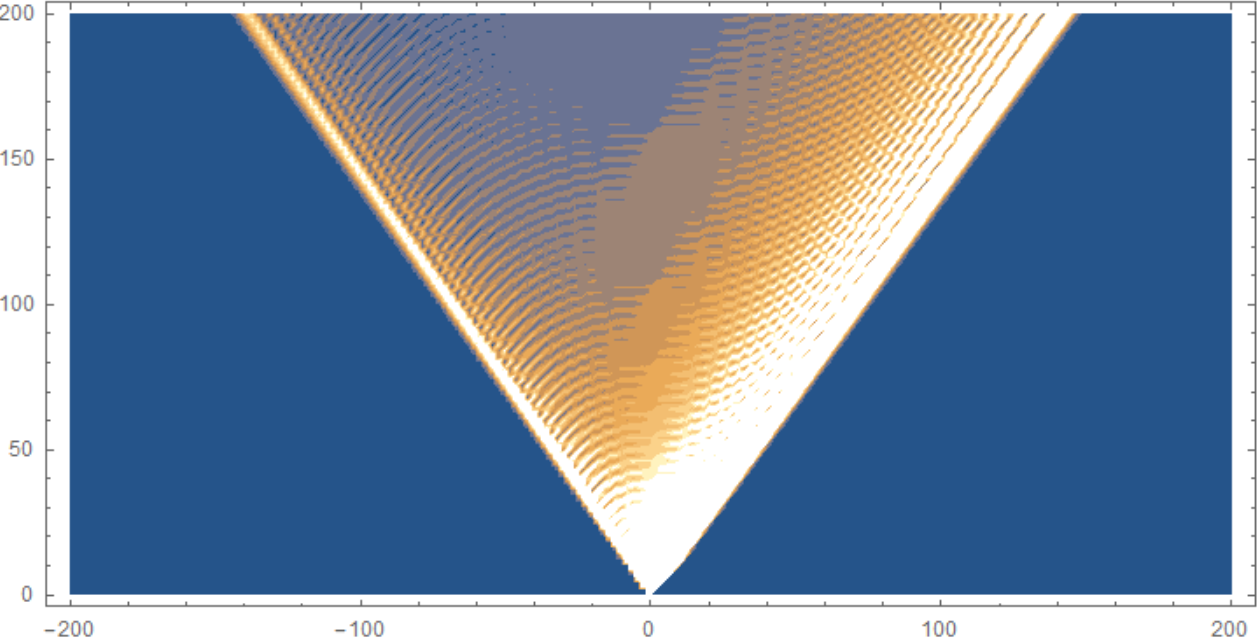}
  \hspace{-0.2cm} & \hspace{-0.2cm}
  \includegraphics[width=0.31\textwidth]
  {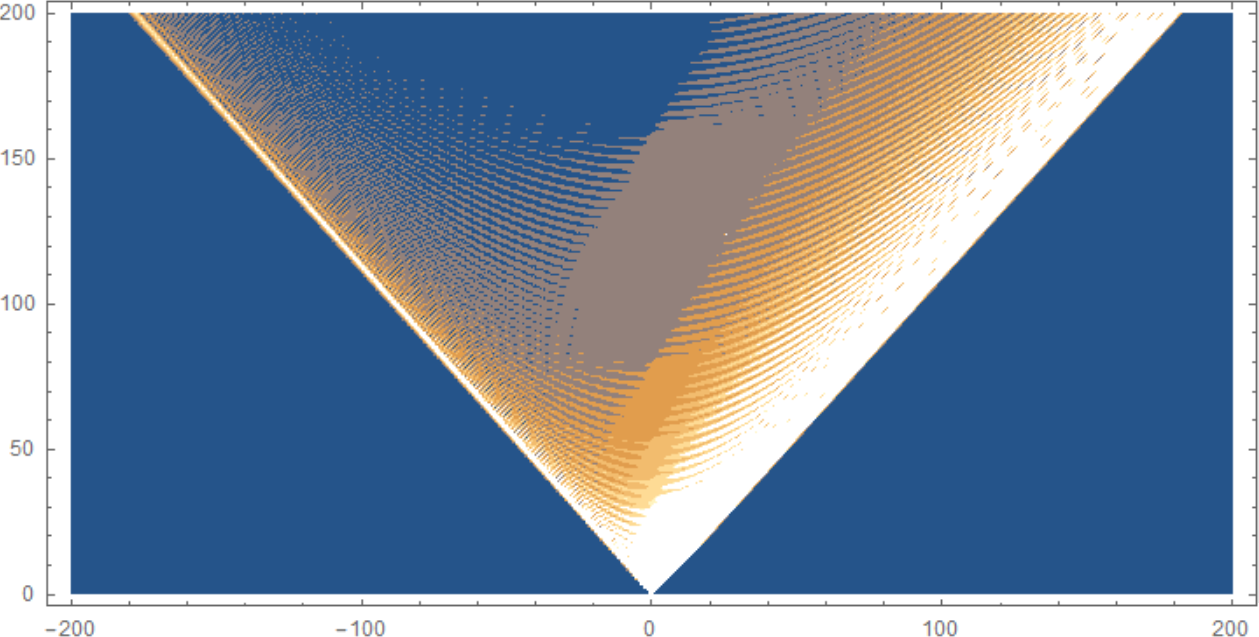}
  \hspace{-0.2cm} & \hspace{-0.2cm}
  \includegraphics[width=0.31\textwidth]
  {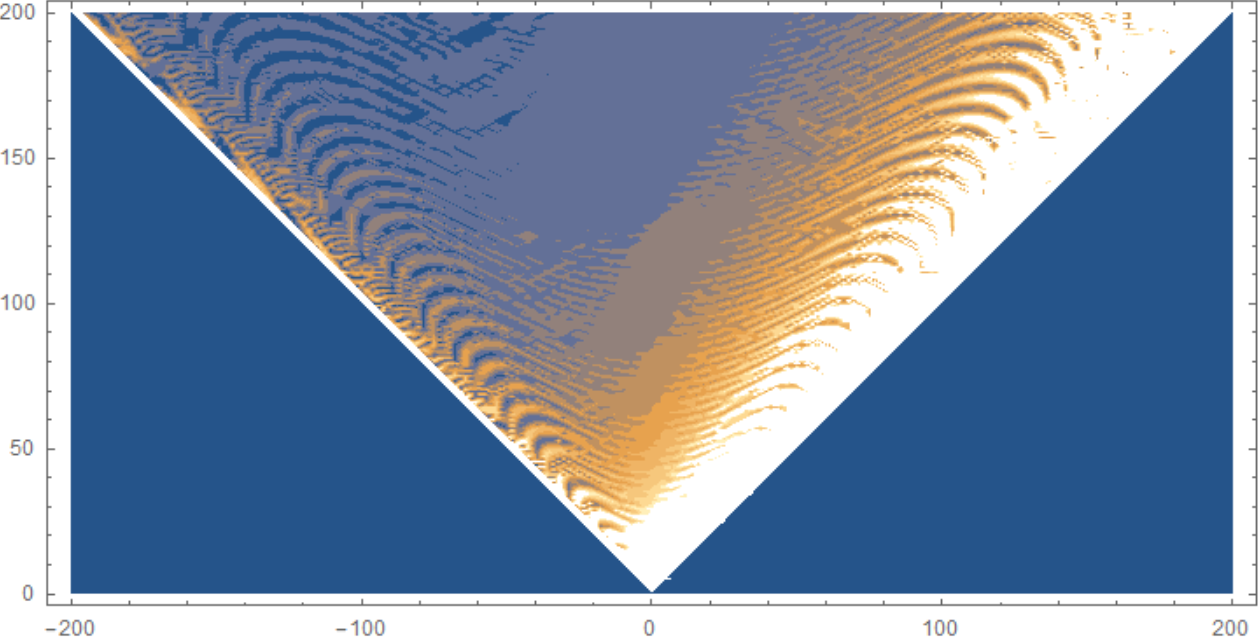}
  \hspace{-0.2cm} \\
  \includegraphics[width=0.20\textwidth]
  {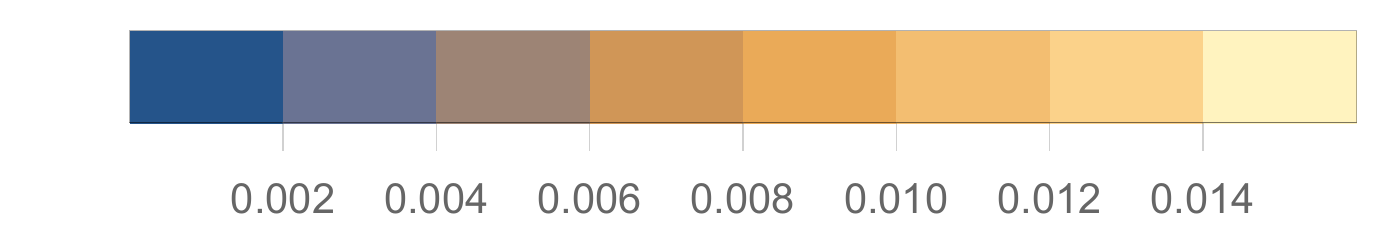}  &
  \includegraphics[width=0.15\textwidth]
  {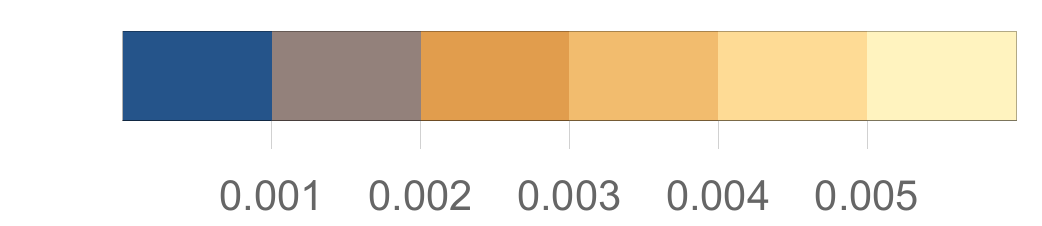} &
  \includegraphics[width=0.22\textwidth]
  {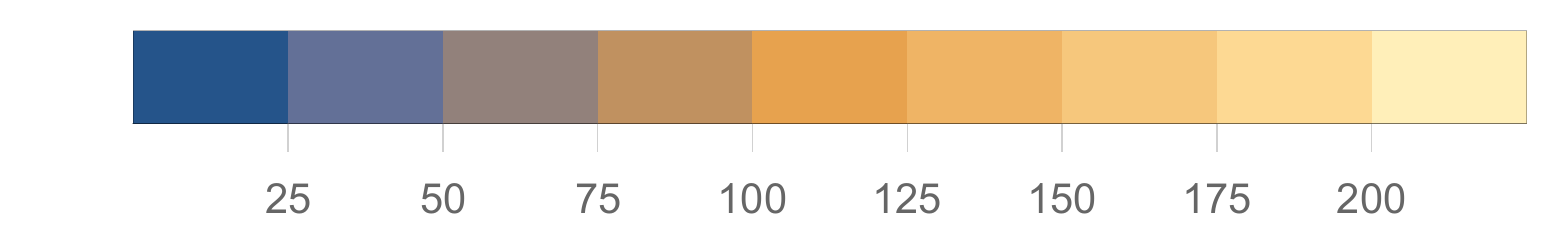}  \\
  \end{tabular}
  \caption{Вероятность обнаружить электрон в малом квадратике вокруг данной точки (белым показаны сильные колебания этой вероятности). Слева: в базовой модели из~\S\ref{sec-basic}  (ср.~\cite[рис.~6]{Yepez-05}). В центре: в уточнении из~\S\ref{sec-mass} для меньшего размера клетки. Справа: в непрерывной теории. В последнем случае показана относительная плотность вероятности.
  } \label{P-contour}
  \vspace{-0.2cm}
\end{figure}

\addcontentsline{toc}{myshrinkalt}{}

\subsection{Мотивировка}

Простейший способ понять, что это за модель, --- это классический \emph{двухщелевой опыт} (см.~рисунок~\ref{Double-slit}). В данном опыте (\emph{когерентный}) пучок электронов направлен на пластину с двумя параллельными щелями. Часть пучка, прошедшая через щели, наблюдается на экране за пластиной. Если одна из щелей закрыта, то луч освещает на экране пятно. Если обе щели открыты, можно было бы ожидать пятно большего размера, но на самом деле мы видим последовательность ярких и темных полос (\emph{интерферограмму}).

Это показывает, что электроны ведут себя, как волны: волны проходят через обе щели, и вклады двух путей либо усиливают, либо гасят друг друга в зависимости от конечных фаз.

Далее, если пропускать электроны через щели по одному, то, как и ожидалось, на экране появляются отдельные точки. Однако примечательно, что та же интерферограмма со светлыми и темными полосами возникает, если электроны накапливать. Нельзя предсказать, где конкретный электрон столкнется с экраном; все, что мы можем сделать, это вычислить вероятность обнаружения электрона в заданном месте.

Метод Фейнмана \emph{суммирования по путям} (или \emph{интеграла по путям}) для вычисления таких вероятностей --- присвоить фазы всем возможным путям и просуммировать полученные волны (см.~\cite{Feynman, Feynman-Gibbs}). \emph{Шашки Фейнмана} (или \emph{шахматная доска Фейнмана}) --- особенно простое комбинаторное правило для этих фаз в случае, когда электрон свободно движется (а лучше сказать, прыгает)
по прямой.

\begin{figure}[htbp]
  \centering
  \includegraphics[width=0.24\textwidth]{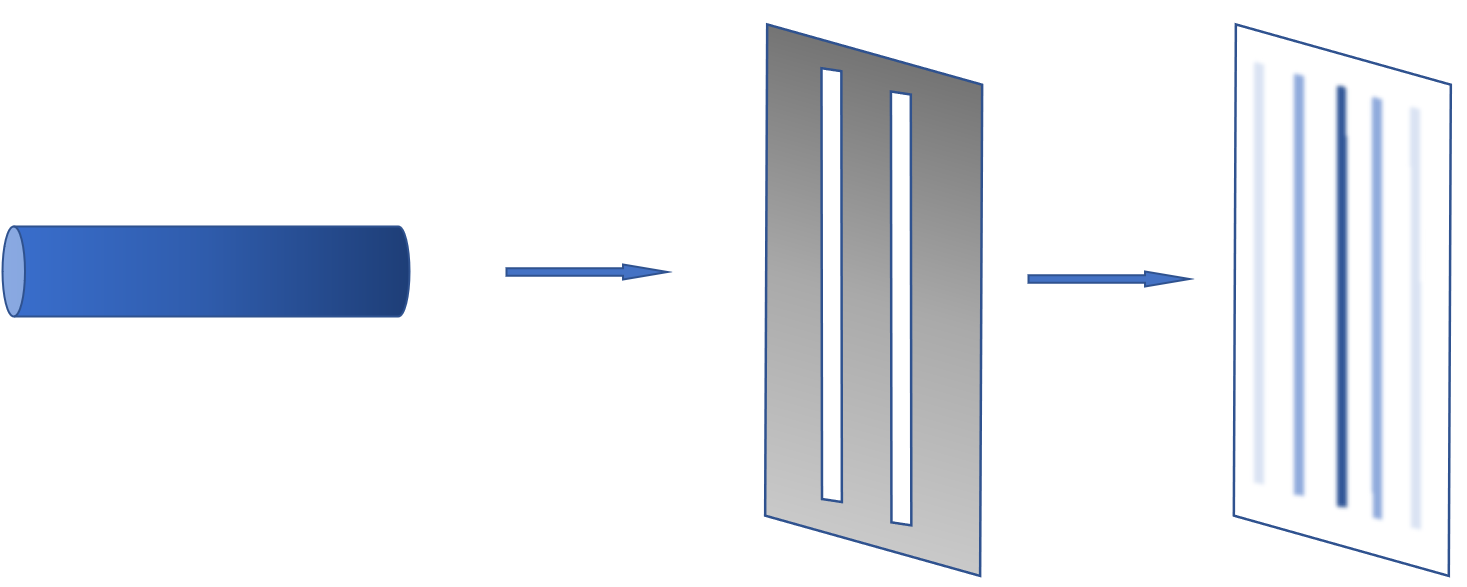}
  \caption{
  Двухщелевой опыт}\label{Double-slit}
  \vspace{-0.2cm}
\end{figure}

\addcontentsline{toc}{myshrinkalt}{}

\subsection{Исторический обзор}
\label{ssec-background}

\textbf{Начало.}
Шашки Фейнмана были придуманы Р.~Фейнманом в 1940-х годах \cite{Schweber-86} и впервые опубликованы в 1965 году в известной книге \cite{Feynman-Gibbs}. 
В этой книге в задаче~2.6 строится функция на решетке с малым шагом $\varepsilon$ (называемая \emph{ядром}; см.~\eqref{eq-def-mass}) и ставится следующая проблема:
\begin{quotation}
  Покажите, что если интервал времени очень велик ($t_b-t_a\gg \hbar/mc^2$), а средняя скорость мала [$x_b-x_a\ll c(t_b-t_a)$], то ядро [если не считать множителя $\exp(imc^2/\hbar)(t_b-t_a)$] совпадает с выражением для
  свободной частицы [см.~(3.3)].
\end{quotation}
Математически это означает, что ядро (деленное на $2i\varepsilon\exp[(-imc^2/\hbar)(t_b-t_a)]$) асимптотически
равно \emph{ядру свободной частицы}~\eqref{eq-free-particle-kernel} (это и есть уравнение~(3.3) из~\cite{Feynman-Gibbs}) в тройном пределе, когда время стремится к бесконечности, а средняя скорость и шаг решетки стремятся к нулю (см.~таблицу~\ref{table-propagators} и рис.~\ref{fig-triple-limit}). Деление на шаг решетки и его стремление к нулю подразумевалось в проблеме, иначе упомянутый ``исключительный'' множитель был бы другим (см.~пример~\ref{p-Feynman-couterexample}).
Мы покажем, что это математическое утверждение, хотя и неверно буквально, выполняется при небольших ограничениях (см.~следствие~\ref{cor-feynman-problem}).

Хотя проблема Фейнмана может показаться простой физикам-теоретикам, даже первого шага математического решения (опровержения утверждения в том виде, как оно сформулировано) не удается найти в литературе. Как обычно, основная трудность в том, чтобы доказать сходимость, а не угадать предел.

\begin{table}[hb]
  \centering
  \begin{tabular}{|l|c|c|c|l|}
    \hline
    пропагатор & \hspace{-0.2cm}непрерывный\hspace{-0.2cm} & \hspace{-0.1cm}решеточный\hspace{-0.1cm} & контекст & ссылки \\
    \hline
    ядро свободной частицы & \eqref{eq-free-particle-kernel} &
    - & 
    квантовая механика &  \cite[(3-3)]{Feynman-Gibbs} \\
    \hline
    запаздывающий пропагатор & \eqref{eq-relativistic-propagator},\eqref{eq-double-fourier-retarded} & \eqref{eq-def-mass} &
    релятивистская 
    & ср.~\cite[(13)]{Jacobson-Schulman-84}
    \\
    для спина $1/2$ & & & квантовая механика & и~\cite[(2-27)]{Feynman-Gibbs}
    \\
    \hline
    фейнмановский пропагатор & \eqref{eq-feynman-propagator},\eqref{eq-double-fourier-feynman} & \eqref{eq-def-anti1} &
    квантовая теория поля & ср.~\cite[\S9F]{Bender-etal-94} \\
     для спина $1/2$ & & & & \\
    \hline
  \end{tabular}
  \caption{Выражения для \emph{пропагаторов} частицы, свободно движущейся в $1$ пространственном и $1$ временном измерении. Квадрат нормы пропагатора имеет смысл относительной плотности вероятности обнаружить частицу в данной точке или, иначе, плотности заряда в этой точке.} \label{table-propagators}
\end{table}

\begin{figure}[htbp]
  \vspace{-0.5cm}
  \centering
  \includegraphics[width=0.90\textwidth]{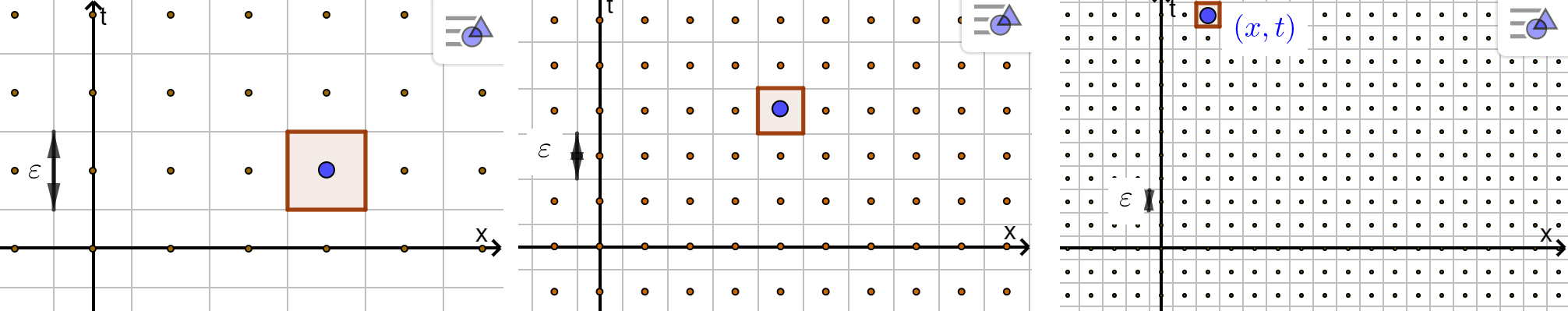}
  \caption{
  Фейнмановский тройной предел: $t\to+\infty$, $x/t\to 0$, $\varepsilon\to 0$
  }\label{fig-triple-limit}
\end{figure}

В 1972 году Дж.~Нарликар \cite{Narlikar-72} обнаружил, что указанное ядро воспроизводит \emph{запаздывающий пропагатор для спина $1/2$} в другом пределе, когда шаг решетки стремится к нулю, а время остается фиксированным  (см.~таблицу~\ref{table-propagators}, рис.~\ref{fig-limit} и~\ref{P-contour}, следствие~\ref{cor-uniform}).
В 1984 году \new{Т.~Джекобсон}--Л.~Шульман \cite{Jacobson-Schulman-84} повторили этот вывод, применили \emph{метод стационарной фазы} (среди других блестящих идей) и нашли вероятность смены направления движения электрона (ср.~с теоремой~\ref{p-right-prob}). Замечательные работы этого периода не содержат математических доказательств, а лишь приближенные вычисления без оценки погрешности.


\begin{figure}[htbp]
  \centering
  \includegraphics[width=0.285\textwidth]{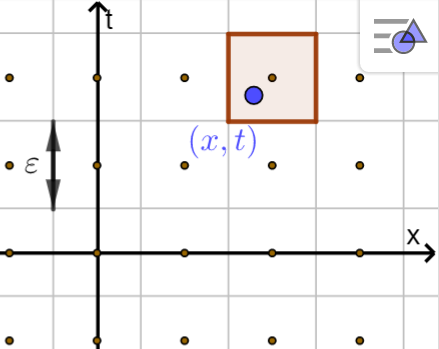}
  \includegraphics[width=0.293\textwidth]{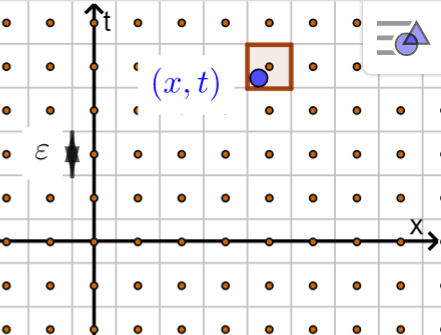}
  \includegraphics[width=0.30\textwidth]{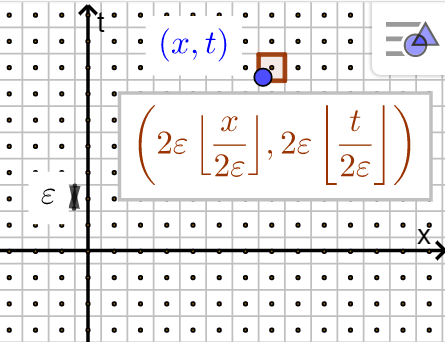}
  \caption{Непрерывный предел: точка $(x,t)$ фиксирована, а шаг $\varepsilon$ решетки стремится к нулю.
  }\label{fig-limit}
\end{figure}

\textbf{Модель Изинга.} В 1981 Х.~Герш заметил, что шашки Фейнмана можно рассматривать как одномерную модель Изинга
при \emph{мнимой} температуре или весах ребер (см.~\S\ref{ssec-interpretation} и \cite{Gersch-81}, \cite[\S3]{Jacobson-Schulman-84}). Мнимые значения этих величин обычны в физике (например, в квантовой теории поля или цепях переменного тока). Из-за мнимости вклады большинства конфигураций сокращаются, что делает модель крайне нетривиальной, не смотря на одномерность. В частности, модель демонстрирует фазовый переход (см. рисунки~\ref{P-contour} и \ref{fig-distribution}). Удивительно, что ранее об этом не упоминалось. Фазовые переходы изучались только в более сложных одномерных моделях Изинга~\cite[\S III]{Jones-66}, \cite{Matveev-Shrock-97}, не смотря на известный эквивалентный результат, которые мы собираемся обсудить сейчас
(см.~теорему~\ref{th-limiting-distribution}(B) и~следствие~\ref{cor-free}).




\begin{figure}[ht]
  \centering
  \includegraphics[width=0.3\textwidth]
  {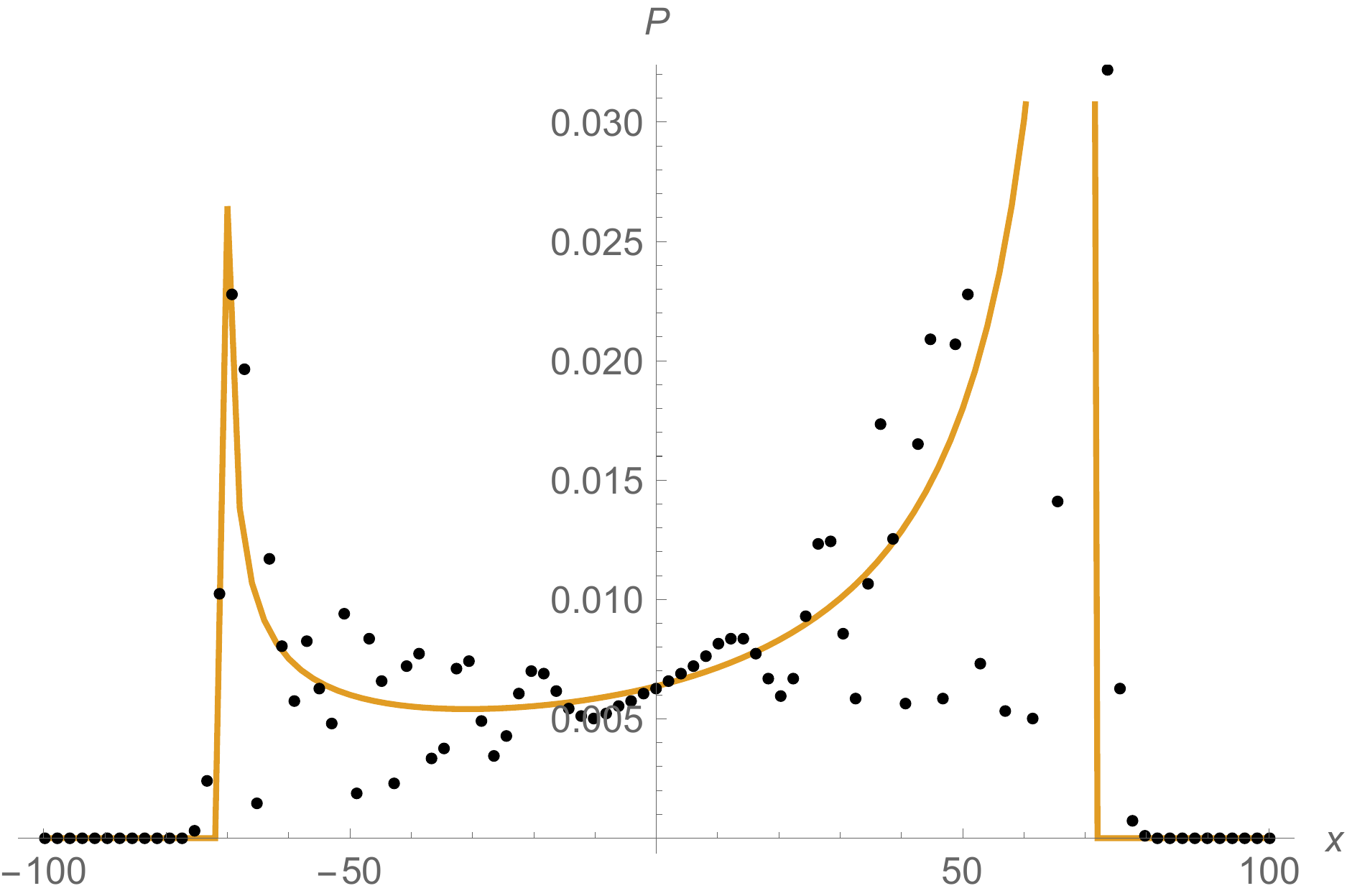}
  \includegraphics[width=0.3\textwidth]
  {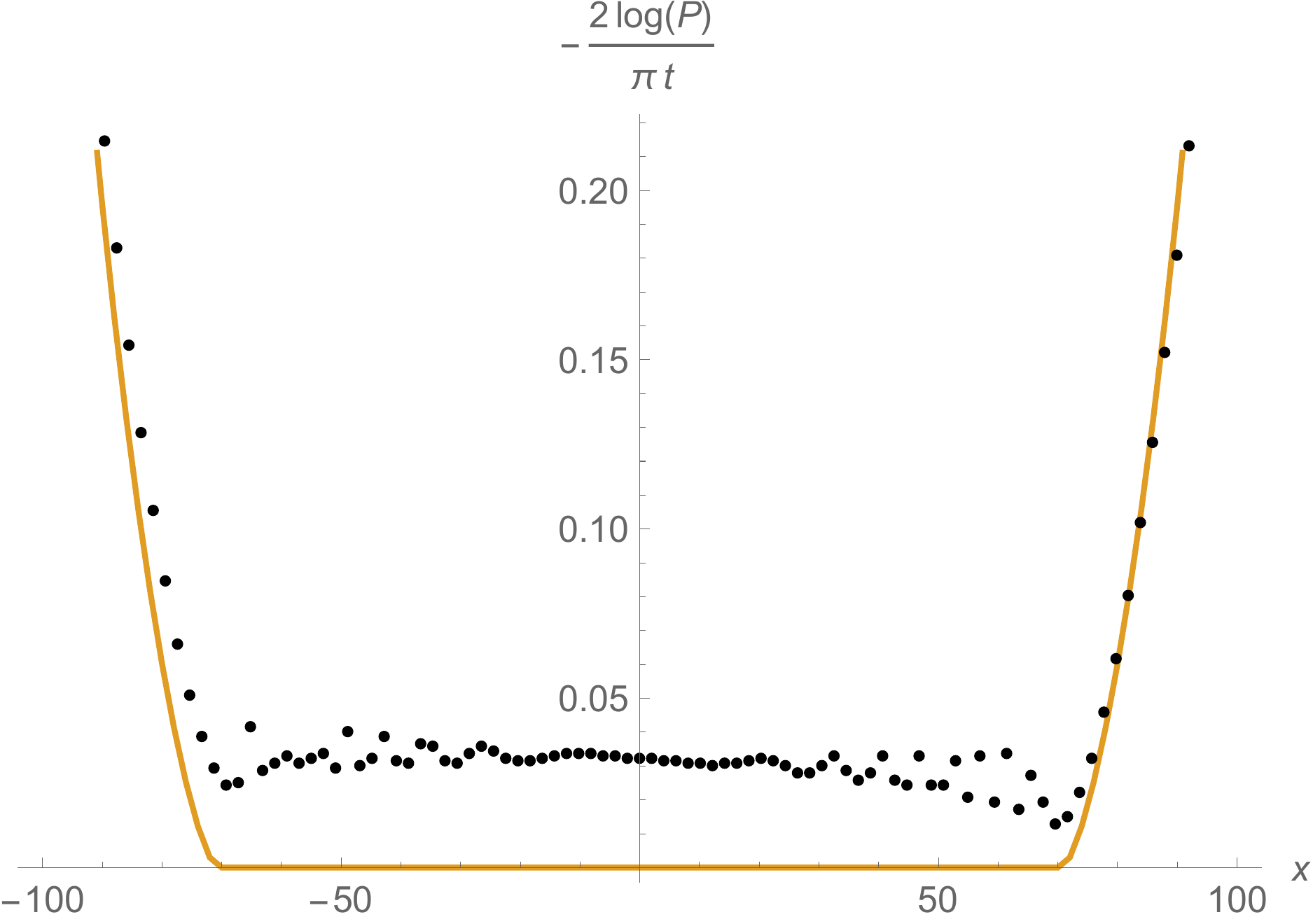}
  \includegraphics[width=0.3\textwidth]
  {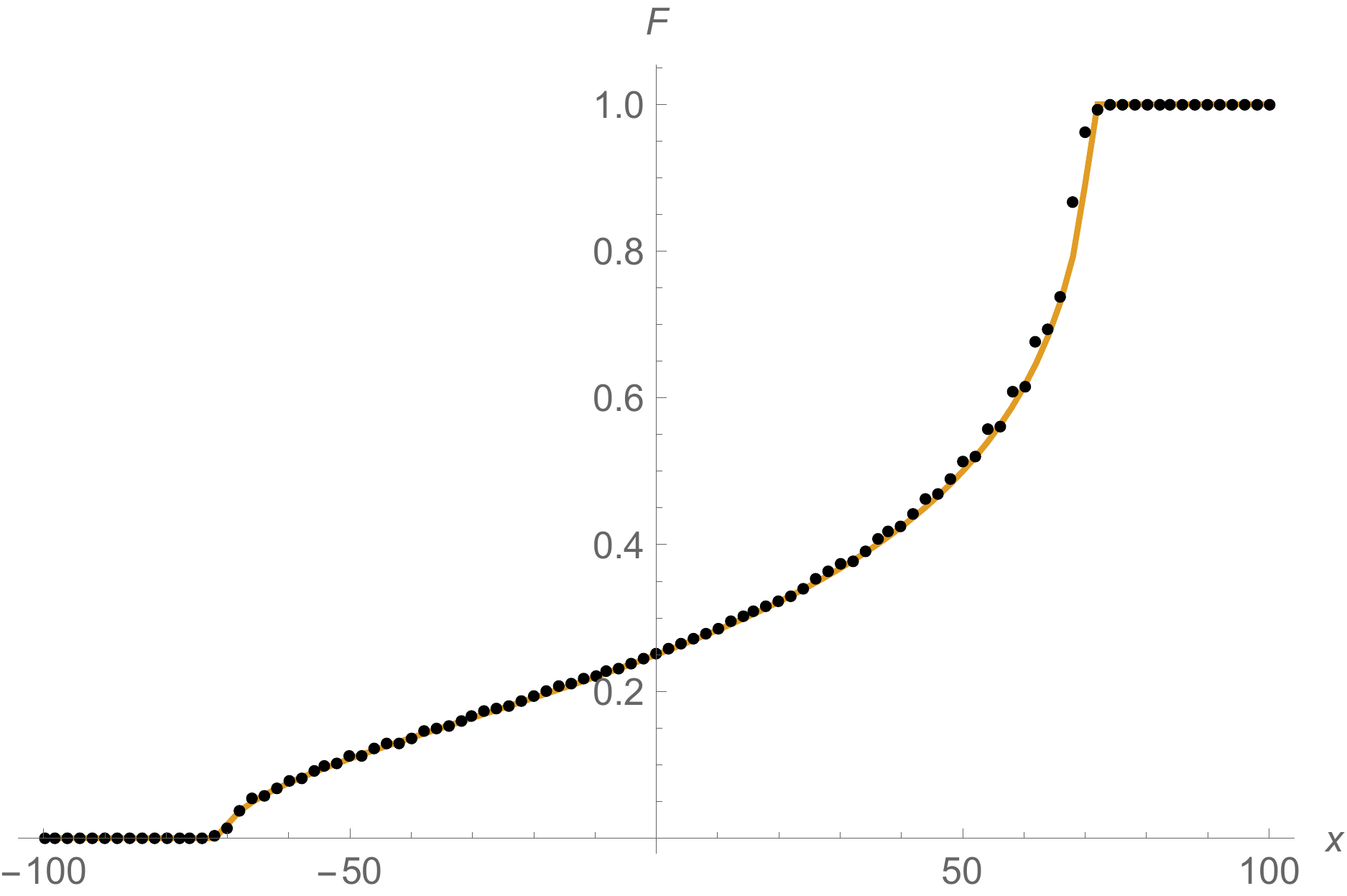}
  \caption{Распределение положения~$x$ электрона в момент времени $t=100$ в естественной системе единиц для базовой модели из~\S\ref{sec-basic} (слева, точки).
  Его нормализованный логарифм (в центре, точки) и функция распределения (справа, точки). Их (слабые) маштабированные пределы при $t\to\infty$ (кривые). 
  Кривая в центре также является графиком (минус мнимой части) предельной плотности свободной энергии в модели Изинга.
  Точки неаналитичности кривых отражают фазовый переход.}
  \label{fig-distribution}
\end{figure}


\textbf{Квантовое блуждание.} В 2001 году А.~Амбаинис, Е.~Бах, А.~Наяк, А.~Вишванат и Дж.~Ватрус совершили прорыв \cite{Ambainis-etal-01}. Они изучали шашки Фейнмана под названиями \emph{одномерное квантовое блуждание} и \emph{блуждание Адамара}; эти нетривиально определяемые модели совершенно эквивалентны фейнмановской.
Они установили предельное поведение модели при больших временах (см. теорему~\ref{th-ergenium}). Они обнаружили несколько поразительных свойств, резко контрастирующих как с непрерывной теорией, так и с классическим случайным блужданием. Во-первых, наиболее вероятная средняя скорость электрона в модели равна $1/\sqrt{2}$ от скорости света, и вероятность превысить это значение крайне мала (см.~рис.~\ref{fig-distribution} и~\ref{P-contour} слева и теорему~\ref{th-limiting-distribution}(B)).
Во-вторых, если поместить поглощающую границу сразу слева от начальной точки, то вероятность поглощения электрона будет равна $2/\pi$.
В-третьих, если установить еще одну поглощающую границу в точке с координатой $x>0$, то вероятность поглощения электрона на левой границе возрастет, достигая $1/\sqrt{2}$ в пределе при $x\to+\infty$. Напомним, что для классического блуждания обе вероятности равны $1$. Еще они нашли множество комбинаторных соотношений и выразили упомянутое выше ядро через значения многочленов Якоби в конкретной точке (см. замечание~\ref{rem-hypergeo}; ср.~\cite[\S2]{Stanley-04}).

Н.~Конно изучал \emph{квантовое блуждание со сносом} \cite{Konno-05, Konno-08}, которое по существу также эквивалентно модели Фейнмана (см.~замечание~\ref{rem-qw}). Он нашел распределение положения электрона в (слабом) пределе больших времен (см. рисунок~\ref{fig-distribution} и теорему~\ref{th-limiting-distribution}(B)). Этот результат был доказан математически Г.~Гримметом--С.~Джэнсоном--П.~Скудо \cite{Grimmett-Janson-Scudo-04}. В выдающейся работе \cite{Sunada-Tate-12} 2012 года Т.~Сунада--Т.Татэ нашли и доказали асимптотику для указанного распределения в пределе больших времен (см.~теоремы~\ref{th-ergenium}--\ref{th-outside}). Это мощный результат, но даже он недостаточен для решения проблемы Фейнмана, так как в нем оценка остаточного члена неравномерна по шагу решетки. \new{В 2018 М.~Маеда и соавторы получили оценку максимума указанного распределения при больших временах
\cite[теорема~2.1]{Maeda-etal-18}.}

\new{}
Квантовые блуждания были обобщены на произвольные графы и применены для построения квантовых алгоритмов (пример реализации см.~на рис.~\ref{tab-implementation}
и в~\cite{Georgopoulus-etal-21}). Для получения более детальной информации по данному направлению рекомендуем обзоры
М.Дж.~Кантеро--Ф.А.~Грюнбаума--Л.~Морала--Л.~Веласкеса,
Н.~Конно, Дж.~Кемпе и С.Е.~Венегас-Андрака
\cite{Cantero-etal-12, Konno-08, Kempe-09, Venegas-Andraca-12, Konno-20}.


\begin{figure}
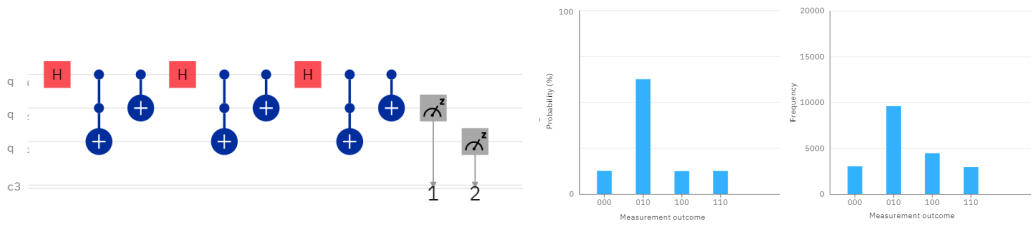

\vspace{-0.3cm}
\begin{center}
\includegraphics[width=0.45\textwidth]{circuit.png}
\includegraphics[width=0.2\textwidth]{bar-chart-simulator-v3.png}
\includegraphics[width=0.2\textwidth]{bar-chart-lima-v2.png}
\end{center}
\vspace{-0.7cm}
\caption{Реализация базовой модели из~\S\ref{sec-basic} на квантовом компьютере с использованием языка квантовых схем (слева). На выходе получается случайная битовая строка, кодирующая положение $x$ электрона в момент времени $t=4$. Строки 000, 010, 100, 110 кодируют $x=4,2,0,-2$ соответственно.
Распределение величины $x$ (в центре) и гистограмма для квантового компьютера IBM Lima (справа) \cite[\S19]{SU-2}. См.~подробнее в~\cite{Georgopoulus-etal-21}.}
\label{tab-implementation}
\end{figure}

\textbf{Решеточные квантовые теории поля.} В более общем контексте, это направление в сторону создания решеточной квантовой теории поля с ``метрикой'' \emph{Минковского}, в которой как пространство, так и время дискретны~\cite{Bender-etal-94}.
В 1970-х годах Ф.~Вегнер и К.~Вильсон предложили \emph{решеточную калибровочную теорию} как средство вычислений в калибровочной теории, описывающей все известные взаимодействия, кроме гравитации; см.~научно-популярное введение в~\cite{Maldacena-16}.
Кульминацией стало теоретическое вычисление массы протона с ошибкой менее $2\%$ в некотором смысле.
Эта теория \emph{евклидова} в том смысле, что время в ней является \emph{мнимой} величиной.
Асимптотика пропагатора для (безмассового)
\emph{евклидова} решеточного уравнения Дирака \cite[теорема~4.3]{Kenyon-02}
сыграла главную роль в непрерывном пределе модели Изинга, найденного Д.~Челкаком--С.~Смирновым 
\cite{Chelkak-Smirnov-11}. 
Аналогичным образом, асимптотики для пропагатора с ``метрикой'' \emph{Минковского} (теоремы~\ref{th-ergenium}--\ref{th-main})
могут пригодиться для отсутствующей пока решеточной квантовой теории поля с такой метрикой. Несколько авторов утверждают, что в шашках Фейнмана \emph{нет удвоения фермионов}, присущего евклидовым решеточным теориям, а условия запрещающей теоремы Нильсена--Ниномия обходятся \cite{Bialynicki-Birula-94, Foster-Jacobson-17}.

В литературе обсуждались уточнения модели Фейнмана.
Например, в районе 1990-х годов Б.~Гаво--Л.~Шульман и Г.~Орд добавили в нее электромагнитное поле \cite{Schulman-Gaveau-89, Ord}. Тогда они не достигли ни сохранения заряда, ни обобщения на неабелевы калибровочные теории; это сделано в определении~\ref{def-external}.
Другой пример~--- добавление массовой матрицы П.~Джизбой
\cite{Jizba-15}.


Давней мечтой было включить в модель пути шашки, поворачивающие вспять во времени или образующие циклы \cite[p.~481--483]{Schweber-86}, \cite{Jacobson-85}; это означало бы рождение электрон-позитронных пар, знаменуя переход от квантовой механики к квантовой теории поля. Искали комбинаторную модель, дающую \emph{фейнмановский}, а не \emph{запаздывающий} пропагатор в непрерывном пределе (см.~таблицу~\ref{table-propagators}). 
Известные конструкции (такие, как \emph{прыжковое разложение})
не приводили к фейнмановскому пропагатору из-за некоторых решеточных артефактов
(так, название работы \cite{Ord-Gualtieri} может ввести в заблуждение: фейнмановский пропагатор там не обсуждается). В \emph{безмассовом} случае \emph{некомбинаторное} построение
фейнмановского пропагатора на решетке имеется у К.~Бендера--Л.~Мэда--К.~Милтона--Л.~Шарпа \cite[\S9F]{Bender-etal-94} и~\cite[\S IV]{Bender-etal-85}. 
В определении~\ref{def-anti-combi}, наконец, дано искомое комбинаторное построение.

Другая давняя открытая проблема --- обобщить модель на настоящий \emph{$4$-мерный} мир. В нобелеской лекции Р.~Фейнман упоминает о своих неудачных попытках. Есть ряд новых подходов, скажем, в работе Б.~Фостера--\new{Т.~Джекобсона} 2017 года \cite{Foster-Jacobson-17}. Но те не так просты и красивы, как исходная двумерная модель, что отмечается в самой работе \cite[\S7.1]{Foster-Jacobson-17}.

\textbf{О физических и математических работах.}
Физическая литература на эту тему весьма обширна \cite{Venegas-Andraca-12}, и мы не можем упомянуть все замечательные работы в этом кратком обзоре. Удивительно, что в литературе мы не нашли свойства \emph{концентрации меры} при стремлении шага решетки к нулю, которое сразу обращает на себя внимание (см.~следствие~\ref{cor-concentration}). Многие статьи написаны настолько хорошо, что \emph{физические} теоремы и доказательства в них могут быть неосторожно приняты за  \emph{математические} (см.~конец~\S\ref{sec-proofs}).
Нам известно всего несколько математических работ на эту тему, таких как \cite{Grimmett-Janson-Scudo-04,Sunada-Tate-12,Maeda-etal-18}.


\addcontentsline{toc}{myshrinkalt}{}

\subsection{Основные результаты} \label{ssec-contributions}

Мы приводим математическое доказательство гипотезы Фейнмана 1965 года о том, что его дискретная модель согласуется с непрерывной, а именно, воспроизводит обычное квантовомеханическое ядро свободной частицы при больших временах, малой средней скорости и малом шаге решетки (см.~следствие~\ref{cor-feynman-problem}). Мы вычисляем асимптотики решеточного пропагатора при больших временах и при малом шаге решетки, равномерные по параметрам модели (см.~теоремы~\ref{th-ergenium} и~\ref{th-main}). Мы впервые отмечаем и доказываем концентрацию меры в непрерывном пределе: средняя скорость электрона, испущенного точечным источником, близка к скорости света с высокой вероятностью (см.~следствие~\ref{cor-concentration}). Эти результаты можно интерпретировать как асимптотические свойства диаграмм Юнга (см.~следствие~\ref{cor-young}) и многочленов Якоби (см.~замечание~\ref{rem-hypergeo}).

Все эти результаты  математически доказываются впервые. Для их формулировки достаточно лишь определения~\ref{def-mass}.
В определениях~\ref{def-external} и~\ref{def-anti-combi} мы
устанавливаем связь с решеточной калибровочной теорией и производим вторичное квантование модели, превращая шашки Фейнмана в полноценную квантовую теорию поля.

\addcontentsline{toc}{myshrinkalt}{}

\subsection{Структура работы и план дальнейших исследований}

Вначале мы даем определения и точные формулировки результатов, и в процессе приводим примеры, иллюстрирующие простейшие концепции квантовой механики и не требующие предварительных знаний. Это те самые примеры, которые Р.~Фейнман приводит первыми в своих собственных книгах. Так, шашки Фейнмана (см.~\S\ref{sec-mass}) --- первый конкретный пример во всей книге~\cite{Feynman-Gibbs}. Отражение от тонкой пленки (см.~\S\ref{sec-medium})  --- первый пример в~\cite{Feynman}; ср.~с рисунками 10--11 там. Поэтому мы надеемся, что эти примеры могут прояснить суть происходящего читателям, не знакомым с квантовой теорией.

Мы начинаем с простейшего (и грубого) частного случая модели и шаг за шагом уточняем ее в каждом следующем разделе. Перед каждым уточнением мы резюмируем: физический \emph{вопрос}, которого оно касается; упрощающие \emph{предположения}, которые оно отбрасывает или  дополнительно накладывает; экспериментальные и теоретические \emph{результаты}, которые оно воспроизводит. Некоторые уточнения (\S\S\ref{sec-medium}--\ref{sec-creation}) лишь анонсируются, и будут обсуждаться подробно в последующих публикациях. Наша цель ---
\emph{(1+1)-мерная решеточная квантовая электродинамика} (``КЭД''), но последние шаги на пути к ней (упомянутые в~\S\ref{sec-QED}) до сих пор не сделаны. Открытые проблемы собраны в~\S\ref{sec-open}.
Для удобного ориентирования приводим схему зависимости уточнений: 
$$
\xymatrix{
\boxed{\text{\ref{sec-basic}. Базовая модель}} \ar[r]\ar[d]\ar[rd]
&
\boxed{\text{\ref{sec-mass}. Масса}} \ar[r]\ar[rd]
&
\boxed{\text{\ref{sec-source}. Источник}} \ar[r]
&
\boxed{\text{\ref{sec-medium}. Среда}}
\\
\boxed{\text{\ref{sec-spin}. Спин}} \ar[rd]
&
\boxed{\text{\ref{sec-external field}. Внешнее поле }} \ar[rd] 
&
\boxed{\text{\red{\ref{sec-creation}}. Античастицы}} \ar[d]
&
\\
&
\boxed{\text{\red{\ref{sec-IdParAnt}.} Тождественные частицы}} \ar[r]
&
\boxed{\text{\red{\ref{sec-QED}}. КЭД}}
&
}
$$
\comment
$$
\xymatrix{
\boxed{\text{\ref{sec-basic}. Basic model}} \ar[dr]
&
\boxed{\text{\ref{sec-external field}. External field }} \ar[r]
&
\boxed{\text{\red{\ref{sec-interaction}. Interaction 
}}} \ar[dr] 
&
\\
\boxed{\text{\ref{sec-medium}. Medium}}
&
\boxed{\text{\ref{sec-spin}. Spin}} \ar[r]\ar[d]\ar[u]
&
\boxed{\text{\red{\ref{sec-IdParAnt}.} Identical particles}} \ar[d]
\red{\ar[u]}
&
\boxed{\text{\red{\ref{sec-QED}}. QED}}
\\
\boxed{\text{\ref{sec-source}. Source}} \ar[u]
&
\boxed{\text{\ref{sec-mass}. Mass}} \ar[r]\ar[l]
&
\boxed{\text{\red{\ref{sec-creation}}. Antiparticles}} \ar[ur]
}
$$
\endcomment

Мы надеемся, что это перспективный путь, чтобы сделать квантовую теорию поля математически строгой и алгоритмической. \emph{Алгоритмическая} квантовая теория поля --- это такая, которая для данной экспериментально измеряемой величины и числа $\delta>0$ давала бы \emph{точную формулировку} алгоритма, предсказывающего эту величину с точностью~$\delta$.
(Разумеется, предсказание не обязано согласовываться с экспериментом для $\delta$, меньших точности самой теории.) См.~игрушечный пример в алгоритме~\ref{alg-main}. Это было бы продолжением \emph{конструктивной} квантовой теории поля
(пока что далекой от алгоритмической). Применения квантовой теории к компьютерным наукам сейчас популярны, но и противоположное направление может быть плодотворно. Алгоритмическое мышление --- еще и способ сделать предмет понятным неспециалистам, как происходит, например, в алгебраической топологии.

Данная статья написана на математическом уровне строгости в том смысле, что все определения, соглашения и теоремы (в.т.ч. следствия, предложения, леммы) нужно понимать буквально. Теоремы остаются верными, даже если их вырезать из текста. Доказательства теорем используют формулировки, но не доказательства других теорем. Большинство формулировок гораздо менее техничны, чем доказательства; поэтому доказательства собраны в отдельном разделе (\S\ref{sec-proofs}), а длинные вычисления --- в~\cite{SU-2}. В процессе доказательств мы даем не требующее предварительных знаний введение в основные методы изучения модели Фейнмана: комбинаторные тождества, интеграл Фурье, метод моментов, метод стационарной фазы. Замечания неформальны и обычно не используются в остальном тексте (а поэтому их можно пропустить). Текст вне определений, теорем, доказательств также формально не используется.


\addcontentsline{toc}{myshrink}{}

\section{Базовая модель (блуждание Адамара)} \label{sec-basic}


{
\hrule
\footnotesize
	\noindent\textbf{Вопрос:} какова вероятность обнаружения электрона в клетке $(x,t)$, если он был испущен из $(0,0)$?

	\noindent\textbf{Предположения:} нет самодействия, нет рождения электрон-позитронных пар, фиксированные масса и шаг решётки, точечный источник; электрон движется либо по плоскости, ``смещаясь равномерно вдоль оси~$t$'', либо по прямой (и тогда $t$ --- это время)
	
	\noindent\textbf{Результаты:} двухщелевой опыт  (качественное объяснение), сохранение заряда, предельное распределение при больших временах

\hrule
}

\medskip

\subsection{Определение и примеры}

Мы начнём с неформального определения в духе~\cite{Feynman}, а затем дадим точное.

На бесконечной шахматной доске шашка ходит на соседнюю по диагонали клетку, влево-вверх или вправо-вверх (см. рис.~\ref{Checker-paths} слева). Каждому пути $s$ шашки сопоставим вектор $a(s)$ на плоскости следующим образом (см. рис.~\ref{Checker-paths} в центре). Рассмотрим секундомер, стрелка которого вращается по мере движения шашки. В начале движения стрелка секундомера направлена вверх.
Пока шашка движется вдоль прямой, стрелка не двигается, а после каждого поворота шашки она поворачивается на $90^\circ$ по часовой стрелке (независимо от того, в какую сторону повернула шашка). Итоговое положение стрелки --- это направление вектора ${a}(s)$. Длина вектора равна $2^{(t-1)/2}$, где $t$ --- общее число ходов шашки (это просто нормировка). Полученный в итоге вектор и есть $a(s)$. Например, для пути $s$ на рис.~\ref{Checker-paths} (вверху в центре) получаем $a(s)=(0,-1/2)$. 

\begin{figure}[htbp]
\begin{center}
\begin{tabular}{l}
\includegraphics[width=3.5cm]{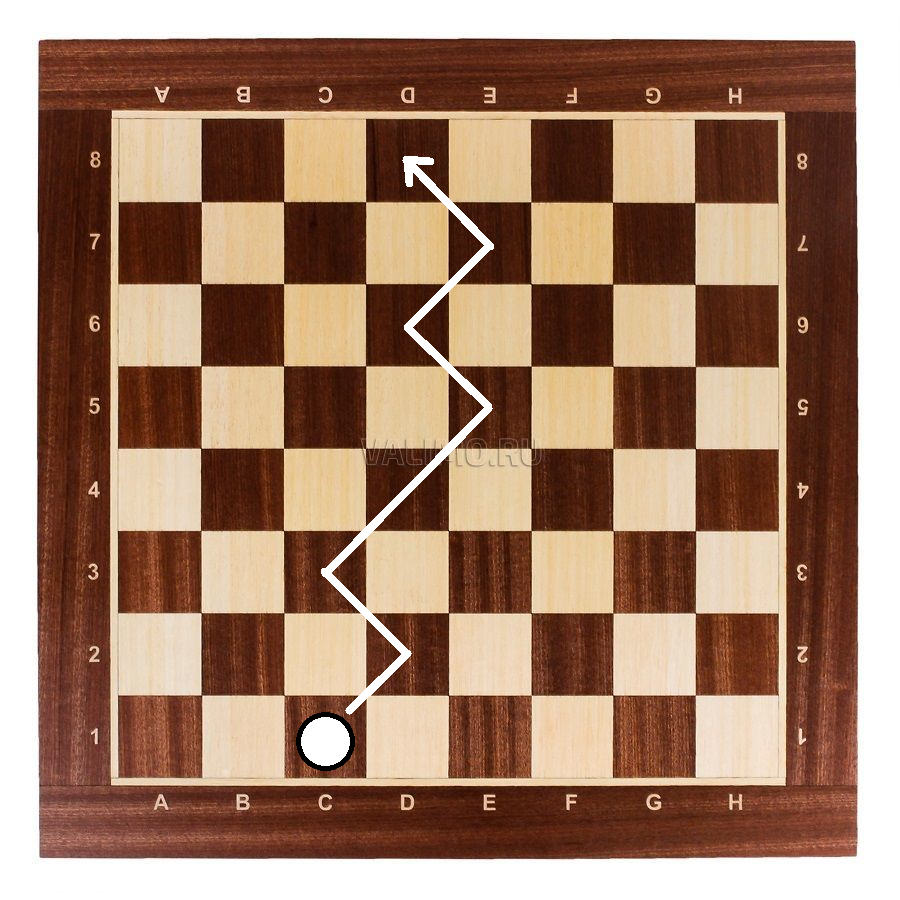}
\end{tabular}
\begin{tabular}{lccc}

\includegraphics[width=1.4cm]{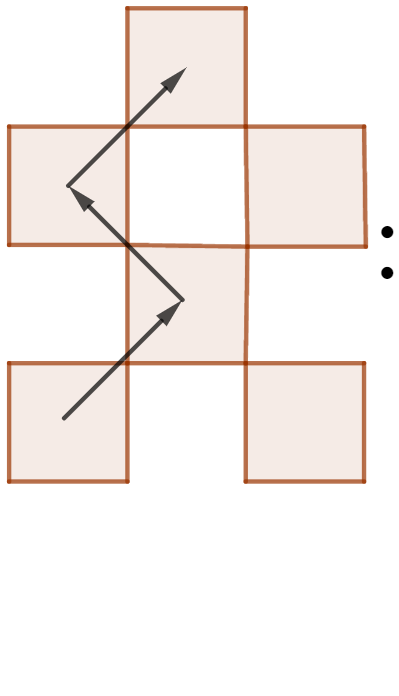} &
\includegraphics[width=1.5cm]{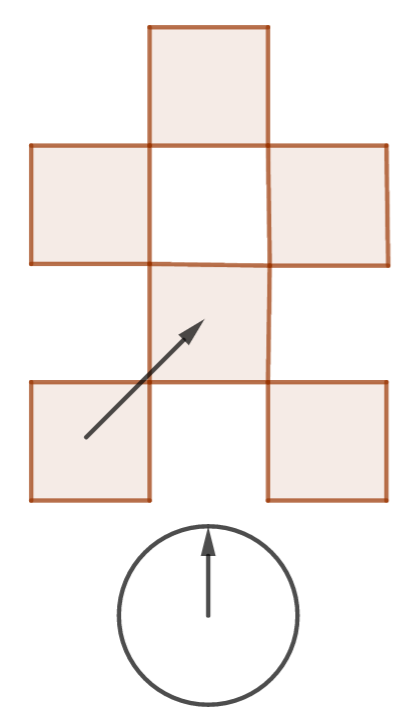} &
\includegraphics[width=1.5cm]{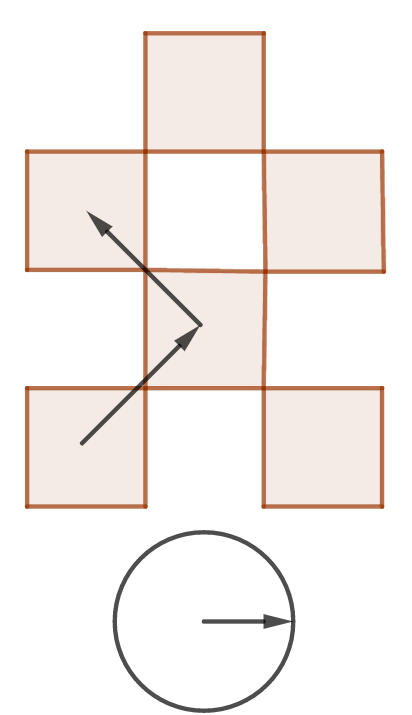} &
\includegraphics[width=1.5cm]{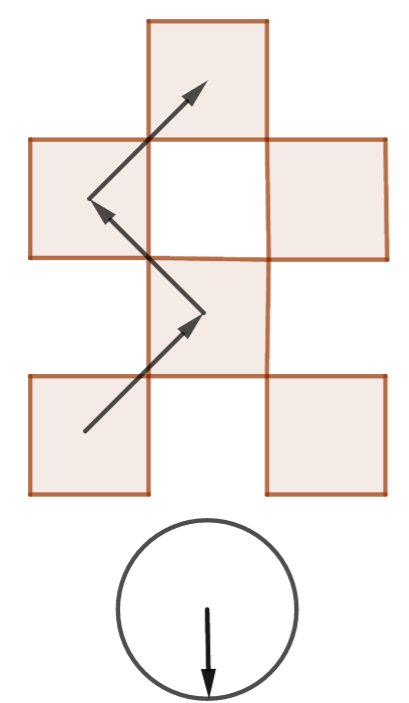} \\
\includegraphics[width=1.5cm]{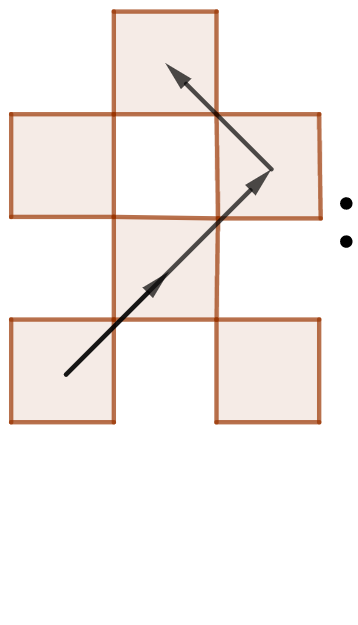} &
\includegraphics[width=1.5cm]{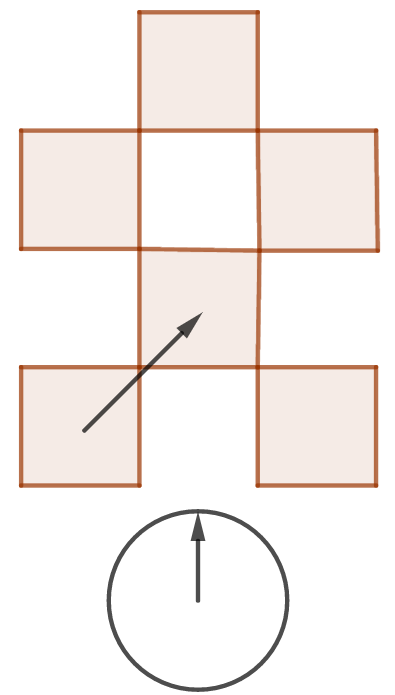} &
\includegraphics[width=1.5cm]{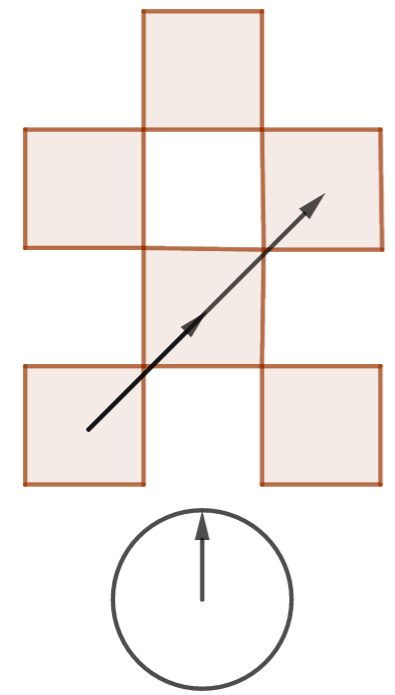} &
\includegraphics[width=1.5cm]{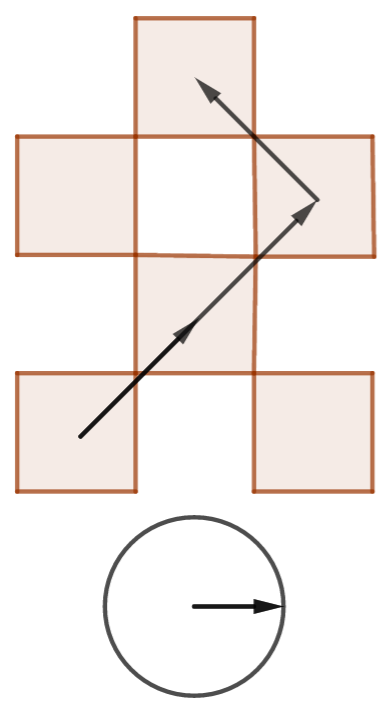}
\end{tabular}
\begin{tabular}{r}
\new{}
\includegraphics[width=3.8cm]{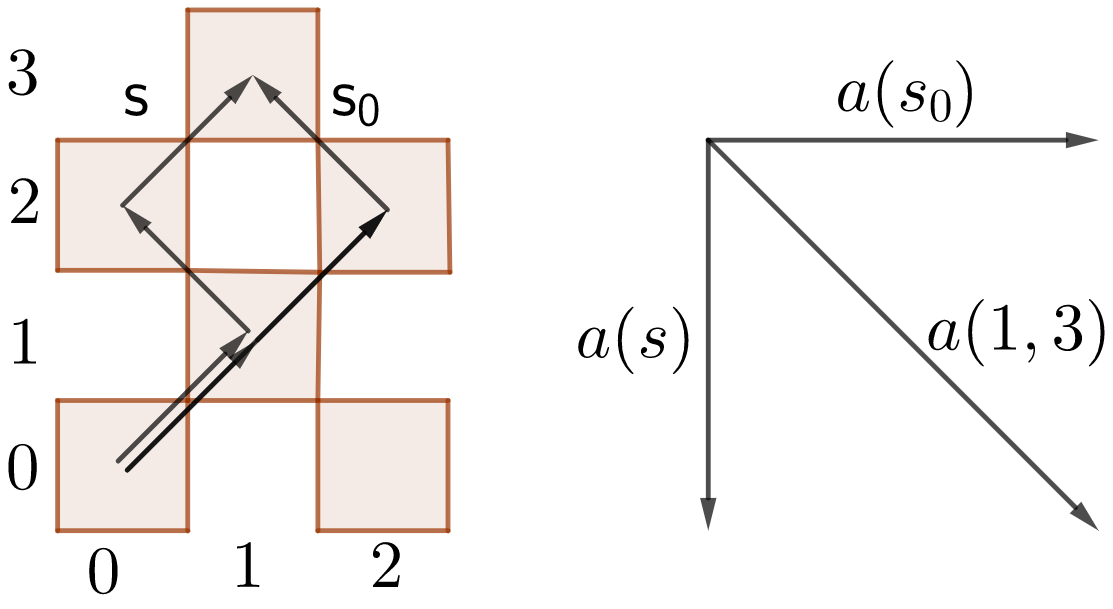}
\end{tabular}
\end{center}
\vspace{-0.6cm}
\caption{ Путь шашки (слева). Векторы, сопоставленные путям (в центре) и клетке (справа).}
\label{Checker-paths}
\vspace{-0.3cm}
\end{figure}

Обозначим $a(x,t):=\sum_s a(s)$, где сумма ведется по всем путям шашки из клетки $(0,0)$ в клетку $(x,t)$, \emph{начинающихся с хода вправо-вверх}. Например, $a(1,3)=(0,-1/2)+(1/2,0)=(1/2,-1/2)$;
см. рис.~\ref{Checker-paths} справа.
Квадрат длины вектора $a(x,t)$ называется \emph{вероятностью обнаружения в точке $(x,t)$ электрона, испущенного из точки $(0,0)$}. (Данная терминология поясняется в \S\ref{ssec-interpretation}). Сам вектор ${a}(x,t)$ называется \emph{стрелкой} \cite[Figure~6]{Feynman}.

Резюмируем описанную конструкцию строго.

\begin{definition} \label{def-basic} \emph{Путь шашки} --- это такая конечная последовательность целых точек плоскости, что вектор из каждой точки (кроме последней) к следующей равен либо $(1,1)$, либо $(-1,1)$. \emph{Поворот} --- это такая точка пути  (не первая и не последняя), что вектор, соединяющий эту точку с предыдущей, ортогонален вектору, соединяющему её со следующей.
\emph{Стрелка} --- это комплексное число
$$
{a}(x,t):=2^{(1-t)/2}\,i\,\sum_s (-i)^{\mathrm{turns}(s)},
$$
где сумма берётся по всем путям $s$ шашки из $(0,0)$ в $(x,t)$ с первым шагом в $(1,1)$, а $\mathrm{turns}(s)$ обозначает общее число поворотов в $s$.
Здесь и далее пустая сумма по определению считается равной $0$.
Обозначим  $$P(x,t):=|{a}(x,t)|^2,
\qquad a_1(x,t):=\mathrm{Re}\,{a}(x,t),
\qquad a_2(x,t):=\mathrm{Im}\,{a}(x,t).
$$
Точки (или клетки) $(x,t)$ с чётным и нечетным $x+t$ называются \emph{чёрными} и \emph{белыми} соответственно.
\end{definition}

\begin{figure}[htbp]
\begin{center}
\includegraphics[width=0.40\textwidth]{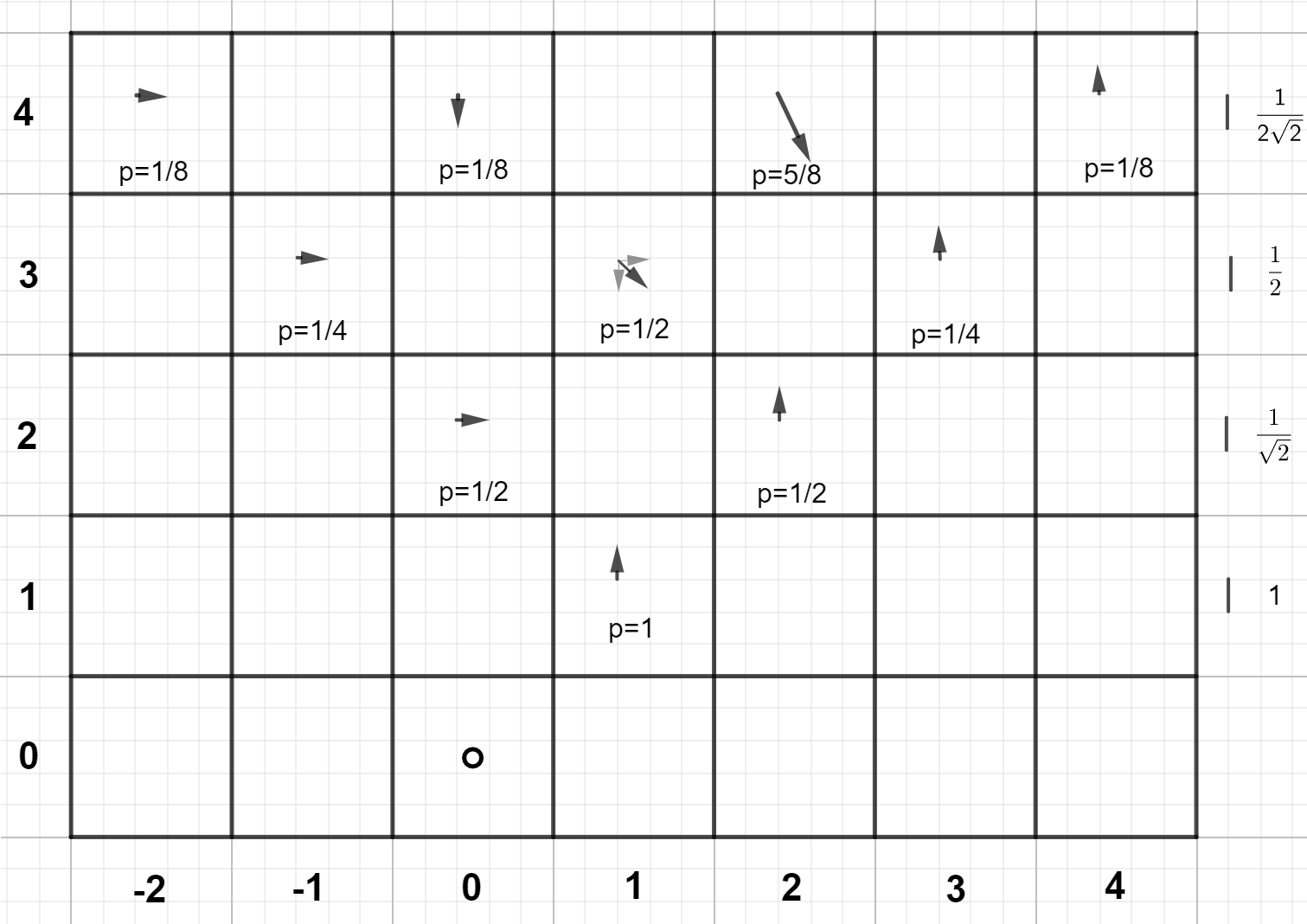}
\includegraphics[width=0.50\textwidth]{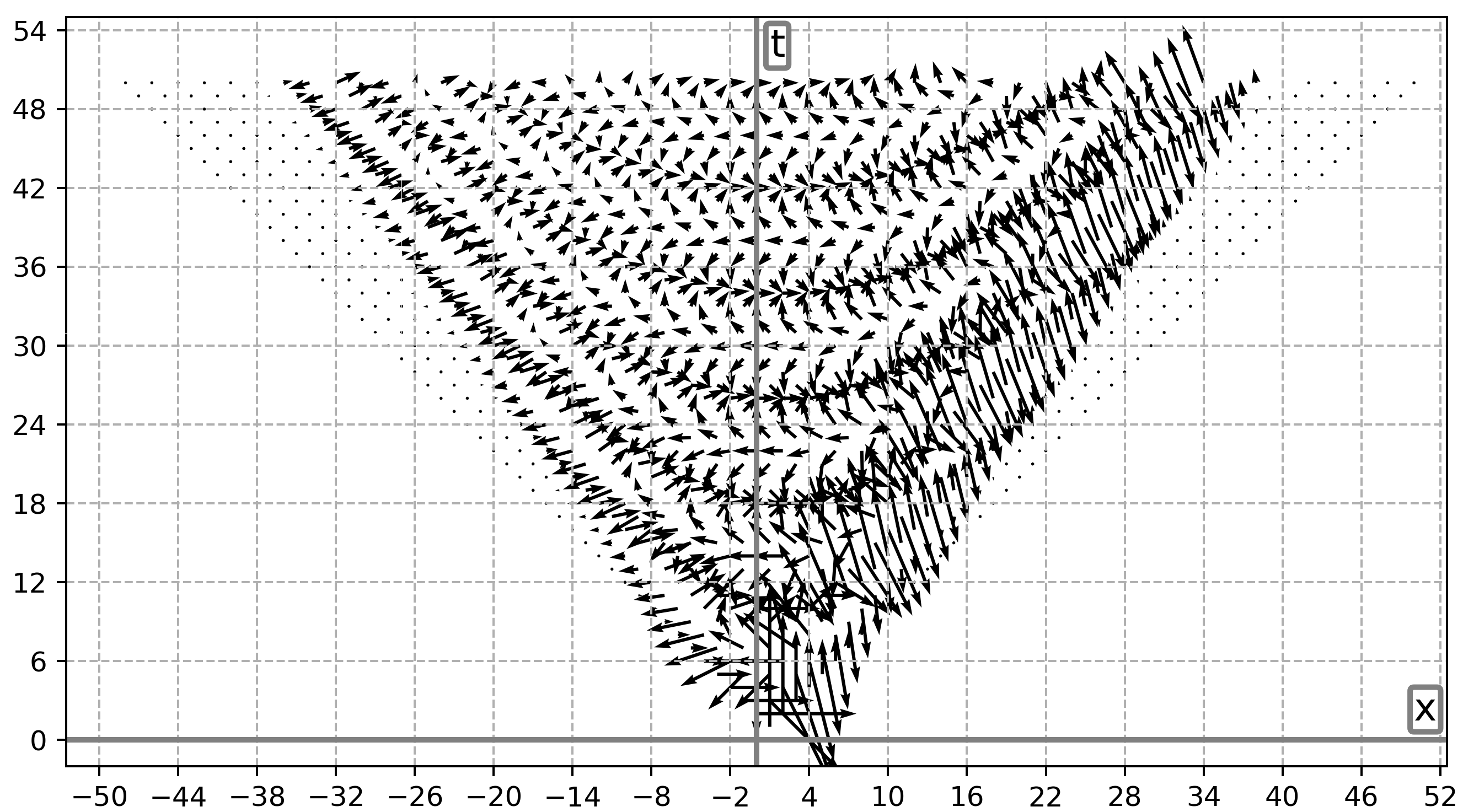}
\end{center}
\vspace{-0.4cm}
\caption{
Стрелки  $a(x,t)$ и вероятности $P(x,t)$ для малых $x,t$ (слева)\new{;} масштаб зависит от строки. Стрелки $10\cdot a(x,t)$ для $t\le 50$ (справа).}
\label{fig-arrows}
\end{figure}


\begin{table}[htbp]
\begin{tabular}{|c*{8}{|>{\centering\arraybackslash}p{55pt}}|}
\hline
	$4$&$\frac{1}{2\sqrt{2}}$ &&$-\frac{i}{2\sqrt{2}}$ &&$\frac{1-2i}{2\sqrt{2}}$ &&$\frac{i}{2\sqrt{2}}$ \\
	\hline
	$3$&&$\frac{1}{2}$&&$\frac{1-i}{2}$&&$\frac{i}{2}$&\\
	\hline
	$2$&&& $\frac{1}{\sqrt{2}}$&&$\frac{i}{\sqrt{2}}$&&\\
	\hline
	$1$&&&&$i$&&&\\
	\hline
	\diagbox[dir=SW,height=18pt]{t}{x}&$-2$&$-1$&$0$&$1$&$2$&$3$&$4$\\
	\hline
\end{tabular}
\caption{Стрелки $a(x,t)$ для малых $x,t$.}
\label{table-a}
\end{table}

Рис.~\ref{fig-arrows} и таблица~\ref{table-a} изображают стрелки $a(x,t)$ и вероятности $P(x,t)$ для малых $x,t$.
Рис.~\ref{a-1000} изображает графики $P(x,1000)$, $a_1(x,1000)$ и $a_2(x,1000)$ как функций от чётного~$x$. Мы видим, что с изменением конечного положения $x$ в некоторый фиксированный большой момент времени $t$, сразу после достижения пика вероятность резко падает до очень малых, но всё ещё ненулевых значений. Интересно неожиданное положение пика, далёкое от $x = t$. На рис.~\ref{P-contour} слева цвет точки $(x,t)$, где $x+t$ чётно, показывает значение $P(x,t)$. Обратим внимание, что стороны  видимого на рисунке угла --- это \emph{не} прямые $t=\pm x$, как можно было бы ожидать. Стороны этого угла --- это $t=\pm \sqrt{2}\,x$ (см. теорему~\ref{th-limiting-distribution}(A)).
\begin{figure}[htbp]
\begin{tabular}{ccc}
\includegraphics[width=0.31\textwidth]{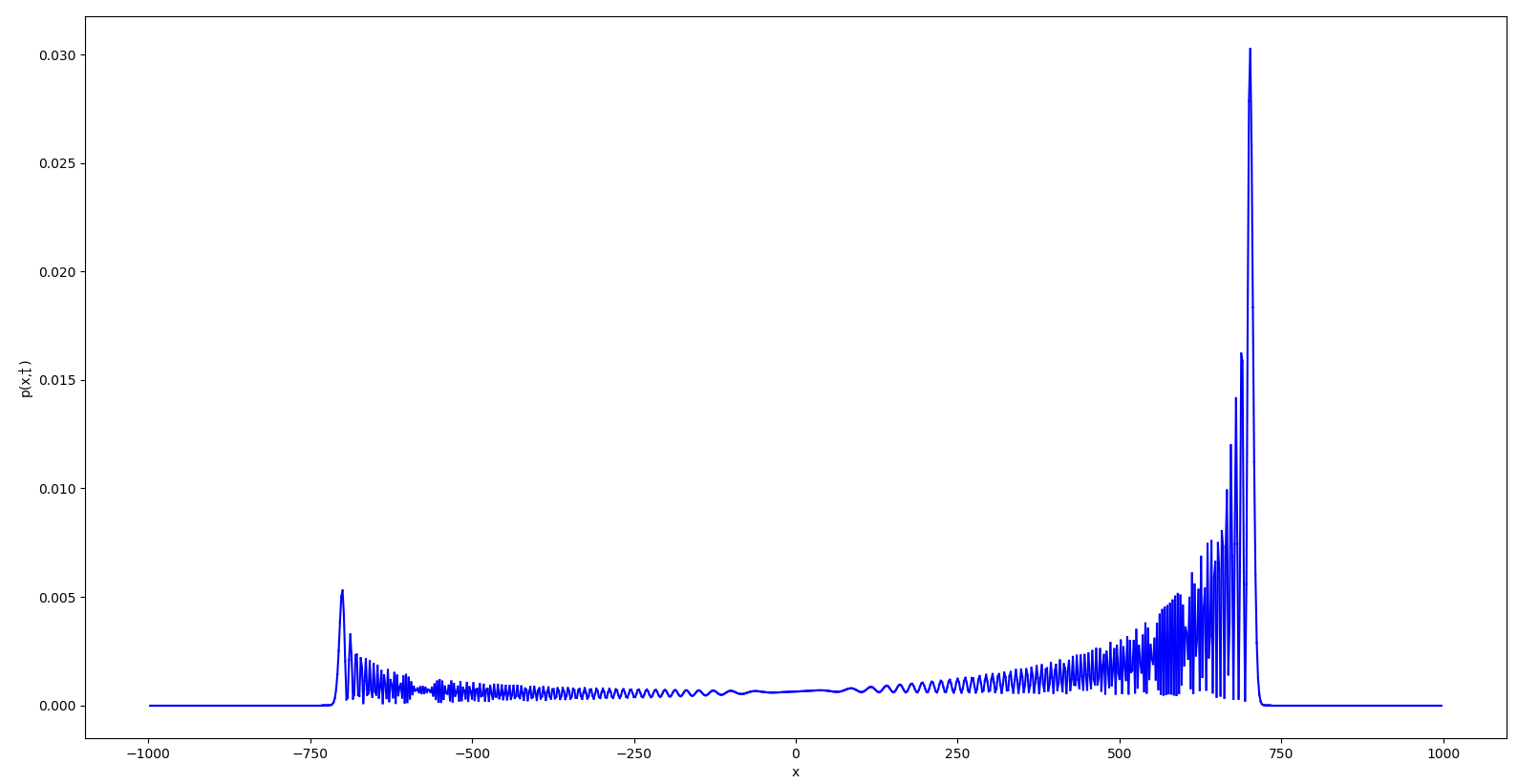}&
\includegraphics[width=0.31\textwidth]{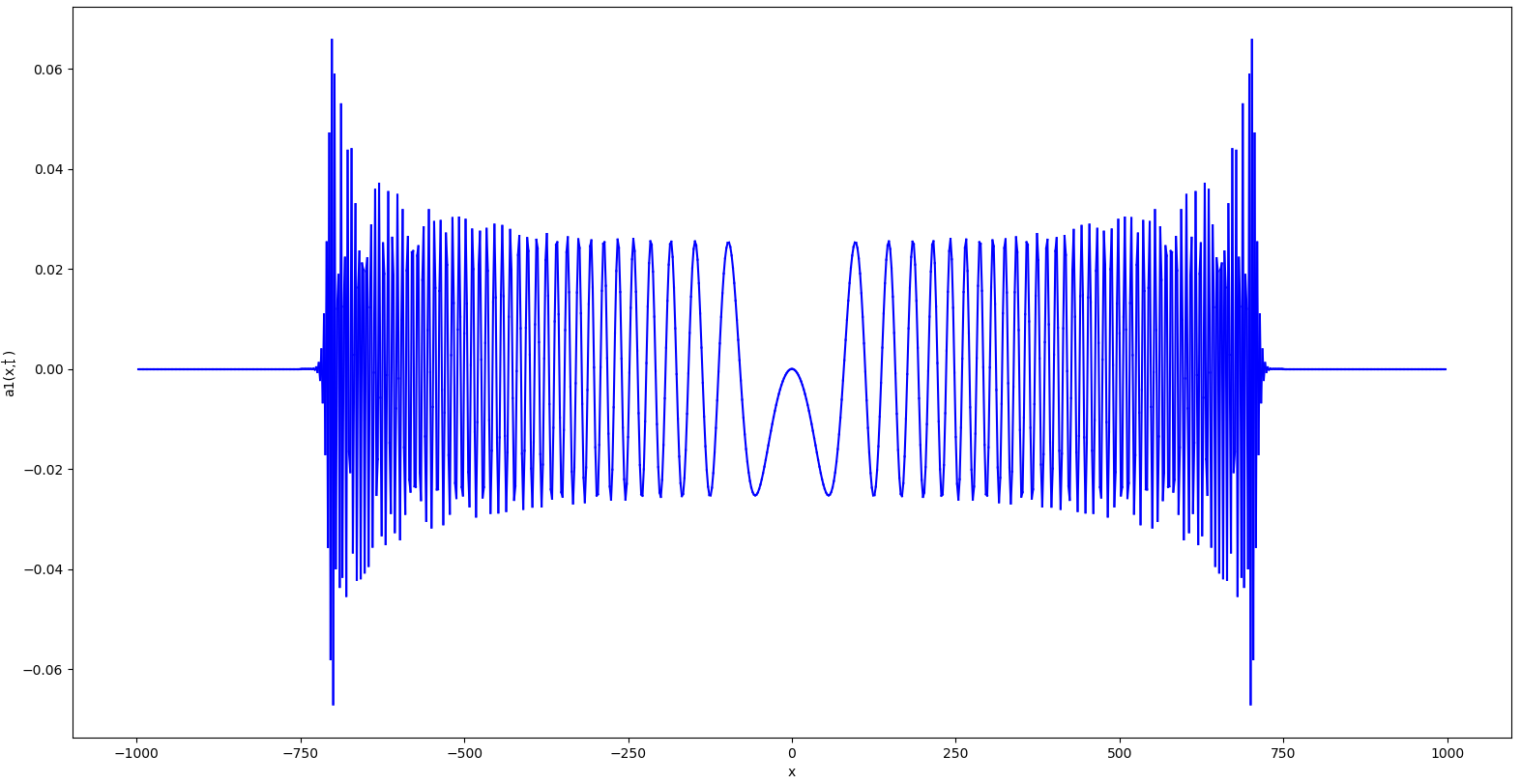} &
\includegraphics[width=0.31\textwidth]{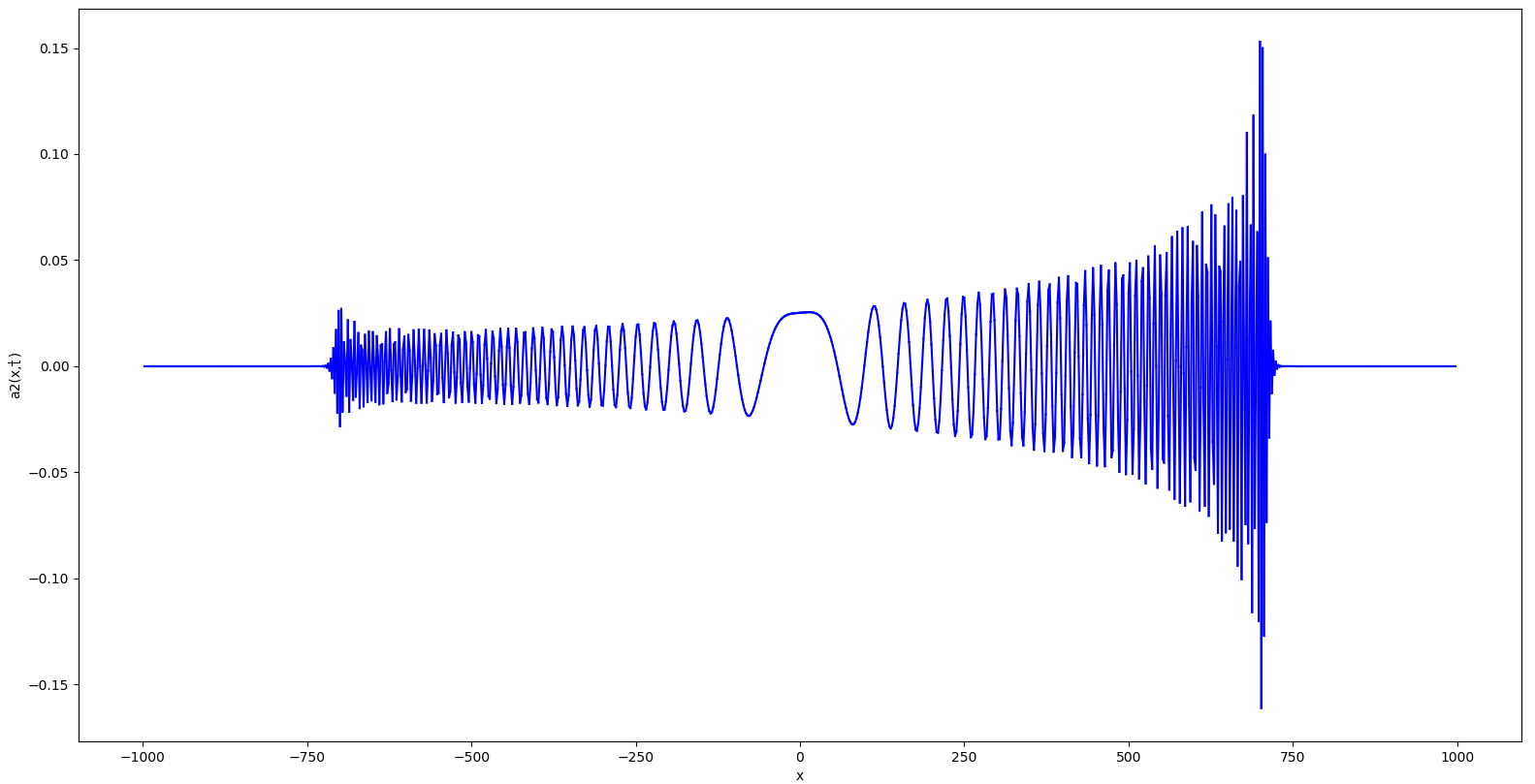}
\end{tabular}
\caption{
Графики функций $P(x,1000),a_1(x,1000),a_2(x,1000)$ при четном $x$}
\label{a-1000}
\end{figure}

\comment
\mscomm{Which of the numerous tables and figures should remain, and which should be removed?}

In the following table, the number in the cell $(x,t)$ is $P (x,t)$, and an empty cell $(x,t)$ means that $P(x,t)=0$.
\smallskip

\begin{tabular}{|c*{8}{|>{\centering\arraybackslash}p{55pt}}|}
	\hline
	$4$&$1/8$ &&$1/8$ &&$5/8$ &&$1/8$ \\
	\hline
	$3$&&$1/4$&&$ 1/2$&&$1/4$&\\
	\hline
	$2$&&& $1/2$&&$1/2$&&\\
	\hline
	$1$&&&&$1$&&&\\
	\hline
	\diagbox[dir=SW,height=18pt]{t}{x}&$-2$&$-1$&$0$&$1$&$2$&$3$&$4$\\
	\hline
\end{tabular}
\vskip8pt

 {Rotating the checkerboard (see Figure~\ref{Checker-paths} to the right) through $45^\circ$ clockwise, writing $a(x,t)$ in each black square, and dropping the normalization factor, we get a compact way to show the vectors (Table~\ref{table}).}


{\small
$$
\begin{array}{cccccccc}
 (0,1) & (0,1) & (0,1) & (0,1) & (0,1) & (0,1) & (0,1) & (0,1) \\
 (1,0) & (1,-1) & (1,-2) & (1,-3) & (1,-4) & (1,-5) & (1,-6) & (1,-7) \\
 (1,0) & (0,-1) & (-1,-1) & (-2,0) & (-3,2) & (-4,5) & (-5,9) & (-6,14) \\
 (1,0) & (-1,-1) & (-2,0) & (-2,2) & (-1,4) & (1,5) & (4,4) & (8,0) \\
 (1,0) & (-2,-1) & (-2,1) & (0,3) & (3,3) & (6,0) & (8,-6) & (8,-14) \\
 (1,0) & (-3,-1) & (-1,2) & (3,3) & (6,0) & (6,-6) & (2,-12) & (-6,-14) \\
 (1,0) & (-4,-1) & (1,3) & (6,2) & (6,-4) & (0,-10) & (-10,-10) & (-20,0) \\
 (1,0) & (-5,-1) & (4,4) & (8,0) & (2,-8) & (-10,-10) & (-20,0) & (-20,20) \\
\end{array}
$$
}

\begin{table}[htbp]
{\small
\begin{tabular}{|c|c*{8}{>{\centering\arraybackslash}p{45pt}}c|}
    \hline
 $ 8$ & $ 1 $ & $ -6-i $ & $ 8+5 i $ & $ 8-3 i $ & $ -6-11 i $ & $ -20-5 i $ & $ -20+15 i $ & $ 35 i $ & $ 35+35 i $ \\
 $ 7$ & $ 1 $ & $ -5-i $ & $ 4+4 i $ & $ 8 $ & $ 2-8 i $ & $ -10-10 i $ & $ -20 $ & $ -20+20 i $ & $ -5+40 i $ \\
 $ 6$ & $ 1 $ & $ -4-i $ & $ 1+3 i $ & $ 6+2 i $ & $ 6-4 i $ & $ -10 i $ & $ -10-10 i $ & $ -20 $ & $ -25+20 i $ \\
 $ 5$ & $ 1 $ & $ -3-i $ & $ -1+2 i $ & $ 3+3 i $ & $ 6 $ & $ 6-6 i $ & $ 2-12 i $ & $ -6-14 i $ & $ -17-8 i $ \\
 $ 4$ & $ 1 $ & $ -2-i $ & $ -2+i $ & $ 3 i $ & $ 3+3 i $ & $ 6 $ & $ 8-6 i $ & $ 8-14 i $ & $ 5-22 i $ \\
 $ 3$ & $ 1 $ & $ -1-i $ & $ -2 $ & $ -2+2 i $ & $ -1+4 i $ & $ 1+5 i $ & $ 4+4 i $ & $ 8 $ & $ 13-8 i $ \\
 $ 2$ & $ 1 $ & $ -i $ & $ -1-i $ & $ -2 $ & $ -3+2 i $ & $ -4+5 i $ & $ -5+9 i $ & $ -6+14 i $ & $ -7+20 i $ \\
 $ 1$ & $ 1 $ & $ 1-i $ & $ 1-2 i $ & $ 1-3 i $ & $ 1-4 i $ & $ 1-5 i $ & $ 1-6 i $ & $ 1-7 i $ & $ 1-8 i $ \\
 $ 0$ & $ i $ & $ i $ & $ i $ & $ i $ & $ i $ & $ i $ & $ i $ & $ i $ & $ i $ \\
 	\hline	\diagbox[dir=SW,height=18pt]{$n$}{$k$}&
 $1$&$2$&$3$&$4$&$5$&$6$&$7$&$8$&$9$\\
	\hline
\end{tabular}
}
\caption{The values $2^{(n+k-1)/2}{a}(k-n,n+k)$ for $0\le n\le 8$, $1\le k\le 9$.}
\label{table}
\end{table}

\endcomment

\subsection{Физическая интерпретация} \label{ssec-interpretation}

Обсудим физическую интерпретацию модели и покажем, что ей действительно удаётся передать невероятное поведение электронов. Есть две разные интерпретации (см.~табл.~\ref{table-interpretations}).

\begin{table}[htbp]
\begin{center}
\begin{tabular}{|l|p{0.37\textwidth}|p{0.45\textwidth}|}
  \hline
  объект & стандартная интерпретация & интерпретация через спиновую цепочку \\
  \hline
  $s$ & путь & конфигурация ``$+$'' и ``$-$'' в один ряд \\
  $\mathrm{turns}(s)$ & число поворотов & половина энергии конфигурации \\
  $t$ & время & объём \\
  $x$ & положение & разность между количеством ``$+$'' и ``$-$''  \\
  $x/t$ & средняя скорость &  намагниченность \\
  $a(x,t)$ & амплитуда вероятности & статистическая сумма с точностью до константы \\
  $P(x,t)$ & вероятность & квадрат нормы статистической суммы\\
  $\dfrac{4i}{\pi t}\log a(x,t)$ &
  \vspace{-0.6cm}
  нормализованный логарифм амплитуды &
  \vspace{-0.6cm}
  плотность свободной энергии\\
  $\dfrac{i\,a_2(x,t)}{a(x,t)}$ &
  \vspace{-0.6cm}
  условная амплитуда вероятности последнего хода вправо-вверх &
  \vspace{-0.6cm}
 ``вероятность'' совпадения знаков на концах спиновой цепочки
  \\  \hline
\end{tabular}
\end{center}
  \vspace{-0.3cm}
  \caption{Физическая интерпретация шашек Фейнмана} \label{table-interpretations}
\end{table}

{\bf Стандартная интерпретация.
}
Координаты  $x$ и $t$ интерпретируются как положение электрона и время соответственно. Иногда (например, в примере~\ref{p-double-slit}) немного неформально обе координаты интерпретируются как положение, и мы {считаем}, что электрон выполняет ``классическое равномерное движение'' по оси $t$. Мы используем естественную систему единиц, в которой скорость света, постоянные Планка и Больцмана равны $1$. Таким образом, прямые $x = \pm t$ изображают движение со скоростью света. Любой путь шашки лежит не ниже обеих прямых, т.е. в световом конусе. Это согласуется с теорией относительности: скорость электрона не может превышать скорость света.

Чтобы думать о $P(x,t)$ как о вероятности, нужно рассматривать координату $t$ как фиксированную, а клетки $(-t,t), (-t+2,t), \dots, (t,t)$ --- как всевозможные исходы эксперимента. Например, $t$-ая горизонталь может быть экраном, детектирующим электрон. Впоследствии мы увидим, что сумма всех чисел $P(x,t)$ на одной горизонтали в самом деле равна $1$ (предложение~\ref{p-probability-conservation}), и поэтому мы можем рассматривать их как ``вероятности''. Отметим, что вероятность обнаружить электрон в множестве $X\subset\mathbb{Z}$ --- это
$P(X,t):=\sum_{x\in X}P(x,t)=\sum_{x\in X}|a(x,t)|^2$, а не $\left|\sum_{x\in X}a(x,t)\right|^2$ (ср.~\cite[рис.~50]{Feynman}).

В реальности точно измерить положение электрона невозможно. Фундаментальным ограничением является \emph{приведённая комптоновская длина волны}
$\lambda=1/m\approx 4\cdot 10^{-13}$~м, где $m$ --- масса электрона. На физическом языке, базовая модель аппроксимирует континуум решёткой с шагом ровно $\lambda$. Но это всё ещё грубое приближение: нужен ещё меньший шаг, чтобы избежать накопления ошибки аппроксимации при больших $x$ и $t$.
Например, рис.~\ref{P-contour} и~\ref{fig-correlation} слева изображают эффект решётки с конечным шагом --- \emph{перенормировку скорости света}: средняя скорость $x/t$ не может превосходить $1/\sqrt{2}$ от скорости света с большой вероятностью. (Объяснение в физических терминах: решёточная регуляризация обрезает расстояния меньше шага решетки, следовательно, малые длины волн, следовательно, большие импульсы и, следовательно, большие скорости.)
Более точная
модель представлена в~\S\ref{sec-mass}: ср. графики на рис.~\ref{P-contour}.


\begin{figure}[htbp]
  \centering
  \includegraphics[width=0.24\textwidth]
  {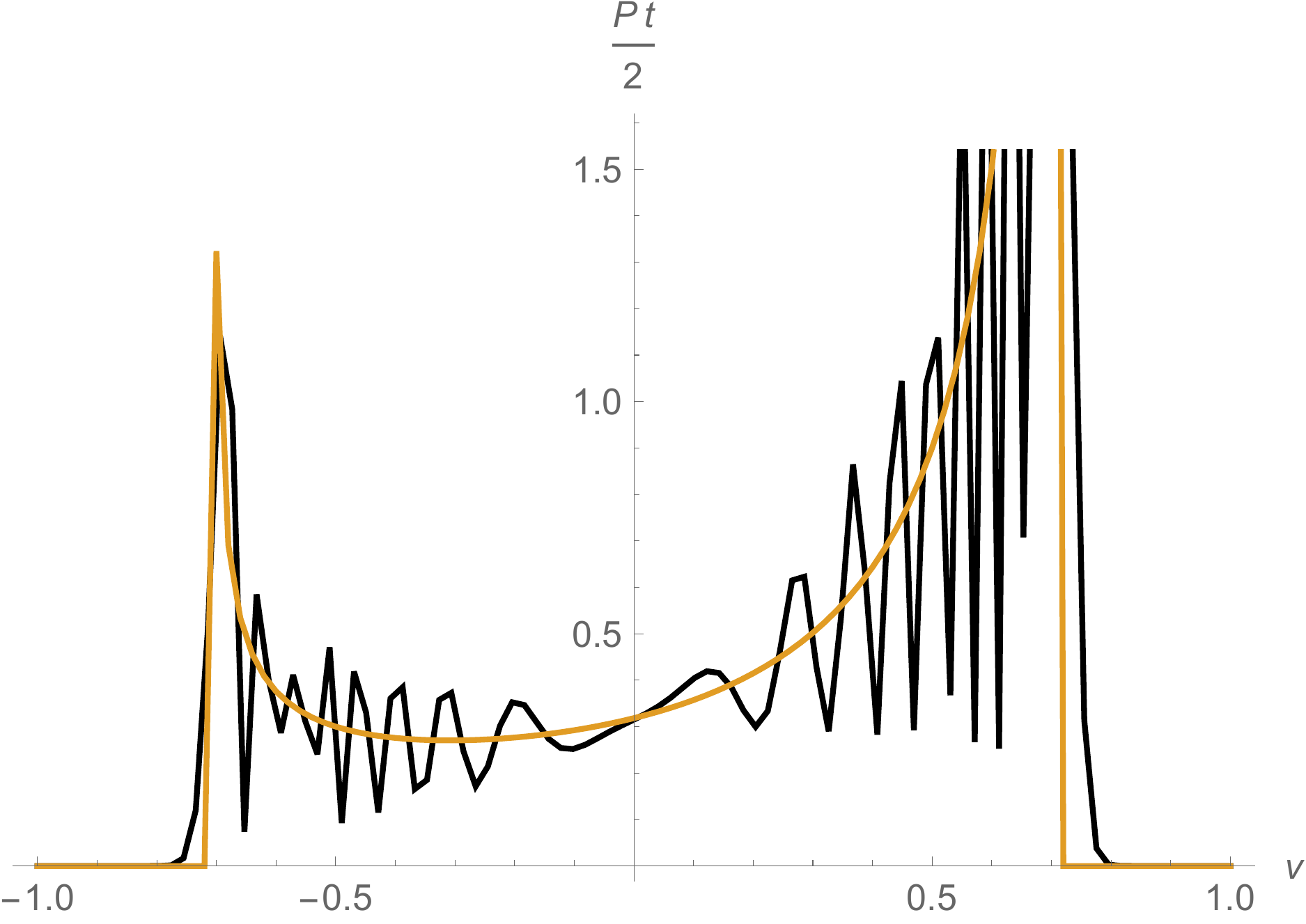}
  \includegraphics[width=0.24\textwidth]
  {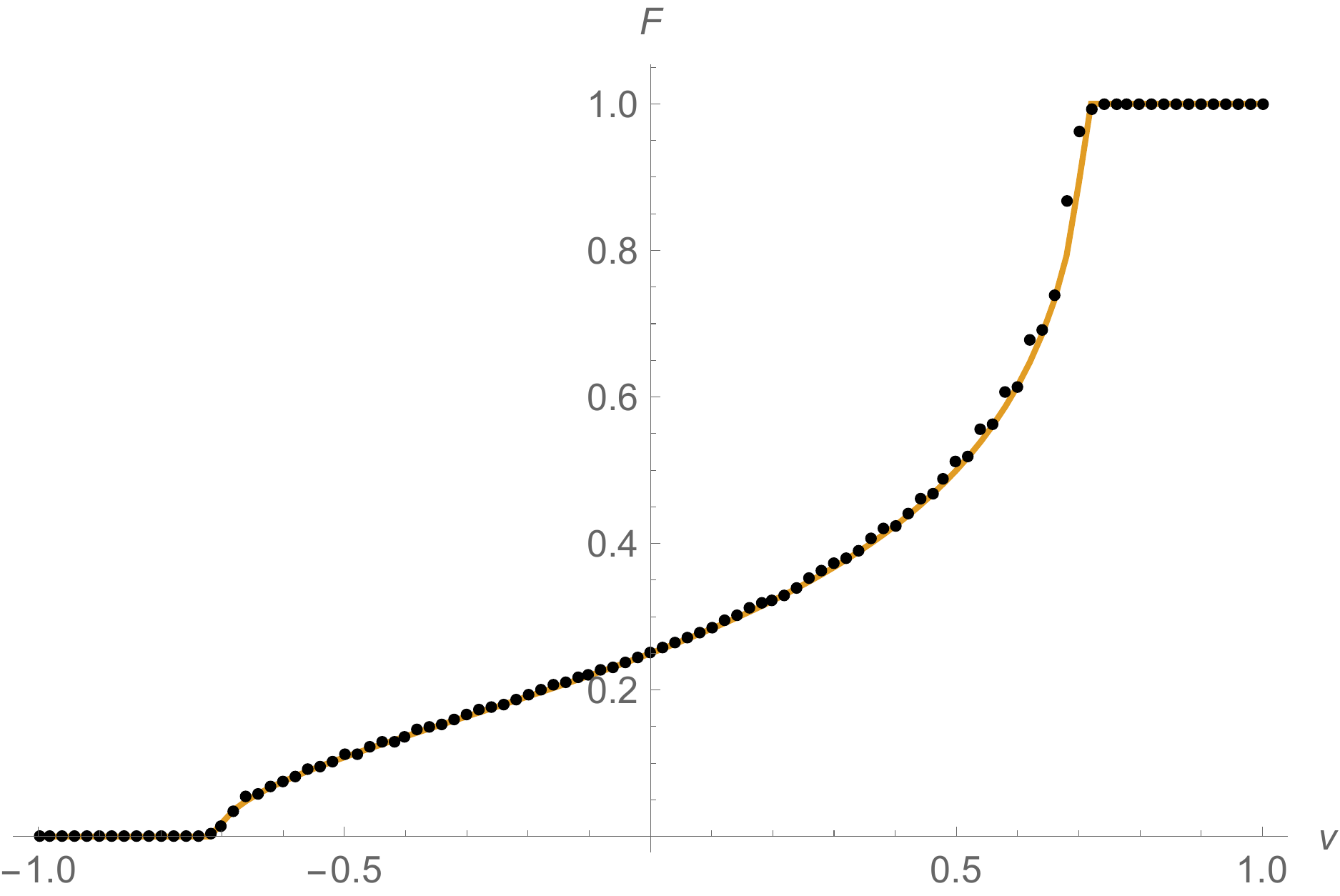}
  \includegraphics[width=0.24\textwidth]
  {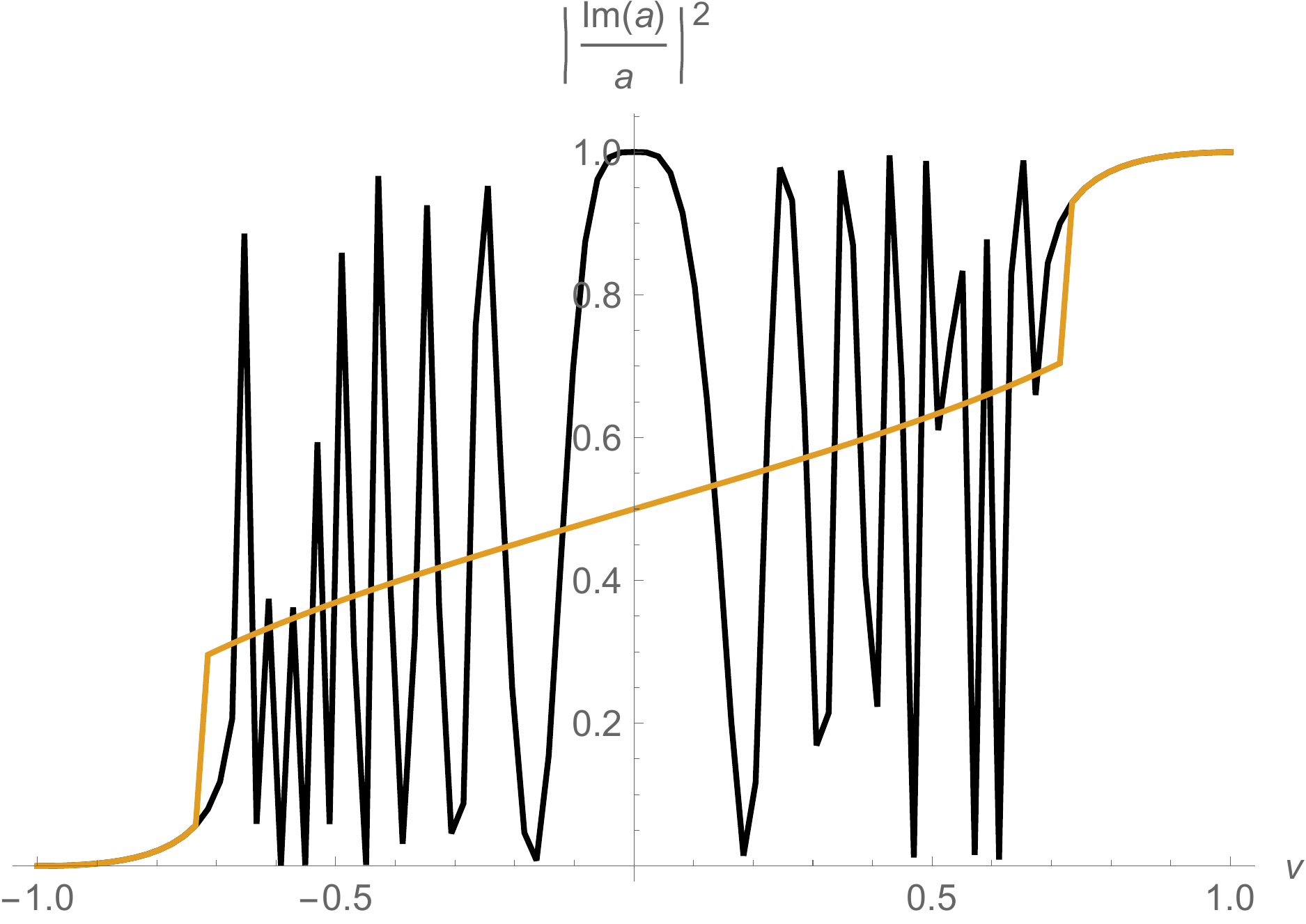}
  \includegraphics[width=0.24\textwidth]
  {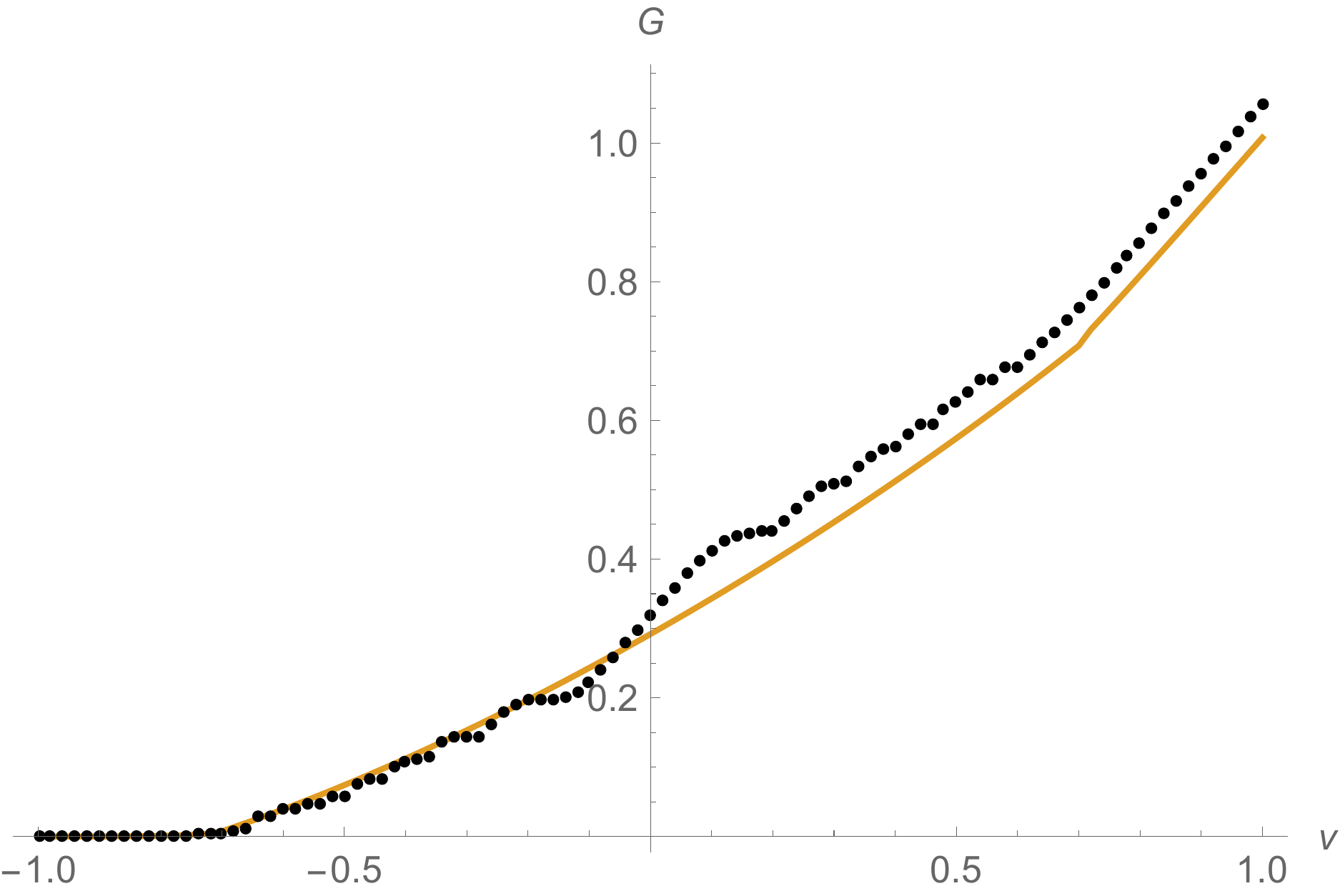}\\
  \includegraphics[width=0.24\textwidth]
  {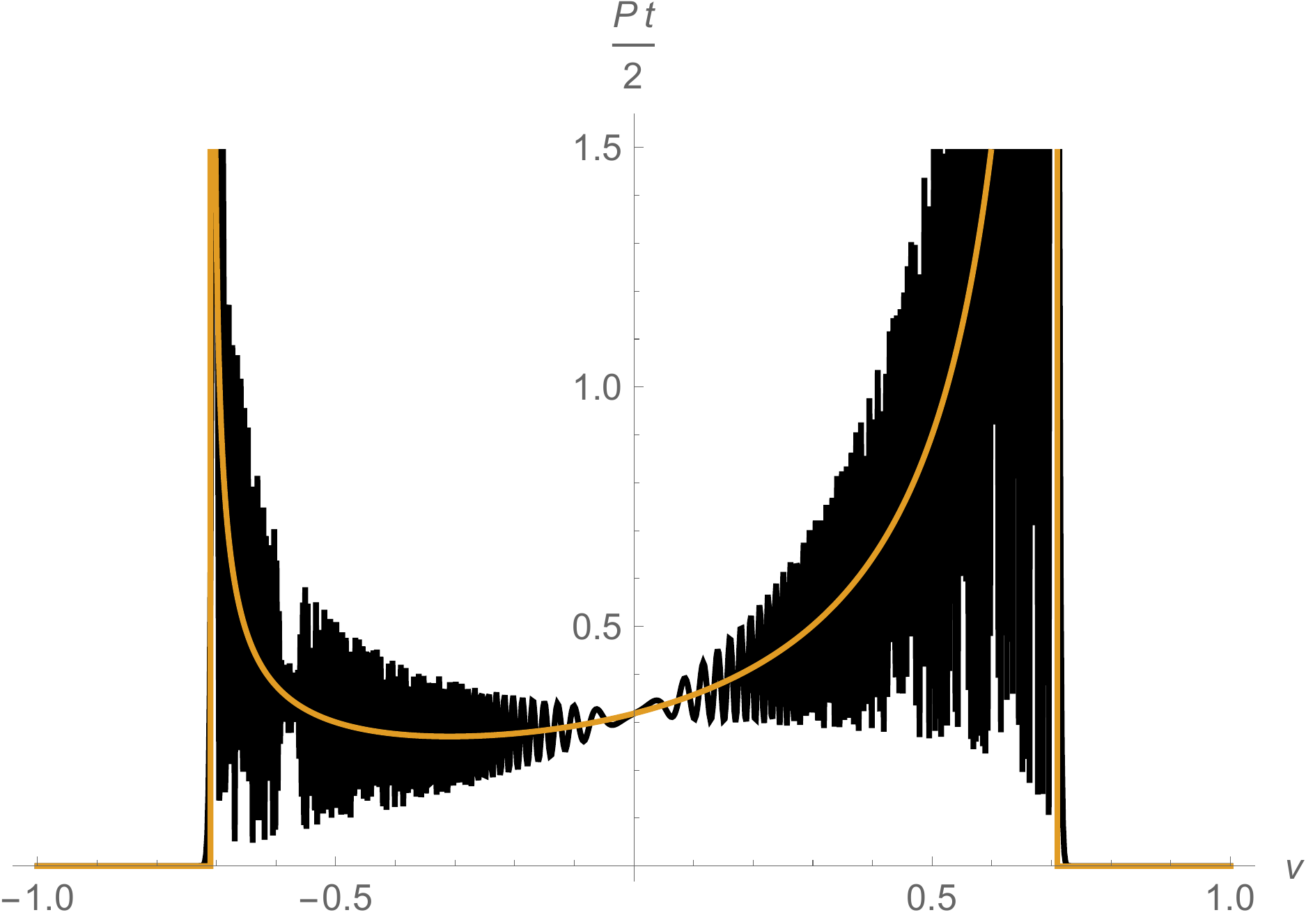}
  \includegraphics[width=0.24\textwidth]
  {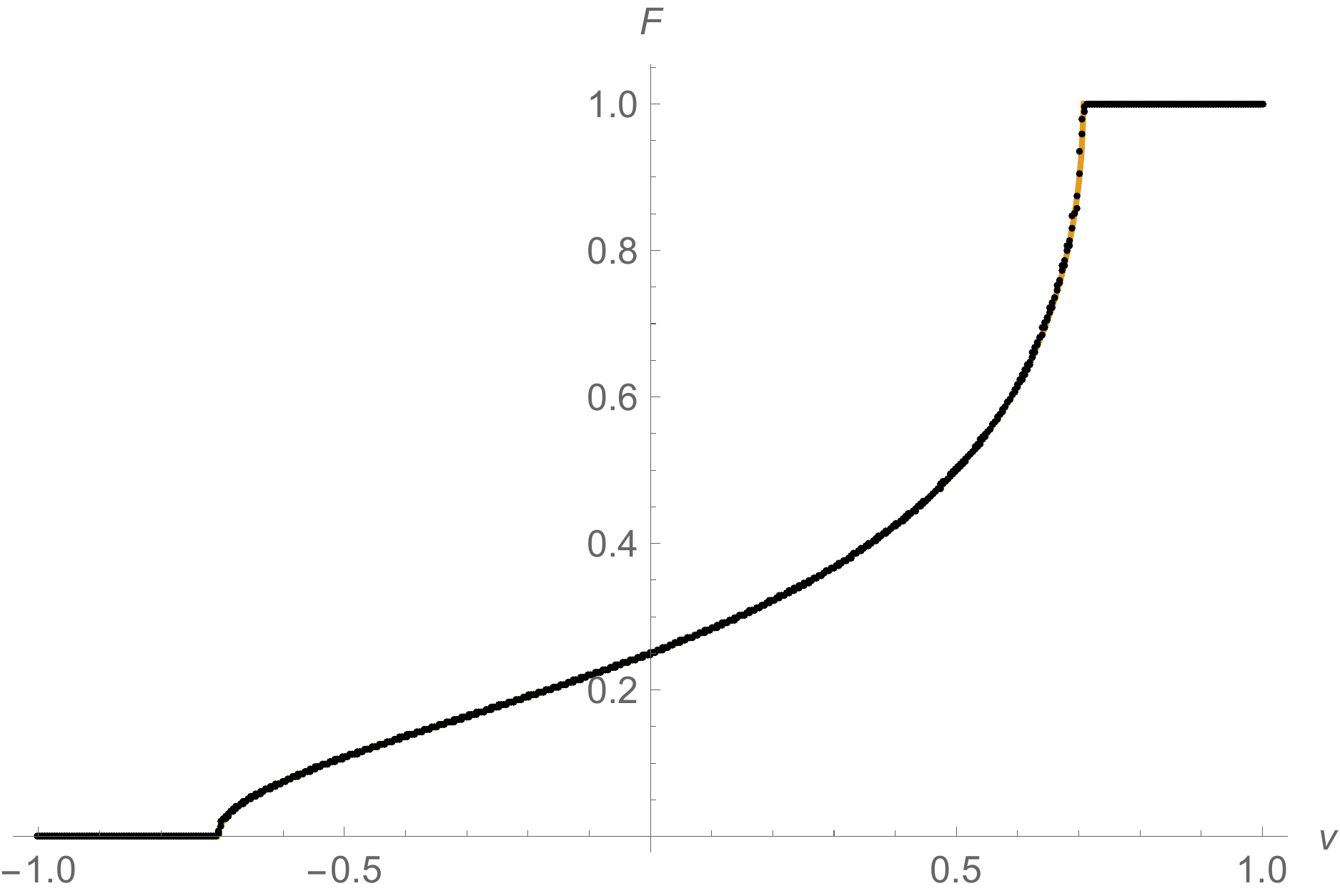}
  \includegraphics[width=0.24\textwidth]
  {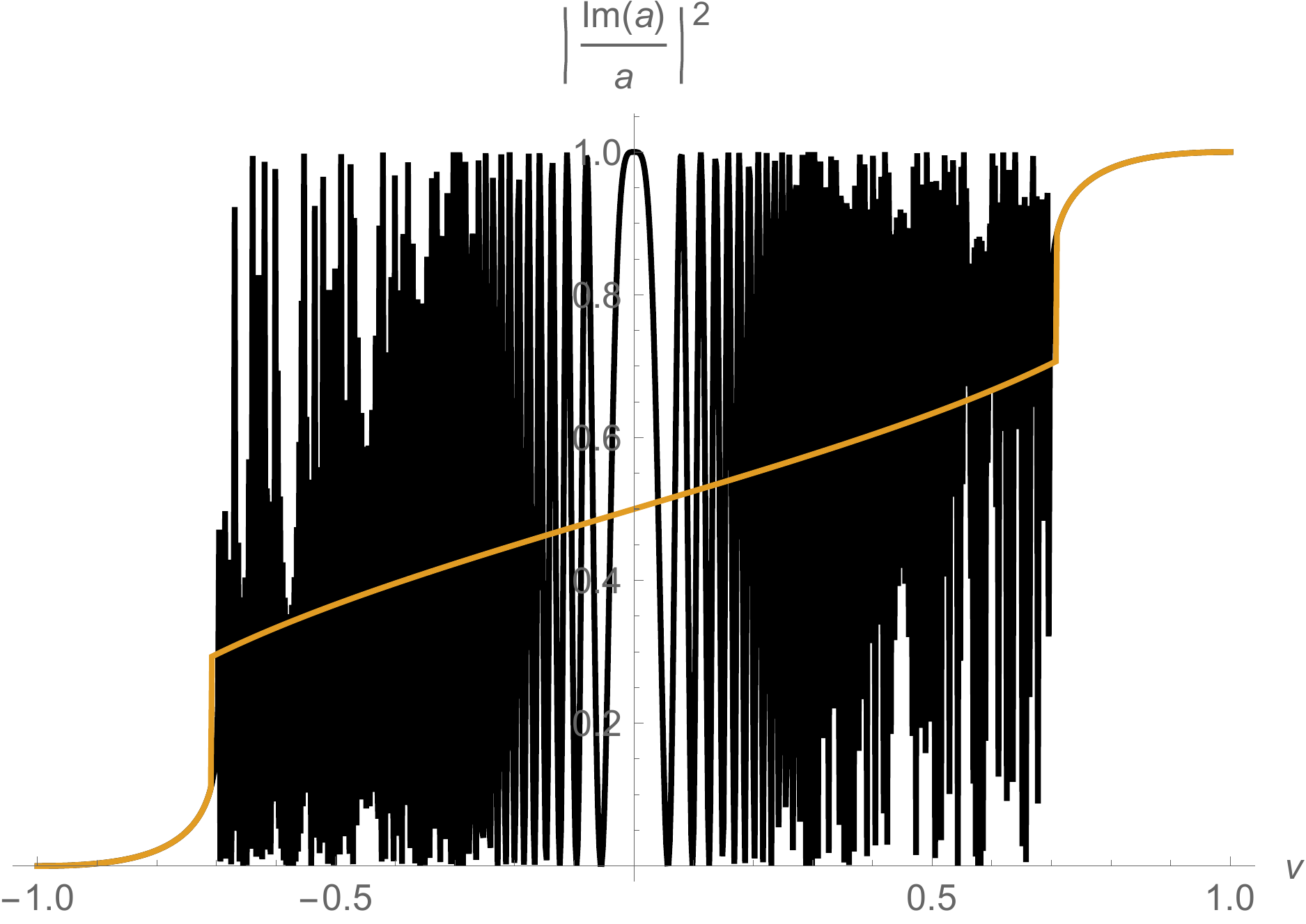}
  \includegraphics[width=0.24\textwidth]
  {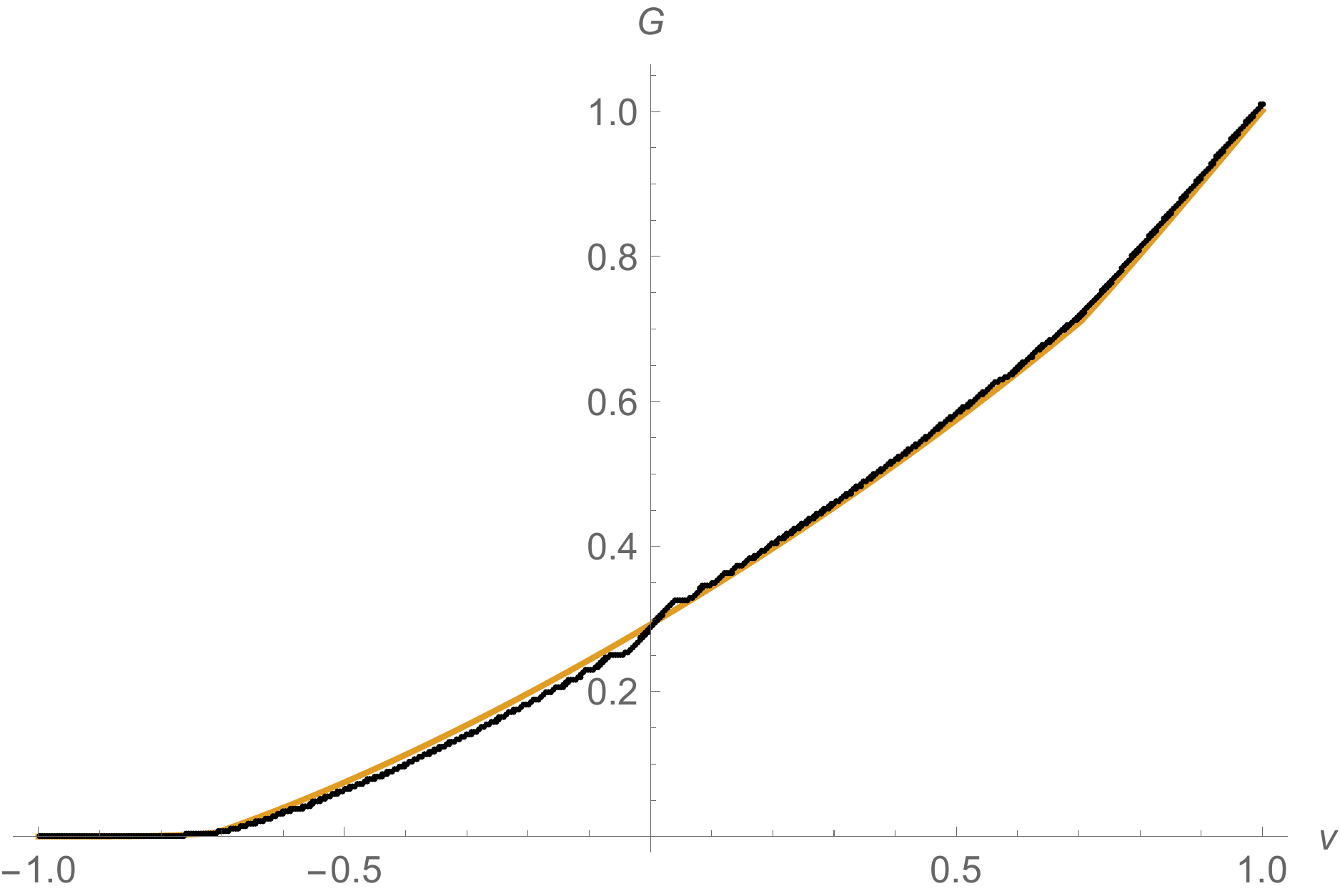}
  \caption{
  Графики $\frac{t}{2}P(2\lceil \frac{vt}{2}\rceil,t)$, $F_t(v):=\sum\limits_{x\le vt}P(x,t)$, $\left|\frac{a_2(2\lceil {vt}/{2}\rceil,t)}{a(2\lceil {vt}/{2}\rceil,t)}\right|^2$,
  $G_t(v):=\sum\limits_{x\le vt} \frac{2}{t}\left|\frac{a_2(x,t)}{a(x,t)}\right|^2$ (тёмные) 
  при $t=100$ (сверху), $t=1000$ (снизу), и их слабые пределы при $t\to\infty$ (светлые). 
  }
  \label{fig-correlation}
\end{figure}

Теперь мы собираемся показать, что этой модели действительно удаётся передать невероятное поведение эдектронов. (Для получения правильных количественных результатов, таких как точная форма интерферограммы, необходимо усовершенствование модели, использующее когерентный источник; см.~\S\ref{sec-source}.)

\emph{Вероятность обнаружения электрона в клетке $(x,t)$ при поглощении в подмножестве $B\subset\mathbb{Z}^2$}
определяется аналогично $P(x,t)$, только сумма производится по путям шашки $s$, не имеющим общих точек с $B$, кроме, быть может, $(0,0)$ и $(x,t)$. Обозначим её $P(x,t \text{ минуя } B)$. Неформально это значит, что появляется дополнительный исход эксперимента, при котором электрон поглощается и не достигает экрана.
В следующем примере мы рассматриваем две чёрных клетки $(\pm1,1)$ на горизонтали $t=1$ как две щели в горизонтальной пластине (ср.~рис.~\ref{Double-slit}).

\begin{example}[Двухщелевой опыт] \label{p-double-slit} Различные пути нельзя рассмотривать как ``взаимоисключающие'':
$$P(0,4) \ne P(0,4 \text{ минуя } \{(2,2)\})+P(0,4 \text{ минуя } \{(0,2)\}).$$
Иногда поглощение может увеличить вероятность:
$P(0,4) = 1/8 < 1/4 = P(0,4 \text{ минуя } \{(2,2)\})$.
\end{example}

Стандартная интерпретация шашек Фейнмана также известна под названиями \emph{блуждание Адамара}, или, более общо, как \emph{одномерное квантовое блуждание} и \emph{квантовый решеточный газ}. Все эти понятия эквивалентны, но приводят к обобщениям модели в разных направлениях \cite{Venegas-Andraca-12, Konno-08, Yepez-05}. Например, объединение усовершенствований из \S\ref{sec-mass}--\ref{sec-medium} дает общее неоднордное квантовое блуждание.

Удивительные свойства квантовых блужданий, упомянутые в~\S\ref{ssec-background}, точно формулируются так:
$$
\sum_{t=1}^{\infty}P(0,t \text{ минуя } \{x=0\})
= \frac{2}{\pi}
\le \sum_{t=1}^{\infty}P(0,t \text{ минуя } \{x=0\}\cup \{x=n\})
\to\frac{1}{\sqrt{2}}\text{ при }n\to+\infty.
$$
Недавно М.~Дмитриев \cite{Dmitriev-22} нашел $\sum_{t=1}^{\infty}P(n,t \text{ минуя } \{x=n\})$ для нескольких $n\ne 0$ (см.~проблему~\ref{p-pi}). Похожие числа возникают для простого случайного блуждания на $\mathbb{Z}^2$ \cite[таблица~2]{Pakharev-Skopenkov-Ustinov}.

\smallskip
 {\textbf{Интерпретация через спиновую цепочку.} Есть
 \emph{совершенно другая} физическая интерпретация той же самой модели: одномерная модель Изинга с мнимой температурой и фиксированной намагниченностью.}

 {Напомним, что \emph{конфигурация} в модели Изинга --- это последовательность $\sigma=(\sigma_1,\dots,\sigma_t)$  фиксированной длины, состоящая из чисел $\pm 1$. \emph{Намагниченностью} и \emph{энергией} конфигурации называются $\sum_{k=1}^{t}\sigma_k/t$ и $H(\sigma)=\sum_{k=1}^{t-1}(1-\sigma_k\sigma_{k+1})$ соответственно. \emph{Вероятностью} конфигурации считается $e^{-\beta H(\sigma)}/Z(\beta)$, где \emph{обратная температура} $\beta=1/T>0$ --- некоторый параметр, а \emph{статистическая сумма} $Z(\beta):=\sum_\sigma e^{-\beta H(\sigma)}$ --- нормирующий множитель. На конфигурацию $\sigma$ обычно накладываются дополнительные ограничения.}

 {Теперь, двигая шашку вдоль пути $s$, будем писать знак ``$+$'' для каждого хода вправо-вверх и ``$-$'' для каждого хода влево-вверх; см.~рис.~\ref{fig-young} слева. Полученная последовательность знаков --- некоторая конфигурация в модели Изинга. Количество поворотов в $s$ равно половине энергии конфигурации, а отношение конечных $x$- и $t$-координат равно намагниченности. 
Поэтому $a(x,t)=\sum_s a(s)$ совпадает (с точностью до не зависящего от $x$ множителя) со статистической суммой модели Изинга при \emph{мнимой} обратной температуре $\beta=i\pi/4$ и фиксированной намагниченности $x/t$:}
\begin{equation*}
a(x,t)=2^{(1-t)/2}\,i\,
\sum_{\substack{(\sigma_1,\dots,\sigma_t)\in\{+1,-1\}^t:\\
\sigma_1=+1,\quad\sum_{k=1}^{t-1}\sigma_k=x}}
\exp\left(\frac{i\pi}{4}\sum_{k=1}^{t-1}(\sigma_k\sigma_{k+1}-1)\right).
\end{equation*}

\begin{figure}[htbp]
  \centering
  \includegraphics[width=0.2\textwidth]{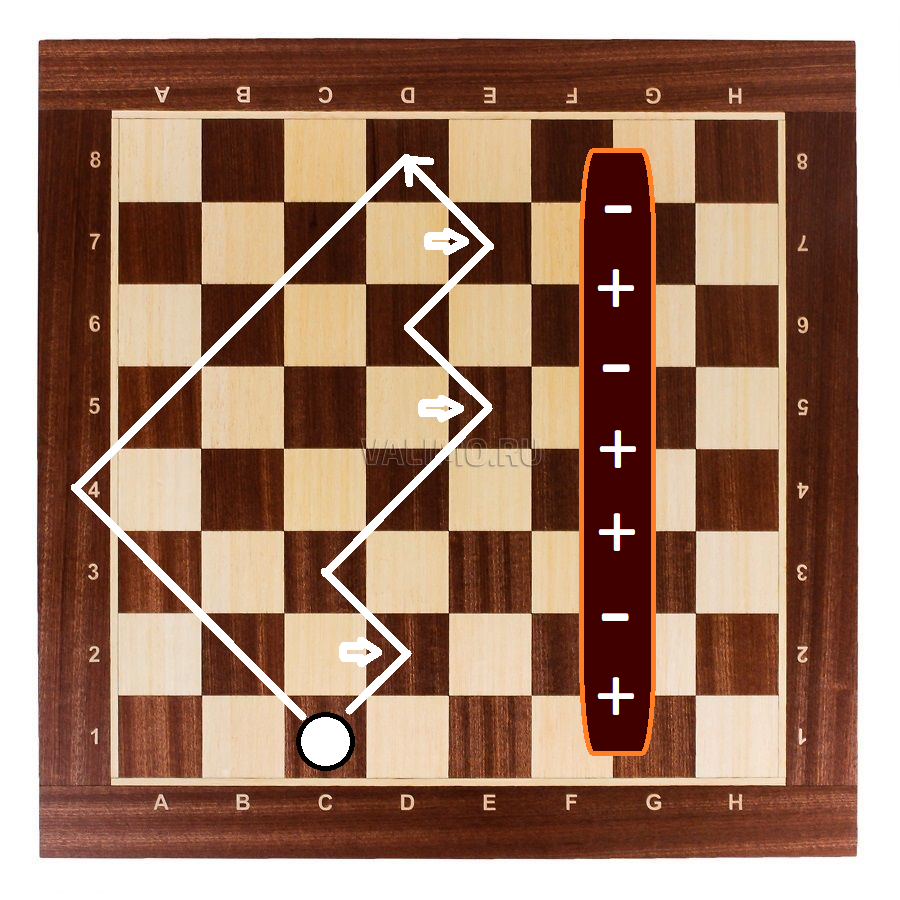}
  \includegraphics[width=0.2\textwidth]{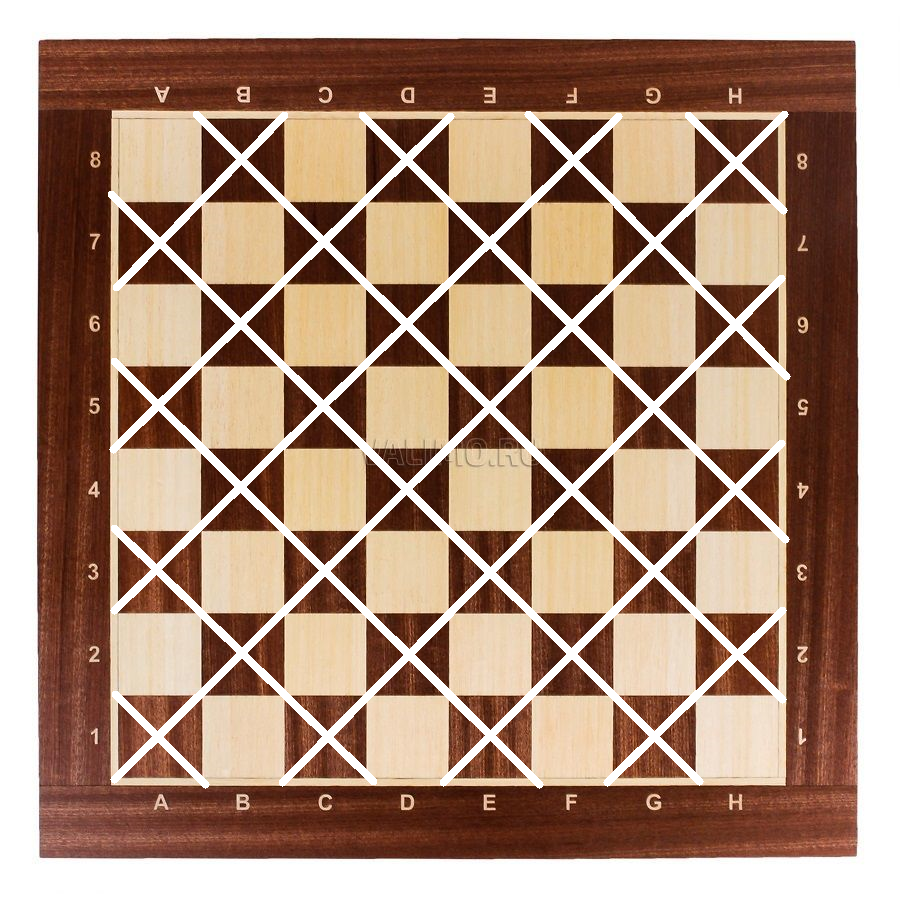}
  \caption{ {Диаграмма Юнга (стрелки указывают на ступени) и модель Изинга \new{(слева)}. Вспомогательная решетка \new{(справа)}.}}
  \label{fig-young}
\end{figure}

 {Отметим принципиальное отличие полученной интерпретации через спиновую цепочку как от обычной модели Изинга, так и от стандартной интерпретации выше. В последних двух моделях
 намагниченность $x/t$ и средняя скорость~$x/t$ были \emph{случайными величинами}; теперь же намагниченность~$x/t$ (которую не следует путать с внешним магнитным полем) является
  \emph{внешним ограничением}. Конфигурационное пространство для интерпретации через спиновую цепочку состоит из последовательностей  ``$+$'' и ``$-$'' с \emph{фиксированным} количеством ``$+$'' и~``$-$''. Суммирование по конфигурациям с разными~$x$ или~$t$ не имело бы никакого смысла: например, напомним, что $P(X,t)=\sum_{x\in X}|a(x,t)|^2$, а не $\left|\sum_{x\in X}a(x,t)\right|^2$.}

{Изменяя намагниченность $x/t$, рассматриваемую как внешнее ограничение, мы наблюдаем \emph{фазовый переход}:
предельная \emph{плотность свободной энергии} (точнее, ее мнимая часть) $-\log a(x,t)/\beta t$ неаналитична, когда $x/t$ проходит через значение $\pm 1/\sqrt{2}$ (см.~рис.~\ref{fig-distribution} в центре и следствие~\ref{cor-free}). Фазовый переход проявляется при $t\to\infty$.
 \label{equal-signs}
 Интересно было бы изучить другие параметры порядка, например,
 \emph{``вероятность'' $i\,a_2(x,t)/a(x,t)$ одинаковых знаков на концах спиновой цепочки} (см.~рис.~\ref{fig-correlation} и проблемы~\ref{p-correlation}--\ref{p-correlation2}). Эти величины являются комплексными числами, просто потому что температура мнимая.
}

 {(Комментарий для специалистов: фазовый переход не связан с накоплением нулей статистической суммы на плоскости комплексного параметра $\beta$, как в \cite[\S III]{Jones-66} и~\cite{Matveev-Shrock-97}. В нашей ситуации $\beta=i\pi/4$
 фиксировано, \emph{вещественный} параметр $x/t$ изменяется, а  статистическая сумма $a(x,t)$ совсем не имеет нулей~\cite[теорема~1]{Novikov-20}.)
}


\smallskip
 {\textbf{Интерпретация через диаграммы Юнга.} Наши результаты также имеют комбинаторную интерпретацию.

 {
\emph{Число ступеней} (или \emph{внутренних углов})
диаграммы Юнга с $w$ столбцами высот $x_1,\dots,x_w$ соответственно --- это количество элементов в множестве $\{x_1,\dots,x_w\}$; см.~рис.~\ref{fig-young} \new{слева}. \label{page-young} Тогда величина $2^{(h+w-1)/2}\,a_1(h-w,h+w)$
равна разности между количеством диаграмм Юнга с нечетным и с четным числом ступеней, имеющих ровно $w$ столбцов и $h$ строк.}
}

 {
Интересное поведение начинается уже для $h=w$ (см.~предложение~\ref{cor-coefficients}). При четном $h=w$ наша разность обращается в ноль. При нечетном $h=w=2n+1$ она равна $(-1)^n\binom{2n}{n}$. Такая периодичность с периодом $4$ примерно сохраняется, пока $h$ близко к $w$ \cite[теорема~2]{SU-20}.
При фиксированном полупериметре $h+w$ наша разность медленно колеблется с ростом $h/w$, достигает пика при $h/w\approx 3+2\sqrt{2}$, а затем резко падает до очень малых значений (см.~следствие~\ref{cor-young} и~теоремы~\ref{th-ergenium}--\ref{th-outside}).
}

 {
Аналогично, величина $2^{(h+w-1)/2}a_2(h-w,h+w)$ равна  разности между количеством диаграмм Юнга с нечетным и с четным числом ступеней, имеющих ровно $w$ столбцов и \emph{меньше}, чем $h$ строк. Поведение аналогично. Уточнение в~\S\ref{sec-mass} связано с \emph{характерными многочленами Стенли}~\cite[\S2]{Stanley-04}.
}

\smallskip
\textbf{Обсуждение определения}. Сравним определение~\ref{def-basic} с определениями в литературе.

Обозначение ``$a$'' происходит от слов ``\emph{a}rrow'' и ``probability \emph{a}mplitude''. Другие названия --- ``волновая функция'', ``ядро'', ``функция Грина'', ``пропагатор''. Более традиционные обозначения --- ``$\psi$'',  ``$K$'', ``$G$'', ``$\Delta$'', ``$S$'', в зависимости от контекста. Мы будем использовать нейтральное обозначение, подходящее для всех контекстов.

Множитель $i$ и знак минуса в определении непринципиальны (и отсутствуют в первоисточнике \cite[задача~2.6]{Feynman-Gibbs}). Они появляются из-за обычного начального положения и направления вращения стрелки секундомера, и сокращают число минусов впоследствии.

Нормирующий множитель $2^{(1-t)/2}$ можно объяснить по аналогии с классическим случайным блужданием. Если бы шашка просто выполняла случайное блуждание, выбирая одно из двух возможных направлений на каждом шагу (после обязательного первого хода вправо-вверх), то число $|a(s)|^2=2^{1-t}$ было бы вероятностью пути~$s$. К данной аналогии следует относиться с осторожностью: в квантовой теории ``вероятность пути'' \emph{не} имеет абсолютно никакого смысла (вспомните пример~\ref{p-double-slit}). Причина в том, что путь --- это не то, что можно измерить:  измерение положения электрона в один момент времени сильно влияет на его движение для всех последующих моментов.

Более концептуально было бы также фиксировать направление \emph{последнего} хода пути $s$ (см. \cite[внизу стр.~35]{Feynman-Gibbs}). К счастью, такая фиксация не требуется в данной работе (и поэтому не делается), но становится критически важной в дальнейших уточнениях модели (см.~объяснение в~\S\ref{sec-spin}).

Можно задаться вопросом, откуда взялась эта модель. Следуя Фейнману, мы не будем пытаться объяснить или ``вывести'' её физически. Эту квантовую модель нельзя получить из классической с помощью фейнмановского суммирования по путям: просто не существует ясного классического аналога частицы со спином $1/2$ (ср.~\S\ref{sec-spin}) и настоящего понимания \emph{спина}. ``Вывод'' модели в \cite{Ambainis-etal-01, Bialynicki-Birula-94, Narlikar-72} требует намного более сложных понятий, чем сама модель. Истинной мотивировкой для этой модели является её простота, согласованность с базовыми принципами (такими как сохранение вероятности) и с экспериментом (что в данном случае означает правильный непрерывный предел, см. следствие~\ref{cor-uniform}).

\subsection{Тождества и асимптотики}

Сформулируем несколько известных свойств нашей модели. {Доказательства даны в~\S\ref{ssec-proofs-basic}.

Для начала отметим, что координаты стрелки $a_1(x,t)$ и $a_2(x,t)$ удовлетворяют следующему рекуррентному соотношению.

\begin{proposition}[Уравнение Дирака] \label{p-Dirac} Для каждого целого $x$ и для каждого положительного целого $t$ выполнено
\begin{align*}
a_1(x,t+1)&=\frac{1}{\sqrt{2}}a_2(x+1,t)+\frac{1}{\sqrt{2}}a_1(x+1,t);\\
a_2(x,t+1)&=\frac{1}{\sqrt{2}}a_2(x-1,t)-\frac{1}{\sqrt{2}}a_1(x-1,t).
\end{align*}
\end{proposition}

Эти уравнения имитируют $(1+1)$-мерное уравнение Дирака в базисе Вейля \cite[(19.4) и~(3.31)]{Peskin-Schroeder}
\begin{equation}\label{eq-continuum-Dirac}
\begin{pmatrix}
m  & \partial/\partial x-\partial/\partial t \\
\partial/\partial x+\partial/\partial t & m
\end{pmatrix}
\begin{pmatrix}
a_2(x,t) \\ a_1(x,t)
\end{pmatrix}=0,
\end{equation}
в котором произодные заменены на конечные разности, $m = 1$, а также добавлен нормирующий множитель $1/\sqrt{2}$.
Для уточнения модели в~\S\ref{sec-mass} эта аналогия становится прозрачной (см.~замечание~\ref{rem-equation-limit}).
Базис Вейля не единственен, поэтому существуют несколько форм уравнения~\eqref{eq-continuum-Dirac}; ср.~\cite[(1)]{Jacobson-Schulman-84}.

Уравнение Дирака влечёт сохранение вероятности.

\begin{proposition}[Сохранение вероятности/заряда] \label{p-probability-conservation}
Для каждого целого $t\ge 1$ выполнено $\sum\limits_{x\in\mathbb{Z}}\!P(x,t)=1$.
\end{proposition}

\comment

\begin{proposition}[Klein--Gordon equation]\label{p-Klein-Gordon} For each integer $x$ and each integer $t\ge 2$ we have
\begin{align*}
\sqrt{2}\, a(x,t+1)+\sqrt{2}\, a(x,t-1)-a(x-1,t)-a(x+1,t)=0.
\end{align*}
\end{proposition}

The analogy to the usual Klein--Gordon equation is going to become transparent in \S\ref{sec-mass}; see Proposition~\ref{p-Klein-Gordon-mass}.

\begin{proposition}[Symmetry]\label{p-symmetry}
For each integer $x$ and each positive integer $t$ we have $$a_1(x,t)=a_1(-x,t)\qquad\text{and}\qquad
a_2(x,t)+a_1(x,t)=a_2(2-x,t)+a_1(2-x,t).
$$
\end{proposition}

\begin{proposition}[Huygens' principle]\label{p-Huygens} For each integers $x$ and $t>t'>0$ we have
\begin{align*}
a_1(x,t)  &=\sum \limits_{x'\in\mathbb{Z}} \left[ a_2(x',t')a_1(x-x'+1,t-t'+1) + a_1(x',t')a_2(x'-x+1,t-t'+1) \right],\\
a_2(x,t)  &= \sum \limits_{x'\in\mathbb{Z}} \left[ a_2(x',t')a_2(x-x'+1,t-t'+1) - a_1(x',t')a_1(x'-x+1,t-t'+1) \right].
\end{align*}
\end{proposition}

Informally, Huygens' principle means that each black square $(x',t')$ on the $t'$-th horizontal acts like an independent point source, with the amplitude and phase determined by $a(x',t')$.



 {To write a formula for $a(x,t)$, set $\binom{n}{k}:=0$ for $k<0<n$ or $k>n>0$. Denote
$\theta(x):=\begin{cases}
                            1, & \mbox{if } x\ge0, \\
                            0, & \mbox{if } x<0.
                          \end{cases}$
}

\endcomment

Существует ``явная'' формула для $a_1(x,t)$ и $a_2(x,t)$
(альтернативные формулы приведены в приложении~\ref{app-formula}).

\begin{proposition}[``Явная'' формула] \label{Feynman-binom}
Для всех целых $|x|<t$, где $x+t$ чётно, выполнено
\begin{align*}
a_1(x,t)&=2^{(1-t)/2}\sum_{r=0}^{(t-|x|)/2}(-1)^r \binom{(x+t-2)/2}{r}\binom{(t-x-2)/2}{r},\\
a_2(x,t)&=2^{(1-t)/2}\sum_{r=1}^{(t-|x|)/2}(-1)^r \binom{(x+t-2)/2}{r}\binom{(t-x-2)/2}{r-1}.
\end{align*}
\end{proposition}





Следующее предложение является прямым следствием \new{<<явной>> формулы}.

\begin{proposition}[Частные значения]
\label{cor-coefficients}
Для всех целых $1\le k\le t-1$ значения $a_1(-t+2k,t)$ и $a_2(-t+2k,t)$ --- это коэффициенты перед $z^{t-k-1}$ и $z^{t-k}$ в разложении многочлена $2^{(1-t)/2}(1+z)^{t-k-1}(1-z)^{k-1}$. В частности,
\begin{align*}
a_1(0,4n+2)&=\frac{(-1)^n}{2^{(4n+1)/2}}\binom{2n}{n},
&
a_1(0,4n)&=0,\\
a_2(0,4n+2)&=0,
&
a_2(0,4n)&=\frac{(-1)^n}{2^{(4n-1)/2}}\binom{2n-1}{n}.
\end{align*}
\end{proposition}

В \S\ref{ssec-mass-properties} мы приведём больше тождеств.
Последовательности $2^{(t-1)/2}a_1(x,t)$ и $2^{(t-1)/2}a_2(x,t)$
есть в онлайн-энциклопедии целочисленных последовательностей~\cite[A098593 и A104967]{oeis}.

\comment

\begin{corollary}
\label{cor-alt-symmetry} 
For each $x,t\in\mathbb{Z}$, $t>0$, we have   $(t-x)\,a_2(x,t)=(t+x-2)\,a_2(2-x,t)$.
\end{corollary}

\begin{remark}\label{cor-hypergeometric}
By definition of the Gauss hypergeometric function,
for each $|x|<t$ such that $x+t$ is even,
\begin{align*}
a_1(x,t)&=2^{(1-t)/2}
\,{}_2F_1\left(1 - \frac{x+t}{2}, 1 + \frac{x-t}{2}, 1; -1\right),\\
a_2(x,t)&=2^{(1-t)/2 }\left(1 - \frac{x+t}{2}\right)
\,{}_2F_1\left(2 - \frac{x+t}{2}, 1 + \frac{x-t}{2}, 2; -1\right).
\end{align*}
\end{remark}


Propositions~\ref{p-Dirac} and~\ref{Feynman-binom} and their generalizations in \S\ref{ssec-mass-properties} are versions of well-known results. Surprisingly, Proposition~\ref{p-probability-conservation} seems to be new; cf.~the massless case in~\cite{Bialynicki-Birula-94}.
For the other formulae from~\S\ref{ssec-mass-properties} and Appendix~\ref{app-formula}, we have found no analogues in the literature.
Numerous other identities have been found by I.~Novikov in numeric experiments \cite[\S4]{Novikov-20}.

\endcomment




Следующий примечательный результат был замечен в \cite[\S4]{Ambainis-etal-01} (см. рисунки~\ref{fig-distribution} и~\ref{fig-correlation}), сформулирован и обоснован эвристически в~\cite[теорема~1]{Konno-05}, и доказан математически в~\cite[теорема~1]{Grimmett-Janson-Scudo-04}. 
Краткое изложение последнего доказательства дано в \S\ref{ssec-proofs-moments}, обобщения~---
в~\S\ref{ssec-mass-asymptotic}.



\begin{theorem}[Предельное распределение при больших временах; см. рисунок~\ref{fig-correlation}]
\label{th-limiting-distribution}
$\mathbf{(A)}$ Для любого $v\in\mathbb{R}$ верно, что
\begin{equation*}
\lim_{t\to\infty}\sum_{x\le vt}P(x,t)
= F(v):=
\begin{cases}
  0, & \mbox{при } v\le -1/\sqrt{2}; \\
  \frac{1}{\pi}\arccos\frac{1-2v}{\sqrt{2}(1-v)},
  & \mbox{при } |v|<1/\sqrt{2}; \\ 
  1, & \mbox{при }v\ge 1/\sqrt{2}.
\end{cases}
\end{equation*}
$\mathbf{(B)}$ Имеет место следующая сходимость по распределению при $t\to\infty$:
$$
tP(\lceil vt\rceil,t)\overset{d}\to F'(v)=
\begin{cases}
  \frac{1}{\pi (1-v)\sqrt{1-2v^2}}, & \mbox{при } |v|< 1/\sqrt{2}; \\
  0, & \mbox{при } |v|\ge 1/\sqrt{2}.
\end{cases}
$$
$\mathbf{(C)}$ Для любого целого $r\ge 0$ верно, что
$
\lim_{t\to\infty}\sum_{x\in\mathbb{Z}} \left(\frac{x}{t}\right)^r P(x,t)=
\int_{-1}^{1}
v^r F'(v)\,dv. 
$
\end{theorem}

Теорема~\ref{th-limiting-distribution}(B) демонстрирует фазовый переход в шашках Фейнмана, если интерпретировать их как модель Изинга при мнимой температуре и фиксированной намагниченности. Напомним, что тогда намагниченность $v$ является внешним ограничением (а не случайной величиной), а $P(\lceil vt\rceil,t)$ --- это квадрат нормы статистической суммы (а не вероятность). Предел по распределению для $tP(\lceil vt\rceil,t)$ разрывен в точке $v=\pm 1/\sqrt{2}$, что отражает фазовый переход (ср.~со~следствием~\ref{cor-free}). 

\comment

Первый полученный нами новый результат --- это аналитическое приближение ${a}(x,t)$ точное для малого $|x|/t$ (см. рисунок~\ref{fig-approximation}).
Это решает аналог проблемы Фейнмана 
для базовой модели (ср.~следствие~\ref{cor-feynman-problem}).

\begin{theorem}[Предел при больших временах вблизи начала координат] \label{Feynman-convergence}
Для любых целых $x,t$ таких, что $|x|<t^{3/4}$ и $x+t$ четно, верно, что 
\begin{align}\label{eq-Feynman-convergence}
{a}(x,t)&=i\sqrt{\frac{2}{\pi t}}
\exp \left(-\frac{i\pi t}{4}+\frac{i x^2}{2t}\right)
+i\frac{2x}{\sqrt{\pi t^3}}\,
\cos\left(-\frac{\pi (t+1)}{4}+\frac{x^2}{2t}\right)
+{O}\left(\frac{\log^{2}t}{t^{3/2}}+\frac{x^{4}}{t^{7/2}}\right),
\\ \label{eq-Feynman-convergence2}
P(x,t)&=\frac{2}{\pi t}\left[1+\frac{x}{t}\left(1+\sqrt{2}
\cos\left(\frac{\pi (2t+1)}{4}-\frac{x^2}{t}\right)\right)
+{O}\left(\frac{\log^{2}t}{t}+\frac{x^{4}}{t^3}\right)\right].
\end{align}
\end{theorem}

Напомним, что 
$f(x,t)={O}\left(g(x,t)\right)$ означает, что существует константа $C$ (не зависящая от $x,t$) такая, что для любых $x,t$, удовлетворяющих условиям теоремы, верно, что $|f(x,t)|\le C\,g(x,t)$.

\begin{figure}[htbp]
\begin{center}
\includegraphics[width=0.32\textwidth]{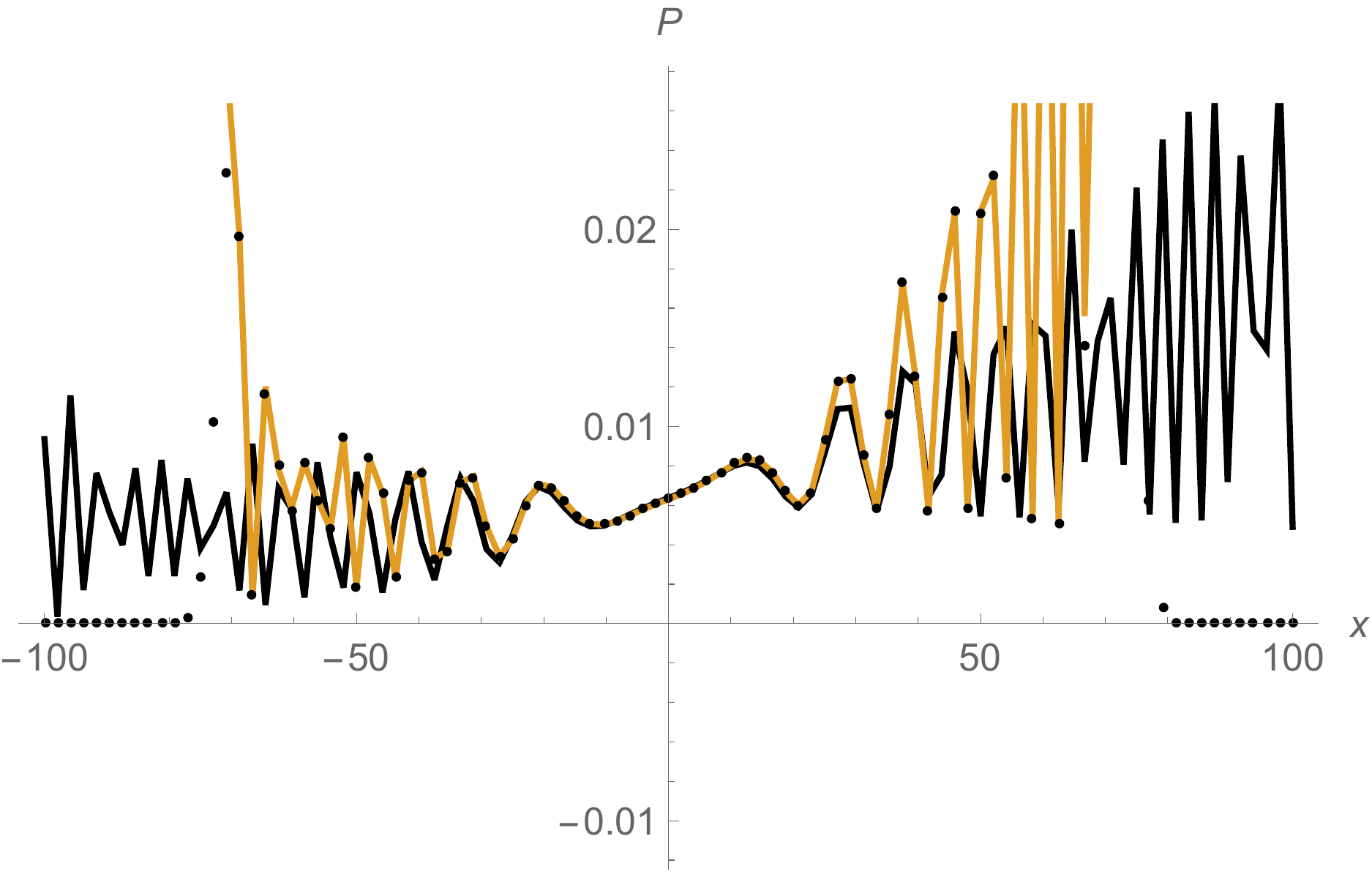}
\includegraphics[width=0.32\textwidth]{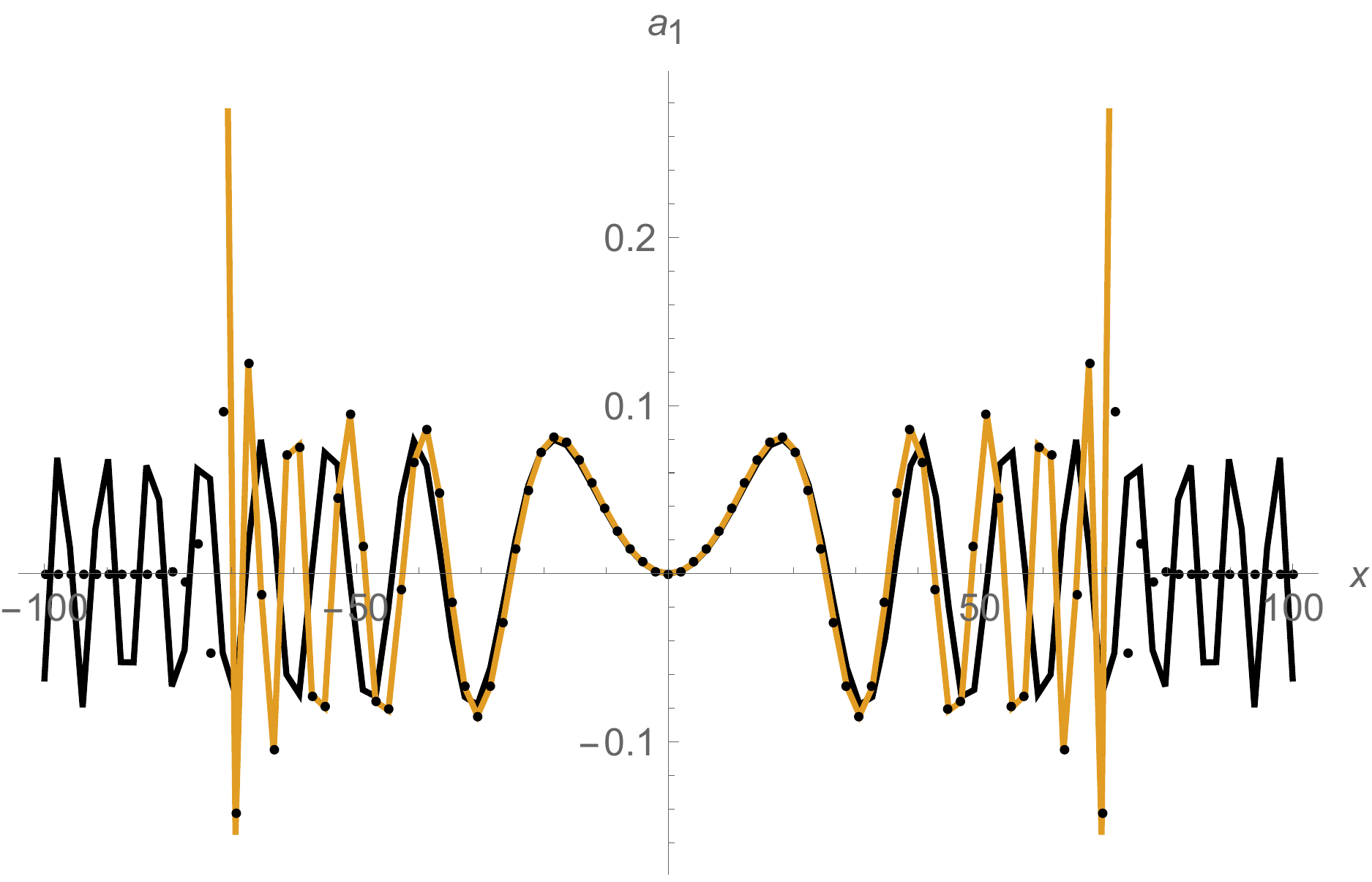}
\includegraphics[width=0.32\textwidth]{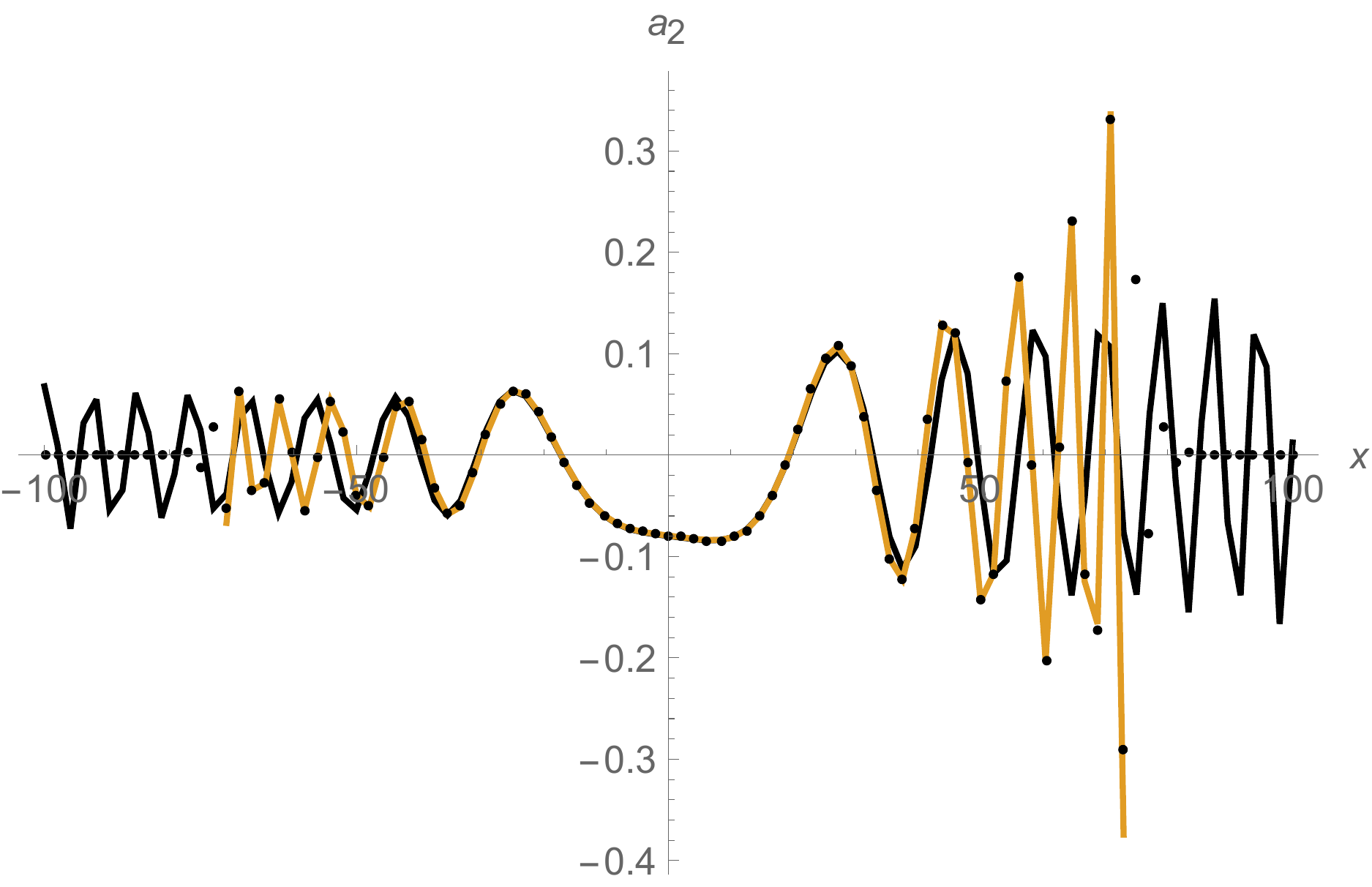}
\end{center}
\vspace{-0.6cm}
\caption{Графики функций $P(x,100)$, $a_1(x,100)$, $a_2(x,100)$ (точки) и их аналитических приближений из теоремы~\ref{Feynman-convergence} (темная линия) и теоремы~\ref{th-ergenium} (светлая).}
\label{fig-approximation}
\end{figure}

Таким образом, для  
$x/t=o(t^{-1/4})$ модель асимптотически воспроизводит \emph{ядро свободной частицы}
\begin{equation}\label{eq-free-particle-kernel}
K(x,t)=\sqrt{\frac{m}{2\pi t}}
\exp \left(\frac{i m x^2}{2t}-\frac{i\pi}{4}\right).
\end{equation}
Первое слагаемое в~\eqref{eq-Feynman-convergence} является ядром~\eqref{eq-free-particle-kernel} для $m=1$, умноженным на $2i^{3/2}e^{-i\pi t/4}$. Последний ``исключительный'' множитель существенно отличается от такового в задаче Фейнмана: $-it$ домножается на $\pi/4$ в экспоненте. Этот дополнительный $\pi/4$ обусловлен конечным шагом решетки $1/m$. Второе слагаемое в~\eqref{eq-Feynman-convergence}  является своего рода релятивистской поправкой, ответственной за зависимость вероятности $P(x,t)$ от $x$. 
Второе слагаемое совпадает со слагаемым в асимптотическом разложении функций Бесселя в запаздывающем пропагаторе~\eqref{eq-relativistic-propagator} для спина $1/2$ при $m=1$, с точностью до коэффициента $\pi/4$ в аргументе косинуса. Мы увидим, что предположение $|x|<t^{3/4}$ является существенным в задаче Фейнмана (см. пример~\ref{p-Feynman-couterexample}). 
Хотя теорема~\ref{Feynman-convergence} легко выводится из теоремы~\ref{th-ergenium}, мы приводим прямое альтернативное доказательство в ~\S\ref{ssec-proofs-alternative}.

\endcomment

\addcontentsline{toc}{myshrink}{}

\section{Масса (квантовое блуждание со сносом)} \label{sec-mass}


{
\hrule
\footnotesize
\noindent\textbf{Вопрос:} какова вероятность обнаружения в клетке $(x,t)$ электрона массой $m$, испущенного из $(0,0)$?
	
	\noindent\textbf{Предположения:} масса и сторона клетки теперь любые.
	
	\noindent\textbf{Результаты:} аналитические выражения для вероятности при больших временах или малой длине стороны клетки, концентрация меры.
\hrule
}

\bigskip

\addcontentsline{toc}{myshrinkalt}{}

\subsection{Тождества}\label{ssec-mass-properties}


\begin{definition} \label{def-mass} 
Зафиксируем $\varepsilon>0$ и $m\ge 0$, называемые \emph{шагом решетки} и \emph{массой частицы} соответственно. Рассмотрим решетку $\varepsilon\mathbb{Z}^2
=\{\,(x,t):x/\varepsilon,t/\varepsilon\in\mathbb{Z}\,\}$ (см.~рис.~\ref{fig-triple-limit}). \emph{Путь $s$ шашки} --- это такая конечная последовательность точек решетки, что вектор из каждой точки (кроме последней) к следующей равен либо $(\varepsilon,\varepsilon)$, либо $(-\varepsilon,\varepsilon)$. Обозначим через $\mathrm{turns}(s)$ количество таких
точек пути $s$ (не считая первую и последнюю), что вектор, соединяющий эту точку с предыдущей, ортогонален вектору,  соединяющему её со следующей. Для каждого
$(x,t)\in\varepsilon\mathbb{Z}^2$, где $t>0$, обозначим через
\begin{equation}\label{eq-def-mass}
{a}(x,t,m,\varepsilon)
:=(1+m^2\varepsilon^2)^{(1-t/\varepsilon)/2}\,i\,\sum_s
(-im\varepsilon)^{\mathrm{turns}(s)}
\end{equation}
сумму по всем путям $s$ шашки на $\varepsilon\mathbb{Z}^2$ из $(0,0)$ в $(x,t)$ с первым шагом в $(\varepsilon,\varepsilon)$. Обозначим
$$P(x,t,m,\varepsilon):=|{a}(x,t,m,\varepsilon)|^2,
\quad
a_1(x,t,m,\varepsilon):=\mathrm{Re}\,a(x,t,m,\varepsilon),
\quad
a_2(x,t,m,\varepsilon):=\mathrm{Im}\,a(x,t,m,\varepsilon).
$$
\end{definition}

\begin{remark}\label{rem-particular}
В частности, $P(x,t)=P(x,t,1,1)$ и 
$a(x,t)=a(x,t,1,1)
=a(x\varepsilon,t\varepsilon,1/\varepsilon,\varepsilon)$.
\end{remark}

Величина $P(x,t,m,\varepsilon)$ интерпретируется как
вероятность обнаружения электрона массой $m$ в квадрате $\varepsilon\times\varepsilon$ с центром $(x,t)$, если электрон был испущен из начала координат. Отметим, что величина $m\varepsilon$, а значит и $P(x,t,m,\varepsilon)$, безразмерна в естественных единицах, в которых $\hbar=c=1$.
На рис.~\ref{P-contour} в центре цвет точки $(x,t)$ показывает значение $P(x,t,1,0.5)$ (если $x+t$ целое). Таблица~\ref{table-am} изображает $a(x,t,m,\varepsilon)$ при малых $x$ и $t$.
Недавно И.~Новиков элегантно доказал, что эта вероятность не обращается в нуль внутри светового конуса: $P(x,t,m,\varepsilon)\ne 0$ при $m>0$, $|x|<t$ и четном $(x+t)/\varepsilon$ \cite[теорема~1]{Novikov-20}.

\begin{table}[htbp]
\begin{tabular}{|c|c|c|c|c|c|c|c|}
\hline
$4\varepsilon$&$\frac{m \varepsilon}{(1+m^2\varepsilon^2)^{3/2}}$&&$\frac{(m \varepsilon- m^3 \varepsilon^3) - m^2 \varepsilon^2 i}{(1+m^2\varepsilon^2)^{3/2}}$&&$\frac{m \varepsilon - 2 m^2 \varepsilon^2 i}{(1+m^2\varepsilon^2)^{3/2}}$&&$\frac{\new{i}}{(1+m^2\varepsilon^2)^{3/2}}$\\
\hline
$3\varepsilon$&&$\frac{m \varepsilon}{1+m^2\varepsilon^2}$&&$\frac{m \varepsilon - m^2 \varepsilon^2 i}{1+m^2\varepsilon^2}$&&$\frac{\new{i}}{(1+m^2\varepsilon^2)}$&\\
\hline
$2\varepsilon$&&&$\frac{m\varepsilon}{\sqrt{1+m^2\varepsilon^2}}$&&$\frac{\new{i}}{\sqrt{1+m^2\varepsilon^2}} $&&\\
\hline
$\varepsilon$&&&&$i$&&&\\
\hline
\diagbox[dir=SW,height=21pt]{$t$}{$x$}&$-2\varepsilon$&$-\varepsilon$&$0$&$\varepsilon$&$2\varepsilon$&$3\varepsilon$&$4\varepsilon$ \\
\hline
\end{tabular}
\caption{(И.~Новиков \cite{Novikov-20}) Величины $a(x,t,m,\varepsilon)$ при малых $x$ и $t$.}
\label{table-am}
\end{table}

\begin{example}[Граничные значения] \label{ex-boundary-values}
Имеем
$a(t,t,m,\varepsilon)=i(1+m^2\varepsilon^2)^{(1-t/\varepsilon)/2}$ и
$a(2\varepsilon-t,t,m,\varepsilon)
=m\varepsilon(1+m^2\varepsilon^2)^{(1-t/\varepsilon)/2}$
для всех $t\in \varepsilon\mathbb{Z}$, $t>0$, и
$a(x,t,m,\varepsilon)=0$ для всех $x>t$ или $x\le -t$.
\end{example}

\begin{example}[Безмассовые и тяжёлые частицы]\label{p-massless}  Для каждого $(x,t)\in \varepsilon\mathbb{Z}^2$, где $t>0$, выполнено
$$P(x,t,0,\varepsilon)=\begin{cases}
1, &\text{\!при }x=t;\\
0,   &\text{\!при }x\ne t.
\end{cases}
\quad\text{и}\quad
\lim_{m \to \infty}P(x,t,m,\varepsilon)=
\begin{cases}
1, &\text{\!если } x=0 \text{ или } \varepsilon, \text{ и }\frac{x+t}{\varepsilon}\text{ четно};\\
0,   &\text{\! иначе}.
\end{cases}$$
\end{example}


Перечислим известные комбинаторные свойства модели \cite{Venegas-Andraca-12, Konno-08}; см.~простые доказательства в~\S\ref{ssec-proofs-basic}.

\begin{proposition}[Уравнение Дирака]\label{p-mass}
Для каждого  $(x,t)\in \varepsilon\mathbb{Z}^2$, где $t>0$, выполнено
\begin{align}\label{eq-Dirac-mass1}
a_1(x,t+\varepsilon,m, \varepsilon) &= \frac{1}{\sqrt{1+m^2\varepsilon^2}}
(a_1(x+\varepsilon,t,m, \varepsilon)
+ m \varepsilon\, a_2(x+\varepsilon,t,m, \varepsilon)),\\
\label{eq-Dirac-mass2}
a_2(x,t+\varepsilon,m, \varepsilon) &= \frac{1}{\sqrt{1+m^2\varepsilon^2}}
(a_2(x-\varepsilon,t,m, \varepsilon)
- m \varepsilon\, a_1(x-\varepsilon,t,m, \varepsilon)).
\end{align}
\end{proposition}

\begin{remark}\label{rem-equation-limit} Это уравнение переходит в уравнение Дирака~\eqref{eq-continuum-Dirac} в непрерывном пределе $\varepsilon\to 0$: для функций $a_1,a_2\colon\mathbb{R}\times (0,+\infty)\to \mathbb{R}$ класса $C^2$, удовлетворяющих  \eqref{eq-Dirac-mass1}--\eqref{eq-Dirac-mass2} на $\varepsilon\mathbb{Z}^2$, левая часть уравнения~\eqref{eq-continuum-Dirac} есть
${O}_m\left(\varepsilon\cdot(\|a_1\|_{C^2}+\|a_2\|_{C^2})\right)$.
\end{remark}

\begin{proposition}[Сохранение вероятности] \label{p-mass2}
Для каждого  $t\in\varepsilon\mathbb{Z}$, где $t>0$, выполнено $\sum\limits
_{x\in\varepsilon\mathbb{Z}}P(x,t,m, \varepsilon)=1$.
\end{proposition}

\begin{proposition}[Уравнение Клейна--Гордона] \label{p-Klein-Gordon-mass} Для каждого $(x,t)\in \varepsilon\mathbb{Z}^2$, где $t>\varepsilon$, выполнено
\begin{align*}
\sqrt{1+m^2\varepsilon^2}\,a(x,t+\varepsilon,m, \varepsilon)
+\sqrt{1+m^2\varepsilon^2}\,a(x,t-\varepsilon,m, \varepsilon)
-a(x+\varepsilon,t,m, \varepsilon)-a(x-\varepsilon,t,m, \varepsilon)=0.
\end{align*}
\end{proposition}

Это уравнение переходит в \emph{уравнение Клейна--Гордона}
$\tfrac{\partial^2 a}{\partial t^2}-\tfrac{\partial^2 a}{\partial x^2}+m^2a=0$
при $\varepsilon\to 0$.

\begin{proposition}[Симметрия]\label{p-symmetry-mass} 
Для каждого $(x,t)\in \varepsilon\mathbb{Z}^2$, где $t>0$, выполнено
\begin{gather*}
  a_1(x,t,m, \varepsilon)=a_1(-x,t,m, \varepsilon),
  \qquad\qquad
  (t-x)\,a_2(x,t,m, \varepsilon)
  =(t+x-2\varepsilon)\,a_2(2\varepsilon-x,t,m, \varepsilon),
  \\
  a_1(x,t,m, \varepsilon)+m\varepsilon\, a_2(x,t,m,\varepsilon)
  =a_1(2\varepsilon-x,t,m, \varepsilon)+m\varepsilon\, a_2(2\varepsilon-x,t,m, \varepsilon).
\end{gather*}
\end{proposition}

\begin{proposition}[Принцип Гюйгенса]\label{p-Huygens} Для всех $x,t,t'\in\varepsilon\mathbb{Z}$, где $t>t'>0$, выполнено
\begin{align*}
\hspace{-1cm}a_1(x,t,m, \varepsilon)  &=\sum \limits_{x'\in\varepsilon\mathbb{Z}}
\left[ a_2(x',t',m, \varepsilon) a_1(x-x'+\varepsilon,t-t'+\varepsilon,m, \varepsilon) + a_1(x',t',m, \varepsilon)
a_2(x'-x+\varepsilon,t-t'+\varepsilon,m, \varepsilon) \right],\\
\hspace{-1cm}a_2(x,t,m, \varepsilon)  &= \sum \limits_{x'\in\varepsilon\mathbb{Z}}
\left[ a_2(x',t',m, \varepsilon)
a_2(x-x'+\varepsilon,t-t'+\varepsilon,m, \varepsilon) - a_1(x',t',m, \varepsilon)
a_1(x'-x+\varepsilon,t-t'+\varepsilon,m, \varepsilon) \right].
\end{align*}
\end{proposition}

Неформально, принцип Гюйгенса означает, что каждая черная клетка $(x',t')$ на горизонтали с номером $t'/\varepsilon$ действует как независимый точечный источник с амплитудой и фазой, определяемыми величиной $a(x',t',m, \varepsilon)$.

\begin{proposition}[Одновременн\'ое реккурентное соотношение]
\label{p-equal-time} Для каждого $(x,t)\in\varepsilon\mathbb{Z}^2$, где $t>0$, выполнено
\begin{multline}\label{eq-p-equal-time}
(x+\varepsilon)((x-\varepsilon)^2-(t-\varepsilon)^2)
a_1(x-2\varepsilon,t,m,\varepsilon) +(x-\varepsilon)((x+\varepsilon)^2-(t-\varepsilon)^2)
a_1(x+2\varepsilon,t,m,\varepsilon)=
\\
=2x \left((1+2m^2\varepsilon^2)(x^2-\varepsilon^2)-(t-\varepsilon)^2\right)
a_1(x,t,m,\varepsilon),
\end{multline}
\vspace{-1.0cm}
\begin{multline*}
x((x-2\varepsilon)^2-t^2)a_2(x-2\varepsilon,t,m,\varepsilon) +(x-2\varepsilon)(x^2-(t-2\varepsilon)^2)
a_2(x+2\varepsilon,t,m,\varepsilon)=
\\
=2(x-\varepsilon)
\left((1+2m^2\varepsilon^2)(x^2-2\varepsilon x)-t^2+2\varepsilon t \right)
a_2(x,t,m,\varepsilon).
\end{multline*}
\end{proposition}

Это позволяет быстро вычислять $a_{1}(x,t)$ и $a_2(x,t)$ на далеких горизонталях, начиная с $x=2\varepsilon-t$ и $x=t$ соответственно (см.~пример~\ref{ex-boundary-values}).



\begin{proposition}[``Явная'' формула] \label{p-mass3}
Для каждого $(x,t)\in\varepsilon\mathbb{Z}^2$, такого что $|x|<t$ и $(x+t)/\varepsilon$ четно, выполнено
\begin{align}
a_1(x,t,m,\varepsilon) &=
(1+m^2\varepsilon^2)^{(1-t/\varepsilon)/2}
\sum_{r=0}^{\frac{t-|x|}{2\varepsilon}}(-1)^r \binom{(x+t)/2\varepsilon-1}{r}\binom{(t-x)/2\varepsilon-1}{r}
(m\varepsilon)^{2r+1},
\label{eq1-p-mass}\\
a_2(x,t,m,\varepsilon)&=
(1+m^2\varepsilon^2)^{(1-t/\varepsilon)/2}
\sum_{r=1}^{\frac{t-|x|}{2\varepsilon}}(-1)^r \binom{(x+t)/2\varepsilon-1}{r}\binom{(t-x)/2\varepsilon-1}{r-1}
(m\varepsilon)^{2r}\new{.}
\label{eq2-p-mass}
\end{align}
\end{proposition}

\begin{remark}\label{rem-hypergeo}
  Для всех $|x|\ge t$ имеем $a(x,t,m,\varepsilon)=0$, кроме случая \new{}$0<x=t\in\varepsilon\mathbb{Z}$, когда $a(t,t,m,\varepsilon)=(1+m^2\varepsilon^2)^{(1-t/\varepsilon)/2}i$. Подчеркнем, что предложение \emph{неприменимо} при $|x|\ge t$.
%

По определению гипергеометрической функции Гаусса,
та же формула переписывается так:
\begin{align*}
a_1(x,t,m,\varepsilon)
&=m\varepsilon\left({1+m^2\varepsilon^2}\right)^{(1-t/\varepsilon)/2}
\cdot{}_2F_1\left(1 - \frac{x+t}{2\varepsilon}, 1 + \frac{x-t}{2\varepsilon}; 1; -m^2\varepsilon^2\right),\\
a_2(x,t,m,\varepsilon)
&=m^2\varepsilon^2\left({1+m^2\varepsilon^2}\right)^{(1-t/\varepsilon)/2}
\cdot{}_2F_1\left(2 - \frac{x+t}{2\varepsilon}, 1 + \frac{x-t}{2\varepsilon}; 2; -m^2\varepsilon^2\right)\cdot\left(1 - \frac{x+t}{2\varepsilon}\right).
\end{align*}
Это дает много новых тождеств. Например, соотношения смежности Гаусса связывают значения $a(x,t,m,\varepsilon)$ в любых $3$ соседних точках решетки; ср.~с предложениями~\ref{p-mass} и~\ref{p-equal-time}. В терминах многочленов Якоби,
\begin{align*}
a_1(x,t,m,\varepsilon)
&=
m\varepsilon(1+m^2\varepsilon^2)^{(x/\varepsilon-1)/2}
P_{(x+t)/2\varepsilon-1}^{(0, -x/\varepsilon)} \left(\frac{1-m^2\varepsilon^2}{1+m^2\varepsilon^2}\right),\\
a_2(x,t,m,\varepsilon)
&=
-m^2\varepsilon^2(1+m^2\varepsilon^2)^{(x/\varepsilon - 3)/2}
P_{(x+t)/2\varepsilon - 2}^{(1, 1 - x/\varepsilon)} \left(\frac{1-m^2\varepsilon^2}{1+m^2\varepsilon^2}\right).
\end{align*}
Есть аналогичное выражение через многочлены Кравчука (ср.~с предложением~\ref{cor-coefficients}).
В терминах характерных многочленов Стенли (определенных в \cite[\S2]{Stanley-04}),
$$
a_2(0,t,m,\varepsilon)
=(-1)^{t/2\varepsilon-1}(1+m^2\varepsilon^2)^{(1-t/\varepsilon)/2}
\left(\tfrac{t}{2\varepsilon}-1\right)G_{t/2\varepsilon-1}(1;m^2\varepsilon^2).
$$
\end{remark}

\begin{proposition}[Интеграл Фурье] \label{cor-fourier-integral}
Положим $\omega_p:=\frac{1}{\varepsilon}\arccos(\frac{\cos p\varepsilon}{\sqrt{1+m^2\varepsilon^2}})$. Тогда для каждого $m>0$ и $(x,t)\in \varepsilon\mathbb{Z}^2$, такого что $t>0$ и $(x+t)/\varepsilon$ четно, выполнено
\begin{align*}
  a_1(x, t,m,\varepsilon)&=
  \frac{im\varepsilon^2}{2\pi}
  \int_{-\pi/\varepsilon}^{\pi/\varepsilon}
  \frac{e^{i p x-i\omega_p(t-\varepsilon)}\,dp}
  {\sqrt{m^2\varepsilon^2+\sin^2(p\varepsilon)}},\\
  a_2(x, t,m,\varepsilon)&=
  \frac{\varepsilon}{2\pi}\int_{-\pi/\varepsilon}^{\pi/\varepsilon}
  \left(1+
  \frac{\sin (p\varepsilon)} {\sqrt{m^2\varepsilon^2+\sin^2(p\varepsilon)}}\right) e^{ip(x-\varepsilon)-i\omega_p(t-\varepsilon)}\,dp.
\end{align*}
\end{proposition}


Интеграл Фурье представляет волну, испущенную точечным источником, как суперпозицию волн с длиной волны $2\pi/p$ и частотой~$\omega_p$.

\begin{proposition}[Преобразование Фурье по всему пространству-времени] \label{cor-double-fourier}
Обозначим $\delta_{x\varepsilon}:=1$, если $x=\varepsilon$, и $\delta_{x\varepsilon}:=0$, если $x\ne\varepsilon$.
Для каждого $m>0$ и $(x,t)\in \varepsilon\mathbb{Z}^2$, такого что $t>0$ и $(x+t)/\varepsilon$ четно, выполнено
\begin{align*}
  a_1(x,t,m,\varepsilon)&=
  \lim_{\delta\to+0}
  \frac{m\varepsilon^3}{4\pi^2}
  \int_{-\pi/\varepsilon}^{\pi/\varepsilon}
  \int_{-\pi/\varepsilon}^{\pi/\varepsilon}
  \frac{e^{i p x-i\omega(t-\varepsilon)}\,d\omega dp} {\sqrt{1+m^2\varepsilon^2}\cos(\omega\varepsilon)
  -\cos(p\varepsilon)-i\delta},\\
  a_2(x,t,m,\varepsilon)&=
  \lim_{\delta\to+0}
  \frac{-i\varepsilon^2}{4\pi^2}
  \int_{-\pi/\varepsilon}^{\pi/\varepsilon}
  \int_{-\pi/\varepsilon}^{\pi/\varepsilon}
  \frac{\sqrt{1+m^2\varepsilon^2}\sin(\omega\varepsilon)+\sin(p\varepsilon)} {\sqrt{1+m^2\varepsilon^2}\cos(\omega\varepsilon)
  -\cos(p\varepsilon)-i\delta}
  e^{i p (x-\varepsilon)-i\omega(t-\varepsilon)}
  \,d\omega dp+\delta_{x\varepsilon}\delta_{t\varepsilon}.
\end{align*}
\end{proposition}


\addcontentsline{toc}{myshrinkalt}{}

\subsection{Асимптотики}\label{ssec-mass-asymptotic}

\textbf{Предел больших времен.} Для фиксированных $m,\varepsilon$ и большого $t$ функция $a(x,t,m,\varepsilon)$
\begin{description}
  \item[(A)] колеблется на отрезке между (примерными) пиками $x=\pm \frac{t}{\sqrt{1+m^2\varepsilon^2}}$ (см.~теорему~\ref{th-ergenium});
  \item[(B)] приближается функцией Эйри вблизи пиков (см.~теорему~\ref{th-Airy});
  \item[(C)] экспоненциально убывает вне отрезка, ограниченного пиками (см.~теорему~\ref{th-outside} и рис.~\ref{fig-approx-error}).
\end{description}

Начнем с обсуждения режима~\textbf{(A)}. Сформулируем основную теорему: аналитическое приближение функции $a(x,t,m,\varepsilon)$, которое является очень точным при $x/t$, расположенным между пиками, но не слишком близко к ним (рис.~\ref{fig-approx-error} слева).

\begin{figure}[htbp]
  \centering
 \begin{tabular}{ccc}
\includegraphics[width=0.3\textwidth]{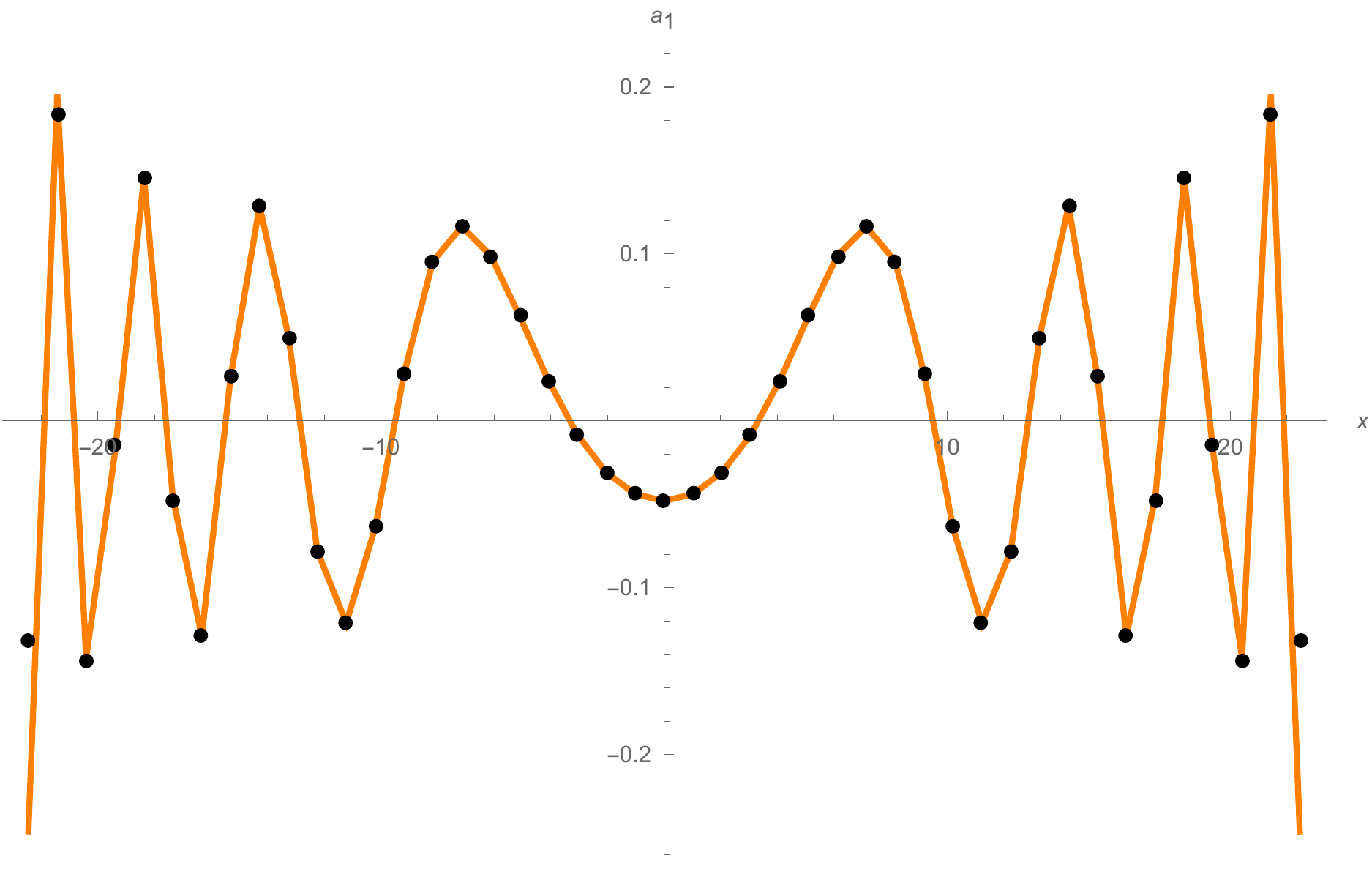}
&\includegraphics[width=0.3\textwidth]{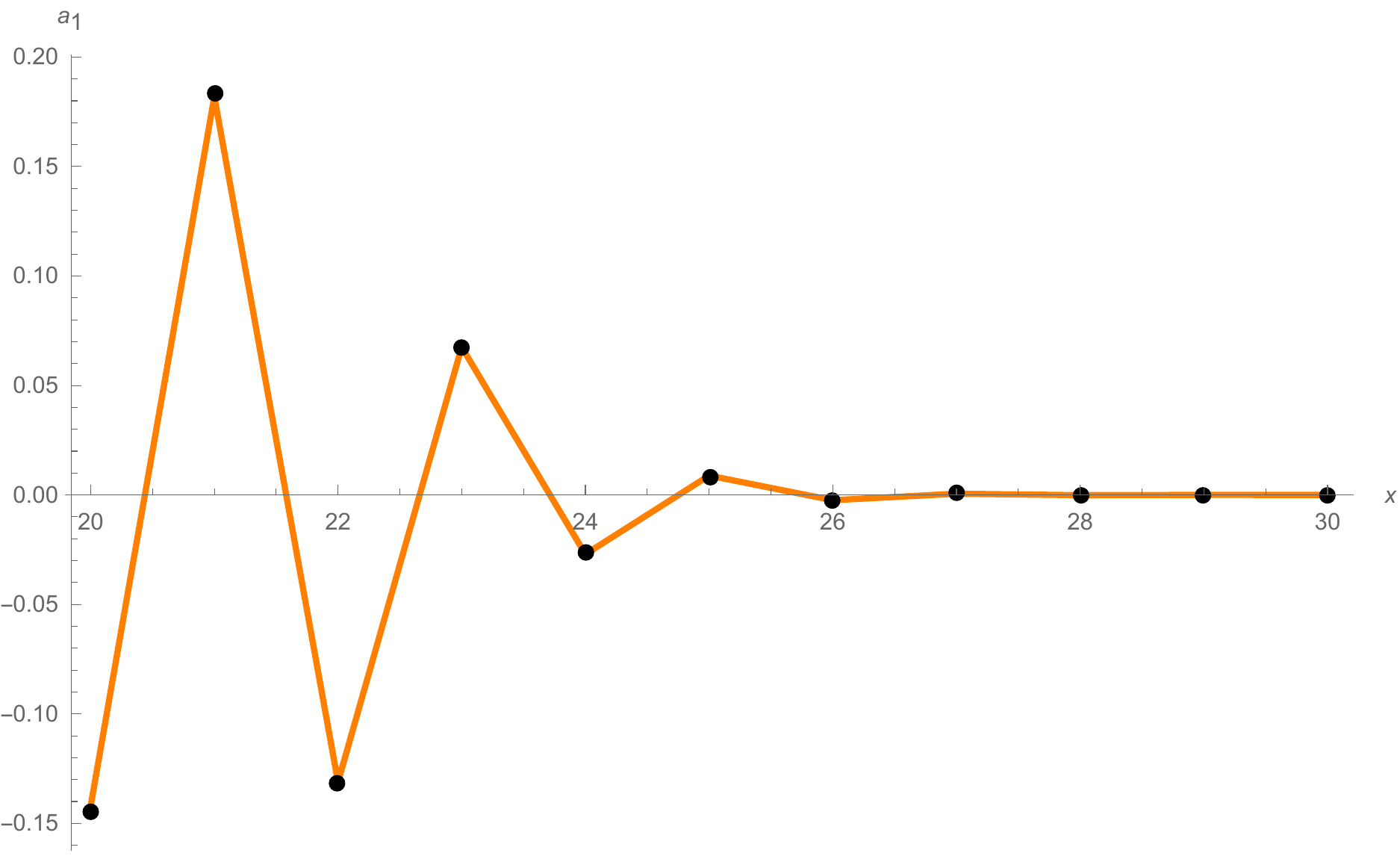} &\includegraphics[width=0.3\textwidth]{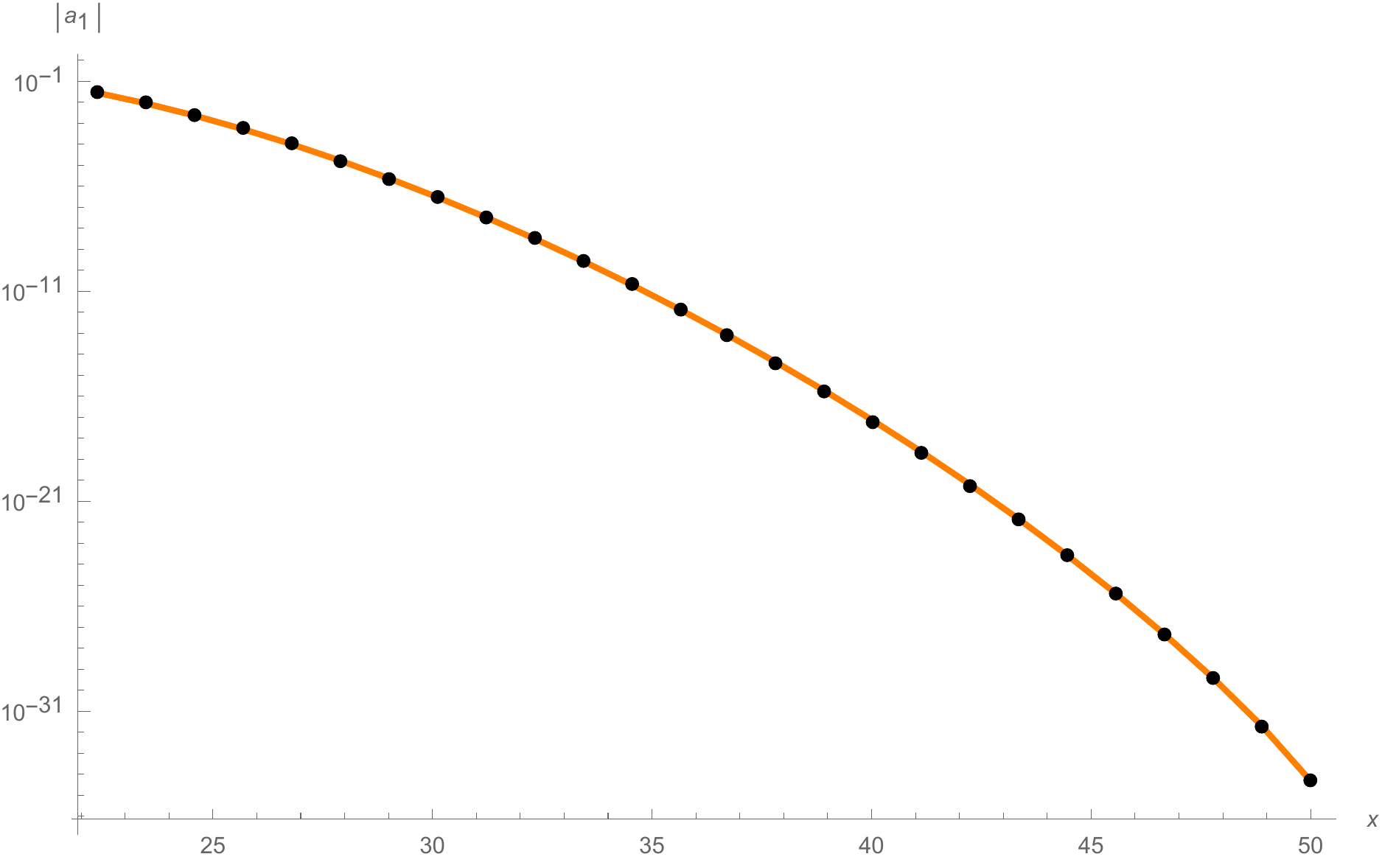}\\
\includegraphics[width=0.3\textwidth]{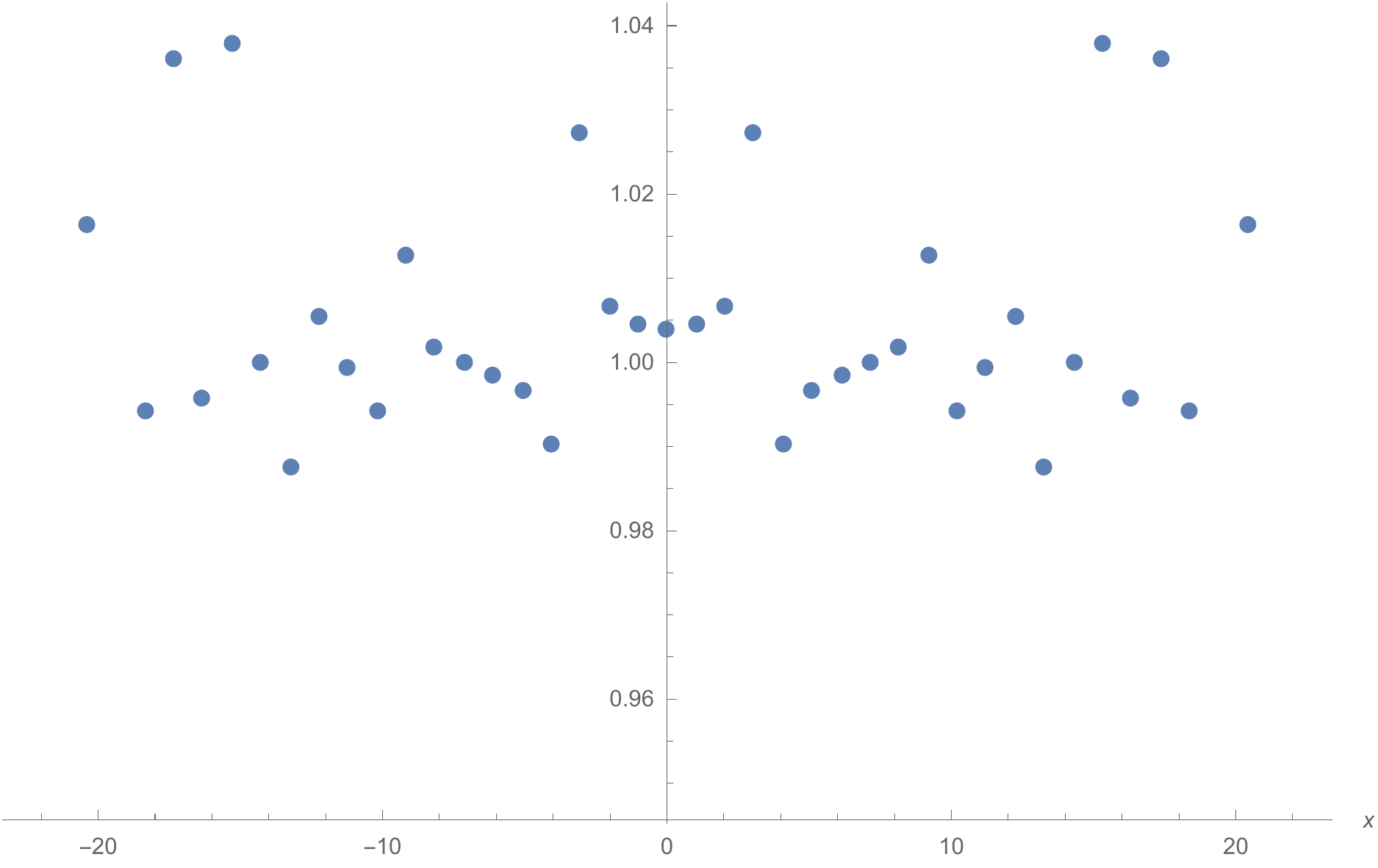}
&
\includegraphics[width=0.3\textwidth]{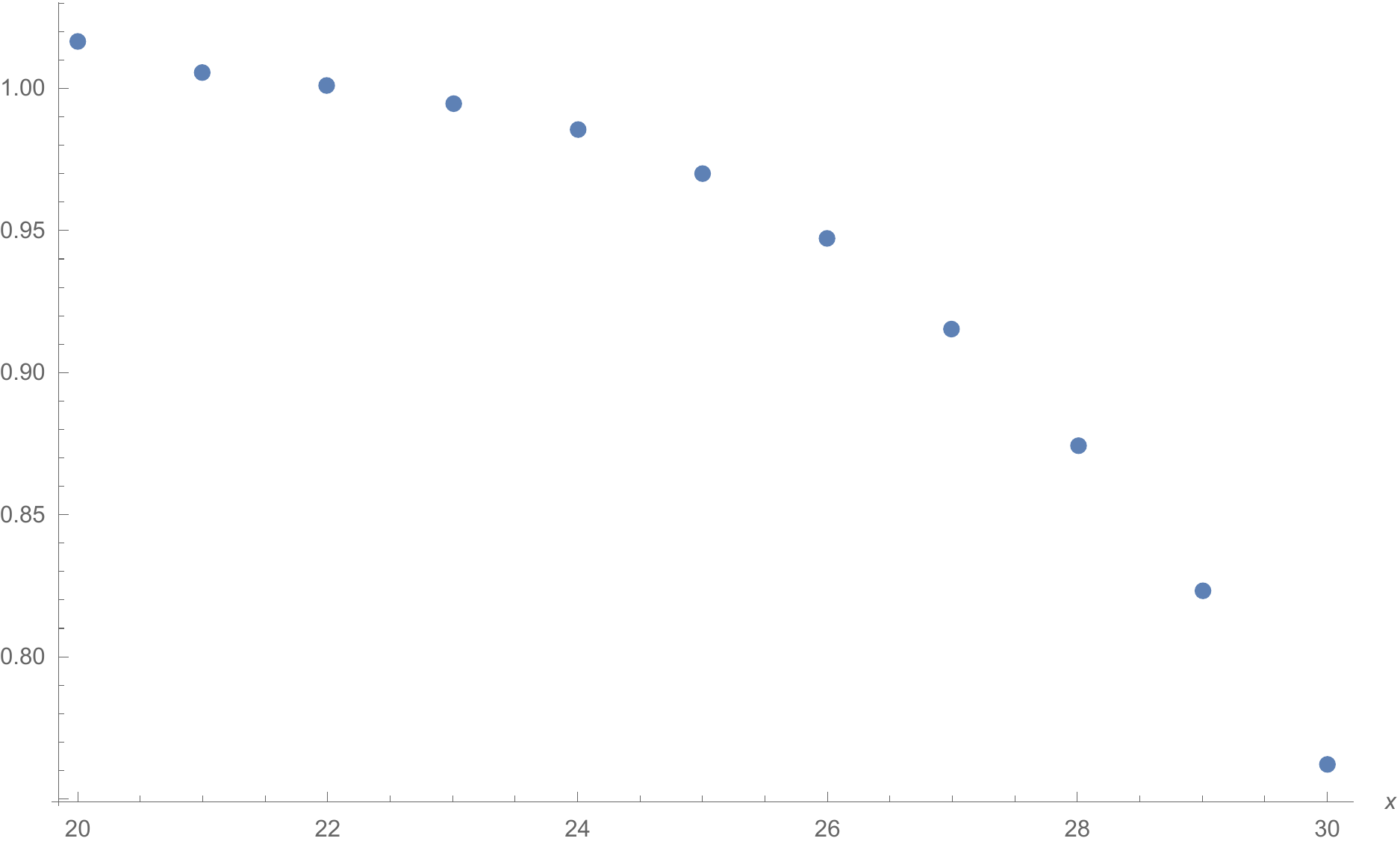}
&
\includegraphics[width=0.3\textwidth]{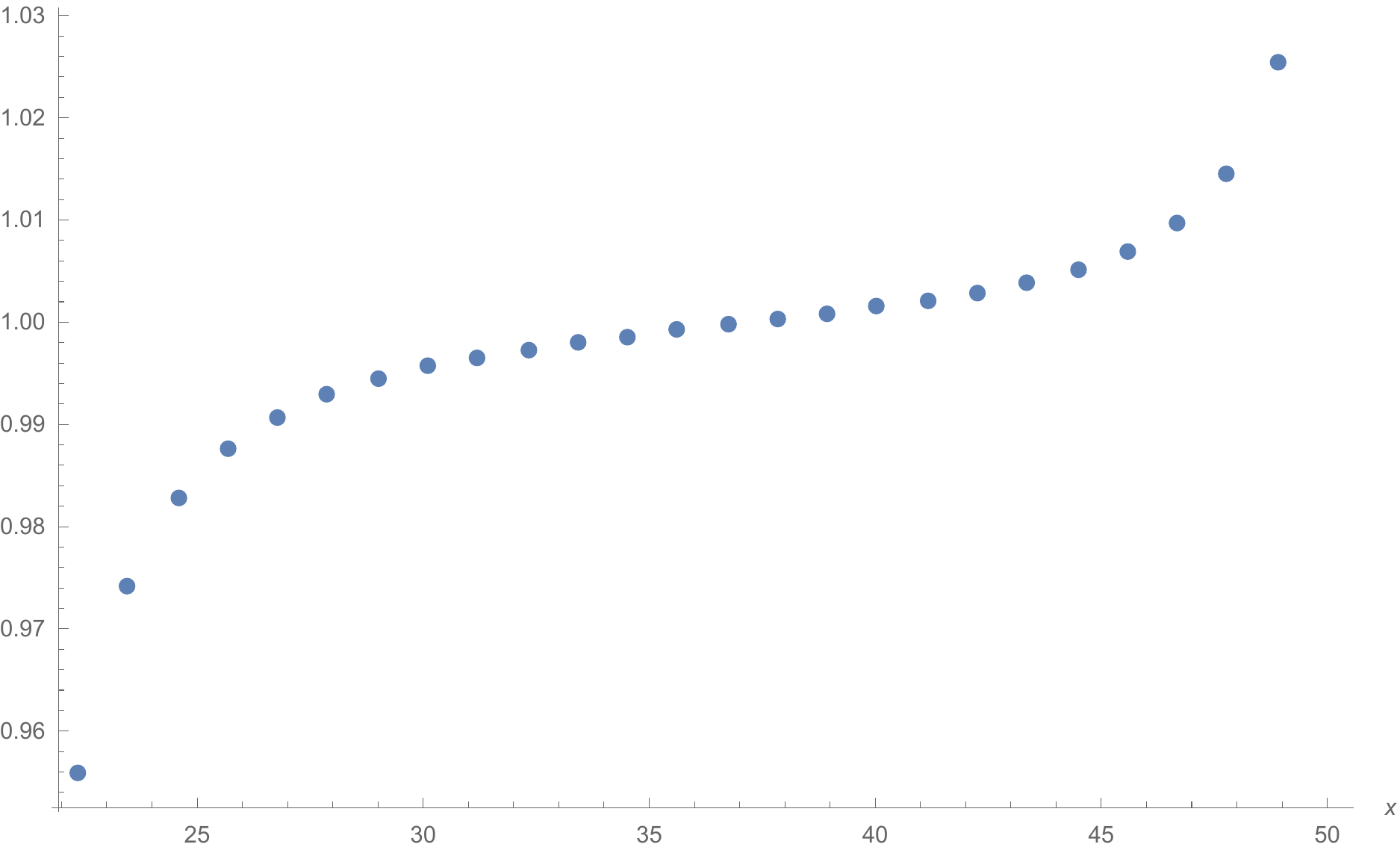}
\end{tabular}
  \caption{Графики функции $a_1(x,50,4,0.5)$ (точки сверху), ее аналитического приближения (кривые сверху) из теорем~\ref{th-ergenium} (слева), \ref{th-Airy} (в центре) и~\ref{th-outside} (справа), и их отношения (снизу). График вверху справа изображает модуль этой функции в логарифмической шкале}
%
  \label{fig-approx-error}
\end{figure}

\begin{theorem}[Асимптотика между пиками при больших временах; см.~рис.~\ref{fig-approx-error} слева] \label{th-ergenium}
Для каждого $\delta>0$ найдется $C_\delta>0$ такое, что при всех $m,\varepsilon>0$ и всех $(x,t)\in\varepsilon\mathbb{Z}^2$, удовлетворяющих неравенствам
\begin{equation}\label{eq-case-A}
 |x|/t<1/\sqrt{1+m^2\varepsilon^2}-\delta, \qquad
 \varepsilon\le 1/m, \qquad t>C_\delta/m,
\end{equation}
выполнены асимптотики
\begin{align}\label{eq-ergenium-re}
{a}_1\left(x,t+\varepsilon,m,\varepsilon\right)
&={\varepsilon}\sqrt{\frac{2m}{\pi}}
\left(t^2-(1+m^2\varepsilon^2)x^2\right)^{-1/4}
\sin \theta(x,t,m,\varepsilon)
+O_\delta\left(\frac{\varepsilon}{m^{1/2}t^{3/2}}\right),\\
\label{eq-ergenium-im}
{a}_2\left(x+\varepsilon,t+\varepsilon,m,\varepsilon\right)
&={\varepsilon}\sqrt{\frac{2m}{\pi}}
\left(t^2-(1+m^2\varepsilon^2)x^2\right)^{-1/4}\sqrt{\frac{t+x}{t-x}}
\cos \theta(x,t,m,\varepsilon)
+O_\delta\left(\frac{\varepsilon}{m^{1/2}t^{3/2}}\right),
\intertext{для четных и нечетных $(x+t)/\varepsilon$ соответственно, где}
\label{eq-theta}
\theta(x,t,m,\varepsilon)&:=
\frac{t}{\varepsilon}\arcsin
\frac{m\varepsilon t} {\sqrt{\left(1+m^2\varepsilon^2\right)\left(t^2-x^2\right)}}
-\frac{x}{\varepsilon}\arcsin
\frac{m\varepsilon x}{\sqrt{t^2-x^2}}
+\frac{\pi }{4}.
\end{align}
\end{theorem}


Здесь и далее запись $f(x,t,m,\varepsilon)
=g(x,t,m,\varepsilon)+{O}_\delta\left(h(x,t,m,\varepsilon)\right)$ означает, что существует константа $C(\delta)$ (зависящая от $\delta$, но \emph{не} от $x,t,m,\varepsilon$), такая что для всех $x,t,m,\varepsilon,\delta$, удовлетворяющих условию теоремы, выполнено $|f(x,t,m,\varepsilon)-g(x,t,m,\varepsilon)|\le C(\delta)\,h(x,t,m,\varepsilon)$. 

Главные члены в теореме~\ref{th-ergenium} были \new{впервые} найдены
в~\cite[теорема~2]{Ambainis-etal-01} в частном случае $m=\varepsilon=1$. Остаточные члены были оценены в~\cite[предложение~2.2]{Sunada-Tate-12}; эта оценка имела такой же порядок по $t$, но была \emph{неравномерна} по $m$ и $\varepsilon$ (и \emph{не могла} быть равномерной в том виде, в котором там приводилась, так как иначе получилось бы противоречие с следствием~\ref{cor-uniform} при $\varepsilon\to 0$). Получение равномерных оценок~\eqref{eq-ergenium-re}--\eqref{eq-ergenium-im} потребовало изысканного варианта метода стационарной фазы и дополнительных шагов~3--4 в доказательстве (см.~\S\ref{ssec-proofs-main}). Другие методы (асимптотика Дарбу для многочленов Лежандра и круговой метод Харди--Литтлвуда)
были применены в \cite[теорема~3]{Bogdanov-20} и \cite[теорема~2]{SU-20} для получения~\eqref{eq-ergenium-re} в частном случае $x=0$ и для более грубой асимптотики при $x$ вблизи $0$.



Теорема~\ref{th-ergenium} имеет ряд интересных следствий.
Во-первых, она позволяет перейти к слабому пределу по времени (см.~рис.~\ref{fig-correlation}). По сравнению с теоремой~\ref{th-limiting-distribution}, это дает сходимость в более сильном смысле, недостижимому методом моментов и методом~\cite[теорема~1.1]{Sunada-Tate-12}.

\begin{corollary}[Предельное распределение при больших временах] 
\label{th-limiting-distribution-mass}
Для каждого $m>0$ и $\varepsilon\le 1/m$ выполнено
\begin{equation*}
\lim_{\substack{t\to\infty\\t\in\varepsilon\mathbb{Z}}}
\sum_{\substack{x\le vt\\ x\in\varepsilon\mathbb{Z}}}P(x,t,m,\varepsilon)
\rightrightarrows F(v,m,\varepsilon):=
  \begin{cases}
  0, & \mbox{если } v\le -1/\sqrt{1+m^2\varepsilon^2}; \\
  \frac{1}{\pi}\arccos\frac{1-(1+m^2\varepsilon^2)v}{\sqrt{1+m^2\varepsilon^2}(1-v)},
     & \mbox{если } |v|<1/\sqrt{1+m^2\varepsilon^2}; \\
  1, & \mbox{если }v\ge 1/\sqrt{1+m^2\varepsilon^2}
\end{cases}
\end{equation*}
при $t\to\infty$ равномерно по $v$.
\end{corollary}

Напомним, что все наши результаты можно сформулировать в терминах диаграмм Юнга или многочленов Якоби.

\begin{corollary}[Ступени диаграмм Юнга; см.~рис.~\ref{fig-young}] \label{cor-young}
Обозначим через $n_+(h\times w)$ и $n_-(h\times w)$ количество
диаграмм Юнга ровно с $h$ строками и $w$ столбцами, имеющими четное и нечетное число ступеней соответственно (см.~определение на стр.~\pageref{page-young}).
Тогда для почти всех $r>1$ выполнено
\begin{equation*}
\limsup_{\substack{w\to\infty}}
\frac{\sqrt{w}}{2^{(r+1)w/2}}
\left|n_+(\lceil rw\rceil\times w)-n_-(\lceil rw\rceil\times w)\right|=
\begin{cases}
  \frac{1}{\sqrt{\pi}}(6r-r^2-1)^{-1/4},
  & \mbox{при } r< 3+2\sqrt{2}; \\
  0, & \mbox{при }r> 3+2\sqrt{2}.
\end{cases}
\end{equation*}
\end{corollary}

Для режимов~\textbf{(B)} и~\textbf{(C)} вблизи и снаружи от пиков сформулируем две теоремы Сунады--Татэ без подробного обсуждения доказательств.
(Строго говоря, главные члены в литературе~\cite[теорема~2]{Ambainis-etal-01}, \cite[предложения~2.2, 3.1, 4.1]{Sunada-Tate-12}  немного отличаются от~\eqref{eq-ergenium-re}, \eqref{eq-Airy-re}, \eqref{eq-outside-re}; но разница лежит в пределах оценки остаточного члена. С практической точки зрения, последние три приближения лучше на несколько порядков.)

Режим~\textbf{(B)} вблизи пиков описывается \emph{функцией Эйри}
$$
\mathrm{Ai}(\lambda) := \frac{1}{\pi} \int_{0}^{+\infty} \cos\left(\lambda p + \frac{p^{3}}{3}\right) \,dp.
$$

\begin{theorem}[Асимптотика вблизи пиков при больших временах; см.~рис.~\ref{fig-approx-error} в центре]  \label{th-Airy} \textup{\cite[предложение~3.1]{Sunada-Tate-12}}\new{}
Для любых $m,\varepsilon,\Delta>0$ и любой последовательности $(x_n,t_n)\in\varepsilon\mathbb{Z}^2$, такой что
\begin{equation}\label{eq-Airy}
|x_n/t_n-1/\sqrt{1+m^2\varepsilon^2}|<\Delta\, t_n^{-2/3}, \qquad t_n=n\varepsilon,
\end{equation}
выполнено
\begin{align}\label{eq-Airy-re}
{a}_1\left(\pm x_n,t_n+\varepsilon,m,\varepsilon\right)
&=
\left(-1\right)^{\frac{t_n - x_n - \varepsilon}{2\varepsilon}}  m\varepsilon\left( \frac{2}{m^2\varepsilon t_n} \right)^{1/3}
\mathrm{Ai} \left( \Delta(x_n,t_n,m,\varepsilon) \right)+ O_{m,\varepsilon,\{x_n\}}\left(n^{-2/3}\right),\\
\label{eq-Airy-im}
\hspace{-0.6cm}
a_2\left(\pm x_n+\varepsilon,t_n+\varepsilon,m,\varepsilon \right)
&= \left(-1\right)^{\frac{t_n - x_n}{2\varepsilon}}\left( \sqrt{1+m^2\varepsilon^2} \pm 1 \right)\left( \frac{2}{m^2\varepsilon t_n} \right)^{1/3} \mathrm{Ai} \left(\Delta(x_n,t_n,m,\varepsilon) \right)+ O_{m,\varepsilon,\{x_n\}}\left(n^{-2/3}\right)
\intertext{при $(x_n+t_n)/\varepsilon$ нечетных и четных соответственно, где знаки $\pm$ согласованы и}
\label{eq-Delta}
\Delta(x_n,t_n,m,\varepsilon)&:=\left( \frac{2}{m^2\varepsilon t_n} \right)^{1/3} \frac{\sqrt{1+m^2\varepsilon^2}\,x_n-t_n}{\varepsilon}.
\end{align}
\end{theorem}

Мы пишем $f(x_n,t_n,m,\varepsilon)
=O_{m,\varepsilon,\{x_n\}}\left(g(n)\right)$, если найдется константа $C(m,\varepsilon,\{x_n\})$ (зависящая от $m,\varepsilon$ и всей последовательности $\{x_n\}_{n=1}^\infty$, но \emph{не} от $n$), такая что для каждого $n$ выполнено $|f(x_n,t_n,m,\varepsilon)|\le C(m,\varepsilon,\{x_n\})g(n)$.

Недавно П.~Закорко продолжила \eqref{eq-Airy-re}--\eqref{eq-Airy-im} на б\'ольшую окрестность пиков \new{}\cite{Zakorko-21}.

Экспоненциальное убывание снаружи от пиков было сформулировано без доказательства в~\cite[Theorem~1]{Ambainis-etal-01}. Доказательство появилось только через десятилетие, когда была установлена следующая асимптотика.

\begin{theorem}[Асимптотика снаружи от пиков при больших временах; см.~рис.~\ref{fig-approx-error} справа] \label{th-outside} \textup{\cite[предложение~4.1]{Sunada-Tate-12}}\new{}
Для любых  $m,\varepsilon,\Delta>0$, $v\in (-1;1)$ и любой последовательности $(x_n,t_n)\in\varepsilon\mathbb{Z}^2$, удовлетворяющей условиям
\begin{equation}\label{eq-outside}
 1/\sqrt{1+m^2\varepsilon^2}<|v|<1, \qquad
 |x_n-vt_n|<\Delta, \qquad t_n=n\varepsilon,
\end{equation}
выполнено \new{}\new{}
\begin{align}\notag
{a}_1\left(x_n,t_n+\varepsilon,m,\varepsilon\right)
&=
{\varepsilon}\sqrt{\frac{m}{2\pi t_n}}
\frac{(-1)^{(t_n-|x_n|-\varepsilon)/2\varepsilon}}{\left((1+m^2\varepsilon^2)v^2-1\right)^{1/4}}\cdot
\\
\label{eq-outside-re}
&\cdot \exp\left(-\frac{t_n}{2\varepsilon}
H\left(\frac{x_n}{t_n},m,\varepsilon\right)\right)
\left(1+O_{m,\varepsilon,\{x_n\}}\left(\frac{1}{n}\right)\right),
\\
\notag
{a}_2\left(x_n+\varepsilon,t_n+\varepsilon,m,\varepsilon\right)
&={\varepsilon}\sqrt{\frac{m}{2\pi t_n}}
\frac{(-1)^{(t_n-|x_n|)/2\varepsilon}}
{\left((1+m^2\varepsilon^2)v^2-1\right)^{1/4}}
\sqrt{\frac{1+v}{1-v}}\cdot
\\
\label{eq-outside-im}
&\cdot \exp\left(-\frac{t_n}{2\varepsilon}
H\left(\frac{x_n}{t_n},m,\varepsilon\right)\right)
\left(1+O_{m,\varepsilon,\{x_n\}}\left(\frac{1}{n}\right)\right)
\intertext{при $(x_n+t_n)/\varepsilon$ нечетных и четных соответственно, где}
\label{eq-hermite}
H(v,m,\varepsilon)&:=
-2\,\mathrm{arcosh}\,
\frac{m\varepsilon} {\sqrt{\left(1+m^2\varepsilon^2\right)\left(1-v^2\right)}}
+2|v|\,\mathrm{arcosh}\,
\frac{m\varepsilon |v|}{\sqrt{1-v^2}}.
\end{align}
\end{theorem}

Функция~\eqref{eq-hermite} положительна и выпукла по $v$ при $\frac{1}{\sqrt{1+m^2\varepsilon^2}}<|v|<1$ \cite[теорема~1.4]{Sunada-Tate-12}.

\begin{corollary}[Предельная плотность свободной энергии; см.~рис.~\ref{fig-distribution} в центре] \label{cor-free}
Для любых $m,\varepsilon>0$, $v\in(-1;1)$ и функции $H(v,m,\varepsilon)$, заданной равенством~\eqref{eq-hermite}, выполнено
\begin{equation*}
\lim_{\substack{t\to\infty\\ t\in 2\varepsilon\mathbb{Z}}}
\frac{1}{t}\log P\left(2\varepsilon\left\lceil \frac{vt}{2\varepsilon}\right\rceil,t,m,\varepsilon\right)
=
\begin{cases}
-H(v,m,\varepsilon)/\varepsilon,
  & \mbox{при } |v|>1/\sqrt{1+m^2\varepsilon^2};\\
0,& \mbox{при } |v|<1/\sqrt{1+m^2\varepsilon^2}.
\end{cases}
\end{equation*}
\end{corollary}

Это означает фазовый переход в модели Изинга в сильном смысле: предельная плотность свободной энергии (мнимая часть которой пропорциональна выражению в левой части) неаналитична по $v$.

\textbf{Фейнмановский тройной предел.} Теорема~\ref{th-ergenium} позволяет перейти к пределу $(1/t,x/t,\varepsilon)\to 0$ следующим образом.

\begin{corollary}[Более простая и грубая асимптотика] \label{cor-intermediate-asymptotic-form}
В предположениях теоремы~\ref{th-ergenium} выполнено
\begin{equation}\label{eq-intermediate-asymptotic-form}
{a}\left(x,t,{m},{\varepsilon}\right)
=\varepsilon\sqrt{\frac{2m}{\pi t}}
\exp\left(-im\sqrt{t^2-x^2}+\frac{i\pi}{4}\right) \left(1+{O}_\delta
\left(\frac{1}{mt}+\frac{|x|}{t}+m^3\varepsilon^2t\right)\right).
\end{equation}
\end{corollary}

\begin{corollary}[Фейнмановский тройной предел; см.~рис.~\ref{fig-triple-limit}] \label{cor-feynman-problem}  
Для каждого $m\ge 0$ и каждой последовательности $(x_n,t_n,\varepsilon_n)$, такой что $(x_n,t_n)\in \varepsilon_n\mathbb{Z}^2$, число $(x_n+t_n)/\varepsilon_n$ четно и
\begin{equation}\label{eq-feynman-problem-assumptions}
1/t_n,\quad x_n/t_n^{3/4},\quad \varepsilon_nt_n^{1/2}\quad\to\quad 0
 \qquad \text{при }\quad n\to\infty,
\end{equation}
выполнена эквивалентность последовательностей
\begin{equation}\label{eq-feynman-problem}
\frac{1}{2i\varepsilon_n}{a}\left(x_n,t_n,{m},{\varepsilon_n}\right)
\sim\sqrt{\frac{m}{2\pi t_n}}
\exp\left(-imt_n-\frac{i\pi}{4}+\frac{imx_n^2}{2t_n}\right)
\qquad\text{ при }\quad n\to\infty.
\end{equation}
\end{corollary}

Для эквивалентности~\eqref{eq-feynman-problem} предположения~\eqref{eq-feynman-problem-assumptions} существенны и точны, как показывает следующий пример.

\begin{example}\label{p-Feynman-couterexample}
Эквивалентность~\eqref{eq-feynman-problem} не выполняется ни для $(x_n,t_n,\varepsilon_n)=(n^3,n^4,1/n^4)$, ни для $(0,4n^2,1/{2n})$, ни для $(0,2n\varepsilon,\varepsilon)$, где $\varepsilon=\mathrm{const}<1/m$: 
в этих случаях предел отношения левой части к правой равен $e^{im/8}$, $e^{im^3/3}$ и не существует соответственно, а не равен~$1$.
(Отсутствие последнего предела подтверждает, что стремление к нулю шага решетки подразумевалось в проблеме Фейнмана.)
\end{example}

Следствие~\ref{cor-feynman-problem} решает проблему Фейнмана (и более того, исправляет ее формулировку, раскрывая точные ограничения). Основная трудность здесь в переходе к \emph{тройному}, а не \emph{повторному} пределу. Нам неизвестно подхода, который мог бы решить проблему без доказательства теоремы~\ref{th-ergenium} целиком.  Например, асимптотика Дарбу для многочленов Якоби (см.~замечание~\ref{rem-hypergeo}) подходит для вычисления повторного предела, когда сначала $t\to+\infty$, а потом $\varepsilon\to 0$, что дает (более слабый) результат, уже не зависящий от $x$.
Ни асимптотика Дарбу, ни Мелера--Гейне, ни более новая асимптотика из~\cite{Kuijlaars-Martinez-Finkelshtein-04} не применимы, когда $1/m\varepsilon$ или $x/\varepsilon$ неограничены. То же касается и \cite[предложение~2.2]{Sunada-Tate-12}, так как оценка остаточного члена там неравномерна по $m$ и $\varepsilon$. Наоборот, следующая теорема подходит для вычисления повторного предела, когда сначала $\varepsilon\to 0$, потом $x/t\to 0$, а потом $t\to+\infty$, но не тройного предела, так как остаточный член взрывается при $t\to+\infty$.

\smallskip
%



\textbf{Непрерывный предел.}
Предел при $\varepsilon\to 0$ 
описывается \emph{функциями Бесселя 1-го рода}: 
\begin{align*}
J_0(z)&:=\sum_{k=0}^\infty (-1)^k\frac{(z/2)^{2k}}{(k!)^2},
&
J_1(z)&:=\sum_{k=0}^\infty (-1)^k\frac{(z/2)^{2k+1}}{k!(k+1)!}.
\end{align*}

\begin{theorem}[Асимптотика в непрерывном пределе] \label{th-main}
Для любых $m,\varepsilon,\delta>0$ и $(x,t)\in\varepsilon\mathbb{Z}^2$ таких, что $(x+t)/\varepsilon$ четно, $t-|x|\ge\delta$, и $\varepsilon<
\delta \,e^{-3ms}/16$, где $s:=\sqrt{t^2-x^2}$, имеем
$$
{a}\left(x,t,{m},{\varepsilon}\right)
={m}{\varepsilon}\left(J_0(ms)
-i\,\frac{t+x}{s}\,J_1(ms)
+O\left(\frac{\varepsilon}{\delta}\log^2\frac{\delta}{\varepsilon}
\cdot e^{m^2t^2}\right)
\right).
$$
\end{theorem}

Напомним, что 
$f(x,t,{m},{\varepsilon})={O}\left(g(x,t,{m},{\varepsilon})\right)$ означает, что существует константа $C$ (не зависящая от $x,t,{m},{\varepsilon}$) такая, что для любых $x,t,{m},{\varepsilon}$, удовлетворяющих условиям теоремы, верно, что $|f(x,t,{m},{\varepsilon})|\le C\,g(x,t,{m},{\varepsilon})$.


Главный член в теореме~\ref{th-main} был вычислен в \cite[\S1]{Narlikar-72}. Численный эксперимент показывает, что остаточный член убывает быстрее, чем утверждается в теореме 
(см. таблицу~\ref{table-error}, вычисленную в \cite[\S14]{SU-2}).

В следующем следствии мы приближаем фиксированную точку плоскости  $(x,t)$ точками решетки  $\left(2\varepsilon\!\left\lceil \frac{x}{2\varepsilon}\right\rceil,2\varepsilon\!\left\lceil \frac{t}{2\varepsilon}\right\rceil\right)$ (см. рисунок~\ref{fig-limit}). Множители $2$ делают эти точки решетки достижимыми для шашки.

\begin{corollary}[Равномерный непрерывный предел; см.~ рис.~\ref{fig-limit}] \label{cor-uniform} Для любого фиксированного $m\ge 0$ имеем
\begin{equation}\label{eq-uniform-limit}
\frac{1}{2\varepsilon}\,{a}\left(2\varepsilon\!\left\lceil \frac{x}{2\varepsilon}\right\rceil,2\varepsilon\!\left\lceil \frac{t}{2\varepsilon}\right\rceil,{m},{\varepsilon}\right)
\rightrightarrows \frac{m}{2}\,J_0(m\sqrt{t^2-x^2})
-i\,\frac{m}{2}\sqrt{\frac{t+x}{t-x}}\,J_1(m\sqrt{t^2-x^2})
\end{equation}
при $\varepsilon\to 0$ равномерно на компактных подмножествах угла $|x|<t$.
\end{corollary}

Доказательство \emph{поточечной} сходимости проще и представлено в приложении~\ref{app-pointwise}.

\begin{corollary}[Концентрация меры]\label{cor-concentration} 
Для любых $t,m,\delta>0$ 
имеем
\begin{equation*}
\vspace{-0.8cm}
\end{equation*}
\begin{equation*}
  \sum_{x\in\varepsilon\mathbb{Z}\,:\,
  0\le t-|x|\le \delta} P(x,t,m,\varepsilon)\to 1\qquad\text{при }\quad
  \varepsilon\to 0
  \quad\text{ таких, что }\quad\frac{t}{2\varepsilon}\in \mathbb{Z}.
\end{equation*}
\end{corollary}

Этот результат, хотя и ожидаем, в литературе не обнаружен. Элементарное доказательство представлено в~\S\ref{ssec-proofs-continuum}. Отметим существенное различие между непрерывным пределом и пределом при больших временах: из следствия~\ref{th-limiting-distribution-mass} видим, что концентрации меры при  $t\to\infty$ для фиксированного $\varepsilon$ \emph{не наблюдается}.



\begin{table}[htb]
  \centering
  \begin{tabular}{ccc}
    \hline
    $\varepsilon$ & $5\varepsilon\log_{10}^2(5\varepsilon)$ &
    $\max\limits_{x\in (-0.8,0.8)\cap2\varepsilon\mathbb{Z}}
    \left|\frac{1}{2\varepsilon}a(x,1,10,\varepsilon)
    -G^R_{11}(x,1)-iG^R_{12}(x,1)\right|$
    \\
    \hline
    0.02 & 0.1 & 1.1 \\
    0.002 & 0.04 & 0.06 \\
    0.0002 & 0.009 &  0.006 \\
    \hline
  \end{tabular}
  \caption{
  Приближение запаздывающего пропагатора~\eqref{eq-relativistic-propagator} для спина $1/2$ с помощью шашек Фейнмана ($m=10$, $\delta=0.2$, $t=1$)}\label{table-error}
\end{table}

\addcontentsline{toc}{myshrinkalt}{}

\subsection{Физическая интерпретация} \label{ssec-mass-interpretation}

Обсудим смысл фейнмановского тройного предела и непрерывного предела. В этом разделе мы опустим некоторые технические определения, не используемые в дальнейшем.

Правая часть выражения~\eqref{eq-feynman-problem} с точностью до множителя $e^{-imt_n}$ есть \emph{ядро свободной частицы}
\begin{equation}\label{eq-free-particle-kernel}
K(x,t)=\sqrt{\frac{m}{2\pi t}}
\exp \left(\frac{i m x^2}{2t}-\frac{i\pi}{4}\right),
\end{equation}
описывающее движение по прямой свободной нерелятивистской частицы, испущенной из начала координат.

Предел~\eqref{eq-uniform-limit} воспроизводит запаздывающий пропагатор со спином $1/2$, описывающий движение электрона вдоль прямой.
Точнее, \emph{запаздывающий пропагатор со спином $1/2$}, или \emph{запаздывающая функция Грина для уравнения Дирака}~\eqref{eq-continuum-Dirac}, является матрично-значной обобщенной функцией медленного роста $G^R(x,t)=(G^R_{kl}(x,t))$ на $\mathbb{R}^2$, зануляющейся при $t<|x|$ и удовлетворяющей уравнению
\begin{equation}\label{eq-Green}
\begin{pmatrix}
m  & \partial/\partial x-\partial/\partial t \\
\partial/\partial x+\partial/\partial t & m
\end{pmatrix}
\begin{pmatrix}
G^R_{11}(x,t) & G^R_{12}(x,t) \\ G^R_{21}(x,t) & G^R_{22}(x,t)
\end{pmatrix}=\delta(x)\delta(t)\begin{pmatrix}
1 & 0 \\ 0 & 1
\end{pmatrix},
\end{equation}
где $\delta(x)$ --- дельта-функция Дирака. 
Пропагатор задаётся формулой (ср.~\cite[(13)]{Jacobson-Schulman-84}, \cite[(3.117)]{Peskin-Schroeder})
\begin{equation}\label{eq-relativistic-propagator}
G^R(x,t)=\frac{m}{2}\,
\begin{pmatrix}
J_0(m\sqrt{t^2-x^2}) &
-\sqrt{\frac{t+x}{t-x}}\,J_1(m\sqrt{t^2-x^2}) \\ \sqrt{\frac{t-x}{t+x}}\,J_1(m\sqrt{t^2-x^2}) &
J_0(m\sqrt{t^2-x^2})
\end{pmatrix}\qquad\text{при }|x|<t.
\end{equation}
Кроме того, $G^R(x,t)$ содержит обобщенную функцию с носителем на прямых $t=\pm x$, которая не наблюдается в пределе~\eqref{eq-uniform-limit}, и поэтому не указана здесь.
Более привычно выражение 
(ср.~предложение~\ref{cor-double-fourier})
\begin{equation}\label{eq-double-fourier-retarded}
G^R(x,t)=
  \frac{1}{4\pi^2}
  \int_{-\infty}^{+\infty}
  \int_{-\infty}^{+\infty}
  \lim_{\delta\to+0}
  \begin{pmatrix}
  m & -ip-i\omega \\
  -ip+i\omega & m
  \end{pmatrix}
  \frac{ e^{i p x-i\omega t}\,dpd\omega}
  {m^2+p^2-(\omega+i\delta)^2},
\end{equation}
где предел берется в смысле слабой топологии матрично-значных обобщенных функций медленного роста, а интеграл понимается, как преобразование Фурье обобщенных функций медленного роста (ср.~\cite[(6.47)]{Folland}).

''Квадрат'' пропагатора $G^R_{11}(x,t)^2+G^R_{12}(x,t)^2$ определен некорректно (так как содержит квадрат дельта-функции Дирака с носителем на прямых $t=\pm x$). Поэтому пропагатор не имеет вероятностной интерпретации, а глобальное сохранение заряда (предложение~\ref{p-mass2}) не имеет непрерывного аналога.
 Например, получаем следующий парадокс:
$\int_{(-t,t)} (G^R_{11}(x,t)^2+G^R_{12}(x,t)^2)\,dx\new{\approx m^2}t/2\ne\mathrm{const}$. Физическое объяснение: прямая $t=x$ несет бесконечный заряд, перетекающий внутрь угла $|x|<t$. Можно интерпретировать ''квадрат'' пропагатора при $|x|\ne t$ как относительную плотность вероятности или плотность заряда (см.~рис.~\ref{P-contour} \new{справа}).

В интерпретации через спиновую цепочку, пропагатор является пределом статистической суммы одномерной модели Изинга при обратной температуре $\beta=i\pi/4-\log(m\varepsilon)/2$. Это именно те значения $\beta$, для которых возможен фазовый переход~\cite{Matveev-Shrock-97}.

Нормировочный множитель $1/2\varepsilon$ перед ``${a}$'' в \eqref{eq-uniform-limit}
может быть объяснен, как деление на длину, приходящуюся на одну черную точку решетки на горизонтальной прямой.
На более глубоком уровне он объясняется нормировкой~$G^R(x,t)$, возникающей из~\eqref{eq-Green}.

Теорема~\ref{th-main}--- игрушечный результат в  \emph{алгоритмической квантовой теории}: она определяет шаг решетки, необходимый для вычисления пропагатора с заданной точностью. Пока это не имеет большого значения, ведь пропагатор и так имеет известное аналитическое выражение и, на самом деле, не может быть измерен экспериментально; эффективность алгоритма также не рассматривается. Но это только первый шаг.

\begin{algorithm}[Алгоритм приближения запаздывающего пропагатора~\eqref{eq-relativistic-propagator} для спина $1/2$] \label{alg-main}

Ввод: масса $m>0$, координаты $|x|<t$, уровень точности $\Delta$.

Вывод: приближенное значение $G_{kl}$ для $G_{kl}^R(x,t)$ на расстоянии не более $\Delta$ от истинного значения~\eqref{eq-relativistic-propagator}.

Алгоритм: вычислим $G_{kl}=\tfrac{(-1)^{(k-1)l}}{2\varepsilon}
\,{a}_{(k+l)\!\!\mod2+1}\left(2\varepsilon\!\left\lceil \tfrac{(-1)^{(k-1)l}x}{2\varepsilon}\right\rceil,2\varepsilon\!\left\lceil \tfrac{t}{2\varepsilon}\right\rceil,{m},{\varepsilon}\right)$
по формуле~\eqref{eq-def-mass} для
$$
\varepsilon=(t-|x|)\min\left\{
\frac{1}{16\, e^{3mt}},\left(\frac{\Delta}{9C\, m\,e^{m^2t^2}}\right)^3
\right\},\qquad\text{где }\quad C=100.
$$
\end{algorithm}

Здесь мы использовали явную оценку для константы $C$, подразумеваемой в обозначении О-большое в теореме~\ref{th-main}; она легко восстанавливается из доказательства теоремы. Теорема и оценка остаются верными, если в них ${a}\left(x,t,{m},{\varepsilon}\right)$ при $(x,t)\in\varepsilon\mathbb{Z}^2$ заменить на
${a}\left(2\varepsilon\!\left\lceil \tfrac{x}{2\varepsilon}\right\rceil,2\varepsilon\!\left\lceil \tfrac{t}{2\varepsilon}\right\rceil,{m},{\varepsilon}\right)$ при произвольных $(x,t)\in\mathbb{R}^2$.

\begin{remark}\label{rem-qw} Общее (однородное) \emph{одномерное квантовое блуждание} по существу сводится к уточнению модели, рассмотренному в этом разделе. А именно, рассмотрим уравнение
\begin{align*}
\psi_1(x,t+1)&=\hphantom{-}a\,\psi_1(x+1,t)+b \,\psi_2(x+1,t), \\
\psi_2(x,t+1)&=-\bar b\,\psi_1(x-1,t)+\bar a \,\psi_2(x-1,t);
\end{align*}
где $a,b\in\mathbb{C}-\{0\}$ и $|a|^2+|b|^2=1$, с начальным условием $\psi_1(x,0)=\psi_2(x,0)=0$ для каждого $x\ne 0$ \cite[\S1]{Sunada-Tate-12}. Прямой проверкой с использованием \eqref{eq-Dirac-mass1}--\eqref{eq-Dirac-mass2} устанавливается
общее решение:
\begin{align*}
  \psi_1(x,t)&=
  \left(\frac{a}{|a|}\right)^{-x}
  \left(
  \frac{|a|b}{a|b|}\,a_1\left(x+1,t+1,\left|\frac{b}{a}\right|,1\right)\psi_2(0,0)
  +a_2\left(1-x,t+1,\left|\frac{b}{a}\right|,1\right)\psi_1(0,0)
  \right),\\
  \psi_2(x,t)&=
  \left(\frac{a}{|a|}\right)^{-x}
  \left(
  a_2\left(x+1,t+1,\left|\frac{b}{a}\right|,1\right)\psi_2(0,0)
  -\frac{a|b|}{|a|b}\,a_1\left(1-x,t+1,\left|\frac{b}{a}\right|,1\right)\psi_1(0,0)
  \right).
\end{align*}
\end{remark}

\addcontentsline{toc}{myshrink}{}

\section{Спин} \red{\label{sec-spin}}

\addcontentsline{toc}{myshrink}{}

{
\hrule
\footnotesize
\noindent\textbf{Вопрос:} какова вероятность обнаружить правый электрон в $(x, t)$, если правый электрон был испущен из $(0,0)$?

\noindent\textbf{Предположения:} теперь учитывается  хиральность электрона.

\noindent\textbf{Результат:} вероятность смены хиральности.
\hrule
}
\bigskip


Особенностью модели является то, что \emph{спин} электрона возникает естественным образом, а не добавляется искусственно.

Почти напрашивается взгляд на электрон, как находящийся в одном из двух состояний в зависимости от направления последнего хода шашки: \emph{движущийся вправо} или \emph{движущийся влево} (или просто <<\emph{правый}>> или <<\emph{левый}>> для краткости).

\emph{Вероятность найти правый 
электрон в клетке $(x,t)$, если правый электрон был испущен из клетки $(0,0)$}, --- это квадрат длины вектора $\sum_s {a}(s)$, где сумма берется по всем путям из клетки $(0,0)$ в клетку $(x,t)$, которые начинаются и заканчиваются ходом вправо-вверх. 
\emph{Вероятность найти левый электрон} определяется аналогично, только сумма берется по путям, которые начинаются с хода вправо-вверх, но заканчиваются ходом влево-вверх. Очевидно, что эти вероятности равны $a_2(x, t)^2$ и $a_1(x, t)^2$ соответственно, потому что последний ход направлен вверх-вправо тогда и только тогда, когда число ходов четное.

Эти правый и левый электроны в точности являются $(1+1)$-мерным аналогом \emph{хиральных} состояний  для частиц со спином $1/2$ \cite[\S19.1]{Peskin-Schroeder}. Действительно, известно, что компоненты $a_2(x,t)$ и $a_1(x,t)$ в уравнении Дирака~\eqref{eq-continuum-Dirac} в базисе Вейля интерпретируются как волновые функции правых и левых частиц соответственно. Связь с направлением движения
становится очевидной для $m=0$: общее решение уравнения~\eqref{eq-continuum-Dirac} --- это $(a_2(x,t),a_1(x,t))=(a_2(x-t,0),a_1(x+t,0))$; таким образом, максимумы функций $a_2(x,t)$ и $a_1(x,t)$ (если таковые имеются) с ростом~$t$ перемещаются вправо и влево соответственно. 
Обратим внимание, что в $3$ или более измерениях, спин \emph{не является} направлением движения
и не может быть объяснен в неквантовых терминах.

Это дает более концептуальную интерпретацию модели:
результатом эксперимента является пара (конечная координата  $x$, направление последнего хода), тогда как конечная координата $t$ фиксирована.
Вероятности обнаружения правого/левого электрона 
являются фундаментальными. В последующих уточнениях модели $a_1(x,t)$ и $a_2(x,t)$  станут комплексными числами, и тогда $P(x,t)$ должно \emph{определяться} как 
$|{a}_1(x,t)|^2+|{a}_2(x,t)|^2$,
а не формуле $P(x,t)=|{a}(x,t)|^2=|{a}_1(x,t)+i{a}_2(x,t)|^2$ выше, которая является совпадением.

\begin{theorem}[Вероятность смены хиральности] \label{p-right-prob}
Для целых $t>0$ выполнено равенство
$\sum\limits_{x\in\mathbb{Z}}a_{1}(x,t)^2
=\frac{1}{2\sqrt{2}}+{O}\left(\frac{1}{\sqrt{t}}\right).$
\end{theorem}

См.~рисунок~\ref{spin-reversal} для иллюстрации и сравнения с уточнением из~\S\ref{sec-external field}. Физическая интерпретация теоремы ограничена: в непрерывной теории  вероятность смены хиральности (для электрона, испускаемого точечным источником) плохо определена аналогично ``квадрату'' пропагатора (см.~\S\ref{ssec-mass-interpretation}). Соответствующая более разумная величина изучается в~\cite[p.~381]{Jacobson-Schulman-84} (ср.~с~проблемой~\ref{p-correlation2}).
Недавно И.~Богданов обобщил теорему на любую массу и шаг решетки (см.~определение~\ref{def-mass}): если  $0\le m \varepsilon\le 1$, то
$\lim_{t\to+\infty,t\in\varepsilon\mathbb{Z}} \sum_{x\in\varepsilon\mathbb{Z}} a_1(x, t, m,\varepsilon)^2 = \frac{m\varepsilon}{2\sqrt{1+m^2\varepsilon^2}}$ \cite[теорема~2]{Bogdanov-20}.
Это подтвердило гипотезу И.~Гайдай-Турлова--Т.~Ковалева--А.~Львова.



\begin{figure}[htbp]
\begin{center}
\includegraphics[width=0.3\textwidth]{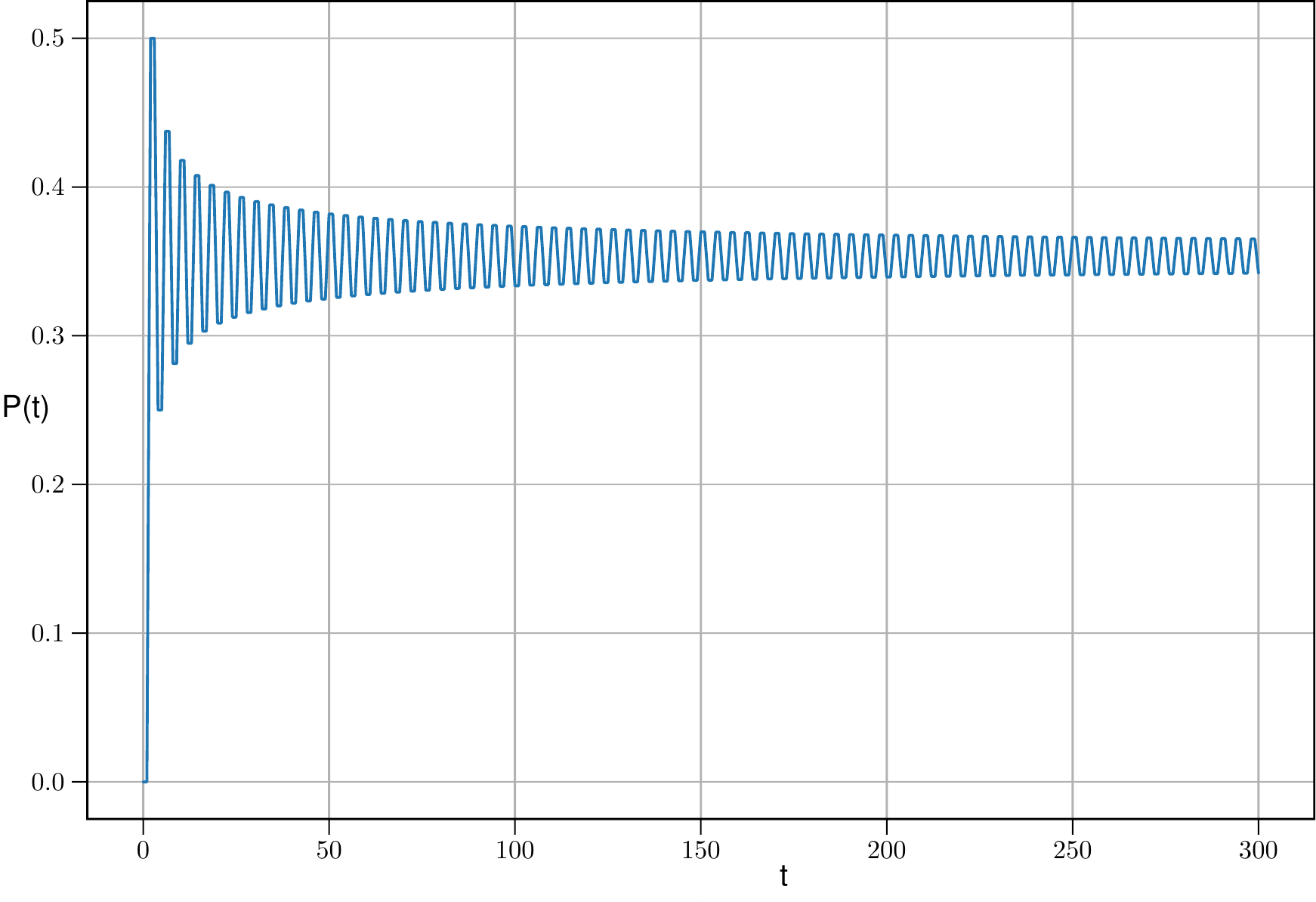}
\includegraphics[width=0.3\textwidth]{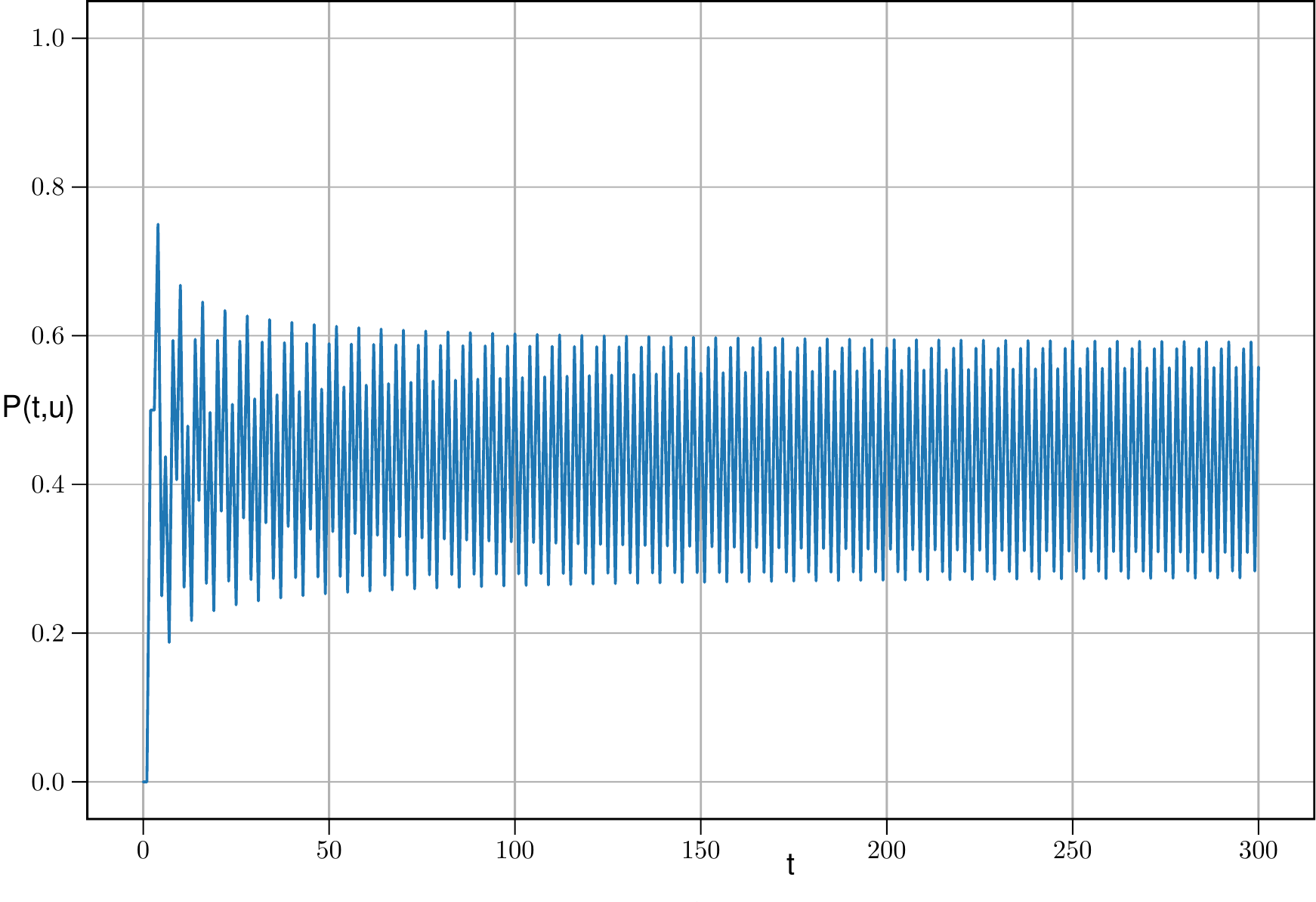}
\end{center}
\vspace{-0.8cm}
\caption{
Графики вероятностей $P(t)=\sum_{x\in\mathbb{Z}}a_{1}(x,t)^2$ и $P(t,u)=\sum_{x\in\mathbb{Z}}a_{1}(x,t,u)^2$
смены хиральности при выключенном и включенном электромагнитном поле соответственно}
\label{spin-reversal}
\end{figure}

\section{Внешнее поле (неоднородное квантовое блуждание)} \label{sec-external field}
\addcontentsline{toc}{myshrink}{}
{
\hrule
\small
\noindent\textbf{Вопрос:} какова вероятность обнаружения электрона в клетке $(x,t)$ , если он двигается в заданном электромагнитном поле $u$?

\noindent\textbf{Предположения:} 
поле есть только в плоскости $Oxt$, электрон не воздействует на него.

\noindent\textbf{Результаты:} 'прецессия' спина в магнитном поле(качественное объяснение), сохранение заряда.
\hrule
}

\bigskip


{
\begin{wrapfigure}{r}{6.2cm}
\vspace{-0.4cm}
\includegraphics[width=6cm]{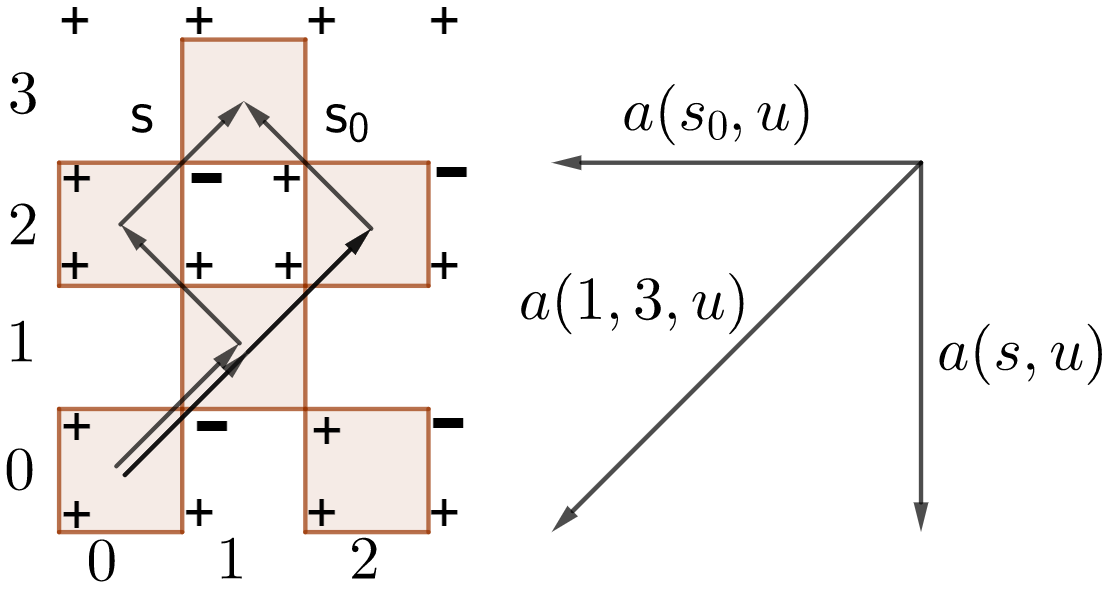}
\vspace{-0.2cm}
\caption{Пути в поле\new{}}
\label{homogeneous-field}
\vspace{-0.2cm}
\end{wrapfigure}
В нашей модели внешнее поле не добавляется искусственно, а возникает естественным образом. Мы начнём с неформального определения, затем дадим точное, и, наконец, выведем сохранение заряда.

В базовой модели стрелка часов не поворачивалась, пока шашка двигалась прямо. Напрашивается изменение модели: поворачивать стрелку равномерно во время движения. Это не изменит модель по сути: поскольку все пути из начального положения в конечное имеют одинаковую длину, то и векторы поворачиваются на один и тот же угол, что не влияет на вероятности. Более интересная модификация получается, когда текущий угол поворота зависит от положения шашки.
Именно это делает электромагнитное поле. В дальнейшем, для простоты, угол поворота будет принимать только два значения $0^\circ$ и $180^\circ$, что означает домножение на $\pm 1$.

}
Таким образом, электромагнитное поле понимается, как фиксированная расстановка $u$ чисел $ \pm 1 $ в вершинах квадратов-клеток. Например, на рисунке~\ref{homogeneous-field}, поле равно $-1$ в правой верхней вершине каждой клетки $(x,t)$, у которой $x$ и $t$ чётны. Модифицируем определение вектора ${a}(s)$, меняя его направление на противоположное каждый раз, когда путь проходит через вершину с полем $-1$. Обозначим получившийся вектор как ${a}(s,u)$. Определим ${a}(x,t,u)$ и $P(x,t,u)$ аналогично ${a}(x,t)$ и $P(x,t)$, заменяя в их определении ${a}(s)$ на ${a}(s,u)$. Например, если $u$ тождественно равно $+1$, то $P(x,t,u)=P(x,t)$.

Давайте слегка переформулируем эту конструкцию, чтобы сделать её связь с решёточной калибровочной теорией более прозрачной. Введём вспомогательную квадратную решётку с вершинами в центрах чёрных клеток (см. рис.~\ref{fig-young} справа) --- это тот граф, по которому фактически движется шашка.

\begin{definition} \label{def-external}
\emph{Вспомогательным ребром} назовём отрезок, соединяющий две ближайших целых точки с чётной суммой координат. Пусть $u$ --- любое отображение из множества всех вспомогательных ребер в  $\{+1,-1\}$. Обозначим
$$
{a}(x,t,u):=2^{(1-t)/2}\,i\,\sum_s (-i)^{t(s)}u(s_0s_1)u(s_1s_2)\dots u(s_{t-1}s_t)
$$
сумму по всем путям $s=(s_0,s_1,\dots,s_t)$ шашки, таким что $s_0=(0,0)$, $s_1=(1,1)$,  $s_t=(x,t)$. Положим $P(x,t,u)=|{a}(x,t,u)|^2$. Аналогично определим $a_1(x,t,u)$ и $a_2(x,t,u)$, только добавим условия $s_{t-1} = (x + 1, t - 1)$ и $s_{t-1}=(x - 1, t-1)$ соответственно. Для полуцелых $x,t$ обозначим через $u(x,t)$ значение $u$ на вспомогательном ребре с серединой $(x,t)$.
\end{definition}

\begin{remark} \label{rem-external}
Здесь поле $u$ --- фиксированное классическое внешнее поле, электрон не воздействует на него.

Это определение аналогично исторически первому построению калибровочной теории, принадлежащему Вейлю--Фоку--Лондону, и даёт связь шашек Фейнмана с решёточной $\mathbb{Z}/2\mathbb{Z}$-калибровочной теорией Вегнера--Вильсона. В частности, оно даёт правильный спин $1$ для электромагнитного поля: функция на множестве рёбер - это дискретный аналог \emph{векторного} поля, т.е. поля \emph{спина} $1$. Несмотря на то, что этот способ связи с калибровочным полем является классическим, для шашек Фейнмана он ранее явно не применялся (ср. с ~\cite[стр.~36]{Schulman-Gaveau-89}); он сильно отличается от подхода в \cite{Ord} и интуиции диаграмм Фейнмана \cite{Feynman}.

Например, поле $u$ на рис.~\ref{homogeneous-field} имеет вид $u(s_1s_2)=\exp\left(-i\int_{s_1}^{s_2}(A_0\,dt+A_1\,dx)\right)$ для каждого вспомогательного ребра $s_1s_2$, где $(A_0,A_1):=(\pi/2)(x+1/2,x+1/2)$ --- вектор-потенциал постоянного однородного электромагнитного поля.

Для произвольной калибровочной группы числа $a_1(x,t,u)$ и $a_2(x,t,u)$ определяются аналогично, только $u$ становится отображением из множества вспомогательных ребер в некоторую матричную группу, скажем, $U(1)$ или $SU(n)$.  Тогда полагаем по определению $P(x,t,u):=\sum_{k} \left(|(a_1(x,t,u))_{k1}|^2 + |(a_2(x,t,u))_{k1}|^2\right)$, где $(a_j)_{kl}$ --- элементы матрицы ~$a_j$.
\end{remark}

\begin{example}[``Прецессия'' спина в магнитном поле] \label{p-precession}
  Положим $u(x+1/2,t+1/2)=-1$, если $x$ и $t$ чётные, и $u(x+1/2,t+1/2)=+1$ иначе (``\emph{однородное магнитное поле}'', см. рисунок~\ref{homogeneous-field}). 
  Тогда график вероятности обнаружения левого электрона $P(t,u)=\sum_{x\in\mathbb{Z}}a_1(x,t,u)^2$(см.~\S\ref{sec-spin}) изображен на рисунке~\ref{spin-reversal} справа. Видно, что он стремится к "периодическому режиму" при $t\to\infty$ (см. проблему~\ref{p-precession-prove}).
\end{example}

Следующие предложения доказываются аналогично предложениям~\ref{p-mass}--\ref{p-mass2}, только добавляется множитель  $u(x \pm \tfrac{1}{2},t+\tfrac{1}{2})$, так как последний ход пути проходит через вершину $(x \pm \tfrac{1}{2},t-\tfrac{1}{2})$.

\begin{proposition}[Уравнение Дирака в электромагнитном поле]\label{p-Dirac-external}  Для любого целого $x$ и любого целого $t\ge 1$ выполнено 
\begin{align*}
a_1(x,t+1,u) &= \frac{1}{\sqrt{2}} u\left(x+\frac{1}{2},t+\frac{1}{2}\right)(a_1(x+1,t,u) + a_2(x+1,t,u)),\\
a_2(x,t+1,u) &= \frac{1}{\sqrt{2}} u\left(x-\frac{1}{2},t+\frac{1}{2}\right)(a_2(x-1,t,u) - a_1(x-1,t,u)).
\end{align*}
\end{proposition}

\begin{proposition}[Сохранение вероятности/заряда] \label{p-probabilitt-conservation-external} Для любого целого $t\ge 1$ выполнено $\sum\limits_{x\in\mathbb{Z}}P(x,t,u)=1$.
\end{proposition}

Недавно Ф.~Ожегов нашел аналоги ``явной'' формулы (предложения~\ref{p-mass3}) и непрерывного предела (следствия~\ref{cor-uniform}) для ``однородного поля'' из примера~\ref{p-precession} \cite[теоремы 3--4]{Ozhegov-21}.

\section{Источник
}\label{sec-source}

\addcontentsline{toc}{myshrink}{}

{
\hrule
\footnotesize
\noindent\textbf{Вопрос:} какова вероятность обнаружения в $(x,t)$ электрона, испущенного источником с длиной волны~$\lambda$?

\noindent\textbf{Предположения:} теперь источник реалистичный.

\noindent\textbf{Результаты:}
распространение волн, дисперсионное соотношение.
\hrule
}

\bigskip

Реалистичный источник испускает волны, а не электроны, локализованные в одной точке $x=0$ (как в базовой модели). Это означает решение уравнения Дирака~\eqref{eq-Dirac-mass1}--\eqref{eq-Dirac-mass2} c \emph{(квази-)периодическими начальными условиями}. 

Для формулировки результата удобно переписать уравнение Дирака ~\eqref{eq-Dirac-mass1}--\eqref{eq-Dirac-mass2}, используя обозначения
$$
\tilde a_1(x,t)=a_1(x,t+\varepsilon,m, \varepsilon), \qquad
\tilde a_2(x,t)=a_2(x+\varepsilon,t+\varepsilon,m, \varepsilon),
$$
чтобы оно приняло форму 
\begin{align}\label{eq-Dirac-source1}
\tilde a_1(x,t) &= \frac{1}{\sqrt{1+m^2\varepsilon^2}}
(\tilde a_1(x+\varepsilon,t-\varepsilon)
+ m \varepsilon\, \tilde a_2(x,t-\varepsilon)),\\
\label{eq-Dirac-source2}
\tilde a_2(x,t) &= \frac{1}{\sqrt{1+m^2\varepsilon^2}}
(\tilde a_2(x-\varepsilon,t-\varepsilon)
- m \varepsilon\, \tilde a_1(x,t-\varepsilon)).
\end{align}
Следующее предложение доказывается прямой проверкой (доступной в \cite[\S12]{SU-2}).

\begin{proposition}[Распространение волн, дисперсионное соотношение]
Уравнения ~\eqref{eq-Dirac-source1}--\eqref{eq-Dirac-source2} с начальными условиями
\begin{align*}
&\\[-1.0cm]
 \tilde a_1(x,0)&= \tilde a_1(0,0)e^{2\pi ix/\lambda},\\
 \tilde a_2(x,0)&= \tilde a_2(0,0)e^{2\pi ix/\lambda};\\[-1.0cm]
\end{align*}
имеют единственное решение
\begin{align}\label{eq-solution1}
\tilde a_1(x,t)&= \hphantom{i}a\cos\tfrac{\alpha}{2}\, e^{2\pi i(x/\lambda+t/T)} +
\hphantom{i}b\sin\tfrac{\alpha}{2}\,e^{2\pi i(x/\lambda-t/T)},\\
         \label{eq-solution2}
\tilde a_2(x,t)&= i a\sin\tfrac{\alpha}{2}\, e^{2\pi i(x/\lambda+t/T)}
-ib\cos\tfrac{\alpha}{2}\, e^{2\pi i(x/\lambda-t/T)},
\end{align}
где числа $T\ge 2$, $\alpha\in [0,\pi]$, и $a,b\in\mathbb{C}$ задаются соотношениями
\begin{align*}
\cos (2\pi\varepsilon/T)
&=\frac{\cos (2\pi\varepsilon/\lambda)}{\sqrt{1+m^2\varepsilon^2}},&
\cot \alpha
&=\frac{\sin (2\pi\varepsilon/\lambda)}{m\varepsilon}, &
\begin{aligned}
a&= \tilde a_1(0,0)\cos\tfrac{\alpha}{2}-i\tilde a_2(0,0)\sin\tfrac{\alpha}{2},\\
b&= \tilde a_1(0,0)\sin\tfrac{\alpha}{2}+i\tilde a_2(0,0)\cos\tfrac{\alpha}{2}.
\end{aligned}
\end{align*}
\end{proposition}

\begin{remark}
Чтобы получить решения непрерывного уравнения Дирака ~\eqref{eq-continuum-Dirac}, достаточно в выражении ~\eqref{eq-solution1}--\eqref{eq-solution2} переопределить $2\pi/T$ и $\alpha$ с помощью соотношений $\tfrac{4\pi^2}{T^2}=\tfrac{4\pi^2}{\lambda^2}+m^2$ и $\cot\alpha=2\pi/m\lambda$ соответственно. В непрерывной (соответственно, дискретной) постановке задачи эти величины --- гипотенуза и угол прямоугольного треугольника с катетами $2\pi/\lambda$ и $m$ (соответственно, $(\arctan m\varepsilon)/\varepsilon$), лежащего на плоскости (соответственно, сфере радиуса $1/\varepsilon$). Это новая  неожиданная интерпретация в терминах сферической геометрии.

\new{Физически, \eqref{eq-solution1}--\eqref{eq-solution2}
описывают волну с периодом $T$ и длиной волны $\lambda$. Соотношение между $T$ и $\lambda$ называется \emph{дисперсионным соотношением}. \emph{Соотношения Планка} и \emph{де Бройля} утверждают, что $E:=2\pi \hbar/T$ и $p:=2\pi \hbar/\lambda$ --- энергия и импульс волны (напомним, что
$\hbar=c=1$ в наших единицах). Дисперсионное соотношение выше
переходит в \emph{формулу Эйнштейна} $E=mc^2$ в пределе при $\varepsilon\to 0$ и $\lambda\to\infty$.}

Комментарий для специалистов: замена $a$ и $b$ операторами рождения и уничтожения, т.е.~вторичное квантование решеточного уравнения Дирака, ведет к модели из~\S\ref{sec-creation}.
\end{remark}

Для следующих уточнений модели мы только анонсируем результаты, которые будут обсуждаться в последующих публикациях.

\section{Среда}\label{sec-medium}

\addcontentsline{toc}{myshrink}{}

{
\hrule

\footnotesize
\noindent\textbf{Вопрос:} какая доля света данного цвета отражается стеклянной пластиной данной ширины?

\noindent\textbf{Предположения:} угол падения~--- $90^\circ$, нет поляризации и дисперсии света; масса частицы теперь зависит от~$x$, но не от цвета.

\noindent\textbf{Результаты:} отражение от тонкой плёнки (количественное объяснение). 

\hrule
}

\bigskip

Шашки Фейнмана можно применить для описания распространения света в прозрачной среде, например в стекле.
Свет распространяется так, как если бы он приобрел  внутри среды некоторую ненулевую массу и потенциальную энергию (зависящие от показателя преломления), в то время снаружи  масса и потенциальная энергия остаются равными нулю.
В целом, модель непригодна для описания света, 
частичное отражение является лишь замечательным исключением.
Обратим внимание, что аналогичные классические явления описываются квантовыми моделями, см.~\cite[\S2.7]{Venegas-Andraca-12}.

Как мы покажем в последующей публикации, шашки Фейнмана позволяют дать строгий вывод следующей известной формулы для доли $P$ света с длиной волны $\lambda$, отраженного от прозрачной пластины шириной $L$ и показателем преломления~$n$:
$$
P=\frac{(n^2-1)^2}{(n^2+1)^2+4n^2\cot^2(2\pi Ln/\lambda)}.
$$

Это делает научно-популярное обсуждение частичного отражения, изложенное Фейнманом в \cite{Feynman}, полностью строгим и показывает, что его модель имеет экспериментально подтвержденные предсказания в реальном мире, а не только в двумерном.

\section{Тождественные частицы} \label{sec-IdParAnt}

\addcontentsline{toc}{myshrink}{}

{
\hrule
\footnotesize
\noindent\textbf{Вопрос:} какова вероятность обнаружения электронов 
в клетках
$F$ и $F'$,
	если они 
испущены из $A$ и~$A'$?

\noindent\textbf{Предположения:} электронов несколько.

\noindent\textbf{Результаты:} принцип запрета Паули, локальность, сохранение заряда.
\hrule
}

\bigskip

Здесь мы анонсируем простое обобщение модели шашек Фейнмана, которое позволяет описывать движение нескольких электронов с соблюдением принципов запрета, локальности и сохранения вероятности (ср.~\cite[\S4.2]{Yepez-05}).

\begin{definition}\label{def-identical}
Зафиксируем целые точки $A=(0,0)$, $A'=(x_0,0)$, \new{}$F=(x,t)$, \new{}$F'=(x',t)$ и их диагональных соседей $B=(1,1)$, $B'=(x_0+1,1)$,
\new{}$E=(x-1,t-1)$, \new{}$E'=(x'-1,t-1)$, где $x_0\ne 0$ и $x'\ge x$.
Положим 
$$
{a}(AB,A'B'\to EF,E'F'):=
\sum_{\substack{s:AB\to EF\\s':A'B'\to E'F'}} {a}(s){a}(s')-
\sum_{\substack{s:AB\to E'F'\\s':A'B'\to EF}} {a}(s){a}(s'),
$$
где первое суммирование проходит по всем парам, состоящим из пути  шашки $s$, начинающегося с хода $AB$ и заканчивающегося ходом $EF$, и пути $s'$, начинающегося с хода $A'B'$ и заканчивающегося ходом $E'F'$, а во втором суммировании завершающие ходы меняются местами.

Величина $P(AB,A'B'\to EF,E'F'):=\left|{a}(AB,A'B'\to EF,E'F')\right|^2$ называется \emph{вероятностью 
найти правые электроны в $F$ и $F'$, если они испущены из $A$ и~$A'$}.
Аналогично величина $P(AB,A'B'\to EF,E'F')$ определяется и в случае, когда \new{}$E=(x\pm 1,t-1)$, \new{}$E'=(x'\pm 1,t-1)$. В случае одинаковых знаков здесь необходимо дополнительно наложить условие $x'\ge x$, в случае разных знаков значения $x'$ и $x$ могут быть произвольными.
\end{definition}




\section{Античастицы} \label{sec-creation}


{
\hrule
\footnotesize
\noindent\textbf{Вопрос:} каково матожидание заряда в клетке $(x,t)$, если электрон был испущен из клетки   $(0,0)$?

\noindent\textbf{Предположения:} 
электрон-позитронные пары теперь могут рождаться и аннигилировать;  ось $t$~--- время.

\noindent\textbf{Результаты:} 
фейнмановский пропагатор для спина $1/2$ в непрерывном пределе, аналитическое выражение для предела при больших временах.
\hrule
}

\bigskip





\addcontentsline{toc}{myshrinkalt}{}

\subsection{Тождества и асимптотики}

Здесь мы предлагаем совершенно новое уточнение модели шашек Фейнмана~---\emph{фейнман\-ов\-ские антишашки}. В этой модели помимо движения исходного электрона учитывается возможность рождения и аннигиляции электрон-позитронных пар.
Уточнение заключается в том, что в интеграле Фурье (см.~предложение~\ref{cor-fourier-integral}) величине $(x+t)/\varepsilon$ теперь мы разрешаем теперь принимать нечётные значения. Таким образом, в дополнение к чёрным клеткам шахматной доски мы вычисляем значения того же интеграла и в белых клетках.
Это эквивалентно вторичному квантованию дискретного уравнения Дирака~\eqref{eq-Dirac-source1}--\eqref{eq-Dirac-source2}, которое нам не потребуется производить (ср.~\cite[\S9F]{Bender-etal-94} и \cite[\S IV]{Bender-etal-85} для безмассового случая).
Ведь истинная мотивировка новой модели --- в замечательной согласованности с исходной и в появлении фейнмановского пропагатора~\eqref{eq-feynman-propagator} для спина $1/2$ в непрерывном пределе (см.~рисунок~\ref{fig-approx-b}). Мы также даем комбинаторное определение (см.~определение~\ref{def-anti-combi}).

\begin{figure}[htbp]
  \centering
\includegraphics[width=0.24\textwidth]{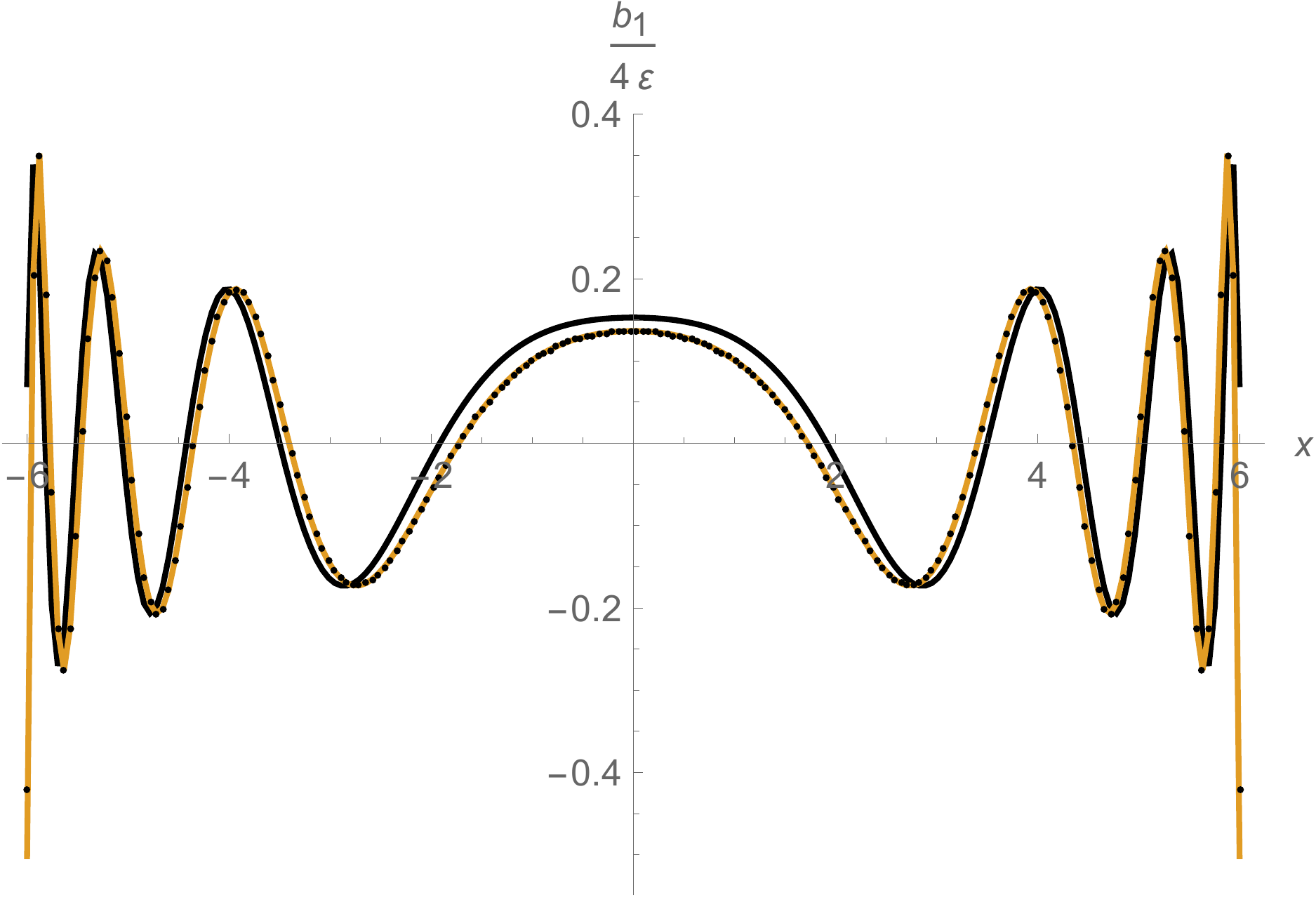}
\includegraphics[width=0.24\textwidth]{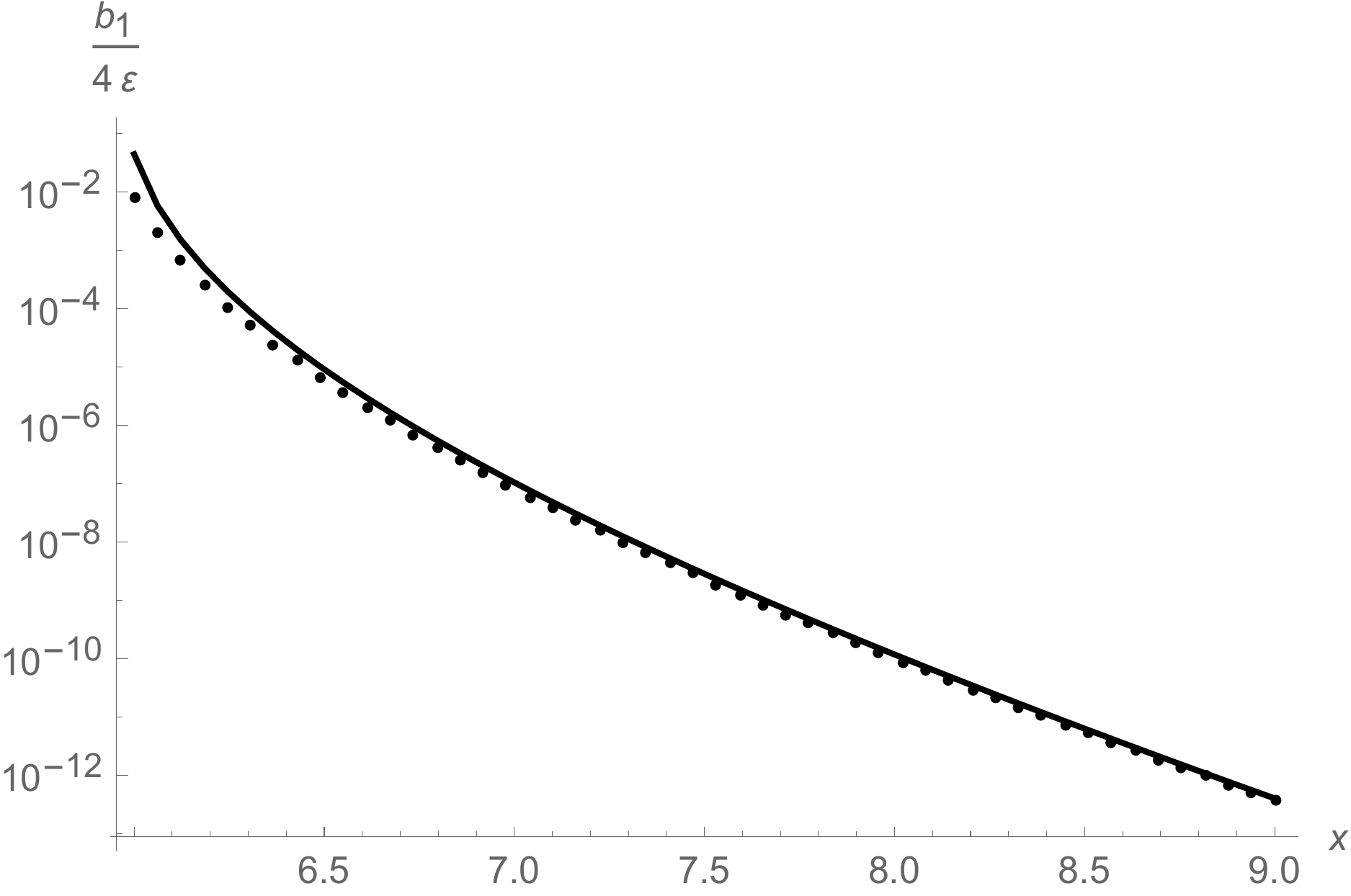}
\includegraphics[width=0.24\textwidth]{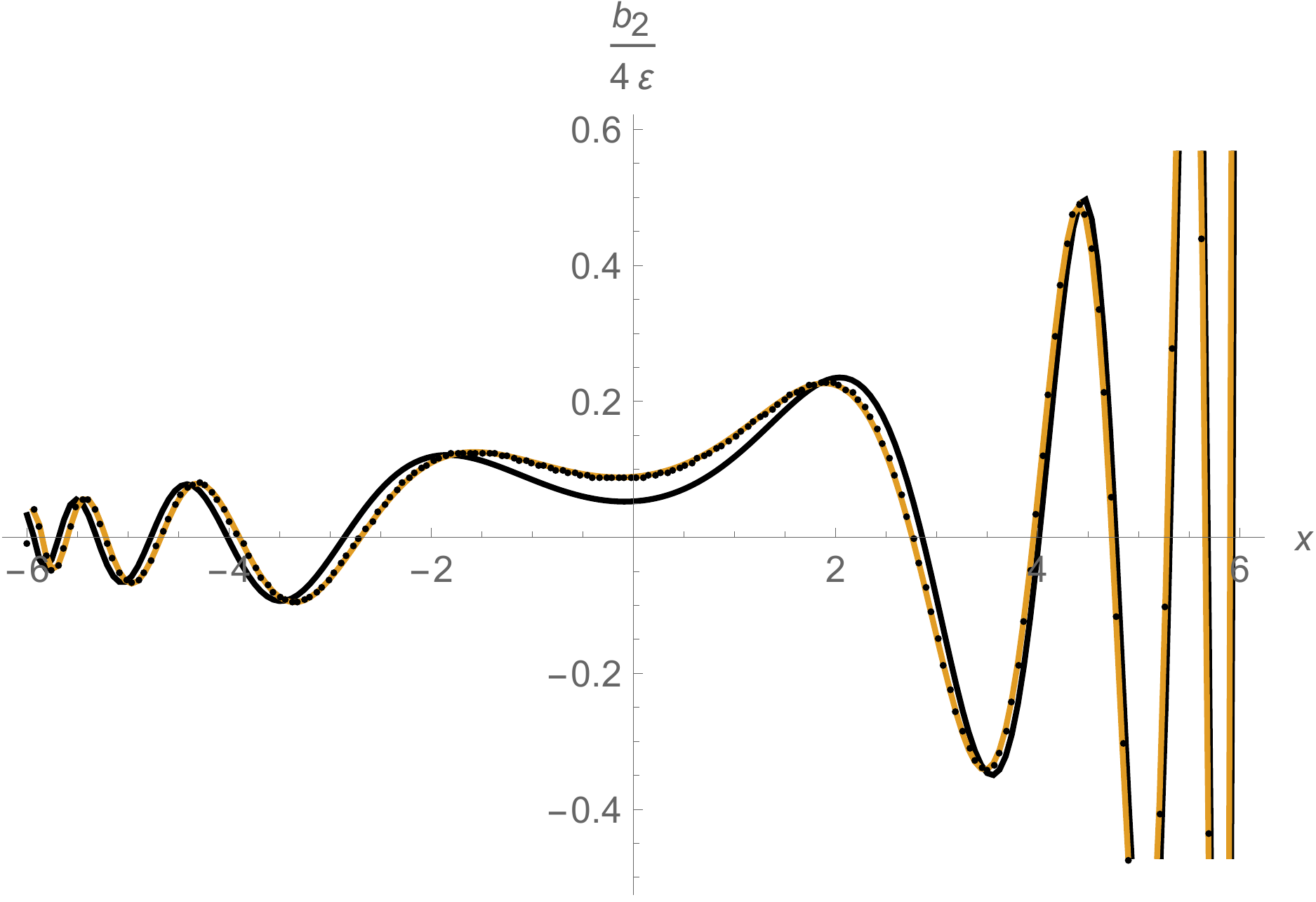}
\includegraphics[width=0.24\textwidth]{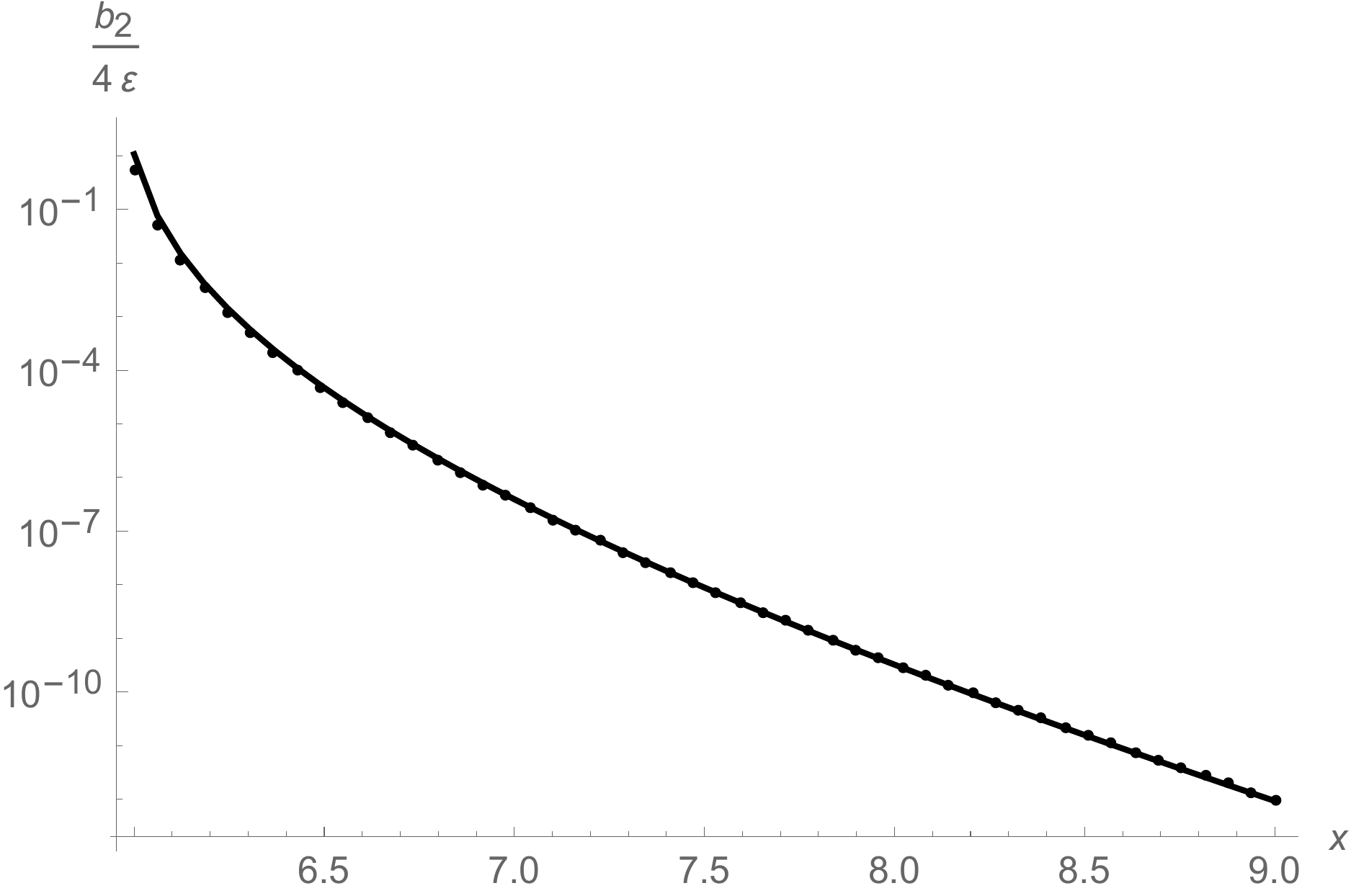}
  \caption{Графики функций $b_1(x,6,4,0.03)/0.12$ (точки слева),  $b_2(x,6,4,0.03)/0.12$ (точки справа), их аналитического приближения из теоремы~\ref{th-anti-ergenium} (светлая кривая), а также мнимой части фейнмановского пропагатора
  $\mathrm{Im}\,G^F_{11}(x,6)$ (темная кривая слева) и $\mathrm{Im}\,G^F_{12}(x,6)$ (темная кривая справа), задаваемого равенством~\eqref{eq-feynman-propagator} при $m=4$ и $t=6$.} 
  \label{fig-approx-b}
\end{figure}

\begin{definition} \label{def-anti} (См.~Предложение~\ref{cor-fourier-integral} и рис.~\ref{fig-approx-b}.)
Зафиксируем параметры $m,\varepsilon>0$. Для всех $(x,t)\in \varepsilon\mathbb{Z}^2$ положим $\omega_p:=\frac{1}{\varepsilon}\arccos(\frac{\cos p\varepsilon}{\sqrt{1+m^2\varepsilon^2}})$ и
\begin{equation}
  \label{eq-def-anti1}
  \begin{aligned}
  A_1(x,t,m,\varepsilon)&:=
  \pm\frac{im\varepsilon^2}{2\pi}
  \int_{-\pi/\varepsilon}^{\pi/\varepsilon}
  \frac{e^{i p x-i\omega_p(t-\varepsilon)}\,dp}
  {\sqrt{m^2\varepsilon^2+\sin^2(p\varepsilon)}};\\
  A_2(x,t,m,\varepsilon)&:=
  \pm\frac{\varepsilon}{2\pi}\int_{-\pi/\varepsilon}^{\pi/\varepsilon}
  \left(1+
  \frac{\sin (p\varepsilon)} {\sqrt{m^2\varepsilon^2+\sin^2(p\varepsilon)}}\right) e^{ip(x-\varepsilon)-i\omega_p(t-\varepsilon)}\,dp,
  \end{aligned}
\end{equation}
где в обоих случаях знак минус выбирается, если $t\le 0$ и $(x+t)/\varepsilon$ чётно. При $m=0$ определим $A_2(x,t,m,\varepsilon)$ той же формулой и положим $A_1(x,t,0,\varepsilon):=0$. В частности, $A_k(x,t,m,\varepsilon)=a_k(x,t,m,\varepsilon)$, если $(x+t)/\varepsilon$ чётно, $t>0$ и $k=1,2$.  Положим по определению $A_k(x,t,m,\varepsilon)=:ib_k(x,t,m,\varepsilon)$ при нечётном 
$(x+t)/\varepsilon$. Положим также $b_k(x,t,m,\varepsilon):=0$, если $(x+t)/\varepsilon$ чётно.
\end{definition}


Можно показать, что числа $A_k(x,t,m,\varepsilon)$ чисто мнимые при нечётном $(x+t)/\varepsilon$. Таким образом, действительная и мнимая части ``живут'' на черных и белых клетках соответственно, аналогично тому, как определяются дискретные аналитические функции (см.~\cite{Chelkak-Smirnov-11}). Соглашение о знаке для $t\le 0$ продиктовано аналогией с ``непрерывной'' теорией (см.~\eqref{eq-def-anti2} и~\eqref{eq-double-fourier-feynman}).

\begin{example} \label{ex-simplest-values} Значение $b_1(0,1,1,1)=\Gamma(\tfrac{1}{4})^2/(2\pi)^{3/2}=
\tfrac{2}{\pi}K(i)=:G\approx 0.83463$ представляет собой константу Гаусса,
а $-b_2(0,1,1,1)=2\sqrt{2\pi}/\Gamma(\tfrac{1}{4})^2
=\tfrac{2}{\pi}(E(i)-K(i))=1/\pi G=:L'\approx 0.38138$~--- обратная величина к лемнискатной константе $L$. Здесь функции $K(z)$ и $E(z)$ --- полные эллиптические интегралы 1-го и 2-го рода соответственно (см.~\textup{\cite[\S6.1]{Finch-03}}).
\end{example}

Остальные значения являются более сложными иррациональностями (см. таблицу~\ref{table-a4}).

\begin{table}[ht]
  \centering
  $b_1(x,t,1,1)$\\[0.1cm]
\begin{tabular}{|c*{6}{|>{\centering\arraybackslash}p{55pt}}|}
\hline
$2$& $\frac{G-L'}{\sqrt{2}}$ & & $\frac{G-L'}{\sqrt{2}}$ & & $\frac{7G-15L'}{3\sqrt{2}}$ \\
\hline
$1$& & $G$ & & $G-2L'$ & \\
\hline
$0$& $\frac{G-L'}{\sqrt{2}}$ & & $\frac{G-L'}{\sqrt{2}}$ & & $\frac{7G-15L'}{3\sqrt{2}}$ \\
\hline
$-1$& & $-L'$ & & $\frac{2G-3L'}{3}$ & \\
\hline
\diagbox[dir=SW,height=18pt]{$t$}{$x$}&$-1$&$0$&$1$&$2$&$3$\\
\hline
\end{tabular}
%
\\[0.2cm]
   $b_2(x,t,1,1)$\\[0.1cm]
\begin{tabular}{|c*{6}{|>{\centering\arraybackslash}p{55pt}}|}
\hline
$2$& $\frac{G-3L'}{3\sqrt{2}}$ & & $\frac{-G-L'}{\sqrt{2}}$ & & $\frac{-G+3L'}{\sqrt{2}}$ \\
\hline
$1$& & $-L'$ & & $L'$ & \\
\hline
$0$& $\frac{G-3L'}{\sqrt{2}}$ & & $\frac{G+L'}{\sqrt{2}}$ & & $\frac{-G+3L'}{3\sqrt{2}}$ \\
\hline
$-1$& & $G$ & & $\frac{G}{3}$ & \\
\hline
\diagbox[dir=SW,height=18pt]{$t$}{$x$}&$-1$&$0$&$1$&$2$&$3$\\
\hline
\end{tabular}
\qquad
  \caption{Значения $b_1(x,t,1,1)$ и $b_2(x,t,1,1)$ для малых $x,t$ (см. определение~\ref{def-anti} и пример~\ref{ex-simplest-values})}
  \label{table-a4}
\end{table}

Мы анонсируем, что аналоги предложений~\ref{p-mass}--\ref{p-symmetry-mass} и~~\ref{p-equal-time} остаются верными буквально, если $a_1$ и $a_2$ заменить на $b_1$ и $b_2$ соответственно (при этом предположение $t>0$ можно отбросить).
В качестве следствия получается, что $2^{(t-1)/2}b_1(x,t,1,1)$ и $2^{(t-1)/2}b_2(x,t,1,1)$ для всех $(x,t)\in \mathbb{Z}^2$ являются  линейными комбинациями с рациональными коэффициентами 
постоянной Гаусса~$G$ и обратной лемнискатной константы~$L'$ .
В общем случае мы анонсируем, что значения $b_1$ и $b_2$ можно ``явно'' выразить через гипергеометрическую функцию Гаусса: для $m,\varepsilon>0$, $(x,t)\in\varepsilon\mathbb{Z}^2$ с нечётным $(x+t)/\varepsilon$ выполнено
\begin{align*}
b_1(x,t,m,\varepsilon)
&=\left({1+m^2\varepsilon^2}\right)^{\frac{1}{2}-\frac{t}{2\varepsilon}}
(-m^2\varepsilon^2)^{\frac{t-|x|}{2\varepsilon}-\frac{1}{2}}
\binom{\frac{t+|x|}{2\varepsilon}-1}{|x|/\varepsilon}
\\&
\cdot{}_2F_1\left(1 + \frac{|x|-t}{2\varepsilon}, 1 + \frac{|x|-t}{2\varepsilon}; 1+\frac{|x|}{\varepsilon}; -\frac{1}{m^2\varepsilon^2}\right),\\
b_2(x+\varepsilon,t+\varepsilon,m,\varepsilon)
&=\left({1+m^2\varepsilon^2}\right)^{-\frac{t}{2\varepsilon}}
(m\varepsilon)^{\frac{t-|x|}{\varepsilon}}
(-1)^{\frac{t-|x|}{2\varepsilon}+\frac{1}{2}}
\binom{\frac{t+|x|}{2\varepsilon}-1+\theta(x)}{|x|/\varepsilon}
\,{}
\\&\cdot{}_2F_1\left(\frac{|x|-t}{2\varepsilon}, 1+\frac{|x|-t}{2\varepsilon}; 1+\frac{|x|}{\varepsilon}; -\frac{1}{m^2\varepsilon^2}\right),\qquad
\text{где }
\theta(x):=\begin{cases}
                            1, & \mbox{при } x\ge0, \\
                            0, & \mbox{при } x<0;
                          \end{cases}
\end{align*}
и ${}_2F_{1}(p,q;r;z)$~--- главная ветвь гипергеометрической функции. Доказательство проводится индукцией по $t/\varepsilon\ge 1$: база проверяется с помощью соотношений \new{} \cite[9.112, 
9.131.1, 9.134.3, \red{9.137.15}]{Gradstein-Ryzhik-63}, а шаг использует  аналоги равенств~\eqref{eq-Dirac-mass1}--\eqref{eq-Dirac-mass2} для $b_1$ и $b_2$, а также \new{формулы} \cite[9.137.\red{11,12,18}]{Gradstein-Ryzhik-63}.

\begin{remark} \label{rem-jacobi-functions} (Ср. с~замечанием~\ref{rem-hypergeo}) Эти выражения могут быть переписаны как \emph{функции Якоби второго рода} полуцелого порядка (см. определение в \cite[(4.61.1)]{Szego-39}). Например, для всех $(x,t)\in\varepsilon\mathbb{Z}^2$ таких, что $|x|>t$ и $(x+t)/\varepsilon$ нечётно, выполняется равенство
$$
b_1(x,t,m,\varepsilon)
=\frac{2m\varepsilon}{\pi}
\left({1+m^2\varepsilon^2}\right)^{(t/\varepsilon-1)/2}
Q_{(|x|-t)/2\varepsilon}^{(0,t/\varepsilon-1)}(1+2m^2\varepsilon^2).
$$
\end{remark}

%
%
%

\begin{remark} \label{rem-probability-to-die} Число $b_1(x,\varepsilon,m,\varepsilon)$ равно величине $(1+\sqrt{1+m^2\varepsilon^2})/m\varepsilon$,
умноженной на вероятность того, что плоское простое случайное блуждание по белым клеткам заканчивается в клетке
$(x,\varepsilon)$, при условии, что оно начинается в $(0,\varepsilon)$ и заканчивается с вероятностью $1-1/\sqrt{1+m^2\varepsilon^2}$ перед каждым шагом. Ничего подобного не известно для $b_{1}(x,t,m,\varepsilon)$ и $b_{2}(x,t,m,\varepsilon)$ при $t\ne \varepsilon$.
\end{remark}

Следующие результаты доказываются почти дословно так же, как предложение~\ref{cor-double-fourier} и теорема~\ref{th-ergenium}. 
Единственным отличием является смена знака слагаемых с участием $f_-(p)$ в~\eqref{eq-oscillatory-integral}, \eqref{eq-main-term}, \eqref{eq-boundary-term}, \eqref{eq-oscillatory-integral-a2}; аналоги лемм~\ref{l-main-term} и~\ref{l-boundary-term} получаются затем непосредственной проверкой.



\begin{proposition}[Преобразование Фурье по всему пространству-времени] \label{cor-double-fourier-anti}
Пусть $\delta_{x\varepsilon}:=1$, если $x=\varepsilon$, и $\delta_{x\varepsilon}:=0$, если $x\ne\varepsilon$.
Тогда для любого $m>0$ и любых $(x,t)\in \varepsilon\mathbb{Z}^2$ выполняются равенства
\begin{equation}
  \label{eq-def-anti2}
  \begin{aligned}
  A_1(x,t,m,\varepsilon)&=
  \lim_{\delta\to+0}
  \frac{m\varepsilon^3}{4\pi^2}
  \int_{-\pi/\varepsilon}^{\pi/\varepsilon}
  \int_{-\pi/\varepsilon}^{\pi/\varepsilon}
  \frac{ e^{i p x-i\omega(t-\varepsilon)}\,d\omega dp} {\sqrt{1+m^2\varepsilon^2}\cos(\omega\varepsilon)
  -\cos(p\varepsilon)-i\delta},\\
  \hspace{-0.5cm}A_2(x,t,m,\varepsilon)&=
  \lim_{\delta\to+0}
  \frac{-i\varepsilon^2}{4\pi^2}
  \int_{-\pi/\varepsilon}^{\pi/\varepsilon}
  \int_{-\pi/\varepsilon}^{\pi/\varepsilon}
  \frac{\sqrt{1+m^2\varepsilon^2}\sin(\omega\varepsilon)+\sin(p\varepsilon)} {\sqrt{1+m^2\varepsilon^2}\cos(\omega\varepsilon)
  -\cos(p\varepsilon)-i\delta}
  e^{i p (x-\varepsilon)-i\omega(t-\varepsilon)}
  \,d\omega dp+\delta_{x\varepsilon}\delta_{t\varepsilon}.
  \end{aligned}
  \end{equation}
\end{proposition}

\begin{theorem}
[Асимптотика при больших временах; см.~рис.~\ref{fig-approx-b}]
\label{th-anti-ergenium}
Для любого $\delta>0$ существует $C_\delta>0$ такое, что для любых $m,\varepsilon>0$ и любых $(x,t)\in\varepsilon\mathbb{Z}^2$, удовлетворяющих~\eqref{eq-case-A} выполняются равенства
\begin{align*}
{b}_1\left(x,t+\varepsilon,m,\varepsilon\right)
&={\varepsilon}\sqrt{\frac{2m}{\pi}}
\left(t^2-(1+m^2\varepsilon^2)x^2\right)^{-1/4}
\cos \theta(x,t,m,\varepsilon)
+O_\delta\left(\frac{\varepsilon}{m^{1/2}t^{3/2}}\right),\\
{b}_2\left(x+\varepsilon,t+\varepsilon,m,\varepsilon\right)
&=-{\varepsilon}\sqrt{\frac{2m}{\pi}}
\left(t^2-(1+m^2\varepsilon^2)x^2\right)^{-1/4}\sqrt{\frac{t+x}{t-x}}
\sin \theta(x,t,m,\varepsilon)
+O_\delta\left(\frac{\varepsilon}{m^{1/2}t^{3/2}}\right),
\end{align*}
для чётных и нечётных $(x+t)/\varepsilon$ соответственно. Здесь функция $\theta(x,t,m,\varepsilon)$ задана формулой~\eqref{eq-theta}.
\end{theorem}



\addcontentsline{toc}{myshrinkalt}{}

\subsection{Физическая интерпретация}

Величину $\frac{1}{2}|{A}_1\left(x,t,m,\varepsilon\right)|^2+
\frac{1}{2}|{A}_2\left(x,t,m,\varepsilon\right)|^2$ можно интерпретировать
как \emph{математическое ожидание заряда} в клетке $(x,t)$ с $t>0$ (единица измерения~--- заряд электрона).
Однако эти числа больше нельзя интерпретировать как вероятности обнаружить электрон в данной клетке. Причина заключается в том, что теперь результаты эксперимента не являются взаимоисключающими: можно обнаружить электрон в двух разных клетках одновременно.
В этом нет ничего загадочного: любое измерение обязательно влияет на электрон, и этого влияния может быть достаточно, чтобы создать электрон-позитронную пару из вакуума. Таким образом, можно обнаружить новорожденный электрон в дополнение к первоначальному, и нет никакого способа отличить один от другого.
(Более формальное объяснение для специалистов: два состояния в пространстве Фока, представляющие электрон, локализованный в удаленных областях, не являются взаимно ортогональными; их скалярное произведение 
по сути описывается фейнмановским пропагатором).

Мы анонсируем, что рассматриваемая модель в непрерывном пределе воспроизводит фейнмановский, а не запаздывающий пропагатор (см.~рис.~\ref{fig-approx-b}).
\emph{Фейнмановский пропагатор для спина $1/2$} имеет вид
\begin{equation}\label{eq-feynman-propagator}
G^F(x,t)=
\begin{cases}
\dfrac{m}{4}\,
\begin{pmatrix}
J_0(ms)-iY_0(ms) &
-{\frac{t+x}{s}}\left(J_1(ms)-iY_1(ms)\right) \\ {\frac{t-x}{s}}\left(J_1(ms)-iY_1(ms)\right) &
J_0(ms)-iY_0(ms)
\end{pmatrix}, &\text{при }|x|<|t|;\\
\dfrac{im}{2\pi}\,
\begin{pmatrix}
K_0(ms) & {\frac{t+x}{s}}\,K_1(ms)  \\
{\frac{x-t}{s}}\,K_1(ms) & K_0(ms)
\end{pmatrix}, &\text{при }|x|>|t|;
\end{cases}
\end{equation}
где $Y_n(z)$ и $K_n(z)$~--- соответственно функции Бесселя 2-го рода и модифицированные функции Бесселя 2-го рода, и $s:=\sqrt{\left|t^2-x^2\right|}$.
Кроме этого, есть обобщенная функция с носителем на прямых $t=\pm x$, которую мы не указываем. Фейнмановский пропагатор удовлетворяет уравнению~\eqref{eq-Green}. Мы видим, что он имеет дополнительную мнимую часть (и общий коэффициент $1/2$) по сравнению с запаздывающим пропагатором~\eqref{eq-relativistic-propagator}. В частности, он не равен нулю при $|x|>|t|$: аннигиляция электрона в одной точке и рождение в другой может привести к кажущемуся движению электрона быстрее скорости света.

Более известным представлением для пропагатора Фейнмана является преобразование Фурье 
(ср.~\eqref{eq-double-fourier-retarded} и~\cite[(6.51)]{Folland})
\begin{equation}\label{eq-double-fourier-feynman}
G^F(x,t)=
  \frac{1}{4\pi^2}
  \int_{-\infty}^{+\infty}
  \int_{-\infty}^{+\infty}
  \lim_{\delta\to+0}
  \begin{pmatrix}
  m & -ip-i\omega \\
  -ip+i\omega & m
  \end{pmatrix}
  \frac{ e^{i p x-i\omega t}\,dpd\omega}
  {m^2+p^2-\omega^2-i\delta}.
\end{equation}

В целом, небольшая поправка, вносимая этим уточнением модели, скорее отражает некоторые фундаментальные ограничения на измерения, чем добавляет что-то значимое к описанию движения. Введённое уточнение следует рассматривать только как ингредиент для более реалистичных моделей с взаимодействием.

\addcontentsline{toc}{myshrinkalt}{}

\subsection{Комбинаторное определение}

Неформально, комбинаторное определение фейнмановских антишашек (определение~\ref{def-anti-combi}) получается из определения фейнмановских шашек (определение~\ref{def-mass}) следующей четырёхшаговой процедурой:
\begin{description}
\item[Шаг 1:] масса частицы приобретает небольшую мнимую часть, которую в конечном счёте мы  устремляем к нулю;
\item[Шаг 2:] в то время как вещественная масса ``отвечает'' за повороты в центрах квадратов, наличие мнимой массы позволяет делать повороты в их \emph{вершинах} (см.~рис.~\ref{fig-1x1} слева);
\item[Шаг 3:] бесконечная решетка заменяется тором, размеры которого в конечном счете мы устремяем к бесконечности;
\item[Шаг 4:] сумма по путям шашки заменяется отношением сумм по конфигурациям петель.
\end{description}
Шаг~2 здесь является совершенно новым, в то время как остальные шаги стандартны.
Шаг~2 отражает общий принцип, согласно которому действительная и мнимая части комплексной величины всегда располагаются на двойственных решётках.

\begin{figure}[htbp]
  \centering
\includegraphics[height=2.5cm]{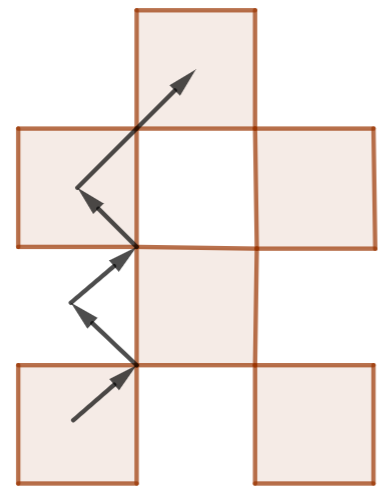}\qquad
\includegraphics[height=2.5cm]{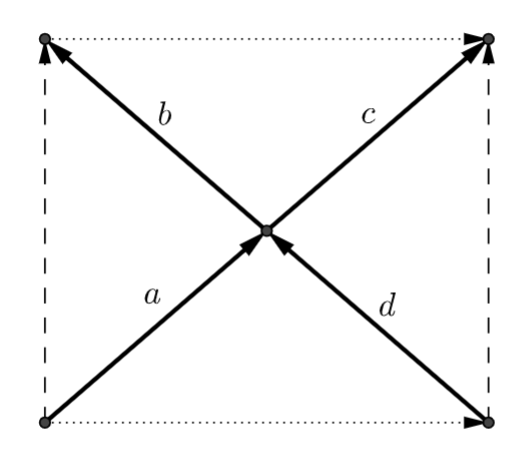}
\includegraphics[height=2.5cm]{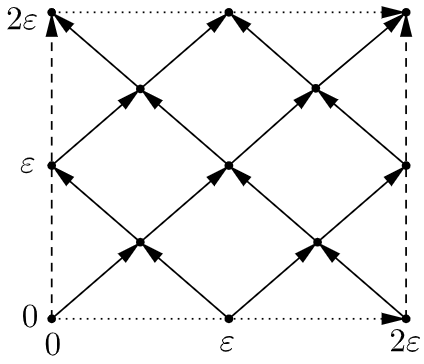}
\caption{Обобщенный путь шашки (слева) и решетки размеров $1$ и $2$ (справа); см.~пример~\ref{ex-1x1}.
}
  \label{fig-1x1}
\end{figure}

\begin{definition} \label{def-anti-combi} (См.~рис.~\ref{fig-1x1})
Зафиксируем параметры $T\in\mathbb{Z}$ и $\varepsilon,m,\delta>0$, называемые \emph{размером решётки, шагом решётки, массой частицы}, и \emph{малой мнимой массой} соответственно. Будем предполагать, что $T>0$ и $\delta<1/2$. \emph{Решёткой} назовём следующее фактормножество:
$$
\faktor{
\{\,(x,t)\in[0,T\varepsilon]^2:
2x/\varepsilon,2t/\varepsilon,(x+t)/\varepsilon\in\mathbb{Z}\,\}
}
{
\forall x,t:(x,0)\sim (x,T\varepsilon)\,\&\,(0,t)\sim (T\varepsilon,t).
}
$$
(По сути это тор, полученный из квадрата $[0,T\varepsilon]^2$ отождествлением противоположных сторон.)
Точка решетки $(x,t)$ называется \emph{чётной} (соответственно, \emph{нечётной}), если отношение $2x/\varepsilon$ является чётным (соответственно, нечётным) числом.
\emph{Ребро}~--- это вектор, начинающийся в точке решётки $(x,t)$ и заканчивающийся в точке рёшетки $(x+\varepsilon/2,t+\varepsilon/2)$ или $(x-\varepsilon/2,t+\varepsilon/2)$.

\emph{Обобщенный путь шашки}~--- это конечная последовательность различных рёбер, таких, что конечная точка каждого ребра является начальной точкой следующего. \emph{Цикл} определяется аналогично, только последовательность имеет единственное повторение: первое и последнее рёбра совпадают, и есть еще хотя бы одно ребро между ними. (В частности, обобщенный путь шашки, у которого конечная точка последнего ребра является начальной точкой первого, еще \emph{не} является циклом: требуется совпадение первого и последнего \emph{рёбер}. Первое и последнее рёбра обобщенного \emph{пути} шашки совпадают, только если в пути ровно одно ребро. В нашей постановке путь \emph{никогда} не является циклом.) \emph{Изменение начального ребра} цикла означает удаление первого ребра из последовательности, затем циклическую перестановку, а затем добавление последнего ребра полученной последовательности в начало. \emph{Петля}~--- это цикл с точностью до изменения начального ребра.



\emph{Узел} пути или петли $s$~--- это \new{упорядоченная} пара последовательных ребер в $s$ \new{(порядок ребер в паре такой же, как в~$s$)}.
\emph{Поворот}~--- это узел, у которого эти два ребра ортогональны. Узел или поворот называется \emph{чётным} (соответственно \emph{нечётным}), если \new{конечная точка первого ребра в паре}
чётная (соответственно нечётная). Обозначим через
$\mathrm{eventurns}(s)$, $\mathrm{oddturns}(s)$, $\mathrm{evennodes}(s)$, $\mathrm{oddnodes}(s)$ количество чётных и нечётных поворотов и узлов в $s$ соответственно. \emph{Стрелка} (или \emph{вес}) пути или петли $s$~--- это число
\begin{equation*}
{A}(s,m,\varepsilon,\delta)
:=\pm\frac{(-im\varepsilon)^{\mathrm{oddturns}(s)}
(-\delta)^{\mathrm{eventurns}(s)}}
{(1+m^2\varepsilon^2)^{\mathrm{oddnodes}(s)/2}
(1-\delta^2)^{\mathrm{evennodes}(s)/2}},
\end{equation*}
где знак минус выбирается в случае, когда $s$ является петлёй.

Набор путей или петель \emph{рёберно непересекающийся}, если
никакие два из них не имеют общего ребра. Рёберно непересекающийся набор петель называется \emph{конфигурацией петель}. \emph{Конфигурация петель с источником \new{$a$} и стоком $f$}~--- это рёберно непересекающийся набор из любого количества петель и ровно одного обобщенного пути шашки, начинающегося с ребра \new{$a$} и заканчивающегося ребром $f$. 
\emph{Стрелка} $A(S,m,\varepsilon,\delta)$ конфигурации петель $S$ (возможно, с источником и стоком)~--- это произведение стрелок всех петель и путей в конфигурации. Пустое произведение по определению равно $1$.

\emph{Стрелка от ребра \new{$a$} к ребру $f$} (или \emph{пропагатор на конечной решетке})~--- это число \new{}\new{}
\begin{equation*}
{A}(a\to f,m,\varepsilon,\delta,T)
:=\frac{\sum\limits_{\substack{\text{конфигурации петель $S$}\\
\text{с источником $a$ и стоком $f$}}}A(S,m,\varepsilon,\delta)}
{\sum\limits_{\substack{\text{конфигурации петель $S$}}}A(S,m,\varepsilon,\delta)}.
\end{equation*}

Теперь рассмотрим точку $(x,t)\in (\varepsilon\mathbb{Z})^2$ и положим $ x':=x\mod T\varepsilon$, $t':=t\mod T\varepsilon$. Пусть \new{} $a_0,f_1,f_2$~--- рёбра, начинающиеся в точках $(0,0)$, $(x',t')$, $(x',t')$ и заканчивающиеся в точках $(\varepsilon/2,\varepsilon/2)$, $(x'-\varepsilon/2,t'+\varepsilon/2)$, $(x'+\varepsilon/2,t'+\varepsilon/2)$ соответственно.
\emph{Стрелка точки $(x,t)$} (или \emph{пропагатор на бесконечной решетке})~--- это пара комплексных чисел \new{}
\begin{equation*}
\widetilde{A}_k(x,t,m,\varepsilon):=-2(-i)^k\,
\lim_{\delta\searrow 0}\lim_{\substack{T\to\infty}}
{A}(a_0\to f_k,m,\varepsilon,\delta,T)
\qquad\text{для $k=1,2$.}
\end{equation*}
\end{definition}

\begin{example} \label{ex-1x1} (см.~рис.~\ref{fig-1x1} в центре) Решётка размера $1$ \new{расположена в квадрате} $[0,\varepsilon]^2$ с отождествлёнными противоположными сторонами. \new{Она состоит из двух точек --- центра и отождествленных вершин квадрата --- и имеет} $4$ ребра $a,b,c,d$. Обобщенные пути шашки $abdc$, $acdb$, $bacd$ различны, хотя содержат одинаковые ребра. Их стрелки равны $\tfrac{-m^2\varepsilon^2}{\sqrt{1-\delta^2}(1+m^2\varepsilon^2)}$,
$\tfrac{-\delta}{\sqrt{1-\delta^2}(1+m^2\varepsilon^2)}$,
$\tfrac{\delta^2}{(1-\delta^2)\sqrt{1+m^2\varepsilon^2}}$
соответственно. Эти пути отличаются от циклов $acdba$, $bacdb$. Эти два цикла определяют одну и ту же петлю со стрелкой $\tfrac{-\delta^2}{(1-\delta^2)(1+m^2\varepsilon^2)}$. Всего есть $9$ конфигураций петель: $\varnothing$, $\{aba\}$, $\{cdc\}$, $\{aca\}$, $\{bdb\}$, $\{abdca\}$, $\{acdba\}$, $\{aba,cdc\}$, $\{aca,bdb\}$. Их стрелки соответственно равны $1$, ${-im\varepsilon\delta}/n$, ${-im\varepsilon\delta}/n$,
${-1}/n$, ${-1}/n$, ${m^2\varepsilon^2}/n^2$, ${-\delta^2}/n^2$,
${-m^2\varepsilon^2\delta^2}/n^2$, ${1}/n^2$, где $n:={\sqrt{1-\delta^2}\sqrt{1+m^2\varepsilon^2}}$.
\end{example}


Неформально, петли образуют \emph{море Дирака} из электронов, заполняющих все пространство, а ребра, не входящие в петли, образуют пути дырок в этом море, то есть античастиц.

Мы анонсируем, что определения~\ref{def-anti} и~\ref{def-anti-combi} эквивалентны в том смысле, что $\widetilde{A}_1(x,t,m,\varepsilon)={A}_1(x,t+\varepsilon,m,\varepsilon)$ и $\widetilde{A}_2(x,t,m,\varepsilon)={A}_2(x+\varepsilon,t+\varepsilon,m,\varepsilon)$;
ср.~\S\ref{sec-source}. Идея доказательства заключается в том, что оба определения фактически строят матричные элементы ``обратного'' оператора к одному и тому же дискретному оператору Дирака: первое~--- через преобразование Фурье, а второе~--- через отношение определителей.

\addcontentsline{toc}{myshrink}{}

\section{На пороге к $(1+1)$-мерной квантовой электродинамике
}
\label{sec-QED}

\addcontentsline{toc}{myshrink}{}

{
\hrule
\footnotesize
\noindent\textbf{Вопрос:} какова вероятность обнаружения электронов (или электрона и позитрона) с  импульсами $q$ и $q'$ в далёком будущем, если они были испущены с импульсами $p$ и $p'$ в далёком прошлом?

\noindent\textbf{Предположения:} теперь есть взаимодействие; все упрощения отброшены, кроме тех, которые были по умолчанию:
нет ядерных сил, нет гравитации, электрон движется только вдоль оси $x$, ось $t$ --- время.

\noindent\textbf{Результаты:} отталкивание одинаковых зарядов и притяжение противоположных зарядов 
(ожидается качественное обоснование).
\hrule
}

\bigskip

Построение требуемой модели является открытой проблемой, поскольку  для этого, в частности, требуется математически строгое построение решеточной калибровочной теории для ``метрики'' \emph{Минковского}.

\section{Открытые проблемы}\label{sec-open}

\addcontentsline{toc}{myshrink}{}






Начнем с проблем, опирающихся только на определение~\ref{def-basic}. Графики показывают, что для фиксированного $t$ вероятнее всего обнаружить электрон вблизи точки $x=t/\sqrt{2}$ (см.~рис.~\ref{a-1000} слева, теоремы~\ref{th-limiting-distribution}(B) и~\ref{th-ergenium}).
Хотя это было замечено 20 лет назад, следующий вопрос все еще остается открытым.

\begin{pr} \label{P-ugol-eng} (А.~Даниярходжаев--Ф.~Куянов; см.~рис.~\ref{a-1000} слева.)
Обозначим через $x_{\max}(t)$ точку, в которой $P(x,t)$ достигает максимума при фиксированном $t$. Ограничено ли $x_{\max}(t)-t/\sqrt{2}$ при $t\to\infty$?
\end{pr}

Эту проблему делает сложной то обстоятельство, что поведение вероятности $P(x,t)$ известно только вблизи $x=t/\sqrt{2}$ (теорема~\ref{th-Airy}) и вдали от $x=t/\sqrt{2}$ (теоремы~\ref{th-ergenium} и~\ref{th-outside}), но не на промежуточных расстояниях.

\begin{pr} (Ср.~\cite{Sunada-Tate-12}) Найти асимптотику вероятности $P(x,t)$ при $t\to\infty$, равномерную по $x\in [-t;t]$.
\end{pr}

\begin{pr} (С.~Нечаев; см.~рис.~\ref{a-1000} слева) Найти положения ``широких провалов'' (участков, где колебания меньше) в графике вероятности $P(x,t)$ для фиксированного большого $t$.
(Ср.~\eqref{eq-ergenium-re}--\eqref{eq-ergenium-im}.)
\end{pr}

%

Цель следующих двух проблем --- изучить фазовый переход с помощью различных параметров порядка (см.~стр.~\pageref{equal-signs}).
В частности, мы предполагаем, что предельная ``вероятность'' одинаковых знаков на концах спиновой цепочки, так же как и предельная ``вероятность'' одинаковых знаков на концах и в середине, неаналитичны в точке $v=\pm 1/\sqrt{2}$.


\begin{pr} \label{p-correlation}
(См.~рис.~\ref{fig-correlation}.)
Для каждого $0<v<1/\sqrt{2}$ докажите равенство
$$
\lim_{t\to\infty}\sum\limits_{0\le x\le vt} \frac{2}{t}\left|\frac{a_2(x,t)}{a(x,t)}\right|^2=
\frac{1}{2} \left(1+v-\sqrt{1-v^2}+\log \frac{1+\sqrt{1-v^2}}{2}\right).
$$
Вычислите этот же предел при $1/\sqrt{2}<v<1$.
(Ср.~с доказательством следствия~\ref{th-limiting-distribution-mass} в~\S\ref{ssec-proofs-phase}.)
\end{pr}

\begin{pr} \label{p-correlation2} (Ср.~\cite[p.~381]{Jacobson-Schulman-84}.)
Найдите слабый предел  
$\lim_{t\to\infty}\left|\sum_{x\in\mathbb{Z}}
\dfrac{a_2(x,t)^2}{a_2(2\lceil vt\rceil-1,2t-1)}\right|^2$.
\end{pr}




Следующая проблема о вероятностях поглощения; она опирается на определение перед примером~\ref{p-double-slit}.

\begin{pr} \label{p-pi} ((Г.~Минаев--И.~Русских); ср.~\cite[\S5]{Ambainis-etal-01},\,\cite{Dmitriev-22},\,\cite[\S4]{Novikov-20})
Для всех $n\in\mathbb{Z}$ найдите \emph{вероятность поглощения} $\sum_{t=1}^{\infty}P(n,t \,\text{минуя}\, \{x=n\})$.
Найдите слабый предел и асимптотику вероятности
$P(x,t \text{ минуя } \{x=0\})$ при $t\to\infty$. (Ср.~теоремы~\ref{th-limiting-distribution}--\ref{th-ergenium}.)
\end{pr}


Следующая проблема обобщает и уточняет проблему~\ref{P-ugol-eng} выше. Она опирается на определение~\ref{def-mass}.

\begin{pr} \label{P-ugol-mass}
(А.~Даниярходжаев--Ф.~Куянов)
Обозначим через $x_{\max}=x_{\max}(t,m,\varepsilon)$ точку, в которой $P(x):=P(x,t,m,\varepsilon)$ достигает максимума. Является ли величина $x_{\max}/\varepsilon-t/\varepsilon\sqrt{1+m^2\varepsilon^2}$ равномерно ограниченной? Убывает ли $P(x)$ при $x>x_{\max}$?
Найдите асимптотику для $a(x,t,m,\varepsilon)$ при $t\to \infty$ для всех $x\in [-t,t]$, равномерную по $x,m,\varepsilon$.
\end{pr}

\begin{pr} (М.~Бланк--С.~Шлосман) Остается ли ограниченным количество смен знака функции $a_1(x):=a_1(x,t,m,\varepsilon)$ на отрезке $[-t,t]$ при $\varepsilon\to 0$ для фиксированных $t,m$?
\end{pr}

Следствие~\ref{cor-uniform} дает равномерный предел на компактных подмножествах угла $|x|<t$, а значит, упускает основной вклад в вероятность. Теперь нас интересует слабый предел, детектирующий пик.

\begin{pr} \label{p-weak} Найти слабые пределы  $\lim\limits_{\varepsilon\to 0}  \frac{1}{2\varepsilon}\, {a}\left(2\varepsilon\!\left\lceil \frac{x}{2\varepsilon}\right\rceil,2\varepsilon\!\left\lceil \frac{t}{2\varepsilon}\right\rceil,{m},{\varepsilon}\right)$ и $\lim\limits_{\varepsilon\to 0}  \frac{1}{4\varepsilon^2}\, {P}\left(2\varepsilon\!\left\lceil \frac{x}{2\varepsilon}\right\rceil,2\varepsilon\!\left\lceil \frac{t}{2\varepsilon}\right\rceil,{m},{\varepsilon}\right)$
как распределений на $\mathbb{R}^2$. Равен ли первый предел
пропагатору~\eqref{eq-double-fourier-retarded}, включая обобщенную функцию на прямых $t=\pm x$? Каков физический смысл второго предела (он может придать смысл некорректно определенному квадрату пропагатора)?
\end{pr}


%


Следующая проблема --- построить непрерывный аналог шашек Фейнмана.

\begin{pr} (М.~Лифшиц)
Рассмотрим $(-im\varepsilon)^{\mathrm{turns}(s)}$ как заряд на множестве всех путей $s$ шашки из $(0,0)$ в $(x,t)$, начинающихся и заканчивающихся ходом вправо-вверх. Сходится ли  этот заряд при $\varepsilon\to 0$ (в слабом или каком-то еще смысле) к заряду на пространстве всех непрерывных функций $[0,t]\to\mathbb{R}$ с граничными значениями $0$ и $x$ соответсвенно?
\end{pr}

Следующая проблема, опирающаяся на определение~\ref{def-external}, продемонстрировала бы ``прецессию спина''.

\begin{pr} \label{p-precession-prove} (см.~рисунок~\ref{spin-reversal} справа; ср.~\cite{Ozhegov-21})
Становится ли $P(x)=\sum_{x\in\mathbb{Z}} a_1(x, t, u)^2$ асимптотически периодической при $t\to\infty$, если $u(x+\tfrac{1}{2},t+\tfrac{1}{2})=(-1)^{(x-1)(t-1)}$?
Найти слабые пределы вероятности $P(x,t,u)$ и асимптотики функций
$a_k(x,t,u)$ при $t\to\infty$. (Ср.~теоремы~\ref{th-limiting-distribution}--\ref{th-ergenium}.)
\end{pr}




Определим ${a}(x,t,m,\varepsilon,u)$ аналогично  ${a}(x,t,m,\varepsilon)$ и ${a}(x,t,u)$, объединяя определения~\ref{def-mass}--\ref{def-external} и замечание~\ref{rem-external}. Мы предполагаем, что такое определение воспроизводит уравнение Дирака в электромагнитном поле.

\begin{pr} (Ср.~\cite{Schulman-Gaveau-89}) Зафиксируем $A_0(x,t),A_1(x,t)\in C^2(\mathbb{R}^2)$. Для каждого вспомогательного ребра $s_1s_2$ положим
$$u(s_1s_2)
:=\exp\left(-i\int_{s_1}^{s_2}\left(A_0(x,t)\,dt+A_1(x,t)\,dx\right)\right).
$$
Обозначим $\psi_k(x,t):=\lim_{\varepsilon\to 0}  \frac{1}{2\varepsilon}\, {a}_k\left(2\varepsilon\!\left\lceil \frac{x}{2\varepsilon}\right\rceil,2\varepsilon\!\left\lceil \frac{t}{2\varepsilon}\right\rceil,{m},{\varepsilon},u\right)$ при $k=1,2$ \new{и $t>0$}.
Удовлетворяет ли этот предел уравнению \new{}
$$
\begin{pmatrix}
m  & \partial/\partial x-\partial/\partial t +iA_0(x,t)-iA_1(x,t)\\
\partial/\partial x+\partial/\partial t - iA_0(x,t)-iA_1(x,t) & m
\end{pmatrix}
\begin{pmatrix}
\psi_2(x,t) \\ \psi_1(x,t)
\end{pmatrix}=0
?
$$
\end{pr}

Следующие две проблемы опираются на определение~\ref{def-anti}.

\begin{pr} (Ср.~следствие~\ref{th-limiting-distribution-mass})
Докажите равенство
$$
\lim\limits_{\substack{t\to\infty\\t\in\varepsilon\mathbb{Z}}}
\sum\limits_{\substack{x\le vt\\x\in\varepsilon\mathbb{Z}}}
\dfrac{\left|A_1(x,t,{m},{\varepsilon})\right|^2
+\left|A_2(x,t,{m},{\varepsilon})\right|^2}{2}=F(v,m,\varepsilon).
$$
\end{pr}

\begin{pr} (Ср.~теоремы~\ref{th-outside} и~\ref{th-anti-ergenium}) Найти асимптотики функций
$A_k(x,t,{m},{\varepsilon})$ для $|x|>|t|$ при $t\to \infty$.
\end{pr}


Последняя проблема неформальна. Она стоит уже полвека.

\begin{pr} (Р.~Фейнман; ср.~\cite{Foster-Jacobson-17}) Обобщите модель на $4$ измерения так, чтобы
$$\lim_{\varepsilon\to 0 }  \frac{1}{2\varepsilon}\, {a}\left(2\varepsilon\!\left\lceil \frac{x}{2\varepsilon}\right\rceil,2\varepsilon\!\left\lceil \frac{y}{2\varepsilon}\right\rceil,2\varepsilon\!\left\lceil \frac{z}{2\varepsilon}\right\rceil,2\varepsilon\!\left\lceil \frac{t}{2\varepsilon}\right\rceil,{m},{\varepsilon}\right)$$
совпал с запаздывающим пропагатором для спина $1/2$, на этот раз в $3$ пространственных и $1$ временном измерении.
\end{pr}

\section{Доказательства}\label{sec-proofs}

Приведем схему, показывающую
зависимость вышеприведенных результатов и последующих подразделов:
{\footnotesize
$$
\xymatrix{
\boxed{\text{\ref{ssec-proofs-moments}. (теорема~\ref{th-limiting-distribution})}}
&
\boxed{\text{\ref{ssec-proofs-basic}. (предложения~\ref{p-Dirac}--\ref{cor-double-fourier})}}
\ar[l]_{\ref{p-mass},\ref{cor-fourier-integral}}
\ar[ld]_{\ref{p-Dirac}-\ref{cor-coefficients},\ref{p-symmetry-mass},\ref{p-Huygens}}
\ar[d]^{\ref{p-mass2},\ref{p-mass3}}
\ar[rrd]^{\ref{p-mass2}}
\ar[r]^{\ref{p-mass},\ref{cor-fourier-integral}}
&
\boxed{\text{\ref{ssec-proofs-main}. (теорема~\ref{th-ergenium})}} \ar[r]\ar[rd]
&
\boxed{\text{\ref{ssec-proofs-feynman}. (следствия~\ref{cor-intermediate-asymptotic-form}--\ref{cor-feynman-problem})}}
\\
\boxed{\text{\ref{ssec-proofs-spin}. (теорема~\ref{p-right-prob})}}
&
\boxed{\text{\ref{ssec-proofs-continuum}. (теорема~\ref{th-main}, следствия~\ref{cor-uniform}--\ref{cor-concentration})}}
&
&
\boxed{\text{\ref{ssec-proofs-phase}. (следствия~\ref{th-limiting-distribution-mass}--\ref{cor-free})}}
}
$$
}
Номера используемых предложений указаны над стрелками. Предложения~\ref{p-Klein-Gordon-mass}, \ref{p-equal-time}, и~\ref{cor-double-fourier} не применяются для доказательства основных результатов. 

В процессе доказательств мы даем не требующее предварительных знаний введение в используемые методы. Некоторые доказательства проще, чем оригинальные.

\addcontentsline{toc}{myshrinkalt}{}

\subsection{Тождества: элементарная комбинаторика (предложения~\ref{p-Dirac}--\ref{cor-double-fourier}) }
\label{ssec-proofs-basic}

Докажем тождества из~\S\ref{sec-mass}; тождества из \S\ref{sec-basic} --- это частный случай, когда $m=\varepsilon=1$. 

\begin{proof}[Доказательство предложений~\ref{p-Dirac} и~\ref{p-mass}]
Выведем формулу для $a_2(x,t,m,\varepsilon)$. Рассмотрим произвольный путь $s$ шашки на $\varepsilon\mathbb{Z}^2$ из $(0,0)$ в $(x,t)$ с первым ходом в $(\varepsilon,\varepsilon)$.
Обозначим $a(s,m \varepsilon):=i(-im \varepsilon)^{\mathrm{turns}(s)}(1+m^2\varepsilon^2)^{(1-t/\varepsilon)/2}$.

Последний ход пути $s$ сделан либо из $(x-\varepsilon,t-\varepsilon)$, либо из $(x+\varepsilon,t-\varepsilon)$. Если он сделан из $(x+\varepsilon,t-\varepsilon)$,
то число $\mathrm{turns}(s)$ нечётно и $s$ не дает вклада в $a_2(x,t,m,\varepsilon)$. Поэтому далее будем предполагать, что последний ход сделан из $(x-\varepsilon,t-\varepsilon)$.
Обозначим через $s'$ путь $s$ без последнего хода. Если направления последних ходов в $s$ и $s'$ совпадают, то
$a(s,m \varepsilon) = \frac{1}{\sqrt{1+m^2\varepsilon^2}} a(s',m \varepsilon)$, иначе $a(s,m \varepsilon) = \frac{-im \varepsilon}{\sqrt{1+m^2\varepsilon^2}}a(s',m \varepsilon)
= \frac{m \varepsilon}{\sqrt{1+m^2\varepsilon^2}}(\mathrm{Im}\,a(s',m \varepsilon)-i\mathrm{Re}\,a(s',m \varepsilon))$.

Суммирование по всем путям $s'$ дает требуемое уравнение
\begin{multline*}
a_2(x,t,m, \varepsilon) =
\mathrm{Im}\sum_{s\ni(x-\varepsilon,t-\varepsilon)} a(s,m\varepsilon)
=\sum_{s'\ni(x-2\varepsilon,t-2\varepsilon)} \frac{\mathrm{Im}\,a(s',m\varepsilon)}{\sqrt{1+m^2\varepsilon^2}}-
\sum_{s'\ni(x,t-2\varepsilon)} \frac{m\varepsilon\,\mathrm{Re}\,a(s',m\varepsilon)}
{\sqrt{1+m^2\varepsilon^2}} \\
=\frac{1}{\sqrt{1+m^2\varepsilon^2}}\left( a_2(x-\varepsilon,t-\varepsilon,m, \varepsilon)
-  m \varepsilon\, a_1(x-\varepsilon,t-\varepsilon,m, \varepsilon)\right).
\end{multline*}
Реккурентная формула для $a_1(x,t+\varepsilon,m,\varepsilon)$ доказывается аналогично.
\end{proof}

\begin{proof}[Доказательство предложений~\ref{p-probability-conservation} и~\ref{p-mass2}]
Доказательство проведём индукцией по $t/\varepsilon$. База $t/\varepsilon=1$ очевидна. Шаг индукции получается из следующих вычислений, опирающихся на предложение~\ref{p-mass}:
\begin{multline*}
    \sum \limits_{x \varepsilon\in \mathbb{Z}} P(x,t+\varepsilon,m, \varepsilon)
    = \sum \limits_{x \in \varepsilon\mathbb{Z}}
    \left[ a_1(x, t+\varepsilon,m, \varepsilon)^2
    + a_2(x, t+\varepsilon,m, \varepsilon)^2 \right]
    \\ \hspace{-0.8cm} = \tfrac{1}{1+m^2\varepsilon^2}
    \left(\sum \limits_{x \in \varepsilon\mathbb{Z}}
    [a_1(x+\varepsilon,t,m, \varepsilon)
    + m \varepsilon\, a_2(x+\varepsilon, t,m, \varepsilon)]^2  + \sum \limits_{x \in \varepsilon\mathbb{Z}}
    [a_2(x-\varepsilon,t,m, \varepsilon)
    - m \varepsilon\, a_1(x-\varepsilon, t,m, \varepsilon)]^2\right)
    \\= \tfrac{1}{1+m^2\varepsilon^2}\left(
    \sum \limits_{x \in \varepsilon\mathbb{Z}}
    [a_1(x,t,m, \varepsilon)
    + m \varepsilon\, a_2(x, t,m, \varepsilon)]^2  + \sum \limits_{x \in \varepsilon\mathbb{Z}}
    [a_2(x,t,m, \varepsilon)
    - m \varepsilon\, a_1(x, t,m, \varepsilon)]^2\right)
    \\= \sum \limits_{x \in \varepsilon\mathbb{Z}}
    \left[a_1(x,t,m, \varepsilon)^2
    + a_2(x, t,m, \varepsilon)^2\right]
    = \sum \limits_{x \in \varepsilon\mathbb{Z}}P(x,t,m, \varepsilon).\\[-1.5cm]
\end{multline*}
\end{proof}

\begin{lemma}[Сопряженное уравнение Дирака] \label{p-Dirac-conjugate} Для всех $(x,t)\in \varepsilon\mathbb{Z}^2$, где $t>\varepsilon$, выполняются равенства
\begin{align*}
a_1(x,t-\varepsilon,m,\varepsilon)
&=\frac{1}{\sqrt{1+m^2\varepsilon^2}}
(a_1(x-\varepsilon,t,m,\varepsilon)
-m\varepsilon\, a_2(x+\varepsilon,t,m,\varepsilon));\\
a_2(x,t-\varepsilon,m,\varepsilon)
&=\frac{1}{\sqrt{1+m^2\varepsilon^2}}
(m\varepsilon\, a_1(x-\varepsilon,t,m,\varepsilon)
+a_2(x+\varepsilon,t,m,\varepsilon)).
\end{align*}
\end{lemma}

\begin{proof}[Доказательство леммы~\ref{p-Dirac-conjugate}]
Покажем, как второе уравнение получается из предложения~\ref{p-mass}. Для этого в~\eqref{eq-Dirac-mass1} и~\eqref{eq-Dirac-mass2} заменим $(x,t)$ на $(x-\varepsilon,t-\varepsilon)$ и $(x+\varepsilon,t-\varepsilon)$  соответственно
и сложим полученные равенства с коэффициентами $m\varepsilon/\sqrt{1+m^2\varepsilon^2}$ и $1/\sqrt{1+m^2\varepsilon^2}$. Первое уравнение получается аналогично.
\end{proof}

\begin{proof}[Доказательство предложения
~\ref{p-Klein-Gordon-mass}]
Действительная часть искомого уравнения является суммой первых уравнений из леммы~\ref{p-Dirac-conjugate}  и предложения~\ref{p-mass}. Мнимая часть~--- это сумма вторых уравнений.
\end{proof}

\begin{proof}[Доказательство предложения
~\ref{p-symmetry-mass}]
Докажем первое тождество. Для каждого пути $s$ обозначим через $s'$ отражение пути $s$ относительно оси $t$, а через $s''$ путь, состоящий из тех же ходов, что и $s'$, но сделанных в обратном порядке. 

Рассмотрим путь $s$ из $(0,0)$ в $(x,t)$ с первым ходом вправо-вверх, у которого $\mathrm{turns}(s)$ нечётно (пути с чётным $\mathrm{turns}(s)$ не дают вклада в $a_1(x,t,m,\varepsilon)$). Последний ход $s$ направлен влево-вверх. Таким образом, последний ход пути $s'$ направлен вправо-вверх, а значит, первый ход пути $s''$ направлен вправо-вверх. Оба пути $s'$ и $s''$ заканчиваются в $(-x,t)$, так как изменение порядка ходов не влияет на конечную точку.

Таким образом, $s\mapsto s''$~--- это биекция между путями с концом в $(x,t)$ и $(-x,t)$ с нечётным числом поворотов. Значит, $a_1(x,t,m,\varepsilon)=a_1(-x,t,m,\varepsilon)$.

Второе тождество мы доказываем индукцией по $t/\varepsilon$ (это доказательство было найдено и написано Е.~Колпаковым). База индукции ($t/\varepsilon=1$ и $t/\varepsilon=2$)~---  очевидна.

Шаг индукции. Возьмем $t\ge 3\varepsilon$. Применяя индуктивное предположение к трём точкам $(x-\varepsilon,t-\varepsilon), (x+\varepsilon,t-\varepsilon), (x,t-2\varepsilon)$ и пользуясь тождеством $a_1(x,t,m,\varepsilon)=a_1(-x,t,m,\varepsilon)$, получаем
\begin{align*}
  (t-x)a_2(x-\varepsilon,t-\varepsilon, m, \varepsilon)&=(x+t-4\varepsilon)a_2(3\varepsilon-x,t-\varepsilon, m, \varepsilon),\\
  (t-x-2\varepsilon)a_2(x+\varepsilon,t-\varepsilon, m, \varepsilon)&=(x+t-2\varepsilon)a_2(\varepsilon-x,t-\varepsilon, m, \varepsilon),\\
  (t-x-2\varepsilon)a_2(x,t-2\varepsilon, m, \varepsilon)&=(x+t-4\varepsilon)a_2(2\varepsilon-x,t-2\varepsilon, m, \varepsilon),\\
   a_1(x-\varepsilon,t-\varepsilon, m, \varepsilon)&= a_1(\varepsilon-x,t-\varepsilon, m, \varepsilon).
\end{align*}
Суммируя эти уравнения с коэффициентами $1,1,-\sqrt{1+m^2\varepsilon^2},-2m\varepsilon^2$ соответственно, находим, что
\begin{multline*}
(t-x)\left(a_2(x-\varepsilon,t-\varepsilon, m, \varepsilon)+a_2(x+\varepsilon,t-\varepsilon, m, \varepsilon)-\sqrt{1+m^2\varepsilon^2}\,a_2(x,t-2\varepsilon, m, \varepsilon)\right)\\
 -2m\varepsilon^2\, a_1(x-\varepsilon,t-\varepsilon, m, \varepsilon) -2\varepsilon\, a_2(x+\varepsilon,t-\varepsilon, m, \varepsilon)+2\varepsilon\sqrt{1+m^2\varepsilon^2}\,a_2(x,t-2\varepsilon, m, \varepsilon)= \\
=-2m\varepsilon^2\, a_1(\varepsilon-x,t-\varepsilon, m, \varepsilon) -2\varepsilon\, a_2(3\varepsilon-x,t-\varepsilon, m, \varepsilon)+2\varepsilon\sqrt{1+m^2\varepsilon^2}\,a_2(2\varepsilon-x,t-2\varepsilon, m, \varepsilon)\\
 +(t+x-2\varepsilon)\left(a_2(3\varepsilon-x,t-\varepsilon, m, \varepsilon)+a_2(\varepsilon-x,t-\varepsilon, m, \varepsilon)-\sqrt{1+m^2\varepsilon^2}\,a_2(2\varepsilon-x,t-2\varepsilon, m, \varepsilon)\right).
\end{multline*}
Из леммы~\ref{p-Dirac-conjugate} следует, что три слагаемых во 2-й строке сокращаются с тремя слагаемыми в 3-й строке.
Применяя уравнение Клейна~--- Гордона (предложение~\ref{p-Klein-Gordon-mass}) к выражениям в 1-й и 4-й строках и сокращая полученное выражение на  сомножитель $\sqrt{1+m^2\varepsilon^2}$, мы получаем искомое тождество
$$(t-x)a_2(x,t, m, \varepsilon) =(t+x-2\varepsilon)a_2(2\varepsilon-x,t,m, \varepsilon).$$

Третье тождество следует из первого тождества и предложения~\ref{p-mass}:
\begin{multline*}
a_1(x,t,m,\varepsilon)+m\varepsilon\, a_2(x,t,m,\varepsilon)
=\sqrt{1+m^2\varepsilon^2}\, a_1(x-\varepsilon,t+\varepsilon,m,\varepsilon)
=\\=\sqrt{1+m^2\varepsilon^2}\, a_1(\varepsilon-x,t+\varepsilon,m,\varepsilon)
=a_1(2\varepsilon-x,t,m,\varepsilon)
+m\varepsilon\, a_2(2\varepsilon-x,t,m,\varepsilon).\\[-1.5cm]
\end{multline*}
\end{proof}

Первое и третье тождества можно также доказать одновременно индукцией по $t/\varepsilon$, используя предложение~\ref{p-mass}.

\begin{figure}[htb]
  \centering
  \includegraphics[width=0.2\textwidth]{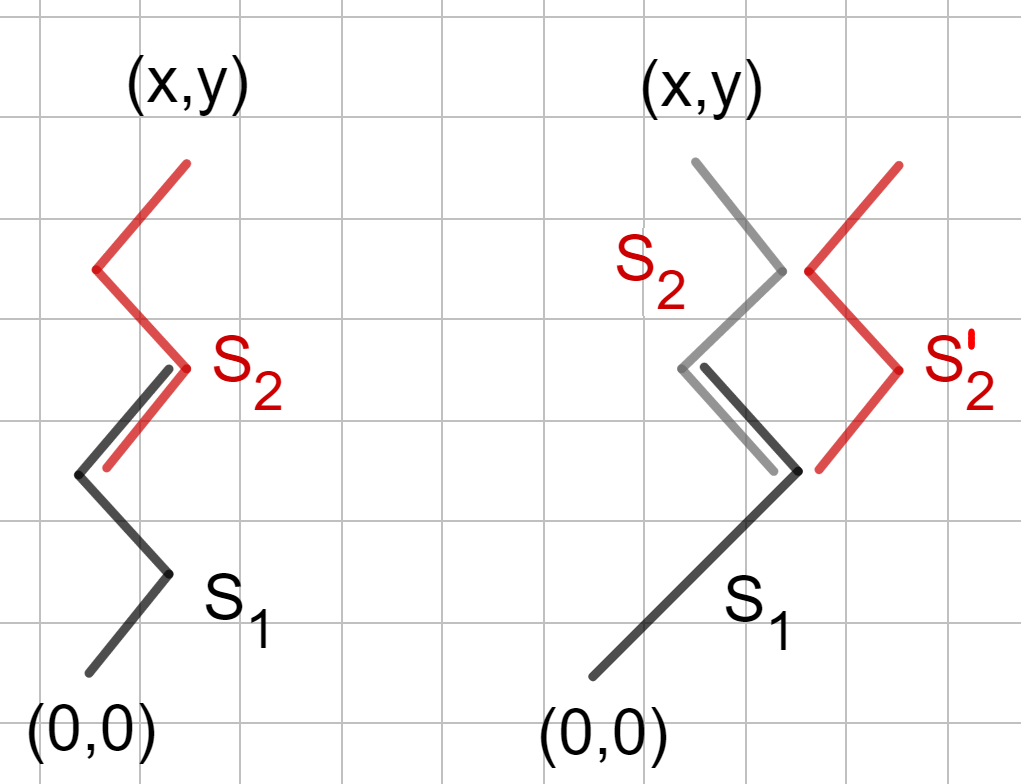}
  \vspace{-0.3cm}
  \caption{
  Путь, разделённый на две части (см. доказательство предложения~\ref{p-Huygens}). 
  }\label{fig-huygens}
\end{figure}

\begin{proof}[Доказательство предложения~\ref{p-Huygens}]
Рассмотрим путь $s$ из $(0,0)$ в $(x,t)$. Обозначим через $(x',t')$ точку, в которой $s$ пересекает прямую $t = t'$. 
Обозначим через $s_1$ часть пути $s$, которая ведёт из $(0,0)$ в $(x',t')$, а через $s_2$ часть пути $s$, которая начинается в точке пересечения пути $s$ с прямой $t = t'-\varepsilon$ и заканчивается в $(x,t)$ (см.~рис.~\ref{fig-huygens}). Сдвинем $s_2$ так, чтобы он начинался в $(0,0)$.

Положим
$a(s,m \varepsilon)
:=i(-im \varepsilon)^{\mathrm{turns}(s)}
(1+m^2\varepsilon^2)^{(1-t/\varepsilon)/2}$.
Так как $\mathrm{turns}(s)=\mathrm{turns}(s_1)+\mathrm{turns}(s_2)$, то
\begin{equation*}
\mathrm{Re}\,a(s,m\varepsilon) =
 \begin{cases}
   \mathrm{Re}\,a(s_1,m\varepsilon)
   \mathrm{Im}\,a(s_2,m\varepsilon), & \text{если ход в $(x',t')$ влево-вверх},\\
   \mathrm{Im}\,a(s_1,m\varepsilon)
   \mathrm{Re}\,a(s_2,m\varepsilon), & \text{если ход в $(x',t')$ вправо-вверх}.
 \end{cases}
\end{equation*}
В первом случае заменим путь $s_2$ путем $s_2'$, полученным отражением относительно прямой $x=0$ (и начинающимся в начале координат). Мы получаем, что  $\mathrm{Im}\,a(s_2',m\varepsilon)= \mathrm{Im}\,a(s_2,m\varepsilon)$. Значит,
\begin{multline*}
a_1(x,t,m,\varepsilon) = \sum\limits_{s}\mathrm{Re}\,a(s,m\varepsilon)
= \sum \limits_{x'} \sum \limits_{s \ni (x',t')} \mathrm{Re}\,a(s,m\varepsilon) \\
= \sum \limits_{x'} \left(\sum \limits_{s \ni (x',t'),(x'-\varepsilon,t'-\varepsilon)}
\mathrm{Im}\,a(s_1,m\varepsilon)
\mathrm{Re}\,a(s_2,m\varepsilon)
+ \sum \limits_{s \ni (x',t'),(x'+\varepsilon,t'-\varepsilon)}
\mathrm{Re}\,a(s_1,m\varepsilon)
\mathrm{Im}\,a(s_2',m\varepsilon)\right) \\
= \sum \limits_{x'} \left[ a_2(x',t',m,\varepsilon)a_1(x-x'+\varepsilon,t-t'+\varepsilon,m,\varepsilon) +a_1(x',t',m,\varepsilon)a_2(x'-x+\varepsilon,t-t'+\varepsilon,m,\varepsilon) \right].
\end{multline*}
Формула для $a_2(x,t,m,\varepsilon)$ доказывается аналогично.
\end{proof}


\comment

\begin{remark} \mscomm{!!! Notational conflict -- letter $t$ !!!}
Formulae illustrating the Huygens principle have the form of a convolution.
Define two series of Laurent polynomials: 
$$P_n=P_n(t) = \sum \limits_{m = -n+2}^{n}  a_1(m,n)t^m,\qquad Q_n=Q_n(t) = \sum \limits_{m = -n+2}^{n}  a_2(m,n)t^m.$$
These polynomials can be considered as generating functions of the sequences $a_{1,2}(m,n)$ with fixed $n$.
By Propositions~\ref{p-symmetry-mass} and~\ref{p-Huygens} we have
\begin{multline*}
a_2(x,t)  = \sum \limits_{x'} \left[ a_2(x',t')a_2(x-x'+1,t-t'+1) - a_1(x',t')a_1(x'-x+1,t-t'+1) \right] =  \\ = \sum \limits_{x'} \left[ a_2(x',t')a_2(x-x'+1,t-t'+1) - a_1(x',t')a_1(x-x'-1,t-t'+1) \right].
\end{multline*}
This formula is equivalent to the following \textit{addition theorem} for $Q$-polynomials:
\blue{$$Q_{n+m} = \frac{1}{t}Q_{n}Q_{m+1} - tP_{n}P_{m+1}.$$}
Also from the Propositions~\ref{p-symmetry-mass} and~\ref{p-Huygens} we have
\begin{multline*}
a_1(x,t)  =\sum \limits_{x'} \left[ a_2(x',t')a_1(x-x'+1,t-t'+1) + a_1(x',t')a_2(x'-x+1,t-t'+1) \right] =\\= \sum \limits_{x'} [ a_2(x',t')a_1(x-x'+1,t-t'+1) +\\+ a_1(x',t')(a_2(x-x'+1,t-t'+1)+a_1(x-x'+1,t-t'+1) - a_1(x-x'-1,t-t'+1) )].
\end{multline*}
This formula respectively is equivalent to the following addition theorem for $P$-polynomials:
\blue{$$P_{n+m} = \frac{1}{t}\left(Q_nP_{m+1} +P_nQ_{m+1}\right)+ \left( \frac{1}{t} - t\right)P_nP_{m+1}.$$}
In particular, substitution of $m=1$ gives the Dirac equation written in terms of
$P$- and $Q$-polynomials:
\begin{equation*}
\begin{cases}
P_{n+1} &=  \frac{1}{t}Q_nP_{2} + P_n \cdot (\frac{1}{t}Q_{2} + \frac{1}{t}P_{2} - tP_{2}) = \frac{1}{t\sqrt{2}}(Q_n + P_n),\\
Q_{n+1}&=  \frac{1}{t}Q_{n}Q_{2} - tP_{n}P_{2} = \frac{t}{\sqrt{2}}(Q_n - P_n).
\end{cases}
\end{equation*}
It follows immediately that
\begin{multline*}
\begin{cases}
t\sqrt{2}P_{n+1} - Q_n &= P_n = Q_n - \frac{\sqrt{2}}{t}Q_{n+1}  \\
t\sqrt{2}P_{n+1} - P_n &= Q_n = P_n + \frac{\sqrt{2}}{t}Q_{n+1}
\end{cases}
\Rightarrow
\begin{cases}
P_{n+1} &= \frac{\sqrt{2}}{t}Q_n -  \frac{Q_{n+1}}{t^2}\\
Q_{n+1}&= t^2P_{n+1} - t\sqrt{2}P_n
\end{cases}
\Rightarrow
\blue{
\begin{cases}
P_{n+1} &=  \frac{1}{\sqrt{2}}(t+\frac{1}{t})P_n  - P_{n-1} \\
Q_{n+1} &=  \frac{1}{\sqrt{2}}(t + \frac{1}{t})Q_n - Q_{n-1}.
\end{cases}}
\end{multline*}
The latter formulae give us yet another recursive definitions for $P_n$ and $Q_n$.
\end{remark}

\endcomment

\begin{proof}[Доказательство предложения~\ref{p-equal-time}]
Обозначим через $f(x,t)$ разность между левой и правой частями выражения~\eqref{eq-p-equal-time}. Введем оператор
$$[\square_m f](x,t)
:=\sqrt{1+m^2\varepsilon^2}\,f(x,t+\varepsilon)
+\sqrt{1+m^2\varepsilon^2}\,f(x,t-\varepsilon)
-f(x+\varepsilon,t)-f(x-\varepsilon,t).$$
Достаточно доказать, что
\begin{equation}\label{eq-sq5}
[\square_m^4 f](x,t)=0 \qquad \text{при }t\ge 5\varepsilon.
\end{equation}
Тогда~\eqref{eq-p-equal-time} будет следовать по индукции по~$t/\varepsilon$: \eqref{eq-sq5} выражает $f(x,t+4\varepsilon)$ как линейную комбинацию $f(x',t')$ с меньшими значениями $t'$; остается проверить, что $f(x,t)=0$ при $t\le 8\varepsilon$ (проверка сделана в~\cite[\S11]{SU-2}).

Чтобы доказать~\eqref{eq-sq5}, запишем $f(x,t)$ в виде
$$
f(x,t)=:p_1(x,t)a(x-2\varepsilon,t,m,\varepsilon)
+p_2(x,t)a(x+2\varepsilon,t,m,\varepsilon)
+p_3(x,t)a(x,t,m,\varepsilon)
$$
для некоторых кубических полиномов $p_k(x,t)$ (см.~\eqref{eq-p-equal-time}), и применим \emph{правило Лейбница}
$$
\square_m(fg)=f\cdot\square_m g
+\sqrt{1+m^2\varepsilon^2}
(\nabla_{t+}f\cdot T_{t+}g-\nabla_{t-}f\cdot T_{t-}g)
-\nabla_{x+}f\cdot T_{x+}g+\nabla_{x-}f\cdot T_{x-}g,
$$
где $[\nabla_{t\pm}f](x,t):=\pm(f(x,t\pm\varepsilon)-f(x,t))$
и $[\nabla_{x\pm}f](x,t):=\pm(f(x\pm\varepsilon,t)-f(x,t))$~---  операторы конечных разностей, $[T_{t\pm}g](x,t):=g(x,t\pm\varepsilon)$
и $[T_{x\pm}g](x,t):=g(x\pm\varepsilon,t)$~---  операторы сдвига. По предложению~\ref{p-Klein-Gordon-mass} выполняется равенство $\square_m\, a(x,t,m,\varepsilon)=0$. Каждый из операторов $\nabla_{t\pm},\nabla_{x\pm}$ уменьшает степень многочлена
, и все вышеперечисленные операторы коммутируют. Значит,~\eqref{eq-sq5} следует из правила Лейбница. Таким образом, мы доказали первое тождество из предложения; второе доказывается аналогично (база индукции проверена в~\cite[\S11]{SU-2}).
\end{proof}

Предложение~\ref{p-equal-time} можно также доказать путем семикратного применения соотношений смежности Гаусса к гипергеометрическому выражению из замечания~\ref{rem-hypergeo}.

\begin{proof}[Доказательство предложений~\ref{Feynman-binom} и~\ref{p-mass3}]
Вычислим $a_1(x,t,m,\varepsilon)$. 
Рассмотрим путь с нечётным числом поворотов; остальные не дают вклада в $a_1(x,t,m,\varepsilon)$. Пусть $2r+1$~--- количество поворотов в пути, а $R$ и $L$~--- количества ходов вправо-вверх и влево-вверх соответственно.
Обозначим через $x_1, x_2, \dots, x_{r+1}$ количества ходов вправо-вверх \emph{перед} первым, третьим, \ldots, последним поворотом соответственно. Обозначим через $y_1, y_2, \dots, y_{r+1}$ количества ходов влево-вверх \emph{после} первого, третьего, \ldots, последнего поворота соответственно. Тогда $x_k,y_k\ge 1$ для всех $1\le k\le r+1$ и
\begin{align*}
R&=x_1+\dots+x_{r+1};\\
L&=y_1+\dots+y_{r+1}.
\end{align*}
Мы свели задачу к комбинаторной: количество путей с $2r+1$ поворотами равно количеству натуральных решений двух полученных уравнений. Для первого уравнения это количество равно количеству способов поставить $r$ перегородок между $R$ шарами, то есть, $\binom{R-1}{r}$. Следовательно,
$$
a_1(x,t,m,\varepsilon) = (1+m^2\varepsilon^2)^{(1-t/\varepsilon)/2}
\sum_{r=0}^{\min\{R,L\}}(-1)^r \binom{R-1}{r}\binom{L-1}{r}(m\varepsilon)^{2r+1}.
$$
Формула~\eqref{eq1-p-mass} получается из равенств
$L+R=t/\varepsilon$ и $R-L=x/\varepsilon$. Формула~\eqref{eq2-p-mass} выводится аналогично.
\end{proof}


\begin{proof}[Доказательство предложения~\ref{cor-fourier-integral}]
Доказательство проводится индукцией по 
$t/\varepsilon$.

База $t/\varepsilon=1$ получается заменой переменной $p\mapsto p+\pi/\varepsilon$: интегралы по $[0;\pi/\varepsilon]$ и $[-\pi/\varepsilon;0]$  для нечетных $x/\varepsilon$ сокращаются и остается
 $(\varepsilon/2\pi)\int_{-\pi/\varepsilon}^{\pi/\varepsilon}e^{ip(x-\varepsilon)}=\delta_{x\varepsilon}$.

Шаг индукции заключается в следующем вычислении (и аналогичном вычислении для  $a_2(x, t+\varepsilon,m,\varepsilon)$):
\begin{multline*}
  a_1(x, t+\varepsilon,m,\varepsilon)
  = \frac{1}{\sqrt{1+m^2\varepsilon^2}} \left(
  a_1(x+\varepsilon,t,m,\varepsilon)
  + m\varepsilon \, a_2(x+\varepsilon,t,m,\varepsilon)\right)
  \\
  = \frac{m\varepsilon^2}{2\pi\sqrt{1+m^2\varepsilon^2}}
  \int_{-\pi/\varepsilon}^{\pi/\varepsilon}
  \left(
  \frac{ie^{ip\varepsilon}}
  {\sqrt{m^2\varepsilon^2+\sin^2(p\varepsilon)}}
  +1+\frac{\sin (p\varepsilon)}
  {\sqrt{m^2\varepsilon^2+\sin^2(p\varepsilon)}}\right)
  e^{i p x-i\omega_p(t-\varepsilon)}\,dp\\
  = \frac{m\varepsilon^2}{2\pi}
  \int_{-\pi/\varepsilon}^{\pi/\varepsilon}
  \frac{\left(i\cos(\omega_p\varepsilon)
  +\sin(\omega_p\varepsilon)\right)
  e^{i p x-i\omega_p(t-\varepsilon)}\,dp}
  {\sqrt{m^2\varepsilon^2+\sin^2(p\varepsilon)}}
  = \frac{im\varepsilon^2}{2\pi}
  \int_{-\pi/\varepsilon}^{\pi/\varepsilon}
  \frac{e^{i p x-i\omega_pt}\,dp}  {\sqrt{m^2\varepsilon^2+\sin^2(p\varepsilon)}}.
\end{multline*}
Здесь 1-ое равенство~--- это предложение~\ref{p-mass}. Второе~--- предположение индукции. Третье следует из соотношений
$\cos\omega_p\varepsilon=\frac{\cos p\varepsilon}{\sqrt{1+m^2\varepsilon^2}}$
и $\sin\omega_p\varepsilon
=\sqrt{1-\frac{\cos^2 p\varepsilon}{1+m^2\varepsilon^2}}=
\sqrt{\frac{m^2\varepsilon^2+\sin^2 p\varepsilon}{1+m^2\varepsilon^2}}$. 
\end{proof}

Альтернативное доказательство предложения~\ref{cor-fourier-integral} может быть получено интегрированием равенств~\eqref{eq-solution1}--\eqref{eq-solution2} по $p=2\pi/\lambda$ и использованием условий $\tilde a_1(0,0)=0$, $\tilde a_2(0,0)=1$.

\begin{proof}[Доказательство предложения~\ref{cor-double-fourier}]
Чтобы доказать формулу для $a_1(x,t,m,\varepsilon)$, возьмем интеграл по $\omega$:
\begin{multline*}
\frac{\varepsilon}{2\pi}\int\limits_{-\pi/\varepsilon}^{\pi/\varepsilon}
  \frac{e^{-i\omega(t-\varepsilon)}\,d\omega} {\sqrt{1+m^2\varepsilon^2}\cos(\omega\varepsilon)
  -\cos(p\varepsilon)-i\delta}
\overset{(*)}{=}\frac{1}{2\pi i}\oint\limits_{|z|=1}
  \frac{2\,z^{t/\varepsilon-1}\,dz} {\sqrt{1+m^2\varepsilon^2}z^2
  -2(\cos(p\varepsilon)+i\delta)z+\sqrt{1+m^2\varepsilon^2}}\\
\overset{(**)}{=}\frac{\left(\left(\cos p\varepsilon+i\delta
- i\sqrt{m^2\varepsilon^2+\sin^2p\varepsilon+\delta^2
-2i\delta\cos p\varepsilon}\right)/ \sqrt{1+m^2\varepsilon^2}\right)^{t/\varepsilon-1}}
{-i\sqrt{m^2\varepsilon^2+\sin^2p\varepsilon+\delta^2
-2i\delta\cos p\varepsilon}}
\overset{(***)}{\rightrightarrows}
\frac{i\,e^{-i\omega_p(t-\varepsilon)}}
{\sqrt{m^2\varepsilon^2+\sin^2p\varepsilon}}
\end{multline*}
при $\delta\to 0$ равномерно по $p$. Здесь мы предполагаем, что $m,t,\delta>0$ и $\delta$ достаточно мало. Равенство~$(*)$ получается заменой переменной $z=e^{-i\omega\varepsilon}$ и изменением направления обхода контура. Чтобы доказать~$(**)$, найдем корни знаменателя:
$$
z_\pm=\frac{\cos p\varepsilon+i\delta
\pm i\sqrt{m^2\varepsilon^2+\sin^2p\varepsilon+\delta^2
-2i\delta\cos p\varepsilon}}{\sqrt{1+m^2\varepsilon^2}},
$$
где $\sqrt{z}$ обозначает значение квадратного корня с положительной действительной частью. Тогда $(**)$ следует из теоремы о вычетах: разложение \new{}
$$
z_\pm=\frac{\cos p\varepsilon\pm i\sqrt{m^2\varepsilon^2+\sin^2p\varepsilon}}{\sqrt{1+m^2\varepsilon^2}}
\left(1\pm \frac{\delta}{\sqrt{m^2\varepsilon^2+\sin^2p\varepsilon}}
+{O}_{m,\varepsilon}\left(\delta^2\right)
\right)
$$
показывает, что для \new{$\delta>0$, достаточно малого в терминах $m$ и $\varepsilon$}, число $z_-$ находится внутри единичного круга, а $z_+$~--- вне его. В~$(***)$ мы полагаем
$\omega_p:=\frac{1}{\varepsilon}\arccos(\frac{\cos p\varepsilon}{\sqrt{1+m^2\varepsilon^2}})$
так, что $\sin\omega_p\varepsilon
=\sqrt{\frac{m^2\varepsilon^2+\sin^2 p\varepsilon}{1+m^2\varepsilon^2}}$
и переходим к пределу $\delta\to 0$. Равномерность сходимости по $p$ обеспечивается условием~$m>0$.

Доказанная равномерная сходимость позволяет поменять местами предел по $\delta$ с интегрированием по $p$. Это приводит к интегралу Фурье для $a_1(x,t,m,\varepsilon)$ из предложения~\ref{cor-fourier-integral}.
Формула для $a_2(x,t,m,\varepsilon)$ доказывается аналогично, при этом случай $t=\varepsilon$ необходимо рассмотреть отдельно. \end{proof}


\comment

\begin{proof}[Доказательство предложения~\ref{p-probability-conservation-external}]
Use induction over $t$.
The step of induction follows immediately from the following computation:
\begin{multline*}
\sum \limits_{x \in \mathbb{Z}} P(x,t+1,u) = \sum \limits_{x \in \mathbb{Z}} \left[ a_1(x, t+1,u)^2 + a_2(x, t+1,u)^2 \right] = \sum \limits_{x \in \mathbb{Z}} a_1(x, t+1,u)^2 + \sum \limits_{x \in \mathbb{Z}} a_2(x, t+1,u)^2 = \\ = \frac{1}{2}\sum \limits_{x \in \mathbb{Z}}u(x+\frac{1}{2},t+\frac{1}{2})^2( a_1(x+1,t,u) + a_2(x+1, t,u))^2  + \frac{1}{2}\sum \limits_{x \in \mathbb{Z}}u(x-\frac{1}{2},t+\frac{1}{2})^2( a_2(x-1,t,u) - a_1(x-1, t,u))^2 = \\ =  \sum \limits_{x \in \mathbb{Z}}\frac{( a_1(x,t,u) + a_2(x, t,u))^2}{2}  + \sum \limits_{x \in \mathbb{Z}}\frac{( a_2(x,t,u) - a_1(x, t,u))^2}{2} = \sum \limits_{x \in \mathbb{Z}} \left[a_1(x,t,u)^2 + a_2(x, t,u)^2\right] = \sum \limits_{x \in \mathbb{Z}} P(x,t,u).
\end{multline*}
\end{proof}

\endcomment

\addcontentsline{toc}{myshrinkalt}{}

\subsection{Фазовый переход: метод моментов (теорема~\ref{th-limiting-distribution})}
\label{ssec-proofs-moments}

В этом подразделе мы даем простое изложение доказательства теоремы~\ref{th-limiting-distribution} из \cite{Grimmett-Janson-Scudo-04}, используя \emph{метод моментов}. Теорема также следует из следствия~\ref{th-limiting-distribution-mass}, полученного другим путем в~\S\ref{ssec-proofs-phase}. Мы опираемся на следующий хорошо известный результат.

\begin{lemma}\label{l-method-moments}
\textup{(см.~\cite[теоремы~30.1--30.2]{Billingsley-95})}
Пусть $f_t\colon\mathbb{R}\to [0,+\infty)$, где $t=0,1,2,\dots$ --- кусочно-непрерывные функции такие, что $\alpha_{r,t}:=\int_{-\infty}^{+\infty}v^rf_t(v)\,dv$ конечен и $\alpha_{0,t}=1$ для любых $r,t=0,1,2,\dots$. Если ряд $\sum_{r=0}^{\infty}\alpha_{r,0}z^r/r!$ имеет положительный радиус сходимости и $\lim_{t\to\infty}\alpha_{r,t}=\alpha_{r,0}$ для любого $r=0,1,2,\dots$, то $f_t$ сходится к $f_0$ по распределению.
\end{lemma}

\begin{proof}[Доказательство теоремы~\ref{th-limiting-distribution}]
Докажем 
(C); тогда (A) и (B) будут следовать из леммы~\ref{l-method-moments} для $f_0(v):=F'(v)$ и \new{}$f_t(v):=tP(\lceil vt \rceil,t)$, потому что  \new{$F'(v)=0$} при $|v|>1$, а следовательно, \new{$\alpha_{r,0}\le \int_{-1}^{+1}|F'(v)|\,dv=1$}.

Перепишем предложение~\ref{cor-fourier-integral} в форме, корректной при \emph{любых} $x,t\in\mathbb{Z}$, независимо от четности: \new{}\new{} 
\begin{equation}\label{eq-rewritten-fourier}
\begin{pmatrix}
  a_1(x,t) \\
  a_2(x,t)
\end{pmatrix}
=\int_{-\pi}^{\pi}
\begin{pmatrix}
  \hat a_1(p,t) \\
  \hat a_2(p,t)
\end{pmatrix}
e^{ip(x-1)}\,\frac{dp}{2\pi}
:=\int_{-\pi}^{\pi}
\begin{pmatrix}
  \hat a_{1+}(p,t)+\hat a_{1-}(p,t) \\
  \hat a_{2+}(p,t)+\hat a_{2-}(p,t)
\end{pmatrix}
e^{ip(x-1)}\,\frac{dp}{2\pi},
\end{equation}
где
\begin{equation}\label{eq-fourier-transform}
\begin{aligned}
\hat a_{1\pm}(p,t)
&= \mp\frac{ie^{ip}}{2\sqrt{1+\sin^2 p}}
e^{\pm i\omega_p(t-1)};\\
\hat a_{2\pm}(p,t)
&= \frac{1}{2}\left(1\mp\frac{\sin p}{\sqrt{1+\sin^2 p}}\right)
e^{\pm i\omega_p(t-1)};
\end{aligned},
\end{equation}
а $\omega_p:=\arccos\frac{\cos p}{\sqrt{2}}$. Теперь~\eqref{eq-rewritten-fourier} выполнено для любых $x,t\in\mathbb{Z}$: действительно, равенство
\new{}\new{}\new{}
\begin{multline*}
\exp\left({-i\omega_{p+\pi}(t-1)+i(p+\pi)(x-1)}\right)
=\exp\left({-i(\pi-\omega_{p})(t-1)+ip(x-1)+i\pi(x-1)}\right)
=\\=(-1)^{(x+t)}\exp\left({i\omega_p(t-1)+ip(x-1)}\right)
\end{multline*}
показывает, что вклады двух слагаемых $\hat a_{k\pm}(p,t)$ в интеграл~\eqref{eq-rewritten-fourier} равны при четных $t+x$ и взаимоуничтожаются при нечетных $t+x$. Слагаемое
$\hat a_{k-}(p,t)$ дает вклад $a_k(x,t)/2$ по предложению~\ref{cor-fourier-integral}.

Используя свойство производных ряда Фурье и теорему Парсиваля, мы получаем
\begin{multline}\label{eq-moment}
\sum_{x\in\mathbb{Z}} \frac{x^r}{t^r} P(x,t)
=
\sum_{x\in\mathbb{Z}}
\begin{pmatrix}
  a_1(x,t) \\
  a_2(x,t)
\end{pmatrix}^*
\frac{x^r}{t^r}
\begin{pmatrix}
  a_1(x,t) \\
  a_2(x,t)
\end{pmatrix}
=
\int\limits_{-\pi}^{\pi}
\begin{pmatrix}
  \hat a_1(p,t) \\
  \hat a_2(p,t)
\end{pmatrix}^*
\frac{i^r}{t^r}
\frac{\partial^r}{\partial p^r}
\begin{pmatrix}
  \hat a_1(p,t) \\
  \hat a_2(p,t)
\end{pmatrix}
\,\frac{dp}{2\pi}.
\end{multline}

Производная оценивается следующим образом:
\begin{equation}\label{eq-derivatives}
\frac{\partial^r}{\partial p^r}\hat a_{k\pm}(p,t)=
\left(\pm i(t-1)\frac{\partial \omega_p}{\partial p}
\right)^r\hat a_{k\pm}(p,t)+O_r(t^{r-1})=
\left(\pm \frac{i(t-1) \sin p}{\sqrt{1+\sin^2 p}}\right)^r\hat a_{k\pm}(p,t)+O_r(t^{r-1}).
\end{equation}
Действительно, дифференцируем~\eqref{eq-fourier-transform} $r$ раз, используя правило Лейбница. Если мы каждый раз будем дифференцировать экспоненциальный множитель $e^{\pm i\omega_p(t-1)},$ то получим главный член. Если мы хотя бы раз дифференцируем множитель, отличный от экспоненциального $e^{\pm i\omega_p(t-1)}$, то мы получаем менее $r$ множителей $(t-1)$, а тогда полученный член есть $O_r(t^{r-1})$ из компактности, потому что он непрерывный и $2\pi$-периодический по $p$.

Подставляя~\eqref{eq-derivatives} в~\eqref{eq-moment},
мы получаем
\begin{multline*}
\sum_{x\in\mathbb{Z}} \frac{x^r}{t^r} P(x,t)=
\int_{-\pi}^{\pi}
\begin{pmatrix}
  \hat a_1(p,t) \\
  \hat a_2(p,t)
\end{pmatrix}^*
\left(\frac{\sin p}{\sqrt{1+\sin^2 p}}\right)^r
\begin{pmatrix}
  (-1)^r \hat a_{1+}(p,t)+\hat a_{1-}(p,t) \\
  (-1)^r \hat a_{2+}(p,t)+\hat a_{2-}(p,t)
\end{pmatrix}
\,\frac{dp}{2\pi}
+O_r\left(\frac{1}{t}\right)\\
=\int_{-\pi}^{\pi}
\left(\frac{\sin p}{\sqrt{1+\sin^2 p}}\right)^r
\frac{1}{2}\left((-1)^r
\left(1-\frac{\sin p}{\sqrt{1+\sin^2 p}}\right)
+1+\frac{\sin p}{\sqrt{1+\sin^2 p}}
\right)
\,\frac{dp}{2\pi}
+O_r\left(\frac{1}{t}\right)\\
=\int_{-\pi/2}^{\pi/2}
\left(\frac{\sin p}{\sqrt{1+\sin^2 p}}\right)^r
\left(1+\frac{\sin p}{\sqrt{1+\sin^2 p}}\right)
\,\frac{dp}{\pi}
+O_r\left(\frac{1}{t}\right)
=
\int_{-1/\sqrt{2}}^{1/\sqrt{2}}
\frac{v^r\,dv}{\pi(1-v)\sqrt{1-2v^2}}
+O_r\left(\frac{1}{t}\right).
\end{multline*}
Здесь второе равенство следует из того, что
$\hat a_{1\pm }(p,t)^*\hat a_{1\pm}(p,t)+\hat a_{2\pm }(p,t)^*\hat a_{2\pm}(p,t)
=\frac{1}{2}\left(1\mp\frac{\sin p}{\sqrt{1+\sin^2 p}}\right)$
и
$\hat a_{1\pm }(p,t)^*\hat a_{1\mp}(p,t)+\hat a_{2\pm }(p,t)^*\hat a_{2\mp}(p,t)
=0$.
Третье получается после замены переменных  $p\mapsto -p$ и $p\mapsto \pi-p$, примененной к интегралу по $[-\pi/2,\pi/2]$. Четвертое получается после замены переменных $v=\sin p/\sqrt{1+\sin^2 p}$, при которой $dp=d\arcsin\frac{v}{\sqrt{1-v^2}}=dv/(1-v^2)\sqrt{1-2v^2}$.
\end{proof}

\addcontentsline{toc}{myshrinkalt}{}

\subsection{Основной результат: метод стационарной фазы (теорема~\ref{th-ergenium})}
\label{ssec-proofs-main}

В этом разделе мы докажем теорему~\ref{th-ergenium}.
Сначала изложим план доказательства, затем докажем теорему по модулю некоторых технических лемм, и, наконец, сами леммы.

План состоит в том, чтобы применить преобразование Фурье, а затем \emph{метод стационарной фазы} к полученному осциллирующему интегралу. Доказательство состоит из $4$ шагов, первые два из которых были известны ранее.
\begin{description}
  \item[Шаг 1:] вычисление главного члена в асимптотике;
  \item[Шаг 2:] оценка остаточного члена, возникающего в окрестностях стационарных точек;
  \item[Шаг 3:] оценка остаточного члена, возникающего в окрестности начала координат;
  \item[Шаг 4:] оценка остаточного члена, возникающего в дополнениях до этих окрестностей.
\end{description}
\smallskip

\begin{proof}[Доказательство теоремы~\ref{th-ergenium} по модулю некоторых лемм]
Выведем асимптотику для\\
$a_1\left(x,t+\varepsilon,{m},{\varepsilon}\right)$; вывод для $a_2\left(x+\varepsilon,t+\varepsilon,{m},{\varepsilon}\right)$ аналогичен и обсуждается в конце доказательства.
Из предложения~\ref{cor-fourier-integral}
и тождества $e^{i\omega_{p+\pi/\varepsilon}t-i(p+\pi/\varepsilon)x}
=e^{i(\pi/\varepsilon-\omega_{p})t-ipx-i\pi x/\varepsilon}
=-e^{-i\omega_pt-ipx}$
при нечетных $(t+x)/\varepsilon$ имеем
\begin{equation}
a_1(x,t+\varepsilon,m,\varepsilon)=
  \frac{m\varepsilon^2}{2\pi i}
  \int_{-\pi/\varepsilon}^{\pi/\varepsilon}
  \frac{e^{i\omega_pt-i p x}} {\sqrt{m^2\varepsilon^2+\sin^2(p\varepsilon)}}\,dp
  =
  \int_{-\pi/2\varepsilon}^{\pi/2\varepsilon}g(p)(e(f_+(p))-e(f_-(p)))\,dp,
  \label{eq-oscillatory-integral}
\end{equation}
где $e(z):=e^{2\pi i z}$, а
\begin{align}\label{eq-parameter-values0}
f_{\pm}(p)&=
\frac{1}{2\pi}
\left(-px\pm\frac{t}{\varepsilon}
\arccos\frac{\cos(p\varepsilon)}{\sqrt{1+m^2\varepsilon^2}}\right),
\\ \label{eq-parameter-values1/2}
g(p)&=
\frac{m\varepsilon^2}{2\pi i \sqrt{m^2\varepsilon^2+\sin^2(p\varepsilon)}}.
\end{align}


\smallskip
\textbf{Шаг 1.} Оценим осциллирующий интеграл~\eqref{eq-oscillatory-integral} с помощью следующего известного результата.

\begin{lemma}[Взвешанный интеграл стационарной фазы] \label{l-weighted-stationary-phase}
\textup{\cite[лемма~5.5.6]{Huxley-96}}
Пусть $f(p)$ --- вещественно-значная функция,
четырежды непрерывно дифференцируемая при $\alpha\le p\le \beta$, пусть $g(p)$ --- вещественно-значная функция, трижды непрерывно дифференцируемая при $\alpha\le p\le \beta$. Предположим, что существуют положительные параметры $M, N, T, U$, такие, что
$$
M\ge \beta-\alpha, \qquad N\ge M/\sqrt{T},
$$
и положительные константы $C_r$ такие, что при $\alpha\le p\le \beta$,
$$
\left|f^{(r)}(p)\right|\le C_rT/M^r, \qquad \left|g^{(s)}(p)\right|\le C_sU/N^s,
$$
при $r = 2, 3, 4$ и $s = 0, 1, 2, 3$, а
$$
f'' (p) \ge T/C_2M^2.
$$
Предположим также, что $f'(p)$ меняет знак с отрицательного на положительный в точке $p = \gamma$, где $\alpha< \gamma< \beta$. Если $T$ достаточно большое в терминах $C_r$, то
\begin{multline}\label{eq-l-weighted-stationary-phase}
  \int_\alpha^\beta g(p)e(f(p))\,dp  =\frac{g(\gamma)e(f(\gamma)+1/8)}{\sqrt{f''(\gamma)}}
  +\frac{g(\beta)e(f(\beta))}{2\pi i f'(\beta)}
  -\frac{g(\alpha)e(f(\alpha))}{2\pi i f'(\alpha)}\\
  +{O}_{C_0,\dots,C_4}
  \left(
  \frac{M^4U}{T^2}\left(1+\frac{M}{N}\right)^2
  \left(\frac{1}{(\gamma-\alpha)^3}+\frac{1}{(\beta-\gamma)^3}
  +\frac{\sqrt{T}}{M^3}\right)
  \right).
\end{multline}
\end{lemma}

Здесь первый член, содержащий значения в стационарной точке $\gamma$, является главным, а граничные члены, содержащие значения в граничных точках $\alpha$ и $\beta$, сократятся на шаге~3.


\begin{lemma}\label{l-stationary-point} \textup{(ср.~\cite[(25)]{Jacobson-Schulman-84}, \cite[\S4]{Ambainis-etal-01})} Предположим~\eqref{eq-case-A};
тогда на 
$[-\frac{\pi}{2\varepsilon},\frac{\pi}{2\varepsilon}]$, функция $f_{\pm}(p)$, задаваемая~\eqref{eq-parameter-values0}, имеет единственную критическую точку
\begin{equation}\label{eq-parameter-values3}
\gamma_\pm=\pm\frac{1}{\varepsilon}\arcsin\frac{m\varepsilon x}{\sqrt{t^2-x^2}}.
\end{equation}
\end{lemma}

Для оценки интеграла~\eqref{eq-oscillatory-integral} мы дважды применим лемму~\ref{l-weighted-stationary-phase} к функциям $f(p)=\pm f_\pm(p)$ в подходящих окрестностях их критических точек $\gamma_\pm$. В случае $f(p)=-f_-(p)$ комплексно сопряжем обе части~\eqref{eq-l-weighted-stationary-phase}. Тогда общий вклад двух полученных главных членов следующий:
\begin{equation}\label{eq-main-term}
\mathrm{MainTerm}:=\frac{g(\gamma_+)e(f_{+}(\gamma_+)+1/8)}{\sqrt{f_{+}''(\gamma_+)}}
-\frac{g(\gamma_-)e(f_{-}(\gamma_-)-1/8)}{\sqrt{-f_{-}''(\gamma_-)}}.
\end{equation}
Прямое, но длинное вычисление (см. \cite[\S2]{SU-2}) дает главный член, заявленный в теореме:

\begin{lemma}\label{l-main-term} \textup{(см.~\cite[\S2]{SU-2})} Предположим~\eqref{eq-case-A}, \eqref{eq-theta}, \eqref{eq-parameter-values0}--\eqref{eq-parameter-values1/2}, \eqref{eq-parameter-values3}; тогда выражение~\eqref{eq-main-term} равняется
\begin{equation*}
\mathrm{MainTerm}=
\varepsilon\sqrt{\frac{2 m}{\pi}}
\left(t^2-(1+m^2\varepsilon^2)x^2\right)^{-1/4}
\sin\theta(x,t,m,\varepsilon).
\end{equation*}
\end{lemma}

\smallskip
\textbf{Шаг 2}. Для оценки остаточного члена нам нужно указать конкретные значения параметров, к которым применяется лемма~\ref{l-weighted-stationary-phase}:
\begin{align}\label{eq-parameter-values1}
M&=N=m, & T&=mt, & U&=\varepsilon.
\end{align}

\begin{lemma}\label{l-checking-the-assumptions}
Если $\varepsilon\le 1/m$, то функции~\eqref{eq-parameter-values0}--\eqref{eq-parameter-values1/2}
и параметры~\eqref{eq-parameter-values1}
удовлетворяют неравенствам
$$
\left|f_\pm^{(r)}(p)\right|\le 3T/M^r, \qquad \left|g^{(s)}(p)\right|\le 3U/N^s \qquad\text{при $p\in\mathbb{R}$, \ $r = 2, 3, 4$, \ $s = 0, 1, 2, 3$.}
$$
\end{lemma}

Нужно также указать отрезок
\begin{equation}\label{eq-parameter-values4}
[\alpha_\pm,\beta_\pm]:=[\gamma_\pm-m\delta/2,\gamma_\pm+m\delta/2].
\end{equation}
Для оценки производной $|f_{\pm}''(p)|$ снизу убедимся, что находимся далеко от её корней $\pm\pi/2\varepsilon$.

\begin{lemma}\label{l-interval}
Предположим~\eqref{eq-case-A} и \eqref{eq-parameter-values3}; тогда
отрезок~\eqref{eq-parameter-values4} содержится в отрезке $[-\pi/2\varepsilon+m\delta/2,\pi/2\varepsilon-m\delta/2]$.
\end{lemma}

Такой предусмотрительный выбор отрезка дает нам следующую, более техничную оценку.

\begin{lemma}\label{l-second-derivative}
Предположим~\eqref{eq-case-A}, \eqref{eq-parameter-values0}, \eqref{eq-parameter-values3} и \eqref{eq-parameter-values4}. Тогда для любого $p\in[\alpha_\pm,\beta_\pm]$  имеем
$$
|f_\pm''(p)|\ge \frac{t\delta^{3/2}}{24\pi m}.
$$
\end{lemma}

Это дает $|f_\pm''(p)|\ge T/C_2M^2$ при $C_2:=24\pi\delta^{-3/2}$ в обозначениях~\eqref{eq-parameter-values1}.
Теперь все предположения леммы~\ref{l-weighted-stationary-phase} проверены ($M\ge \beta_\pm-\alpha_\pm$ и $N\ge M/\sqrt{T}$ автоматически, так как $\delta\le 1$ и $t>C_\delta/m$ из~\eqref{eq-case-A}).
Применим лемму к $g(p)$ и $\pm f_\pm(p)$ на $[\alpha_\pm,\beta_\pm]$ (знак минус перед $f_-(p)$ гарантирует неравенство $f''(p)>0$, а  множитель $i$ внутри $g(p)$ несущественнен при применении леммы). Получаем следующую оценку остаточного члена на этих отрезках.

\begin{lemma}\label{l-error-stationary} \textup{(см.~\cite[\S4]{SU-2})}
Параметры~\eqref{eq-parameter-values3} и \eqref{eq-parameter-values1}--\eqref{eq-parameter-values4} удовлетворяют
$$
\frac{M^4U}{T^2}\left(1+\frac{M}{N}\right)^2
  \left(\frac{1}{(\gamma_\pm-\alpha_\pm)^3}
  +\frac{1}{(\beta_\pm-\gamma_\pm)^3}
  +\frac{\sqrt{T}}{M^3}\right)
={O}_\delta\left(\frac{\varepsilon}{m^{1/2}t^{3/2}}\right).
$$
\end{lemma}

Хотя это только часть остаточного члена, она уже того же порядка, что и в теореме. 

\smallskip
\textbf{Шаг 3.}
Для оценки остаточного члена вне $[\alpha_\pm,\beta_\pm]$ используем другой известный технический результат.

\begin{lemma}[Взвешенный тест первой производной] \label{l-weighted-first-derivative-test}
\textup{\cite[лемма~5.5.5]{Huxley-96}}
Пусть $f(p)$ ---вещественно-значная функция,
трижды непрерывно дифференцируемая при $\alpha\le p\le \beta$, пусть $g(p)$ --- вещественно-значная
функция, дважды непрерывно дифференцируемая при $\alpha\le p\le \beta$. Предположим, что существуют положительные параметры $M, N, T, U$, такие, что
$
M\ge \beta-\alpha,
$
и положительные константы $C_r$ такие, что при $\alpha\le p\le \beta$ выполнено
$$
\left|f^{(r)}(p)\right|\le C_rT/M^r, \qquad \left|g^{(s)}(p)\right|\le C_sU/N^s,
$$
при $r = 2, 3$ и $s = 0, 1, 2$.
Если $f'(p)$ и $f''(p)$ не меняют знак на отрезке $[\alpha,\beta]$, то
\begin{multline*}
\int_\alpha^\beta g(p)e(f(p))\,dp  =\frac{g(\beta)e(f(\beta))}{2\pi i f'(\beta)}
  -\frac{g(\alpha)e(f(\alpha))}{2\pi i f'(\alpha)}\\
  +{O}_{C_0,\dots,C_3}
  \left(
  \frac{TU}{M^2}\left(1+\frac{M}{N}+\frac{M^3\min|f'(p)|}{N^2T}\right)
  \frac{1}{\min|f'(p)|^3}
  \right).
\end{multline*}
\end{lemma}


Эта лемма, в частности, требует, чтобы интервал был достаточно малым. По этой причине мы нарезаем исходный отрезок $[-\pi/2\varepsilon,\pi/2\varepsilon]$ на большое количество отрезков с концами
$$\hspace{-1cm}-\frac{\pi}{2\varepsilon}=
\alpha_{-K}<\beta_{-K}=\alpha_{-K+1}<\beta_{-K+1}
=\alpha_{-K+2}<\dots =\alpha_{i}<\hat\beta_{i}=\alpha_\pm<\beta_\pm=\hat\alpha_j<\beta_j=\dots<
\beta_{K-1}=\frac{\pi}{2\varepsilon}.
$$
Здесь $\alpha_\pm$ и $\beta_\pm$ задаются формулой~\eqref{eq-parameter-values4} выше. Остальные точки задаются так:
\begin{equation}\label{eq-other-points}
\alpha_k=\frac{k\pi}{2\varepsilon K}, \quad \beta_k=\frac{(k+1)\pi}{2\varepsilon K}, \quad
\text{где}\quad K=2\left\lceil\frac{\pi}{m\varepsilon}\right\rceil,
\quad\text{а}\quad k=-K,\dots,i,j+1,\dots,K-1.
\end{equation}
Здесь индексы $i$ и $j$~--- минимальные, для которых выполнено $\frac{(i+1)\pi}{2\varepsilon K}>\alpha_\pm$ и $\frac{(j+1)\pi}{2\varepsilon K}>\beta_\pm$.
Следовательно, все полученные отрезки, кроме $[\alpha_\pm,\beta_\pm]$ и его соседей, имеют одинаковую длину $\frac{\pi}{2\varepsilon K}$. (Хотя более концептуальной является нарезка с использованием геометрической прогрессии, а не арифметической, это не влияет на окончательную оценку.)

Мы уже применили лемму~\ref{l-weighted-stationary-phase} к $[\alpha_\pm,\beta_\pm]$, и теперь мы применим лемму~\ref{l-weighted-first-derivative-test} к каждому из оставшихся отрезков в нарезке для $f(p)= f_\pm(p)$ (на этот раз нет необходимости в смене знака у $f_-(p)$). После суммирования полученных оценок все члены, содержащие значения $f_\pm(\alpha_k)$ и $f_\pm(\beta_k)$ в концевых точках, кроме самой левой и самой правой, сокращаются.
Оставшиеся граничные члены дают
\begin{equation}\label{eq-boundary-term}
\mathrm{BoundaryTerm}:=
\frac{g(\frac{\pi}{2\varepsilon})e(f_+(\frac{\pi}{2\varepsilon}))}
{2\pi i f'_+(\frac{\pi}{2\varepsilon})}
-\frac{g(-\frac{\pi}{2\varepsilon})e(f_+(-\frac{\pi}{2\varepsilon}))}
{2\pi i f'_+(-\frac{\pi}{2\varepsilon})}
-\frac{g(\frac{\pi}{2\varepsilon})e(f_-(\frac{\pi}{2\varepsilon}))}
{2\pi i f'_-(\frac{\pi}{2\varepsilon})}
+\frac{g(-\frac{\pi}{2\varepsilon})e(f_-(-\frac{\pi}{2\varepsilon}))}
{2\pi i f'_-(-\frac{\pi}{2\varepsilon})}.
\end{equation}

\begin{lemma}\label{l-boundary-term} \textup{(см.~\cite[\S5]{SU-2})} Для $(x,t)\in\varepsilon\mathbb{Z}^2$ таких, что $(x+t)/\varepsilon$ нечетно, выражение~\eqref{eq-boundary-term} зануляется.
\end{lemma}

Остается оценить остаточные члены.
Начнем оценку с центральных отрезков $[\alpha_0,\beta_0]$ и $[\alpha_{-1},\beta_{-1}]$
(возможно, без частей, вырезанных отрезком $[\alpha_\pm,\beta_\pm]$); к ним нужен специальный подход. Применим лемму~\ref{l-weighted-first-derivative-test} к этим отрезкам для тех же функций~\eqref{eq-parameter-values0}--\eqref{eq-parameter-values1/2} и тех же значений~\eqref{eq-parameter-values1} параметров $M,N,T,U$, как в шаге 2. Все предположения этой леммы уже были проверены в лемме~\ref{l-checking-the-assumptions}; имеем также
$\beta_0-\alpha_0\le {\pi}/{2\varepsilon K}={\pi}/4\varepsilon \lceil\frac{\pi}{m\varepsilon}\rceil<m=M$ и $f''_\pm(p)\ne 0$.
Следовательно, остается оценить $|f_\pm'(p)|$ снизу.

\begin{lemma}\label{l-first-derivative}
Предположим~\eqref{eq-case-A}, \eqref{eq-parameter-values0}, \eqref{eq-parameter-values3}, \eqref{eq-parameter-values4}; тогда для любого
$p\in [-\pi /2\varepsilon,\pi/2\varepsilon] \setminus[\alpha_\pm,\beta_\pm]$
получаем $$|f_{\pm}'(p)|\ge t\delta^{5/2}/48\pi.$$
\end{lemma}

Тогда остаточный член на центральных отрезках оценивается так.

\begin{lemma}\label{l-error-central} \textup{(см.~\cite[\S6]{SU-2})}
Предположим~\eqref{eq-case-A}, \eqref{eq-parameter-values3} и \eqref{eq-parameter-values4}. Тогда параметры~\eqref{eq-parameter-values1} и функции~\eqref{eq-parameter-values0} удовлетворяют
$$
\frac{TU}{M^2}\left(1+\frac{M}{N}
+\frac{M^3\min_{p\in[\alpha_\pm,\beta_\pm]}|f_\pm'(p)|}{N^2T}\right)
  \frac{1}{\min_{p\in[\alpha_\pm,\beta_\pm]}|f_\pm'(p)|^3}
={O}\left(\frac{\varepsilon}{mt^{2}\delta^{15/2}}\right).
$$
\end{lemma}

Это значение есть ${O}_\delta\left({\varepsilon}/{m^{1/2}t^{3/2}}\right)$, так как $t>C_\delta/m$ из~\eqref{eq-case-A}. Следовательно, остаточный член на центральных отрезках находится в допустимых теоремой пределах.

\bigskip\textbf{Шаг 4.} Для оценки остаточного члена на оставшихся отрезках $[\alpha_k,\beta_k]$, где без потери общности считаем $k>0$, применим лемму~\ref{l-weighted-first-derivative-test} с несколько другими параметрами:
\begin{align}\label{eq-parameter-values-alt}
T&=mt/k, & M&=mk, & U&=\varepsilon/k, & N&=mk. 
\end{align}

\begin{lemma}\label{l-checking-the-assumptions-alt}
Для $0<k<K$ и $\varepsilon\le 1/m$ параметры~\eqref{eq-parameter-values-alt} и \eqref{eq-other-points},  функции~\eqref{eq-parameter-values0}--\eqref{eq-parameter-values1/2} на $[\alpha_k,\beta_k]$ удовлетворяют всем условиям леммы~\ref{l-weighted-first-derivative-test}, кроме, быть может, условия на знак~$f'(p)$.
\end{lemma}

Поскольку окрестность $[\alpha_\pm,\beta_\pm]$ корня производной $f'(p)$ вырезана, получаем, что  $f'(p)$ имеет постоянный знак на оставшихся промежутках и по лемме~\ref{l-weighted-first-derivative-test} их вклад в ошибку оценивается следующим образом.

\begin{lemma}\label{l-error-alt} \textup{(см.~\cite[\S7]{SU-2})}
Пусть~\eqref{eq-case-A}, \eqref{eq-parameter-values3}, \eqref{eq-parameter-values4}, $0<k<K$. Тогда
функции~\eqref{eq-parameter-values0} и параметры~\eqref{eq-parameter-values-alt} удовлетворяют
$$
\frac{TU}{M^2}\left(1+\frac{M}{N}
+\frac{M^3\min_{p\not\in[\alpha_\pm,\beta_\pm]}|f_\pm'(p)|}{N^2T}\right)
  \frac{1}{\min_{p\not\in[\alpha_\pm,\beta_\pm]}|f_\pm'(p)|^3}
={O}\left(\frac{\varepsilon}{k^2mt^{2}\delta^{15/2}}\right).
$$
\end{lemma}

Суммирование по всем $k$ дает остаточный член
$$
\sum_{k=1}^{K}
{O}\left(\frac{\varepsilon}{k^2mt^2\delta^{15/2}}\right)=
{O}\left(\frac{\varepsilon}{mt^2\delta^{15/2}}
\sum_{k=1}^{\infty}\frac{1}{k^2}\right)
={O}_\delta\left(\frac{\varepsilon}{m^{1/2}t^{3/2}}\right),
$$
так как ряд внутри O-большого сходится и $t>C_\delta/m$. Следовательно, итоговый остаточный член на всех отрезках удовлетворяет требованию теоремы, что завершает доказательство асимптотики~\eqref{eq-ergenium-re}.

Вывод асимптотики для $a_2\left(x+\varepsilon,t+\varepsilon,{m},{\varepsilon}\right)$ аналогичен. Из предложения~\ref{cor-fourier-integral} для четных $(x+t)/\varepsilon$ получаем 
\begin{multline}\label{eq-oscillatory-integral-a2}
  a_2\left(x+\varepsilon,t+\varepsilon,{m},{\varepsilon}\right)=
  \frac{\varepsilon}{2\pi}\int_{-\pi/\varepsilon}^{\pi/\varepsilon}
  \left(1+\frac{\sin (p\varepsilon)} {\sqrt{m^2\varepsilon^2+\sin^2(p\varepsilon)}}\right)
  e^{i\omega_pt-i p x}\,dp\\
  =
  \int_{-\pi/2\varepsilon}^{\pi/2\varepsilon}
  \left[g_+(p)e(f_+(p))+g_-(p)e(f_-(p))\right]dp,
\end{multline}
где $f_\pm(p)$ такие же, как выше (см.~\eqref{eq-parameter-values0}), а
\begin{align}\label{eq-parameter-values7}
g_\pm(p)&=\frac{\varepsilon}{2\pi}\left(1\pm
\frac{\sin (p\varepsilon)} {\sqrt{m^2\varepsilon^2+\sin^2(p\varepsilon)}}\right).
\end{align}
Повторим рассуждения шагов~1--4 с заменой $g(p)$ на $g_\pm(p)$. Конкретная форма $g(p)$ использовалась только в леммах~\ref{l-main-term}, \ref{l-checking-the-assumptions}, \ref{l-boundary-term}, \ref{l-checking-the-assumptions-alt}. Аналоги лемм~\ref{l-main-term} и~\ref{l-boundary-term} для $g_\pm(p)$ получаются прямой проверкой \cite[\S13]{SU-2}. Лемма~\ref{l-checking-the-assumptions} справедлива для $g_\pm(p)$: доказательство повторять не нужно, так как $g_\pm(p)=(\varepsilon/t)(f_\pm'(p)+(x+t)/2\pi)$ \cite[\S1]{SU-2}.
Однако параметры~\eqref{eq-parameter-values-alt} и лемму~\ref{l-checking-the-assumptions-alt} нужно заменить следующими (тогда аналог леммы~\ref{l-error-alt} 
выполнен):
\begin{align}\label{eq-parameter-values-alt-a2}
T&=mt/k, & M&=mk, & U&=\varepsilon, & N&=mk^{3/2}. 
\end{align}

\begin{lemma}\label{l-checking-the-assumptions-alt-a2}
Для $0<k<K$ и $\varepsilon\le 1/m$ параметры~\eqref{eq-parameter-values-alt-a2} и \eqref{eq-other-points},  функции~\eqref{eq-parameter-values0} и~\eqref{eq-parameter-values7} на $[\alpha_k,\beta_k]$ удовлетворяют условиям леммы~\ref{l-weighted-first-derivative-test}, кроме, быть может, условия на знак~$f'(p)$.
\end{lemma}

\new{И на этот раз, доказательство повторять не нужно:
лемма~\ref{l-checking-the-assumptions-alt-a2} следует из леммы~\ref{l-checking-the-assumptions-alt} и выражения $g_\pm$ через~$f_\pm'$.}

Это завершает доказательство теоремы~\ref{th-ergenium} по модулю лемм.
\end{proof}

\bigskip
Теперь докажем леммы. Леммы~\ref{l-main-term}, \ref{l-error-stationary}, \ref{l-boundary-term}, \ref{l-error-central}, 
\ref{l-error-alt} доказываются прямой проверкой
~\cite{SU-2}. Следующие выражения \cite[\S1,3]{SU-2} часто используются в доказательствах других лемм:
\begin{align}\label{eq-f'}
f_\pm'(p)&=\frac{1}{2\pi}\left(-x\pm\frac{t \sin  p\varepsilon} {\sqrt{ m^2\varepsilon^2+\sin^2 p\varepsilon}}\right);\\
\label{eq-f''}
f_\pm''(p)&=
\pm\frac{ m^2 \varepsilon^3 t \cos (p \varepsilon)}
{2 \pi \left(m^2\varepsilon^2+\sin^2(p\varepsilon)\right)^{3/2}}
\end{align}

\begin{proof}[Доказательство леммы~\ref{l-stationary-point}]
Используя~\eqref{eq-f'} и решая квадратное уравнение $f_\pm'(p)=0$, как уравнение от $\sin p\varepsilon$, получаем~\eqref{eq-parameter-values3}. Условия~$|x|/t<1/\sqrt{1+m^2\varepsilon^2}$ из~\eqref{eq-case-A} гарантируют существование арксинуса.
\end{proof}


\begin{proof}[Доказательство леммы~\ref{l-checking-the-assumptions}]
Из вычисления производных в~\cite[\S3]{SU-2}
и условия $m\varepsilon\le 1$
получаем
\begin{align*}
|g(p)|&=\frac{m\varepsilon^2}
{2\pi\sqrt{m^2\varepsilon^2+\sin^2(p\varepsilon)}}
\le\frac{m\varepsilon^2}
{2\pi\sqrt{m^2\varepsilon^2+0}}
\le \varepsilon=U, \\
|g^{(1)}(p)|&=\frac{m\varepsilon^3
|\sin(p\varepsilon)\cos(p\varepsilon)|}
{2\pi\left(m^2\varepsilon^2+\sin^2(p\varepsilon)\right)^{3/2}}
\le\frac{m\varepsilon^3|\sin(p\varepsilon)\cos(p\varepsilon)|}
{2\pi\left(m^2\varepsilon^2+0\right)
\left(0+\sin^2(p\varepsilon)\right)^{1/2}}
\le \frac{\varepsilon}{m}=\frac{U}{N}, \\
|g^{(2)}(p)|&=\frac{m\varepsilon^4
\left|m^2\varepsilon^2+\sin^4(p \varepsilon) -2(1+m^2\varepsilon^2)\sin^2(p \varepsilon)\right|}
{2\pi\left(m^2\varepsilon^2+\sin^2(p\varepsilon)\right)^{5/2}}
\le \frac{m\varepsilon^4
\left(3m^2\varepsilon^2+3\sin^2(p \varepsilon) 
\right)}
{2\pi\left(m^2\varepsilon^2+\sin^2(p\varepsilon)\right)^{5/2}}
\le \frac{\varepsilon}{m^2}=\frac{U}{N^2}, \\
|g^{(3)}(p)|&=\frac{m\varepsilon^5
|\sin(p\varepsilon)\cos(p\varepsilon)|\cdot
\left|4m^4\varepsilon^4+9m^2\varepsilon^2+\sin^4(p \varepsilon)
-(6+10m^2\varepsilon^2)\sin^2(p \varepsilon)\right|}
{2\pi\left(m^2\varepsilon^2+\sin^2(p\varepsilon)\right)^{7/2}}
\le\frac{3\varepsilon}{m^3}=\frac{3U}{N^3},
\\
|f^{(2)}_\pm(p)|&=
\frac{ m^2 \varepsilon^3 t \cos (p \varepsilon)}
{2 \pi \left(m^2\varepsilon^2+\sin^2(p\varepsilon)\right)^{3/2}}
\le \frac{t}{m}=\frac{T}{M^2}, \\
|f^{(3)}_\pm(p)|&=\frac{m^2 \varepsilon^4 t |\sin (p \varepsilon)| \left(m^2\varepsilon^2+\cos (2 p \varepsilon)+2\right)}
{2 \pi\left(m^2\varepsilon^2+\sin^2(p\varepsilon)\right)^{5/2}}
\le \frac{4m^2 \varepsilon^4 t |\sin (p \varepsilon)| }
{2 \pi\left(m^2\varepsilon^2+\sin^2(p\varepsilon)\right)^{5/2}}
\le \frac{t}{m^2}=\frac{T}{M^3},\\
|f^{(4)}_\pm(p)|&=\frac{m^2\varepsilon^5t\cos(p\varepsilon) \left|m^4\varepsilon^4+3m^2\varepsilon^2
+4\sin^4(p \varepsilon)-2(6+5m^2\varepsilon^2)\sin^2(p \varepsilon)\right|}
{2 \pi\left(m^2\varepsilon^2+\sin^2(p\varepsilon)\right)^{7/2}}
\le\frac{3t}{m^3}=\frac{3T}{M^4}.\\[-1.4cm]
\end{align*}
\end{proof}

\begin{proof}[Доказательство леммы~\ref{l-interval}]
Лемма следует из следующей цепочки оценок:
\begin{multline*}
\frac{\pi}{2\varepsilon}-|\gamma_\pm|
= \frac{\sin(\pi/2)-\sin|\gamma_\pm\varepsilon|}
{\varepsilon\cos(\theta\varepsilon)}
\ge \frac{\sin(\pi/2)-\sin|\gamma_\pm\varepsilon|}
{\varepsilon\cos(\gamma_\pm\varepsilon)}
= \frac{1-m\varepsilon |x|/\sqrt{t^2-x^2}} {\varepsilon\sqrt{1-m^2\varepsilon^2 x^2/(t^2-x^2)}}\\
= \frac{\sqrt{1-x^2/t^2}-m\varepsilon |x|/t} {\varepsilon\sqrt{1-(1+m^2\varepsilon^2) x^2/t^2}}
\ge\frac{1}{\varepsilon}\left(\sqrt{1-\frac{x^2}{t^2}}- \frac{m\varepsilon |x|}{t}\right)\\
\ge
\frac{1}{\varepsilon}\left(\sqrt{1-\frac{1}{1+m^2\varepsilon^2}}
-\frac{m\varepsilon |x|}{t}\right)
= m \left(\frac{1}{\sqrt{1+m^2\varepsilon^2}}-\frac{|x|}{t}\right)
\ge {m\delta}.
\end{multline*}
Здесь первое равенство выполнено для некоторого $\theta\in[|\gamma_\pm|,\pi/2\varepsilon]$ по теореме Лагранжа. Следующее неравенство выполнено, так как косинус убывает на заданном промежутке. Следующее получается путем подстановки~\eqref{eq-parameter-values3}. Остальное просто, так как $|x|/t<1/\sqrt{1+m^2\varepsilon^2}-\delta$ ввиду~\eqref{eq-case-A}.
\end{proof}

\begin{proof}[Доказательство леммы~\ref{l-second-derivative}]
Докажем лемму для $f_+(p)$ и $\gamma_+\ge 0$; для остальныx знаков --- аналогично. Опустим индекс $+$ в обозначениях $f_+,\alpha_+,\beta_+,\gamma_+$. Лемма следует из следующей цепочки неравенств:
\begin{align*}
|f^{(2)}(p)|&\overset{(*)}{\ge} |f^{(2)}(\beta)|=
\frac{ m^2 \varepsilon^3 t \cos (\beta \varepsilon)}
{2 \pi \left(m^2\varepsilon^2+\sin^2(\beta\varepsilon)\right)^{3/2}}
\overset{(**)}{\ge}
\frac{ m^2 \varepsilon^3 t \cos (\gamma \varepsilon)}
{4 \pi \left(m^2\varepsilon^2+\sin^2(\gamma\varepsilon)
+ 2m^2\varepsilon^2t^2/(t^2-x^2)\right)^{3/2}}\\
&\overset{(***)}{=}
\frac{ m^2 \varepsilon^3 t \sqrt{t^2-(1+m^2\varepsilon^2)x^2}(t^2-x^2)}
{4 \pi \left(3m^2\varepsilon^2t^2\right)^{3/2}}
\ge
\frac{t\sqrt{1-(1+m^2\varepsilon^2)x^2/t^2}(1-x^2/t^2)}
{24 \pi m }
\ge \frac{t\delta^{3/2}}{24 \pi m }.
\end{align*}

Здесь неравенство $(*)$ доказывается следующим образом. Из~\eqref{eq-f''},  $f^{(2)}(p)$ возрастает на $[-\pi/2\varepsilon,0]$ и убывает на $[0,\pi/2\varepsilon]$: ведь эта функция четна, а её числитель убывает на $[0,\pi/2\varepsilon]$ и знаменатель возрастает на $[0,\pi/2\varepsilon]$.
Следовательно, $|f^{(2)}(p)|\ge \min\{|f^{(2)}(\beta)|,|f^{(2)}(\alpha)|\}$ при $p\in[\alpha,\beta]$ из леммы~\ref{l-interval}. Но так как $f^{(2)}(p)$ четно и $\gamma\ge 0$, то из~\eqref{eq-parameter-values4} получаем
$$
|f^{(2)}(\alpha)|=|f^{(2)}(\gamma-m\delta/2)|=
|f^{(2)}(|\gamma-m\delta/2|)|\ge |f^{(2)}(\gamma+m\delta/2)|
=|f^{(2)}(\beta)|.
$$

Неравенство $(**)$ следует из следующих двух оценок. Сначала из леммы~\ref{l-interval} и выпуклости косинуса на промежутке $[\gamma\varepsilon,\pi/2]$ получаем
$$
\cos(\beta\varepsilon)
\ge \cos\left(\frac{\gamma\varepsilon}{2}+\frac{\pi}{4}\right)
\ge \frac{1}{2}
\left(\cos (\gamma\varepsilon)+\cos\frac{\pi}{2}\right)
=\frac{\cos (\gamma\varepsilon)}{2}.
$$
Затем, используя неравентсво $\sin z-\sin w\le z-w$ при $0\le w\le z\le \pi/2$, а потом неравенство  $\delta\le 1$ и~\eqref{eq-parameter-values3}--\eqref{eq-parameter-values4}, получаем
\begin{multline*}
  \sin^2(\beta\varepsilon)-\sin^2(\gamma\varepsilon)
  \le
  \varepsilon(\beta-\gamma)
  \left(\sin(\beta\varepsilon)+\sin(\gamma\varepsilon)\right)
  \le \varepsilon(\beta-\gamma)
  \left(\varepsilon(\beta-\gamma)+2\sin(\gamma\varepsilon)\right)\\
  = \frac{m\varepsilon\delta}{2}
  \left(\frac{m\varepsilon\delta}{2}
  +\frac{2m\varepsilon x}{\sqrt{t^2-x^2}}\right)
  \le \frac{m\varepsilon t}{2\sqrt{t^2-x^2}}
  \left(\frac{2m\varepsilon t}{\sqrt{t^2-x^2}}
  +\frac{2m\varepsilon t}{\sqrt{t^2-x^2}}\right)
  =\frac{2m^2\varepsilon^2t^2}{t^2-x^2}.
\end{multline*}

Равенство $(***)$ получается из~\eqref{eq-parameter-values3}. Оставшиеся оценки очевидны.
\end{proof}

\begin{proof}[Доказательство леммы~\ref{l-first-derivative}]
Из лемм~\ref{l-stationary-point} и~\ref{l-second-derivative} при $p\in [\beta_+,\pi/2\varepsilon]$ имеем
$$
f_+'(p)=f_+'(\gamma_+)+\int_{\gamma_+}^{p}f_+''(p)\,dp\ge 0+(p-\gamma_+)\frac{t\delta^{3/2}}{24 \pi m }
\ge (\beta_+-\gamma_+)\frac{t\delta^{3/2}}{24 \pi m }
=\frac{t\delta^{5/2}}{48 \pi},
$$
так как
$
f_+''(p)\ge 0
$
из~\eqref{eq-f''}.
Для $p\in [-\pi/2\varepsilon,\alpha_+]$ и $f'_-(p)$ доказательство аналогично.
\end{proof}

\begin{proof}[Доказательство леммы~\ref{l-checking-the-assumptions-alt}] Возьмем $p\in [\alpha_k,\beta_k]$.
Из~\eqref{eq-other-points}, неравенства $\sin z\ge z/2$ при $z\in [0,\pi/2]$ и неравенства $m\varepsilon\le 1$ получаем
\begin{align*}
|g(p)|&=\frac{m\varepsilon^2}
{2\pi\sqrt{m^2\varepsilon^2+\sin^2(p\varepsilon)}}
\le\frac{m\varepsilon^2}{2\pi\sin(p\varepsilon)}
\le\frac{m\varepsilon^2}{\pi p\varepsilon}
\le\frac{2m\varepsilon^2 K}{\pi^2 k}
=\frac{4m\varepsilon^2}{\pi^2 k} \left\lceil\frac{\pi}{m\varepsilon}\right\rceil
={O}\left(\frac{\varepsilon}{k}\right)
={O}\left({U}\right), \\
|g^{(1)}(p)|&=\frac{m\varepsilon^3
|\sin(p\varepsilon)\cos(p\varepsilon)|}
{2\pi\left(m^2\varepsilon^2+\sin^2(p\varepsilon)\right)^{3/2}}
\le\frac{m\varepsilon^3}{2\pi\sin^2(p\varepsilon)}
={O}\left(\frac{\varepsilon}{mk^2}\right)
={O}\left(\frac{U}{N}\right), \\
|g^{(2)}(p)|&=\frac{m\varepsilon^4
\left|m^2\varepsilon^2+\sin^4(p \varepsilon) -2(1+m^2\varepsilon^2)\sin^2(p \varepsilon)\right|}
{2\pi\left(m^2\varepsilon^2+\sin^2(p\varepsilon)\right)^{5/2}}
\le \frac{m\varepsilon^4
\left(3m^2\varepsilon^2+3\sin^2(p \varepsilon) 
\right)}
{2\pi\left(m^2\varepsilon^2+\sin^2(p\varepsilon)\right)^{5/2}}
={O}\left(\frac{U}{N^2}\right), \\
|g^{(3)}(p)|&=\frac{m\varepsilon^5
|\sin(p\varepsilon)\cos(p\varepsilon)|\cdot
\left|4m^4\varepsilon^4+9m^2\varepsilon^2+\sin^4(p \varepsilon)
-(6+10m^2\varepsilon^2)\sin^2(p \varepsilon)\right|}
{2\pi\left(m^2\varepsilon^2+\sin^2(p\varepsilon)\right)^{7/2}}
={O}\left(\frac{U}{N^3}\right),
\\
|f^{(2)}_\pm(p)|&=
\frac{ m^2 \varepsilon^3 t \cos (p \varepsilon)}
{2 \pi \left(m^2\varepsilon^2+\sin^2(p\varepsilon)\right)^{3/2}}
={O}\left(\frac{t}{mk^3}\right)
={O}\left(\frac{T}{M^2}\right), \\
|f^{(3)}_\pm(p)|
&=\frac{m^2 \varepsilon^4 t |\sin (p \varepsilon)| \left(m^2\varepsilon^2+\cos (2 p \varepsilon)+2\right)}
{2 \pi\left(m^2\varepsilon^2+\sin^2(p\varepsilon)\right)^{5/2}}
={O}\left(\frac{t}{m^2k^4}\right)
={O}\left(\frac{T}{M^3}\right),\\
|f^{(4)}_\pm(p)|&=\frac{m^2\varepsilon^5t\cos(p\varepsilon) \left|m^4\varepsilon^4+3m^2\varepsilon^2
+4\sin^4(p \varepsilon)-2(6+5m^2\varepsilon^2)\sin^2(p \varepsilon)\right|}
{2 \pi\left(m^2\varepsilon^2+\sin^2(p\varepsilon)\right)^{7/2}}
={O}\left(\frac{t}{m^3k^5}\right)
={O}\left(\frac{T}{M^4}\right).
\end{align*}
Далее, $f_\pm''(p)$ не меняет знак на промежутках $[\alpha_k,\beta_k]$, так как зануляется только в $\pm\pi/2\varepsilon$. Также имеем
$\beta_k-\alpha_k\le {\pi}/{2\varepsilon K}={\pi}/4\varepsilon \lceil\frac{\pi}{m\varepsilon}\rceil<m\le mk=M$.
\end{proof}



\begin{remark} \label{rem-rough}
Аналогично шагам 3--4 выше (но со значительными упрощениями, так как не будет стационарных точек), можно доказать следующее: \emph{для любых $m,\varepsilon,\delta>0$ и любых $(x,t)\in\varepsilon\mathbb{Z}^2$, удовлетворяющих неравенствам
$|x|/t>1/\sqrt{1+m^2\varepsilon^2}+\delta$ и $\varepsilon\le 1/m$, имеем
$
{a}\left(x,t,m,\varepsilon\right)
=O\left(\frac{\varepsilon}{mt^{2}\delta^3}\right)$}
\cite[теорема~3B]{SU-20}.
\end{remark}

\subsection{Предел при больших временах: снова метод стационарной фазы(следствия\,\ref{th-limiting-distribution-mass}-\ref{cor-free})}
\label{ssec-proofs-phase}


В этом разделе мы докажем следствия~\ref{th-limiting-distribution-mass}--\ref{cor-free}. Сначала изложим план, которому будем следовать, затем докажем следствие~\ref{th-limiting-distribution-mass} по модулю технической леммы, затем саму лемму, и, наконец, следствия~\ref{cor-young}--\ref{cor-free}.


План доказательства следствия~\ref{th-limiting-distribution-mass}
(и таких результатов, как проблемы~\ref{p-correlation}--\ref{p-correlation2}) состоит из $4$ шагов
\new{(ср.~с препринтом статьи \cite{Sunada-Tate-12} на сайте arxiv.org; там другая реализация шагов~1--3)}:
\begin{description}
  \item[Шаг 1:] вычисление основного вклада в сумму, основанное на асимптотиках~\eqref{eq-ergenium-re}--\eqref{eq-ergenium-im};
  \item[Шаг 2:] оценка вклада тригонометрической суммы;
  \item[Шаг 3:] оценка остаточного члена, возникающего при замене суммы интегралом;
  \item[Шаг 4:] оценка вклада слагаемых за пределами того интервала, на котором выполняются формулы \eqref{eq-ergenium-re}--\eqref{eq-ergenium-im}.
\end{description}

\begin{proof}[Доказательство следствия~\ref{th-limiting-distribution-mass} по модулю некоторых лемм]
\textbf{Шаг 1.}
Зафиксируем значения параметров $m,\varepsilon,\delta>0$, положим
$n:=1+m^2\varepsilon^2$, \new{$F(v):=F(v,m,\varepsilon)$}, $V:=1/\sqrt{n}-\delta$, и зафиксируем $v$ в пределах $-V\le v\le V$. Докажем, что если $t$ достаточно велико в терминах
$\delta,m,\varepsilon$, то
\begin{equation}\label{eq-distribution-between-peaks}
\sum_{\substack{-Vt < x\le vt\\x\in\varepsilon\mathbb{Z}}} P(x,t,m,\varepsilon)
= F(v)-F(-V)+{O}_{\delta,m,\varepsilon}\left(\sqrt{\frac{\varepsilon}{t}}\right).
\end{equation}
Это утверждение вытекает из последовательности асимптотических формул 
\begin{multline}\label{eq-a1-distribution-between-peaks}
\hspace{-0.5cm}
\sum_{\substack{-V(t+\varepsilon) < x\le v(t+\varepsilon)\\ x\in\varepsilon\mathbb{Z}}}
a_1^2(x,t+\varepsilon,m,\varepsilon)
\overset{(*)}=
\sum_{\substack{-V(t+\varepsilon) < x\le v(t+\varepsilon)\\(x+t)/\varepsilon\text{ odd}}}
\left(\frac{2m\varepsilon^2}{\pi t}
\left(1-\frac{nx^2}{t^2}\right)^{-1/2}\sin^2\theta(x,t,m,\varepsilon)
+{O}_{\delta}\left(\frac{\varepsilon^2}{t^{2}}\right)
\right)
\\
\overset{(**)}=
\sum_{\substack{-Vt< x< vt\\(x+t)/\varepsilon\text{ odd}}}
\frac{m\varepsilon^2}{\pi t}
\left(1-\frac{nx^2}{t^2}\right)^{-1/2}-
\sum_{\substack{-Vt< x< vt\\(x+t)/\varepsilon\text{ odd}}}
\frac{m\varepsilon^2}{\pi t}
\left(1-\frac{nx^2}{t^2}\right)^{-1/2}
\cos 2\theta(x,t,m,\varepsilon)
+{O}_{\delta}\left(\frac{\varepsilon}{t}\right)
\\
\overset{(***)}=
\sum_{\substack{-Vt< x< vt\\(x+t)/\varepsilon\text{ odd}}}
\frac{m\varepsilon\cdot 2\varepsilon/t}{2\pi\sqrt{1-nx^2/t^2}}
+{O}_{\delta,m,\varepsilon}\left(\sqrt{\frac{\varepsilon}{t}}\right)
\overset{(****)}=
\int\limits_{-V}^{v}\frac{m\varepsilon\,dv}{2\pi\sqrt{1-nv^2}}
+{O}_{\delta,m,\varepsilon}\left(\sqrt{\frac{\varepsilon}{t}}\right)
\\
=m\varepsilon\frac{\arcsin(\sqrt{n}v)-\arcsin(-\sqrt{n}V)} {2\pi\sqrt{n}}
+{O}_{\delta,m,\varepsilon}\left(\sqrt{\frac{\varepsilon}{t}}\right)
\end{multline}
и аналогичной асимптотической формулы
\begin{multline*}
\sum_{\substack{-V(t+\varepsilon) < x\le v(t+\varepsilon)\\ x\in\varepsilon\mathbb{Z}}}a_2^2(x,t+\varepsilon,m,\varepsilon)
=\int_{-V}^{v}\frac{m\varepsilon(1+v)\, dv} {2\pi(1-v)\sqrt{1-nv^2}}
+{O}_{\delta,m,\varepsilon}\left(\sqrt{\frac{\varepsilon}{t}}\right)
\\
=F(v)-F(-V)
-m\varepsilon\frac{\arcsin(\sqrt{n}v)-\arcsin(-\sqrt{n}V)} {2\pi\sqrt{n}}
+{O}_{\delta,m,\varepsilon}\left(\sqrt{\frac{\varepsilon}{t}}\right).
\end{multline*}
Здесь равенство $(*)$ следует из теоремы~\ref{th-ergenium}, ведь для достаточно больших $t$ выполняются неравенства $|x|/t\le |v|(t+\varepsilon)/t<V+\delta/2=1/\sqrt{n}-\delta/2$. При этом произведение главного члена и остаточного члена в~\eqref{eq-ergenium-re} оценивается как $\varepsilon^2/t^2$. Асимптотическая формула~$(**)$ выполняется, поскольку число слагаемых меньше $t/\varepsilon$ и (возможно) выпавшие первое и последнее слагаемые меньше чем $m\varepsilon^2/t\sqrt{\delta}$. 

\smallskip\textbf{Шаг 2.} Докажем формулу~$(***)$. Мы используем следующую упрощенную версию метода стационарной фазы.

\begin{lemma}\label{l-trigonometric-sum} \textup{\cite[следствие из теоремы~4]{Karatsuba-93}} В предположениях
леммы~\ref{l-weighted-stationary-phase} (за исключением предположений на $f'(p)$, $g^{(3)}(p)$ и неравенства $N\ge M/\sqrt{T}$), если $M=N$ и $M/C\le T\le CM^2$ для некоторого $C>0$, то
$$
\sum_{\alpha<p<\beta}g(p)e(f(p))
={O}_{C,C_0,\dots,C_4}
  \left(
  \frac{(\beta-\alpha)U\sqrt{T}}{M}+\frac{UM}{\sqrt{T}}
  \right).
$$
\end{lemma}

Для удобства обозначений будем предполагать, что $t/\varepsilon$ нечётно; в противном случае доказательство будет устроено аналогично. С учётом предположения, переменная суммирования может быть представлена в виде $x=2p\varepsilon$, где $p$~--- целое. Применим лемму~\ref{l-trigonometric-sum} к функциям
\begin{align}\label{eq-functions-trig}
  f_{\pm}(p)=\pm\frac{1}{\pi}\theta(2p\varepsilon,t,m,\varepsilon) \qquad\text{и}\qquad
  g(p)=\frac{m\varepsilon^2}{\pi t}
\left(1-\frac{4n\varepsilon^2p^2}{t^2}\right)^{-1/2}
\end{align}
и параметрам
\begin{align}\label{eq-parameter-values-trig}
M&=N=T=t/\varepsilon,  & U&=\varepsilon/t, & \alpha&=-Vt/2\varepsilon, & \beta&=vt/2\varepsilon.
\end{align}

\begin{lemma}\label{l-checking-the-assumptions-trig}
При $\varepsilon\le 1/m$ существуют константы $C,C_0,\dots,C_4$ зависящие от $\delta,m,\varepsilon$, но не от $v,p$ такие, что параметры~\eqref{eq-parameter-values-trig} и функции~\eqref{eq-functions-trig}
удовлетворяют всем предположениям леммы~\ref{l-trigonometric-sum}. \end{lemma}

Значит, формула~$(***)$ действительно выполняется, поскольку параметры~\eqref{eq-parameter-values-trig} удовлетворяют условию
$\frac{(\beta-\alpha)U\sqrt{T}}{M}+\frac{UM}{\sqrt{T}}
={O}\left(\sqrt{\frac{\varepsilon}{t}}\right)$.

\smallskip\textbf{Шаг 3.} Докажем формулу~(****). Для этого мы используем следующий известный результат.

\begin{lemma}[Формула суммирования Эйлера] \label{l-Euler-summation-formula} \textup{\cite[замечание к теореме~1]{Karatsuba-93}}
Если $g(p)$ непрерывно дифференцируема на $[\alpha,\beta]$ и $\rho(p):=1/2-\{p\}$, то
$$
\sum_{\alpha<p<\beta}g(p)
=\int_{\alpha}^{\beta}g(p)\,dp+\rho(\beta)g(\beta)-\rho(\alpha)g(\alpha)
+\int_{\alpha}^{\beta}\rho(p)g'(p)\,dp.
$$
\end{lemma}

Снова предположим без потери общности, что $t/\varepsilon$ нечетно. Применим Лемму~\ref{l-Euler-summation-formula} с теми же $\alpha,\beta,g(p)$ (задаваемыми равенствами~\eqref{eq-functions-trig}--\eqref{eq-parameter-values-trig}), что и на шаге 2.  Согласно Лемме~\ref{l-checking-the-assumptions-trig} выполняются оценки $g(p)=O_{\delta,m,\varepsilon}(\varepsilon/t)$ и $g'(p)=O_{\delta,m,\varepsilon}(\varepsilon^2/t^2)$. Значит, по Лемме~\ref{l-Euler-summation-formula} разница между суммой и интегралом в~(****) есть $O_{\delta,m,\varepsilon}(\varepsilon/t)$, 
что доказывает равенство~\eqref{eq-distribution-between-peaks}.

\smallskip\textbf{Шаг 4.} Докажем следствие для произвольного $v\in (-1/\sqrt{n};1/\sqrt{n})$. Согласно~\eqref{eq-distribution-between-peaks}, для любых
$\delta,m,\varepsilon$ существуют $C_1(\delta,m,\varepsilon)$ и $C_2(\delta,m,\varepsilon)$ такие, что для любого $v\in [-1/\sqrt{n}+\delta,1/\sqrt{n}-\delta]$ и любого $t\ge C_1(\delta,m,\varepsilon)$
выполняется неравенство
\begin{equation*}
\left|\sum_{(-1/\sqrt{n}+\delta)t< x\le vt}P(x,t,m,\varepsilon)-F(v)\right|\le F\left(-\frac{1}{\sqrt{n}}+\delta\right)
+C_2(\delta,m,\varepsilon)\sqrt{\frac{\varepsilon}{t}}.
\end{equation*}

Очевидно, что $C_1(\delta,m,\varepsilon)$ и $C_2(\delta,m,\varepsilon)$ можно считать убывающими функциями от $\delta$: чем больше интервал  $[-\frac{1}{\sqrt{n}}+\delta,\frac{1}{\sqrt{n}}-\delta]$, тем хуже оценки остатков в (*)--(****). Выберем $\delta(t)$ стремящимся к $0$ достаточно медленно, чтобы
$C_1(\delta(t),m,\varepsilon)\le t$ при $t$ достаточно большом в терминах $m,\varepsilon$ и чтобы
$C_2(\delta(t),m,\varepsilon)\sqrt{\frac{\varepsilon}{t}}\to 0$ при $t\to\infty$. Положим $V(t):=\frac{1}{\sqrt{n}}-\delta(t)$.
Тогда, так как $F\left(-\frac{1}{\sqrt{n}}+\delta\right)\to F\left(-\frac{1}{\sqrt{n}}\right)=0$ при $\delta \to 0$ по определению $F(v)$, то
\begin{equation}\label{eq-auxiliary-uniform-convergence}
\sum_{-V(t)t< x\le vt}P(x,t,m,\varepsilon)\rightrightarrows F(v)\qquad\text{при }\quad t\to\infty
\end{equation}
равномерно по $v\in(-\frac{1}{\sqrt{n}};\frac{1}{\sqrt{n}})$.
Аналогично, так как $F\left(\frac{1}{\sqrt{n}}-\delta\right)\to F\left(\frac{1}{\sqrt{n}}\right)=1$ при $\delta \to 0$, то
\begin{equation*}
\sum_{-V(t)t< x\le V(t)t}P(x,t,m,\varepsilon)\to 1\qquad\text{при }\quad t\to\infty.
\end{equation*}
Значит, согласно предложению~\ref{p-mass2},
\begin{align*}
\sum_{x\le -V(t)t}P(x,t,m,\varepsilon)
&=1-\sum_{x>-V(t)t}P(x,t,m,\varepsilon)
\le 1-\sum_{-V(t)t< x\le V(t)t}P(x,t,m,\varepsilon)
\to 0.
\end{align*}
Вместе с~\eqref{eq-auxiliary-uniform-convergence} это доказывает следствие для  $v\in(-1/\sqrt{n};1/\sqrt{n})$.
Для $v\le -1/\sqrt{n}$ и аналогично для $v\ge 1/\sqrt{n}$,
следствие вытекает из оценки
$$\sum_{x\le vt}P(x,t,m,\varepsilon)\le
\sum_{x\le -V(t)t}P(x,t,m,\varepsilon)\to 0.$$
\end{proof}

Докажем  теперь лемму~\ref{l-checking-the-assumptions-trig} и следствия~\ref{cor-young}, \ref{cor-free}.

\begin{proof}[Доказательство леммы~\ref{l-checking-the-assumptions-trig}] Неравенства $M/C\le T\le CM^2$ и $M\ge \beta-\alpha$ 
автоматически выполняются при $C=1$, поскольку $t/\varepsilon$~--- натуральное число и $|V|,|v|\le 1$.
Используя предположения $\varepsilon\le 1/m$, $\alpha\le p\le \beta$, и полагая $C_2:=\max\{1/m\varepsilon,2/\delta^{3/2}\}$ оценим производные (вычисленные в \cite[\S9]{SU-2}) следующим образом:
\begin{align*}
|g(p)|&=\frac{m\varepsilon^2}{\pi t}
\left(1-\frac{4n\varepsilon^2p^2}{t^2}\right)^{-1/2}
\le \frac{m\varepsilon^2}{\pi t\sqrt{1-nV^2}}
\le \frac{m\varepsilon^2}{t\sqrt{\delta}}
\le \frac{\varepsilon}{t\sqrt{\delta}}
={O}_{\delta}\left({U}\right), \\
|g^{(1)}(p)|&=\frac{4m\varepsilon^4 n|p|}{\pi t^3}
\left(1-\frac{4n\varepsilon^2p^2}{t^2}\right)^{-3/2}
\le \frac{2m\varepsilon^3 nVt}
{\pi t^3(1-nV^2)^{3/2}}
\le \frac{m\varepsilon^3 n}{t^2\delta^{3/2}}
\le \frac{2\varepsilon^2}{t^2\delta^{3/2}}
={O}_{\delta}\left(\frac{U}{N}\right), \\
|g^{(2)}(p)|
&=\frac{4m\varepsilon^4 n(8\varepsilon^2 np^2+t^2)}{\pi t^5}
\left(1-\frac{4n\varepsilon^2p^2}{t^2}\right)^{-5/2}
\le \frac{4m\varepsilon^4 n(2nV^2+1)t^2}{\pi t^5(1-nV^2)^{5/2}}
={O}_{\delta}\left(\frac{\varepsilon^3}{t^3}\right)
={O}_{\delta}\left(\frac{U}{N^2}\right), \\
|f^{(2)}(p)|&=\frac{4m\varepsilon^2}{\pi t} \left(1-\frac{4\varepsilon^2p^2}{t^2}\right)^{-1}
\left(1-\frac{4n\varepsilon^2p^2}{t^2}\right)^{-1/2}
\ge\frac{m\varepsilon^2}{t}\ge \frac{T}{C_2M^2},\\
|f^{(2)}(p)|&=\frac{4m\varepsilon^2}{\pi t} \left(1-\frac{4\varepsilon^2p^2}{t^2}\right)^{-1}
\left(1-\frac{4n\varepsilon^2p^2}{t^2}\right)^{-1/2}
\le
\frac{4m\varepsilon^2}{\pi t \left(1-V^2\right)\sqrt{1-nV^2}}
\le
\frac{2m\varepsilon^2}{t\delta^{3/2}}
\le \frac{C_2T}{M^2}, \\
|f^{(3)}(p)|
&=\frac{16m\varepsilon^4
\left|(n+2)pt^2-12n\varepsilon^2 p^3\right|}{\pi t^5}
\left(1-\frac{4\varepsilon^2p^2}{t^2}\right)^{-2}
\left(1-\frac{4n\varepsilon^2p^2}{t^2}\right)^{-3/2}
={O}_{\delta}\left(\frac{T}{M^3}\right),\\
|f^{(4)}(p)|&=
\frac{16m\varepsilon^4\left|768n^2\varepsilon^6p^6 - 48n(2n+5) \varepsilon^4p^4 t^2 + 8(n^2-n+3)\varepsilon^2 p^2 t^4 +(n+2) t^6\right|} {\pi t^9
\left(1-{4\varepsilon^2p^2}/{t^2}\right)^{3}
\left(1-{4n\varepsilon^2p^2}/{t^2}\right)^{5/2}}
={O}_{\delta}\left(\frac{T}{M^4}\right).\\[-1.0cm]
\end{align*}
\end{proof}


\begin{proof}[Доказательство следствия~\ref{cor-young}]
Диаграммы Юнга с $h$ строками и $w$ столбцами находятся в очевидной биекции с путями шашки из $(0,0)$ в $(w-h,w+h)$, проходящими через $(1,1)$ и $(w-h+1,w+h-1)$. Поэтому $n_+(h\times w)-n_-(h\times w)=-2^{(w+h-1)/2}a_1(w-h,w+h)$.
Положим $h:=\lceil rw\rceil$. Применим
теорему~\ref{th-ergenium} и замечание~\ref{rem-rough} (или теорему~\ref{th-outside}) с
$$
\delta=\tfrac{1}{2}\left|\tfrac{1}{\sqrt{2}}-\tfrac{r-1}{r+1}\right|,
\quad m=\varepsilon=1, \quad x=w-h, \quad t=w+h-1.
$$
Это завершает доказательство в случае $r>3+2\sqrt{2}$.
Остается показать, что при $r<3+2\sqrt{2}$ величина~\eqref{eq-theta}  при $w\to\infty$ будет попадать в сколь угодно малые окрестности точек $\frac{\pi}{2}+\pi\mathbb{Z}$.

Пусть $v:=\frac{h-w}{w+h-1}$, $v_0:=\frac{r-1}{r+1}$ и
$$
\theta(vt,t,1,1)=
t\left(\arcsin\frac{1}{\sqrt{2-2v^2}}
-v\arcsin\frac{v}{\sqrt{1-v^2}}\right)+\frac{\pi}{4}
=:t\theta(v)+\frac{\pi}{4}.
$$
Так как $\theta(v)\in C^2[0;1/\sqrt{2}-\delta]$, то можно записать разложение в ряд Тейлора
$$
\theta(vt,t,1,1)
=\frac{\pi}{4}+t\theta(v_0)+t(v-v_0)\theta'(v_0)
+O_\delta\left(t(v-v_0)^2\right).
$$
Подставляя в него
$$
v-v_0=\frac{h-w}{w+h-1}-\frac{r-1}{r+1}
=\frac{2h-2rw+r-1}{(r+1)(w+h-1)}=\frac{2\{-rw\}+r-1}{(r+1)t},
$$
где $h=\lceil rw\rceil=rw+\{-rw\}$, получаем, что
\begin{multline*}
\theta(vt,t,1,1)
=\frac{\pi}{4}+(w+rw+\{-rw\}-1)\theta(v_0)
+\frac{2\{-rw\}+r-1}{(r+1)}\theta'(v_0)
+O_\delta\left(\frac{1}{w}\right)\\
=\frac{\pi}{4}-\theta(v_0)+v_0\theta'(v_0)+w(r+1)\theta(v_0)
+\{-rw\}\left(\theta(v_0)+\tfrac{2}{(r+1)}\theta'(v_0)\right)
+O_\delta\left(\tfrac{1}{w}\right)\\=:\pi(\alpha(r) w+\beta(r)\{-rw\}+\gamma(r))+O_\delta\left(\tfrac{1}{w}\right).
\end{multline*}
Почти для всех $r$ числа
$1,r,\alpha(r)$ линейно независимы над полем рациональных чисел, потому что график функции $\alpha(r)=(r+1)\theta\left(\frac{r-1}{r+1}\right)$ имеет лишь счетное число точек пересечения с прямыми, задаваемыми уравнениями с рациональными коэффициентами.
Следовательно, по теореме Кронекера для любого $\Delta>0$ существует бесконечно много $w$ таких, что
$$
\{-rw\}<\Delta\qquad\text{и}
\qquad \left|\{\alpha(r)w\}+\gamma(r)-\tfrac{1}{2}\right|<\Delta.
$$
Для таких  $w$
\begin{multline*}
\left|\sin\theta(vt,t,1,1)\right|=
1+O\left((1+\beta(r))\Delta\right)+O_\delta\left(\tfrac{1}{w}\right)
\qquad\text{и}\qquad
2^{(r+1)w/2}\le 2^{(w+h)/2}\le 2^{(r+1)w/2+\Delta},\\[-0.4cm]
\end{multline*}
таким образом, утверждение следствия вытекает из~\eqref{eq-ergenium-re}.
\end{proof}

%
%
%
%

Выведем утверждение следствия~\ref{cor-free} из результатов работы~\cite{Sunada-Tate-12}.

\begin{proof}[Доказательство следствия~\ref{cor-free}]
Применим следствие~1.5 и теорему~1.1 из~\cite{Sunada-Tate-12} с
$$
n=t/\varepsilon-1, \quad
y_n=2\varepsilon\left\lceil \frac{vt}{2\varepsilon}\right\rceil-1,  \quad
\xi=v, \quad
\phi=(0,1)^T, \quad
a=1/\sqrt{1+m^2\varepsilon^2}, \quad b=m\varepsilon/\sqrt{1+m^2\varepsilon^2}.
$$
Тогда в случае $|v|>1/\sqrt{1+m^2\varepsilon^2}$ нужное утверждение непосредственно вытекает из следствия~1.5 из~\cite{Sunada-Tate-12}.
В случае $|v|<1/\sqrt{1+m^2\varepsilon^2}$ нужное утверждение вытекает из оценок \new{}\new{}
\begin{equation}\label{eq-probability-bound}
\frac{2m\varepsilon^2}{\pi t(1-v)}
\sqrt{\frac{a-|v|}{a+|v|}}
+{O}_{m,\varepsilon,v}\left(\frac{1}{t^2}\right)
\le P\left(2\varepsilon\left\lceil\frac{vt}{2\varepsilon}\right\rceil,t,m,\varepsilon\right)
\le\frac{2m\varepsilon^2}{\pi t(1-v)}
\sqrt{\frac{a+|v|}{a-|v|}}
+{O}_{m,\varepsilon,v}\left(\frac{1}{t^2}\right),
\end{equation}
где $t\in 2\varepsilon\mathbb{Z}$ и $a:=1/\sqrt{1+m^2\varepsilon^2}$. Оценки~\eqref{eq-probability-bound} следуют из теоремы 1.1 работы \cite{Sunada-Tate-12}, так как фигурирующая там величина $\mathrm{OSC}_n(\xi)$ допускает оценку
$
|\mathrm{OSC}_n(\xi)|\le \sqrt{A(\xi)^2+B(\xi)^2}=|\xi|(1+\xi)/|a|
$
(см. проверку этого факта в \cite[\S18]{SU-2}).
\end{proof}

Оценки \eqref{eq-probability-bound} могут быть также получены с помощью теоремы~\ref{th-ergenium} методом из \S\ref{ssec-proofs-feynman}.

\addcontentsline{toc}{myshrinkalt}{}

\subsection{Решение проблемы Фейнмана: разложения в ряды Тейлора (следствия~\ref{cor-intermediate-asymptotic-form}--\ref{cor-feynman-problem})}
\label{ssec-proofs-feynman}

В этом разделе  с помощью теоремы~\ref{th-ergenium} мы получим решение проблемы Фейнмана. Для этого будет достаточно приблизить функции, фигурирующие в теореме~\ref{th-ergenium}, отрезками рядов Тейлора.

\begin{proof}[Доказательство следствия~\ref{cor-intermediate-asymptotic-form}]
Сначала выведем асимптотическую формулу для функции $\theta(x,t,m,\varepsilon)$, заданной равенством~\eqref{eq-theta}.
Положим $n:=1+m^2\varepsilon^2$. Тогда $1/\sqrt{1+z^2}=1+O(z^2)$,
$\arcsin z=z+{O}\left(z^3\right)$ при $z\in[-1;1]$, и $t/\sqrt{t^2-x^2}<1/\sqrt{1-\sqrt{n}x/t}<1/\sqrt{\delta}$. Значит,
\begin{align*}
\arcsin
\frac{m\varepsilon t} {\sqrt{n\left(t^2-x^2\right)}}
&=\frac{m\varepsilon t} {\sqrt{1+m^2\varepsilon^2}\sqrt{t^2-x^2}}+
{O}\left(
\frac{m^3\varepsilon^3}{n^{3/2}} \left(\frac{t}{\sqrt{t^2-x^2}}\right)^3\right)
=
\frac{m\varepsilon t} {\sqrt{t^2-x^2}}+
{O}_\delta\left(m^3\varepsilon^3\right).
\end{align*}
Пользуясь аналогичной асимптотической формулой для $\arcsin
\frac{m\varepsilon x}{\sqrt{t^2-x^2}}$, получаем, что
$$
\theta(x,t,m,\varepsilon)
=\frac{m t^2} {\sqrt{t^2-x^2}}
-\frac{m x^2} {\sqrt{t^2-x^2}}+\frac{\pi}{4}+
\left(\frac{t+|x|}{\varepsilon}\right)
{O}_\delta\left(m^3\varepsilon^3\right)=
m\sqrt{t^2-x^2}+\frac{\pi}{4}+
{O}_\delta\left(m^3\varepsilon^2t\right).
$$
Так как
\begin{equation*}
 \left|\frac{\partial \sqrt{t^2-x^2}}{\partial t}\right|
 =\frac{t}{\sqrt{t^2-x^2}}<\frac{1}{\sqrt{\delta}}
 \qquad\text{и}\qquad
 \left|\frac{\partial \sqrt{t^2-x^2}}{\partial x}\right|
 =\frac{|x|}{\sqrt{t^2-x^2}}<\frac{1}{\sqrt{\delta}},
\end{equation*}
то из теоремы Лагранжа о конечном приращении следует, что
\begin{align*}
\theta(x,t-\varepsilon,m,\varepsilon)
&=m\sqrt{t^2-x^2}+\frac{\pi}{4}
+{O}_\delta\left(m\varepsilon+m^3\varepsilon^2t\right),\\
\theta(x-\varepsilon,t-\varepsilon,m,\varepsilon)
&=m\sqrt{t^2-x^2}+\frac{\pi}{4}
+{O}_\delta\left(m\varepsilon+m^3\varepsilon^2t\right).
\end{align*}
Рассмотрим оставшиеся сомножители в \eqref{eq-ergenium-re}--\eqref{eq-ergenium-im}.
По теореме Лагранжа, для некоторого $\eta\in[0,nx^2/t^2]$
\begin{multline*}
\left(1-\frac{nx^2}{t^2}\right)^{-1/4}-1
=\frac{nx^2}{t^2}\frac{(1-\eta)^{-5/4}}{4}
\le \frac{nx^2}{t^2}
\left(1-\frac{nx^2}{t^2}\right)^{-5/4}\\
\le  \frac{x^2}{t^2} \left(\frac{1}{\sqrt{n}}-\frac{x}{t}\right)^{-5/4}
\left(\frac{1}{\sqrt{n}}+\frac{x}{t}\right)^{-5/4}
\le \frac{x^2}{t^2}\delta^{-5/2}
={O}_\delta\left(\frac{|x|}{t}\right).
\end{multline*}
Значит, при $t\ge 2\varepsilon$
\begin{equation}\label{eq-auxest1}
\left(1-\frac{nx^2}{(t-\varepsilon)^2}\right)^{-1/4}
=1+{O}_\delta\left(\frac{|x|}{t}\right)
\quad\text{и}\quad
\left(1-\frac{n(x-\varepsilon)^2}{(t-\varepsilon)^2}\right)^{-1/4}
=1+{O}_\delta\left(\frac{|x|+\varepsilon}{t}\right).
\end{equation}
Кроме того,
\begin{equation}\label{eq-auxest2}
\sqrt{\frac{t-\varepsilon+x-\varepsilon}{t-x}}
=\sqrt{1+2\frac{x-\varepsilon}{t-x}}
=1+{O}\left(\frac{x-\varepsilon}{t-x}\right)
=1+{O}_\delta\left(\frac{|x|+\varepsilon}{t}\right).
\end{equation}
Подставляя все полученные асимптотические формулы в~\eqref{eq-ergenium-re}--\eqref{eq-ergenium-im}, приходим к равенствам
\begin{align*}
\mathrm{Re}\,{a}\left(x,t,m,\varepsilon\right)
&={\varepsilon}\sqrt{\frac{2m}{\pi t}}
\left(\sin \left(m\sqrt{t^2-x^2}+\frac{\pi}{4}\right)
+O_\delta\left(
\frac{1}{mt}+\frac{|x|+\varepsilon}{t}+m\varepsilon+m^3\varepsilon^2t
\right)\right),\\
\mathrm{Im}\,{a}\left(x,t,m,\varepsilon\right)
&={\varepsilon}\sqrt{\frac{2m}{\pi t}}
\left(\cos \left(m\sqrt{t^2-x^2}+\frac{\pi}{4}\right)
+O_\delta\left(
\frac{1}{mt}+\frac{|x|+\varepsilon}{t}+m\varepsilon+m^3\varepsilon^2t
\right)
\right).
\end{align*}
По предположению $\varepsilon\le 1/m$. Значит, $m\varepsilon\le \frac{1}{mt}+m^3\varepsilon^2t$ и $\frac{\varepsilon}{t}\le \frac{1}{mt}$, что позволяет переписать остаточные члены в нужном виде.
\end{proof}

\begin{proof}[Доказательство следствия~\ref{cor-feynman-problem}]
Для доказательства равенства~\eqref{eq-feynman-problem} достаточно  в утверждение следствия~\ref{cor-intermediate-asymptotic-form} подставить асимптотическую формулу
$$
\sqrt{t^2-x^2}=t\left(1-\frac{x^2}{2t^2}
+{O}_\delta\left(\frac{x^4}{t^4}\right)\right),
\qquad\text{где }\quad \frac{|x|}{t}<1-\delta.
$$
\end{proof}

\begin{proof}[Доказательство утверждения из примера~\ref{p-Feynman-couterexample}]
В случае, когда $(x_n,t_n,\varepsilon_n)=(n^3,n^4,1/n^4)$, нужное утверждение получается подстановкой в следствие~\ref{cor-intermediate-asymptotic-form} разложения
$$
\sqrt{t^2-x^2}=t\left(1-\frac{x^2}{2t^2}
-\frac{x^4}{8t^4}+{O}\left(\frac{x^6}{t^6}\right)\right),
\qquad\text{где }\quad \frac{|x|}{t}\le
\frac{1}{2}.
$$

В остальных случаях нам необходимо оценить функцию $\theta(\varepsilon,t,m,\varepsilon)$, задаваемую равенством~\eqref{eq-theta}, при $t\ge 2\varepsilon$ и $\varepsilon\le 1/m$. Для таких $t,m,\varepsilon$

\begin{gather*}
m\varepsilon t/\sqrt{(1+m^2\varepsilon^2)(t^2-\varepsilon^2)}\le \sqrt{2/3},\\
\arcsin z-\arcsin w={O}\left(z-w\right)\qquad (\text{при }0<w<z<\sqrt{2/3}),\\
1/\sqrt{1-z^2}-1=O(z^2)\qquad(\text{при }|z|\le 1/2).
\end{gather*}
Значит,
\begin{multline*}
\theta(\varepsilon,t,m,\varepsilon)-\theta(0,t,m,\varepsilon)
=\\=\frac{t}{\varepsilon}\arcsin
\frac{m\varepsilon t} {\sqrt{\left(1+m^2\varepsilon^2\right)\left(t^2-\varepsilon^2\right)}}
-\frac{t}{\varepsilon}\arcsin\frac{m\varepsilon} {\sqrt{1+m^2\varepsilon^2}}
-\arcsin
\frac{m\varepsilon^2}{\sqrt{t^2-\varepsilon^2}}
= {O}\left(\frac{m\varepsilon^2}{t}\right)
={O}\left(\frac{1}{mt}\right).
\end{multline*}
Применим теорему~\ref{th-ergenium} и формулы~\eqref{eq-auxest1}--\eqref{eq-auxest2} при $x=0$ и $\delta=1/\sqrt{2}$. Получим, что
\begin{equation}\label{eq-intermediate-asymptotic-form2}
{a}\left(0,t,{m},{\varepsilon}\right)
=\varepsilon\sqrt{\frac{2m}{\pi t}}
\exp\left(-i\left(\frac{t}{\varepsilon}-1\right)
\arctan(m\varepsilon)+\frac{i\pi}{4}\right)
\left(1+{O}\left(\frac{1}{mt}\right)\right).
\end{equation}

В случае, когда $\varepsilon=\varepsilon_n=1/2n$ и $t=t_n=(2n)^2$, выполняется равенство $\arctan(m\varepsilon)=m\varepsilon-m^3\varepsilon^3/3
+O(m^5\varepsilon^5)$. Поэтому правая часть~\eqref{eq-intermediate-asymptotic-form2} эквивалентна правой части~\eqref{cor-feynman-problem}, умноженной на $e^{im^3/3}$.

В случае, когда $\varepsilon=\varepsilon_n=\mathrm{const}$ и $t=t_n=2n\varepsilon$, отношение правых частей~\eqref{eq-intermediate-asymptotic-form2} и~\eqref{cor-feynman-problem} не имеет предела, поскольку разность $\arctan(m\varepsilon)-m\varepsilon$ не кратна $\pi$ при $0<\varepsilon<1/m$.
\end{proof}

\addcontentsline{toc}{myshrinkalt}{}

\subsection{Непрерывный предел: метод смены хвостов (теорема~\ref{th-main} и следствия~\ref{cor-uniform}--\ref{cor-concentration})
}
\label{ssec-proofs-continuum}


\begin{proof}[Доказательство теоремы~\ref{th-main}] Доказательство основано на \emph{методе смены хвостов} и состоит из  
$5$ шагов:
\begin{description}
\item[Шаг 1:] отбрасывание нормирующего множителя в~\eqref{eq1-p-mass}--\eqref{eq2-p-mass}.

\item[Шаг 2:] отбрасывание слагаемых в~\eqref{eq1-p-mass}--\eqref{eq2-p-mass}, начиная с номера $T$ (
    мы возьмём $T=\lceil\log \frac\delta\varepsilon\rceil$).

\item[Шаг 3:] замена биномиальных коэффициентов степенями в оставшихся слагаемых.

\item[Шаг 4:] замена получившихся сумм степенными рядами.

\item[Шаг 5:] объединение оценок остаточных членов, полученных на предыдущих шагах, для получения оценки итогового остаточного члена.
\end{description}
\smallskip


Выведем асимптотику для
$a_1\left(x,t,{m},{\varepsilon}\right)$; доказательство для $a_2\left(x,t,{m},{\varepsilon}\right)$ аналогично.

\smallskip
\textbf{Шаг 1.} Рассмотрим первый множитель в~\eqref{eq1-p-mass}. Имеем $0\ge (1-t/\varepsilon)/2\ge -t/\varepsilon$, так как $t\ge \delta\ge\varepsilon$. Возводя $(1+m^2\varepsilon^2)$ в соответствующие степени, получаем
$$1\ge \left(1+{m^2\varepsilon^2}{}\right)^{(1-t/\varepsilon)/2}
\ge \left(1+{m^2}{\varepsilon^2}\right)^{-t/\varepsilon}
\ge e^{-{m^2}{\varepsilon^2}t/\varepsilon}\ge 1-m^2 t\varepsilon,$$
где в последних двух неравенствах мы использовали тот факт, что $e^a\ge 1+a$ для любого $a\in\mathbb{R}$. Следовательно,
$$\left(1+{m^2}{\varepsilon^2}\right)^{(1-t/\varepsilon)/2}=
1+O(m^2t\varepsilon).$$

\textbf{Шаг 2.} Рассмотрим  $T$-ю частичную сумму в~\eqref{eq1-p-mass} с $T=\lceil\log \frac\delta\varepsilon\rceil$ слагаемыми.
Общее число слагаемых действительно не меньше, чем $T$,
так как $(t-|x|)/2\varepsilon\ge \delta/2\varepsilon\ge \log (\delta/\varepsilon)$, ввиду неравенств $t-|x|\ge \delta>\varepsilon$ и $e^a\ge 1+a+a^2/2\ge 2a$ для любого $a\ge 0$.

Для $r\ge T$  отношение двух последовательных слагаемых в~\eqref{eq1-p-mass} равняется
\begin{align*}
\left({m}{\varepsilon}\right)^2
\frac{((t+x)/2\varepsilon-1-r)
((t-x)/2\varepsilon-1-r)}{(r+1)^2}
&<\left({m}{\varepsilon}\right)^2
\cdot\frac{(t+x)}{2\varepsilon T}\cdot \frac{(t-x)}{2\varepsilon T}=\frac{m^2s^2}{4T^2}<\frac{1}{2},
\end{align*}
где последнее неравенство следует из того, что
$T=\lceil\log \frac\delta\varepsilon\rceil>\lceil\log e^{3ms}\rceil\ge 3ms$.
Следовательно, 
остаточный член (т.е., сумма по $r\ge T$) меньше, чем сумма геометрической прогресси со знаменателем $\frac{1}{2}.$ Следовательно, из предложения~\ref{p-mass3} имеем
\begin{align*}
a_1\left(x,t,{m},{\varepsilon}\right)
&={m}{\varepsilon}
\left(1+O\left({m^2 t\varepsilon}\right)\right)\cdot
\left[\sum_{r=0}^{T-1}(-1)^r
\binom{(t+x)/2\varepsilon-1}{r}
\binom{(t-x)/2\varepsilon-1}{r}
\left({m}{\varepsilon}\right)^{2r}
+
\right.
\\&+\left.O\left(
\binom{(t+x)/2\varepsilon-1}{T}
\binom{(t-x)/2\varepsilon-1}{T}
\left({m}{\varepsilon}\right)^{2T}
\right)\right].
\end{align*}

\textbf{Шаг 3.} Чтобы аппроксимировать сумму, возьмём целые $L:=(t\pm x)/2\varepsilon$, $r<T$, и преобразуем биномиальные коэффициенты следующим образом:
$$
\binom{L-1}{r}
=\frac{(L-1)
\cdots(L-r)}{r!}=\frac{L^r}{r!}
\left(1-\frac{1}{L}\right)
\cdots\left(1-\frac{r}{L}\right).
$$
Здесь
$$\frac{r}{L}= \frac{2r\varepsilon}{t\pm x}<\frac{2T\varepsilon}{\delta}=
\frac{2\varepsilon}{\delta}
\left\lceil\log \frac{\delta}{\varepsilon}\right\rceil
\le
\frac{2\varepsilon}{\delta}\left(\log \frac{\delta}{\varepsilon}+1\right)
< \frac{1}{2},
$$
так как $\delta/\varepsilon\ge 16$, а $2(\log a+1)/a$ убывает при $a\ge 16$ и меньше $1/2$ при $a=16$.
Применяя неравенство $1-a\ge e^{-2a}$ при $0\le a\le 1/2$,  а потом --- неравенства $1-a\le e^{-a}$ и $L\ge\delta/2\varepsilon$, получаем
$$
\left(1-\frac{1}{L}\right)\cdots\left(1-\frac{r}{L}\right)
\ge
e^{{-2/L}}e^{{-4/L}}\cdots e^{{-2r/L}}=
e^{{-r(r+1)/L}}\ge
e^{-T^2/L}\ge 1-\frac{T^2}{L}
\ge 1-\frac{2T^2\varepsilon}{\delta}.
$$
Следовательно,
\begin{align*}
  \frac{(t\pm x)^r}{r!(2\varepsilon)^r} \ge\binom{(t\pm x)/2\varepsilon-1}{r}
  \ge \frac{(t\pm x)^r}{r!(2\varepsilon)^r}
  \left(1-\frac{2T^2 \varepsilon}{\delta}\right).
\end{align*}
Подставляя результат в выражение для $a_1(x,t,m,\varepsilon)$ из шага~2, получаем
\begin{multline*}
a_1\left(x,t,{m},{\varepsilon}\right)
={m}{\varepsilon}\left(1+O\left({m^2t}{\varepsilon}\right)\right)\cdot\\
\cdot
\left[\sum_{r=0}^{T-1}(-1)^r
\left(
\frac{m}{2}\right)^{2r}\frac{(t^2-x^2)^r}{(r!)^2}
\left(1+O\left(\frac{T^2 \varepsilon}{\delta}\right)
\right)
+O\left(
\left(\frac{m}{2}\right)^{2T}\frac{(t^2-x^2)^T}{(T!)^2}
\left(1+\frac{T^2 \varepsilon}{\delta}
\right)
\right)
\right].
\end{multline*}

Последний остаточный член оценивается следующим образом.
Так как $T!\ge (T/3)^T$
и
$$
T\ge \log \frac\delta\varepsilon\ge 3m\sqrt{t^2-x^2}\ge
\frac{3m}{2}\sqrt{t^2-x^2}\sqrt{e},
$$
получаем
$$
\frac{(t^2-x^2)^T}{(T!)^2}\cdot\left(\frac{m}{2}\right)^{2T}\le\frac{(t^2-x^2)^T}{(T)^{2T}}\cdot\left(\frac{3m}{2}\right)^{2T}\le e^{-T}\le \frac\varepsilon\delta.
$$
Имеем $(\varepsilon/\delta)\left(1+T^2 \varepsilon/\delta\right)=O\left(T^2 \varepsilon/\delta\right)$,
так как $T\ge 1$ и $\varepsilon<\delta$.
Следовательно, искомый остаточный член может быть включен в $O\left(T^2 \varepsilon/\delta\right)$. Имеем
$$
a_1\left(x,t,{m},{\varepsilon}\right)
={m}{\varepsilon}\left(1+O\left({m^2t}{\varepsilon}\right)\right)\cdot
\sum_{r=0}^{T-1}(-1)^r
\left(
\frac{m}{2}\right)^{2r}\frac{(t^2-x^2)^r}{(r!)^2}
\left(1+O\left(\frac{T^2 \varepsilon}{\delta}\right)
\right).
$$
Заметим, что по нашему соглашению об обозначениях, константа в $O\left(T^2 \varepsilon/\delta\right)$ не зависит от $r$.

\textbf{Шаг 4.} Теперь мы можем заменить сумму из $T$ слагаемых рядом, так как ''хвост'' знакопеременного ряда, у которого модуль слагаемых убывает, можно оценить первым слагаемым (которое мы только что оценили):
$$\left|\sum_{r=T}^{\infty}(-1)^r
\left(
\frac{m}{2}\right)^{2r}\frac{(t^2-x^2)^r}{(r!)^2}
\right|
\le\frac{(t^2-x^2)^T}{(T!)^2}\cdot\left(\frac{m}{2}\right)^{2T}\le \frac\varepsilon\delta=O\left(\frac{T^2 \varepsilon}{\delta}\right).
$$
Так как константа в каждом слагаемом $O\left(\frac{T^2 \varepsilon}{\delta}\right)$ одинакова (см. шаг~3), имеем
\begin{align*}
a_1\left(x,t,{m},{\varepsilon}\right)
&={m}{\varepsilon}\left(1+O\left({m^2t\varepsilon}\right)\right)\cdot
\sum_{r=0}^{\infty}(-1)^r
\left(
\frac{m}{2}\right)^{2r}\frac{(t^2-x^2)^r}{(r!)^2}
\left[1+O\left(\frac{T^2 \varepsilon}{\delta}\right)
\right]\\
&=
{m}{\varepsilon}\left(1+O\left({m^2t\varepsilon}\right)\right)\cdot
\left(J_0(ms)+O\left(\frac{T^2 \varepsilon}{\delta}I_0(ms)\right)\right),
\end{align*}
здесь мы используем
\emph{модифицированные функции Бесселя первого рода}:
\begin{align*}
I_0(z)&:=\sum_{k=0}^\infty \frac{(z/2)^{2k}}{(k!)^2},
&
I_1(z)&:=\sum_{k=0}^\infty \frac{(z/2)^{2k+1}}{k!(k+1)!}.
\end{align*}

\textbf{Шаг 5.} Имеем $m^2t\delta\le
m^2(t+|x|)(t-|x|)=m^2s^2\le 9m^2s^2\le
 T^2$. Следовательно, ${m^2t\varepsilon}\, J_0(ms)\le T^2\varepsilon \,I_0(ms)/\delta$ и
$
{m^2t\varepsilon}\le
{T^2\varepsilon}/{\delta}<(\log(\delta/\varepsilon)+1)^2
{\varepsilon}/{\delta}< 2,
$
так как $(a+1)^2/2<e^a$ при $a\ge 0$.
Получаем следующую формулу:
\begin{align*}
a_1\left(x,t,{m},{\varepsilon}\right)
&={m}{\varepsilon}\left(J_0(ms)+O\left(\frac{\varepsilon}{\delta}
\log^2\frac{\delta}{\varepsilon}\, I_0(ms)\right)\right).
\\
\intertext{Аналогично, }
a_2\left(x,t,{m},{\varepsilon}\right)
&={m}{\varepsilon}\left(1+O\left({m^2t\varepsilon}\right)\right)
\cdot\frac{t+x}{\sqrt{t^2-x^2}}\cdot
\sum_{r=1}^{T-1}(-1)^r
\left(
\frac{m}{2}\right)^{2r-1}\frac{(t^2-x^2)^{\frac{2r-1}{2}}}
{(r-1)!r!}
\left[1+O\left(\frac{T^2 \varepsilon}{\delta}\right)\right]\\
&=-{m}{\varepsilon}\cdot\frac{t+x}{s}
\left(J_1(ms)+O\left(\frac{\varepsilon}{\delta}
\log^2\frac{\delta}{\varepsilon}\,I_1(ms)\right)\right).
\end{align*}
Получаем искомую асимптотику для
$a\left(x,t,{m},{\varepsilon}\right),$
так как
\begin{align*}
I_0(ms)&\le \sum_{k=0}^\infty \frac{(ms/2)^{2k}}{k!}=e^{m^2s^2/4}\le e^{m^2t^2},\\
\frac{t+x}{s}I_1(ms)&\le \frac{t+x}{s}\cdot\frac{ms}{2}\sum_{k=0}^\infty \frac{(ms/2)^{2k}}{k!}= m\frac{t+x}{2}\,e^{m^2s^2/4}\le
mt\,e^{m^2t^2/4}\le e^{m^2t^2/2}e^{m^2t^2/4}\le e^{m^2t^2}.\\[-1.7cm]
\end{align*}
\end{proof}

\begin{proof}[Доказательство следствия~\ref{cor-uniform}]
Утверждение следует из теоремы~\ref{th-main}, так как правая часть~\eqref{eq-uniform-limit} равномерно непрерывна на каждом компактном подможестве угла $|x|<t$.
\end{proof}

\begin{proof}[Доказательство следствия~\ref{cor-concentration}]
Так как правая часть~\eqref{eq-uniform-limit} непрерывна на
$[-t+\delta; t-\delta]$, то она ограничена на этом множестве. Так как последовательность, равномерно сходящаяся к ограниченной функции, равномерно ограничена, из следствия~\ref{cor-uniform} модуль левой части~\eqref{eq-uniform-limit} меньше некоторой константы $C_{t,m,\delta}$, зависящей от $t,m,\delta$, но независящей от $x, \varepsilon$. Теперь из предложения~\ref{p-mass2} для ${t}/{2\varepsilon}\in \mathbb{Z}$
имеем
\begin{multline*}
  1-\sum_{\substack{x\in\varepsilon\mathbb{Z}:
  |x|\ge t-\delta}} P(x,t,m,\varepsilon)
  =\sum_{\substack{x\in\varepsilon\mathbb{Z}:
  |x|<t-\delta}} P(x,t,m,\varepsilon)
  =\sum_{\substack{x\in\varepsilon\mathbb{Z}:
  |x|<t-\delta}} 4\varepsilon^2
  \left|\frac{1}{2\varepsilon}a(x,t,m,\varepsilon)\right|^2
  \\
  < 4\varepsilon^2C_{t,m,\delta}^2\frac{t-\delta}{\varepsilon}\to 0
  \quad\text{при }\varepsilon\to 0.
  \\[-1.8cm]
\end{multline*}
\end{proof}

\addcontentsline{toc}{myshrinkalt}{}

\subsection{Вероятность смены хиральности: комбинаторные тождества (теорема~\ref{p-right-prob})}
\label{ssec-proofs-spin}

Хотя теорему~\ref{p-right-prob} можно вывести из~\eqref{eq-a1-distribution-between-peaks}, приведем прямое доказательство, опирающееся только на~\S\ref{ssec-proofs-basic}.


\begin{proof}[Доказательство теоремы~\ref{p-right-prob}]
Обозначим
$$
S_1(t)=\sum_xa_1^2(x,t);\;\;S_2(t)=\sum_xa_2^2(x,t);\;\;S_{12}(t)=\sum_xa_1(x,t)a_2(x,t).
$$
Согласно предложениям \ref{p-Dirac}, \ref{p-symmetry-mass} и
\ref{p-Huygens} получаем
$$
a_1(0,2t)=\frac{1}{\sqrt{2}}\sum_xa_1(x,t)(a_2(x,t)-a_1(x,t))+a_2(x,t)(a_2(x,t)+a_1(x,t))=\frac{1}{\sqrt{2}}(S_2(t)+2S_{12}(t)-S_1(t)).
$$
По определению и предложению \ref{p-Dirac} получаем
$$
S_1(t+1)-S_2(t+1)=2S_{12}(t).
$$
Тогда
$$
S_1(t+1)-S_2(t+1)=S_1(t)-S_2(t)+a_1(0,2t)\sqrt{2}.
$$
Так как  $S_1(t)+S_2(t)=1$ по предложению~\ref{p-probability-conservation}, то мы получаем рекуррентную формулу
$S_1(t+1)=S_1(t)+\frac{1}{\sqrt{2}}a_1(0,2t)$; ср.~\cite[(33)]{Jacobson-Schulman-84}. Из предложения~\ref{cor-coefficients} по индукции следует, что
$$
S_1(t)
=\frac{1}{2}\sum_{k=0}^{\lfloor t/2\rfloor-1}
\frac{1}{(-4)^k}\binom{2k}{k}.
$$

По биному Ньютона мы получаем $\sum_{k=0}^{\infty}\binom{2k}{k}x^k=\frac{1}{\sqrt{1-4x}}$ для всех $x\in \left[-\frac{1}{4},\frac{1}{4}\right)$. Подставляя $x=-\frac{1}{4}$, мы получаем
$\lim_{t\to \infty}\frac{1}{2}\sum_{k=0}^{\lfloor t/2\rfloor-1}\binom{2k}{k}\left(-\frac{1}{4}\right)^k=\frac{1}{2\sqrt{2}}.$
С помощью формулы Стирлинга оцениваем скорость сходимости:
\begin{gather*}
\left|\sum_{x\in\mathbb{Z}}a_{1}(x,t)^2-\frac{1}{2\sqrt{2}}\right|< \frac{1}{2\cdot 4^{\lfloor t/2 \rfloor}}\binom{2\lfloor t/2 \rfloor}{\lfloor t/2 \rfloor} <\frac{e}{2\pi\sqrt{2\lfloor t/2 \rfloor}}<\frac{1}{2\sqrt{t}}.\\[-1.6cm]
\end{gather*}
\end{proof}

\subsection*{Подводные камни}


Наконец, предупредим математически ориентированного читателя. Выдающиеся работы \cite{Ambainis-etal-01, Konno-05, Konno-08}
написаны настолько хорошо, что физические теоремы и доказательства в них могут быть неосторожно приняты за  математические, хотя некоторые из них неверны в том виде, как написаны. Основной источник этих проблем состоит, на самом деле, в применении неверной теоремы из математической статьи~\cite[теорема~3.3]{Chen-Ismail-91}.

Простой контрпример к \cite[теорема~3.3]{Chen-Ismail-91}: $a=b=\alpha=\beta=x=0$ и нечетное $n$. Эти значения автоматически удовлетворяют  предположениям теоремы, то есть, условию~(ii) из \cite[леммы~3.1]{Chen-Ismail-91}. Тогда по замечанию~\ref{rem-hypergeo} и предложению~\ref{cor-coefficients} левая часть равенства~\cite[(2.16)]{Chen-Ismail-91} обращается в ноль. Поэтому она не может быть эквивалентна ненулевой последовательности в правой части. Здесь мы интерпретиуем знак ``$\approx$'' в \cite[(2.16)]{Chen-Ismail-91} как эквивалентность последовательностей, следуя \cite{Konno-05}.
Попытка интерпретировать знак как стремление к нулю разности левой и правой части равенства \cite[(2.16)]{Chen-Ismail-91} обесценила бы результат, так как каждая из частей очевидным образом стремится к нулю по отдельности. \new{}

Хотя \cite{Ambainis-etal-01, Konno-05, Konno-08} отмечают небольшие ошибки в \cite{Chen-Ismail-91}, проблема более серьезна. Известные асимптотики для многочленов Якоби никогда не формулируются как эквивалентность, а вместо этого содержат аддитивный остаточный член. Его оценка трудна даже в частных случаях, изучаемых в~\cite{Kuijlaars-Martinez-Finkelshtein-04}, а случай из замечания~\ref{rem-hypergeo} отмечается, как более сложный \cite[внизу стр.~198]{Kuijlaars-Martinez-Finkelshtein-04}. Поэтому
\cite[теорема~3.3]{Chen-Ismail-91} нужно рассматривать как интересный физический, а не математический результат.

\subsection*{Благодарности}

Работа выполнена при поддержке Министерства науки и высшего
образования Российской Федерации, соглашение \No075-15-2019-1619. Работа была опробована в 2019 году как курсы на летней конференции Турнира городов в Аранджеловаце, летней школе ``Современная математика'' в Дубне, факультете математики Высшей школы экономики в Москве и Независимом Московском Университете. Авторы благодарны всем участникам этих мероприятий за их вклад, особенно  А.~Воропаеву, М.~Дмитриеву, И.~Новикову, Ф.~Ожегову за многочисленные замечания и перевод на русский язык, Е.~Ахмедовой, Р.~Валиевой за набор некоторых частей текста, И.~Богданову, И.~Гайдай-Турлову, А.~Даниярходжаеву,  Т.~Ковалеву, Ф.~Куянову, Г.~Минаеву, И.~Русских, В.~Скопенковой, М.~Федорову за некоторые рисунки, А.~Кудрявцеву, А.~Львову за приложения (этим двум авторам было тогда меньше 16 лет). Авторы благодарны за полезные обсуждения В.~Акулину, Т.~Батеневу, А.~Белавину, М.~Берштейну, М.~Бланку, А.~Бородину, В.~Бухштаберу, \new{Л.~Веласкесу,} А.~Вершику,  \new{Т.~Джекобсону,} И.~Дынникову, П.~Закорко, И.~Ибрагимову, И.~Иванову, Д.У.~Киму, Е.~Колпакову, А.~Куиджлаарсу, С.~Ландо, М.~Лифшицу, \new{М.~Маеда,} В.~Назайкинскому, С.~Нечаеву, С.~Новикову, Г.~Ольшанскому, Ю.~Петровой, И.~Полехину, П.~Пылявскому, А.~Рыбко, И.~Сабитову, А.~Семенову, \new{Т.~Татэ}, С.~Тихомирову, Д.~Трещеву, М.~Христофорову, Г.~Челнокову, В.~Чернышову, С.~Шлосману, Л.~Шульману.

\appendix

\comment
\section{G.~Minaev and I.~Russkikh. More general conservation law \mscomm{! to be deleted !}} \label{app-conservation}

\addcontentsline{toc}{myshrink}{}

\mscomm{This makes sense, only if rewritten using our notation and with a more intuitive proof via the Kirchhoff current law}

\emph{A generalization of conservation law by Gleb Minaev and Ivan Russkikh, participants of Summer conference of Tournament of towns.} Perform the change coordinates $(x,t)\mapsto(t,u)=(\frac{x+t}{2}-1, \frac{t-x}{2})$, i.e., rotate the coordinate system through $45$ degrees clockwise and shift it by the vector $(-1,0)$.

Denote $\vec{b}(t,u):=\vec{a}(t-u+1,t+u+1)$, $Q(t,u):=P(t-u+1,t+u+1)$, $\vec{B}(t,u):=(1+m^2\varepsilon^2)^{(t+u)/2}\:\vec{b}(t,u)$. The coordinates of the vectors $\vec{b}(t,u)$ and $\vec{B}(t,u)$ denote by $b_1(t,u)$, $b_2(t,u)$, $B_1(t,u)$, $B_2(t,u)$.

\textit{Remark.} In the new coordinate system, a checker moves between neighbouring points of the integer lattice rather than
along the diagonals between black squares. Also we suppose that we start at $(0,0)$ and move to any point of $(\mathbb{N}\cup\{0\})^2$, and the additional ``pre-move'' from $(-1,0)$ to $(0,0)$ is taken into account only to compute the number of turns.

For a subset $M\subset(\mathbb{N}\cup\{0\})^2$ denote $\vec{b}(t,u\text{ bypass }M):=\sum_{s}\vec{b}(s)$, where we sum is over all paths $s$ from $(0,0)$ to $(t,u)$, which bypass the points of the set $M$. Analogously, define $Q(t,u\text{ bypass }M)$, $\vec{b}(t,u\text{ bypass }M,\pm)$, $Q(t,u\text{ bypass }M,\pm)$.
Also for a set $M$ denote
$$
Q(M) := \sum_{p\in M}Q(p\text{ bypass }M\setminus\{p\}).
$$

\textit{Remark.}
Note that for an infinite set $M$ the sum becomes infinite as well. The order of a summation is irrelevant because all the summands are positive.

\begin{theorem}
For each finite set $M\subset(\mathbb{N}\cup\{0\})^2$ such that there are no infinite paths from $(0;0)$ bypassing the points of $M$ we have $Q(M)=1$.
\end{theorem}

\begin{proof}
Prove the theorem by induction over $m:=\max_{(t,u)\in M}(t+u)$, i.e., the maximal number of a downwards-right diagonal containing at least one point of the set $M$. The diagonal is called \emph{maximal}.

\textit{Base}: $m=0.$ In this case $M={(0,0)}$ and $Q(M)=Q(0,0)=1$.

\textit{Step.} Let us prove a lemma.

\begin{lemma}\label{up_right_recur_lemm}
If a set $A\subset(\mathbb{N}\cap\{0\})^2$, which does not contain the points $(t,u), (t,u+1),(t+1,u)\in(\mathbb{N}\cap\{0\})^2$ , then
$$
Q(t,u\text{ bypass }A)=Q(t,u+1\text{ bypass }A, -)+Q(t+1,u\text{ bypass }A, +).
$$
\end{lemma}

\begin{proof}
This is straightforward:
\begin{multline*}
Q(t,u+1\text{ bypass }A,-)+Q(t+1,u\text{ bypass }A,+)
=b_1(t,u+1\text{ bypass }A)^2+b_2(t+1,u\text{ bypass }A)^2=\\
=\frac{(b_1(t,u\text{ bypass }A)+ b_2(t,u\text{ bypass }A))^2+(b_2(t,u\text{ bypass }A)- b_1(t,u\text{ bypass }A))^2}{2}=\\
=b_1(t,u\text{ bypass }A)^2+b_2(t,u\text{ bypass }A)^2
=Q(t,u\text{ bypass }A).\\[-1.5cm]
\end{multline*}
\end{proof}

Suppose that the point $(t,u) \in M$ is such that $t+u$ is maximal. Suppose that there is a checker path starting at the point $(0,0)$ and ending at $(t,u)$ with the last move, say, in the upwards direction bypassing all the points of the set $M\setminus\{(t,u)\}$. Then there exists a path to the point $(t,u-1)$, bypassing all the points of the set $M\setminus\{(t,u-1)\}$. Notice that in this case $(t,u-1)\notin M$, because the move to $(t;u)$ is in the upwards direction. Hence there exists a path to the point $(t+1,u-1)$, bypassing all the points of the set $M\setminus\{(t+1,u-1)\}$. If $(t+1,u-1)\notin M$ then there exists a path going to infinity passing through $(t+1,u-1)$ and bypassing all the points of the set $M$. For example, the path turning right at the point $(t,u)$ and only going right from there passes through $(t,u)$ and bypasses all the points in $M$ because the diagonal is maximal. Therefore $(t+1,u-1)\in M$. Notice that $Q(a,b\text{ bypass }A)=Q(a,b\text{ bypass }A,+)+Q(a,b\text{ bypass }A,-)$ for any set $A\subset(\mathbb{N}\cap\{0\})^2$. Denote $K:=M\setminus\{(t,u);(t+1,u-1)\}$. Then we have the following chain of equalities (where we use \ref{up_right_recur_lemm} for the set $A=K$):
\begin{multline*}
Q(M)=Q(K)+Q(t,u\text{ bypass }M\setminus\{(t,u)\})+Q(t+1,u-1\text{ bypass }M\setminus\{(t+1,u-1)\})=\\
=Q(K)+Q(t,u\text{ bypass }M\setminus\{(t,u)\},+)+Q(t,u\text{ bypass }M\setminus\{(t,u)\},-)+\\
+Q(t+1,u-1\text{ bypass }M\setminus\{(t+1,u-1)\},+)+Q(t+1,u-1\text{ bypass }M\setminus\{(t+1,u-1)\},-)=\\
=Q(K)+Q(t,u\text{ bypass }M\setminus\{(t,u)\},+)+\\
+Q(t,u-1\text{ bypass }M\setminus\{(t,u-1)\})+Q(t+1,u-1\text{ bypass }M\setminus\{(t+1,u-1)\},-)
=Q(M\cup \{(t,u-1)\}).
\end{multline*}
Thus if the point $(t,u-1)$ is added to the set $M$ then the probability $Q(M)$ is not changed. 

This way we put new points onto the diagonal $t+u=m-1,$ therefore, the maximal diagonal is not changed. Thus if we add a few points to the set $M$, then the paths bypassing other points of the set $M$ bypass the points of $M$ in the maximal diagonal as well. Therefore, if we remove all the points from the maximal diagonal, then $Q(M)$ is not changed and no infinite path bypassing the points of the set $M$ appears. This way we change $M$ keeping $Q(M)$ fixed but decreasing the number of a maximal diagonal. By the inductive hypothesis the new $Q(M)$ equals $1$, hence the old one also equals $1$.
\end{proof}

\endcomment

\addcontentsline{toc}{myshrink}{}

\section{А.~Кудрявцев. Альтернативные <<явные>> формулы} \label{app-formula}

\addcontentsline{toc}{myshrink}{}


Положим $\binom{n}{k}:=0$ для целых $k<0<n$ или $k>n>0$. Обозначим
$\theta(x):=\begin{cases}
                            1, & \mbox{при } x\ge0, \\
                            0, & \mbox{при } x<0.
                         \end{cases}$

\begin{proposition}[<<Явная>> формула] \label{Feynman-binom-alt}
 {Для всех целых $|x|<t$, таких что $x+t$ четно, выполнено:
\begin{align*}
&\mathbf{(A)} &
a_1(x,t)&= {2^{(1-t)/2}\sum\limits_{r=0}^{(t-|x|)/2}
(-2)^{r}\binom{(x+t-2)/2}{r}\binom{t-r-2}{(x+t-2)/2},}\\
& &
a_2(x,t)&= {2^{(1-t)/2}\sum\limits_{r=0}^{(t-|x|)/2}
(-2)^{r}\binom{(x+t-2)/2}{r}\binom{t-r-2}{(x+t-4)/2}};\\
&\mathbf{(B)} &
a_1(x,t)&=2^{(1-t)/2}
\sum_{r=0}^{(t-|x|)/2} (-1)^r\binom{(t-|x|-2)/2}{r}
\binom{|x|}{(t+|x|-4r-2)/2},\\
& &
a_2(x,t)&=2^{(1-t)/2}
\sum_{r=0}^{(t-|x|)/2} (-1)^r\binom{(t-|x|-2)/2}{r-\theta(x)}
\binom{|x|}{(t+|x|-4r)/2}.
\end{align*}
}
\end{proposition}

\begin{proof}[Доказательство предложения~\ref{Feynman-binom-alt}(A)]
Введем производящие функции
$$\hat a_1(p,q):=2^{n/2}\,\sum\limits_{n>k\ge 0}a_1(2k-n+1,n+1)p^kq^n\qquad
\text{и}\qquad
\hat a_2(p,q):=2^{n/2}\,\sum\limits_{n>k\ge 0}a_1(2k-n+1,n+1)p^kq^n.$$
По предложению~\ref{p-Dirac} получаем следующие уравнения:
\begin{equation*}
 \begin{cases}
   \hat a_1(p,q)-\hat a_1(p,0)=q\cdot(\hat a_2(p,q)+\hat a_1(p,q))\\
   \hat a_2(p,q)-\hat a_2(p,0)=pq\cdot(\hat a_2(p,q)-\hat a_1(p,q))
 \end{cases}
\end{equation*}
Так как $\hat a_1(p,0)=0$ и $\hat a_2(p,0)=1$, то решив систему, получим
$$\hat a_2(p,q)=\dfrac {1-q}{1-q-pq+2pq^2} 
\qquad\text{и}\qquad
\hat a_1(p,q)=\dfrac q{1-q-pq+2pq^2}=q+q^2(1+p-2pq)+q^3(1+p-2pq)^2+\dots
$$
Коэффициент при $p^kq^n$ в $\hat a_1(p,q)$ равен $$\sum\limits_{j=\max(k,n-k)}^n (-2)^{n-j-1}\cdot\binom{j}{n-j-1\quad k-n+j+1 \quad j-k},$$
так как из каждого слагаемого вида $q^{j+1}(1+p-2pq)^j$ берется ровно одна комбинация множителей:
\begin{itemize}
	\item чтобы степень при $q$ была равна $n$, множителей $-2pq$ должно быть $n-j-1$;
	\item чтобы степень при $p$ была равна $k$, множителей $p$ должно быть $k-(n-j-1)$;
	\item чтобы всего множителей было $j$, множителей $1$ должно быть $j-(k-(n-j-1))-(n-j-1)=j-k$.
\end{itemize}
Заменяя переменную суммирования на $r=n-j-1$, получаем требуемое выражение для $a_1(x,t)$.

Выражение для $a_2(x,t)$ следует из выражения для $a_1(x,t)$, предложения~\ref{p-Dirac} и правила Паскаля:
\begin{align*}
a_2(x,t)&=\sqrt{2}\,a_1(x-1,t+1)-a_1(x,t)\\
&=2^{(1-t)/2}\sum\limits_{r=0}^{(t-|x|)/2}
(-2)^{r}\binom{(x+t-2)/2}{r}
\left(\binom{t-r-1}{(x+t-2)/2}-\binom{t-r-2}{(x+t-2)/2}\right)\\
&=2^{(1-t)/2}\sum\limits_{r=0}^{(t-|x|)/2}
(-2)^{r}\binom{(x+t-2)/2}{r}\binom{t-r-2}{(x+t-4)/2}.\\[-1.4cm]
\end{align*}
\end{proof}

\begin{proof}[Доказательство предложения~\ref{Feynman-binom-alt}(B)]
(А.~Воропаев)
Согласно предложению~\ref{cor-coefficients}, для всех $|x|<t$ числа $a_1(x,t)$ и $a_2(x,t)$ равны коэффициентам при $z^{(t-x-2)/2}$ и $z^{(t-x)/2}$ в разложении многочлена
\begin{equation*}
	2^{(1-t)/2}(1+z)^{(t-x-2)/2}(1-z)^{(t+x-2)/2}=
	\begin{cases}
		2^{(1-t)/2}(1-z^2)^{\frac{t-x-2}2}(1-z)^x,
&\text{при } x\ge 0;\\
		2^{(1-t)/2}(1-z^2)^{\frac{t+x-2}2}(1+z)^{-x},
&\text{при } x< 0.
	\end{cases}
\end{equation*}
При $x<0$ из этого сразу получаем требуемое. При $x\ge 0$ нужно еще заменить переменную суммирования на $r'=(t-x-2)/2-r$ или $r'=(t-x)/2-r$ для $a_1(x,t)$ и $a_2(x,t)$ соответственно.
\end{proof}

\section{А.~Львов. Поточечный непрерывный предел} \label{app-pointwise}

\addcontentsline{toc}{myshrink}{}



\begin{theorem}[Поточечный непрерывный предел]  \label{p-convergence}
Для любых вещественных $m\ge 0$ и $|x|<t$ выполнено
\begin{align*}
\lim_{n\to \infty} n\, a_1\left(\frac{2}{n}\left\lfloor \frac{nx}{2}\right\rfloor,\frac{2}{n}\left\lfloor \frac{nt}{2}\right\rfloor,m,\frac{1}{n}\right)
&= m\,J_0(m\sqrt{t^2-x^2});\\
\lim_{n\to \infty} n\, a_2\left(\frac{2}{n}\left\lfloor \frac{nx}{2}\right\rfloor,\frac{2}{n}\left\lfloor \frac{nt}{2}\right\rfloor,m,\frac{1}{n}\right)
&= -m\frac{x+t}{\sqrt{t^2-x^2}}J_1(m\sqrt{t^2-x^2}).
\end{align*}
\end{theorem}


\begin{proof}[Доказательство теоремы~\ref{p-convergence}]
Обозначим $A := \lfloor \frac{nx}{2} \rfloor+\lfloor \frac{nt}{2} \rfloor$ и $B := \lfloor \frac{nt}{2} \rfloor-\lfloor \frac{nx}{2} \rfloor$.
Первый предел вычисляется следующим образом:
\begin{multline*} 
n\, a_1\left(\frac{2}{n}\left\lfloor \frac{nx}{2}\right\rfloor,\frac{2}{n}\left\lfloor \frac{nt}{2}\right\rfloor,m,\frac{1}{n}\right)
=
n\left(1 + \frac{m^2}{n^2}\right)^{\lfloor \frac{nt}{2} \rfloor - \frac{1}{2}} \cdot\sum \limits _{r = 0}^{\infty} (-1)^{r}\binom{A-1}{r}\binom{B-1}{r}\left(\frac{m}{n}\right)^{2r+1}
\\
\thicksim
\sum \limits _{r = 0}^{\infty} (-1)^{r}\binom{A-1}{r}\binom{B-1}{r}\frac{m^{2r+1}}{n^{2r}} =
\sum\limits_{r = 0; 2 | r}^{\infty} \binom{A-1}{r}\binom{B-1}{r}\frac{m^{2r+1}}{n^{2r}} -
\sum\limits_{r = 0; 2 \nmid r}^{\infty} \binom{A-1}{r}\binom{B-1}{r}\frac{m^{2r+1}}{n^{2r}}
\\ \to
\sum\limits_{r = 0; 2 | r}^{\infty}
\frac{(x + t)^r(t - x)^r m^{2r + 1}}{2^{2r}(r!)^2} -
\sum\limits_{r = 0; 2 \nmid r}^{\infty}
\frac{(x + t)^r(t - x)^r m^{2r + 1}}{2^{2r}(r!)^2} =
m\, J_0(m\sqrt{t^2 - x^ 2}) \qquad\text{при } n \rightarrow \infty.
\end{multline*}
Здесь равенство в первой строке --- это предложение~\ref{p-mass3}. Эквивалентность во второй строке следует из
$$
1 \le \left(1 + \frac{m^2}{n^2}\right)^{\lfloor \frac{nt}{2} \rfloor - \frac{1}{2}}
\le \left(1 + \frac{m^2}{n^2}\right)^{nt}
= \sqrt[n]{\left(1 + \frac{m^2}{n^2}\right)^{n^2t}} \thicksim \sqrt[n]{e^{m^2t}} \rightarrow 1 \qquad\text{при } n \rightarrow \infty,
$$
по теореме о двух полицейских. Равенство во второй строке выполнено, так как
$\binom{A-1}{r}\binom{B-1}{r}\frac{m^{2r+1}}{n^{2r}} = 0$ при $r > \max\{A, B\}$, а значит, все три участвующие суммы конечны. Сходимость в третьей строке устанавливается в леммах~\ref{r2}--\ref{l} ниже.
Второй предел в теореме вычисляется аналогично.
\end{proof}



\begin{lemma}
\label{r2}
Для каждого натурального $r$ выполнено
$\lim\limits_{n \rightarrow \infty} \binom{A-1}{r}\binom{B-1}{r}\frac{m^{2r+1}}{n^{2r}} =
 \frac{(x + t)^r(t - x)^r m^{2r + 1}}{2^{2r}(r!)^2}$.
\end{lemma}

\begin{proof}
Имеем $$ 
\binom{A-1}{r}\binom{B-1}{r}\frac{m^{2r+1}}{n^{2r}} =
\frac{(A - 1)\dots(A - r)\cdot(B - 1)\dots(B - r)}{(r!)^2} \cdot \frac{m^{2r + 1}}{n^{2r}} \to
\left(\frac{x + t}{2}\right)^r\left(\frac{t - x}{2}\right)^r \frac{m^{2r + 1}}{(r!)^2}$$
при $n \rightarrow \infty$, так как для каждого $1\le i\le r$
$$\lim\limits_{n \rightarrow \infty} \frac{A - i}{n} = \lim\limits_{n \rightarrow \infty} \frac{A}{n} =
\lim\limits_{n \rightarrow \infty} \frac{\lfloor \frac{nx}{2} \rfloor + \lfloor \frac{nt}{2} \rfloor}{n} = \lim\limits_{n \rightarrow \infty} \frac{\frac{nx}{2} + \frac{nt}{2} + o(n)}{n} = \frac{x + t}{2}$$
и аналогично, $\lim\limits_{n \rightarrow \infty} \frac{B - i}{n} = \frac{t - x}{2}$.
\end{proof}

\begin{lemma}
\label{r3} Для каждого натурального $r$ выполнено
$\binom{A-1}{r}\binom{B-1}{r}\frac{m^{2r+1}}{n^{2r}} \leq \frac{(x + t)^r(t - x)^r m^{2r + 1}}{2^{2r}(r!)^2}$.
\end{lemma}

\begin{proof} Это доказывается аналогично, так как для каждого $1\le i\le r<\min\{A,B\}$ выполнено
\begin{align*}
(A-i)(B-i) \leq 
\left(\left\lfloor \frac{nx}{2} \right\rfloor + \left\lfloor \frac{nt}{2} \right\rfloor - 1\right)\left(\left\lfloor \frac{nt}{2} \right\rfloor - \left\lfloor \frac{nx}{2} \right\rfloor - 1\right) \leq
\left(\frac{nx}{2} + \frac{nt}{2}\right)\left(\frac{nt}{2} - \frac{nx}{2}\right).
\\[-1.2cm]
\end{align*}
\end{proof}










\begin{lemma}
\label{l}
Пусть последовательность
$\{a_k(n)\}_{k=0}^\infty$ неотрицательных последовательностей такова, что $\lim\limits_{n \rightarrow \infty} a_k(n) = b_k$ для каждого $k$; $a_k(n) \leq b_k$ для всех $k, n$; и $\sum\limits_{k = 0}^{\infty}b_k$ конечна. Тогда $\lim\limits_{n \rightarrow \infty} \sum\limits_{k = 0}^{\infty}a_k(n) = \sum\limits_{k = 0}^{\infty}b_k$.
\end{lemma}

\begin{proof}
Обозначим $b := \sum\limits_{k = 0}^{\infty}b_k$.
Тогда $\sum\limits_{k = 0}^{\infty}a_k(n) \leq b$  для всех $n$.
Возьмем $\varepsilon > 0$. Возьмем $N$ такое, что $\sum\limits_{k = 0}^{N}b_k > b - \varepsilon$. Для каждого $k \leq N$ возьмем $M_k$ такое, что для всех $n \geq M_k$ выполнено $ a_k(n) > b_k - \frac{\varepsilon}{2^{k + 1}}$.  Тогда для всех
$n > \max\limits_{0 \leq k \leq N}M_k$
выполнено $\sum\limits_{k = 0}^{\infty}a_k(n) > b - 2\varepsilon$.
Итак, $\lim\limits_{n \rightarrow \infty} \sum\limits_{k = 0}^{\infty}a_k(n) = b.$
\end{proof}


{
\footnotesize

\noindent
\textsc{Михаил Скопенков\\
Национальный исследовательский университет Высшая школа экономики\\ (факультет математики) и\\
Институт проблем передачи информации им. А.А. Харкевича
Российской академии наук} 
\\
\texttt{mikhail.skopenkov\,@\,gmail$\cdot $com} \quad \url{https://users.mccme.ru/mskopenkov/}

\vspace{0.3cm}
\noindent
\textsc{Алексей Устинов\\
Национальный исследовательский университет Высшая школа экономики\\ (факультет компьютерных наук) и\\
Хабаровское отделение Института прикладной математики \\
Дальневосточного отделения Российской академии наук} 
\\
\texttt{Ustinov.Alexey\,@\,gmail$\cdot $com} \quad
\url{http://iam.khv.ru/staff/Ustinov.php}

}